\documentclass[11pt,letterpaper]{article}
\usepackage[T1]{fontenc}
\usepackage[utf8]{inputenc}
\usepackage{times}
\usepackage{amsmath,amssymb,amsfonts,amsthm,mathrsfs}
\usepackage[margin=1in]{geometry}
\usepackage{microtype}
\usepackage{graphicx,booktabs,rotating,placeins}
\usepackage{xcolor,url}
\usepackage[authoryear,round]{natbib}
\usepackage{algorithm2e}
\usepackage{authblk}
\usepackage{enumitem}
\usepackage{titlesec}
\usepackage[colorlinks=true,linkcolor=blue,citecolor=blue,urlcolor=blue,
  bookmarksnumbered=true,bookmarksopen=false]{hyperref}
\hypersetup{
  pdftitle={High-Dimensional Change Point Analysis for Temporally Dependent Data},
  pdfauthor={Xiaoyi Wang, Le Zhou, Jixuan Liu, Long Feng},
  pdfsubject={High-dimensional change point inference for dependent data},
  pdfkeywords={Cauchy combination test, Change point inference, High-dimensional data, Time series}
}
\graphicspath{{Fig/}}
\allowdisplaybreaks[4]
\setlist{topsep=4pt plus 1pt,itemsep=2pt,parsep=0pt}
\titleformat{\section}{\normalfont\large\bfseries}{\thesection}{0.75em}{}
\titleformat{\subsection}{\normalfont\normalsize\bfseries}{\thesubsection}{0.75em}{}
\titlespacing*{\section}{0pt}{15pt plus 3pt minus 2pt}{7pt plus 2pt minus 1pt}
\titlespacing*{\subsection}{0pt}{11pt plus 2pt minus 1pt}{5pt plus 1pt}
\AtBeginDocument{
  \setlength{\abovedisplayskip}{7pt plus 2pt minus 2pt}
  \setlength{\belowdisplayskip}{7pt plus 2pt minus 2pt}
  \setlength{\abovedisplayshortskip}{4pt plus 1pt}
  \setlength{\belowdisplayshortskip}{5pt plus 2pt minus 1pt}
}

\newcommand{\bm}[1]{\boldsymbol{#1}}
\newcommand{\bX}{\boldsymbol X}
\newcommand{\bGam}{\boldsymbol{\Gamma}}
\newcommand{\bOme}{\boldsymbol{\Omega}}
\newcommand{\bmu}{\boldsymbol{\mu}}
\newcommand{\bdelta}{\boldsymbol{\delta}}
\newcommand{\bepsilon}{\boldsymbol{\epsilon}}
\newcommand{\T}{{\!\top\!}}
\newcommand{\ind}[1]{\mathbb{I}(#1)}
\newcommand{\E}{\mathbb{E}}
\DeclareMathOperator{\Var}{Var}
\DeclareMathOperator{\Cov}{Cov}
\DeclareMathOperator{\Cum}{cum}
\DeclareMathOperator{\Diag}{diag}
\DeclareMathOperator{\Tr}{tr}
\DeclareMathOperator{\Sgn}{sgn}
\DeclareMathOperator*{\argmin}{arg\,min}
\DeclareMathOperator*{\argmax}{arg\,max}

\newtheoremstyle{preprintplain}{7pt}{7pt}{\itshape}{}{\bfseries}{.}{0.5em}{}
\newtheoremstyle{preprintremark}{7pt}{7pt}{}{}{\itshape}{.}{0.5em}{}
\theoremstyle{preprintplain}
\newtheorem{theorem}{Theorem}
\newtheorem{lemma}{Lemma}
\newtheorem{proposition}{Proposition}
\newtheorem{corollary}{Corollary}
\theoremstyle{preprintremark}
\newtheorem{assumption}{Assumption}
\newtheorem{remark}{Remark}
\newenvironment{prop}{\begin{proposition}}{\end{proposition}}
\newenvironment{coro}{\begin{corollary}}{\end{corollary}}

\title{\bfseries High-Dimensional Change Point Analysis for Temporally Dependent Data}
\author[1]{Xiaoyi Wang}
\author[2]{Le Zhou}
\author[3]{Jixuan Liu}
\author[3]{Long Feng}
\affil[1]{Faculty of Arts and Sciences, Beijing Normal University, Zhuhai, Guangdong 519087, China}
\affil[2]{Department of Mathematics, Hong Kong Baptist University, Hong Kong}
\affil[3]{School of Statistics and Data Science, Nankai University, Tianjin 300071, China}
\affil[ ]{\href{mailto:wangxy059@bnu.edu.cn}{wangxy059@bnu.edu.cn};
  \href{mailto:lezhou@hkbu.edu.hk}{lezhou@hkbu.edu.hk};
  \href{mailto:jxliu@mail.nankai.edu.cn}{jxliu@mail.nankai.edu.cn};
  \href{mailto:flnankai@nankai.edu.cn}{flnankai@nankai.edu.cn}}

\date{}

\begin{document}
\maketitle
\begin{abstract}
This paper develops adaptive procedures for detecting and locating mean changes in high-dimensional time series. Quadratic CUSUM statistics target dense changes, whereas coordinatewise maximum statistics target sparse changes. Two weighting schemes are considered to accommodate both interior and boundary changes. Under non-Gaussian vector dependence, we establish the limiting distributions, the consistency of the required centering and scale estimators, and prove asymptotic independence between  quadratic and maximum statistics. These results justify Cauchy combination tests. We further establish single change localization and consistent multiple change recovery using wild binary segmentation. Numerical results illustrate the effectiveness of the proposed methods.
\end{abstract}


\section{Introduction}

Change point inference plays a critical role in the analysis of high-dimensional data, where the structure or behavior of a system may shift over time due to external shocks, policy changes, or underlying regime switches. However, in modern applications such as genomics, finance, neuroimaging, and environmental monitoring, the traditional low-dimensional change point detection methods may fail when data are often collected across a large number of variables simultaneously and the dimension may be comparable to or even exceed the sample size. High-dimensional change point inference therefore aims to determine whether a structural break has occurred and to estimate its location, with the unknown sparsity of the change and the potentially complex temporal and cross-sectional dependence in the data.

A widely studied problem is testing for a change in the mean of a high-dimensional vector. Many existing studies focus on constructing test statistics by taking the maximum or the sum over the sample size or across dimensions of the corresponding individual CUSUM-type statistics. 
Sum-type statistics combine information over all coordinates and are typically powerful against dense alternatives; representative contributions include \citet{10.1016/j.jeconom.2009.10.020}, \citet{10.1111/j.1467-9892.2012.00796.x}, \citet{10.1007/s11425-016-0058-5}, and \citet{10.1214/17-aos1610}.  Coordinatewise maximum statistics, in contrast, are designed for sparse alternatives; see, among others, \citet{10.1214/15-aos1347}, \citet{EnikeevaHarchaoui2019}, and \citet{10.1111/rssb.12406}. The minimax detection boundaries for sparse high-dimensional mean changes are studied by \citet{LiuGaoSamworth2021}. Since the relative power of sum-type and max-type methods depends on both the number of affected coordinates and the magnitudes of their changes, neither class is uniformly preferable when the sparsity level is unknown.

Several adaptive tests combine statistics designed for different sparsity regimes.  \citet{10.1111/rssb.12375} introduced an adaptive strategy based on adjusted max-$L_q$ aggregation. \citet{doi:10.1080/01621459.2021.1884562} proposed a family of self-normalized U-statistics corresponding to each $L_q$ norm, providing a unified framework for adaptivity. \citet{wang2023JRSSB} proposed adaptive methods based on the max-$L_\infty$-type statistics \citep{10.1214/15-aos1347,10.1111/rssb.12406} and the sum-$L_2$-type statistic \citep{10.1142/s201032631950014x} and established the asymptotic independence between them. For a comprehensive review of recent advances in high-dimensional change point inference, including these adaptive strategies, see \cite{10.1016/j.jmva.2021.104833}, \cite{wang2023JRSSB} and \cite{meng2024covariate}.

 Most of the corresponding asymptotic theory is developed for independent observations.  For temporally dependent data, \cite{li2019change} proposed a sum-\(L_2\)-type test procedure that aggregates a bias-corrected version of the classic CUSUM test statistic. Meanwhile, \cite{wang2022inference} developed a trimmed version of the self-normalized U-statistic, which excludes pairs of points that are close in time. {Both methods aggregate quadratic terms across coordinates and are designed for dense alternatives.}  For sparse alternatives,  the two max-\(L_\infty\)-type test statistics \citep{10.1214/15-aos1347,10.1111/rssb.12406} can also be applied to temporally dependent high-dimensional data, provided that the variance estimator is replaced by a long-run variance estimator. 

Testing concerns the existence of a change, whereas in many applications its location is the main inferential target. For a single high-dimensional mean change, \citet{WangSamworth2018} estimate a sparse projection direction and locate the change in the projected series, \citet{10.1111/rssb.12406} develop finite-sample inference and identification results, and \citet{WangShao2023} study a U-statistic-based estimator and its limiting distribution.  The self-normalized approach of \citet{wang2022inference} also provides change point estimators for independent and weakly dependent observations.  Compared with global testing, location estimation requires uniform control of the stochastic CUSUM process in a neighborhood of the population maximizer.

The multiple change problem further requires estimation of both the number and the locations of the changes while preventing nearby changes from masking one another.  \citet{WangSamworth2018} incorporate their sparse-projection estimator into wild binary segmentation, and \citet{Cho2016} develop a double-CUSUM binary segmentation for panel mean changes.  \citet{10.1142/s201032631950014x} combine dense and sparse information and estimate multiple locations by screening, model selection, and dynamic programming.  \citet{wang2022inference} combine self-normalized tests with wild binary segmentation, whereas \citet{ChenWangWu2022} use moving-sum screening followed by CUSUM refinement for temporally dependent high-dimensional series.  \citet{PilliatCarpentierVerzelen2023} aggregate local homogeneity tests and establish near-optimal or optimal detection guarantees over a broad class of high-dimensional models.  Wild binary segmentation and related multiscale procedures provide general tools for isolating individual changes \citep{Fryzlewicz2014WBS,BaranowskiChenFryzlewicz2019}. 

{We develop max-$L_2$-type and max-$L_\infty$-type procedures for testing and locating mean changes in temporally dependent high-dimensional data.} %
{For each type of procedures, we consider an unweighted version and a trimmed, weighted version. The quadratic statistics use dependence-adjusted centering and scaling; the coordinatewise maxima use difference-based estimates of marginal long-run variances. Pairing statistics with the same weighting gives two Cauchy combination tests that adapt to unknown signal sparsity. We use the same two components for single and multiple change localization.} %
The main contributions are threefold.
\begingroup

\begin{itemize}
\item[(i)] We propose two max-$L_2$-type statistics for change point inference in high-dimensional time series data with temporal dependence. The numerical results show that our proposed methods can improve power over the sum-\(L_2\)-type procedure in the dense settings considered.
\item[(ii)] We establish a unified asymptotic theory under strong mixing, $m$-dependence, and physical dependence. The results include the limiting distributions of the max-$L_2$- and max-$L_\infty$-type test statistics, their asymptotic independence, consistency under dense and sparse alternatives, and localization rates for the corresponding single change estimators. We also establish uniform consistency of the nuisance-parameter estimators required for the max-$L_\infty$-type statistics. 
\item[(iii)] We extend the max-$L_2$-type, max-$L_\infty$-type, and adaptive methods to multiple change estimation by wild binary segmentation.  At each recursion, a shortest-interval rule among the significant intervals and a separation rule are used to isolate individual changes and prevent duplicate detections.  Under explicit spacing, dependence, and signal conditions, {we prove consistency for the estimated number of changes and derive localization rates for all proposed methods.}
\end{itemize}
\endgroup

The remainder of the paper is organized as follows.  Section~\ref{sec:2-L2} formulates the single change problem and develops the max-$L_2$-type tests with feasible centering and scaling. Section~\ref{sec:3-adaptive} introduces the max-$L_\infty$-type statistics and the adaptive tests.  Section~\ref{sec:location} studies single change localization, and Section~\ref{sec:wbs} develops the multiple change procedures. The simulation study and real data applications are {presented in Sections~\ref{sec:simulation-compatibility} and~\ref{sec:applications}}, respectively. All technical proofs are provided in the Supplementary Material.

\textbf{Notation.}
Throughout, $a\wedge b=\min(a,b)$, $a_{+}=\max(a,0)$ and $\mathbb{I}(\cdot)$ is the indicator function. For a scalar random variable $Z$ and $q\ge1$, let $\|Z\|_q=(\E|Z|^q)^{1/q}$. For vectors and matrices, $\|\cdot\|_2$, $\|\cdot\|_\infty$, $\|\cdot\|_{\mathrm{op}}$, and $\|\cdot\|_{\mathrm F}$ denote the Euclidean, maximum, operator, and Frobenius norms, respectively; $\lambda_{\min}(\cdot)$, $\lambda_{\max}(\cdot)$, and $\Tr(\cdot)$ denote the extreme eigenvalues and trace.  For positive sequences, $a_n\lesssim b_n$ means $a_n\le Cb_n$ for a constant $C$ independent of $n$ and $p$, and $a_n\asymp b_n$ means that both $a_n\lesssim b_n$ and $b_n\lesssim a_n$ hold. For random variables $Y_1,\ldots,Y_m$ with finite $m$th moments, let $\Cum(Y_1,\ldots,Y_m)=\sum_{\pi\in\Pi_m}(|\pi|-1)!(-1)^{|\pi|-1}\prod_{B\in\pi}\E\left(\prod_{a\in B}Y_a\right)$ denote their joint cumulant, where $\Pi_m$ is the collection of all partitions of $\{1,\ldots,m\}$ and $|\pi|$ is the number of blocks in the partition $\pi$. Let $X_n\xrightarrow{d}X$ or $X_n\Rightarrow X$ denote convergence in distribution, with $\Rightarrow$ denoting weak convergence in the specified function space for stochastic processes; $X_n\xrightarrow{p}X$ and $X_n\xrightarrow{\mathrm{a.s.}}X$ denote convergence in probability and almost-sure convergence, respectively.

\section{\texorpdfstring{Problem formulation and max-$L_2$-type tests}{Problem formulation and max-L2-type tests}}\label{sec:2-L2}
Let $\bX_i=(X_{i1},\ldots,X_{ip})^\top$, $i=1,\ldots,n$, be $p$-dimensional observations generated by the single change mean model
\begin{equation}\label{eq:single change-model}
\bX_i=\bmu_0+\bdelta\ind{i>\tau}+\bepsilon_i,
\qquad i=1,\ldots,n,
\end{equation}
where $\bmu_0\in\mathbb R^p$ is the baseline mean, $\bdelta\in\mathbb R^p$ is the jump vector, $\tau\in\{1,\ldots,n-1\}$ is the change location, and $\bepsilon_i=(\epsilon_{i1},\ldots,\epsilon_{ip})^\top$.  The error process $\{\bepsilon_i\}$ is centered and strictly stationary; its tail and temporal-dependence conditions are stated below.  Under the null hypothesis, we set $\bdelta=\mathbf0$ and use $\tau=n$ for notational convenience.  We focus on testing whether there exists a changepoint, that is,
\begin{equation}\label{eq:single change-hypothesis}
H_0:\ \bdelta=\mathbf0
\qquad\text{against}\qquad
H_1:\ \tau\in\{1,\ldots,n-1\}\ \text{and}\ \bdelta\ne\mathbf0.
\end{equation}

For dense alternatives, we consider the quadratic CUSUM statistic
\[
W(k)=\frac{1}{n\sqrt p}\left\|\sum_{i=1}^k\bX_i-\frac{k}{n}\sum_{i=1}^n\bX_i\right\|_2^2,\qquad k=1,\ldots,n-1.
\]
  Let $\bGam(h)=\Cov(\bepsilon_0,\bepsilon_h)$ denote the lag-$h$ covariance matrix, $\bOme=\sum_{h\in\mathbb Z}\bGam(h)$ be the long-run covariance matrix, and define the spectral density matrix by $\mathbf f_p(\lambda)=(2\pi)^{-1}\sum_{h\in\mathbb Z} \bGam(h)e^{-\mathrm i h\lambda}$, for $-\pi\le\lambda\le\pi$. {The centering and long-run scaling of $W(k)$ depend on the lagged covariance matrices}%
  . {We therefore derive the appropriate centering and long-run scale. Their estimators account for the number of available observation pairs at each lag and do not require the lagged covariance matrices to be symmetric.}

{We impose the following assumptions.}
\begin{assumption}%
\label{ass:C1}
For every $p$, $\{\bepsilon_i:i\in\mathbb Z\}$ is a strictly stationary $p$-dimensional process with $\mathrm E(\bepsilon_i)=0$. There exists a finite constant $K_{\mathrm{sg}}<\infty$ such that
\(
\sup_{\|\mathbf u\|_2=1}\sup_{q\ge2}q^{-1/2}\|\mathbf u^\top\bepsilon_i\|_q\le K_{\mathrm{sg}}.
\)
In addition, with $K_{\mathrm{sg}}$ and all constants below independent
of $p$, one of the following conditions holds.
\begin{list}{}{%
  \setlength{\leftmargin}{3.1em}%
  \setlength{\labelwidth}{2.5em}%
  \setlength{\labelsep}{0.6em}%
}
\item[\textnormal{(AM)}] There exist constants $K_1,K_2>0$ and $\varkappa>0$ such that the strong-mixing coefficient of the vector process,
\[
\alpha_{\epsilon,p}(r)=\sup_{t\in\mathbb Z}
\sup_{A\in\sigma(\bepsilon_s:s\le t),\,
B\in\sigma(\bepsilon_s:s\ge t+r)}
|\Pr(A\cap B)-\Pr(A)\Pr(B)|,
\]
satisfies $\alpha_{\epsilon,p}(r)\le K_1\exp(-K_2r^{\varkappa})$ for $r\ge1$.
\item[\textnormal{(MD)}] The process $\{\bepsilon_i:i\in\mathbb Z\}$ is $m_0$-dependent for a fixed integer $m_0\ge0$.
\item[\textnormal{(PD)}] %
There exist an i.i.d. innovation sequence $\{\boldsymbol\xi_i\}_{i\in\mathbb Z}$ and a measurable function $\mathcal G_p$ such that $\bepsilon_i=\mathcal G_p(\boldsymbol\xi_i,\boldsymbol\xi_{i-1},\ldots)$. 
If $\bepsilon_i^{\{r\}}$ is obtained by replacing $\boldsymbol\xi_{i-r}$ by an independent copy $\boldsymbol\xi_{i-r}'$, there exist constants $K_1,K_2>0$ and
$\varkappa>0$ such that, for every $i\in\mathbb Z$,
\[
\sup_{\|\mathbf u\|_2=1}\sup_{q\ge2}q^{-1/2}\|\mathbf u^\top(\bepsilon_i-\bepsilon_i^{\{r\}})\|_q\le K_1\exp(-K_2r^{\varkappa}),\qquad r\ge0.
\]
\end{list}
\end{assumption}

\begin{assumption}%
\label{ass:C2}
 There exist constants $0<C_0<C_1<\infty$ and $c_f>0$ such that
$C_0\le\lambda_{\min}(\bOme)\le\lambda_{\max}(\bOme)\le C_1$, $\inf_{-\pi\le\lambda\le\pi}p^{-1}\Tr\{\mathbf f_p(\lambda)\}\ge c_f$ and $p^{-1}\Tr(\bOme^2)\rightarrow\vartheta_\Omega\in(0,\infty)$ as $p\rightarrow\infty$.
\end{assumption}

For $q\in\{2,3,4\}$, let $\mathfrak P_q=\{\{1,2\},\ldots,\{2q-1,2q\}\}$. %
A partition $\pi$ of $\{1,\ldots,2q\}$ is called connected relative to
$\mathfrak P_q$ if, for every nonempty
$\mathcal A\subsetneq\mathfrak P_q$, the set
$\bigcup_{A\in\mathcal A}A$ cannot be expressed as a union of blocks
of $\pi$. Let $j(a)=j_{\lceil a/2\rceil}$.  For a block $B\in\pi$, let $a_B=\min B$, $B^*=B\setminus\{a_B\}$, and define 
\[
\begin{aligned}
\mathfrak K_B^{\mathrm{cum}}(j_1,\ldots,j_q)
={}&\sum_{(r_a: a\in B^*)
      \in\mathbb Z^{|B|-1}}
\left\{1+\max_{a\in B^*}|r_a|\right\}^{8}
\left|
\Cum\left(
\epsilon_{0,j(a_B)},
(\epsilon_{r_a,j(a)})_{a\in B^*}
\right)
\right|,
\end{aligned}
\]
where $\max_{a\in B^*}|r_a|=0$ when $B^*=\varnothing$.  Define
\[
\mathfrak C_{\pi,p}
=
\sum_{j_1,\ldots,j_q=1}^p
\prod_{B\in\pi}\mathfrak K_B^{\mathrm{cum}}(j_1,\ldots,j_q).
\]

\begin{assumption}%
\label{ass:C3}
For each $q\in\{2,3,4\}$, {there exists $C_q<\infty$} such that
\[
\max_{\pi:\,\pi\text{ connected relative to }\mathfrak P_q}
\mathfrak C_{\pi,p}\le C_qp.
\]
\end{assumption}

Assumption~\ref{ass:C1} covers the $\alpha$-mixing, $m$-dependence, and physical-dependence frameworks commonly used in high-dimensional time-series analysis. %
By Lemma~S.28 in the Supplementary Material,
\[
\|\bGam(h)\|_{\mathrm{op}}\le\varpi_{\mathrm{dep}}(|h|),\qquad
\varpi_{\mathrm{dep}}(r)=
\begin{cases}
C\exp(-cr^{\varkappa}),&\text{under (AM) or (PD)},\\
C\mathbb I(r\le m_0),&\text{under (MD)}.
\end{cases}
\]
Consequently,
\begin{equation}
\eta_m:=\sum_{r>m}(1+r)^8\varpi_{\mathrm{dep}}(r)
\begin{cases}
\le C\exp(-cm^{\varkappa}/2),&\text{under (AM) or (PD)},\\
=0\quad(m\ge m_0),&\text{under (MD)}.
\end{cases}
\label{eq:eta-m}
\end{equation}
Assumption~\ref{ass:C2} implies $\Tr(\bOme^2)\asymp p$ and $\Tr(\bOme^4)/\{\Tr(\bOme^2)\}^2=O(p^{-1})$.
Assumption~\ref{ass:C3} bounds the cumulant sums used to control the covariance and the third and fourth cumulants of the quadratic CUSUM process. These quantities involve four, six, and eight error variables, respectively. In the definition of $\mathfrak K_B^{\mathrm{cum}}$, one time index in each block is set to zero, and the remaining indices are measured relative to it. Thus, the sums depend only on time lags. For temporally independent observations, only the zero-lag terms remain, and the condition bounds sums of cumulants across coordinates at the same time point. In particular, under this assumption, the contribution of fourth-order error cumulants to the covariance of the normalized quadratic CUSUM process is $O(n^{-1})$.

Under $H_0$, {the mean of $W(k)$} is
\[
\mu_k%
=\frac1{n\sqrt p}
\sum_{h=0}^{n-1}\sum_{i=1}^{n-h}
\{2-\ind{h=0}\}c_{i,k}c_{i+h,k}\Tr\{\bGam(h)\},
\]
where $c_{i,k}=\mathbb I(i\le k)-k/n$. We approximate $\mu_k$ by its lag-$M$ truncation
\[
\mu_{k,M}=\frac1{n\sqrt p}
\sum_{h=0}^{M}\sum_{i=1}^{n-h}
\{2-\ind{h=0}\}c_{i,k}c_{i+h,k}\Tr\{\bGam(h)\},
\]
where $M=\lceil(n\wedge p)^{1/8}\rceil$.
Under Assumption~\ref{ass:C1}, replacing $\mu_k$ by
$\mu_{k,M}$ incurs an asymptotically negligible truncation error; see the proof of Theorem~\ref{Th1}.

\begin{theorem}\label{Th1}
Suppose Assumptions~\ref{ass:C1}--\ref{ass:C3}, $M=\lceil(n\wedge p)^{1/8}\rceil$ and $H_0$ hold. Under Assumption~\ref{ass:C1}(AM) or (PD), assume additionally that  $\log p=o(n^{\varkappa/8})$. Then, with the endpoint convention \(W(0)=W(n)=0\) and
\(\mu_{0,M}=\mu_{n,M}=0\),
\[
\frac{W(\lfloor nt\rfloor)-\mu_{\lfloor nt\rfloor,M}}{\omega_p}
\ \Rightarrow\ V(t)
\quad {\text{in } D[0,1]},
\]
as $(n,p)\to\infty$, where $\omega_p=\{2\Tr(\bOme^2)/p\}^{1/2}$ and $D[0,1]$ is the space of cadlag functions. The process $V(t)$ is a centered Gaussian process with 
\[
\E\{V(s)V(t)\}=s^2(1-t)^2,\qquad 0\le s\le t\le1.
\]
\end{theorem}

\begin{remark}
  The Gaussian process $V$ in Theorem 1 admits the Ornstein--Uhlenbeck representation
  \[
  \{V(t):0<t<1\}
  \stackrel{d}{=}
  \left\{
  t(1-t)Z_{\mathrm{OU}}
  \left[\log\{t/(1-t)\}\right]:0<t<1
  \right\},
  \]
  where $Z_{\mathrm{OU}}$ is a centered stationary Ornstein--Uhlenbeck process with covariance
  \[
\E\{Z_{\mathrm{OU}}(u)Z_{\mathrm{OU}}(v)\}
=e^{-|u-v|},\qquad u,v\in\mathbb R.
\]
  The OU representation is derived in the proof of Lemma~S.11
  in the Supplementary Material and is used
  in the proof of Theorem~\ref{Tms-1}(ii).
\end{remark}

We consider both unweighted and variance-standardized quadratic CUSUM processes.
For $\gamma\in\{0,1/2\}$, define
\[
W_\gamma(k)=v_k^{-2\gamma}W(k),
\qquad k\in\mathcal K_\gamma,
\]
where $v_k=t_k(1-t_k)$ with $t_k=k/n$, and
\begin{equation}
\mathcal K_0=\{1,\ldots,n-1\},
\qquad
\mathcal K_{1/2}=\{\lambda_n,\ldots,n-\lambda_n\},
\label{eq:candidate-sets}
\end{equation}
and $\lambda_n\in\{1,\ldots,\lfloor n/2\rfloor\}$ is a pre-specified boundary trimming parameter.
Theorem~\ref{Th1} gives $\Var\{V(t)\}=\{t(1-t)\}^2$.
Thus $W_0(k)=W(k)$ is the unweighted quadratic CUSUM statistic
evaluated over the full set of candidate locations.
In contrast, $W_{1/2}(k)=W(k)/v_k$ is the variance-standardized
quadratic CUSUM statistic evaluated on the trimmed set
$\mathcal K_{1/2}$.
The corresponding null mean and its lag-$M$ truncation are
\[
\mu_{\gamma,k}=v_k^{-2\gamma}\mu_k,
\qquad
\mu_{\gamma,k,M}=v_k^{-2\gamma}\mu_{k,M}.
\]

To construct feasible max-$L_2$ tests, we estimate $\mu_{\gamma,k,M}$ uniformly over $k\in\mathcal K_\gamma$ for each $\gamma\in\{0,1/2\}$, together with the scale $\omega_p$. We first construct the unweighted centering estimator. For $0\le h,\ell\le M$, define
$\mathbf Y_{i,h}^{(M)}=\bX_i-\bX_{i+M+h+1}$ and 
$\mathcal P_{i_1,h;i_2,\ell}^{(M)}=\mathbf Y_{i_1,h}^{(M)\top}\mathbf Y_{i_2,\ell}^{(M)}$. Motivated  by the moving-range construction of \citet{10.1142/s201032631950014x}, we set $\widehat g_h=\{2(n-M-2h-1)\}^{-1}\sum_{i=1}^{n-M-2h-1}\mathcal P_{i+h,h;i,h}^{(M)}$, and estimate $\mu_{k,M}$ by
\[
\widehat\mu_{k,M}=\frac{1}{n\sqrt p}\sum_{h=0}^{M}\sum_{i=1}^{n-h}\{2-\ind{h=0}\}c_{i,k}c_{i+h,k} \widehat g_h,\qquad k=1,\ldots,n-1.
\]
For $\gamma\in\{0,1/2\}$ and $k\in\mathcal K_\gamma$, we then estimate $\mu_{\gamma,k,M}$ by
\[
\widehat\mu_{\gamma,k,M}=v_k^{-2\gamma}\widehat\mu_{k,M}.
\]

 For $0\le h,\ell\le M$, let $\mathcal I_{h\ell}=\{(i_1,i_2): 1\le i_1\le\lfloor n/2\rfloor-M-2h-1, i_1+\lfloor n/2\rfloor\le i_2\le n-M-2\ell-1\}$, $\mathcal T_{h\ell}^{\mathrm F}=\Tr\{\bGam(h)\bGam(\ell)^\top\}$ and $\mathcal T_{h\ell}^{\mathrm C}=\Tr\{\bGam(h)\bGam(\ell)\}$. {Motivated by \citet{zhang2018clt}}, we consider the following difference-based estimators
\begin{align*}
\widehat{\mathcal T}_{h\ell}^{\mathrm F}
&=\frac1{4|\mathcal I_{h\ell}|}
\sum_{(i_1,i_2)\in\mathcal I_{h\ell}}
\mathcal P_{i_1,h;i_2,\ell}^{(M)}\,
\mathcal P_{i_1+h,h;i_2+\ell,\ell}^{(M)},~~
\widehat{\mathcal T}_{h\ell}^{\mathrm C}
&=\frac1{4|\mathcal I_{h\ell}|}
\sum_{(i_1,i_2)\in\mathcal I_{h\ell}}
\mathcal P_{i_1,h;i_2+\ell,\ell}^{(M)}\,
\mathcal P_{i_1+h,h;i_2,\ell}^{(M)}.
\end{align*}
The long-run covariance matrix truncated at lag $M$ is given by $\bOme_M=\bGam(0)+\sum_{h=1}^{M}\{\bGam(h)+\bGam(h)^\top\}$. A direct expansion yields
$\Tr(\bOme_M^2)
={}\mathcal T_{00}^{\mathrm F}
+2\sum_{h=1}^M(\mathcal T_{h0}^{\mathrm F}+\mathcal T_{0h}^{\mathrm F})+2\sum_{h,\ell=1}^M
(\mathcal T_{h\ell}^{\mathrm F}+\mathcal T_{h\ell}^{\mathrm C})$. {Replacing each trace term by its difference-based estimator gives the following estimator of $\omega_p$:}
\begin{align*}
\widehat\omega
=\bigg[\max\bigg\{n^{-1},\frac2p\bigg\{&\widehat{\mathcal T}_{00}^{\mathrm F}
+2\sum_{h=1}^M(\widehat{\mathcal T}_{h0}^{\mathrm F}+\widehat{\mathcal T}_{0h}^{\mathrm F})+2\sum_{h,\ell=1}^M(\widehat{\mathcal T}_{h\ell}^{\mathrm F}+\widehat{\mathcal T}_{h\ell}^{\mathrm C})\bigg\}\bigg\}\bigg]^{1/2}.
\end{align*}
Theorem~\ref{Th2} establishes the consistency rates of the centering estimators $\widehat\mu_{\gamma,k,M}$ and scale estimator $\widehat\omega$.

\begin{theorem}\label{Th2}
Under model~\eqref{eq:single change-model}, suppose
Assumptions~\ref{ass:C1}--\ref{ass:C3}, $M=\lceil(n\wedge p)^{1/8}\rceil$ and $\|\bdelta\|_2^2=o(n\sqrt p/M^2)$ hold.  
Under Assumption~\ref{ass:C1}(AM) or (PD), assume additionally that  
\(
\log p=o(n^{\varkappa/8}).
\)
Let
\begin{align*}
r_{\mu,n}={}&\frac{M}{\sqrt n}+\sqrt p\,M\eta_M+\frac{M^2\|\bdelta\|_2^2}{n\sqrt p},
\quad
r_{\omega,n}={}M^2\left(\frac1{np}+\frac1{n^2}\right)^{1/2}\!\!\!\!+M^2\eta_M+\frac{M^2\|\bdelta\|_2^2}{n\sqrt p}.%
\end{align*}
where $\eta_M$ is defined in \eqref{eq:eta-m}.
Then, for each $\gamma\in\{0,1/2\}$,
\[
\max_{k\in\mathcal K_\gamma}
|\widehat\mu_{\gamma,k,M}-\mu_{\gamma,k,M}|=O_p(r_{\mu,n}),
\qquad
|\widehat\omega^2-\omega_p^2|=O_p(r_{\omega,n}),
\]
where $r_{\mu,n}=o(1)$ and $r_{\omega,n}=o(1)$.

\end{theorem}

Using the estimators above, we now construct the max-$L_2$
test statistics. For $\gamma\in\{0,1/2\}$, define the
centered quadratic CUSUM statistic
\[
Q_{\gamma,k}=W_\gamma(k)-\widehat\mu_{\gamma,k,M}.
\] 

{For $\gamma=0$, the process
$\{Q_{0,k}/\widehat\omega:k\in\mathcal K_0\}$ has a Gaussian limit under
$H_0$, but its finite-sample marginal distributions may be right-skewed.
We use a quadratic correction motivated by the Cornish--Fisher expansion
\citep{cornish1938moments,hall1992bootstrap} to improve finite-sample calibration.
Let $\widehat\mu_{k,M}^{+}=\max\{\widehat\mu_{k,M},v_k/n\}$ so that
the correction denominator is positive. Under a strong positive signal,
an uncapped term proportional to $Q_{0,k}^2$ can exceed the linear term
$Q_{0,k}/\widehat\omega$. We therefore use
$\Delta_k=\min\{Q_{0,k},(3/2)\widehat\mu_{k,M}\}$ in the quadratic
correction, and define} %
\begin{equation*}
    S_0=\max_{k\in\mathcal K_{0}}\left\{\frac{Q_{0,k}}{\widehat\omega}-\frac{\Delta_k^2-\widehat\omega^2v_k^2}{3\widehat\mu_{k,M}^{+}\widehat\omega}\right\}.
\end{equation*}

For the max-$L_2$ statistic with $\gamma=1/2$, %
the Darling--Erd\H{o}s normalization already
provides an asymptotically valid calibration using the
standard Gumbel distribution under the conditions of
Theorem~\ref{Tms-1}(ii); see Section~\ref{sec:simulation-compatibility}
for its finite-sample performance in simulations.
{We therefore define}
\[
S_{1/2}
=\max_{k\in\mathcal K_{1/2}}
\frac{Q_{1/2,k}}{\widehat\omega}.
\]

\begin{theorem}\label{Tms-1}
Suppose the conditions of Theorem~\ref{Th2} hold under $H_0$. Then,
\begin{enumerate}
\item[(i)] As $(n,p)\rightarrow\infty$,
\[
S_{0}
\xrightarrow{d}\sup_{0\le t\le1}V(t),
\]
where $V(t)$ is the Gaussian process in Theorem~\ref{Th1};
\item[(ii)] If $\lambda_n\asymp n^\lambda$ for some $\lambda\in(0,1)$ and $\log\log n=o(p^{1/7})$, %
then, as $(n,p)\rightarrow\infty$,
\[
A_{\mathrm{DE}}(\log h_n)S_{1/2}
-D_{\mathrm{DE}}(\log h_n)
\xrightarrow{d}Z_{\mathrm{Gu}},
\]where $F_{\mathrm{Gu}}(x)=\Pr(Z_{\mathrm{Gu}}\le x)=\exp\{-\exp(-x)\}%
$, $A_{\mathrm{DE}}(x)=\sqrt{2\log x}$, $
D_{\mathrm{DE}}(x)=2\log x+2^{-1}\log\log x-2^{-1}\log\pi$, and $h_n=\{(\lambda_n/n)^{-1}-1\}^2$.
\end{enumerate}
\end{theorem}

\begin{remark}
Let
   \(
   S_0^{(1)}=\max_{k\in\mathcal K_0}Q_{0,k}/\widehat\omega.
   \)
   The proof of Theorem~\ref{Tms-1} establishes
   \[
|S_0-S_0^{(1)}|
\le
\max_{k\in\mathcal K_0}
\left|
\frac{\Delta_k^2-\widehat\omega^2v_k^2}
{3\widehat\mu_{k,M}^{+}\widehat\omega}
\right|
=O_p(p^{-1/2}).
\]
{Thus the quadratic correction does not change the limiting null distribution.}
\end{remark}

The corresponding $p$-values are
\begin{align*}
p_{S_{0}}&=1-F_V\left(S_{0}\right),~~~~p_{S_{1/2}}=1-F_{\mathrm{Gu}}\left\{A_{\mathrm{DE}}(\log h_n)S_{1/2}-D_{\mathrm{DE}}(\log h_n)\right\},
\end{align*}
where $F_V(\cdot)$ is the cumulative distribution function (cdf) of $\sup_{0\leq t\leq 1}V(t)$.

\begin{remark}\label{rem-1}

{We approximate the cdf $F_V$ by Monte Carlo simulation on a grid $0=t_0<t_1<\cdots<t_{T_d}=1$. For each replication $b=1,\ldots,B_V$, generate $\{V_b(t_0),\ldots,V_b(t_{T_d})\}^{\top}\sim N(\mathbf0,\mathbf R_{T_d+1})$, where $(\mathbf R_{T_d+1})_{k\ell}=(t_k\wedge t_\ell)^2\{1-(t_k\vee t_\ell)\}^2$, and let $v_b=\max_kV_b(t_k)$. The empirical distribution of $v_1,\ldots,v_{B_V}$ estimates the grid-maximum distribution and approximates $F_V$ as the grid is refined.} %
\end{remark}

The next proposition gives the consistency of the max-$L_2$ tests under $H_1$. %
\begin{prop}\label{prop:Tms-alter}
Suppose the conditions of Theorem~\ref{Th2} hold under $H_1$ and
$\tau/n\in[\vartheta,1-\vartheta]$ for a fixed
$\vartheta\in(0,1/2)$. Let \(
\Delta_{S,n}=n\|\bdelta\|_2^2/\sqrt p.
\)
\begin{enumerate}
\item[(i)] If $\Delta_{S,n}\to\infty$, then the test based on $S_0$ is consistent.
\item[(ii)] If $\lambda_n\asymp n^\lambda$ for some $\lambda\in(0,1)$, %
 $\Delta_{S,n}/\sqrt{\log\log h_n}\to\infty$, %
 then the test based on $S_{1/2}$ is consistent.
\end{enumerate}
\end{prop}

\section{Adaptive tests}\label{sec:3-adaptive}

The max-$L_2$ tests in Section~\ref{sec:2-L2} aggregate information across all coordinates and are better suited for dense alternatives.  To detect sparse signals, we now introduce coordinatewise standardized CUSUM statistics.  For coordinate $j$, let $\sigma_j^2=(\bOme)_{jj}$ and let $\widehat\sigma_j$ be an estimator of $\sigma_j$.  For $\gamma\in\{0,1/2\}$, define
\begin{equation}\label{CUSUM1}
C_{\gamma,j}(k)
=v_k^{-\gamma}\frac{\sum_{i=1}^k X_{ij}-k/n\sum_{i=1}^n X_{ij}}{\sqrt n\,\widehat\sigma_j},
\qquad k\in\mathcal K_\gamma,\quad j=1,\ldots,p.
\end{equation}
The two max-$L_\infty$ statistics are
\[
M_{\gamma}=\max_{k\in\mathcal K_\gamma}\max_{1\le j\le p}|C_{\gamma,j}(k)|,~~~~\gamma\in\{0,1/2\}.
\]

To estimate the componentwise long-run standard deviations ${\sigma}_j$, $j=1,\ldots,p$, we use the third-order difference-based estimator of \citet{chan2022AoS}. Specifically, $\widetilde\sigma_j^2=\sum_{|r|<\ell_j}K_{\mathrm{LRV}}(r/\ell_j)\widehat c_{r,j}^{D}$ and $\widehat\sigma_j=\{\max(\widetilde\sigma_j^2,n^{-1})\}^{1/2}$, where $\widehat c_{r,j}^{D}=n^{-1}\sum_{i=6\ell_j+|r|+1}^nD_{i,j}^{(3)}D_{i-|r|,j}^{(3)}$, $D_{i,j}^{(3)}=\sum_{r=0}^3d_r^{(3)}X_{i-2r\ell_j,j}$, for $i=6\ell_j+1,\ldots,n$, $\mathbf d^{(3)}=(d_0^{(3)},d_1^{(3)},d_2^{(3)},d_3^{(3)})$ is the normalized third-order difference vector specified in \cite{chan2022AoS} and $K_{\mathrm{LRV}}$ is a symmetric bounded kernel supported on $[-1,1]$ with $K_{\mathrm{LRV}}(0)=1$.
 
\begin{assumption}\label{ass:lrv-M}
    (i) There exist constants $\widetilde q>1/2$ and $B_{\mathrm{LRV}}\ne0$ such that $(K_{\mathrm{LRV}}(x)-1)/|x|^{\widetilde q}\rightarrow B_{\mathrm{LRV}}$ as $x\rightarrow0$. (ii) $n^{c_{\ell^{-}}}\lesssim \ell_{\min}\leq\ell_{\max}\lesssim n^{c_{\ell^{+}}}$ for some fixed $0<c_{\ell^{-}}<c_{\ell^{+}}<1/2$, 
 where $\ell_{\max}=\max_{1\leq j\leq p}\ell_j$, $\ell_{\min}=\min_{1\leq j\leq p}\ell_j$ denote the largest and smallest bandwidths, respectively. (iii) $\|\bdelta\|_\infty=o(\sqrt{n}/\ell_{\max})$.
\end{assumption}
\begin{remark}\label{rem:lrv-polynomial-rate}
The uniform consistency of $\widehat{\sigma}_j$ is established in Lemma S.26 of the Supplementary Material; specifically, $\max_{1\leq j\leq p}\vert\widehat{\sigma}_j/\sigma_j-1\vert=O_p(a_{\sigma,n})$ under $H_0$, where $a_{\sigma,n}=n^{-1/2}\ell_{\max}\{\log(np)\}^{c_{\mathrm{LRV}}}+\ell_{\min}^{-\widetilde q}+\eta_{\ell_{\min}}+n^{-1}\ell_{\max}^2$ with $c_{\mathrm{LRV}}=3/2+1/\varkappa$ under (AM) and $c_{\mathrm{LRV}}=3/2$ under (MD) or (PD). %
Furthermore, let $c_\sigma^\star=(2\widetilde q-1)/\{2(1+2\widetilde q)\}$, $\ell_j\asymp n^{1/(1+2\widetilde q)}$ uniformly in $j$ and $\widetilde q>1/2$. If $\log p=O(n^{\nu})$ with $0<\nu<\{c_\sigma^\star-c\}/c_{\mathrm{LRV}}$, then $\max_{1\leq j\leq p}\vert\widehat{\sigma}_j/\sigma_j-1\vert=O_p(n^{-c})$ with $0<c<c_\sigma^{\star}$.
\end{remark}

The next result extends the Gaussian extreme-value limit of \citet{wang2023JRSSB} to all three temporal-dependence branches in Assumption~\ref{ass:C1}.
Define
\[
\nu_{\mathrm{dep}}=
\begin{cases}
(34+13/\varkappa)^{-1},&\text{under (AM)},\\
1/28,&\text{under (MD)},\\
(34+1/\varkappa)^{-1},&\text{under (PD)}.
\end{cases}
\]
Let $\nu_{c}=\min\{(1/2-c_{\ell^+})/(1+c_{\mathrm{LRV}}),c_{\ell^-}\widetilde q\}$.

\begin{theorem}\label{null:Max}
Suppose Assumptions~\ref{ass:C1}--\ref{ass:lrv-M} and $H_0$ hold. %
Then, 
\begin{enumerate}
\item[(i)] %
If $\log p=O(n^{\nu})$ with 
$0<\nu<\min\{\nu_{\mathrm{dep}},\nu_c\}$,
as $(n,p)\rightarrow\infty$,
\[
2M_{0}^2-\log(2p)\xrightarrow{d}Z_{\mathrm{Gu}},
\]
where $F_{\mathrm{Gu}}(x)=\Pr(Z_{\mathrm{Gu}}\le x)=\exp\{-\exp(-x)\}$.
\item[(ii)] If $\lambda_n\asymp n^\lambda$ for some $\lambda\in(0,1)$ and  $\log p=O(n^{\nu})$ with %
$0<\nu<\min\{\lambda\nu_{\mathrm{dep}},\nu_c\}$%
, as $(n,p)\rightarrow\infty$,
\[
A_{\mathrm{DE}}(p\log h_n)M_{1/2}-D_{\mathrm{DE}}(p\log h_n)\xrightarrow{d}Z_{\mathrm{Gu}},
\]
where $A_{\mathrm{DE}}(x)=\sqrt{2\log x}$, and
$D_{\mathrm{DE}}(x)=2\log x+2^{-1}\log\log x-2^{-1}\log\pi$.
\end{enumerate}
\end{theorem}

By Theorem~\ref{null:Max}, the $p$-values associated with $M_0$ and $M_{1/2}$ are
\begin{align*}
p_{M_{0}}&=1-F_{\mathrm{Gu}}\{2M_{0}^2-\log(2p)\},\\
p_{M_{1/2}}&=1-F_{\mathrm{Gu}}\{
A_{\mathrm{DE}}(p\log h_n)M_{1/2}-D_{\mathrm{DE}}(p\log h_n)\}.
\end{align*}
The next proposition gives the consistency of the max-$L_\infty$ tests under $H_1$.
\begin{prop}\label{prop:consistency-Linf}
  Suppose the conditions of Theorem~\ref{null:Max} hold under $H_1$ and $\tau/n\in[\vartheta,1-\vartheta]$ for a fixed $\vartheta\in(0,1/2)$.%
    \begin{enumerate}
\item[(i)] If $\|\bdelta\|_\infty=o\{\sqrt{n}/(\ell_{\max}\sqrt{\log p})\}$ and $\sqrt n\|\bdelta\|_\infty/\sqrt{\log p}\to\infty$, then the test based on $M_0$ is consistent.

\item[(ii)] If $\lambda_n\asymp n^\lambda$ for some $\lambda\in(0,1)$,
\[
\|\bdelta\|_\infty=o[\sqrt{n}/\{\ell_{\max}\sqrt{\log (p\log h_n)}\}]\quad\text{and}\quad \sqrt n\|\bdelta\|_\infty/\sqrt{\log(p\log h_n)}\to\infty,
\]
then the test based on $M_{1/2}$ is consistent.
\end{enumerate}
\end{prop}

\begin{remark}
    For finite-sample computation, the componentwise long-run variance routine follows the robust change point option of \citet[Section~6.1]{chan2022AoS}. Each coordinate is first pre-adjusted for a small number of large level disturbances and a piecewise-linear trend; the third-order difference with coefficients $\mathbf d^{(3)}\approx(0.1942,0.2809,0.3832,-0.8582)$ is then combined with the parabolic kernel $K_{\mathrm{LRV}}(x)=(1-x^2)_+$ and a data-driven bandwidth $\ell_j$. The following finite-sample formulas are defined for $1\le\ell_j<n/6$, as the bandwidth assumptions ensure eventually. Replacing the original-sample denominator by the available third-difference length multiplies the estimate by $n/(n-6\ell_j)=1+O(\ell_j/n)$; uniformly over coordinates, this discrepancy is absorbed by the existing $\ell_{\max}^2/n$ term in $a_{\sigma,n}$.
    To account for the finite-sample uncertainty of $\widehat\sigma_j^2$, define $s_K(\ell)=\sum_{\vert r\vert<\ell}K_{\mathrm{LRV}}^2(r/\ell)=15^{-1}16\ell-15^{-1}\ell^{-3}$ and $\nu_j=(n-6\ell_j)/\{(7/6)s_K(\ell_j)\}$, and let $R_j\sim\operatorname{Gamma}(\nu_j/2,\nu_j/2)$, with the second argument denoting the rate. The finite-sample calibrated $p$-values for $M_0$ and $M_{1/2}$ are, respectively,
    \[
    p_{M_{0}}^{\mathrm{FLRV}}=1-\exp\left[-2\sum_{j=1}^p\left(1+\frac{4M_{0}^2}{\nu_j}\right)^{-\nu_j/2}\right],
    \]
    and
    \[
    \begin{aligned}
    p_{M_{1/2}}^{\mathrm{FLRV}}
    =1-\exp\Bigg[&-\frac{\exp\{D_{\mathrm{DE}}(p\log h_n)\}}{p}
    \sum_{j=1}^p\psi_{\nu_j}\{A_{\mathrm{DE}}(p\log h_n)M_{1/2}\}\Bigg],\\
    \end{aligned}
    \]
    where $\psi_{\nu}(b)=\E\{\exp(-b\sqrt R)\}$ and $R\sim\operatorname{Gamma}(\nu/2,\nu/2)$.
\end{remark}

In practice, the sparsity level is unknown. The next theorem establishes asymptotic independence between max-$L_2$ and max-$L_{\infty}$ statistics.
\begin{theorem}\label{indnull}
Suppose the conditions of Theorems~\ref{Th2} and \ref{null:Max} hold under $H_0$, and $\log\log n=o(p^{1/7})$. Then, for every $x,y\in\mathbb R$,
\begin{enumerate}
    \item [(i)] If $\log p=O(n^{\nu})$ with $0<\nu<\min\{\nu_{\mathrm{dep}},\nu_c\}$, as $(n,p)\rightarrow \infty$,
    \begin{align*}
    &\Pr\left\{S_{0}\le x,2M_{0}^2-\log(2p)\le y\right\}
    \rightarrow F_V(x)F_{\mathrm{Gu}}(y).
    \end{align*}
    \item [(ii)] If $\lambda_n\asymp n^\lambda$ for some $\lambda\in(0,1)$ and  $\log p=O(n^{\nu})$ with $0<\nu<\min\{\lambda\nu_{\mathrm{dep}},\nu_c\}$, then, as $(n,p)\to\infty$,
    \begin{align*}
    &\Pr\left\{
    S_{1/2}\le \frac{x+D_{\mathrm{DE}}(\log h_n)}{A_{\mathrm{DE}}(\log h_n)}, M_{1/2}\le \frac{y+D_{\mathrm{DE}}(p\log h_n)}{A_{\mathrm{DE}}(p\log h_n)}\right\}
    \rightarrow F_{\mathrm{Gu}}(x)F_{\mathrm{Gu}}(y).
    \end{align*}
\end{enumerate}
\end{theorem}

For $\gamma\in\{0,1/2\}$, we define the adaptive statistics and $p$-values by Cauchy combination \citep{LiuXie2020Cauchy}, that is,
\begin{align*}
T_{CC,\gamma}&=1/2\tan\{\pi(1/2-p_{S_{\gamma}})\}+1/2\tan\{\pi(1/2-p_{M_{\gamma}})\},~~~~p_{CC,\gamma}=1-F_C(T_{CC,\gamma}),
\end{align*}
where $F_C$ is the standard Cauchy distribution function. If a component $p$-value equals one, its
Cauchy transform is defined by the limit $-\infty$, and the combined
$p$-value is set to one. The finite-valued component statistics give
strictly positive component $p$-values, so this convention entails no
sum of opposite infinities.
\begin{prop}\label{prop:cauchy-size}
Under the conditions of Theorem~\ref{indnull}(i), $T_{CC,0}\xrightarrow{d}C(0,1)$, whereas under the conditions of Theorem~\ref{indnull}(ii), $T_{CC,1/2}\xrightarrow{d}C(0,1)$.
\end{prop}

{We next establish asymptotic independence under the following local alternatives:} 
\begingroup
\begin{align*}
H_{1;n,p}:|\mathcal A|=o(p), \bdelta^\top\bOme\bdelta=o\{p/(n\log\log n)\} \text{ and }\sqrt n\,\|\bdelta\|_\infty
=O\left[\{\log(np)\}^{b_\delta}\right],
\end{align*}
\endgroup
where $\mathcal A=\{j:\delta_j\ne0\}$ and $b_\delta\ge0$ is fixed. Let $\mathcal A^c=\{1,\ldots,p\}\setminus \mathcal A$
\begingroup
 and $\chi_{\mathcal A,p}=\max_{j\in\mathcal A^c}\{\boldsymbol{\Omega}_{j,\mathcal A}\boldsymbol{\Omega}_{\mathcal A,\mathcal A}^{-1}\boldsymbol{\Omega}_{\mathcal A,j}/(\bOme)_{jj}\}$. Set $\chi_{\mathcal A,p}=0$ if $\mathcal A=\varnothing$.
\endgroup

\begin{theorem}\label{indalter}
\begingroup
Suppose Assumptions~\ref{ass:C1}--\ref{ass:lrv-M}, $M=\lceil(n\wedge p)^{1/8}\rceil$, $\log\log n=o(p^{1/7})$ and $\chi_{\mathcal A,p}\{\log(np)\}^2\to0$ %
\endgroup
 hold. Under Assumption~\ref{ass:C1}(AM) or (PD), assume additionally that  $\log p=o(n^{\varkappa/8})$. Under $H_{1;n,p}$,
\begin{itemize}
    \item [(i)]  If $\log p=O(n^{\nu})$ with $0<\nu<\min\{\nu_{\mathrm{dep}},\nu_c,(1-c_{\ell^+})/(b_\delta+1/2)\}$ and %
    $\Delta_{S,n}(p^{-1/4}+n^{-3/4})\rightarrow 0$, then, as $(n,p)\rightarrow\infty$,
\[
\Pr(Z_{S,0,n}\le x,Z_{M,0,n}\le y)
-\Pr(Z_{S,0,n}\le x)\Pr(Z_{M,0,n}\le y)\rightarrow0
\]
for every $x,y\in\mathbb R$, where $Z_{S,0,n}=S_0$ and
$Z_{M,0,n}=2M_0^2-\log(2p)$.

\begingroup
\item[(ii)] If $|\mathcal A|=o\{p/(\log\log n)^2\}$, $\lambda_n\asymp n^\lambda$ for some $\lambda\in(0,1)$,  $\log p=O(n^{\nu})$ with $0<\nu<\min\{\lambda\nu_{\mathrm{dep}},\nu_c,(1-c_{\ell^+})/(b_\delta+1/2)\}$ and %
\endgroup
$\sqrt{(\log\log h_n)/n}\Delta_{S,n}(p^{-1/4}+n^{-1/4})\rightarrow0$, then, as $(n,p)\rightarrow\infty$,
\[
\Pr(Z_{S,1/2,n}\le x,Z_{M,1/2,n}\le y)
-\Pr(Z_{S,1/2,n}\le x)\Pr(Z_{M,1/2,n}\le y)\rightarrow0
\]
for every $x,y\in\mathbb R$, where
$Z_{S,1/2,n}=A_{\mathrm{DE}}(\log h_n)S_{1/2}-D_{\mathrm{DE}}(\log h_n)$ and
$Z_{M,1/2,n}=A_{\mathrm{DE}}(p\log h_n)M_{1/2}-D_{\mathrm{DE}}(p\log h_n)$.
\end{itemize}
\end{theorem}

\begin{prop}\label{prop:cauchy-consistency}
\begingroup
Suppose Assumptions~\ref{ass:C1}--\ref{ass:lrv-M}, $M=\lceil(n\wedge p)^{1/8}\rceil$, $\|\bdelta\|_2^2=o(n\sqrt{p}/M^2)$, $\log\log n=o(p^{1/7})$ and $\tau/n\in[\vartheta,1-\vartheta]$ for a fixed $\vartheta\in(0,1/2)$ hold. Under Assumption~\ref{ass:C1}(AM) or (PD), assume additionally that $\log p=o(n^{\varkappa/8})$. Under $H_1$,
\endgroup
\begin{itemize}
\item[(i)] %
If $\log p=O(n^{\nu})$ with $0<\nu<\min\{\nu_{\mathrm{dep}},\nu_c\}$, $\|\bdelta\|_\infty=o\{\sqrt{n}/(\ell_{\max}\sqrt{\log p})\}$ and either $\sqrt n\|\bdelta\|_\infty/\sqrt{\log p}\to\infty$ or $\Delta_{S,n}\rightarrow\infty$, then the test based on $T_{CC,0}$ is consistent.

\item[(ii)] If $\lambda_n\asymp n^\lambda$ for some $\lambda\in(0,1)$, $\log p=O(n^{\nu})$ with $0<\nu<\min\{\lambda\nu_{\mathrm{dep}},\nu_c\}$, %
  $\|\bdelta\|_\infty=o[\sqrt{n}/\{\ell_{\max}\sqrt{\log (p\log h_n)}\}]$ %
 and either $\sqrt n\|\bdelta\|_\infty/\sqrt{\log(p\log h_n)}\to\infty$ or $\Delta_{S,n}/\sqrt{\log\log h_n}\to\infty$, then the test based on $T_{CC,1/2}$ is consistent.
\end{itemize}
\end{prop}

\section{Location estimation}\label{sec:location}
{We estimate a single change location by maximizing the centered quadratic CUSUM statistic or the squared coordinatewise CUSUM statistic. For $\gamma\in\{0,1/2\}$, define}
\begin{align*}
\widehat\tau_{S,\gamma}&=\argmax_{k\in\mathcal K_\gamma}Q_{\gamma,k},~~~~\widehat\tau_{M,\gamma}=\argmax_{k\in\mathcal K_\gamma}\max_{1\le j\le p}|C_{\gamma,j}(k)|^2.
\end{align*}

Recall $\Delta_{S,n}=n\|\bdelta\|_2^2/\sqrt p$ from Proposition~\ref{prop:Tms-alter}, and let $\Delta_{M,n}=\sqrt n\|\bdelta\|_\infty/\sqrt{\log(np)}$.
The following theorems establish the consistency of the change point estimators $\widehat{\tau}_{S,\gamma}$ and $\widehat{\tau}_{M,\gamma}$.

\begin{theorem}\label{consistency-S}
Under model~\eqref{eq:single change-model}, suppose the conditions of Theorem~\ref{Th2} and $\tau/n\in[\vartheta,1-\vartheta]$ for a fixed $\vartheta\in(0,1/2)$ hold. 

\begin{enumerate}
\item[(i)] If $\Delta_{S,n}\to\infty$, then, as $(n,p)\rightarrow\infty$, $\left\vert\widehat\tau_{S,0}-\tau\right\vert/n\xrightarrow{p}0$. More precisely,
\[
\frac{\left\vert\widehat\tau_{S,0}-\tau\right\vert}{n}=O_p\left\{\frac{\sqrt p}{n\|\bdelta\|_2^2}+\frac1{\sqrt{n}\|\bdelta\|_2}\right\}.
\]

\begingroup
\item[(ii)]  If $\lambda_n\asymp n^\lambda$ for some $\lambda\in(0,1)$, $\log\log n=o(p^{1/7})$, and $\Delta_{S,n}/\sqrt{\log\log h_n}\to\infty$, under (AM) or (PD) assume additionally that $\log p=O(n^\nu)$ for some $0<\nu<\lambda\varkappa$; then, as $(n,p)\rightarrow\infty$, $\left\vert\widehat\tau_{S,1/2}-\tau\right\vert/n\xrightarrow{p}0$. More precisely,
\endgroup
\[
\frac{\left\vert\widehat\tau_{S,1/2}-\tau\right\vert}{n}=O_p\left[\sqrt{\log\log h_n}\left\{\frac{\sqrt p}{n\|\bdelta\|_2^2}+\frac1{\sqrt{n}\|\bdelta\|_2}\right\}\right].
\]
\end{enumerate}
\end{theorem}

\begin{theorem}\label{consistency-M}
Under model~\eqref{eq:single change-model}, suppose Assumptions~\ref{ass:C1}--\ref{ass:lrv-M} and $\tau/n\in[\vartheta,1-\vartheta]$ for a fixed $\vartheta\in(0,1/2)$ hold.
\begin{enumerate}
\item[(i)] If  $\log p=O(n^{\nu})$ with $0<\nu<\min\{\nu_{\mathrm{dep}},\nu_c\}$ and $\Delta_{M,n}\to\infty$, then, as $(n,p)\rightarrow\infty$, $\left\vert\widehat\tau_{M,0}-\tau\right\vert/n\xrightarrow{p}0$. More precisely,
\[
\frac{\left\vert\widehat\tau_{M,0}-\tau\right\vert}{n}=O_p\left\{\frac{\sqrt{\log(np)}}{\sqrt n\|\bdelta\|_\infty}+\frac{\log(np)}{n\|\bdelta\|_\infty^2}\right\}.
\]

\item[(ii)] If $\lambda_n\asymp n^\lambda$ for some $\lambda\in(0,1)$, $\log p=O(n^{\nu})$ with $0<\nu<\min\{\lambda\nu_{\mathrm{dep}},\nu_c\}$ and $\Delta_{M,n}\to\infty$, then, as $(n,p)\rightarrow\infty$, $\left\vert\widehat\tau_{M,1/2}-\tau\right\vert/n\xrightarrow{p}0$. More precisely,
\[
\frac{\left\vert\widehat\tau_{M,1/2}-\tau\right\vert}{n}=O_p\left\{\frac{\sqrt{\log(np)}}{\sqrt n\|\bdelta\|_\infty}+\frac{\log(np)}{n\|\bdelta\|_\infty^2}\right\}.
\]

\end{enumerate}
\end{theorem}

{Following \citet{wang2023JRSSB}, we select the location estimator associated with the smaller component $p$-value:}
\[
\widehat\tau_\gamma=
\begin{cases}
\widehat\tau_{S,\gamma},&p_{S_{\gamma}}<p_{M_{\gamma}},\\
\widehat\tau_{M,\gamma},&p_{S_{\gamma}}\ge p_{M_{\gamma}},
\end{cases}
\]
where $\gamma=0,1/2$. Let
\begingroup
\[
\begin{aligned}
R_{S,0,n}&=\frac{\sqrt p}{n\|\bdelta\|_2^2}+\frac1{\sqrt n\|\bdelta\|_2},
&R_{S,1/2,n}&=\sqrt{\log\log h_n}R_{S,0,n},\\
R_{M,n}&=\frac{\sqrt{\log(np)}}{\sqrt n\|\bdelta\|_\infty}
+\frac{\log(np)}{n\|\bdelta\|_\infty^2}.
\end{aligned}
\]
\endgroup
and $\Delta_{S,\gamma,n}=\Delta_{S,n}/\sqrt{\log(np)}\ind{\gamma=0}+\sqrt{\log\log h_n}\Delta_{S,n}/\log(np)\mathbb{I}(\gamma=1/2)$.
\begin{coro}\label{cor:adaptive-location}
Under model~\eqref{eq:single change-model}, suppose the conditions of Theorems~\ref{consistency-S}--\ref{consistency-M} hold. For each $\gamma\in\{0,1/2\}$, if $\Delta_{S,\gamma,n}\rightarrow\infty$ or $\Delta_{M,n}\rightarrow\infty$, then, $\vert \widehat\tau_{\gamma}-\tau\vert /n\xrightarrow{p}0$. More precisely,

\[
\frac{\vert \widehat\tau_\gamma-\tau\vert }{n}
=
\begin{cases}
O_p(R_{S,\gamma,n}), & \Delta_{S,n}/\log(np)\to\infty,\ \Delta_{M,n}=O(1),\\[3pt]
O_p(R_{M,n}), & \Delta_{M,n}\to\infty,\Delta_{S,n}/\log(np)=O(1),\\[3pt]
O_p(R_{S,\gamma,n}\vee R_{M,n}), & \Delta_{S,n}/\log(np)\to\infty,\ \Delta_{M,n}\to\infty.
\end{cases}
\]
\end{coro}

\section{Multiple change point estimation by wild binary segmentation}\label{sec:wbs}

{We extend the max-$L_2$, max-$L_\infty$, and adaptive procedures
to multiple changes using the random-interval construction of
wild binary segmentation (WBS) \citep{Fryzlewicz2014WBS}. At each recursion, we select a shortest significant interval, following the narrowest-over-threshold approach \citep{BaranowskiChenFryzlewicz2019}.}

We now consider a piecewise-constant mean change model with multiple change points,
\begin{equation}\label{eq:wbs-multiple-model}
\bX_i=\boldsymbol{\mu}_0+\sum_{j=1}^{K}\bdelta_j\mathbb{I}(i> \tau_j)+\bepsilon_i,~~~~~~~~i=1,\ldots,n,
\end{equation}
where $\{\tau_1,\ldots,\tau_{K}\}$ are true change points satisfying $0=\tau_0<\tau_1<\cdots<\tau_{K}<\tau_{K+1}=n$. 

For an interval $I=(s,e]$ {of length $m_I=e-s$, a candidate change point $b\in\{s+1,\ldots,e-1\}$, and $\gamma\in\{0,1/2\}$, define}
\begingroup
\begin{align*}
W_{\gamma,I}(b)&=\frac{v_{I,b}^{-2\gamma}}{m_I\sqrt p}
\left\|\sum_{i=s+1}^{b}\bX_i-\frac{m_{I,-}(b)}{m_I}\sum_{i=s+1}^{e}\bX_i\right\|_2^2,\\
C_{\gamma,j,I}(b)&=v_{I,b}^{-\gamma}\dfrac{\sum_{i=s+1}^{b} X_{ij}-t_{I,b}\sum_{i=s+1}^{e} X_{ij}}{\sqrt m_I\,\widehat\sigma_j},
\end{align*}
\endgroup
where $m_{I,-}(b)=b-s$, $m_{I,+}(b)=e-b$, $t_{I,b}=m_{I,-}(b)/m_I$ and $v_{I,b}=t_{I,b}(1-t_{I,b})$. 

For $i\in\{s+1,\ldots,e\}$, let
\(
a_{i,b}^{I}=m_{I,-}(b)^{-1}\ind{i\le b}-m_{I,+}(b)^{-1}\ind{i>b}\). %
Replacing the full-sample CUSUM coefficients in Section~\ref{sec:2-L2} by $a_{i,b}^{I}$ gives
\begin{align}\label{eq:wbs-local-centering}
\widehat\mu_{\gamma,M,I}(b)={}&v_{I,b}^{-2\gamma}\frac{m_{I,-}(b)^2m_{I,+}(b)^2}{m_I^3\sqrt p}\sum_{h=0}^{M}\sum_{i=s+1}^{e-h}\{2-\ind{h=0}\}a_{i,b}^{I}a_{i+h,b}^{I}\widehat g_h,
\end{align}
 and $\mu_{\gamma,M,I}(b)$ is defined analogously.
Let the candidate sets be $\mathcal K_0(I)=\{s+1,\ldots,e-1\}$ and $\mathcal K_{1/2}(I)=\{s+\lambda_{m_I}^{\mathrm W},\ldots,e-\lambda_{m_I}^{\mathrm W}\}$ with a fixed WBS-specific trimming parameter $\lambda_{\mathrm W}\in(0,1/4)$ and $\lambda_m^{\mathrm W}=\lceil m^{\lambda_{\mathrm W}}\rceil$.
Let
\begin{align*}
\mathcal S_\gamma(I)
&=\max_{b\in\mathcal K_\gamma(I)}\{W_{\gamma,I}(b)-\widehat\mu_{\gamma,M,I}(b)\},~~~~\mathcal M_\gamma(I)=\max_{b\in\mathcal K_\gamma(I)}\max_{1\le j\le p}|C_{\gamma,j,I}(b)|,
\end{align*}
be the two local test statistics. The corresponding change point location estimators on interval $I$ are
\begin{align*}
\widehat \tau_{S,\gamma}(I)&=\argmax_{b\in\mathcal K_\gamma(I)}\{W_{\gamma,I}(b)-\widehat\mu_{\gamma,M,I}(b)\},~~~~\widehat \tau_{M,\gamma}(I)=\argmax_{b\in\mathcal K_\gamma(I)}\max_{1\le j\le p}|C_{\gamma,j,I}(b)|.
\end{align*}

{The procedure has two stages. First, draw $N_{\mathrm W,n}$ endpoint pairs independently and uniformly with replacement from $\mathbb I_n^{\mathrm W}=\{(s,e):0\le s<e\le n,\ e-s\ge\ell_{\mathrm W,n}\}$, independently of the observations. These pairs define the intervals $I_r^{\mathrm W}=(S_r^{\mathrm W},E_r^{\mathrm W}]$, $r=1,\ldots,N_{\mathrm W,n}$, with minimum length $\ell_{\mathrm W,n}$.}
{Second, apply the following recursive procedure, excluding a neighborhood of each estimated change to prevent duplicate detections.} For a recursive segment $(s,e]$, we retain the intervals that are contained in $[s,e]$, that is, $\mathcal I_{\mathrm W}(s,e)=\{r:s\le S_r^{\mathrm W}<E_r^{\mathrm W}\le e\}$. Let $\bar\Lambda_{S,\gamma,n}^{\mathrm W}>0$ and $\bar\Lambda_{M,\gamma,n}^{\mathrm W}>0$ {be baseline threshold sequences for the two components.} For each $\gamma\in\{0,1/2\}$, let $c_{S,\gamma}^{\mathrm W}$ and $c_{M,\gamma}^{\mathrm W}$ be finite-sample calibration constants, such that $0<c_-\le c_{\mathfrak c,\gamma}^{\mathrm W}\le c_+<\infty$, $\mathfrak c\in\{\mathrm S,\mathrm M\}$, $\gamma\in\{0,1/2\}$.
The thresholds are defined as $\Lambda_{S,\gamma,n}^{\mathrm W}=c_{S,\gamma}^{\mathrm W}\bar\Lambda_{S,\gamma,n}^{\mathrm W}$, $\Lambda_{M,\gamma,n}^{\mathrm W}=c_{M,\gamma}^{\mathrm W}\bar\Lambda_{M,\gamma,n}^{\mathrm W}$. %

Define the threshold ratios
\[
\mathfrak R_{S,\gamma}^{\mathrm W}(I)=\frac{\mathcal S_\gamma(I)}{\widehat\omega_{\mathrm W}\Lambda_{S,\gamma,n}^{\mathrm W}},\qquad\mathfrak R_{M,\gamma}^{\mathrm W}(I)=\frac{\mathcal M_\gamma(I)}{\Lambda_{M,\gamma,n}^{\mathrm W}},
\]
with \(\widehat\omega_{\mathrm W}=\widehat\omega\vee n^{-1}.\)
For the current segment $(s,e]$, an interval $I$ is declared significant if at least one of its threshold ratios exceeds one. Accordingly, define
$$
\mathcal E_\gamma^{\mathrm W}(s,e)
=
\left\{
(\mathfrak c,r):
\mathfrak c\in\{\mathrm S,\mathrm M\},
\ r\in\mathcal I_{\mathrm W}(s,e),
\ \mathfrak R_{\mathfrak c,\gamma}^{\mathrm W}(I_r^{\mathrm W})>1
\right\}.
$$
If $\mathcal E_\gamma^{\mathrm W}(s,e)$ is empty, the recursion stops on $(s,e]$. Otherwise, we first select the shortest significant interval by setting
$m^\star=\min_{(\mathfrak c,r)\in\mathcal E_\gamma^{\mathrm W}(s,e)}|I_r^{\mathrm W}|$ and we choose
$(\mathfrak c^\star,r^\star)$ with the largest threshold ratio among all significant pairs satisfying $|I_r^{\mathrm W}|=m^\star$. The change point is estimated by $\widehat\tau=\widehat{\tau}_{\mathfrak c^\star,\gamma}(I_{r^\star}^{\mathrm W})$. The segment is then split at $\hat{\tau}$ and the same procedure is applied recursively to the guarded subsegments $(s,\widehat{\tau}-d_{\mathrm W,n}]$ and $(\widehat{\tau}+d_{\mathrm W,n},e]$. The procedure is summarized in Algorithm~\ref{alg:adaptive-guarded-wbs}.

\begin{algorithm}[t]
\caption{Adaptive guarded wild binary segmentation}
\label{alg:adaptive-guarded-wbs}
\DontPrintSemicolon
\SetKwFunction{WBS}{WBS}
\SetKwProg{Fn}{Function}{:}{}

\KwIn{Observations $\{\boldsymbol X_i\}_{i=1}^n$; weighting parameter $\gamma\in\{0,1/2\}$; number of random intervals $N_{\mathrm W,n}$; minimum interval length $\ell_{\mathrm W,n}$; thresholds $\Lambda_{S,\gamma,n}^{\mathrm W}$ and $\Lambda_{M,\gamma,n}^{\mathrm W}$; separation guard $d_{\mathrm W,n}$.}
\KwOut{Estimated change point set $\widehat{\mathcal T}_{\gamma}$.}

Draw $I_r^{\mathrm W}=(S_r^{\mathrm W},E_r^{\mathrm W}]$, $r=1,\ldots,N_{\mathrm W,n}$, independently and uniformly with replacement from $\mathbb I_n^{\mathrm W}=\{(s,e):0\leq s<e\leq n,\ e-s\geq\ell_{\mathrm W,n}\}$.
Set $\widehat{\mathcal T}_{\gamma}\leftarrow\varnothing$\;

\Fn{\WBS{$s,e$}}{
    \If{$e-s<\ell_{\mathrm W,n}$}{\Return}

    Set
    $\mathcal I_{\mathrm W}(s,e)
    \leftarrow
    \{r:s\leq S_r^{\mathrm W}<E_r^{\mathrm W}\leq e\}$\;
    \If{$\mathcal I_{\mathrm W}(s,e)=\varnothing$}{\Return}

    Set
    $\mathcal E_{\gamma}^{\mathrm W}(s,e)
    \leftarrow
    \{(\mathfrak c,r):\mathfrak c\in\{\mathrm S,\mathrm M\},\
    r\in\mathcal I_{\mathrm W}(s,e),\
    \mathfrak R_{\mathfrak c,\gamma}^{\mathrm W}(I_r^{\mathrm W})>1\}$\;
    \If{$\mathcal E_{\gamma}^{\mathrm W}(s,e)=\varnothing$}{\Return}

    Set
    $m^\star
    \leftarrow
    \min_{(\mathfrak c,r)\in\mathcal E_{\gamma}^{\mathrm W}(s,e)}
    |I_r^{\mathrm W}|$\;
    Set
    $\mathcal H_{\gamma}^{\mathrm W}(s,e)
    \leftarrow
    \{(\mathfrak c,r)\in\mathcal E_{\gamma}^{\mathrm W}(s,e):
    |I_r^{\mathrm W}|=m^\star\}$\;

    Choose $(\mathfrak c^\star,r^\star)$ from $\arg\max_{(\mathfrak c,r)\in\mathcal H_{\gamma}^{\mathrm W}(s,e)}\mathfrak R_{\mathfrak c,\gamma}^{\mathrm W}(I_r^{\mathrm W})$;
    Set
    $\widehat{\tau}
    \leftarrow
    \widehat{\tau}_{\mathfrak c^\star,\gamma}(I_{r^\star}^{\mathrm W})$\;
    Update
    $\widehat{\mathcal T}_{\gamma}\leftarrow \widehat{\mathcal T}_{\gamma}\cup\{\widehat{\tau}\}$\;

    \If{$\widehat{\tau}-d_{\mathrm W,n}>s$}{
        \WBS{$s,\widehat{\tau}-d_{\mathrm W,n}$}\;
    }
    \If{$\widehat{\tau}+d_{\mathrm W,n}<e$}{
        \WBS{$\widehat{\tau}+d_{\mathrm W,n},e$}\;
    }
}

\WBS{$0,n$}\;
\Return $\widehat{\mathcal T}_{\gamma}$\;
\end{algorithm}

\begin{remark}
\label{rem:component-specific-wbs}
We extend the max-$L_2$, max-$L_\infty$, and adaptive procedures
to multiple changes using the random-interval construction of
wild binary segmentation (WBS) \citep{Fryzlewicz2014WBS}.
More precisely, for a fixed $\mathfrak c\in\{S,M\}$, %
\[
\mathcal E^W_{\mathfrak c,\gamma}(s,e)=\left\{(\mathfrak c,r): r\in\mathcal I^W(s,e),\ \mathfrak R^W_{\mathfrak c,\gamma}(I_r^W)>1 \right\}.
\]
The shortest-significant-interval and largest-threshold-ratio selection steps are carried out over $\mathcal E^W_{\mathfrak c,\gamma}(s,e)$, and the selected change point location is $\widehat\tau=\widehat\tau_{\mathfrak c,\gamma}(I_{r^\star}^W)$. This location is added to the component-specific estimated change point set $\widehat{\mathcal T}_{\mathfrak c,\gamma}=\{\widehat\tau_{\mathfrak c,\gamma,1}<\cdots<\widehat\tau_{\mathfrak c,\gamma,\widehat K_{\mathfrak c,\gamma}}\}$ with $\widehat K_{\mathfrak c,\gamma}=\vert\widehat{\mathcal T}_{\mathfrak c,\gamma}\vert$, while all remaining guarded-recursion steps are the same as those in Algorithm~\ref{alg:adaptive-guarded-wbs}. 
\end{remark}

We next establish the consistency of the proposed procedures for estimating both the number and the locations of the change points.
Let  $\Delta_n=\min_{0\le j\le K}(\tau_{j+1}-\tau_j)$, $\kappa_{2,n}=\min_{1\le j\le K}\|\bdelta_j\|_2$ and $\kappa_{\infty,n}=\min_{1\le j\le K}\|\bdelta_j\|_\infty$.

\begin{assumption}%
\label{ass:WBS-geometry}
(i) The number of change points $K$ is fixed, and there is a constant $c_\Delta>0$ such that $\Delta_n\ge c_\Delta n$. (ii) The jumps satisfy $\max_{1\leq j\leq K}\|\bdelta_j\|_2^2=o(n\sqrt{p}/M^2)$.
\end{assumption}
\begin{remark}
    Assumption~\ref{ass:WBS-geometry}(i) specifies the number of change points to be fixed and the minimum spacing between consecutive change points to be at least a fixed proportion of $n$.  Together with this fixed-$K$ condition, Assumption~\ref{ass:WBS-geometry}(ii) controls signal contamination in the local centering and scale estimators $\widehat\mu_{\gamma,M,I}$ and $\widehat\omega$, and is used to establish their consistency under the multiple change model~\eqref{eq:wbs-multiple-model}. When $K=1$, it reduces to the signal upper bound in Theorem~\ref{Th2} for the single change setting.
\end{remark}

\begingroup
Let $c_{\mathrm{dep}}\ge1/2$ be a fixed constant implied by Assumption~\ref{ass:C1} such that $\left\|\sum_iw_i\mathbf u^\top\bepsilon_i\right\|_q\le Cq^{c_{\mathrm{dep}}}\|w\|_2$ for every $q\ge2$, every deterministic unit vector $\mathbf u$, and every finitely supported deterministic coefficient sequence $w$; see Lemma S.28 in Supplementary Material for details.
\endgroup
\begin{assumption}%
\label{ass:WBS-tuning}
(i) For a fixed $\kappa_{\mathrm W}\in(0,1)$, $n^{\kappa_{\mathrm W}}\leq\ell_{\mathrm W,n}=o(n)$. (ii) The number of sampled intervals satisfies $N_{\mathrm W,n}\rightarrow\infty$, and $\log N_{\mathrm W,n}=O(\log n)$. (iii) The separation guard satisfies $d_{\mathrm W,n}\rightarrow\infty$, and $d_{\mathrm W,n}=o(\Delta_n)$. (iv)  $n^{-1}\ell_{\max}^2\max_{1\leq j\leq K}\|\bdelta_j\|_{\infty}^2\rightarrow 0$.%

\end{assumption}

 The estimated numbers of change points by max-$L_2$, max-$L_{\infty}$ and adaptive methods are defined by $\widehat K_{S,\gamma}$, $\widehat K_{M,\gamma}$, and $\widehat K_{\gamma}$, respectively. Their ordered estimated locations are denoted by $\widehat{\tau}_{S,\gamma,j}$, $\widehat{\tau}_{M,\gamma,j}$, and $\widehat{\tau}_{\gamma,j}$. Let $B_{S,0,n}^{\mathrm W}=\{N_{\mathrm W,n}\log(2n)\}^{1/4}+p^{-1/2}\{\log (np)\}^{2c_{\mathrm{dep}}}$ and $B_{S,1/2,n}^{\mathrm W}=N_{\mathrm W,n}^{1/4}\{\log(2n)\}^{1/2}+p^{-1/2}\{\log (np)\}^{2c_{\mathrm{dep}}}$.

\begin{theorem}%
\label{thm:wbs-dense}
Suppose Assumptions~\ref{ass:C1}--\ref{ass:C3}, \ref{ass:WBS-geometry}, \ref{ass:WBS-tuning}(i)--(iii) and $M=\lceil(n\wedge p)^{1/8}\rceil$ hold. Under Assumption~\ref{ass:C1}(AM) or (PD), assume additionally that  
\(
\log p=o(n^{\varkappa/8}).
\)
  For $\gamma\in\{0,1/2\}$, suppose $B_{S,\gamma,n}^{\mathrm W}=o(\Lambda_{S,\gamma,n}^{\mathrm W})$, $\Lambda_{S,\gamma,n}^{\mathrm W}=o(p^{-1/2}\Delta_n\kappa_{2,n}^2)$, and
$$\mathbb{I}\{\gamma=1/2\}\Delta_n^{\lambda_{\mathrm W}}+\Delta_n\left\{\frac{B_{S,\gamma,n}^{\mathrm W}}{\Lambda_{S,\gamma,n}^{\mathrm W}}+\frac{\{\log(np)\}^{c_{\mathrm{dep}}}}{p^{1/4}\{\Lambda_{S,\gamma,n}^{\mathrm W}\}^{1/2}}\right\}=o(d_{\mathrm W,n}).$$
Then $\Pr(\widehat K_{S,\gamma}=K)\rightarrow1$, and
$$
\max_{1\leq j\leq K}\left|\widehat{\tau}_{S,\gamma,j}-\tau_j\right|=O_p\left[\mathbb{I}(\gamma=1/2)\Delta_n^{\lambda_{\mathrm W}}+\Delta_n\left\{\frac{B_{S,\gamma,n}^{\mathrm W}}{\Lambda_{S,\gamma,n}^{\mathrm W}}+\frac{\{\log(np)\}^{c_{\mathrm{dep}}}}{p^{1/4}\{\Lambda_{S,\gamma,n}^{\mathrm W}\}^{1/2}}\right\}\right].
$$
In particular,
\(
\frac{1}{\Delta_n}\max_{1\leq j\leq K}\left|\widehat{\tau}_{S,\gamma,j}-\tau_j\right|\xrightarrow{p}0.
\)
\end{theorem}

\begin{theorem}%
\label{thm:wbs-max}
\begingroup
Suppose Assumptions~\ref{ass:C1}--\ref{ass:C3}, \ref{ass:lrv-M}(i)--(ii), \ref{ass:WBS-geometry}(i), \ref{ass:WBS-tuning} hold. For $\gamma\in\{0,1/2\}$, suppose $\log p=O(n^{\nu})$ with $0<\nu<\min\{\lambda_{\gamma}\nu_{\mathrm{dep}},\nu_{c}\}$ where $\lambda_{\gamma}=\mathbb{I}(\gamma=0)+\lambda_{\mathrm W}\mathbb{I}(\gamma=1/2)$ %
\endgroup
, $\{\log(np)\}^{c_{\mathrm{dep}}}=o(\Lambda_{M,\gamma,n}^{\mathrm W})$, $\Lambda_{M,\gamma,n}^{\mathrm W}=o(\sqrt{\Delta_n}\,\kappa_{\infty,n})$, and
$$
\mathbb{I}\{\gamma=1/2\}\Delta_n^{\lambda_{\mathrm W}}+\Delta_n\frac{\{\log(np)\}^{c_{\mathrm{dep}}}}{\Lambda_{M,\gamma,n}^{\mathrm W}}=o(d_{\mathrm W,n}).
$$
Then $\Pr(\widehat K_{M,\gamma}=K)\rightarrow1$, and
$$
\max_{1\leq j\leq K}\left|\widehat{\tau}_{M,\gamma,j}-\tau_j\right|=O_p\left[\mathbb{I}(\gamma=1/2)\Delta_n^{\lambda_{\mathrm W}}+\Delta_n\frac{\{\log(np)\}^{c_{\mathrm{dep}}}}{\Lambda_{M,\gamma,n}^{\mathrm W}}\right].
$$
In particular,
\(
\frac{1}{\Delta_n}\max_{1\leq j\leq K}\left|\widehat{\tau}_{M,\gamma,j}-\tau_j\right|\xrightarrow{p}0.
\)
\end{theorem}

\begin{coro}%
\label{cor:wbs-adaptive}
Suppose Assumptions~\ref{ass:C1}--\ref{ass:C3}, \ref{ass:lrv-M}(i)--(ii), \ref{ass:WBS-geometry}--\ref{ass:WBS-tuning} and $M=\lceil (n\wedge p)^{1/8}\rceil$ hold. Under Assumption~\ref{ass:C1}(AM) or (PD), assume additionally that 
\(
\log p=o(n^{\varkappa/8}).
\)
\begingroup
For a fixed $\gamma\in\{0,1/2\}$, suppose $\log p=O(n^{\nu})$ with $0<\nu<\min\{\lambda_{\gamma}\nu_{\mathrm{dep}},\nu_c\}$, $B_{S,\gamma,n}^{\mathrm W}=o(\Lambda_{S,\gamma,n}^{\mathrm W})$, $\{\log(np)\}^{c_{\mathrm{dep}}}=o(\Lambda_{M,\gamma,n}^{\mathrm W})$ and 
\endgroup
$$
\mathbb{I} \{\gamma=1/2\}\Delta_n^{\lambda_{\mathrm W}}+\Delta_n\left\{\frac{B_{S,\gamma,n}^{\mathrm W}}{\Lambda_{S,\gamma,n}^{\mathrm W}}+\frac{\{\log(np)\}^{c_{\mathrm{dep}}}}{p^{1/4}\{\Lambda_{S,\gamma,n}^{\mathrm W}\}^{1/2}}+\frac{\{\log(np)\}^{c_{\mathrm{dep}}}}{\Lambda_{M,\gamma,n}^{\mathrm W}}\right\}=o(d_{\mathrm W,n}).
$$ For every $j=1,\ldots,K$, assume that at least one of
$$
\frac{
\Delta_n\|\boldsymbol\delta_j\|_2^2
}{
\sqrt p\,\Lambda_{S,\gamma,n}^{\mathrm W}
}
\rightarrow\infty,
\qquad
\frac{
\sqrt{\Delta_n}\|\boldsymbol\delta_j\|_\infty
}{
\Lambda_{M,\gamma,n}^{\mathrm W}
}
\rightarrow\infty
$$
holds. Then,
$\Pr(\widehat K_{\gamma}=K)\rightarrow 1$, and
\begingroup
$$
\begin{aligned}
\max_{1\le j\le K}|\widehat\tau_{\gamma,j}-\tau_j|
=O_p\Bigg[&\mathbb I\{\gamma=1/2\}\Delta_n^{\lambda_{\mathrm W}}\\
&+\Delta_n\left\{\frac{B_{S,\gamma,n}^{\mathrm W}}{\Lambda_{S,\gamma,n}^{\mathrm W}}
+\frac{\{\log(np)\}^{c_{\mathrm{dep}}}}{p^{1/4}\{\Lambda_{S,\gamma,n}^{\mathrm W}\}^{1/2}}
+\frac{\{\log(np)\}^{c_{\mathrm{dep}}}}{\Lambda_{M,\gamma,n}^{\mathrm W}}\right\}\Bigg].
\end{aligned}
$$
\endgroup
In particular,
\(
\frac{1}{\Delta_n}\max_{1\leq j\leq K}\left|\widehat{\tau}_{\gamma,j}-\tau_j\right|\xrightarrow{p}0.
\)
\end{coro}

\section{Simulation results}\label{sec:simulation-compatibility}

{We compare the six proposed procedures with LXZL \citep{li2019change} and WZVS \citep{wang2022inference}, which allow temporal dependence, and with JPYZ \citep{10.1007/s11425-016-0058-5} and WZWY \citep{10.1142/s201032631950014x}, which assume independent observations.}
For the WZVS test, we set the tuning parameter \(\eta_{\mathrm{WZVS}}\) following \cite{wang2022inference}. 
 We set $M=\lceil(n\wedge p)^{1/8}\rceil$ and $\lambda_n=\lfloor\sqrt n\rfloor$ for statistics with $\gamma=1/2$.

We generate data from
\begin{equation*}
\bX_i = \sum\limits_{h=0}^{q_{\mathrm{MA}}} \mathbf{A}_h\boldsymbol\zeta_{i-h},\qquad i=1,\ldots,n,
\end{equation*}
where $\mathbf{A}_0,\dots,\mathbf{A}_{q_{\mathrm{MA}}}$ are $p\times p$ matrices that determine the autocovariance structure, $\boldsymbol\zeta_i=\boldsymbol\Sigma_{\zeta}^{1/2}\bm u_i$, and $\bm u_i=(u_{i1},\ldots,u_{ip})^\top$ has independent coordinates.  {We consider standardized Gaussian errors}
\(
u_{ij}\sim N(0,1)
\)
{and standardized Gamma errors}
\(
u_{ij}\sim 1/2\operatorname{Gamma}(4,1)-2.
\)
For $0\le h\le q_{\mathrm{MA}}$, this data-generating model implies $\Cov(\bX_i,\bX_{i+h})=\bGam(h)$ with
\[
\bGam(h)=\sum_{k=0}^{q_{\mathrm{MA}}-h}\mathbf{A}_k\boldsymbol\Sigma_{\zeta}\mathbf{A}_{k+h}^{\T},\qquad\bGam(-h)=\bGam(h)^\top.
\]
The matrix $\mathbf A_h=(a_{h,ij})_{p\times p}$ is given by
\begin{equation*}
(\mathbf{A}_h)_{i,j}=\left\{
\begin{array}{cl}
\dfrac{\phi}{h}, & \mbox{if}~ |i-j|=0,\\
\dfrac{\phi}{h|i-j|^2}, & \mbox{if}~1 \leq |i-j|\leq \lfloor p\pi_{\mathrm A} \rfloor,\\
0,  & \mbox{if}~ |i-j|> \lfloor p\pi_{\mathrm A} \rfloor,
\end{array}
\right.
\end{equation*}
for $h=1,\dots,q_{\mathrm{MA}}$. Here, $\phi$ controls the temporal cross-coordinate dependence and $\pi_{\mathrm A}$ controls the relative bandwidth of $\mathbf A_h$. We set $\mathbf A_0=\mathbf I_p$ and consider $q_{\mathrm{MA}}\in\{0,2\}$. The two covariance scenarios are
\begin{enumerate}
\item[](S1) $\boldsymbol\Sigma_{\zeta}=\mathbf I_p$, $\phi=0.5$, and $\pi_{\mathrm A}=0.5$;
\item[](S2) $r_\zeta=0.5$, $\boldsymbol\Sigma_{\zeta}=(r_\zeta^{|i-j|})_{1\le i,j\le p}$, $\phi=0.2$, and $\pi_{\mathrm A}=0.2$.
\end{enumerate}

All experiments are based on 500 Monte Carlo replications. %
 In the empirical-size study, we report all six procedures, { $T_{S_{0}}$, $T_{M_{0}}$, $T_{C_{0}}$, $T_{S_{1/2}}$, $T_{M_{1/2}}$, $T_{C_{1/2}}$}, together with the four competing methods. For the Gaussian-process calibration of $S_{0}$,  we take $T_d=B_V=10000$ in Remark~\ref{rem-1}. %

\begin{table}[htbp]
\caption{Empirical sizes (in \%) with normal errors.}
\centering
\scriptsize
\setlength{\tabcolsep}{2.2pt}
\renewcommand\arraystretch{0.8}
\resizebox{\textwidth}{!}{%
\begin{tabular}{ccccccccccccc} \toprule
$n$ & $p$ & $q_{\mathrm{MA}}$ & $T_{S_{0}}$ & $T_{M_{0}}$ & $T_{C_0}$ & $T_{S_{1/2}}$ & $T_{M_{1/2}}$ & $T_{C_{1/2}}$ & LXZL & WZVS & JPYZ & WZWY\\ \midrule
\multicolumn{13}{c}{Scenario (S1)}\\
400 & 250 & 0 & 5.4 & 5.0 & 5.4 & 4.4 & 2.6 & 3.6 & 4.8 & 5.0 & 7.0 & 4.2 \\
400 & 500 & 0 & 4.4 & 4.4 & 4.8 & 5.8 & 3.6 & 4.8 & 3.8 & 5.6 & 5.4 & 3.6 \\
800 & 250 & 0 & 5.2 & 4.2 & 5.4 & 6.4 & 3.6 & 4.0 & 6.0 & 7.2 & 7.4 & 6.2 \\
800 & 500 & 0 & 3.4 & 5.2 & 4.8 & 2.6 & 3.0 & 2.8 & 4.0 & 5.6 & 4.2 & 3.6 \\
400 & 250 & 2 & 3.8 & 4.4 & 4.0 & 3.8 & 3.8 & 3.2 & 5.2 & 6.4 & 100.0 & 100.0 \\
400 & 500 & 2 & 4.0 & 3.8 & 5.2 & 4.6 & 2.6 & 3.6 & 6.6 & 6.4 & 100.0 & 100.0 \\
800 & 250 & 2 & 3.0 & 5.8 & 3.8 & 6.0 & 3.4 & 4.2 & 4.4 & 4.8 & 100.0 & 100.0 \\
800 & 500 & 2 & 4.8 & 6.4 & 5.4 & 5.6 & 4.8 & 5.8 & 7.2 & 5.6 & 100.0 & 100.0 \\
\midrule
\multicolumn{13}{c}{Scenario (S2)}\\
400 & 250 & 0 & 5.0 & 5.6 & 6.0 & 6.6 & 4.2 & 5.8 & 5.6 & 5.8 & 8.8 & 6.2 \\
400 & 500 & 0 & 5.8 & 4.0 & 5.4 & 6.2 & 3.0 & 4.0 & 5.0 & 6.2 & 6.8 & 4.4 \\
800 & 250 & 0 & 4.6 & 4.2 & 3.4 & 7.6 & 2.4 & 5.2 & 5.6 & 5.0 & 7.8 & 6.0 \\
800 & 500 & 0 & 5.8 & 5.8 & 6.4 & 7.6 & 3.6 & 4.8 & 6.2 & 4.8 & 6.8 & 6.0 \\
400 & 250 & 2 & 3.0 & 5.2 & 3.4 & 5.4 & 2.4 & 3.6 & 4.4 & 5.0 & 100.0 & 100.0 \\
400 & 500 & 2 & 3.4 & 4.8 & 4.4 & 4.6 & 2.4 & 2.8 & 5.4 & 5.4 & 100.0 & 100.0 \\
800 & 250 & 2 & 2.8 & 6.2 & 4.0 & 5.6 & 3.4 & 5.0 & 5.2 & 5.4 & 100.0 & 100.0 \\
800 & 500 & 2 & 3.8 & 5.2 & 5.8 & 5.8 & 4.6 & 5.4 & 5.4 & 5.4 & 100.0 & 100.0 \\ \bottomrule
\end{tabular}%
}
\label{tab1-rep500}
\end{table}

\begin{table}[htbp]
\caption{Empirical sizes (in \%) with standardized Gamma errors.}
\centering
\scriptsize
\setlength{\tabcolsep}{2.2pt}
\renewcommand\arraystretch{0.8}
\resizebox{\textwidth}{!}{%
\begin{tabular}{ccccccccccccc} \toprule
$n$ & $p$ & $q_{\mathrm{MA}}$ & $T_{S_{0}}$ & $T_{M_{0}}$ & $T_{C_0}$ & $T_{S_{1/2}}$ & $T_{M_{1/2}}$ & $T_{C_{1/2}}$ & LXZL & WZVS & JPYZ & WZWY\\ \midrule
\multicolumn{13}{c}{Scenario (S1)}\\
400 & 250 & 0 & 3.2 & 3.8 & 3.6 & 5.0 & 7.0 & 7.4 & 6.2 & 6.2 & 6.6 & 5.4 \\
400 & 500 & 0 & 6.2 & 2.8 & 4.8 & 5.8 & 4.2 & 4.4 & 6.2 & 4.4 & 6.4 & 4.8 \\
800 & 250 & 0 & 4.4 & 3.6 & 2.8 & 5.8 & 6.2 & 6.4 & 5.4 & 4.6 & 6.4 & 5.0 \\
800 & 500 & 0 & 5.0 & 5.0 & 4.6 & 3.8 & 4.8 & 2.8 & 3.2 & 5.6 & 6.0 & 3.4 \\
400 & 250 & 2 & 2.6 & 6.2 & 5.4 & 5.6 & 5.4 & 6.0 & 6.2 & 3.4 & 100.0 & 100.0 \\
400 & 500 & 2 & 4.0 & 5.2 & 5.0 & 5.4 & 5.4 & 5.0 & 4.4 & 5.0 & 100.0 & 100.0 \\
800 & 250 & 2 & 2.0 & 4.4 & 2.4 & 6.0 & 6.4 & 5.4 & 4.4 & 3.2 & 100.0 & 100.0 \\
800 & 500 & 2 & 4.8 & 5.8 & 5.2 & 5.6 & 6.0 & 7.2 & 6.0 & 5.0 & 100.0 & 100.0 \\
\midrule
\multicolumn{13}{c}{Scenario (S2)}\\
400 & 250 & 0 & 4.4 & 3.0 & 4.4 & 8.4 & 4.2 & 6.8 & 5.4 & 5.2 & 7.6 & 5.2 \\
400 & 500 & 0 & 6.0 & 3.8 & 5.2 & 7.4 & 4.6 & 6.6 & 6.6 & 7.2 & 7.6 & 6.0 \\
800 & 250 & 0 & 4.0 & 5.2 & 4.8 & 7.2 & 6.0 & 7.2 & 5.4 & 4.0 & 6.2 & 5.6 \\
800 & 500 & 0 & 4.2 & 4.0 & 5.2 & 5.6 & 6.2 & 5.8 & 3.6 & 6.2 & 6.2 & 3.4 \\
400 & 250 & 2 & 2.8 & 4.4 & 3.2 & 5.8 & 3.0 & 4.2 & 4.6 & 4.2 & 100.0 & 100.0 \\
400 & 500 & 2 & 2.8 & 5.6 & 5.6 & 3.8 & 5.0 & 3.8 & 5.0 & 5.2 & 100.0 & 100.0 \\
800 & 250 & 2 & 4.2 & 5.6 & 5.4 & 7.0 & 5.4 & 5.8 & 6.6 & 4.8 & 100.0 & 100.0 \\
800 & 500 & 2 & 3.4 & 6.8 & 5.6 & 4.4 & 4.8 & 5.6 & 4.0 & 4.8 & 100.0 & 100.0 \\ \bottomrule
\end{tabular}%
}
\label{tab2-rep500}
\end{table}
Tables~\ref{tab1-rep500}--\ref{tab2-rep500} report the empirical size for the normal and standardized Gamma errors, respectively. Under normal errors, the six proposed procedures have empirical sizes between 2.4\% and 7.6\% over all settings. Under standardized Gamma errors, $T_{S_0}$, $T_{M_0}$, $T_{C_0}$ remain between 2.0\% and 6.8\%, while the max-$L_2$ test $T_{S_{1/2}}$ and its combination $T_{C_{1/2}}$ are mildly liberal in a few asymmetric settings, with maxima of 8.4\% and 7.4\%, respectively, {in the settings considered.} The two dependent-data competitors, LXZL and WZVS, also remain broadly close to the nominal level. In contrast, JPYZ and WZWY are reasonably calibrated when $q_{\mathrm{MA}}=0$ but reject in every replication when $q_{\mathrm{MA}}=2$, {showing severe size distortion in these serially dependent designs.}

For the power comparison, we set the first $s$ coordinates to have a common mean change, $\delta_i=c_\tau\sqrt{\log p/(ns)}$ for $i=1,\ldots,s$ and $\delta_i=0$ for $i>s$, where $s$ controls the sparsity of the alternative. We take $s\in\{1,2,3,5,7,9,20,30,40,50\}$, set $c_\tau=20$ for $\tau=0.2n$ and $c_\tau=15$ for $\tau=0.5n$, and report results for $n=400$, $p=500$, Scenario (S2), and $q_{\mathrm{MA}}=2$. %

\begin{figure}[htbp]
    \centering
    \includegraphics[width=0.9\linewidth]{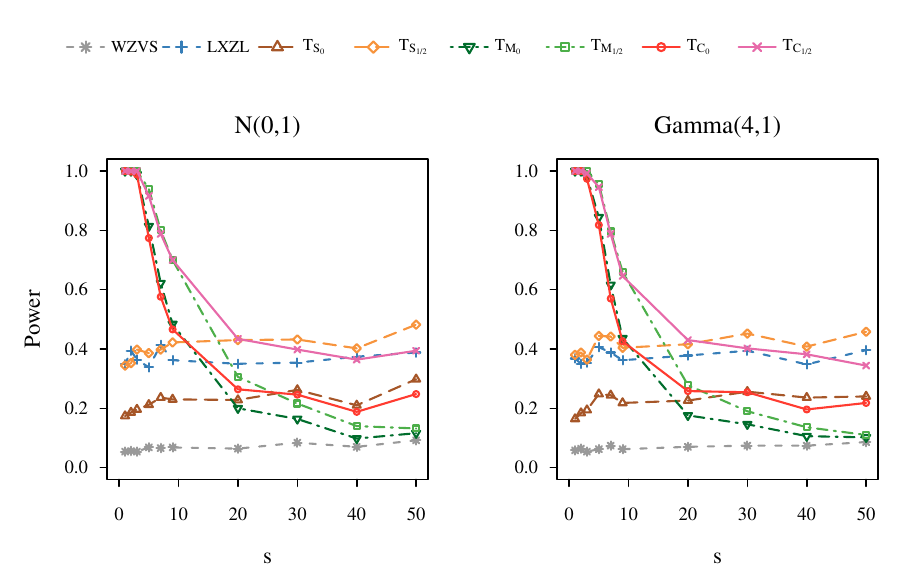}
    \caption{Empirical power under Scenario (S2), $q_{\mathrm{MA}}=2$, and $\tau=0.2n$.  The left and right panels correspond to normal and standardized $\operatorname{Gamma}(4,1)$ errors, respectively.}
    \label{fig:power_s2_02}
\end{figure}

\begin{figure}[htbp]
    \centering
    \includegraphics[width=0.9\linewidth]{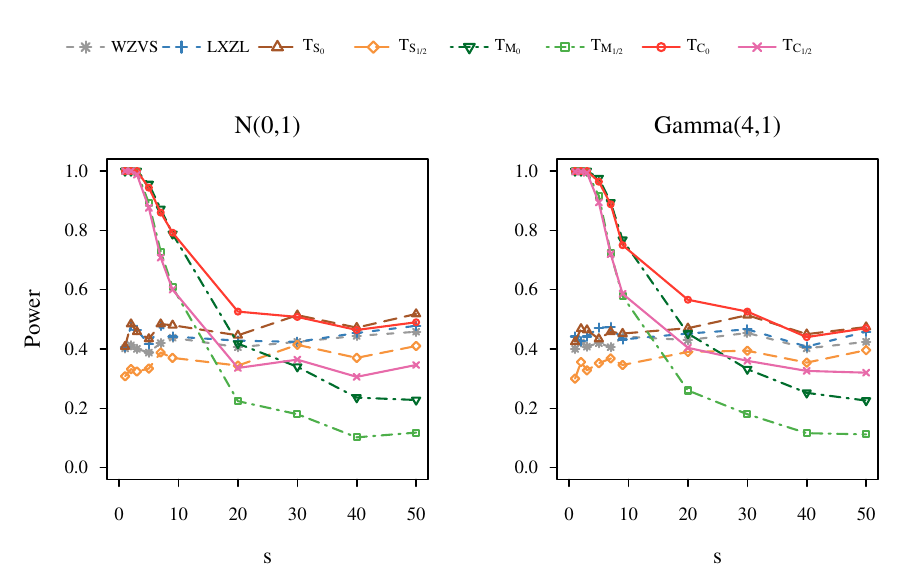}
    \caption{Empirical power under Scenario (S2), $q_{\mathrm{MA}}=2$, and $\tau=0.5n$. The left and right panels correspond to normal and standardized $\operatorname{Gamma}(4,1)$ errors, respectively.}
    \label{fig:power_s2_05}
\end{figure}
Figures~\ref{fig:power_s2_02}--\ref{fig:power_s2_05} report the empirical power under different change point locations. The experiment moves from sparse alternatives with strong coordinatewise changes to denser alternatives with weaker coordinatewise changes while holding the total $L_2$ signal constant as $s$ increases.

{When $\tau=0.5n$, the unweighted statistics generally have higher power once $s$ is moderate. When $\tau=0.2n$, the variance-standardized statistics generally have higher power. The coordinatewise maximum tests are more powerful for sparse alternatives, whereas the quadratic tests gain power as $s$ increases. For small $s$, $T_{C_\gamma}$ closely follows $T_{M_\gamma}$; for moderate and large $s$, it retains power as the maximum component loses power. The curves are similar under the two error distributions considered.}
{For small $s$, $T_{M_\gamma}$ and $T_{C_\gamma}$ have higher power than LXZL and WZVS. At $\tau=0.2n$, $T_{S_{1/2}}$ and $T_{C_{1/2}}$ generally outperform LXZL, and most proposed procedures outperform WZVS. At $\tau=0.5n$, $T_{S_0}$ and $T_{C_0}$ generally outperform both competitors.}

\begin{figure}[htbp]
    \centering
    \includegraphics[width=0.9\linewidth]{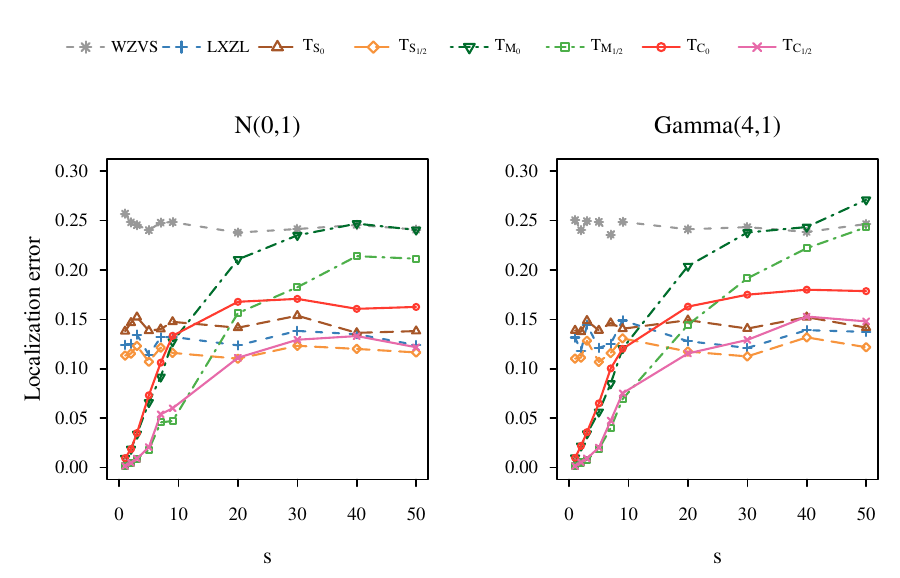}
    \caption{Localization error under Scenario (S2), $q_{\mathrm{MA}}=2$, and $\tau=0.2n$. The left and right panels correspond to normal and standardized $\operatorname{Gamma}(4,1)$ errors, respectively. %
    }
    \label{fig:acc_s2_02}
\end{figure}

\begin{figure}[htbp]
    \centering
    \includegraphics[width=0.9\linewidth]{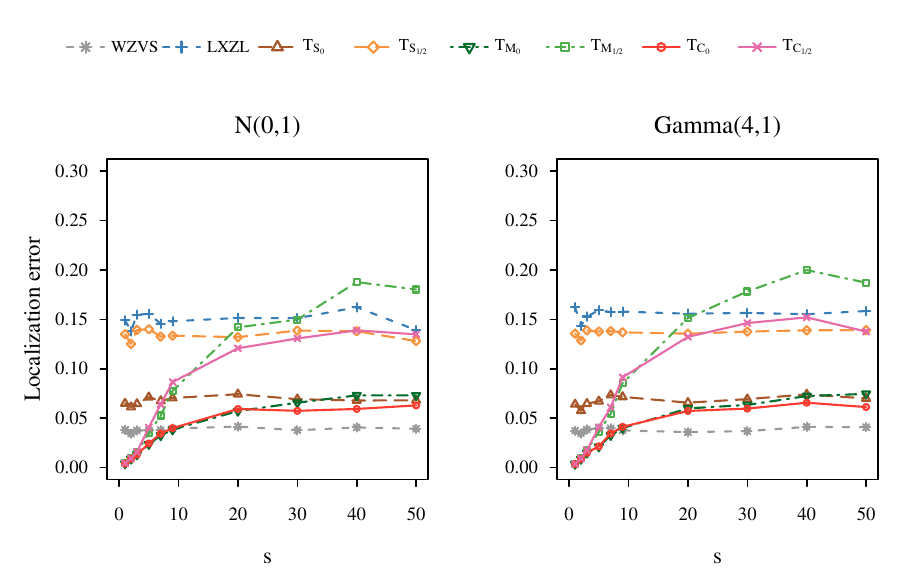}
    \caption{Localization error under Scenario (S2), $q_{\mathrm{MA}}=2$, and $\tau=0.5n$. The left and right panels correspond to normal and standardized $\operatorname{Gamma}(4,1)$ errors, respectively.
    }
    \label{fig:acc_s2_05}
\end{figure}
We next evaluate single change localization under the same design as before. Figures~\ref{fig:acc_s2_02}--\ref{fig:acc_s2_05} present the localization error, measured by the absolute distance between the estimated and true changepoints, scaled by the sample size $n$.

{The location estimators associated with $T_{M_\gamma}$ have smaller errors under sparse alternatives, whereas those associated with $T_{S_\gamma}$ are more accurate under dense alternatives. The adaptive estimators maintain relatively small errors across the sparsity levels considered.} %
{At $\tau=0.5n$, the estimators associated with $T_{M_\gamma}$ and $T_{C_\gamma}$ are more accurate for small $s$, while WZVS has a slight advantage for moderate $s$. Most proposed estimators outperform LXZL. At $\tau=0.2n$, the estimators associated with $T_{S_{1/2}}$ and $T_{C_{1/2}}$ are competitive with or more accurate than LXZL, and most proposed estimators outperform WZVS.}

We finally examine multiple change estimation in a design that compares all eight procedures. We set $n=400$ and $p=500$ and use Scenario~(S2) with $q_{\mathrm{MA}}=2$. The errors are either normal or standardized Gamma, as in the preceding experiments. The two true change points are $\tau_1=120$, $\tau_2=280$ and the mean sequence follows the three-segment pattern $\bm 0$, $\mu_s\bm v_s$, and $\bm 0$, where the first $s$ coordinates of $\bm v_s$ equal one and the remaining coordinates equal zero. We take $s\in\{1,2,3,5,7,9,20,30,40,50\}$, and $\mu_s=\{12+22s^{-1/2}\}\sqrt{\log p/n}$.
Thus, the two jumps have equal magnitude and opposite signs, while the affected dimension ranges from very sparse to moderately dense. We compare WZVS, LXZL, and the six proposed procedures $T_{S_{0}}$, $T_{M_{0}}$, $T_{C_{0}}$, $T_{S_{1/2}}$, $T_{M_{1/2}}$ and $T_{C_{1/2}}$.  The proposed WBS procedures use 300 random intervals, a minimum interval length of 20, and a separation guard of 60.  We take {the common baseline threshold sequence} $\bar\Lambda_{S,\gamma,n}^{\mathrm W}=\bar\Lambda_{M,\gamma,n}^{\mathrm W}=n^{1/4}$ and use the calibration constants $c_{S,0}^{\mathrm W}=0.30$, $c_{S,1/2}^{\mathrm W}=1$, $c_{M,0}^{\mathrm W}=1$, and $c_{M,1/2}^{\mathrm W}=2$. Thus the four actual component thresholds are $0.30 n^{1/4}$, $n^{1/4}$, $n^{1/4}$, and $2n^{1/4}$, respectively.  WZVS is implemented through \texttt{SNSeg} with confidence level $0.95$, and LXZL through \texttt{HDcpDetect} with significance level $0.05$, $M$-threshold $0.05$, and no power enhancement. %

For a method with estimated change point set $\widehat{\mathfrak T}$, let $\mathfrak T=\{\tau_1,\tau_2\}$ and $\widehat K=|\widehat{\mathfrak T}|$.  We report the two directed localization errors $\mathrm{OE}=\max_{\tau\in\mathfrak T}\min_{\widehat\tau\in\widehat{\mathfrak T}}\vert\widehat\tau-\tau\vert$, $\mathrm{UE}=\max_{\widehat\tau\in\widehat{\mathfrak T}}\min_{\tau\in\mathfrak T}\vert\widehat\tau-\tau\vert$, together with the absolute error in the estimated number of change points, $|\widehat{K}-K|$.  
\begin{figure}[!htb]
    \centering
    \includegraphics[width=0.98\linewidth]{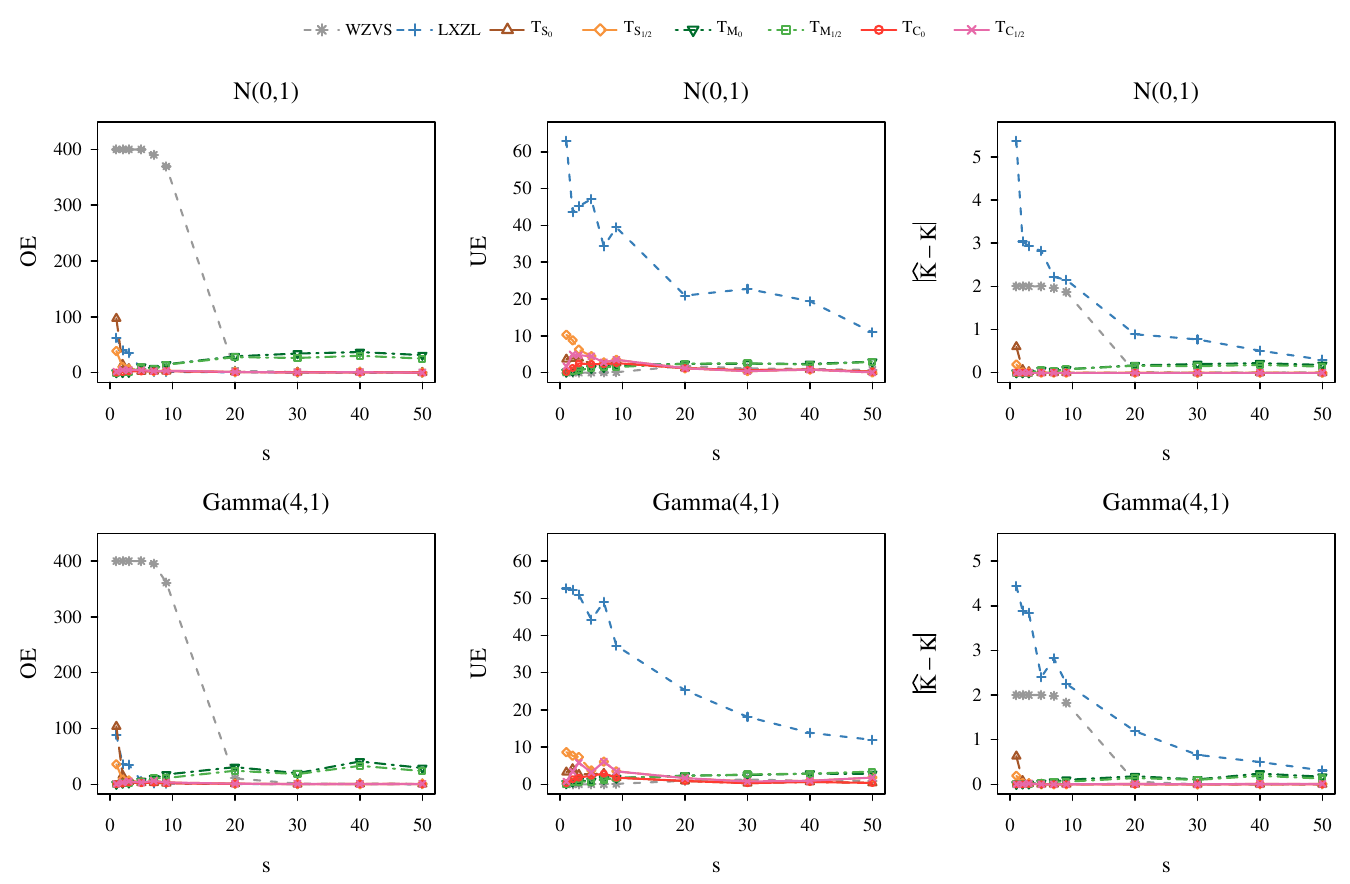}
    \caption{Multiple change estimation errors under Scenario~(S2), $q_{\mathrm{MA}}=2$, $n=400$, and $p=500$.  The left, middle, and right columns report the average OE, UE, and $\vert \widehat{K}-K\vert$, respectively.  %
    }
    \label{fig:wbs-eight-summary}
\end{figure}
{Figure~\ref{fig:wbs-eight-summary} shows the same sparsity-dependent pattern as the single change experiments: $T_{M_\gamma}$ has smaller localization errors for small $s$, $T_{S_\gamma}$ improves as the signal becomes denser, and $T_{C_\gamma}$ maintains relatively small errors across sparsity levels.} %
{WZVS has large OE and an average absolute change-count error close to two under sparse and intermediate alternatives, with smaller errors in the dense regime. LXZL has smaller OE as $s$ increases, but its UE and absolute change-count error remain larger, especially for sparse alternatives.} %

\section{Real data applications}\label{sec:applications}

\subsection{FRED-MD macroeconomic panel}\label{sec:fredmd}

We apply the proposed procedures to the June 2026 vintage of FRED-MD, a monthly database of U.S. macroeconomic indicators maintained by the Federal Reserve Bank of St. Louis \citep{McCrackenNg2016}.  The database covers output and income, the labour market, housing, consumption and orders, money and credit, interest and exchange rates, prices, and financial markets. We use the official transformation codes supplied with FRED-MD and retain the longest balanced panel with no missing observations. The resulting data set contains $n=789$ monthly observations from January 1960 to September 2025 and $p=119$ series. Variables with a later starting date or internal missing values are excluded.%

{Table~\ref{tab:fred-single} reports the single change results. WZVS and LXZL do not reject the null hypothesis, whereas all six proposed procedures reject at the $5\%$ level. The two quadratic procedures estimate a change in April 2022, while the two maximum and two adaptive procedures select June 2020. Thus the quadratic and coordinatewise maximum criteria identify different locations in this panel.}

\begin{table}[htbp]
\centering
\small
\begin{tabular}{lccc}
\toprule
Procedure & Estimated date & $p$-value & Reject at $5\%$ \\
\midrule
WZVS & 2004-12 & $>0.20$ & No \\
LXZL & 2022-04 & $0.6654$ & No \\
$T_{S_{0}}$ & 2022-04 & $\leq 9.999\times10^{-5}$ & Yes \\
$T_{S_{1/2}}$ & 2022-04 & $1.100\times10^{-21}$ & Yes \\
$T_{M_{0}}$ & 2020-06 & $2.330\times10^{-20}$ & Yes \\
$T_{M_{1/2}}$ & 2020-06 & $8.174\times10^{-26}$ & Yes \\
$T_{C_0}$ & 2020-06 & $\leq 4.660\times10^{-20}$ & Yes \\
$T_{C_{1/2}}$ & 2020-06 & $\leq 1.635\times10^{-25}$ & Yes \\
\bottomrule
\end{tabular}
\caption{Single change analysis of the FRED-MD panel.}
\label{tab:fred-single}
\end{table}

We next apply the guarded wild binary segmentation procedures in Section~\ref{sec:wbs}. The estimated multiple change points are summarized in Table~\ref{tab:fred-multiple} and Figure~\ref{fig:fred-multiple}. WZVS detects only one change, in June 2017, and LXZL detects none. By contrast, each proposed procedure identifies several change points. The two $T_{S_{\gamma}}$ procedures each return six changes, and their paired locations differ by at most 10 months. The five locations common to $T_{M_{\gamma}}$, $\gamma\in\{0,1/2\}$ differ by at most 9 months, with $T_{M_{1/2}}$ additionally detecting November 1970 and March 1995. The $T_{C_{\gamma}}$ procedures share six locations within 13 months, and $T_{C_{1/2}}$ also detects March 1995. In particular, the sparse and adaptive procedures agree closely around 1978, 1985, 2000--2001, 2011, and 2020. 

\begin{table}[!htbp]
\centering
\footnotesize
\begin{tabular}{@{}lcp{0.70\textwidth}@{}}
\toprule
Procedure & Number & Estimated dates \\
\midrule
WZVS & 1 & 2017-06 \\
LXZL & 0 & -- \\
$T_{S_{0}}$ & 6 & 1975-04; 1982-11; 1992-02; 2001-02; 2009-08; 2020-05 \\
$T_{S_{1/2}}$ & 6 & 1975-04; 1982-12; 1991-04; 2001-02; 2009-05; 2020-05 \\
$T_{M_{0}}$ & 5 & 1978-01; 1985-10; 2000-04; 2011-01; 2020-06 \\
$T_{M_{1/2}}$ & 7 & 1970-11; 1978-01; 1985-10; 1995-03; 2001-01; 2011-01; 2020-06 \\
$T_{C_0}$ & 6 & 1969-10; 1978-01; 1985-10; 2000-04; 2011-01; 2020-06 \\
$T_{C_{1/2}}$ & 7 & 1970-11; 1978-01; 1985-10; 1995-03; 2001-01; 2011-01; 2020-05 \\
\bottomrule
\end{tabular}
\caption{Multiple change estimates for the FRED-MD panel.}
\label{tab:fred-multiple}
\end{table}

\begin{figure}[!htbp]
\centering
\includegraphics[width=0.8\linewidth]{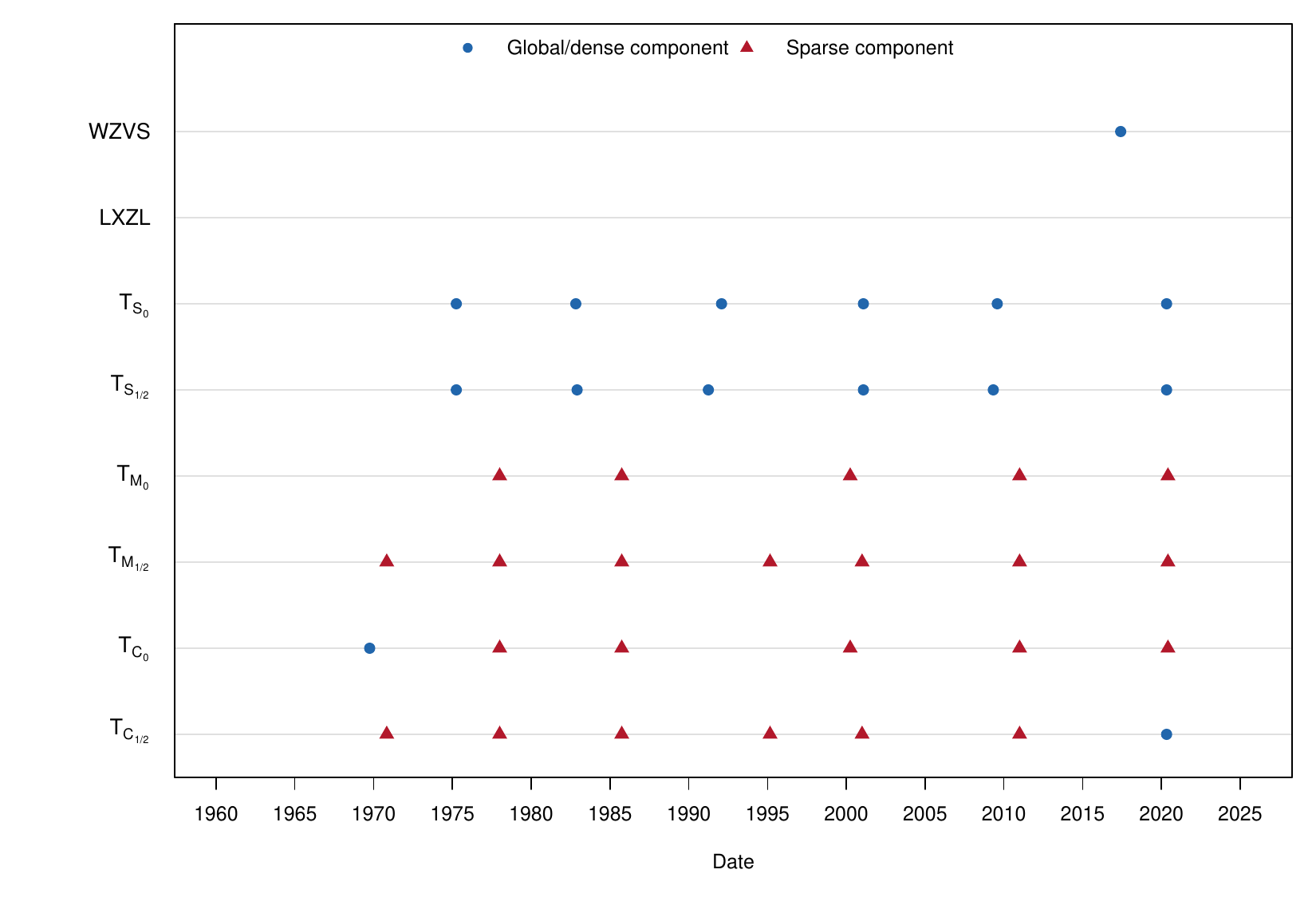}
\caption{Estimated multiple change points from all eight procedures for the FRED-MD panel.  Circles and triangles denote changes attributed to the global or dense component and the sparse component, respectively.}
\label{fig:fred-multiple}
\end{figure}

{The estimated dates describe changes in the joint mean of the transformed series. Without a known set of change points for comparison, their proximity to economic events does not by itself establish either the causes of the changes or the accuracy of the procedures.}

\subsection{Electricity-consumption panel}\label{sec:electricity}

We next analyze the Electricity Load Diagrams 2011--2014 data set from the UCI Machine Learning Repository \citep{Trindade2015Electricity}.  The source contains electricity consumption for 370 clients observed every 15 minutes on the raw load scale, with time stamps reported in Portuguese local time.  We shift each time stamp back by 15 minutes before aggregating the observations to daily means, so that an interval ending at midnight is assigned to the preceding day.  The analysis uses the period from January 2012 through December 2014, thereby avoiding most of the artificial zero-load histories in 2011 for clients that entered the database later.  We remove series with missing daily values, zero temporal variance, or a nonpositive daily mean.  One further series, \texttt{MT\_362}, is excluded because it accounts for approximately $96.7\%$ of the sum of the coordinatewise temporal variances and would otherwise dominate the raw-scale quadratic statistics.  The resulting panel has $n=1096$ daily observations and $p=289$ client series.  Even after this exclusion, the largest remaining series, \texttt{MT\_196}, contributes approximately $61.1\%$ of the total coordinatewise variance, so the panel remains strongly heterogeneous. %

Table~\ref{tab:electricity-single} reports the single change analysis.  WZVS does not reject the constant-mean null and is therefore excluded from the subsequent multiple change comparison.  {LXZL and all six proposed procedures reject at the $5\%$ level. LXZL estimates a change in October 2012.}  The two $T_{S_{\gamma}}$ procedures locate a change in October 2012, the two $T_{M_{\gamma}}$ {procedures estimate a change in December 2012}, and the two $T_{C_{\gamma}}$ select between these two locations according to the smaller component $p$-value.  %

\begin{table}[!htbp]
\centering
\small
\begin{tabular}{lccc}
\toprule
Procedure & Estimated date & $p$-value & Reject at $5\%$ \\
\midrule
WZVS & 2012-09-14 & $0.10< p<0.20$ & No \\
LXZL & 2012-10-13 & $9.738\times10^{-31}$ & Yes \\
$T_{S_0}$ & 2012-10-14 & $\leq 9.999\times10^{-5}$ & Yes \\
$T_{S_{1/2}}$ & 2012-10-13 & $4.273\times10^{-216}$ & Yes \\
$T_{M_0}$ & 2012-12-14 & $2.094\times10^{-50}$ & Yes \\
$T_{M_{1/2}}$ & 2012-12-13 & $7.663\times10^{-47}$ & Yes \\
$T_{C_0}$ & 2012-12-14 & $\leq 4.188\times10^{-50}$ & Yes \\
$T_{C_{1/2}}$ & 2012-10-13 & $\leq 8.546\times10^{-216}$ & Yes \\
\bottomrule
\end{tabular}
\caption{Single change analysis of the Electricity Load Diagrams panel.}
\label{tab:electricity-single}
\end{table}

For multiple change estimation, we retain LXZL and all six proposed procedures {and apply WBS}.  Table~\ref{tab:electricity-multiple} and Figure~\ref{fig:electricity-multiple} show four to six estimated changes for the proposed procedures.  {The estimated changes are concentrated in} May--October 2012, spring to autumn 2013, and mid-to-late 2014. $T_{C_{0}}$ detects changes in May 2012, October 2012, April 2013, September 2013, April 2014, and September 2014, whereas $T_{C_{1/2}}$ detects changes in May 2012, October 2012, June 2013, October 2013, and July 2014.

\begin{table}[!htbp]
\centering
\footnotesize
\begin{tabular}{@{}lcp{0.70\textwidth}@{}}
\toprule
Procedure & Number & Estimated dates \\
\midrule
LXZL & 6 & 2012-05; 2012-10; 2013-06; 2013-09; 2014-06; 2014-10 \\
$T_{S_0}$ & 6 & 2012-05; 2012-10; 2013-04; 2013-09; 2014-04; 2014-09 \\
$T_{S_{1/2}}$ & 4 & 2012-05; 2012-10; 2013-06; 2014-07 \\
$T_{M_0}$ & 5 & 2012-04; 2012-10; 2013-06; 2014-01; 2014-07 \\
$T_{M_{1/2}}$ & 5 & 2012-04; 2012-11; 2013-06; 2013-10; 2014-07 \\
$T_{C_0}$ & 6 & 2012-05; 2012-10; 2013-04; 2013-09; 2014-04; 2014-09 \\
$T_{C_{1/2}}$ & 5 & 2012-05; 2012-10; 2013-06; 2013-10; 2014-07 \\
\bottomrule
\end{tabular}
\caption{Tuned multiple change estimates for the Electricity Load Diagrams panel after excluding WZVS.}
\label{tab:electricity-multiple}
\end{table}

\begin{figure}[!htbp]
\centering
\includegraphics[width=0.8\linewidth]{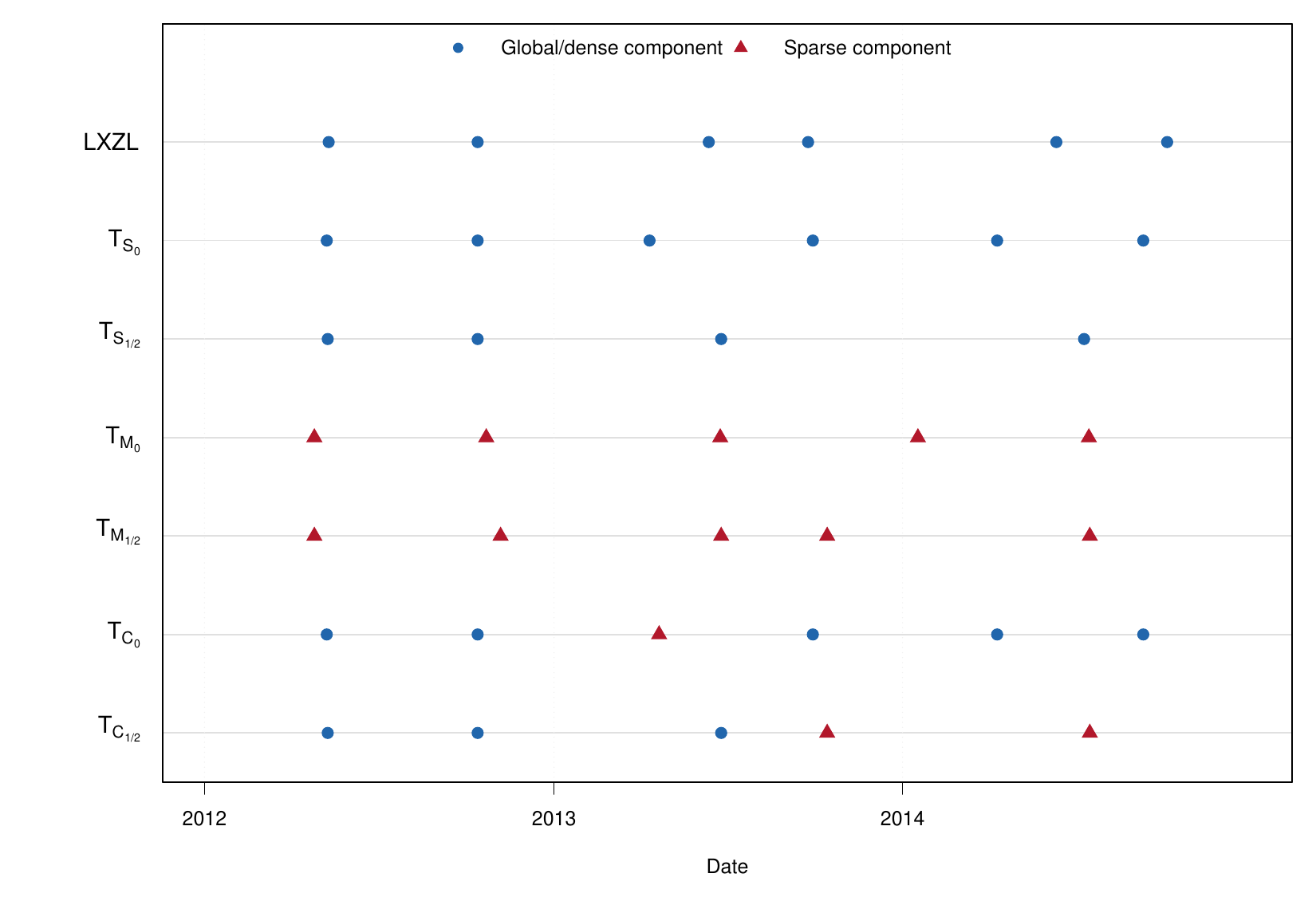}
\caption{Tuned multiple change estimates for the Electricity Load Diagrams panel after excluding WZVS.  Circles and triangles denote changes attributed to the global or dense component and the sparse component, respectively.}
\label{fig:electricity-multiple}
\end{figure}

{October 2012 is selected by LXZL, $T_{S_0}$, $T_{S_{1/2}}$, $T_{M_0}$, $T_{C_0}$, and $T_{C_{1/2}}$; $T_{M_{1/2}}$ selects November 2012. In 2013, LXZL, $T_{S_{1/2}}$, $T_{M_0}$, $T_{M_{1/2}}$, and $T_{C_{1/2}}$ select dates around June, whereas $T_{S_0}$ and $T_{C_0}$ select April and September. In 2014, the variance-standardized procedures select dates around July, LXZL selects June and October, and $T_{S_0}$ and $T_{C_0}$ select April and September. These differences show how the weighting and aggregation choices affect the estimated segmentation; the dates alone do not identify the clients responsible for a change or distinguish seasonal variation from a structural break.}

\section{Conclusion}

We develop two max-$L_2$-type tests for detecting dense mean changes in high-dimensional time series under temporal dependence. Building on these tests, we combine the max-$L_2$-type statistics with max-$L_\infty$-type statistics designed for sparse alternatives. We establish the limiting distributions of the component statistics, validate the required nuisance-parameter estimators, and prove the asymptotic independence between the matched max-$L_2$- and max-$L_\infty$-type statistics. These results justify two Cauchy combination tests that adapt to unknown signal sparsity. We further study the corresponding single change estimators and extend the component and adaptive procedures to multiple change estimation through guarded wild binary segmentation. The simulations evaluate finite-sample size, power, and localization performance, and the two applications demonstrate the use of the procedures on macroeconomic and electricity-load panels.

There remain several important directions for future research. First, accurate estimation of the temporal dependence parameters, specifically the lag-truncation parameter $M$ and the long-run variance, remains a challenging task in change point detection and warrants further investigation. Second, several studies have considered change point problems for high-dimensional time series under strong factor-driven cross-sectional dependence \citep{BarigozziChoFryzlewicz2018,HorvathLiuRiceZhao2022}. Extending our sparsity-adaptive framework to such strong-factor structures while retaining the temporal-dependence setting would be an interesting direction for future research.

\section*{Supplementary material}
The Supplementary Material, appended after the references, contains detailed proofs of the theoretical results, uniform consistency arguments for the long-run variance estimators, and auxiliary moment and concentration inequalities.

\FloatBarrier
\bibliographystyle{biometrika}
\bibliography{ref}

\clearpage
\phantomsection
\addcontentsline{toc}{section}{Supplementary Material}
\begin{center}
{\Large\bfseries Supplementary Material}
\end{center}
\medskip

\setcounter{section}{0}
\setcounter{subsection}{0}
\setcounter{equation}{0}
\setcounter{theorem}{0}
\setcounter{lemma}{0}
\setcounter{proposition}{0}
\setcounter{corollary}{0}
\setcounter{assumption}{0}
\setcounter{remark}{0}
\setcounter{figure}{0}
\setcounter{table}{0}
\renewcommand{\thesection}{S.\arabic{section}}
\renewcommand{\theequation}{S.\arabic{equation}}
\renewcommand{\thetheorem}{S.\arabic{theorem}}
\renewcommand{\thelemma}{S.\arabic{lemma}}
\renewcommand{\theproposition}{S.\arabic{proposition}}
\renewcommand{\thecorollary}{S.\arabic{corollary}}
\renewcommand{\theassumption}{S\arabic{assumption}}
\renewcommand{\theremark}{S.\arabic{remark}}
\renewcommand{\thefigure}{S\arabic{figure}}
\renewcommand{\thetable}{S\arabic{table}}
\renewcommand{\theHsection}{supp.\arabic{section}}
\renewcommand{\theHequation}{supp.\arabic{equation}}
\renewcommand{\theHtheorem}{supp.\arabic{theorem}}
\renewcommand{\theHlemma}{supp.\arabic{lemma}}
\renewcommand{\theHproposition}{supp.\arabic{proposition}}
\renewcommand{\theHcorollary}{supp.\arabic{corollary}}
\renewcommand{\theHassumption}{supp.\arabic{assumption}}
\renewcommand{\theHremark}{supp.\arabic{remark}}
\renewcommand{\theHfigure}{supp.\arabic{figure}}
\renewcommand{\theHtable}{supp.\arabic{table}}
\renewcommand{\ind}[1]{\mathbb{I}\{#1\}}

{This supplement proves } Theorems~\ref{Th1}--\ref{consistency-M}, the multiple-change Theorems~\ref{thm:wbs-dense}--\ref{thm:wbs-max}, Propositions~\ref{prop:Tms-alter}, \ref{prop:consistency-Linf}, \ref{prop:cauchy-size}, and \ref{prop:cauchy-consistency}, and Corollaries~\ref{cor:adaptive-location} and~\ref{cor:wbs-adaptive}.  

\section{{{Consequences of Assumption~\ref{ass:C1}}}
}

Throughout the supplement, $C_1,C_2>0$ denote constants independent of $n$ and $p$.  Their values may change between displays.  Recall that
\[
\varpi_{\mathrm{dep}}(r)=C_1\exp(-C_2r^{\varkappa}),
\qquad
\eta_m=\sum_{r>m}(1+r)^8\varpi_{\mathrm{dep}}(r),
\]
with the convention $\varpi_{\mathrm{dep}}(r)=C\ind{r\le m_0}$ under fixed $m_0$-dependence.  Assumption~\ref{ass:C1} gives, for every $A>0$, $m>1$,
\[
\eta_m\le C\exp(-cm^{\varkappa}),
\qquad
\eta_m\le C_A m^{-A}.
\]
For a deterministic sequence $\boldsymbol w=(w_i)_{i\in \mathbb{Z}}$, let $\|\boldsymbol w\|_r=(\sum_{i\in \mathbb{Z}}|w_i|^r)^{1/r}$ and $\|\boldsymbol w\|_\infty=\max_{i\in \mathbb{Z}}|w_i|$.

\section{Proofs of Theorems~\ref{Th1}--\ref{Tms-1} and Proposition~\ref{prop:Tms-alter}}
\label{sec:supp-dense-l2}

\subsection{Notation}

For $k=0,\ldots,n$, define $c_{i,k}=\ind{i\le k}-k/n$. Let $\{\bepsilon_i^G:i\in\mathbb Z\}$ be a centered stationary Gaussian
vector process satisfying
\(
\Cov(\bepsilon_i^G,\bepsilon_{i+h}^G)=\bGam(h),
\)
for
\(h\in\mathbb Z\).

For $t\in[0,1]$, define
\begin{equation}\label{eq:def-noise-cusum-processes}
\mathbf U_n(t)
=
\frac1{\sqrt n}\sum_{i=1}^n
c_{i,\lfloor nt\rfloor}\bepsilon_i,
\qquad
\mathbf U_n^G(t)
=
\frac1{\sqrt n}\sum_{i=1}^n
c_{i,\lfloor nt\rfloor}\bepsilon_i^G.
\end{equation}
Let $B_1,\ldots,B_p$ be independent standard Brownian bridges and let
\(
\mathbf B(t)=\{B_1(t),\ldots,B_p(t)\}^{\top}.
\)
Let \(K_B(s,t)=s\wedge t-st\) denote the Brownian-bridge covariance kernel.

For $t\in[0,1]$ and a deterministic vector
$\mathbf h\in\mathbb R^p$, define
\begin{align}
Q_{n,p}(t;\mathbf h)
&=
\frac{\|\mathbf U_n(t)+\mathbf h\|_2^2
-\E\|\mathbf U_n(t)\|_2^2}{\sqrt p},
\label{eq:def-Qnp-shifted}\\
Q_{n,p}^G(t;\mathbf h)
&=
\frac{\|\mathbf U_n^G(t)+\mathbf h\|_2^2
-\E\|\mathbf U_n^G(t)\|_2^2}{\sqrt p},
\label{eq:def-QnpG-shifted}\\
G_p(t;\mathbf h)
&=
\frac{\|\bOme^{1/2}\mathbf B(t)+\mathbf h\|_2^2
-\E\|\bOme^{1/2}\mathbf B(t)\|_2^2}
{\sqrt p\,\omega_p}.
\label{eq:def-Gp-shifted}
\end{align}
For $\mathbf h=\mathbf0$, let $Q_{n,p}(t)=Q_{n,p}(t;\mathbf0)$, $Q_{n,p}^G(t)=Q_{n,p}^G(t;\mathbf0)$ and $G_p(t)=G_p(t;\mathbf0)$. Notice that \(Q_{n,p}(t)=W(\lfloor nt\rfloor)-\mu_{\lfloor nt\rfloor}.\)

\subsection{Auxiliary Lemmas}
\label{subsec:supp-dense-cusum-tools}

\begin{lemma}\label{lem:cusum-cov-general}
Under Assumption~\ref{ass:C1}, uniformly in $s,t\in[0,1]$,
\[
\left\|\Cov\{\mathbf U_n(s),\mathbf U_n(t)\}
-K_B(s,t)\bOme\right\|_{\mathrm{op}}
\le \frac Cn.
\]
Moreover, for $0\le s\le t\le1$ and
$\Delta_{s,t}=t-s+n^{-1}$, 
\[
\|\Cov\{\mathbf U_n(t)-\mathbf U_n(s)\}\|_{\mathrm{op}}
\le C\Delta_{s,t}.
\]
\end{lemma}

\begin{proof}
\begingroup
We first compare the lag coefficients with the bridge kernel, and then
apply the same covariance expansion to an increment. Set
\begin{align*}
k&=\lfloor ns\rfloor,\qquad \ell=\lfloor nt\rfloor,\qquad
c_{i,k}=0\quad(i\notin\{1,\ldots,n\}),\\
b_{n,h}(s,t)&=\frac1n\sum_{i\in\mathbb Z}c_{i,k}c_{i+h,\ell}.
\end{align*}
Stationarity gives
\begin{align*}
\Cov\{\mathbf U_n(s),\mathbf U_n(t)\}
&=\sum_{h\in\mathbb Z}b_{n,h}(s,t)\bGam(h),\\
b_{n,0}(s,t)&=K_B(k/n,\ell/n),\qquad
|b_{n,0}(s,t)-K_B(s,t)|\le 2/n.
\end{align*}
The zero-extended sequence $c_{i,\ell}$ has total variation at most two.
Consequently, for $|h|<n$,
\begin{align*}
|b_{n,h}(s,t)-b_{n,0}(s,t)|
&\le\frac1n\sum_i|c_{i+h,\ell}-c_{i,\ell}|\\
&\le\frac{|h|}{n}\sum_i|c_{i+1,\ell}-c_{i,\ell}|
\le\frac{2|h|}{n}.
\end{align*}
Since $b_{n,h}=0$ for $|h|\ge n$, Lemma~\ref{lem:vector-basic}(i) yields
\begin{align*}
\left\|\Cov\{\mathbf U_n(s),\mathbf U_n(t)\}
-K_B(s,t)\bOme\right\|_{\mathrm{op}}
&\le \frac Cn\sum_{|h|<n}(1+|h|)\|\bGam(h)\|_{\mathrm{op}}
 +C\sum_{|h|\ge n}\|\bGam(h)\|_{\mathrm{op}}\\
&\le\frac Cn\sum_{h\in\mathbb Z}(1+|h|)\|\bGam(h)\|_{\mathrm{op}}
\le C/n.
\end{align*}
For the increment, let $d_i=c_{i,\ell}-c_{i,k}$. Then
\begin{align*}
\frac1n\sum_i d_i^2
&=\frac{\ell-k}{n}\left(1-\frac{\ell-k}{n}\right)
\le t-s+n^{-1}=\Delta_{s,t},\\
\left\|\Cov\{\mathbf U_n(t)-\mathbf U_n(s)\}\right\|_{\mathrm{op}}
&\le\frac1n\sum_h\|\bGam(h)\|_{\mathrm{op}}
       \sum_i|d_id_{i+h}|\\
&\le\frac{\|d\|_2^2}{n}\sum_h\|\bGam(h)\|_{\mathrm{op}}
\le C\Delta_{s,t}.
\end{align*}
\endgroup
\end{proof}

\begin{lemma}\label{lem:quadratic-cumulants}
Under Assumptions~\ref{ass:C1}--\ref{ass:C3}, %
for $0\le s\le t\le1$,
\[
\left|\Cov\{Q_{n,p}(s),Q_{n,p}(t)\}
-2K_B(s,t)^2\frac{\Tr(\bOme^2)}p\right|
\le Cn^{-1}.
\]
Furthermore,
with 
\(
\Delta_{s,t}=t-s+n^{-1},
\)
\[
\E|Q_{n,p}(t)-Q_{n,p}(s)|^4
\le C\Delta_{s,t}^2,
\]
\end{lemma}

\begin{proof}
\begingroup
The covariance claim follows from the fourth-cumulant remainder and
Lemma~\ref{lem:cusum-cov-general}; the increment claim follows from a
centered bilinear representation. Let
\begin{align*}
x_i(u)&=n^{-1/2}c_{i,\lfloor nu\rfloor}\ind{1\le i\le n},
\qquad \mathbf C_n(s,t)=\Cov\{\mathbf U_n(s),\mathbf U_n(t)\},\\
R_{4,n}(s,t)&=\frac1p\sum_{a,b=1}^p\sum_{i_1,i_2,i_3,i_4}
x_{i_1}(s)x_{i_2}(s)x_{i_3}(t)x_{i_4}(t)
\Cum(\epsilon_{i_1,a},\epsilon_{i_2,a},
      \epsilon_{i_3,b},\epsilon_{i_4,b}).
\end{align*}
The fourth-order moment--cumulant identity gives
\[
\Cov\{Q_{n,p}(s),Q_{n,p}(t)\}
=\frac2p\Tr\{\mathbf C_n(s,t)\mathbf C_n(t,s)\}+R_{4,n}(s,t).
\]
For $\pi=\{\{1,2,3,4\}\}$, the single-block bound in
Lemma~\ref{lem:cumulant-coefficient-contraction} gives
\begin{align}
|R_{4,n}(s,t)|
&\le p^{-1}\mathfrak C_{\pi,p}
\|x(s)\|_1\|x(s)\|_\infty\|x(t)\|_\infty^2\notag\\
&\le p^{-1}(Cp)(C\sqrt n)(Cn^{-1/2})^3
\le Cn^{-1}.\label{LS4-1}
\end{align}
Here $\mathfrak C_{\pi,p}\le Cp$ follows from
Assumption~\ref{ass:C3}. Lemma~\ref{lem:cusum-cov-general} and
Assumption~\ref{ass:C2} further imply
\begin{align}
&\left|\frac2p\Tr\{\mathbf C_n(s,t)\mathbf C_n(t,s)\}
-\frac2pK_B(s,t)^2\Tr(\bOme^2)\right|\notag\\
&\quad\le\frac2p
\|\mathbf C_n(s,t)-K_B(s,t)\bOme\|_{\mathrm F}
\{\|\mathbf C_n(t,s)\|_{\mathrm F}
  +|K_B(s,t)|\|\bOme\|_{\mathrm F}\}
\le Cn^{-1}.\label{LS4-2}
\end{align}
This proves the covariance claim.
\endgroup

\begingroup
For the increment, set
\[
u_i=x_i(t)-x_i(s),\qquad v_i=x_i(t)+x_i(s),\qquad
\mathbf A=\sum_i u_i\bepsilon_i,\qquad
\mathbf B=\sum_i v_i\bepsilon_i.
\]
With $k=\lfloor ns\rfloor$ and $\ell=\lfloor nt\rfloor$,
\begin{align*}
Q_{n,p}(t)-Q_{n,p}(s)
&=p^{-1/2}\{\mathbf A^\top\mathbf B-
                  \E(\mathbf A^\top\mathbf B)\},\\
\|u\|_2^2
&=\frac{\ell-k}{n}\left(1-\frac{\ell-k}{n}\right)
\le\Delta_{s,t},\qquad
\|v\|_2\le\|x(t)\|_2+\|x(s)\|_2\le1.
\end{align*}
Lemma~\ref{lem:bilinear-quadratic-bound}, whose proof uses only
Assumptions~\ref{ass:C1}--\ref{ass:C3}, now gives
\[
\E|Q_{n,p}(t)-Q_{n,p}(s)|^4
\le C\|u\|_2^4\|v\|_2^4
\le C\Delta_{s,t}^2.
\]
\endgroup
\end{proof}

\subsection{Proof of Theorem~\ref{Th1}}
\label{subsec:supp-dense-gaussian-limit}

\begin{lemma}
\label{lem:block-comparison-finite}
Suppose Assumptions~\ref{ass:C1}--\ref{ass:C3} hold.  Fix any $A>0$.  Under Assumption~\ref{ass:C1}(MD), choose any fixed integer
$s_n>m_0$.  Under
Assumption~\ref{ass:C1}(AM) or (PD), set
\(
s_n=
\left\lceil c_0\{\log(np)\}^{1/\varkappa}\right\rceil,
\)
where $c_0=c_0(A)$ is sufficiently large.  Let $b_n$ be a positive integer
sequence such that
\(
s_n=o(b_n),
b_n=o(n),
\)
and define
\[
\rho_n^{\mathrm{fd}}
=
\left(\frac{s_n}{b_n}\right)^{1/2}
+
\left(\frac{b_n}{n}\right)^{1/2}
+
(np)^{-A}.
\]
Then $\rho_n^{\mathrm{fd}}\to0$.  Fix $m\ge1$, $0<t_1<\cdots<t_m<1$, and deterministic vectors
$\mathbf h_1,\ldots,\mathbf h_m\in\mathbb R^p$. Then, for every $f\in C_b^3(\mathbb R^m)$,
\begin{align*}
&\left|
\E f\{Q_{n,p}(t_r;\mathbf h_r):1\le r\le m\}
-
\E f\{Q_{n,p}^G(t_r;\mathbf h_r):1\le r\le m\}
\right|\\
&\qquad\le
C_m\|f\|_{C^3}\rho_n^{\mathrm{fd}}
\left\{
1+\max_{1\le r\le m}
\frac{\|\mathbf h_r\|_2}{\sqrt p}
\right\}^{3}.
\end{align*}
\end{lemma}

\begin{proof}
Recall the definitions of $Q_{n,p}(t;\mathbf h)$ and
$Q_{n,p}^G(t;\mathbf h)$ in \eqref{eq:def-Qnp-shifted} and
\eqref{eq:def-QnpG-shifted}, respectively.
{We compare the $m$-dimensional vectors evaluated at $(t_r,\mathbf h_r)$, $r\le m$, in four steps:}
\[
Q_{n,p}
\ \rightarrow\
Q_{n,p}^{I}
\ \rightarrow\
Q_{n,p}^{I,\mathrm{ind}}
\ \rightarrow\
Q_{n,p}^{I,G,\mathrm{ind}}
\ \rightarrow\
Q_{n,p}^{G}.
\]
{Step~1 deletes} the separating small blocks
and the terminal remainder, thereby comparing the full statistic
$Q_{n,p}$ with the retained-block statistic $Q_{n,p}^{I}$.  Step~2 removes
the dependence between the retained blocks and constructs
$Q_{n,p}^{I,\mathrm{ind}}$: this is exact under (MD), uses an
$s_n$-dependent approximation under (PD), and uses sequential independent-block
replacement under (AM).  Step~3 replaces these mutually independent retained
blocks one at a time by centered Gaussian blocks having the same within-block
covariances, yielding $Q_{n,p}^{I,G,\mathrm{ind}}$.  Step~4 completes the
comparison with $Q_{n,p}^{G}$.  Throughout the proof,
let \(H=\max_{r\le m}\|\mathbf h_r\|_2/\sqrt p\).

Step 1. We partition $\{1,\ldots,n\}$ into several large blocks of
length $b_n$ and small blocks of length $s_n$, with the terminal remainder
treated as an additional small block. For $1\le r\le m$, let $k_r=\lfloor nt_r\rfloor$ and
$a_{i,r}=n^{-1/2}c_{i,k_r}$.  Let
$L=\lfloor n/(b_n+s_n)\rfloor$.  For $\ell=1,\ldots,L$, define
\begin{align*}
I_\ell
&=\{(\ell-1)(b_n+s_n)+1,\ldots,(\ell-1)(b_n+s_n)+b_n\},\\
J_\ell
&=\{(\ell-1)(b_n+s_n)+b_n+1,\ldots,\ell(b_n+s_n)\},
\end{align*}
and let $J_{L+1}$ contain the remaining indices.  Let
$\mathcal J=J_1\cup\cdots\cup J_{L+1}$.  {Expanding the definitions gives}
\[
\max_{1\le r\le m}\sum_{i\in\mathcal J}a_{i,r}^2
\le C\left(\frac{s_n}{b_n}+\frac{b_n}{n}\right)
=:(r_n^{\mathrm{del}})^2.
\]

For the deletion step, let $\mathbf U_n(t_r)=\mathbf U_r^I+\mathbf U_r^{\mathrm{del}}$, $\mathbf U_r^I=\sum_{\ell=1}^L\sum_{i\in I_\ell}a_{i,r}\bepsilon_i$ and $\mathbf U_r^{\mathrm{del}}=\sum_{i\in\mathcal J}a_{i,r}\bepsilon_i$, and define
\[
Q_{n,p}^{I}(t_r;\mathbf h_r)
=
\frac{\|\mathbf U_r^I+\mathbf h_r\|_2^2
-\E\|\mathbf U_r^I\|_2^2}{\sqrt p}.
\]
Then
\begin{align*}
&Q_{n,p}(t_r;\mathbf h_r)
-Q_{n,p}^{I}(t_r;\mathbf h_r)\\
&\quad=
\frac2{\sqrt p}
\left[
(\mathbf U_r^I+\mathbf h_r)^\top\mathbf U_r^{\mathrm{del}}
-\E\{(\mathbf U_r^I+\mathbf h_r)^\top
       \mathbf U_r^{\mathrm{del}}\}
\right]
+
\frac1{\sqrt p}
\left[
\|\mathbf U_r^{\mathrm{del}}\|_2^2
-\E\|\mathbf U_r^{\mathrm{del}}\|_2^2
\right].
\end{align*}
To verify the corresponding $L_4$ bound, define the coefficient sequences $a_{i,r}^{I}=a_{i,r}\ind{i\notin\mathcal J}$ and $a_{i,r}^{\mathrm{del}}=a_{i,r}\ind{i\in\mathcal J}$. Then $\mathbf U_r^I=\sum_i a_{i,r}^{I}\bepsilon_i$, $\mathbf U_r^{\mathrm{del}}=\sum_i a_{i,r}^{\mathrm{del}}\bepsilon_i$ and the coefficient bounds above give $\|a_r^I\|_2\le \|a_r\|_2\le C$ and $\|a_r^{\mathrm{del}}\|_2=(\sum_{i\in\mathcal J}a_{i,r}^2)^{1/2}\le r_n^{\mathrm{del}}$. Since $\E(\mathbf U_r^{\mathrm{del}})=\mathbf0$,
\begin{align*}
(\mathbf U_r^I+\mathbf h_r)^\top\mathbf U_r^{\mathrm{del}}
-\E\{(\mathbf U_r^I+\mathbf h_r)^\top
       \mathbf U_r^{\mathrm{del}}\}
=
(\mathbf U_r^I)^\top\mathbf U_r^{\mathrm{del}}
-\E\{(\mathbf U_r^I)^\top\mathbf U_r^{\mathrm{del}}\}
+\mathbf h_r^\top\mathbf U_r^{\mathrm{del}}.
\end{align*}
Lemma~\ref{lem:bilinear-quadratic-bound}, first with
$(a,b)=(a_r^I,a_r^{\mathrm{del}})$ and then with
$a=b=a_r^{\mathrm{del}}$, yields
\begin{align*}
\left\|
\frac{(\mathbf U_r^I)^\top\mathbf U_r^{\mathrm{del}}
-\E\{(\mathbf U_r^I)^\top\mathbf U_r^{\mathrm{del}}\}}
{\sqrt p}
\right\|_4
&\le
C\|a_r^I\|_2\|a_r^{\mathrm{del}}\|_2
\le Cr_n^{\mathrm{del}},\\
\left\|
\frac{\|\mathbf U_r^{\mathrm{del}}\|_2^2
-\E\|\mathbf U_r^{\mathrm{del}}\|_2^2}
{\sqrt p}
\right\|_4
&\le
C\|a_r^{\mathrm{del}}\|_2^2
\le C\{r_n^{\mathrm{del}}\}^2.
\end{align*}
Moreover, Lemma~\ref{lem:vector-basic}(ii), applied with $q=4$ and
$\mathbf u=\mathbf h_r$, gives
\[
\left\|
\frac{\mathbf h_r^\top\mathbf U_r^{\mathrm{del}}}{\sqrt p}
\right\|_4
\le
C\frac{\|\mathbf h_r\|_2}{\sqrt p}
\|a_r^{\mathrm{del}}\|_2
\le
Cr_n^{\mathrm{del}}\frac{\|\mathbf h_r\|_2}{\sqrt p}.
\]
Because $r_n^{\mathrm{del}}=o(1)$, the triangle inequality applied to the
preceding decomposition gives, for all sufficiently large $n$,
\[
\max_{1\le r\le m}
\|Q_{n,p}(t_r;\mathbf h_r)
-Q_{n,p}^{I}(t_r;\mathbf h_r)\|_4
\le
Cr_n^{\mathrm{del}}
\left\{
1+\max_{r\le m}\frac{\|\mathbf h_r\|_2}{\sqrt p}
\right\}.
\]
By the mean-value theorem,
\begin{align*}
\left|
\E f\{Q_{n,p}(t_r;\mathbf h_r):r\le m\}
-
\E f\{Q_{n,p}^{I}(t_r;\mathbf h_r):r\le m\}
\right|
\le
C_m\|f\|_{C^1}r_n^{\mathrm{del}}
\left\{
1+\max_{r\le m}\frac{\|\mathbf h_r\|_2}{\sqrt p}
\right\}.
\end{align*}

Step 2. We next reduce the retained-block array to mutually independent block contributions.  For $\ell\le L$ and $r\le m$, define
\[
\qquad
\mathbf U_r^{I_\ell}=\sum_{i\in I_\ell}a_{i,r}\bepsilon_i,
\qquad
w_{\ell r}^2=\sum_{i\in I_\ell}a_{i,r}^2.
\]
Since $m$ and the $t_r$ are fixed,
\[
\max_{\ell\le L,r\le m}w_{\ell r}\le C(b_n/n)^{1/2}.
\]

(a) Under Assumption~\ref{ass:C1}(AM), define
$\mathcal B_\ell=(\bepsilon_i:i\in I_\ell)$.  For each $\ell\le L$, let
\begin{equation}
\mathcal B_\ell^*
:= (\bepsilon_i^{(\ell,*)}:i\in I_\ell),
\qquad
\mathcal B_\ell^*\overset{d}=\mathcal B_\ell.\label{eq:AM-copied-retained-blocks}
\end{equation}
The copied blocks $\mathcal B_1^*,\ldots,\mathcal B_L^*$ are chosen
mutually independent and jointly independent of
$\mathcal B_1,\ldots,\mathcal B_L$.  We claim that
\begin{align}
&\left|
\E f\{Q_{n,p}^{I}(t_r;\mathbf h_r):r\le m\}
-
\E f\{Q_{n,p}^{I,*}(t_r;\mathbf h_r):r\le m\}
\right|\notag\\
&\qquad\le
C_m\|f\|_{C^3}
\left\{
\left(\frac{b_n}{n}\right)^{1/2}
+(np)^{-A}
\right\}(1+H)^3.
\label{eq:fd-AM-block-independence}
\end{align}
For each
$r\le m$, define
\[
\mathbf U_r^{I_\ell,*}
=\sum_{i\in I_\ell}a_{i,r}\bepsilon_i^{(\ell,*)},
\qquad
\mathbf U_r^{I,*}=\sum_{\ell=1}^L\mathbf U_r^{I_\ell,*},
\]
and 
\[
Q_{n,p}^{I,*}(t_r;\mathbf h_r)
=
\frac{\|\mathbf U_r^{I,*}+\mathbf h_r\|_2^2
-\E\|\mathbf U_r^{I,*}\|_2^2}{\sqrt p}.
\]
{We replace the blocks from left to right, one at a time.} At the $\ell$th replacement, the
past blocks have already been replaced and the future original blocks are
left unchanged.  Hence the contribution from all blocks other than
$I_\ell$ is $\mathbf S_{\ell r}
=\sum_{j<\ell}\mathbf U_r^{I_j,*}
 +\sum_{j>\ell}\mathbf U_r^{I_j}$. For $\diamond\in\{0,*\}$, let
\(
\mathbf U_{\ell r}^{(0)}=\mathbf U_r^{I_\ell},
\mathbf U_{\ell r}^{(*)}=\mathbf U_r^{I_\ell,*},
\)
and introduce the centered interpolation
\[
q_{\ell r}^{(\diamond)}(t)
=
\frac{\|\mathbf S_{\ell r}+t\mathbf U_{\ell r}^{(\diamond)}
       +\mathbf h_r\|_2^2
-\E\|\mathbf S_{\ell r}+t\mathbf U_{\ell r}^{(\diamond)}\|_2^2}
{\sqrt p},
\qquad 0\le t\le1.
\]
Define
\[
F_\ell^{(\diamond)}(t)
=f\{q_{\ell1}^{(\diamond)}(t),\ldots,
     q_{\ell m}^{(\diamond)}(t)\},
\]
and let $F_\ell(t)=F_\ell^{(0)}(t)$ and
$F_\ell^*(t)=F_\ell^{(*)}(t)$. We have $F_\ell(0)=F_\ell^*(0)$, and hence
\begin{align*}
\E f\{Q_{n,p}^{I}(t_r;\mathbf h_r):r\le m\}
-
\E f\{Q_{n,p}^{I,*}(t_r;\mathbf h_r):r\le m\}
=
\sum_{\ell=1}^L
\left\{\E F_\ell(1)-\E F_\ell^*(1)\right\}.
\end{align*}
{Taylor expansion to second order gives}
\begin{align*}
\E F_\ell(1)-\E F_\ell^*(1)
={}&\E\{\dot F_\ell(0)-\dot F_\ell^*(0)\}+\frac12\E\{\ddot F_\ell(0)-\ddot F_\ell^*(0)\}
+\E(R_{\ell,3}-R_{\ell,3}^*).
\end{align*}
The first two derivatives of $q_{\ell r}^{(\diamond)}(t)$ are
\begin{align*}
\dot q_{\ell r}^{(\diamond)}(t)
&=\frac2{\sqrt p}\left[
(\mathbf S_{\ell r}+\mathbf h_r)^\top\mathbf U_{\ell r}^{(\diamond)}
-\E\{\mathbf S_{\ell r}^\top\mathbf U_{\ell r}^{(\diamond)}\}
+t\left\{\|\mathbf U_{\ell r}^{(\diamond)}\|_2^2
-\E\|\mathbf U_{\ell r}^{(\diamond)}\|_2^2\right\}
\right],\\
\ddot q_{\ell r}^{(\diamond)}(t)
&=\frac2{\sqrt p}
\left\{\|\mathbf U_{\ell r}^{(\diamond)}\|_2^2
-\E\|\mathbf U_{\ell r}^{(\diamond)}\|_2^2\right\}.
\end{align*}
Lemmas~\ref{lem:vector-basic} and~\ref{lem:bilinear-quadratic-bound} give,
uniformly over $\diamond\in\{0,*\}$ and $0\le t\le1$, $\|\dot q_{\ell r}^{(\diamond)}(t)\|_4
\le Cw_{\ell r}(1+H)$ and
$\|\ddot q_{\ell r}^{(\diamond)}(t)\|_4
\le Cw_{\ell r}^2$. Let $\mathbf q_{\ell,0}=(q_{\ell r}^{(0)}(0):r\le m)$, which is also
$(q_{\ell r}^{(*)}(0):r\le m)$.  The chain rule gives
\begin{align*}
\dot F_\ell^{(\diamond)}(0)
&=\sum_{r=1}^m
f_r(\mathbf q_{\ell,0})\dot q_{\ell r}^{(\diamond)}(0),\\
\ddot F_\ell^{(\diamond)}(0)
&=\sum_{r,u=1}^m
f_{ru}(\mathbf q_{\ell,0})
\dot q_{\ell r}^{(\diamond)}(0)
\dot q_{\ell u}^{(\diamond)}(0)
+\sum_{r=1}^m
f_r(\mathbf q_{\ell,0})\ddot q_{\ell r}^{(\diamond)}(0).
\end{align*}

To make the remainder $R_{\ell,3}$ and $R_{\ell,3}^*$ explicit, define $\mathbf q_\ell^{(\diamond)}(t)=(q_{\ell r}^{(\diamond)}(t):r\le m)$. Because $q_{\ell r}^{(\diamond)}(t)$ is quadratic in $t$,
$d^3q_{\ell r}^{(\diamond)}(t)/dt^3=0$.  Differentiating the chain-rule
formula once more therefore gives
\begin{align*}
\frac{d^3}{dt^3}F_\ell^{(\diamond)}(t)
={}&\sum_{r,u,v=1}^m
f_{ruv}\{\mathbf q_\ell^{(\diamond)}(t)\}
\dot q_{\ell r}^{(\diamond)}(t)
\dot q_{\ell u}^{(\diamond)}(t)
\dot q_{\ell v}^{(\diamond)}(t)+3\sum_{r,u=1}^m
f_{ru}\{\mathbf q_\ell^{(\diamond)}(t)\}
\ddot q_{\ell r}^{(\diamond)}(t)
\dot q_{\ell u}^{(\diamond)}(t).
\end{align*}
The integral remainder is
\[
R_{\ell,3}^{(\diamond)}
=\frac12\int_0^1(1-t)^2
\frac{d^3}{dt^3}F_\ell^{(\diamond)}(t)\,dt,
\qquad
R_{\ell,3}=R_{\ell,3}^{(0)},
\quad
R_{\ell,3}^*=R_{\ell,3}^{(*)}.
\]
For the first type of product, H\"older inequality and monotonicity of
$L_q$ norms give
\begin{align*}
&\E\left|
\dot q_{\ell r}^{(\diamond)}(t)
\dot q_{\ell u}^{(\diamond)}(t)
\dot q_{\ell v}^{(\diamond)}(t)
\right|\le
\|\dot q_{\ell r}^{(\diamond)}(t)\|_3
\|\dot q_{\ell u}^{(\diamond)}(t)\|_3
\|\dot q_{\ell v}^{(\diamond)}(t)\|_3
\le Cw_{\ell r}w_{\ell u}w_{\ell v}(1+H)^3.
\end{align*}
For the second type, Cauchy--Schwarz gives
\[
\E\left|
\ddot q_{\ell r}^{(\diamond)}(t)
\dot q_{\ell u}^{(\diamond)}(t)
\right|
\le
\|\ddot q_{\ell r}^{(\diamond)}(t)\|_2
\|\dot q_{\ell u}^{(\diamond)}(t)\|_2
\le Cw_{\ell r}^2w_{\ell u}(1+H).
\]
Using the preceding derivative-moment bounds, the fact that $m$ is fixed,
and $(1+H)\ge1$, and applying the estimates for
$\diamond\in\{0,*\}$, we obtain
\[
\E|R_{\ell,3}|+\E|R_{\ell,3}^*|
\le C_m\|f\|_{C^3}
\max_{r\le m}w_{\ell r}^3(1+H)^3.
\]

It remains to control the first two Taylor terms.  Let
\(
\mathcal G_{\ell,+}
=\sigma(\mathcal B_j:j>\ell)
 \vee\sigma(\mathcal B_j^*:j<\ell).
\)
The copied past blocks are independent of both the current block and the
future original blocks.  Adjoining them to the future $\sigma$-field therefore
does not increase the strong-mixing coefficient, and
\(
\alpha\{\sigma(\mathcal B_\ell),\mathcal G_{\ell,+}\}
\le\alpha_{\epsilon,p}(s_n).
\)
{Moreover, $\mathbf q_{\ell,0}$ and $\mathbf S_{\ell r}$ are $\mathcal G_{\ell,+}$-measurable.} 

Define $\zeta_{\ell r}^{(\diamond)}
=(\mathbf S_{\ell r}+\mathbf h_r)^\top
  \mathbf U_{\ell r}^{(\diamond)}$, $
\mu_{\ell r}^{(\diamond)}
=\E\{\mathbf S_{\ell r}^\top
  \mathbf U_{\ell r}^{(\diamond)}\}$ and $
\xi_{\ell r}^{(\diamond)}
=\|\mathbf U_{\ell r}^{(\diamond)}\|_2^2
  -\E\|\mathbf U_{\ell r}^{(\diamond)}\|_2^2$. At point $t=0$,
\[
\dot q_{\ell r}^{(\diamond)}(0)
=\frac2{\sqrt p}\{\zeta_{\ell r}^{(\diamond)}
-\mu_{\ell r}^{(\diamond)}\},
\qquad
\ddot q_{\ell r}^{(\diamond)}(0)
=\frac2{\sqrt p}\xi_{\ell r}^{(\diamond)}.
\]
Since
\(
\sigma\!(\mathbf U_{\ell r}^{(*)}:r\le m)
\ \perp\!\!\!\perp\
\mathcal G_{\ell,+}.
\)
Consequently, $\mu_{\ell r}^{(*)}=0$.

We first treat the first-order Taylor difference.  The chain rule gives
\begin{align*}
\E\{\dot F_\ell(0)-\dot F_\ell^*(0)\}
=\frac2{\sqrt p}\sum_{r=1}^m
\left[
\E\{f_r(\mathbf q_{\ell,0})
(\zeta_{\ell r}^{(0)}-\zeta_{\ell r}^{(*)})\}
-\mu_{\ell r}^{(0)}\E f_r(\mathbf q_{\ell,0})
\right].
\end{align*}
Expanding the first expectation coordinatewise gives
\begin{align*}
&\E\{f_r(\mathbf q_{\ell,0})
(\zeta_{\ell r}^{(0)}-\zeta_{\ell r}^{(*)})\}\\
&\quad=\sum_{a=1}^p\sum_{i\in I_\ell}a_{i,r}
\Big[
\E\{f_r(\mathbf q_{\ell,0})(S_{\ell r,a}+h_{r,a})
\epsilon_{i,a}\}
-
\E\{f_r(\mathbf q_{\ell,0})(S_{\ell r,a}+h_{r,a})
\epsilon_{i,a}^{(\ell,*)}\}
\Big],
\end{align*}
where the second expectation is zero.  Since
\(
\mu_{\ell r}^{(0)}
=\sum_{a=1}^p\sum_{i\in I_\ell}a_{i,r}
\E(S_{\ell r,a}\epsilon_{i,a}),
\)
define the $\mathcal G_{\ell,+}$-measurable variable
\[
Y_{\ell r,a}^{(1)}
=f_r(\mathbf q_{\ell,0})(S_{\ell r,a}+h_{r,a})
-\E f_r(\mathbf q_{\ell,0})S_{\ell r,a}.
\]
The entire first-order difference is therefore
\[
\E\{\dot F_\ell(0)-\dot F_\ell^*(0)\}
=\frac2{\sqrt p}
\sum_{r=1}^m\sum_{a=1}^p\sum_{i\in I_\ell}a_{i,r}
\Cov(\epsilon_{i,a},Y_{\ell r,a}^{(1)}).
\]
For each fixed $r\le m$ and $a\le p$, Lemma~\ref{lem:vector-basic}(ii)
gives
\(
\|U_{r,a}^{I_j,*}\|_4\le Cw_{jr},
\|\sum_{j>\ell}U_{r,a}^{I_j}\|_4\le C.
\)
Since
$U_{r,a}^{I_1,*},\ldots,U_{r,a}^{I_{\ell-1},*}$
are mutually independent, Rosenthal's inequality at order $4$, together
with $\sum_jw_{jr}^2\le C$, gives
\(
\|\sum_{j<\ell}U_{r,a}^{I_j,*}\|_4\le C.
\)
Consequently,
\[
\sup_{i,a}\|\epsilon_{i,a}\|_4
+
\sup_{\ell,r,a}\|S_{\ell r,a}\|_4\le C,
\qquad
\|Y_{\ell r,a}^{(1)}\|_4
\le C\|f\|_{C^1}(1+|h_{r,a}|).
\]
For $i\in I_\ell$, $\epsilon_{i,a}$ is
$\sigma(\mathcal B_\ell)$-measurable, whereas
$Y_{\ell r,a}^{(1)}$ is $\mathcal G_{\ell,+}$-measurable, and
\[
\alpha\{\sigma(\mathcal B_\ell),\mathcal G_{\ell,+}\}
\le\alpha_{\epsilon,p}(s_n).
\]
Davydov inequality with the two moment exponents equal to $4$ therefore gives
\begin{align*}
\left|\Cov(\epsilon_{i,a},Y_{\ell r,a}^{(1)})\right|
&\le
8\alpha_{\epsilon,p}(s_n)^{1-1/4-1/4}
\|\epsilon_{i,a}\|_4\|Y_{\ell r,a}^{(1)}\|_4\\
&\le
C\|f\|_{C^1}(1+|h_{r,a}|)
\alpha_{\epsilon,p}(s_n)^{1/2}.
\end{align*}
Together with
\[
\sum_{i\in I_\ell}|a_{i,r}|
\le\sqrt{b_n}\,w_{\ell r}\le Cb_n/\sqrt n,
\qquad
\sum_{a=1}^p(1+|h_{r,a}|)\le p(1+H),
\]
this yields
\begin{equation}
\left|\E\{\dot F_\ell(0)-\dot F_\ell^*(0)\}\right|
\le C_m\|f\|_{C^1}b_n\sqrt{\frac pn}\,(1+H)
\alpha_{\epsilon,p}(s_n)^{1/2}.
\label{eq:fd-AM-first-Taylor-bound}
\end{equation}

For the second-order Taylor difference, direct substitution gives
\begin{align*}
\E\{\ddot F_\ell(0)-\ddot F_\ell^*(0)\}
&=\frac4p\sum_{r,u=1}^m
\E\Big[f_{ru}(\mathbf q_{\ell,0})
\Big\{(\zeta_{\ell r}^{(0)}-\mu_{\ell r}^{(0)})
      (\zeta_{\ell u}^{(0)}-\mu_{\ell u}^{(0)})
      -\zeta_{\ell r}^{(*)}\zeta_{\ell u}^{(*)}\Big\}\Big]\\
&\qquad\quad+
\frac2{\sqrt p}\sum_{r=1}^m
\E\{f_r(\mathbf q_{\ell,0})
(\xi_{\ell r}^{(0)}-\xi_{\ell r}^{(*)})\},
\end{align*}
with
\begin{align*}
\zeta_{\ell r}^{(0)}\zeta_{\ell u}^{(0)}
={}&\sum_{a,b=1}^p\sum_{i,j\in I_\ell}
a_{i,r}a_{j,u}
(S_{\ell r,a}+h_{r,a})(S_{\ell u,b}+h_{u,b})
\epsilon_{i,a}\epsilon_{j,b},\\
\xi_{\ell r}^{(0)}
={}&\sum_{a=1}^p\sum_{i,j\in I_\ell}a_{i,r}a_{j,r}
\{\epsilon_{i,a}\epsilon_{j,a}
-\E(\epsilon_{i,a}\epsilon_{j,a})\},
\end{align*}
and the analogous copied formulas.  We have
\begin{align*}
& (\zeta_{\ell r}^{(0)}-\mu_{\ell r}^{(0)})
  (\zeta_{\ell u}^{(0)}-\mu_{\ell u}^{(0)})
  -\zeta_{\ell r}^{(*)}\zeta_{\ell u}^{(*)}\\
&\quad=
\left\{
\zeta_{\ell r}^{(0)}\zeta_{\ell u}^{(0)}
-\zeta_{\ell r}^{(*)}\zeta_{\ell u}^{(*)}
\right\}
-\mu_{\ell r}^{(0)}\zeta_{\ell u}^{(0)}
-\mu_{\ell u}^{(0)}\zeta_{\ell r}^{(0)}
+\mu_{\ell r}^{(0)}\mu_{\ell u}^{(0)}.
\end{align*}
Let
\begin{align*}
X_{\ell,ijab}
&=\epsilon_{i,a}\epsilon_{j,b}
-\E(\epsilon_{i,a}\epsilon_{j,b}),\qquad
X_{\ell,ijab}^{*}
=
\epsilon_{i,a}^{(\ell,*)}\epsilon_{j,b}^{(\ell,*)}
-\E\{\epsilon_{i,a}^{(\ell,*)}\epsilon_{j,b}^{(\ell,*)}\}\\
Y_{\ell,r}^{(2,\xi)}
&=f_r(\mathbf q_{\ell,0}),\qquad
Y_{\ell,ruab}^{(2,\zeta)}
=f_{ru}(\mathbf q_{\ell,0})
(S_{\ell r,a}+h_{r,a})(S_{\ell u,b}+h_{u,b}).
\end{align*}
Then $X_{\ell,ijab}^{*}\overset d=X_{\ell,ijab}$ and
$X_{\ell,ijab}^{*}$ is independent of $\mathcal G_{\ell,+}$.  Hence, for
every $\mathcal G_{\ell,+}$-measurable random variable $Y$,
\(
\E\{Y(X_{\ell,ijab}-X_{\ell,ijab}^{*})\}
=
\Cov(X_{\ell,ijab},Y).
\)
Then,
\begin{align*}
\E\left[
f_{ru}(\mathbf q_{\ell,0})
\left\{
\zeta_{\ell r}^{(0)}\zeta_{\ell u}^{(0)}
-\zeta_{\ell r}^{(*)}\zeta_{\ell u}^{(*)}
\right\}
\right]&=
\sum_{a,b=1}^p\sum_{i,j\in I_\ell}
a_{i,r}a_{j,u}
\Cov\{X_{\ell,ijab},Y_{\ell,ruab}^{(2,\zeta)}\},\\
\E\left[
f_r(\mathbf q_{\ell,0})
\{\xi_{\ell r}^{(0)}-\xi_{\ell r}^{(*)}\}
\right]&=
\sum_{a=1}^p\sum_{i,j\in I_\ell}
a_{i,r}a_{j,r}
\Cov\{X_{\ell,ijaa},Y_{\ell,r}^{(2,\xi)}\}.
\end{align*}
For each fixed $1\le r\le m$ and $1\le a\le p$, Lemma~\ref{lem:vector-basic}(ii),
with $q=8$, gives
\(
\|U_{r,a}^{I_j,*}\|_8\le Cw_{jr},
\|\sum_{j>\ell}U_{r,a}^{I_j}\|_8\le C.
\)
Since
$U_{r,a}^{I_1,*},\ldots,U_{r,a}^{I_{\ell-1},*}$
are mutually independent, Rosenthal's inequality at order $8$, together
with $\sum_jw_{jr}^2\le C$, gives
\(
\|\sum_{j<\ell}U_{r,a}^{I_j,*}\|_8\le C.
\)
The same lemma, applied to the current original and copied blocks, gives
\[
\|U_{\ell r,a}^{(0)}\|_8
+\|U_{\ell r,a}^{(*)}\|_8
\le Cw_{\ell r}.
\]
Thus
\[
\sup_{i,a}\|\epsilon_{i,a}\|_8
+
\sup_{\ell,r,a}\|S_{\ell r,a}\|_8\le C,
\qquad
\|U_{\ell r,a}^{(0)}\|_8
+\|U_{\ell r,a}^{(*)}\|_8
\le Cw_{\ell r}.
\]
Consequently,
\begin{align*}
\|X_{\ell,ijab}\|_4\le C,\quad
\|Y_{\ell,r}^{(2,\xi)}\|_4\le \|f\|_{C^1},\quad
\|Y_{\ell,ruab}^{(2,\zeta)}\|_4
\le C\|f\|_{C^2}
(1+|h_{r,a}|)(1+|h_{u,b}|).
\end{align*}
Davydov inequality therefore yields
\begin{align*}
\left|\Cov\{X_{\ell,ijab},Y_{\ell,ruab}^{(2,\zeta)}\}\right|
&\le C\|f\|_{C^2}(1+|h_{r,a}|)(1+|h_{u,b}|)
\alpha_{\epsilon,p}(s_n)^{1/2},\\
\left|\Cov\{X_{\ell,ijaa},Y_{\ell,r}^{(2,\xi)}\}\right|
&\le C\|f\|_{C^1}\alpha_{\epsilon,p}(s_n)^{1/2}.
\end{align*}

Now, we consider the last three terms in the product expansion.
Since
\[
\mu_{\ell r}^{(0)}
=\sum_{a=1}^p
\Cov(S_{\ell r,a},U_{\ell r,a}^{(0)}),
\]
Davydov inequality and the preceding moment bounds imply
\[
|\mu_{\ell r}^{(0)}|
\le C\alpha_{\epsilon,p}(s_n)^{1/2}
\sum_{a=1}^p
\|S_{\ell r,a}\|_4\|U_{\ell r,a}^{(0)}\|_4
\le Cp w_{\ell r}\alpha_{\epsilon,p}(s_n)^{1/2}.
\]
Moreover,
\begin{align*}
\E|\zeta_{\ell r}^{(0)}|
\le\sum_{a=1}^p
(\|S_{\ell r,a}\|_2+|h_{r,a}|)
\|U_{\ell r,a}^{(0)}\|_2
\le Cw_{\ell r}
\left(p+\sum_{a=1}^p|h_{r,a}|\right)
\le Cp w_{\ell r}(1+H).
\end{align*}
Thus, uniformly in $r,u\le m$,
\begin{align*}
&\frac1p\left|
\E\left[f_{ru}(\mathbf q_{\ell,0})
\left\{-\mu_{\ell r}^{(0)}\zeta_{\ell u}^{(0)}
-\mu_{\ell u}^{(0)}\zeta_{\ell r}^{(0)}
+\mu_{\ell r}^{(0)}\mu_{\ell u}^{(0)}\right\}
\right]\right|\\
&\qquad\le C\|f\|_{C^2}
 p w_{\ell r}w_{\ell u}(1+H)^2
\alpha_{\epsilon,p}(s_n)^{1/2},
\end{align*}
where $\alpha_{\epsilon,p}(s_n)\le1$ is used for the product of the two
$\mu$ terms.  Combining these estimates and summing the at most
$C_m(b_n p)^2$ main scalar terms gives, for a fixed finite $C_*$,
\begin{equation}
\left|\E\{\ddot F_\ell(0)-\ddot F_\ell^*(0)\}\right|
\le C_m\|f\|_{C^2}(1+H)^2(np)^{C_*}
\alpha_{\epsilon,p}(s_n)^{1/2}.
\label{eq:fd-AM-second-Taylor-bound}
\end{equation}
Consequently,
\begin{align*}
\left|\E\{\dot F_\ell(0)-\dot F_\ell^*(0)\}\right|
+
\left|\E\{\ddot F_\ell(0)-\ddot F_\ell^*(0)\}\right|
\le
C_m\|f\|_{C^2}(1+H)^2
(np)^{C_*}\alpha_{\epsilon,p}(s_n)^{1/2}.
\end{align*}
Since $b_n\le n$ and
$\alpha_{\epsilon,p}(s_n)^{1/2}\le C e^{-cs_n^{\varkappa}}$, increasing
$c_0=c_0(A)$ so that
\(
cc_0^{\varkappa}\ge A+C_*+3
\)
ensures that the last display is bounded by
\(
C_m\|f\|_{C^2}(1+H)^2(np)^{-A-2}.
\)

Finally, $L\le n$ and the block-weight bound gives
\(
\sum_{\ell=1}^L\max_{r\le m}w_{\ell r}^3
\le C\left(b_n/n\right)^{1/2}.
\)
Summing the Taylor identities over $\ell$ therefore gives \eqref{eq:fd-AM-block-independence}.

(b) Under Assumption~\ref{ass:C1}(MD), the retained blocks are already
independent once $s_n>m_0$.

(c) Under Assumption~\ref{ass:C1}(PD), recall that
\(
\bepsilon_i
=\mathcal G_p(\boldsymbol\xi_i,\boldsymbol\xi_{i-1},\ldots),
\)
where $\{\boldsymbol\xi_i:i\in\mathbb Z\}$ are i.i.d. random vectors.
For the above $s_n$, define
\begin{equation}
\mathcal H_{i,s_n}=\sigma(\boldsymbol\xi_i,\ldots,\boldsymbol\xi_{i-s_n}),
\qquad
\bepsilon_i^{(s_n)}=\E(\bepsilon_i\mid\mathcal H_{i,s_n}),
\qquad
\mathbf R_i^{(s_n)}=\bepsilon_i-\bepsilon_i^{(s_n)}.\label{eq:physical-dependent-truncation}
\end{equation}
The process $\{\bepsilon_i^{(s_n)}\}$ is $s_n$-dependent.  
Let $\mathcal I=I_1\cup\cdots\cup I_L$ and let
\begin{align*}
\mathbf U_r^{I_\ell,(s_n)}
&=\sum_{i\in I_\ell}a_{i,r}\bepsilon_i^{(s_n)},
\qquad
\mathbf U_r^{I,(s_n)}
=\sum_{\ell=1}^L\mathbf U_r^{I_\ell,(s_n)},\\
\mathbf D_r^{I,(s_n)}
&=\sum_{\ell=1}^L\sum_{i\in I_\ell}
a_{i,r}\mathbf R_i^{(s_n)}
=\mathbf U_r^I-\mathbf U_r^{I,(s_n)}.
\end{align*}
The vectors
$\{(\mathbf U_r^{I_\ell,(s_n)}:r\le m):1\le\ell\le L\}$ are mutually
independent.
Define
\[
Q_{n,p}^{I,(s_n)}(t_r;\mathbf h_r)
=
\frac{\|\mathbf U_r^{I,(s_n)}+\mathbf h_r\|_2^2
-\E\|\mathbf U_r^{I,(s_n)}\|_2^2}{\sqrt p}.
\]
Then
\begingroup
\begin{align}
&Q_{n,p}^{I}(t_r;\mathbf h_r)
-Q_{n,p}^{I,(s_n)}(t_r;\mathbf h_r)\notag\\
&\quad=
\frac2{\sqrt p}\left[
(\mathbf U_r^{I,(s_n)}+\mathbf h_r)^\top\mathbf D_r^{I,(s_n)}
-\E\{(\mathbf U_r^{I,(s_n)}+\mathbf h_r)^\top
       \mathbf D_r^{I,(s_n)}\}
\right]\notag\\
&\qquad+\frac1{\sqrt p}\left[
\|\mathbf D_r^{I,(s_n)}\|_2^2
-\E\|\mathbf D_r^{I,(s_n)}\|_2^2
\right].
\label{eq:physical-residual-decomposition}
\end{align}
\endgroup

Let
$\{\boldsymbol\xi_j':j\in\mathbb Z\}$ be an independent copy of
$\{\boldsymbol\xi_j:j\in\mathbb Z\}$, and define
\[
\tilde{\bepsilon}_i^{(s_n)}
=
\mathcal G_p(
\boldsymbol\xi_i,\ldots,\boldsymbol\xi_{i-s_n},
\boldsymbol\xi_{i-s_n-1}',\boldsymbol\xi_{i-s_n-2}',\ldots).
\]
Set $\bepsilon_i^{[s_n]}=\bepsilon_i$.  For each $r>s_n$, let
$\bepsilon_i^{[r]}$ be obtained from $\bepsilon_i^{[r-1]}$ by replacing
$\boldsymbol\xi_{i-r}$ with $\boldsymbol\xi_{i-r}'$.  Then
$\bepsilon_i^{[r]}\to\tilde{\bepsilon}_i^{(s_n)}$ in $L^q$.
Writing
$\mathcal F_i=\sigma(\boldsymbol\xi_i,\boldsymbol\xi_{i-1},\ldots)$, we have
\[
\mathbf R_i^{(s_n)}
=
\E\{\bepsilon_i-\widetilde{\bepsilon}_i^{(s_n)}
\mid\mathcal F_i\}.
\]
Hence conditional Jensen's inequality, Minkowski's inequality, and
Assumption~\ref{ass:C1}(PD) give, for every $q\ge2$,
\begin{align*}
\sup_{\|\mathbf u\|_2=1}
\|\mathbf u^\top\mathbf R_i^{(s_n)}\|_q
&\le
\sup_{\|\mathbf u\|_2=1}
\|\mathbf u^\top
(\bepsilon_i-\tilde{\bepsilon}_i^{(s_n)})\|_q\\
&\le
\sum_{r=s_n+1}^{\infty}
\sup_{\|\mathbf u\|_2=1}
\|\mathbf u^\top
(\bepsilon_i^{[r-1]}-\bepsilon_i^{[r]})\|_q\\
&\le
Cq^{1/2}\sum_{r=s_n+1}^{\infty}e^{-cr^{\varkappa}}
\le Cq^{1/2}e^{-cs_n^{\varkappa}}.
\end{align*}
\begingroup
To obtain the weighted residual bound, let $\mathbf R_i^{(s_n),\{r\}}$
denote the residual with innovation $\boldsymbol\xi_{i-r}$ replaced,
and let $\mathcal P_jZ=\E(Z\mid\mathcal F_j)-\E(Z\mid\mathcal F_{j-1})$.
Conditional Jensen's inequality and the preceding coupling give
\begin{align*}
\delta_{R,q}(r)
:=\sup_{\|\mathbf u\|_2=1}
\|\mathbf u^\top(\mathbf R_i^{(s_n)}-
                         \mathbf R_i^{(s_n),\{r\}})\|_q
\le Cq^{1/2}\begin{cases}
e^{-cs_n^{\varkappa}},&0\le r\le s_n,\\
e^{-cr^{\varkappa}},&r>s_n,
\end{cases}\\
\sum_{r\ge0}\delta_{R,q}(r)\le Cq^{1/2}e^{-c's_n^{\varkappa}},\qquad
\|\mathcal P_{i-r}(\mathbf u^\top\mathbf R_i^{(s_n)})\|_q
\le\|\mathbf u\|_2\delta_{R,q}(r).
\end{align*}
For $r>s_n$,
\begin{align*}
\mathbf R_i^{(s_n),\{r\}}
=\bepsilon_i^{\{r\}}-\bepsilon_i^{(s_n)},\qquad
\mathbf R_i^{(s_n)}-\mathbf R_i^{(s_n),\{r\}}
=\bepsilon_i-\bepsilon_i^{\{r\}}.
\end{align*}
For $0\le r\le s_n$ and $\|\mathbf u\|_2=1$,
\begin{align*}
\|\mathbf u^\top(\mathbf R_i^{(s_n)}
                    -\mathbf R_i^{(s_n),\{r\}})\|_q
&\le\|\mathbf u^\top\mathbf R_i^{(s_n)}\|_q
    +\|\mathbf u^\top\mathbf R_i^{(s_n),\{r\}}\|_q\\
&=2\|\mathbf u^\top\mathbf R_i^{(s_n)}\|_q
\le C\sqrt q\,e^{-cs_n^{\varkappa}}.
\end{align*}
For every finitely supported deterministic $d$ and
$\mathbf u\in\mathbb R^p$,
\begin{align*}
\left\|\sum_i d_i\mathbf u^\top\mathbf R_i^{(s_n)}\right\|_q
&\le\sum_{r\ge0}
\left\|\sum_i d_i\mathcal P_{i-r}
                (\mathbf u^\top\mathbf R_i^{(s_n)})\right\|_q\\
&\le\sqrt{q-1}\sum_{r\ge0}
\left\{\sum_i d_i^2
\|\mathcal P_{i-r}(\mathbf u^\top\mathbf R_i^{(s_n)})\|_q^2
\right\}^{1/2}\\
&\le C\sqrt q\,\|d\|_2\|\mathbf u\|_2
       \sum_{r\ge0}\delta_{R,q}(r)
\le Cq e^{-c's_n^{\varkappa}}\|d\|_2\|\mathbf u\|_2.
\end{align*}
Since $q\le q^{c_{\mathrm{dep}}+1}$, this proves the bound used below:
\endgroup
\begin{equation}
\left\|\sum_i d_i\mathbf u^\top\mathbf R_i^{(s_n)}\right\|_q
\le Cq^{c_{\mathrm{dep}}+1}e^{-cs_n^{\varkappa}}
\|d\|_2\|\mathbf u\|_2.
\label{eq:weighted-physical-residual}
\end{equation}

Taking $q=2$,
$d_i=a_{i,r}\ind{i\in\mathcal I}$, and $\mathbf u=\mathbf e_j$
in~\eqref{eq:weighted-physical-residual}, and using
$\sum_{i\in\mathcal I}a_{i,r}^2\le C$, gives
\(
\max_{r\le m,\,j\le p}\|D_{r,j}^{I,(s_n)}\|_2
\le Ce^{-cs_n^{\varkappa}}.
\)
Lemma~\ref{lem:vector-basic}(ii) and
$\mathbf U_r^{I,(s_n)}=\mathbf U_r^I-\mathbf D_r^{I,(s_n)}$ also give
\(
\max_{r\le m,\,j\le p}\|U_{r,j}^{I,(s_n)}\|_2\le C.
\)
For the first term in~\eqref{eq:physical-residual-decomposition}, the H\"older
and Minkowski inequalities yield
\begin{align*}
&\left\|
(\mathbf U_r^{I,(s_n)}+\mathbf h_r)^\top\mathbf D_r^{I,(s_n)}
-\E\{(\mathbf U_r^{I,(s_n)}+\mathbf h_r)^\top
       \mathbf D_r^{I,(s_n)}\}
\right\|_1\\
&\quad\le
2\sum_{j=1}^p
\|U_{r,j}^{I,(s_n)}+h_{r,j}\|_2
\|D_{r,j}^{I,(s_n)}\|_2\\
&\quad\le
Ce^{-cs_n^{\varkappa}}\left(p+\|\mathbf h_r\|_1\right)
\le Ce^{-cs_n^{\varkappa}}\left(p+\sqrt p\,\|\mathbf h_r\|_2\right)
\le Cpe^{-cs_n^{\varkappa}}(1+H).
\end{align*}
For the second term,
\begin{align*}
\left\|
\|\mathbf D_r^{I,(s_n)}\|_2^2
-\E\|\mathbf D_r^{I,(s_n)}\|_2^2
\right\|_1
\le 2\E\|\mathbf D_r^{I,(s_n)}\|_2^2
=2\sum_{j=1}^p\|D_{r,j}^{I,(s_n)}\|_2^2
\le Cpe^{-2cs_n^{\varkappa}}.
\end{align*}
Consequently, since $e^{-2cs_n^{\varkappa}}\le e^{-cs_n^{\varkappa}}$,
\[
\max_{r\le m}
\|Q_{n,p}^{I}(t_r;\mathbf h_r)
-Q_{n,p}^{I,(s_n)}(t_r;\mathbf h_r)\|_1
\le C\sqrt p\,e^{-cs_n^{\varkappa}}(1+H).
\]
Since
$s_n=\lceil c_0\{\log(np)\}^{1/\varkappa}\rceil$,
\(
\sqrt p\,e^{-cs_n^{\varkappa}}
\le \sqrt p\,(np)^{-cc_0^{\varkappa}}
\le (np)^{1/2-cc_0^{\varkappa}}.
\)
For every fixed $A>0$, choose $c_0=c_0(A)$ sufficiently large that
\(
cc_0^{\varkappa}\ge A+\frac12,
\)
we get
\begin{align*}
&\left|
\E f\{Q_{n,p}^{I}(t_r;\mathbf h_r):r\le m\}
-
\E f\{Q_{n,p}^{I,(s_n)}(t_r;\mathbf h_r):r\le m\}
\right|\notag\\
&\qquad\le
C_m\|f\|_{C^1}(np)^{-A}(1+H).
\end{align*}
Thus all three branches of Assumption~\ref{ass:C1} lead to independent
retained blocks with the displayed comparison error.

Step 3. 
To cover the three branches of Step~2 with one notation, let
\[
\mathbf V_{\ell r}
=
\begin{cases}
\mathbf U_r^{I_\ell,*},
&\text{under Assumption~\ref{ass:C1}(AM)},\\
\mathbf U_r^{I_\ell},
&\text{under Assumption~\ref{ass:C1}(MD)},\\
\mathbf U_r^{I_\ell,(s_n)},
&\text{under Assumption~\ref{ass:C1}(PD)}.
\end{cases}
\]
For every branch, the stacked block vectors
$\mathbf V_\ell=(\mathbf V_{\ell r}:r\le m)$ are mutually independent and
centered.  Define
\[
\mathbf V_r=\sum_{\ell=1}^L\mathbf V_{\ell r},
\qquad
Q_{n,p}^{I,\mathrm{ind}}(t_r;\mathbf h_r)
=
\frac{\|\mathbf V_r+\mathbf h_r\|_2^2
-\E\|\mathbf V_r\|_2^2}{\sqrt p}.
\]
Thus $Q_{n,p}^{I,\mathrm{ind}}$ equals $Q_{n,p}^{I,*}$ in branch~(AM),
$Q_{n,p}^{I}$ in branch~(MD), and $Q_{n,p}^{I,(s_n)}$ in branch~(PD).

For every $\ell$, let
$\mathbf V_\ell^G=(\mathbf V_{\ell r}^G:r\le m)$ be a centered jointly
Gaussian vector satisfying
\(
\Cov(\mathbf V_{\ell r}^G,\mathbf V_{\ell u}^G)
=
\Cov(\mathbf V_{\ell r},\mathbf V_{\ell u}),
\)
for \(r,u\le m,\) and $\mathbf V_1^G,\ldots,\mathbf V_L^G$ are mutually independent and
independent of 
$\{\mathbf V_1,\ldots,\mathbf V_L\}$.
Define
\[
\mathbf V_r^G=\sum_{\ell=1}^L\mathbf V_{\ell r}^G,
\qquad
Q_{n,p}^{I,G,\mathrm{ind}}(t_r;\mathbf h_r)
=
\frac{\|\mathbf V_r^G+\mathbf h_r\|_2^2
-\E\|\mathbf V_r^G\|_2^2}{\sqrt p}.
\]
The goal of this step is to prove
\begin{align}
&\left|
\E f\{Q_{n,p}^{I,\mathrm{ind}}(t_r;\mathbf h_r):r\le m\}
-
\E f\{Q_{n,p}^{I,G,\mathrm{ind}}(t_r;\mathbf h_r):r\le m\}
\right|\nonumber\\
&\qquad\le
C_m\|f\|_{C^3}
\left(\frac{b_n}{n}\right)^{1/2}(1+H)^3.
\label{eq:fd-block-gaussian-replacement}
\end{align}
\begingroup
Under (PD), we transfer fourth moments from independent copies of the
original blocks; Assumption~\ref{ass:C3} is not asserted for conditional
projections. For each retained block, construct a pair
$(\bepsilon_i^*,\bepsilon_i^{*,(s_n)}:i\in I_\ell)$ from one innovation
sequence, independently across $\ell$, and define
\[
\mathbf Y(x)=\sum_{\ell}\sum_{i\in I_\ell}x_i\bepsilon_i^*,\quad
\mathbf Y_s(x)=\sum_{\ell}\sum_{i\in I_\ell}x_i\bepsilon_i^{*,(s_n)},\quad
\mathbf D_s(x)=\mathbf Y(x)-\mathbf Y_s(x).
\]
The weighted residual bound at order eight, and Rosenthal's inequality
for the independent centered block sums, give
\begin{align*}
\max_{a\le p}\|D_{s,a}(x)\|_8
&\le C e^{-cs_n^{\varkappa}}
 \left[\left\{\sum_\ell\|x\mathbf1_{I_\ell}\|_2^2\right\}^{1/2}
       +\left\{\sum_\ell\|x\mathbf1_{I_\ell}\|_2^8\right\}^{1/8}\right]
\le C e^{-cs_n^{\varkappa}}\|x\|_2,\\
\max_{a\le p}\{\|Y_a(x)\|_8+\|Y_{s,a}(x)\|_8\}
&\le C\|x\|_2.
\end{align*}
For any vector array $\mathbf Z(x)$ linear in $x$, let
\[
\mathcal B_{\mathbf Z}(x,y)
=p^{-1/2}\{\mathbf Z(x)^\top\mathbf Z(y)
                   -\E[\mathbf Z(x)^\top\mathbf Z(y)]\}.
\]
The identity
\[
\mathbf Y(x)^\top\mathbf Y(y)-\mathbf Y_s(x)^\top\mathbf Y_s(y)
=\mathbf D_s(x)^\top\mathbf Y(y)
 +\mathbf Y_s(x)^\top\mathbf D_s(y)
\]
and coordinatewise H\"older give
\begin{align*}
\|\mathcal B_{\mathbf Y}(x,y)-\mathcal B_{\mathbf Y_s}(x,y)\|_4
&\le\frac2{\sqrt p}\sum_{a=1}^p
 \{\|D_{s,a}(x)\|_8\|Y_a(y)\|_8
  +\|Y_{s,a}(x)\|_8\|D_{s,a}(y)\|_8\}\\
&\le C\sqrt p\,e^{-cs_n^{\varkappa}}\|x\|_2\|y\|_2.
\end{align*}

For the Gaussian transfer, stack the retained observations, let
$\boldsymbol\Sigma$ and $\boldsymbol\Sigma_s$ be the covariance
matrices of their independent original and projected copies, and use
one $\mathbf z\sim N(\mathbf0,\mathbf I)$ to define
\[
\mathbf g=\boldsymbol\Sigma^{1/2}\mathbf z,\qquad
\mathbf g_s=\boldsymbol\Sigma_s^{1/2}\mathbf z.
\]
The preceding order-two residual bound and covariance summability imply
\begin{align*}
\|\boldsymbol\Sigma-\boldsymbol\Sigma_s\|_{\mathrm{op}}
&\le Cn e^{-cs_n^{\varkappa}},\qquad
\|\boldsymbol\Sigma\|_{\mathrm{op}}\le C,\\
\|\boldsymbol\Sigma_s\|_{\mathrm{op}}
&\le\|\boldsymbol\Sigma\|_{\mathrm{op}}
+\|\boldsymbol\Sigma_s-\boldsymbol\Sigma\|_{\mathrm{op}}
\le C\{1+n e^{-cs_n^{\varkappa}}\}\le C,\\
\|\boldsymbol\Sigma^{1/2}-\boldsymbol\Sigma_s^{1/2}\|_{\mathrm{op}}
&\le C\|\boldsymbol\Sigma-\boldsymbol\Sigma_s\|_{\mathrm{op}}^{1/2}
\le C\sqrt n\,e^{-c's_n^{\varkappa}}.
\end{align*}
Let $\mathbf Y_G(x)=\sum_i x_i\mathbf g_i$ and
$\mathbf Y_{G,s}(x)=\sum_i x_i\mathbf g_{s,i}$. Gaussian moments give
\begin{align}
\max_a\|Y_{G,a}(x)-Y_{G,s,a}(x)\|_8
&\le C\sqrt n\,e^{-c's_n^{\varkappa}}\|x\|_2,\notag\\
\|\mathcal B_{\mathbf Y_G}(x,y)-
       \mathcal B_{\mathbf Y_{G,s}}(x,y)\|_4
&\le C\sqrt{np}\,e^{-c's_n^{\varkappa}}\|x\|_2\|y\|_2.
\label{eq:fd-PD-gaussian-bilinear-transfer}
\end{align}
Choose $\mathbf z$ independently of the copied innovation sequences.
For any $\mathcal L_G\subseteq\{1,\ldots,L\}$, define, for $i\in I_\ell$,
\[
(\widetilde{\bepsilon}_i,\widetilde{\bepsilon}_i^{(s_n)})
=\begin{cases}
(\bepsilon_i^*,\bepsilon_i^{*,(s_n)}),&\ell\notin\mathcal L_G,\\
(\mathbf g_i,\mathbf g_{s,i}),&\ell\in\mathcal L_G.
\end{cases}
\]
For $2\le r\le8$, $i_1,\ldots,i_r\in\bigcup_{\ell=1}^L I_\ell$,
and $j_1,\ldots,j_r\le p$,
\[
\begin{aligned}
\Cum(\widetilde\epsilon_{i_1,j_1},\ldots,
      \widetilde\epsilon_{i_r,j_r})
=\begin{cases}
\Cum(\epsilon_{i_1,j_1},\ldots,\epsilon_{i_r,j_r}),
 &i_1,\ldots,i_r\in I_\ell,\ \ell\notin\mathcal L_G,\\
\Cov(\epsilon_{i_1,j_1},\epsilon_{i_2,j_2}),
 &r=2,\ i_1,i_2\in I_\ell,\ \ell\in\mathcal L_G,\\
0,&\text{otherwise}.
\end{cases}
\end{aligned}
\]
In particular,
\[
\left|\Cum(\widetilde\epsilon_{i_1,j_1},\ldots,
           \widetilde\epsilon_{i_r,j_r})\right|
\le\left|\Cum(\epsilon_{i_1,j_1},\ldots,\epsilon_{i_r,j_r})\right|.
\]
Writing
$\mathbf Y(x)=\sum_i x_i\widetilde{\bepsilon}_i$ and
$\mathbf Y_s(x)=\sum_i x_i\widetilde{\bepsilon}_i^{(s_n)}$,
the preceding coupling bounds and the absolute contractions in
Assumption~\ref{ass:C3} therefore give, uniformly in $\mathcal L_G$,
\begin{align}
\|\mathcal B_{\mathbf Y}(x,y)\|_4&\le C\|x\|_2\|y\|_2,\notag\\
\|\mathcal B_{\mathbf Y_s}(x,y)\|_4
&\le C\{1+\sqrt{np}e^{-c's_n^{\varkappa}}\}\|x\|_2\|y\|_2
\le C\|x\|_2\|y\|_2.
\label{eq:fd-hybrid-bilinear-moment}
\end{align}
The last inequality uses
$\sqrt{np}e^{-c's_n^{\varkappa}}
\le(np)^{1/2-c'c_0^{\varkappa}}\le1$, after increasing $c_0$.
This establishes the fourth moments required for every interpolation
below.

\endgroup
Following the blockwise replacement technique established in Step~2(a), define
\[
\mathbf T_{\ell r}
:=\sum_{j<\ell}\mathbf V_{jr}^G
 +\sum_{j>\ell}\mathbf V_{jr},
\qquad
\mathbf U_{\ell r}^{(0)}=\mathbf V_{\ell r},
\qquad
\mathbf U_{\ell r}^{(G)}=\mathbf V_{\ell r}^G.
\]
For $\diamond\in\{0,G\}$ and $0\le t\le1$, define
\[
q_{\ell r}^{(\diamond)}(t)
=
\frac{
\|\mathbf T_{\ell r}+t\mathbf U_{\ell r}^{(\diamond)}
+\mathbf h_r\|_2^2
-
\E\|\mathbf T_{\ell r}+t\mathbf U_{\ell r}^{(\diamond)}\|_2^2
}{\sqrt p},
\]
and
\[
F_\ell^{(\diamond)}(t)
=f\{q_{\ell1}^{(\diamond)}(t),\ldots,
q_{\ell m}^{(\diamond)}(t)\}.
\]
For $t=0$, we have
\[
\mathbf q_{\ell,0}
:=\bigl(q_{\ell r}^{(0)}(0):r\le m\bigr)
=\bigl(q_{\ell r}^{(G)}(0):r\le m\bigr),
\qquad
F_\ell^{(0)}(0)=F_\ell^{(G)}(0).
\]
It therefore suffices to prove, for every $\ell\le L$,
\[
\left|
\E F_\ell^{(0)}(1)-\E F_\ell^{(G)}(1)
\right|
\le
C_m\|f\|_{C^3}
\max_{r\le m}w_{\ell r}^3(1+H)^3.
\]

Let
\(
\mathcal H_\ell
=\sigma(\mathbf V_j,\mathbf V_j^G:j\ne\ell).
\)
Then $\mathbf T_{\ell r}$ and $\mathbf q_{\ell,0}$ are
$\mathcal H_\ell$-measurable, while, for $\diamond\in\{0,G\}$,
\[
\sigma(\mathbf U_{\ell r}^{(\diamond)}:r\le m)
\ \perp\!\!\!\perp\ \mathcal H_\ell,
\qquad
\E\mathbf U_{\ell r}^{(\diamond)}=\mathbf0,
\quad r\le m.
\]
Consequently, for every $r,u\le m$ and $\diamond\in\{0,G\}$,
\begin{align*}
\E\{\mathbf U_{\ell r}^{(\diamond)}\mid\mathcal H_\ell\}
=\mathbf0,\qquad
\E\{\mathbf U_{\ell r}^{(\diamond)}
       \mathbf U_{\ell u}^{(\diamond)\top}\mid\mathcal H_\ell\}
=\Cov(\mathbf V_{\ell r},\mathbf V_{\ell u}).
\end{align*}
Differentiating the preceding interpolation at $t=0$ gives
\begin{align*}
\dot q_{\ell r}^{(\diamond)}(0)
=\frac2{\sqrt p}
(\mathbf T_{\ell r}+\mathbf h_r)^\top
\mathbf U_{\ell r}^{(\diamond)},\qquad
\ddot q_{\ell r}^{(\diamond)}(0)
=\frac2{\sqrt p}
\left\{\|\mathbf U_{\ell r}^{(\diamond)}\|_2^2
-\E\|\mathbf U_{\ell r}^{(\diamond)}\|_2^2\right\}.
\end{align*}
Consequently,
\begin{align*}
\E\{\dot q_{\ell r}^{(\diamond)}(0)\mid\mathcal H_\ell\}
&=0,\\
\E\{\dot q_{\ell r}^{(\diamond)}(0)
       \dot q_{\ell u}^{(\diamond)}(0)\mid\mathcal H_\ell\}
&=\frac4p
(\mathbf T_{\ell r}+\mathbf h_r)^\top
\Cov(\mathbf V_{\ell r},\mathbf V_{\ell u})
(\mathbf T_{\ell u}+\mathbf h_u),\\
\E\{\ddot q_{\ell r}^{(\diamond)}(0)\mid\mathcal H_\ell\}
&=0.
\end{align*}
Since the derivatives of
$f$ evaluated at $\mathbf q_{\ell,0}$ are $\mathcal H_\ell$-measurable, the
chain-rule identities from Step~2(AM) now give
\[
\E\{\dot F_\ell^{(0)}(0)-\dot F_\ell^{(G)}(0)\}=0,
\qquad
\E\{\ddot F_\ell^{(0)}(0)-\ddot F_\ell^{(G)}(0)\}=0.
\]
Hence only the third-order Taylor remainder remains. It suffices to control the products of three first derivatives and those of one first and one second derivative of \(q_{\ell r}^{(\diamond)}(t)\).

\begingroup
Apply \eqref{eq:fd-hybrid-bilinear-moment} with
$x_i=a_{i,r}\ind{i\in I_\ell}$ and
$y_i=a_{i,r}\ind{i\in\mathcal I\setminus I_\ell}$, and then with
$y=x$. Since $\|x\|_2=w_{\ell r}$ and $\|y\|_2\le C$,
\begin{align*}
&\left\|p^{-1/2}\{\mathbf T_{\ell r}^{\top}\mathbf U_{\ell r}^{(\diamond)}
-\E(\mathbf T_{\ell r}^{\top}\mathbf U_{\ell r}^{(\diamond)})\}\right\|_4
\le Cw_{\ell r},\\
&\left\|p^{-1/2}\{\|\mathbf U_{\ell r}^{(\diamond)}\|_2^2
-\E\|\mathbf U_{\ell r}^{(\diamond)}\|_2^2\}\right\|_4
\le Cw_{\ell r}^2,\\
&\left\|p^{-1/2}\mathbf h_r^\top\mathbf U_{\ell r}^{(\diamond)}\right\|_4
\le Cw_{\ell r}\|\mathbf h_r\|_2/\sqrt p.
\end{align*}
The last line follows from Lemma~\ref{lem:vector-basic}(ii),
\eqref{eq:weighted-physical-residual} under (PD), and covariance matching
for a Gaussian block. Substitution into the exact derivatives from
Step~2 gives, uniformly in $t\in[0,1]$ and $\diamond\in\{0,G\}$,
\[
\|\dot q_{\ell r}^{(\diamond)}(t)\|_4
\le C\{w_{\ell r}(1+H)+tw_{\ell r}^2\}
\le Cw_{\ell r}(1+H),\qquad
\|\ddot q_{\ell r}^{(\diamond)}(t)\|_4\le Cw_{\ell r}^2.
\]
Applying the integral Taylor remainder and H\"older bounds already
displayed in Step~2(a) therefore yields
\[
\E|R_{\ell,3}^{(0)}|+\E|R_{\ell,3}^{(G)}|
\le C_m\|f\|_{C^3}\max_{r\le m}w_{\ell r}^3(1+H)^3.
\]
\endgroup
Summing over $\ell$ and using the block-weight bound gives Equation~\eqref{eq:fd-block-gaussian-replacement}.

Step 4. The result in Step~1 applies to the stationary
Gaussian sequence $\{\bepsilon_i^G:i\in\mathbb Z\}$; hence adding back the deleted indices contributes at most
$C_m\|f\|_{C^1}r_n^{\mathrm{del}}(1+H)$. Define 
\[
\mathbf U_r^{I_\ell,G}
=\sum_{i\in I_\ell}a_{i,r}\bepsilon_i^G,
\qquad
\mathbf U_r^{I,G}
=\sum_{\ell=1}^L\mathbf U_r^{I_\ell,G},
\]
and 
\[
Q_{n,p}^{I,G}(t_r;\mathbf h_r)
=
\frac{\|\mathbf U_r^{I,G}+\mathbf h_r\|_2^2
-\E\|\mathbf U_r^{I,G}\|_2^2}{\sqrt p}.
\]
It remains to prove
\begin{align}
&\left|
\E f\{Q_{n,p}^{I,G,\mathrm{ind}}(t_r;\mathbf h_r):r\le m\}
-
\E f\{Q_{n,p}^{I,G}(t_r;\mathbf h_r):r\le m\}
\right|\notag\\
&\qquad\le
C_m\|f\|_{C^2}
\{\eta_{s_n}+e^{-cs_n^{\varkappa}}\}(1+H)^2.
\label{eq:fd-Gaussian-covariance-restoration}
\end{align}
For this comparison, define
\[
\boldsymbol\Sigma_0
=\Cov\{(\mathbf V_r^G:r\le m)\},
\qquad
\boldsymbol\Sigma_1
=\Cov\{(\mathbf U_r^{I,G}:r\le m)\},
\]
{let $\boldsymbol\Sigma_{\nu,ru}$ denote the $(r,u)$ block of}
$\boldsymbol\Sigma_\nu$, and set
\(
\boldsymbol\Delta_{ru}^{\mathrm{rest}}
=
\boldsymbol\Sigma_{1,ru}-\boldsymbol\Sigma_{0,ru}.
\)
Since $\mathbf V_1^G,\ldots,\mathbf V_L^G$ are mutually independent,
\(
\boldsymbol\Sigma_{0,ru}
=\sum_{\ell=1}^L
\Cov(\mathbf V_{\ell r}^G,\mathbf V_{\ell u}^G).
\)
Furthermore, we have
\[
\boldsymbol\Sigma_{1,ru}
=\sum_{\ell=1}^L
\Cov(\mathbf U_r^{I_\ell,G},\mathbf U_u^{I_\ell,G})
+\boldsymbol\Delta_{ru}^{\mathrm{cross}}
\]
with
\[
\boldsymbol\Delta_{ru}^{\mathrm{cross}}
:=\sum_{\substack{\ell,j=1\\\ell\ne j}}^L
\Cov(\mathbf U_r^{I_\ell,G},\mathbf U_u^{I_j,G}).
\]
Because $\{\bepsilon_i^G\}$ and $\{\bepsilon_i\}$ have the same covariance
function,
\(
\Cov(\mathbf U_r^{I_\ell,G},\mathbf U_u^{I_\ell,G})
=
\Cov(\mathbf U_r^{I_\ell},\mathbf U_u^{I_\ell}).
\)
In branches~(AM) and~(MD), the Gaussian blocks in Step~3 have exactly this
within-block covariance.  Thus, in branch~(AM),
\(
\boldsymbol\Delta_{ru}^{\mathrm{rest}}
=\boldsymbol\Delta_{ru}^{\mathrm{cross}}.
\)

In branch~(MD), $s_n>m_0$, while the distance between any two indices in
distinct retained blocks exceeds $s_n$.  Hence, for $\ell\ne j$,
\(
\Cov(\mathbf U_r^{I_\ell,G},\mathbf U_u^{I_j,G})=\mathbf0,
\boldsymbol\Delta_{ru}^{\mathrm{cross}}
=\boldsymbol\Delta_{ru}^{\mathrm{rest}}=\mathbf0.
\)
The centered Gaussian vectors at the two endpoints therefore have the same
covariance and hence the same distribution:
\[
(\mathbf V_r^G:r\le m)\ \stackrel{d}{=}\
(\mathbf U_r^{I,G}:r\le m).
\]
Consequently, the left-hand side of
\eqref{eq:fd-Gaussian-covariance-restoration} is zero under~(MD).

For branches~(AM) and~(PD), every pair of distinct retained blocks is
separated by more than $s_n$ indices.  Extending $a_{i,r}=0$ outside
$\{1,\ldots,n\}$, Lemma~\ref{lem:vector-basic}(i) and Cauchy--Schwarz give
\begin{align*}
\|\boldsymbol\Delta_{ru}^{\mathrm{cross}}\|_{\mathrm{op}}
&\le
\sum_{\substack{\ell,j=1\\ \ell\ne j}}^L
\sum_{i\in I_\ell}\sum_{k\in I_j}
|a_{i,r}a_{k,u}|\,\|\bGam(k-i)\|_{\mathrm{op}}\\
&\le
\sum_{|h|>s_n}\varpi_{\mathrm{dep}}(|h|)
\sum_i|a_{i,r}a_{i+h,u}|\\
&\le
\left(\sum_i a_{i,r}^2\right)^{1/2}
\left(\sum_i a_{i,u}^2\right)^{1/2}
\sum_{|h|>s_n}\varpi_{\mathrm{dep}}(|h|)
\le C\eta_{s_n}.
\end{align*}
Since every $p\times p$ matrix $\mathbf A$ satisfies
$\|\mathbf A\|_{\mathrm F}\le\sqrt p\|\mathbf A\|_{\mathrm{op}}$, it follows that
\(
\|\boldsymbol\Delta_{ru}^{\mathrm{cross}}\|_{\mathrm F}
\le C\sqrt p\,\eta_{s_n}.
\)

In branch~(PD), 
$\mathbf V_{\ell r}=\mathbf U_r^{I_\ell,(s_n)}$, and
\(
\Cov(\mathbf V_{\ell r}^G,\mathbf V_{\ell u}^G)
=
\Cov(\mathbf U_r^{I_\ell,(s_n)},\mathbf U_u^{I_\ell,(s_n)}).
\)
Define the aggregate within-block discrepancy
\[
\boldsymbol\Delta_{ru}^{(s_n)}
:=\sum_{\ell=1}^L
\left[
\Cov(\mathbf U_r^{I_\ell},\mathbf U_u^{I_\ell})
-
\Cov(\mathbf U_r^{I_\ell,(s_n)},\mathbf U_u^{I_\ell,(s_n)})
\right].
\]
The covariance difference in this branch therefore decomposes exactly as
\[
\boldsymbol\Delta_{ru}^{\mathrm{rest}}
=\boldsymbol\Delta_{ru}^{\mathrm{cross}}
+\boldsymbol\Delta_{ru}^{(s_n)}.
\]
To bound the second term, let
\[
\mathbf D_{\ell r}^{(s_n)}
=\sum_{i\in I_\ell}a_{i,r}\mathbf R_i^{(s_n)}.
\]
The weighted residual bound gives
\[
\sup_{\|\mathbf u\|_2=1}
\|\mathbf u^\top\mathbf D_{\ell r}^{(s_n)}\|_2
\le Ce^{-cs_n^{\varkappa}}w_{\ell r}.
\]
Since
$\mathbf U_r^{I_\ell,(s_n)}=\mathbf U_r^{I_\ell}-\mathbf D_{\ell r}^{(s_n)}$,
\begin{align*}
&\Cov(\mathbf U_r^{I_\ell},\mathbf U_u^{I_\ell})
-
\Cov(\mathbf U_r^{I_\ell,(s_n)},\mathbf U_u^{I_\ell,(s_n)})\\
&\qquad=
\Cov(\mathbf D_{\ell r}^{(s_n)},\mathbf U_u^{I_\ell})
+
\Cov(\mathbf U_r^{I_\ell},\mathbf D_{\ell u}^{(s_n)})
-
\Cov(\mathbf D_{\ell r}^{(s_n)},\mathbf D_{\ell u}^{(s_n)}).
\end{align*}
Cauchy--Schwarz and
$\sum_\ell w_{\ell r}w_{\ell u}\le C$ imply
\[
\|\boldsymbol\Delta_{ru}^{(s_n)}\|_{\mathrm{op}}
\le Ce^{-cs_n^{\varkappa}},
\qquad
\|\boldsymbol\Delta_{ru}^{(s_n)}\|_{\mathrm F}
\le C\sqrt p\,e^{-cs_n^{\varkappa}}.
\]
Consequently, uniformly over all three branches and $r,u\le m$,
\begin{equation}
\|\boldsymbol\Delta_{ru}^{\mathrm{rest}}\|_{\mathrm{op}}
\le C\{\eta_{s_n}+e^{-cs_n^{\varkappa}}\},
\qquad
\|\boldsymbol\Delta_{ru}^{\mathrm{rest}}\|_{\mathrm F}
\le C\sqrt p\,\{\eta_{s_n}+e^{-cs_n^{\varkappa}}\}.
\label{eq:fd-GC-discrepancy}
\end{equation}

For $0\le\theta\le1$, let
\[
\boldsymbol\Sigma_\theta
=(1-\theta)\boldsymbol\Sigma_0+\theta\boldsymbol\Sigma_1,
\qquad
\mathbf X_\theta
=(\mathbf X_{\theta,r}:r\le m)
\sim N(\mathbf0,\boldsymbol\Sigma_\theta).
\]
The covariance-tail and block-weight bounds, followed by convexity, give
\begin{equation}
\sup_{0\le\theta\le1}\max_{r,u\le m}
\|\boldsymbol\Sigma_{\theta,ru}\|_{\mathrm{op}}\le C,
\qquad
\sup_{0\le\theta\le1}\max_{r,u\le m}
\|\boldsymbol\Sigma_{\theta,ru}\|_{\mathrm F}\le C\sqrt p.
\label{eq:fd-GC-interpolating-covariance}
\end{equation}
If an endpoint covariance matrix is singular, perform the calculation with
$\boldsymbol\Sigma_\theta+\varepsilon\mathbf I_{mp}$ and then let
$\varepsilon\downarrow0$; all the following bounds are uniform in
$0<\varepsilon\le1$.

Define
\[
q_{\theta,r}
=
\frac{\|\mathbf X_{\theta,r}+\mathbf h_r\|_2^2
-\Tr(\boldsymbol\Sigma_{\theta,rr})}{\sqrt p},
\qquad
\mathbf q_\theta=(q_{\theta,1},\ldots,q_{\theta,m})^\top,
\qquad
\Psi(\theta)=\E f(\mathbf q_\theta).
\]
The two endpoints are the statistics in the target display above.  Let
$f_r=\partial_r f$ and $f_{ru}=\partial_{ru}f$.  Since
\[
\frac{\partial q_{\theta,r}}{\partial\theta}
=-\frac{\Tr(\boldsymbol\Delta_{rr}^{\mathrm{rest}})}{\sqrt p},
\qquad
\nabla_{\mathbf X_{\theta,r}}q_{\theta,r}
=\frac{2(\mathbf X_{\theta,r}+\mathbf h_r)}{\sqrt p},
\qquad
\nabla_{\mathbf X_{\theta,r}}^2q_{\theta,r}
=\frac{2\mathbf I_p}{\sqrt p},
\]
the Gaussian covariance-interpolation identity yields
\begin{align*}
\Psi^{\prime}(\theta)
={}&-\frac1{\sqrt p}\sum_{r=1}^m
\Tr(\boldsymbol\Delta_{rr}^{\mathrm{rest}})
\E f_r(\mathbf q_\theta)\\
&+\frac12\sum_{r,u=1}^m\sum_{a,b=1}^p
(\boldsymbol\Delta_{ru}^{\mathrm{rest}})_{ab}
\E\left[
\partial_{x_{r,a}x_{u,b}}
f\{q_{\theta,1},\ldots,q_{\theta,m}\}
\right].
\end{align*}
The chain rule gives
\begin{align*}
&\partial_{x_{r,a}x_{u,b}}
f\{q_{\theta,1},\ldots,q_{\theta,m}\}\\
&\qquad=
\frac4p f_{ru}(\mathbf q_\theta)
(X_{\theta,r,a}+h_{r,a})(X_{\theta,u,b}+h_{u,b})
+\frac2{\sqrt p}\ind{r=u}\delta_{ab}f_r(\mathbf q_\theta).
\end{align*}
After substitution, the last term cancels the first line of the preceding
interpolation identity.  Hence
\begin{equation}
\Psi^{\prime}(\theta)
=\frac2p\sum_{r,u=1}^m
\E\left[
f_{ru}(\mathbf q_\theta)
(\mathbf X_{\theta,r}+\mathbf h_r)^\top
\boldsymbol\Delta_{ru}^{\mathrm{rest}}
(\mathbf X_{\theta,u}+\mathbf h_u)
\right].
\label{eq:fd-GC-derivative}
\end{equation}

It remains to bound this contraction.  The Gaussian fourth-moment formula
and the preceding interpolating-covariance bounds imply
\begin{align*}
\Var\left(
\mathbf X_{\theta,r}^\top
\boldsymbol\Delta_{ru}^{\mathrm{rest}}
\mathbf X_{\theta,u}
\right)
&=\Tr\left(
\boldsymbol\Delta_{ru}^{\mathrm{rest}}
\boldsymbol\Sigma_{\theta,uu}
\boldsymbol\Delta_{ru}^{\mathrm{rest}\top}
\boldsymbol\Sigma_{\theta,rr}
\right)\\
&\quad+
\Tr\left(
\boldsymbol\Delta_{ru}^{\mathrm{rest}}
\boldsymbol\Sigma_{\theta,ur}
\boldsymbol\Delta_{ru}^{\mathrm{rest}}
\boldsymbol\Sigma_{\theta,ur}
\right)
\le C\|\boldsymbol\Delta_{ru}^{\mathrm{rest}}\|_{\mathrm F}^2.
\end{align*}
Therefore,
\begin{align*}
\E\left|
\mathbf X_{\theta,r}^\top
\boldsymbol\Delta_{ru}^{\mathrm{rest}}
\mathbf X_{\theta,u}
\right|
\le
\|\boldsymbol\Delta_{ru}^{\mathrm{rest}}\|_{\mathrm F}
\|\boldsymbol\Sigma_{\theta,ur}\|_{\mathrm F}
+C\|\boldsymbol\Delta_{ru}^{\mathrm{rest}}\|_{\mathrm F}
\le Cp\{\eta_{s_n}+e^{-cs_n^{\varkappa}}\}.
\end{align*}
Cauchy--Schwarz and the operator-norm bound for
$\boldsymbol\Delta_{ru}^{\mathrm{rest}}$ also give
\begin{align*}
&\E\left|
\mathbf h_r^\top\boldsymbol\Delta_{ru}^{\mathrm{rest}}
\mathbf X_{\theta,u}
\right|
+
\E\left|
\mathbf X_{\theta,r}^\top\boldsymbol\Delta_{ru}^{\mathrm{rest}}
\mathbf h_u
\right|
\le
C\{\eta_{s_n}+e^{-cs_n^{\varkappa}}\}
(\|\mathbf h_r\|_2+\|\mathbf h_u\|_2),\\
&\left|
\mathbf h_r^\top\boldsymbol\Delta_{ru}^{\mathrm{rest}}
\mathbf h_u
\right|
\le
C\{\eta_{s_n}+e^{-cs_n^{\varkappa}}\}
\|\mathbf h_r\|_2\|\mathbf h_u\|_2.
\end{align*}
Consequently,
\begin{equation}
\frac1p\E\left|
(\mathbf X_{\theta,r}+\mathbf h_r)^\top
\boldsymbol\Delta_{ru}^{\mathrm{rest}}
(\mathbf X_{\theta,u}+\mathbf h_u)
\right|
\le
C\{\eta_{s_n}+e^{-cs_n^{\varkappa}}\}(1+H)^2.
\label{eq:fd-GC-contraction}
\end{equation}
Since $m$ is fixed,
\[
\sup_{0\le\theta\le1}|\Psi^{\prime}(\theta)|
\le
C_m\|f\|_{C^2}
\{\eta_{s_n}+e^{-cs_n^{\varkappa}}\}(1+H)^2.
\]
Integrating over $\theta\in[0,1]$ gives the target bound.  We refer to the
preceding covariance interpolation, cancellation, and contraction estimate
as the Gaussian covariance comparison calculation (GC).

Finally, $\eta_{s_n}$ and $e^{-cs_n^{\varkappa}}$ decrease exponentially, so,
under (AM) and (PD), a sufficiently large $c_0=c_0(A)$ gives
\(
\eta_{s_n}+e^{-cs_n^{\varkappa}}
\le C(np)^{-A}.
\)
Since
\[
r_n^{\mathrm{del}}
\le C\left\{
\left(\frac{s_n}{b_n}\right)^{1/2}
+
\left(\frac{b_n}{n}\right)^{1/2}
\right\},
\]
combining Steps~1--4 gives the stated bound.
\end{proof}

\begin{lemma}
\label{lem:gaussian-cov-restoration}
Fix an integer $m\ge1$, time points $0<t_1<\cdots<t_m<1$, and deterministic vectors $\mathbf h_1,\ldots,\mathbf h_m\in\mathbb R^p$. Then, under Assumptions~\ref{ass:C1}--\ref{ass:C2}, for every $f\in C_b^3(\mathbb R^m)$,
\begin{align*}
&\left|
\E f\left\{\frac{Q_{n,p}^G(t_r;\mathbf h_r)}{\omega_p}:1\le r\le m\right\}
-
\E f\left\{G_p(t_r;\mathbf h_r):1\le r\le m\right\}
\right|\\
&\qquad\le
C_m\|f\|_{C^2}n^{-1}
\left\{1+\max_{1\le r\le m}
\frac{\|\mathbf h_r\|_2}{\sqrt p}\right\}^{2}.
\end{align*}
\end{lemma}

\begin{proof}
Recall the definitions of $Q_{n,p}^G(t;\mathbf h)$ and
$G_p(t;\mathbf h)$ in \eqref{eq:def-QnpG-shifted} and
\eqref{eq:def-Gp-shifted}, respectively.
Let
\(
\tau_p=\sqrt p\,\omega_p=\{2\Tr(\bOme^2)\}^{1/2},
\)
\(
H_n=\max_{1\le r\le m}\frac{\|\mathbf h_r\|_2}{\sqrt p}.
\)
Define
\[
\mathbf Y^{(0)}
=(\mathbf U_n^G(t_r):1\le r\le m),
\qquad
\mathbf Y^{(1)}
=(\bOme^{1/2}\mathbf B(t_r):1\le r\le m),
\]
and let $\boldsymbol\Sigma_\nu=\Cov(\mathbf Y^{(\nu)})$, with
$\boldsymbol\Sigma_{\nu,ru}$ denoting its $(r,u)$ block.  Then
\[
\boldsymbol\Sigma_{0,ru}
=\Cov\{\mathbf U_n^G(t_r),\mathbf U_n^G(t_u)\},
\qquad
\boldsymbol\Sigma_{1,ru}=K_B(t_r,t_u)\bOme.
\]
Thus, for
$\boldsymbol\Delta_{ru}
=\boldsymbol\Sigma_{1,ru}-\boldsymbol\Sigma_{0,ru}$,
Lemma~\ref{lem:cusum-cov-general} gives
\[
\max_{r,u\le m}\|\boldsymbol\Delta_{ru}\|_{\mathrm{op}}
\le\frac Cn,
\qquad
\max_{r,u\le m}\|\boldsymbol\Delta_{ru}\|_{\mathrm F}
\le\frac{C\sqrt p}{n}.
\]
The same lemma, Assumption~\ref{ass:C2}, and
$\boldsymbol\Sigma_{1,ru}=K_B(t_r,t_u)\bOme$ also give
\[
\max_{\nu\in\{0,1\}}\max_{r,u\le m}
\|\boldsymbol\Sigma_{\nu,ru}\|_{\mathrm{op}}\le C,
\qquad
\max_{\nu\in\{0,1\}}\max_{r,u\le m}
\|\boldsymbol\Sigma_{\nu,ru}\|_{\mathrm F}\le C\sqrt p.
\]
Moreover, $\tau_p^2=2\Tr(\bOme^2)\asymp p$.  We therefore apply the
Gaussian covariance comparison calculation (GC) from Step~4 of the proof of
Lemma~\ref{lem:block-comparison-finite}. By the definitions in \eqref{eq:def-QnpG-shifted} and \eqref{eq:def-Gp-shifted},
\begin{align*}
\frac{\|\mathbf Y_r^{(0)}+\mathbf h_r\|_2^2
-\Tr(\boldsymbol\Sigma_{0,rr})}{\tau_p}
&=\frac{Q_{n,p}^G(t_r;\mathbf h_r)}{\omega_p},\\
\frac{\|\mathbf Y_r^{(1)}+\mathbf h_r\|_2^2
-\Tr(\boldsymbol\Sigma_{1,rr})}{\tau_p}
&=G_p(t_r;\mathbf h_r).
\end{align*}
\begingroup
For completeness, set
$\boldsymbol\Sigma_\theta=(1-\theta)\boldsymbol\Sigma_0+
\theta\boldsymbol\Sigma_1$, take
$\mathbf X_\theta\sim N(\mathbf0,\boldsymbol\Sigma_\theta)$, and let
$q_{\theta,r}=\{\|\mathbf X_{\theta,r}+\mathbf h_r\|_2^2-
\Tr(\boldsymbol\Sigma_{\theta,rr})\}/\tau_p$.
The trace cancellation in \eqref{eq:fd-GC-derivative} gives
\begin{align*}
\left|\frac{d}{d\theta}\E f(\mathbf q_\theta)\right|
&=\left|\frac2{\tau_p^2}\sum_{r,u}
\E\{f_{ru}(\mathbf q_\theta)
(\mathbf X_{\theta,r}+\mathbf h_r)^\top
\boldsymbol\Delta_{ru}(\mathbf X_{\theta,u}+\mathbf h_u)\}\right|\\
&\le\frac{C\|f\|_{C^2}}{np}\sum_{r,u}
\{\E\|\mathbf X_{\theta,r}+\mathbf h_r\|_2^2\}^{1/2}
\{\E\|\mathbf X_{\theta,u}+\mathbf h_u\|_2^2\}^{1/2}\\
&\le C_m\|f\|_{C^2}n^{-1}(1+H_n)^2.
\end{align*}
Integrating over $\theta\in[0,1]$ yields
\endgroup
\begin{align*}
&\left|
\E f\left\{\frac{Q_{n,p}^G(t_r;\mathbf h_r)}{\omega_p}:1\le r\le m\right\}
-
\E f\left\{G_p(t_r;\mathbf h_r):1\le r\le m\right\}
\right|\\
&\qquad\le
C_m\|f\|_{C^2}n^{-1}(1+H_n)^2,
\end{align*}
which is the asserted bound.
\end{proof}

\begin{lemma}\label{lem:ideal-gaussian-dense}
Under Assumption~\ref{ass:C2}, as $p\rightarrow\infty$, $G_p(t)\Rightarrow V(t) $ in $D[0,1]$, where $V$ is centered Gaussian with
\(
\Cov\{V(s),V(t)\}=K_B(s,t)^2.
\)
\end{lemma}

\begin{proof}
\begingroup
We prove finite-dimensional convergence by a triangular-array CLT and
tightness by a fourth-moment increment bound. Recall
\eqref{eq:def-Gp-shifted}, and let $\lambda_{1,p},\ldots,\lambda_{p,p}$ be the eigenvalues of $\bOme$.  Orthogonal invariance of $\mathbf B(t)=(B_1(t),\ldots,B_p(t))^\top$ gives
\endgroup
\begin{align*}
G_p(t)=&\frac{\mathbf B(t)^\top\bOme\mathbf B(t)
-t(1-t)\Tr(\bOme)}{\{2\Tr(\bOme^2)\}^{1/2}}=\sum_{j=1}^p
\frac{\lambda_{j,p}}{\{2\Tr(\bOme^2)\}^{1/2}}
\{B_j(t)^2-t(1-t)\}.
\end{align*}
Fix $0<t_1<\cdots<t_m<1$ and $a_1,\ldots,a_m\in\mathbb R$.  Let
\[
Z_{j,p}=
\frac{\lambda_{j,p}}{\{2\Tr(\bOme^2)\}^{1/2}}
\sum_{r=1}^ma_r\{B_j(t_r)^2-t_r(1-t_r)\}.
\]
The $Z_{j,p}$ are independent and centered, and
\begin{align*}
\sum_{j=1}^p\Var(Z_{j,p})
&=\sum_{r,u=1}^ma_ra_uK_B(t_r,t_u)^2.
\end{align*}
Because Gaussian fourth moments are finite uniformly over the fixed time points,
\begin{align*}
\sum_{j=1}^p\E Z_{j,p}^4
&\le C_m
\frac{\sum_{j=1}^p\lambda_{j,p}^4}
     {\{\Tr(\bOme^2)\}^2}\le C_m
\frac{\|\bOme\|_{\mathrm{op}}^2}
     {\Tr(\bOme^2)}\le C_mp^{-1}.
\end{align*}
For every $\varepsilon>0$,
\[
\sum_j\E\{Z_{j,p}^2\ind{|Z_{j,p}|>\varepsilon}\}
\le\varepsilon^{-2}\sum_j\E Z_{j,p}^4
\le C_m\varepsilon^{-2}p^{-1}=o(1).
\]
The Lindeberg--Feller theorem and the Cram\'er--Wold device give 
\[
\bigl(G_p(t_1),\ldots,G_p(t_m)\bigr)
\Rightarrow
\bigl(V(t_1),\ldots,V(t_m)\bigr).
\]

For tightness, fix $0\le s\le t\le1$, let $d=t-s$, and define
\[
B_j^{-}(s,t)=B_j(t)-B_j(s),\qquad
B_j^{+}(s,t)=B_j(t)+B_j(s),
\]
\[
\Upsilon_j(s,t)
=B_j^{-}(s,t)B_j^{+}(s,t)
-\E\{B_j^{-}(s,t)B_j^{+}(s,t)\},\qquad
c_{j,p}^{\Omega}=\frac{\lambda_{j,p}}{\{2\Tr(\bOme^2)\}^{1/2}}.
\]
Then
\[
G_p(t)-G_p(s)=\sum_{j=1}^pc_{j,p}^{\Omega}\Upsilon_j(s,t),
\]
and the $\Upsilon_j(s,t)$ are independent and centered.  The Brownian-bridge covariance gives
\[
\Var\{B_j^{-}(s,t)\}\le d,\qquad
\Var\{B_j^{+}(s,t)\}\le 1,
\qquad
|\Cov\{B_j^{-}(s,t),B_j^{+}(s,t)\}|\le2d^{1/2}.
\]
For a centered product of two jointly Gaussian variables,
\[
\E\{\Upsilon_j(s,t)^2\}
=\Var\{B_j^{-}(s,t)\}\Var\{B_j^{+}(s,t)\}
 +\Cov\{B_j^{-}(s,t),B_j^{+}(s,t)\}^2
\le Cd.
\]
Writing
$\{B_j^{-}(s,t),B_j^{+}(s,t)\}^{\top}=\mathbf L_j\mathbf z_j$
with $\mathbf z_j\sim N(\mathbf0,\mathbf I_2)$, and setting
$\mathbf A=\frac12\left(\begin{smallmatrix}0&1\\1&0\end{smallmatrix}\right)$,
we have
\[
\Upsilon_j(s,t)
=\mathbf z_j^{\top}(\mathbf L_j^{\top}\mathbf A\mathbf L_j)\mathbf z_j
 -\Tr(\mathbf L_j^{\top}\mathbf A\mathbf L_j).
\]
Let $\kappa_{1j}$ and $\kappa_{2j}$ be the eigenvalues of the symmetric matrix
$\mathbf L_j^{\top}\mathbf A\mathbf L_j$.  Its orthogonal diagonalization,
together with the rotational invariance of $\mathbf z_j$, gives
\[
\Upsilon_j(s,t)\stackrel{d}{=}
\kappa_{1j}(z_{1j}^2-1)+\kappa_{2j}(z_{2j}^2-1).
\]
\begingroup
Using $\E(z^2-1)^2=2$ and $\E(z^2-1)^4=60$ gives
\begin{align*}
\E\{\Upsilon_j(s,t)^4\}
&=60(\kappa_{1j}^4+\kappa_{2j}^4)
 +24\kappa_{1j}^2\kappa_{2j}^2\\
&\le60(\kappa_{1j}^2+\kappa_{2j}^2)^2
=15[\E\{\Upsilon_j(s,t)^2\}]^2\le Cd^2.
\end{align*}
Independence across $j$ therefore yields
\endgroup
\begin{align*}
\E|G_p(t)-G_p(s)|^4
&=\sum_{j=1}^p(c_{j,p}^{\Omega})^4\E\{\Upsilon_j(s,t)^4\}+6\sum_{j<\ell}(c_{j,p}^{\Omega})^2(c_{\ell,p}^{\Omega})^2
   \E\{\Upsilon_j(s,t)^2\}\E\{\Upsilon_\ell(s,t)^2\}\\
&\le Cd^2\left\{\sum_{j=1}^p(c_{j,p}^{\Omega})^2\right\}^2
\le Cd^2.
\end{align*}
The Kolmogorov--Chentsov criterion now proves tightness in $D[0,1]$.
\end{proof}

\begin{proof}[Proof of Theorem~\ref{Th1}]
 Under $H_0$, $W(k)-\mu_{k}=Q_{n,p}(t)$, and therefore
\begin{equation}\label{eq:Th1-decomposition}
\frac{W(\lfloor nt\rfloor)-\mu_{\lfloor nt\rfloor,M}}{\omega_p}
=
\frac{Q_{n,p}(t)}{\omega_p}+R_{n,p}(t),
\qquad
R_{n,p}(t)
=
\frac{\mu_{\lfloor nt\rfloor}-\mu_{\lfloor nt\rfloor,M}}{\omega_p}.
\end{equation}

For $t\in[0,1]$, let $k=\lfloor nt\rfloor$. We first control $R_{n,p}$ in~\eqref{eq:Th1-decomposition}. 
For $1\le k\le n-1$, the formulas for $\mu_k$ and $\mu_{k,M}$ in
the main text give
\begin{equation}\label{eq:Th1-mu-difference}
\mu_k-\mu_{k,M}
=
\frac2{n\sqrt p}
\sum_{h=M+1}^{n-1}
\left\{\sum_{i=1}^{n-h}c_{i,k}c_{i+h,k}\right\}
\Tr\{\bGam(h)\}.
\end{equation}
For every $h>M$,
\(
|\Tr\{\bGam(h)\}|
\le p\|\bGam(h)\|_{\mathrm{op}},
\)
and the Cauchy--Schwarz inequality gives
\begin{align*}
\left|
\sum_{i=1}^{n-h}c_{i,k}c_{i+h,k}
\right|
&\le
\left(\sum_{i=1}^{n-h}c_{i,k}^2\right)^{1/2}
\left(\sum_{i=1}^{n-h}c_{i+h,k}^2\right)^{1/2}\\
&\le
\sum_{i=1}^nc_{i,k}^2
=
k\left(1-\frac{k}{n}\right)^2
+(n-k)\left(\frac{k}{n}\right)^2
=nv_k
\le\frac n4.
\end{align*}
For $k=0,n$, $c_{i,k}=0$ and hence
$\mu_k=\mu_{k,M}=0$. \eqref{eq:Th1-mu-difference} and Lemma~\ref{lem:vector-basic}(i) yield
\begin{align}
\max_{0\le k\le n}|\mu_k-\mu_{k,M}|
&\le
C\sqrt p\sum_{h>M}\|\bGam(h)\|_{\mathrm{op}}\notag\\
&\le
C\sqrt p\sum_{h>M}\varpi_{\mathrm{dep}}(h)
\le C\sqrt p\,\eta_M,
\label{eq:Th1-centering-truncation}
\end{align}
where
\(
\eta_M=\sum_{r>M}(1+r)^8\varpi_{\mathrm{dep}}(r).
\)
Under Assumption~\ref{ass:C1}(MD), this bound is zero once $M\ge m_0$.
Under Assumption~\ref{ass:C1}(AM) or (PD),
\[
\sqrt p\,\eta_M
\le
C\exp\left\{
\frac12\log p-c(n\wedge p)^{\varkappa/8}
\right\}
=o(1).
\]
Indeed, it holds when $p\le n$, while for $p>n$ it follows
with \(\log p=o(n^{\varkappa/8})\). Hence
\begingroup
\begin{equation}\label{eq:Th1-centering}
\sup_{0\le t\le 1}|R_{n,p}(t)|=o(1).
\end{equation}
\endgroup

\begingroup
We next establish finite-dimensional convergence. Choose
$b_n=\lceil\sqrt{ns_n}\rceil$ in
Lemma~\ref{lem:block-comparison-finite}. Its block conditions follow from
\[
s_n/n\longrightarrow0,\qquad
s_n/b_n\le\sqrt{s_n/n},\qquad
b_n/n\le\sqrt{s_n/n}+n^{-1}.
\]
Here $s_n$ is fixed under (MD), and
$s_n=O\{(\log n+\log p)^{1/\varkappa}\}=o(n)$ under (AM) or (PD).
\endgroup
Since $\omega_p$ is bounded away from zero and infinity, Lemma~\ref{lem:block-comparison-finite} with
$\mathbf h_1=\cdots=\mathbf h_m=\mathbf0$, gives
\begin{align*}
\left|
\E f\left\{\frac{Q_{n,p}(t_r)}{\omega_p}:1\le r\le m\right\}
-
\E f\left\{\frac{Q_{n,p}^{G}(t_r)}{\omega_p}:1\le r\le m\right\}
\right|
\le
C_m\|f\|_{C^3}\rho_n^{\mathrm{fd}}
=o(1).
\end{align*}
Lemma~\ref{lem:gaussian-cov-restoration} yields
\[
\left|
\E f\left\{\frac{Q_{n,p}^{G}(t_r)}{\omega_p}:1\le r\le m\right\}
-
\E f\{G_p(t_r):1\le r\le m\}
\right|
\le C_m\|f\|_{C^2}n^{-1}
=o(1).
\]
Combining these two results with Lemma~\ref{lem:ideal-gaussian-dense},
\begin{equation}\label{eq:Th1-fidi}
\left\{
\frac{Q_{n,p}(t_r)}{\omega_p}:1\le r\le m
\right\}
\Rightarrow
\{V(t_r):1\le r\le m\}.
\end{equation}

It remains to prove tightness. Since Assumption~\ref{ass:C2} implies
\(
0<c\le\omega_p\le C<\infty
\)
for all sufficiently large $p$, it suffices to prove tightness of
$Q_{n,p}$. Let $\widetilde Q_{n,p}$ be the polygonal interpolation of
$Q_{n,p}$ on the grid $\{k/n:0\le k\le n\}$. For \(u\in[k/n,(k+1)/n]\), let \(\lambda_u
=\frac{u-k/n}{(k+1)/n-k/n}
\in[0,1],\)
\[\widetilde Q_{n,p}(u)
=
(1-\lambda_u)Q_{n,p}\left(\frac{k}{n}\right)
+
\lambda_u Q_{n,p}\left(\frac{k+1}{n}\right),\]
with \(0\le\lambda\le1.\)
{For \(t,s\in[k/n,(k+1)/n]\),}
\[
\widetilde Q_{n,p}(t)-\widetilde Q_{n,p}(s)
=n(t-s)\left[Q_{n,p}\left(\frac{k+1}{n}\right)-Q_{n,p}\left(\frac{k}{n}\right)\right].
\]
Lemma~\ref{lem:quadratic-cumulants} gives
\(
\E|Q_{n,p}(t)-Q_{n,p}(s)|^4
\le C(t-s+n^{-1})^2.
\)
\begingroup
For two points in one mesh interval, this gives
\[
\E|\widetilde Q_{n,p}(t)-\widetilde Q_{n,p}(s)|^4
\le Cn^2|t-s|^4\le C|t-s|^2.
\]
For arbitrary $s<t$, split $[s,t]$ into its two possible end pieces
inside mesh intervals and its middle grid interval. If their lengths
are $d_1,d_2,d_3$ (zero lengths are allowed), the preceding bound and
Lemma~\ref{lem:quadratic-cumulants} on the nonempty grid interval give
\begin{align*}
\|\widetilde Q_{n,p}(t)-\widetilde Q_{n,p}(s)\|_4
&\le C\sum_{j=1}^3\sqrt{d_j}
\le C\sqrt{3(d_1+d_2+d_3)}
=C\sqrt{3(t-s)}.
\end{align*}
Thus, by the Kolmogorov--Chentsov criterion, $\{\widetilde Q_{n,p}\}$ is tight in $C[0,1]$. Moreover,
\endgroup
\begin{align*}
&\Pr\left\{
\sup_{0\le t\le1}|Q_{n,p}(t)-\widetilde Q_{n,p}(t)|
>\varepsilon
\right\}\\
&\qquad\le
\varepsilon^{-4}\sum_{k=0}^{n-1}
\E\left|
Q_{n,p}\left(\frac{k+1}{n}\right)
-Q_{n,p}\left(\frac{k}{n}\right)
\right|^4
\le C\varepsilon^{-4}n^{-1}\rightarrow 0.
\end{align*}
Thus $\{Q_{n,p}\}$ is tight in $D[0,1]$. 
Combining this with~\eqref{eq:Th1-fidi} yields
\[
\frac{Q_{n,p}(t)}{\omega_p}
\Rightarrow V(t)
\qquad\text{in }D[0,1].
\]
Finally, \eqref{eq:Th1-decomposition}, \eqref{eq:Th1-centering}, and
Slutsky's theorem prove the convergence of $\{W(\lfloor nt\rfloor)-\mu_{\lfloor nt\rfloor,M}\}/\omega_p$ in Theorem~\ref{Th1}.
\end{proof}

\subsection{{Proof of Theorem~\ref{Th2}}%
}

For $0\le h,\ell\le M$, define
\[\mathbf D_{i,h}^{(M)}=\bepsilon_i-\bepsilon_{i+M+h+1},
\qquad
\mathcal D_{i_1,h;i_2,\ell}^{(M)}={\mathbf D_{i_1,h}^{(M)}}^\top
\mathbf D_{i_2,\ell}^{(M)}.
\]
Let
\[
Z_{i,h}={\mathbf D_{i+h,h}^{(M)}}^\top\mathbf D_{i,h}^{(M)}-\E\{{\mathbf D_{i+h,h}^{(M)}}^\top\mathbf D_{i,h}^{(M)}\}.
\]
and
\begin{align*}
H_{h,\ell}^{\mathrm F}(i_1,i_2)
&=\mathcal D_{i_1,h;i_2,\ell}^{(M)}\mathcal D_{i_1+h,h;i_2+\ell,\ell}^{(M)},
\qquad
H_{h,\ell}^{\mathrm C}(i_1,i_2)=\mathcal D_{i_1,h;i_2+\ell,\ell}^{(M)}\mathcal D_{i_1+h,h;i_2,\ell}^{(M)}.
\end{align*}

\begin{lemma}\label{lem:difference-product-general}
Suppose Assumptions~\ref{ass:C1}--\ref{ass:C3} hold. (i) Uniformly for $0\le h,\ell\le M=o(n)$,
\[
\sum_{i_1=1}^{N_h}\sum_{i_2=1}^{N_\ell}|\Cov(Z_{i_1,h},Z_{i_2,\ell})|\le Cnp,
\]
where $N_h=n-M-2h-1$ and $N_\ell=n-M-2\ell-1$.

(ii) %
Under branch~(AM) or (PD), suppose additionally that
$\log p=o(n^{\varkappa})$. Then, uniformly for
$0\le h,\ell,h',\ell'\le M=o(n)$ and
$\nu,\nu'\in\{\mathrm F,\mathrm C\}$,
\[
\left|\Cov(\widehat{\mathcal T}_{h\ell}^{\nu,\epsilon},
\widehat{\mathcal T}_{h'\ell'}^{\nu',\epsilon})\right|
\le C\left(\frac pn+\frac{p^2}{n^2}\right),
\]
where the superscript $\epsilon$ denotes the noise-only version of the
estimator in the main paper.
\end{lemma}

\begin{proof}
(i) %
By the product--cumulant identity,
\begin{align*}
\Cov(Z_{i_1,h},Z_{i_2,\ell})
={}&\sum_{a,b=1}^p
\Cov(D_{i_1+h,h,a}^{(M)},D_{i_2+\ell,\ell,b}^{(M)})
\Cov(D_{i_1,h,a}^{(M)},D_{i_2,\ell,b}^{(M)})\\
&+\sum_{a,b=1}^p
\Cov(D_{i_1+h,h,a}^{(M)},D_{i_2,\ell,b}^{(M)})
\Cov(D_{i_1,h,a}^{(M)},D_{i_2+\ell,\ell,b}^{(M)})\\
&+\sum_{a,b=1}^p
\Cum(D_{i_1+h,h,a}^{(M)},D_{i_1,h,a}^{(M)},
D_{i_2+\ell,\ell,b}^{(M)},D_{i_2,\ell,b}^{(M)}).%
\end{align*}

\begingroup
For the first covariance pairing, define $b_h=M+h+1$ and
$b_\ell=M+\ell+1$, and let $\mathcal O_{h\ell}^{\mathrm{mr}}=\{b b_\ell-a b_h:a,b\in\{0,1\}\}$. Expansion of the differences and Lemma~\ref{lem:vector-basic}(i) give
\begin{align*}
\|\Cov(\mathbf D_{i_1,h}^{(M)},
       \mathbf D_{i_2,\ell}^{(M)})\|_{\mathrm F}
=&\left\|\sum_{a,b\in\{0,1\}}
(-1)^{a+b}
\boldsymbol\Gamma
\bigl(i_2-i_1+bb_\ell-ab_h\bigr)\right\|_{\mathrm F}\\
\le{}& C\sqrt p\sum_{u\in\mathcal O_{h\ell}^{\mathrm{mr}}}
\varpi_{\mathrm{dep}}(|i_2-i_1+u|).
\end{align*}
\endgroup
Similarly, $\|\Cov(
\mathbf D_{i_1+h,h}^{(M)},
\mathbf D_{i_2+\ell,\ell}^{(M)})\|_{\mathrm F}$ has the bound $C\sqrt p\sum_{u\in\widetilde{\mathcal O}_{h\ell}^{\mathrm{mr}}}
\varpi_{\mathrm{dep}}(|i_2-i_1+u|)$ with a set
$\widetilde{\mathcal O}_{h\ell}^{\mathrm{mr}}=\{\ell-h+b b_\ell-a b_h:a,b\in\{0,1\}\}$. Hence Cauchy--Schwarz inequality yields
\begingroup
\begin{align*}
&\sum_{i_1=1}^{N_h}\sum_{i_2=1}^{N_\ell}\sum_{a,b=1}^p
|\Cov(D_{i_1+h,h,a}^{(M)},D_{i_2+\ell,\ell,b}^{(M)})\Cov(D_{i_1,h,a}^{(M)},D_{i_2,\ell,b}^{(M)})|\\
&\quad\le Cp\sum_{i_1=1}^n\sum_{r\in\mathbb Z}
\left\{\sum_{u\in\mathcal O_{h\ell}^{\mathrm{mr}}}
\varpi_{\mathrm{dep}}(|r+u|)\right\}
\left\{\sum_{v\in\widetilde{\mathcal O}_{h\ell}^{\mathrm{mr}}}
\varpi_{\mathrm{dep}}(|r+v|)\right\}\le Cnp,
\end{align*}
The convolution bound is uniform in its shifts because
\[
\sum_{r\in\mathbb Z}\varpi_{\mathrm{dep}}(|r+u|)
\varpi_{\mathrm{dep}}(|r+v|)
\le\left\{\sum_r\varpi_{\mathrm{dep}}(|r|)^2\right\}^{1/2}
    \left\{\sum_r\varpi_{\mathrm{dep}}(|r|)^2\right\}^{1/2}
\le C.
\]
For the second pairing, replace the two shift sets by
\[
\{bb_\ell-ab_h-h:a,b\in\{0,1\}\},\qquad
\{bb_\ell-ab_h+\ell:a,b\in\{0,1\}\}.
\]
The same displayed convolution bound gives its contribution $Cnp$.

\endgroup

For the fourth-cumulant term, we have
\begin{align*}
&\Cum(D_{i_1+h,h,a}^{(M)},D_{i_1,h,a}^{(M)},
D_{i_2+\ell,\ell,b}^{(M)},D_{i_2,\ell,b}^{(M)})\\
&\quad=
\sum_{\boldsymbol\iota\in\{0,1\}^4}
(-1)^{\iota_1+\cdots+\iota_4}
\Cum\bigl(
\epsilon_{i_1+h+\iota_1b_h,a},
\epsilon_{i_1+\iota_2b_h,a},
\epsilon_{i_2+\ell+\iota_3b_\ell,b},
\epsilon_{i_2+\iota_4b_\ell,b}
\bigr).
\end{align*}
Let $J_n=\{1,\ldots,n\}$, hence
\begingroup
\begin{align*}
&\sum_{i_1=1}^{N_h}\sum_{i_2=1}^{N_\ell}\sum_{a,b=1}^p
\left|
\Cum\bigl(
\epsilon_{i_1+h+\iota_1b_h,a},
\epsilon_{i_1+\iota_2b_h,a},
\epsilon_{i_2+\ell+\iota_3b_\ell,b},
\epsilon_{i_2+\iota_4b_\ell,b}
\bigr)
\right|\\
&\quad\le
\sum_{a,b=1}^p\sum_{u_1,u_2,u_3,u_4\in J_n}
\left|
\Cum(\epsilon_{u_1,a},\epsilon_{u_2,a},
\epsilon_{u_3,b},\epsilon_{u_4,b})
\right|.
\end{align*}
\endgroup
Apply the single-block assertion of
Lemma~\ref{lem:cumulant-coefficient-contraction} with $q=2$,
$\pi=\{\{1,2,3,4\}\}$, which is connected relative to
$\mathfrak P_2$, and $x_{r,i}=\ind{i\in J_n}$ for $r=1,\ldots,4$.
Assumption~\ref{ass:C3} gives
$\mathfrak C_{\pi,p}\le Cp$, and therefore
\begin{align*}
\sum_{a,b=1}^p\sum_{u_1,u_2,u_3,u_4\in J_n}
\left|
\Cum(\epsilon_{u_1,a},\epsilon_{u_2,a},
\epsilon_{u_3,b},\epsilon_{u_4,b})
\right|\le
\mathfrak C_{\pi,p}\|x_1\|_1
\prod_{r=2}^4\|x_r\|_\infty
\le Cnp.
\end{align*}
Summing over the $2^4$ choices of $\boldsymbol\iota$ consequently yields
\[
\sum_{i_1,i_2}\sum_{a,b=1}^p
|\Cum(D_{i_1+h,h,a}^{(M)},D_{i_1,h,a}^{(M)},
D_{i_2+\ell,\ell,b}^{(M)},D_{i_2,\ell,b}^{(M)})|
\le Cnp.
\]
This proves the first assertion.

(ii) Recall
$D_{u,q,j}^{(M)}=
\epsilon_{u,j}-\epsilon_{u+b_q,j},$
with
$b_q=M+q+1.$ {Expanding the difference gives} \(D_{u,q,j}^{(M)}
=
\sum_{\iota_r\in\{0,1\}}
(-1)^{\iota_r}
\epsilon_{u+\iota_rb_q,j}.\)
Define
\begin{align*}
\boldsymbol\tau_{h\ell}^{\mathrm F}(i_1,i_2)
=(i_1,i_2,i_1+h,i_2+\ell),\qquad
\boldsymbol\tau_{h\ell}^{\mathrm C}(i_1,i_2)
=(i_1,i_2+\ell,i_1+h,i_2),\\
\mathbf m_{h\ell}=(h,\ell,h,\ell),
\qquad
\mathbf j(a,b)=(a,a,b,b).
\end{align*}
For $\boldsymbol\iota=(\iota_1,\dots,\iota_4)$, with $|\boldsymbol\iota|=\sum_{r=1}^4\iota_r$, expansion of the four
differences gives, for $\nu\in\{\mathrm F,\mathrm C\}$,
\begin{align}
H_{h,\ell}^{\nu}(i_1,i_2)
={}&\sum_{a,b=1}^p
\sum_{\boldsymbol\iota\in\{0,1\}^4}
(-1)^{|\boldsymbol\iota|}
\prod_{r=1}^4
\epsilon_{\tau_{h\ell,r}^{\nu}(i_1,i_2)
+\iota_rb_{m_{h\ell,r}},\,j_r(a,b)},
\label{eq:difference-kernel-expansion}
\end{align}
where $\tau_{h\ell,r}^{\nu}(i_1,i_2)
$ is the $r$th component of $\boldsymbol\tau_{h\ell}^{\nu}(i_1,i_2)$, $m_{h\ell,r}$ is the $r$th component of $\mathbf m_{h\ell}$, and $j_r(a,b)$ is the $r$th component of $\mathbf j(a,b)$.

Fix $\boldsymbol\iota,\boldsymbol\iota'\in\{0,1\}^4$, $(i_1,i_2)\in\mathcal I_{h\ell}$, $(i_1',i_2')\in\mathcal I_{h'\ell'}$, $a,b,c,d=1,\dots,p$. For $1\le r\le4$, we  define
\begin{align*}
\xi_r
=\epsilon_{\tau_{h\ell,r}^{\nu}(i_1,i_2)
+\iota_rb_{m_{h\ell,r}},\,j_r(a,b)},
\qquad
\xi_{4+r}
=\epsilon_{\tau_{h'\ell',r}^{\nu'}(i_1',i_2')
+\iota'_rb_{m_{h'\ell',r}},\,j_r(c,d)}.
\end{align*}
Let $\Pi_8$ be the set of partitions of $\{1,\ldots,8\}$ and set
\[
\Pi_8^\times
=
\left\{
\pi\in\Pi_8:
\text{some }B\in\pi\text{ satisfies }
B\cap\{1,\ldots,4\}\ne\varnothing,
\ B\cap\{5,\ldots,8\}\ne\varnothing
\right\}.
\]
The moment--cumulant formula gives 
\begin{align}
\Cov\left(\prod_{r=1}^4\xi_r,
            \prod_{r=5}^8\xi_r\right)
={}&
\E\left(\prod_{r=1}^8\xi_r\right)
-
\E\left(\prod_{r=1}^4\xi_r\right)\,
\E\left(\prod_{r=5}^8\xi_r\right)=
\sum_{\pi\in\Pi_8^\times}
\prod_{B\in\pi}\Cum(\xi_r:r\in B).
\label{eq:product-covariance-partitions}
\end{align}
Then, \eqref{eq:difference-kernel-expansion}--
\eqref{eq:product-covariance-partitions} imply
\begin{align}
&\sum_{(i_1,i_2)\in\mathcal I_{h\ell}}
\sum_{(i_1',i_2')\in\mathcal I_{h'\ell'}}
\left|
\Cov\{H_{h,\ell}^{\nu}(i_1,i_2),H_{h',\ell'}^{\nu'}(i_1',i_2')\}
\right|
\le
\sum_{\boldsymbol\iota,\boldsymbol\iota'\in\{0,1\}^4}
\sum_{\pi\in\Pi_8^\times}
S_\pi(\boldsymbol\iota,\boldsymbol\iota'),
\label{eq:S9-partition-reduction}
\end{align}
where
\begin{align*}
S_\pi(\boldsymbol\iota,\boldsymbol\iota')
:={}&
\sum_{(i_1,i_2)\in\mathcal I_{h\ell}}
\sum_{(i_1',i_2')\in\mathcal I_{h'\ell'}}
\sum_{a,b,c,d=1}^p
\prod_{B\in\pi}|\Cum(\xi_r:r\in B)|.
\end{align*}

For $u=1,\ldots,4$, let $P_u=\{2u-1,2u\}$. By the definitions of \(\mathbf j(a,b)=(a,a,b,b)\) and \(\mathbf j(c,d)=(c,c,d,d)\), the coordinate indices of \(\xi_1,\ldots,\xi_8\) are \((a,a,b,b,c,c,d,d)\). Thus, the pairs
$P_1,\ldots,P_4$ correspond respectively to the repeated coordinate
indices $(a,a)$, $(b,b)$, $(c,c)$, and $(d,d)$.  We call $P_u$ and $P_v$ directly linked under
$\pi$ and let $P_u\leftrightarrow_\pi P_v$ denote this relation if
\[
\exists B\in\pi:\qquad
B\cap P_u\ne\varnothing,\qquad B\cap P_v\ne\varnothing.
\]
Define
\[
P_u\sim_\pi P_v
\quad\Longleftrightarrow\quad
\begin{gathered}
\text{there exist }J\ge0\text{ and }P_{u_0},\ldots,P_{u_J}\in\mathfrak P_4
\text{ such that}\\
P_{u_0}=P_u,\qquad P_{u_J}=P_v,\qquad
P_{u_{j-1}}\leftrightarrow_\pi P_{u_j}\quad (j=1,\ldots,J).
\end{gathered}
\]
For $P_u\in\mathfrak P_4$, set
\[
[P_u]_\pi:=\{P_v\in\mathfrak P_4:P_u\sim_\pi P_v\},
\qquad
\mathcal C(\pi):=\{[P_u]_\pi:P_u\in\mathfrak P_4\}.
\]
{Thus $\mathcal C(\pi)$ is a partition of the four pairs $P_1,\ldots,P_4$ into connected components: two pairs belong to the same component if they can be joined by a sequence of cumulant blocks in $\pi$.}  
For example,
\[
\pi=\big\{\{1,5\},\{2,6\},\{3,7\},\{4,8\}\big\}
\quad\rightarrow\quad
\mathcal C(\pi)=\big\{\{P_1,P_3\},\{P_2,P_4\}\big\},
\]
because the first two cumulant blocks $\{1,5\}$ and $\{2,6\}$ link $P_1$ with $P_3$, whereas
the last two link $P_2$ with $P_4$.  %
More generally, for $C\in\mathcal C(\pi)$, let
\[
U_C=\bigcup_{P_u\in C}P_u,
\qquad
\pi_C=\{B\in\pi:B\subseteq U_C\}.
\]
Then $C$ is a maximal collection of pair blocks linked through
cumulant blocks.  If $q_C=|C|\ge2$, the induced partition $\pi_C$ is
connected relative to $\mathfrak P_{q_C}$ after relabeling, and hence
\begin{align}
\mathfrak C_{\pi_C,p}\le C_{q_C}p
\qquad(q_C=2,3,4)
\label{eq:S9-component-contraction}
\end{align}
by Assumption~\ref{ass:C3}.

\begingroup
Define $\psi(r)=(1+|r|)^{-4}$.  Uniformly in every integer shift $e$,
\begin{align}
K_n(e)
:={}&\sum_{u=1}^n\sum_{v=1}^n\psi(v-u+e)
\le n\sum_{r\in\mathbb Z}(1+|r|)^{-4}\le Cn.
\label{eq:S9-temporal-convolution}
\end{align}
To apply this bound, let $\xi_a=\epsilon_{\rho_a,j(a)}$ for the
primitive times fixed above. For $B\in\pi$, let
\[
R_B=\max_{a\in B}|\rho_a-\rho_{a_B}|,\qquad a_B=\min B.
\]
The definition of $\mathfrak K_B^{\mathrm{cum}}$ gives the pointwise bound
\begin{equation}
|\Cum(\xi_a:a\in B)|
\le\mathfrak K_B^{\mathrm{cum}}(j_1,\ldots,j_4)(1+R_B)^{-8}.
\label{eq:S9-pointwise-lag-envelope}
\end{equation}
In particular, for $a,b\in B$,
\[
|\rho_a-\rho_b|\le2R_B,
\qquad (1+R_B)^{-8}\le C\psi(\rho_a-\rho_b).
\]
\endgroup

We now divide $\Pi_8^\times$ into three classes.  First, we consider the case that $C(\pi)$ has a singleton class, that is, there exists
$\{P_u\}\in\mathcal C(\pi)$ for some $u$. It means that $\{\{2u-1\},\{2u\}\}\subset\pi$ or $P_u\in\pi$. Since the $\xi_r$ are centered, a nonzero contribution therefore requires
$P_u\in\pi$, which contributes
$\Cov(\xi_{2u-1},\xi_{2u})$.  Without loss of generality, let $\xi_{2u-1}=\epsilon_{\rho_1,j}$ and $\xi_{2u}=\epsilon_{\rho_2,j}$. Then
\begin{align}
|\rho_1-\rho_2|
\ge \lfloor n/2\rfloor-3M-2
\ge n/3
\label{eq:S9-singleton-gap}
\end{align}
for all sufficiently large $n$.  Lemma~\ref{lem:vector-basic}(i) yields
\begin{align}
\sum_{j=1}^p
|\Cov(\epsilon_{\rho_1,j},\epsilon_{\rho_2,j})|
&\le p\|\bGam(\rho_2-\rho_1)\|_{\mathrm{op}}\nonumber\\
&\le
\begin{cases}
0, & \text{under (MD),}\cr
Cp\exp(-cn^{\varkappa}), & \text{under (AM) or (PD).}
\end{cases}
\label{eq:S9-singleton-covariance}
\end{align}
The remaining fixed-order cumulants are uniformly bounded.  Therefore
\begin{align}
\sum_{\substack{\pi\in\Pi_8^\times:\
\mathcal C(\pi)\text{ has a singleton class}}}
S_\pi(\boldsymbol\iota,\boldsymbol\iota')
&\le
\begin{cases}
0, & \text{under (MD),}\cr
Cn^4p^4\exp(-cn^{\varkappa}),
& \text{under (AM) or (PD).}
\end{cases}
\label{eq:S9-singleton-total}
\end{align}

Next suppose that $\mathcal C(\pi)=\{\mathfrak P_4\}$.  Equivalently,
\begingroup
$\pi$ itself is connected relative to $\mathfrak P_4$. Choose a block
$B$ and positions $a\in B\cap\{1,\ldots,4\}$ and
$b\in B\cap\{5,\ldots,8\}$. Their time difference has the form
\[
\rho_b-\rho_a=v-u+e_\pi,
\qquad u\in\{i_1,i_2\},\quad v\in\{i_1',i_2'\},
\]
where $e_\pi$ depends only on the fixed lags and difference choices.
Equation~\eqref{eq:S9-pointwise-lag-envelope} implies
\[
\sum_{a,b,c,d=1}^p\prod_{D\in\pi}|\Cum(\xi_r:r\in D)|
\le C\mathfrak C_{\pi,p}\psi(v-u+e_\pi).
\]
Summing the other two base times freely and then using
\eqref{eq:S9-temporal-convolution} gives
\endgroup
\begin{align}
S_\pi(\boldsymbol\iota,\boldsymbol\iota')
&\le Cn^2K_n(e_\pi)\mathfrak C_{\pi,p}
\le Cn^3p.
\label{eq:S9-connected-total}
\end{align}

Finally, suppose that $\mathcal C(\pi)$ has no singleton class and
contains more than one class.  Then necessarily
\[
\mathcal C(\pi)=\{C_1,C_2\},
\qquad |C_1|=|C_2|=2.
\]
The decomposition
$\{\{P_1,P_2\},\{P_3,P_4\}\}$ is impossible because it contains no
cumulant block meeting both $\{1,\ldots,4\}$ and
$\{5,\ldots,8\}$, contrary to $\pi\in\Pi_8^\times$.  Hence the only
two possible decompositions are
\[
\mathcal C(\pi)
\in
\left\{
\big\{\{P_1,P_3\},\{P_2,P_4\}\big\},
\big\{\{P_1,P_4\},\{P_2,P_3\}\big\}
\right\}.
\]
In particular, each $C_j=\{P_r,P_u\}$ contains one pair block from
$\{P_1,P_2\}$ and one from $\{P_3,P_4\}$.  Let
$P_v=\{L_v,R_v\}$, where $L_v=2v-1$ and $R_v=2v$ index primitive
variables from the first and second sample halves, respectively.  Since
the $\xi_r$ are centered, a nonzero contribution cannot contain a
singleton cumulant block.  Moreover, $\pi_{C_j}$ is connected relative
to $P_r$ and $P_u$.  As $U_{C_j}$ has four elements, these facts leave
exactly the following possibilities:
\[
\pi_{C_j}=\{U_{C_j}\},
\qquad
\pi_{C_j}=\big\{\{L_r,L_u\},\{R_r,R_u\}\big\},
\qquad
\pi_{C_j}=\big\{\{L_r,R_u\},\{R_r,L_u\}\big\}.
\]
\begingroup
Choose $C_1=\{P_r,P_u\}$, with $r\in\{1,2\}$ and
$u\in\{3,4\}$. For its three possible induced partitions define
\[
(d_1,d_2)=
\begin{cases}
(\rho_{L_u}-\rho_{L_r},\rho_{R_u}-\rho_{R_r}),
 &\pi_{C_1}=\{U_{C_1}\}
   \text{ or }\{\{L_r,L_u\},\{R_r,R_u\}\},\\
(\rho_{R_u}-\rho_{L_r},\rho_{L_u}-\rho_{R_r}),
 &\pi_{C_1}=\{\{L_r,R_u\},\{R_r,L_u\}\}.
\end{cases}
\]
For one fourth-order block, $|d_1|,|d_2|\le2R_{U_{C_1}}$, so
\[
(1+|d_1|)^4(1+|d_2|)^4
\le2^8(1+R_{U_{C_1}})^8,
\qquad
|\Cum(\xi_a:a\in U_{C_1})|
\le C\mathfrak K_{U_{C_1}}^{\mathrm{cum}}\psi(d_1)\psi(d_2).
\]
For either two-covariance partition, each block contains the two
endpoints of one $d_i$. Thus
\begin{align*}
\prod_{B\in\pi_{C_1}}|\Cum(\xi_a:a\in B)|
&\le\left(\prod_{B\in\pi_{C_1}}\mathfrak K_B^{\mathrm{cum}}\right)
       (1+|d_1|)^{-8}(1+|d_2|)^{-8}\\
&\le\left(\prod_{B\in\pi_{C_1}}\mathfrak K_B^{\mathrm{cum}}\right)
       \psi(d_1)\psi(d_2).
\end{align*}
For $C_2$, use \eqref{eq:S9-pointwise-lag-envelope} without its decay
factor. The coordinate indices of $C_1$ and $C_2$ are disjoint; hence
\[
\sum_{a,b,c,d=1}^p\prod_{B\in\pi}|\Cum(\xi_r:r\in B)|
\le C\mathfrak C_{\pi_{C_1},p}\mathfrak C_{\pi_{C_2},p}
       \psi(d_1)\psi(d_2).
\]
Each pair $P_v$ contains one time based on the first sample half and
one based on the second. Consequently the four base times occur exactly
once in $(d_1,d_2)$: after a permutation of $(i_1,i_2,i_1',i_2')$,
\[
d_1=v_1-u_1+e_{\pi,1},\qquad
d_2=v_2-u_2+e_{\pi,2},
\]
with fixed integer shifts $e_{\pi,1},e_{\pi,2}$.
Enlarging the valid
\endgroup
summation ranges to $\{1,\ldots,n\}$ therefore gives
\begin{align*}
\sum_{u_1,v_1=1}^n\sum_{u_2,v_2=1}^n
\psi(v_1-u_1+e_{\pi,1})
\psi(v_2-u_2+e_{\pi,2})
=K_n(e_{\pi,1})K_n(e_{\pi,2}).
\end{align*}
Consequently, Lemma~\ref{lem:cumulant-coefficient-contraction} and
\eqref{eq:S9-component-contraction}--
\eqref{eq:S9-temporal-convolution} give
\begin{align}
S_\pi(\boldsymbol\iota,\boldsymbol\iota')
&\le
C K_n(e_{\pi,1})K_n(e_{\pi,2})
\mathfrak C_{\pi_{C_1},p}
\mathfrak C_{\pi_{C_2},p}
\le Cn^2p^2.
\label{eq:S9-two-component-total}
\end{align}

There are only $2^8$ difference choices and finitely many partitions of
$\{1,\ldots,8\}$.  Combining
\eqref{eq:S9-partition-reduction},
\eqref{eq:S9-singleton-total},
\eqref{eq:S9-connected-total}, and
\eqref{eq:S9-two-component-total} yields
\begin{align}
&\sum_{(i_1,i_2)\in\mathcal I_{h\ell}}
\sum_{(i_1',i_2')\in\mathcal I_{h'\ell'}}
\left|
\Cov\{H_{h,\ell}^{\nu}(i_1,i_2),H_{h',\ell'}^{\nu'}(i_1',i_2')\}
\right|\nonumber\\
&\qquad\le
C\left\{n^3p+n^2p^2+n^4p^4\exp(-cn^{\varkappa})\right\},
\label{eq:S9-kernel-covariance-sum}
\end{align}
where the last term is zero under (MD).  Since
\[
\widehat{\mathcal T}_{h\ell}^{\nu,\epsilon}
=\frac1{4|\mathcal I_{h\ell}|}
\sum_{(i_1,i_2)\in\mathcal I_{h\ell}}H_{h,\ell}^{\nu}(i_1,i_2),
\qquad
c n^2\le|\mathcal I_{h\ell}|\le Cn^2,
\]
\eqref{eq:S9-kernel-covariance-sum} implies
\begin{align*}
\left|\Cov(\widehat{\mathcal T}_{h\ell}^{\nu,\epsilon},
\widehat{\mathcal T}_{h'\ell'}^{\nu',\epsilon})\right|
&\le
\frac{C}{n^4}
\left\{n^3p+n^2p^2+n^4p^4\exp(-cn^{\varkappa})\right\}\\
&\le
C\left\{\frac pn+\frac{p^2}{n^2}
+p^4\exp(-cn^{\varkappa})\right\}.
\end{align*}
Under branch~(AM) or (PD),
\[
n p^3\exp(-cn^{\varkappa})
=\exp\{\log n+3\log p-cn^{\varkappa}\}=o(1),
\]
and hence $p^4\exp(-cn^{\varkappa})=o(p/n)$; under branch~(MD) this
term is zero.  Thus, for all three branches,
\[
\left|\Cov(\widehat{\mathcal T}_{h\ell}^{\nu,\epsilon},
\widehat{\mathcal T}_{h'\ell'}^{\nu',\epsilon})\right|
\le C\left(\frac pn+\frac{p^2}{n^2}\right).
\]
This proves the second assertion uniformly over
$\nu,\nu'\in\{\mathrm F,\mathrm C\}$.
\end{proof}

\begin{lemma}\label{lem:centering-general}
Under the assumptions of Theorem~\ref{Th2},
\[
\max_{\gamma\in\{0,1/2\}}\max_{k\in\mathcal K_\gamma}
|\widehat\mu_{\gamma,k,M}-\mu_{\gamma,k,M}|
=O_p(r_{\mu,n}),
\]
where
\(
r_{\mu,n}
=M/\sqrt n
+\sqrt p\,M\eta_M
+M^2\|\bdelta\|_2^2/(n\sqrt p).
\)
\end{lemma}

\begin{proof}
Define
\[
a_{k,h}
=
\frac{2-\ind{h=0}}{n\sqrt p}
\sum_{i=1}^{n-h}c_{i,k}c_{i+h,k},
\qquad
a_{\gamma,k,h}=v_k^{-2\gamma}a_{k,h}.
\]
By definition,
\[
\widehat\mu_{\gamma,k,M}
=\sum_{h=0}^M a_{\gamma,k,h}\widehat g_h,
\qquad
\mu_{\gamma,k,M}
=\sum_{h=0}^M a_{\gamma,k,h}g_h.
\]
For $0\le h\le M$, define the noise-only estimators
\[
\widehat\mu_{\gamma,k,M}^\epsilon
=\sum_{h=0}^M a_{\gamma,k,h}\widehat g_h^\epsilon,
\qquad
\widehat g_h^\epsilon
=\frac1{2N_h}\sum_{i=1}^{N_h}
(\mathbf D_{i+h,h}^{(M)})^\top\mathbf D_{i,h}^{(M)}.
\]
Thus
\begin{align}
\widehat\mu_{\gamma,k,M}-\mu_{\gamma,k,M}
={}&
\{\widehat\mu_{\gamma,k,M}
-\widehat\mu_{\gamma,k,M}^\epsilon\}
+\{\widehat\mu_{\gamma,k,M}^\epsilon
-\mu_{\gamma,k,M}\}.
\label{eq:S10-centering-decomposition}
\end{align}

We first control the first difference in
\eqref{eq:S10-centering-decomposition}.  Under model (1), we have
\[
\mathbf Y_{i,h}^{(M)}
=\mathbf D_{i,h}^{(M)}+\mathbf d_{i,h},
\qquad
\mathbf d_{i,h}=\bdelta
\{\ind{i>\tau}-\ind{i+M+h+1>\tau}\}.
\]
Define
\[
\mathcal I_h^{\mathrm{aff}}=
\{1\le i\le N_h:\mathbf d_{i,h}\ne\mathbf0\text{ or }
       \mathbf d_{i+h,h}\ne\mathbf0\},
\]
where $N_h=n-M-2h-1$. It
satisfies $|\mathcal I_h^{\mathrm{aff}}|\le CM$, and
$\|\mathbf d_{i,h}\|_2\le\|\bdelta\|_2$.  {Expanding the definitions gives}
\begin{align*}
\widehat g_h-\widehat g_h^{\epsilon}
={}&\frac1{2N_h}\sum_{i=1}^{N_h}
(\mathbf D_{i+h,h}^{(M)})^\top\mathbf d_{i,h}
+\frac1{2N_h}\sum_{i=1}^{N_h}
\mathbf d_{i+h,h}^\top\mathbf D_{i,h}^{(M)}
+\frac1{2N_h}\sum_{i=1}^{N_h}
\mathbf d_{i+h,h}^\top\mathbf d_{i,h}.
\end{align*}
Since the vectors $\mathbf D_{i,h}^{(M)}$ are centered,
\[
\E(\widehat g_h-\widehat g_h^\epsilon)
=\frac1{2N_h}\sum_{i=1}^{N_h}
\mathbf d_{i+h,h}^\top\mathbf d_{i,h},
\qquad
|\E(\widehat g_h-\widehat g_h^\epsilon)|
\le C\frac{M\|\bdelta\|_2^2}{n}.
\]
\begingroup
Let $b_{i,h}=\ind{i>\tau}-\ind{i+M+h+1>\tau}$, so
$\mathbf d_{i,h}=b_{i,h}\bdelta$, and extend sums over $i$ by zero
outside $\{1,\ldots,N_h\}$. Then
\begin{align*}
w_{u,h}={}&\frac12\sum_{i=1}^{N_h}
b_{i,h}\{\ind{u=i+h}-\ind{u=i+M+2h+1}\}\\
&+\frac12\sum_{i=1}^{N_h}
b_{i+h,h}\{\ind{u=i}-\ind{u=i+M+h+1}\}.
\end{align*}
Consequently,
\begin{align*}
\widehat g_h-\widehat g_h^{\epsilon}
-\E(\widehat g_h-\widehat g_h^{\epsilon})
&=N_h^{-1}\sum_u w_{u,h}\bdelta^\top\bepsilon_u,\\
\max_u|w_{u,h}|&\le C,\qquad
\sum_u\ind{w_{u,h}\ne0}\le CM,\qquad
\sum_u w_{u,h}^2\le CM.
\end{align*}
Lemma~\ref{lem:vector-basic}(ii) gives
\[
\left\|\widehat g_h-\widehat g_h^{\epsilon}
-\E(\widehat g_h-\widehat g_h^{\epsilon})\right\|_2
\le\frac C{N_h}\|w_{\cdot,h}\|_2\|\bdelta\|_2
\le C\frac{M^{1/2}\|\bdelta\|_2}{n}.
\]
\endgroup
Since
\(
\left|
\sum_{i=1}^{n-h}c_{i,k}c_{i+h,k}
\right|
\le nv_k,
\) then $|a_{k,h}|\le Cv_kp^{-1/2}$,
$\sum_{h=0}^M|a_{k,h}|
\le CMv_kp^{-1/2}$.
Consequently, uniformly over $k\in\mathcal K_\gamma$, $|a_{\gamma,k,h}|\le Cp^{-1/2}$, $\sum_{h=0}^M|a_{\gamma,k,h}|\le CMp^{-1/2}$.
{It follows that}
\begin{align*}
&\max_{\gamma,k}
\left|\sum_{h=0}^Ma_{\gamma,k,h}
\E(\widehat g_h-\widehat g_h^{\epsilon})\right|
\le C\frac{M^2\|\bdelta\|_2^2}{n\sqrt p},
\end{align*}
and, by Minkowski's inequality,
\begin{align*}
&\left\|
\max_{\gamma,k}
\left|\sum_{h=0}^Ma_{\gamma,k,h}
\{(\widehat g_h-\widehat g_h^{\epsilon})
-\E(\widehat g_h-\widehat g_h^{\epsilon})\}\right|
\right\|_2\\
&\qquad\le\frac C{\sqrt p}\sum_{h=0}^M
\left\|(\widehat g_h-\widehat g_h^{\epsilon})
-\E(\widehat g_h-\widehat g_h^{\epsilon})\right\|_2
\le C\frac{M^{3/2}\|\bdelta\|_2}
{n\sqrt p}.
\end{align*}
Consequently,
\begin{align}
\max_{\gamma,k}|\widehat\mu_{\gamma,k,M}-\widehat\mu_{\gamma,k,M}^\epsilon|=O_p\left\{\frac{M^2\|\bdelta\|_2^2}{n\sqrt p}+\frac{M^{3/2}\|\bdelta\|_2}{n\sqrt p}\right\}.
\label{eq:S10-signal-contamination}
\end{align}

We next control the second difference in
\eqref{eq:S10-centering-decomposition}.  Stationarity of the noise
process gives
\begin{align*}
\E\{(\mathbf D_{i+h,h}^{(M)})^\top
\mathbf D_{i,h}^{(M)}\}
={}&2g_h-g_{M+1}-g_{M+2h+1},
\end{align*}
and
\[
\E(\widehat g_h^\epsilon)-g_h
=-\frac12g_{M+1}-\frac12g_{M+2h+1}.
\]
By Lemma~\ref{lem:vector-basic}(i),
\(
|g_r|\le p\|\bGam(r)\|_{\mathrm{op}}
\le p\varpi_{\mathrm{dep}}(r),
\)
so, uniformly for $0\le h\le M$,
\[
|\E(\widehat g_h^\epsilon)-g_h|
\le Cp\{\varpi_{\mathrm{dep}}(M+1)
        +\varpi_{\mathrm{dep}}(M+2h+1)\}
\le Cp\eta_M.
\]
It follows that
\begin{align}
\max_{\gamma,k}
\left|\sum_{h=0}^M a_{\gamma,k,h}
\{\E(\widehat g_h^\epsilon)-g_h\}\right|
\le C\sqrt p\,M\eta_M.
\label{eq:S10-noise-bias}
\end{align}
Since $M=o(n)$, $N_h\ge n/2$ uniformly for $h\le M$ and all
sufficiently large $n$.  Lemma~\ref{lem:difference-product-general}
therefore yields
\[
\Var(\widehat g_h^\epsilon)
\le\frac{C}{N_h^2}\sum_{i_1,i_2=1}^{N_h}
|\Cov(Z_{i_1,h},Z_{i_2,h})|
\le C\frac pn.
\]
Define
\[
\mathbf X_g^\epsilon
=(\widehat g_0^\epsilon-\E\widehat g_0^\epsilon,\ldots,
\widehat g_M^\epsilon-\E\widehat g_M^\epsilon)^\top,
\qquad
\mathbf a_{\gamma,k}
=(a_{\gamma,k,0},\ldots,a_{\gamma,k,M})^\top.
\]
Then
\[
\max_{\gamma,k}\|\mathbf a_{\gamma,k}\|_2^2
\le C\frac Mp,
\qquad
\E\|\mathbf X_g^\epsilon\|_2^2
=\sum_{h=0}^M\Var(\widehat g_h^\epsilon)
\le C\frac{Mp}{n}.
\]
Consequently,
\begin{align*}
\E\left[
\max_{\gamma,k}
\left|\sum_{h=0}^Ma_{\gamma,k,h}
(\widehat g_h^\epsilon-\E\widehat g_h^\epsilon)\right|^2
\right]
\le
\max_{\gamma,k}\|\mathbf a_{\gamma,k}\|_2^2
\E\|\mathbf X_g^\epsilon\|_2^2
\le C\frac{M^2}{n}.
\end{align*}
Together with \eqref{eq:S10-noise-bias}, this gives
\begin{align}
\max_{\gamma,k}
|\widehat\mu_{\gamma,k,M}^\epsilon-\mu_{\gamma,k,M}|
=O_p\left(\frac M{\sqrt n}+\sqrt p\,M\eta_M\right).
\label{eq:S10-noise-centering}
\end{align}

Finally, Assumption~\ref{ass:C2} gives
$(\bdelta^\top\bOme\bdelta)^{1/2}\le C_1\|\bdelta\|_2$.
Since $x\le1+x^2$ for $x\ge0$, and the choice of $M$ in
Theorem~\ref{Th2} satisfies $1\le M\le np$ for all sufficiently large
$n$,
\[
\frac{M^{3/2}\|\bdelta\|_2}
{n\sqrt p}
\le C\left\{
\frac{M^{3/2}}{n\sqrt p}
+\frac{M^{3/2}\|\bdelta\|_2^2}{n\sqrt p}
\right\}
\le C\left\{
\frac M{\sqrt n}
+\frac{M^2\|\bdelta\|_2^2}{n\sqrt p}
\right\}.
\]
Combining \eqref{eq:S10-centering-decomposition},
\eqref{eq:S10-signal-contamination}, and
\eqref{eq:S10-noise-centering} gives
\[
\max_{\gamma\in\{0,1/2\}}\max_{k\in\mathcal K_\gamma}
|\widehat\mu_{\gamma,k,M}-\mu_{\gamma,k,M}|
=O_p\left\{
\frac M{\sqrt n}
+\sqrt p\,M\eta_M
+\frac{M^2\|\bdelta\|_2^2}{n\sqrt p}
\right\}
=O_p(r_{\mu,n}).
\]

\end{proof}

\begin{lemma}\label{lem:scale-general}
Under the assumptions of Theorem~\ref{Th2}, $|\widehat\omega^2-\omega_p^2|=O_p(r_{\omega,n})$.
\end{lemma}

\begin{proof}
Define the untruncated scale estimator and its noise-only counterpart by
\[
\widetilde\omega^2
=\frac2p\sum_{h,\ell=0}^M\sum_{\nu\in\{\mathrm F,\mathrm C\}}
w_{h\ell}^{\nu}\widehat{\mathcal T}_{h\ell}^{\nu},
\qquad
\widetilde\omega_\epsilon^2
=\frac2p\sum_{h,\ell=0}^M\sum_{\nu\in\{\mathrm F,\mathrm C\}}
w_{h\ell}^{\nu}\widehat{\mathcal T}_{h\ell}^{\nu,\epsilon},
\]
where
\[
w_{00}^{\mathrm F}=1,\qquad
w_{h0}^{\mathrm F}=w_{0h}^{\mathrm F}=2\quad(1\le h\le M),
\qquad
w_{h\ell}^{\mathrm F}=w_{h\ell}^{\mathrm C}=2\quad(1\le h,\ell\le M).
\]
Thus
$\max_{h,\ell,\nu}|w_{h\ell}^{\nu}|\le2$, and
\begin{align}
\widetilde\omega^2-\omega_p^2
={}&
\{\widetilde\omega^2-\widetilde\omega_\epsilon^2\}
+\{\widetilde\omega_\epsilon^2-\omega_p^2\}.
\label{eq:S11-scale-decomposition}
\end{align}

We first control the first difference in
\begingroup
\eqref{eq:S11-scale-decomposition}.  Under model~\eqref{eq:single change-model}, recall
\endgroup
\[
\mathbf Y_{i,h}^{(M)}
=\mathbf D_{i,h}^{(M)}+\mathbf d_{i,h},
\qquad
\mathbf d_{i,h}=\bdelta
\{\ind{i>\tau}-\ind{i+M+h+1>\tau}\}.
\]
Define
\[
\mathcal I_{h,0}^{\mathrm{aff}}
=\{i:\mathbf d_{i,h}\ne\mathbf0\},
\qquad
\mathcal I_{h,1}^{\mathrm{aff}}
=\{i:\mathbf d_{i+h,h}\ne\mathbf0\}.
\]
Both sets have cardinality at most $CM$.  Using
$\mathbf Y_{i,h}^{(M)}=\mathbf D_{i,h}^{(M)}+\mathbf d_{i,h}$ and
{multilinearity, for $\nu\in\{\mathrm F,\mathrm C\}$ we obtain}
\begin{align}
\widehat{\mathcal T}_{h\ell}^{\nu}
-\widehat{\mathcal T}_{h\ell}^{\nu,\epsilon}
=\sum_{r=1}^4\Delta_{h\ell}^{\nu,(r)},
\label{eq:S11-delta-expansion}
\end{align}
\begingroup
To define these terms, let
$q_j=m_{h\ell,j}$ and $u_j=\tau_{h\ell,j}^{\nu}(i_1,i_2)$ using
\eqref{eq:difference-kernel-expansion}, and set
\[
\mathbf E_j^{A}(i_1,i_2)=
\begin{cases}
\mathbf d_{u_j,q_j},&j\in A,\\
\mathbf D_{u_j,q_j}^{(M)},&j\notin A.
\end{cases}
\]
Then, for $r=1,\ldots,4$,
\[
\Delta_{h\ell}^{\nu,(r)}
=\frac1{4|\mathcal I_{h\ell}|}
\sum_{\substack{A\subseteq\{1,2,3,4\}\\|A|=r}}
\sum_{(i_1,i_2)\in\mathcal I_{h\ell}}
\{(\mathbf E_1^A)^\top\mathbf E_2^A\}
\{(\mathbf E_3^A)^\top\mathbf E_4^A\}.
\]
\endgroup
The expressions indexed by $\nu=\mathrm F$ and $\nu=\mathrm C$ differ
only by interchanging the two vectors associated with $i_2$ and $i_2+\ell$.
This exchange changes neither the support counts nor the moment bounds,
so it suffices to give representative terms for $\nu=\mathrm F$ below;
all resulting estimates hold uniformly over
$\nu\in\{\mathrm F,\mathrm C\}$.

For $r=1$, a representative term is
\[
(\mathbf d_{i_1,h}^{\top}\mathbf D_{i_2,\ell}^{(M)})
\{(\mathbf D_{i_1+h,h}^{(M)})^\top
  \mathbf D_{i_2+\ell,\ell}^{(M)}\}.
\]
Thus,
\begingroup
\[
\begin{aligned}
&\left\|
\frac{1}{4|\mathcal I_{h\ell}|}
\sum_{(i_1,i_2)\in\mathcal I_{h\ell}}
(\mathbf d_{i_1,h}^{\top}\mathbf D_{i_2,\ell}^{(M)})
\{(\mathbf D_{i_1+h,h}^{(M)})^\top
\mathbf D_{i_2+\ell,\ell}^{(M)}\}
\right\|_2\\
&\quad\le
\frac{C}{n^2}
\sum_{i_1\in\mathcal I_{h,0}^{\mathrm{aff}}}\sum_{i_2}
\left\|
(\bdelta^\top\mathbf D_{i_2,\ell}^{(M)})
\{(\mathbf D_{i_1+h,h}^{(M)})^\top
\mathbf D_{i_2+\ell,\ell}^{(M)}\}
\right\|_2.
\end{aligned}
\]
\endgroup
Since $\mathbf D_{i_2,\ell}^{(M)}$ is a difference of two innovations,
Lemma~\ref{lem:vector-basic}(ii), applied with $q=4$, and
$\lambda_{\min}(\bOme)\ge C_0$ in Assumption~\ref{ass:C2} give
\[
\|\bdelta^\top\mathbf D_{i_2,\ell}^{(M)}\|_4
\le C\|\bdelta\|_2
\le C(\bdelta^\top\bOme\bdelta)^{1/2}.
\]
Lemma~\ref{lem:bilinear-quadratic-bound}, applied to the two
difference coefficient sequences, gives
\[
\left\|
(\mathbf D_{i_1+h,h}^{(M)})^\top\mathbf D_{i_2+\ell,\ell}^{(M)}
-\E\{(\mathbf D_{i_1+h,h}^{(M)})^\top
       \mathbf D_{i_2+\ell,\ell}^{(M)}\}
\right\|_4
\le C\sqrt p.
\]
Moreover, $(i_1,i_2)\in\mathcal I_{h\ell}$ separates the two differences by
at least $n/3$.  Lemma~\ref{lem:vector-basic}(i) therefore yields
\begin{align*}
\left|
\E\{(\mathbf D_{i_1+h,h}^{(M)})^\top
       \mathbf D_{i_2+\ell,\ell}^{(M)}\}
\right|
\le
p\left\|
\Cov(\mathbf D_{i_1+h,h}^{(M)},\mathbf D_{i_2+\ell,\ell}^{(M)})
\right\|_{\mathrm{op}}
\le Cp\,\varpi_{\mathrm{dep}}(n/3)
\le C\sqrt p.
\end{align*}
Consequently,
\[
\left\|(\mathbf D_{i_1+h,h}^{(M)})^\top
\mathbf D_{i_2+\ell,\ell}^{(M)}\right\|_4
\le C\sqrt p.
\]
H\"older's and Minkowski's inequalities now yield
\[
\|\Delta_{h\ell}^{\nu,(1)}\|_2
\le C\frac{M\sqrt p
(\bdelta^\top\bOme\bdelta)^{1/2}}n.
\]

For $r=2$, there are two types of terms, distinguished by the summation
indices of their deterministic vectors.  One representative term whose
two deterministic vectors are both indexed by $i_1$ is
\(
(\mathbf d_{i_1,h}^{\top}\mathbf D_{i_2,\ell}^{(M)})
(\mathbf d_{i_1+h,h}^{\top}\mathbf D_{i_2+\ell,\ell}^{(M)}).
\)
This term can be nonzero only if
\(
i_1\in\mathcal I_{h,0}^{\mathrm{aff}}
\cap\mathcal I_{h,1}^{\mathrm{aff}},
\)
which gives at most $CM$ choices for $i_1$.  For every such $i_1$, there
are at most $Cn$ admissible values of $i_2$ satisfying
$(i_1,i_2)\in\mathcal I_{h\ell}$.  It is therefore supported on at most
$CnM$ pairs.  The analogous term containing
$\mathbf d_{i_2,\ell}$ and $\mathbf d_{i_2+\ell,\ell}$ has the same support count. Another representative term with one deterministic vector indexed by $i_1$
and the other indexed by $i_2$ is
\(
(\mathbf d_{i_1,h}^{\top}\mathbf d_{i_2,\ell})
\{(\mathbf D_{i_1+h,h}^{(M)})^\top
  \mathbf D_{i_2+\ell,\ell}^{(M)}\}.
\)
This term vanishes unless
\(
i_1\in\mathcal I_{h,0}^{\mathrm{aff}},
i_2\in\mathcal I_{\ell,0}^{\mathrm{aff}}.
\)
Thus there are at most $CM^2$ admissible pairs $(i_1,i_2)$.  Moreover,
\begin{align*}
\left\|
(\mathbf d_{i_1,h}^{\top}\mathbf D_{i_2,\ell}^{(M)})
(\mathbf d_{i_1+h,h}^{\top}\mathbf D_{i_2+\ell,\ell}^{(M)})
\right\|_2
\le C\|\bdelta\|_2^2,
\qquad
\left\|
(\mathbf d_{i_1,h}^{\top}\mathbf d_{i_2,\ell})
\{(\mathbf D_{i_1+h,h}^{(M)})^\top
  \mathbf D_{i_2+\ell,\ell}^{(M)}\}
\right\|_2
\le C\sqrt p\,\|\bdelta\|_2^2.
\end{align*}
Consequently,
\[
\|\Delta_{h\ell}^{\nu,(2)}\|_2
\le C\left\{
\frac{M\|\bdelta\|_2^2}{n}
+\frac{M^2\sqrt p\|\bdelta\|_2^2}{n^2}
\right\}.
\]

For $r=3$, a representative term is
\(
(\mathbf d_{i_1,h}^{\top}\mathbf d_{i_2,\ell})
(\mathbf d_{i_1+h,h}^{\top}\mathbf D_{i_2+\ell,\ell}^{(M)}).
\)
\begingroup
Its support and moment bounds are
\begin{align*}
\#\{(i_1,i_2):\mathbf d_{i_1,h}\ne0,\ \mathbf d_{i_2,\ell}\ne0\}&\le CM^2,\\
\|(\mathbf d_{i_1,h}^\top\mathbf d_{i_2,\ell})
 (\mathbf d_{i_1+h,h}^\top\mathbf D_{i_2+\ell,\ell}^{(M)})\|_2
&\le\|\bdelta\|_2^2\|\bdelta^\top\mathbf D_{i_2+\ell,\ell}^{(M)}\|_2
\le C\|\bdelta\|_2^3.
\end{align*}
Thus
\[
\|\Delta_{h\ell}^{\nu,(3)}\|_2
\le Cn^{-2}M^2\|\bdelta\|_2^3.
\]
\endgroup

{For $r=4$, the term}
\(
(\mathbf d_{i_1,h}^{\top}\mathbf d_{i_2,\ell})
(\mathbf d_{i_1+h,h}^{\top}\mathbf d_{i_2+\ell,\ell})
\)
is deterministic, is supported on at most $CM^2$ pairs, and is bounded
by $\|\bdelta\|_2^4$.  Therefore
\[
|\Delta_{h\ell}^{\nu,(4)}|
\le C\frac{M^2\|\bdelta\|_2^4}{n^2}.
\]
Denote
\(
B_{\delta,n}=\frac{M^2\|\bdelta\|_2^2}{n\sqrt p}.
\)
Since $M\lesssim\sqrt p$ and $M^2\lesssim n$, summing the preceding estimates
over the $O(M^2)$ triples $(h,\ell,\nu)$ and using
$x^{3/2}\le x+x^2$ for $x\ge0$ give
\begin{align*}
&\left\|
\frac2p\sum_{h,\ell=0}^M\sum_{\nu}
 w_{h\ell}^{\nu}
\Delta_{h\ell}^{\nu,(1)}
\right\|_2
\le C\frac{M^3(\bdelta^\top\bOme\bdelta)^{1/2}}
{n\sqrt p},\\
&\frac2p\sum_{h,\ell=0}^M\sum_{\nu}
|w_{h\ell}^{\nu}|
\sum_{r=2}^4\|\Delta_{h\ell}^{\nu,(r)}\|_2
\le C\left\{\frac{M^2\|\bdelta\|_2^2}{n\sqrt p}+\frac{M^4\|\bdelta\|_2^4}{n^2p}\right\}.
\end{align*}
It follows that
\begin{align}
|\widetilde\omega^2-\widetilde\omega_\epsilon^2|
=O_p\left\{
\frac{M^3(\bdelta^\top\bOme\bdelta)^{1/2}}{n\sqrt p}
+\frac{M^2\|\bdelta\|_2^2}{n\sqrt p}+\frac{M^4\|\bdelta\|_2^4}{n^2p}
\right\}.
\label{eq:S11-signal-contamination}
\end{align}

We next control the second difference in
\eqref{eq:S11-scale-decomposition}.  Define
\[
\bOme_M=\bGam(0)+\sum_{h=1}^M
\{\bGam(h)+\bGam(h)^\top\},
\qquad
\omega_{p,M}^2=\frac2p\Tr(\bOme_M^2).
\]
We have the further decomposition
\[
\begin{aligned}
\widetilde\omega_\epsilon^2-\omega_p^2
=\{\widetilde\omega_\epsilon^2-\E\widetilde\omega_\epsilon^2\}
+\{\E\widetilde\omega_\epsilon^2-\omega_{p,M}^2\}
+\{\omega_{p,M}^2-\omega_p^2\}.
\end{aligned}
\]

By Lemma~\ref{lem:difference-product-general}(ii),
\begin{align*}
\Var(\widetilde\omega_\epsilon^2)
=\Var\left[
\frac2p\sum_{h,\ell=0}^M\sum_{\nu}
w_{h\ell}^{\nu}\widehat{\mathcal T}_{h\ell}^{\nu,\epsilon}
\right]
\le\frac{CM^4}{p^2}
\left(\frac p{n}+\frac{p^2}{n^2}\right)
\le CM^4\left(\frac1{np}+\frac1{n^2}\right).
\end{align*}
Therefore
\[
\widetilde\omega_\epsilon^2-
\E\widetilde\omega_\epsilon^2
=O_p\left\{M^2\left(\frac1{np}+\frac1{n^2}\right)^{1/2}\right\}.
\]

{To bound the bias, expand the expectations as}
\[
\E\widetilde\omega_\epsilon^2
=\frac2p\sum_{h,\ell=0}^M\sum_{\nu}
w_{h\ell}^{\nu}\E\widehat{\mathcal T}_{h\ell}^{\nu,\epsilon},
\qquad
\E\widehat{\mathcal T}_{h\ell}^{\nu,\epsilon}
=\frac1{4|\mathcal I_{h\ell}|}
\sum_{(i_1,i_2)\in\mathcal I_{h\ell}}
\E H_{h,\ell}^{\nu}(i_1,i_2).
\]
We now expand each expectation $\E H_{h,\ell}^{\nu}(i_1,i_2)$.  For centered vectors
$\mathbf x,\mathbf y,\mathbf u,\mathbf v$, the complete
product--cumulant identity is
\begin{align*}
\E\{(\mathbf x^\top\mathbf y)(\mathbf u^\top\mathbf v)\}
={}&\sum_{a,b}\E(x_ay_a)\E(u_bv_b)
+\sum_{a,b}\E(x_au_b)\E(y_av_b)\\
&+\sum_{a,b}\E(x_av_b)\E(y_au_b)
+\sum_{a,b}\Cum(x_a,y_a,u_b,v_b).
\end{align*}
For $\nu=\mathrm F$, take
\(
(\mathbf x,\mathbf y,\mathbf u,\mathbf v)
=(\mathbf D_{i_1,h}^{(M)},\mathbf D_{i_2,\ell}^{(M)},
  \mathbf D_{i_1+h,h}^{(M)},\mathbf D_{i_2+\ell,\ell}^{(M)}),
\)
and, for $\nu=\mathrm C$, interchange
$\mathbf D_{i_2,\ell}^{(M)}$ and $\mathbf D_{i_2+\ell,\ell}^{(M)}$.
Define
\[
\mathbf A_h^{(M)}
:=\Cov(\mathbf D_{i_1,h}^{(M)},\mathbf D_{i_1+h,h}^{(M)})
=\Cov(\mathbf D_{0,h}^{(M)},\mathbf D_{h,h}^{(M)}).
\]
A direct expansion gives
\[
\mathbf A_h^{(M)}
=2\bGam(h)-\bGam(M+2h+1)-\bGam(M+1)^\top,
\qquad
\|\mathbf A_h^{(M)}-2\bGam(h)\|_{\mathrm{op}}
\le C\eta_M.
\]
Consequently, the second covariance pairing produces the main trace:
\begin{align*}
\E H_{h,\ell}^{\mathrm F}(i_1,i_2)
=\Tr\{\mathbf A_h^{(M)}
        (\mathbf A_\ell^{(M)})^\top\}
 +R_{h,\ell}^{\mathrm F}(i_1,i_2),
 \qquad
\E H_{h,\ell}^{\mathrm C}(i_1,i_2)
=\Tr\{\mathbf A_h^{(M)}\mathbf A_\ell^{(M)}\}
 +R_{h,\ell}^{\mathrm C}(i_1,i_2).
\end{align*}
Here $R_{h,\ell}^{\nu}(i_1,i_2)$ is the sum of the first and third covariance
pairings in the complete identity and its fourth-cumulant term.  %
For $(i_1,i_2)\in\mathcal I_{h\ell}$, every primitive index associated with
$i_1$ is separated from every primitive index associated with $i_2$ by
at least $n/3$.  Lemma~\ref{lem:vector-basic}(i) therefore gives, 
\begin{align*}
\left|\sum_{a,b}\E(x_ay_a)\E(u_bv_b)\right|
\le Cp^2\varpi_{\mathrm{dep}}(n/3)^2,
\qquad
\left|\sum_{a,b}\E(x_av_b)\E(y_au_b)\right|
\le Cp\varpi_{\mathrm{dep}}(n/3)^2.
\end{align*}
\begingroup
Expanding the four differences produces $2^4$ fourth cumulants.
For each, the largest anchored relative lag is at least $n/6$.
Equation~\eqref{eq:S9-pointwise-lag-envelope} and
Assumption~\ref{ass:C3} therefore give
\[
\left|\sum_{a,b}\Cum(x_a,y_a,u_b,v_b)\right|
\le 2^4 Cn^{-8}
\mathfrak C_{\{\{1,2,3,4\}\},p}
\le Cp n^{-8}.
\]
\endgroup
Consequently, uniformly over $(i_1,i_2)\in\mathcal I_{h\ell}$,
\[
|R_{h,\ell}^{\nu}(i_1,i_2)|
\le C\{p^2\varpi_{\mathrm{dep}}(n/3)^2+pn^{-8}\},
\qquad \nu\in\{\mathrm F,\mathrm C\}.
\]
Moreover,
\(
\|\mathbf A_h^{(M)}-2\bGam(h)\|_{\mathrm F}
\le C\sqrt p\,\eta_M,
\|\mathbf A_h^{(M)}\|_{\mathrm F}
+\|\bGam(h)\|_{\mathrm F}\le C\sqrt p.
\)
Hence, by the Frobenius inner-product inequality,
\begin{align*}
&\left|\frac14\Tr\{\mathbf A_h^{(M)}
  (\mathbf A_\ell^{(M)})^\top\}-\mathcal T_{h\ell}^{\mathrm F}\right|
+\left|\frac14\Tr\{\mathbf A_h^{(M)}
  \mathbf A_\ell^{(M)}\}-\mathcal T_{h\ell}^{\mathrm C}\right|\\
&\qquad\le
C\|\mathbf A_h^{(M)}-2\bGam(h)\|_{\mathrm F}
  \|\mathbf A_\ell^{(M)}\|_{\mathrm F}
+C\|\bGam(h)\|_{\mathrm F}
  \|\mathbf A_\ell^{(M)}-2\bGam(\ell)\|_{\mathrm F}
\le Cp\eta_M.
\end{align*}
Combining this trace bound with the preceding representation of
$\E\widehat{\mathcal T}_{h\ell}^{\nu,\epsilon}$ gives
\begin{align*}
\left|\E\widehat{\mathcal T}_{h\ell}^{\nu,\epsilon}
-\mathcal T_{h\ell}^{\nu}\right|
\le Cp\{\eta_M+p\varpi_{\mathrm{dep}}(n/3)^2+n^{-8}\},
\qquad \nu\in\{\mathrm F,\mathrm C\}.
\end{align*}
Therefore,
\[
|\E\widetilde\omega_\epsilon^2-\omega_{p,M}^2|
\le CM^2\{\eta_M+p\varpi_{\mathrm{dep}}(n/3)^2+n^{-8}\}.
\]

Moreover,
\(
\|\bOme-\bOme_M\|_{\mathrm{op}}\le C\eta_M,
\|\bOme\|_{\mathrm{op}}+\|\bOme_M\|_{\mathrm{op}}\le C,
\)
and hence
\[
|\omega_{p,M}^2-\omega_p^2|
\le\frac2p\|\bOme_M-\bOme\|_{\mathrm F}
\|\bOme_M+\bOme\|_{\mathrm F}
\le C\eta_M.
\]

Consequently,
\begin{align}
|\widetilde\omega_\epsilon^2-\omega_p^2|
=O_p\left\{
M^2\left(\frac1{np}+\frac1{n^2}\right)^{1/2}
+M^2\eta_M
+M^2p\varpi_{\mathrm{dep}}(n/3)^2
+M^2n^{-8}
\right\}.
\label{eq:S11-noise-scale}
\end{align}

\begingroup
To simplify these rates, let
$A_n=M^2\{(np)^{-1}+n^{-2}\}^{1/2}$ and retain
$B_{\delta,n}=M^2\|\bdelta\|_2^2/(n\sqrt p)=o(1)$.
Since $M\asymp(n\wedge p)^{1/8}$, we have
\begin{align*}
\frac{M^4}{n\sqrt p}
&\le C A_n,\\
\frac{M^3(\bdelta^\top\bOme\bdelta)^{1/2}}{n\sqrt p}
&\le C\left(B_{\delta,n}\frac{M^4}{n\sqrt p}\right)^{1/2}
\le C(B_{\delta,n}+A_n),\\
B_{\delta,n}^2&=o(B_{\delta,n}),\qquad
M^2n^{-8}=o(A_n).
\end{align*}
The first inequality follows from $M^2\le C\sqrt n$ and
$A_n\ge M^2/\sqrt{np}$. Under (AM) or (PD),
\[
\frac{M^2p\varpi_{\mathrm{dep}}(n/3)^2}{A_n}
\le Cnp\exp(-cn^{\varkappa})
=\exp\{O(1)+\log n+\log p-cn^{\varkappa}\}=o(1);
\]
under (MD), its numerator is eventually zero.
\endgroup
Combining \eqref{eq:S11-scale-decomposition},
\eqref{eq:S11-signal-contamination}, and
\begingroup
\eqref{eq:S11-noise-scale} now gives
\endgroup
\begin{align}
|\widetilde\omega^2-\omega_p^2|
=O_p\Bigg\{&
M^2\left(\frac1{np}+\frac1{n^2}\right)^{1/2}
+M^2\eta_M+M^2p\varpi_{\mathrm{dep}}(n/3)^2+M^2n^{-8}\notag\\
&+\frac{M^3(\bdelta^\top\bOme\bdelta)^{1/2}}{n\sqrt p}
+\frac{M^2\|\bdelta\|_2^2}{n\sqrt p}
+\left(\frac{M^2\|\bdelta\|_2^2}{n\sqrt p}\right)^2
\Bigg\}
=O_p(r_{\omega,n}).
\label{eq:S11-sharp-scale-rate}
\end{align}

\begingroup
Since $\widehat\omega^2=\max\{\widetilde\omega^2,n^{-1}\}$ and $\omega_p^2\ge 2C_0^2>n^{-1}$ for all sufficiently large $n$,
the map $x\mapsto\max(x,n^{-1})$ is one-Lipschitz,
\[
|\widehat\omega^2-\omega_p^2|
=|\max(\widetilde\omega^2,n^{-1})-\max(\omega_p^2,n^{-1})|
\le|\widetilde\omega^2-\omega_p^2|
=O_p(r_{\omega,n}).
\]
\endgroup
\end{proof}

\begin{proof}[Proof of Theorem~\ref{Th2}]
Lemmas~\ref{lem:centering-general} and~\ref{lem:scale-general}
give the two stated stochastic rates.  In particular,
\begin{align}
\max_{1\le k<n}
\frac{|\widehat\mu_{k,M}-\mu_{k,M}|}{v_k}
&\le
\max_{1\le k<n}\sum_{h=0}^M
\frac{|a_{k,h}|}{v_k}|\widehat g_h-g_h|\le \frac{2}{\sqrt p}\sum_{h=0}^M
|\widehat g_h-g_h|
=O_p(r_{\mu,n}).
\label{eq:Th2-weighted-centering}
\end{align}
We first verify that $r_{\mu,n}\to0$.  Since
$M\asymp(n\wedge p)^{1/8}$,
\begingroup
$M/\sqrt n\to0$.  Under (MD), $\eta_M=0$ eventually. Under (AM) or (PD),
\begin{align*}
\sqrt p\,M\eta_M
&\le C\exp\{\tfrac12\log p+\log M-c(n\wedge p)^{\varkappa/8}\}
=o(1),\\
\tfrac12\log p+\log M
&=\begin{cases}
O(\log p)=o(p^{\varkappa/8}),&p\le n,\\
o(n^{\varkappa/8}),&p>n.
\end{cases}
\end{align*}
\endgroup
The signal condition in Theorem~\ref{Th2}  gives
$M^2\|\bdelta\|_2^2/(n\sqrt p)=o(1)$, and hence
$r_{\mu,n}\to0$.

We next verify that
$r_{\omega,n}\to0$.  Let $m_n=n\wedge p$.  Since
$M\asymp m_n^{1/8}$,
\[
M^4\left(\frac1{np}+\frac1{n^2}\right)
\le C\left(\frac1{n\sqrt p}+\frac1{n^{3/2}}\right)=o(1),
\qquad
M^2n^{-8}=o(1).
\]
Under (MD), $\eta_M=0$ and
$\varpi_{\mathrm{dep}}(n/3)=0$ eventually.  Under (AM) or (PD),
$\log p=o(n^{\varkappa/8})$ and the exponential decay of
$\eta_M$ and $\varpi_{\mathrm{dep}}$ give
\[
M^2\eta_M+M^2p\varpi_{\mathrm{dep}}(n/3)^2=o(1).
\]
Assumption~\ref{ass:C2} gives
\[
\left\{
\frac{M^3(\bdelta^\top\bOme\bdelta)^{1/2}}{n\sqrt p}
\right\}^2
\le C \frac{M^2\|\bdelta\|_2^2}{n\sqrt p}\frac{M^4}{n\sqrt p}
=o(1).
\]
Hence every term in $r_{\omega,n}$ tends to zero.

Finally, $\omega_p^2=2p^{-1}\Tr(\bOme^2)\ge2C_0^2$ and
$|\widehat\omega^2-\omega_p^2|=O_p(r_{\omega,n})=o_p(1)$.  Thus
$\Pr(\widehat\omega\ge C_0)\to1$, and, on this event,
\[
\left|\frac{\widehat\omega}{\omega_p}-1\right|
=\frac{|\widehat\omega^2-\omega_p^2|}
{\omega_p(\widehat\omega+\omega_p)}
\le C|\widehat\omega^2-\omega_p^2|
=O_p(r_{\omega,n}).
\]
\end{proof}

\subsection{{Proof of Theorem~\ref{Tms-1} and Proposition~\ref{prop:Tms-alter}}
}

For $t\in(0,1)$ and $\gamma\in\{0,1/2\}$, define $v(t)=t(1-t)$.
For a deterministic vector $\mathbf h\in\mathbb R^p$, let
\begin{equation}
\begin{aligned}
&Q_{\gamma,n,p}(t;\mathbf h)
:=v(t)^{-2\gamma}\frac{
\|\mathbf U_n(t)+\mathbf h\|_2^2
-\E\|\mathbf U_n(t)\|_2^2
}{\sqrt p},\\
&Q_{\gamma,n,p}^{G}(t;\mathbf h)
:=v(t)^{-2\gamma}\frac{
\|\mathbf U_n^G(t)+\mathbf h\|_2^2
-\E\|\mathbf U_n^G(t)\|_2^2
}{\sqrt p},\\
&G_{p,\gamma}(t;\mathbf h)
:=v(t)^{-2\gamma}\frac{
\|\bOme^{1/2}\mathbf B(t)+\mathbf h\|_2^2
-\E\|\bOme^{1/2}\mathbf B(t)\|_2^2
}{\sqrt p\,\omega_p}.
\end{aligned}
\label{eq:def-boundary-quadratic-shifted}
\end{equation}

{For $k\in\mathcal K_\gamma$, let $t_k=k/n$ and $v_k=v(t_k)$.
For a deterministic family $\mathbf h=(\mathbf h_k:k\in\mathcal K_\gamma)$,
define the quadratic arrays}
\begin{equation}
\begin{aligned}
\mathcal Q_{\gamma,n}(\mathbf h)
&=\{Q_{\gamma,n,p}(t_k;\mathbf h_k):k\in\mathcal K_\gamma\},\\
\mathcal Q_{\gamma,n}^G(\mathbf h)
&=\{Q_{\gamma,n,p}^G(t_k;\mathbf h_k):k\in\mathcal K_\gamma\},\\
\mathcal Q_{\gamma,n}^{B}(\mathbf h)
&=\{G_{p,\gamma}(t_k;\mathbf h_k):k\in\mathcal K_\gamma\}.
\end{aligned}
\label{eq:def-boundary-quadratic-arrays-shifted}
\end{equation}
For zero shifts, with
$\mathbf0=(\mathbf0:k\in\mathcal K_\gamma)$, let
\begin{align*}
Q_{\gamma,n,p}(t_k)&=Q_{\gamma,n,p}(t_k;\mathbf0),&
Q_{\gamma,n,p}^{G}(t_k)&=Q_{\gamma,n,p}^{G}(t_k;\mathbf0),&
G_{p,\gamma}(t_k)&=G_{p,\gamma}(t_k;\mathbf0),\\
\mathcal Q_{\gamma,n}&=\mathcal Q_{\gamma,n}(\mathbf0),&
\mathcal Q_{\gamma,n}^{G}&=\mathcal Q_{\gamma,n}^{G}(\mathbf0),&
\mathcal Q_{\gamma,n}^{B}&=\mathcal Q_{\gamma,n}^{B}(\mathbf0).
\end{align*}
Furthermore, we define %
three corresponding max-$L_2$ statistics, for $\gamma\in\{0,1/2\}$,
\begin{equation}
\begin{aligned}
S_{n,\gamma}^{\epsilon}
&=\max_{k\in\mathcal K_\gamma}
\frac{Q_{\gamma,n,p}(t_k)}{\omega_p},
\qquad
S_{n,\gamma}^{G}
&=\max_{k\in\mathcal K_\gamma}
\frac{Q_{\gamma,n,p}^{G}(t_k)}{\omega_p},
\qquad
S_{n,\gamma}^{B}
&=\max_{k\in\mathcal K_\gamma}G_{p,\gamma}(t_k).
\end{aligned}
\label{eq:def-dense-max-statistics}
\end{equation}

For a finite vector $\mathbf z=(z_1,\ldots,z_N)^\top$, define
\[
\operatorname{smax}_{\beta}(\mathbf z)
=
\beta^{-1}\log\left(\sum_{a=1}^N e^{\beta z_a}\right).
\]
We repeatedly use the elementary log-sum-exp bound
\begin{equation}
0\le
\operatorname{smax}_{\beta}(\mathbf z)-\max_{1\le a\le N}z_a
\le
\beta^{-1}\log N,
\label{eq:smax-logsumexp-bound}
\end{equation}
and, by \citet{chernozhukov2015comparison},
\begin{equation}
\sum_a |\partial_a\operatorname{smax}_{\beta}|
=1,\qquad
\sum_{a,b} |\partial_{ab}\operatorname{smax}_{\beta}|
\le 2\beta,\qquad
\sum_{a,b,c} |\partial_{abc}\operatorname{smax}_{\beta}|
\le 6\beta^2.
\label{eq:smax-derivative-bounds}
\end{equation}

\begin{lemma}%
\label{lem:quadratic-only-uniform-comparison}
Suppose Assumptions~\ref{ass:C1}--\ref{ass:C3} hold. When
$\gamma=1/2$, assume additionally that
\begingroup
$\lambda_n\asymp n^\lambda$ for some fixed $\lambda\in(0,1)$.
Let $\lambda_0=1$ and $\lambda_{1/2}=\lambda$. Under (AM) or (PD),
assume $\log p=O(n^\nu)$ for some $0<\nu<\lambda_\gamma\varkappa$.
Let $\ell_n=\log(n)+1$, $\beta_n=\ell_n^2$,
$c_{\mathrm q}=4+3c_{\mathrm{dep}}$, and
$\kappa_{\mathrm q}=2c_{\mathrm q}+4$, where $c_{\mathrm{dep}}$ is
given in \eqref{eq:branchwise-cdep-choice}. Set
\[
s_n=\begin{cases}
\lceil c_0\{\log(np)\}^{1/\varkappa}\rceil,&\text{under (AM) or (PD)},\\
m_0+1,&\text{under (MD)},
\end{cases}
\qquad b_n=\lceil s_n\ell_n^{\kappa_{\mathrm q}}\rceil.
\]
For every fixed $A>0$, choose $c_0=c_0(A)$ sufficiently large. Then, for $\gamma\in\{0,1/2\}$, a deterministic family
\endgroup
$\mathbf h=(\mathbf h_k:k\in\mathcal K_\gamma)$, 
and every $f\in C_b^3(\mathbb R)$,
\begin{align}
&\left|
\E f\!\left\{
\operatorname{smax}_{\beta_n}
\left(\omega_p^{-1}\mathcal Q_{\gamma,n}(\mathbf h)\right)
\right\}
-
\E f\!\left\{
\operatorname{smax}_{\beta_n}
\left(\omega_p^{-1}\mathcal Q_{\gamma,n}^{G}(\mathbf h)\right)
\right\}
\right|
\le
C\|f\|_{C^3}r_{n,\gamma}^{\mathrm{quad}}
\{1+H_{\gamma,n}(\mathbf h)\}^{3},
\label{eq:quadratic-only-block-comparison}
\end{align}
where
\[
r_{n,\gamma}^{\mathrm{quad}}
=
\ell_n^{c_{\mathrm q}}
\left\{
\left(\frac{s_n}{b_n}\right)^{1/2}
+\left(\frac{b_n}{\Lambda_{\gamma,n}}\right)^{1/2}
\right\}
+n^{-A},
\quad
H_{\gamma,n}(\mathbf h)
=
\max_{k\in\mathcal K_\gamma}
v_k^{-\gamma}\frac{\|\mathbf h_k\|_2}{\sqrt p},
\]
\begingroup
with $\Lambda_{0,n}=n$ and $\Lambda_{1/2,n}=\lambda_n$. In particular, $(\log\log n)^{3/2}r_{n,\gamma}^{\mathrm{quad}}\to0$.
\endgroup
\end{lemma}

\begin{proof}
\begingroup
We prove the comparison by deleting gaps, making the retained blocks
independent, replacing them by Gaussian blocks, and restoring the
covariance. {To make the comparison uniform over $k$, Step~2 bounds the maxima of the random Taylor derivatives.}
\endgroup
Taking $N=|\mathcal K_\gamma|\le n$ and $\beta=\beta_n$ in
\eqref{eq:smax-logsumexp-bound} gives, for
$\mathbf z=(z_k:k\in\mathcal K_{\gamma})$,
\begin{equation}
0\le
\operatorname{smax}_{\beta_n}(\mathbf z)-\max_{k\in\mathcal K_{\gamma}} z_k
\le\frac{\log n}{\beta_n}.
\label{eq:quadratic-only-smax-error}
\end{equation}
The derivative bounds used below are \eqref{eq:smax-derivative-bounds} with
$N=|\mathcal K_\gamma|$ and $\beta=\beta_n$.

We follow the four-step route of
Lemma~\ref{lem:block-comparison-finite}, now uniformly over
$k\in\mathcal K_{\gamma}$:
\[
\mathcal Q_{\gamma,n}(\mathbf h)
\rightarrow \mathcal Q_{\gamma,n}^{I}(\mathbf h)
\rightarrow \mathcal Q_{\gamma,n}^{I,\mathrm{ind}}(\mathbf h)
\rightarrow \mathcal Q_{\gamma,n}^{I,G,\mathrm{ind}}(\mathbf h)
\rightarrow \mathcal Q_{\gamma,n}^{G}(\mathbf h).
\]
The replacement arguments in these four steps are the same as those in
Lemma~\ref{lem:block-comparison-finite}.  It remains to verify that the
block-weight, moment, smooth-maximum derivative, and covariance-discrepancy
bounds used there hold uniformly over $k\in\mathcal K_{\gamma}$; these
uniform bounds are established below. 

Step 1. Partition $\{1,\ldots,n\}$ into retained blocks $I_\ell$
of length $b_n$, separated by deleted blocks $J_\ell$ of length $s_n$,
and define
\begin{equation}
a_{i,(\gamma,k)}
=v_k^{-\gamma}\frac{c_{i,k}}{\sqrt n},
\qquad
w_{\ell, (\gamma,k)}^2
=\sum_{i\in I_\ell}a_{i,(\gamma,k)}^2.
\label{eq:quadratic-only-weighted-coefficients}
\end{equation}
{For $\gamma=0$, the coefficient bounds in the proof of Lemma~\ref{lem:block-comparison-finite} give} 
\begingroup
\begin{align}
\sum_\ell\sup_{k\in\mathcal K_{0}}w_{\ell, (0,k)}^3
&\le C\frac{n}{b_n}\left(\frac{b_n}{n}\right)^{3/2}\le C\left(\frac{b_n}{n}\right)^{1/2},
\label{eq:quadratic-only-gamma0-block-weights}
\\
\sup_{k\in\mathcal K_{0}}
\sum_{i\in\cup_\ell J_\ell}a_{i,(0,k)}^2
&\le C\left\{
\frac{s_n}{b_n}+\frac{b_n}{n}
\right\}:=\{r_{n,0}^{\mathrm{del}}\}^2.
\label{eq:quadratic-only-gamma0-deleted-weights}
\end{align}
For $\gamma=1/2$ and $k\le n/2$, $a^2_{i,(1/2,k)}\le C/k$ if
$i\le k$, whereas $a^2_{i,(1/2,k)}\le1/n$ if $i>k$.
Let $r_\ell$ be the left endpoint of $I_\ell$ and
$r_\ell'=n-r_\ell-b_n+2$ the left endpoint of its reflected block.
First restrict the supremum to $k\le n/2$.
\endgroup
\[
w_{\ell,(1/2,k)}^2 = \sum_{i\in I_\ell\cap[1,k]}a_{i,(1/2,k)}^2
+\sum_{i\in I_\ell\cap(k,n]}a_{i,(1/2,k)}^2
\le C\frac{b_n}{\max(k,r_\ell,\lambda_n)}+C\frac {b_n}{n}.
\]
\begingroup
Since \(\lambda_n\asymp n^\lambda\), \(b_n=o(\lambda_n)\), \(
|\mathcal J\cap[1,k]|\le C\left(ks_n/b_n+b_n\right),
m_{\mathcal J}=|\mathcal J|\le C\left(ns_n/b_n+b_n\right),
\), it follows that
\begin{align}
\sum_{\ell}\sup_{k\in\mathcal K_{1/2}} w_{\ell,(1/2,k)}^3
&\le C\sum_{\ell=1}^L\left[
 \left\{\frac{b_n}{\max(\lambda_n,r_\ell)}\right\}^{3/2}
+\left\{\frac{b_n}{\max(\lambda_n,r_\ell')}\right\}^{3/2}
+\left(\frac{b_n}{n}\right)^{3/2}\right]\nonumber\\
&\le C\sum_{r\ge0}\left(
\frac{b_n}{\lambda_n+r b_n}\right)^{3/2}
+C\frac{n}{b_n}\left(\frac{b_n}{n}\right)^{3/2}\nonumber\\
&\le{} C\left\{
\left(\frac{b_n}{n}\right)^{1/2}
+\left(\frac{b_n}{\lambda_n}\right)^{1/2}
\right\},
\label{eq:quadratic-only-block-weights}\\
\sup_{k\in\mathcal K_{1/2}}
\sum_{i\in\cup_\ell J_\ell}a_{i,(1/2,k)}^2
&\le{} C\left\{\frac{|\mathcal J\cap[1,k]|}{k}
+\frac{k|\mathcal J\cap(k,n]|}{n^2}\right\} \nonumber\\
&\le{} C\left\{
\frac{s_n}{b_n}+\frac{b_n}{n}+\frac{b_n}{\lambda_n}
\right\}:=\{r_{n,1/2}^{\mathrm{del}}\}^2.
\label{eq:quadratic-only-deleted-weights}
\end{align}
\endgroup
Time reversal gives the same bound for $k>n/2$.

For $\gamma\in\{0,1/2\}$, recall $\mathcal J=\cup_\ell J_\ell$, define
\begin{equation}
a_{i,(\gamma,k)}^{I}=a_{i,(\gamma,k)}\ind{i\notin\mathcal J},
\qquad
a_{i,(\gamma,k)}^{\mathrm{del}}=a_{i,(\gamma,k)}\ind{i\in\mathcal J},
\label{eq:quadratic-only-retained-deleted-coefficients}
\end{equation}
and let
\begin{equation}
\mathbf U_{\gamma,k}=\sum\limits_{i=1}^n a_{i,(\gamma,k)}\mathbf \epsilon_i,
\quad
\mathbf U_{\gamma,k}^{I}
=\sum\limits_{i=1}^n a_{i,(\gamma,k)}^{I}\mathbf \epsilon_i, 
\quad
\mathbf U_{\gamma,k}^{\mathrm{del}}=\mathbf U_{\gamma,k}-\mathbf U_{\gamma,k}^{I}
=\sum\limits_{i=1}^n a_{i,(\gamma,k)}^{\mathrm{del}}\mathbf \epsilon_i.
\label{eq:quadratic-only-retained-deleted-cusums}
\end{equation}
\begingroup
Within this proof abbreviate $Q_{\gamma,n,p}$ to $Q_{\gamma,n}$.
Then
\endgroup
\begin{align*}
Q_{\gamma,n}(t_k;\mathbf h_k)=
\frac{1}{\sqrt p}\{\|\mathbf U_{\gamma,k}+\widetilde{\mathbf h}_{\gamma,k}\|_2^2
-\E\|\mathbf U_{\gamma,k}\|_2^2\},
\\
Q_{\gamma,n}^{I}(t_k;\mathbf h_k)=
\frac{1}{\sqrt p}\{\|\mathbf U_{\gamma,k}^{I}+\widetilde{\mathbf h}_{\gamma,k}\|_2^2
-\E\|\mathbf U_{\gamma,k}^{I}\|_2^2\},
\end{align*}
where $\widetilde{\mathbf h}_{\gamma,k}=v_k^{-\gamma}\mathbf h_k$.
Let
$\mathcal Q_{\gamma,n}^{I}(\mathbf h)
=\{Q_{\gamma,n}^{I}(t_k;\mathbf h_k):k\in\mathcal K_\gamma\}$.

For $\mathbf h_k=\mathbf0$, let
$Q_{\gamma,n}^{I}(t_k)=Q_{\gamma,n}^{I}(t_k;\mathbf0)$. 
Furthermore, we define the coefficients
\[
\mathbf a_{\gamma,k}=(a_{i,(\gamma,k)})_{i=1}^n,
\qquad
\mathbf a_{\gamma,k}^I=(a_{i,(\gamma,k)}^I)_{i=1}^n,
\qquad
\mathbf a_{\gamma,k}^{\mathrm{del}}
=(a_{i,(\gamma,k)}^{\mathrm{del}})_{i=1}^n,
\]
and, for $k_1<k_2$,
\[
\Delta\mathbf a_{\gamma}
=\mathbf a_{\gamma,k_2}-\mathbf a_{\gamma,k_1},
\qquad
\Delta\mathbf a_{\gamma}^I
=\mathbf a_{\gamma,k_2}^I-\mathbf a_{\gamma,k_1}^I,
\qquad
\Delta\mathbf a_{\gamma}^{\mathrm{del}}
=\mathbf a_{\gamma,k_2}^{\mathrm{del}}
-\mathbf a_{\gamma,k_1}^{\mathrm{del}}.
\]
The full coefficient norms hold for every $1\le k<n$, whereas the
retained and deleted coefficient bounds below are restricted to the
indicated sets:
\begin{equation}
\begin{aligned}
&\|\mathbf a_{0,k}\|_2=\sqrt{v_k}\le\frac12,
&\quad
\|\mathbf a_{1/2,k}\|_2&=1,
&&1\le k<n,\\
&\|\mathbf a_{0,k}^I\|_2\le\sqrt{v_k}\le\frac12,
&
\|\mathbf a_{0,k}^{\mathrm{del}}\|_2&\le r_{n,0}^{\mathrm{del}},
&&k\in\mathcal K_0,\\
&\|\mathbf a_{1/2,k}^I\|_2\le1,
&\quad
\|\mathbf a_{1/2,k}^{\mathrm{del}}\|_2&\le r_{n,1/2}^{\mathrm{del}},
&&k\in\mathcal K_{1/2}.
\end{aligned}
\label{eq:quadratic-only-coefficient-norms}
\end{equation}
For $\gamma=1/2$ and every $1\le k<n$, define
$u_k=\log\{k/(n-k)\}$ and rewrite
\[
a_{i,(1/2,k)}
=n^{-1/2}\exp(-u_k/2)\ind{i\le k}
-n^{-1/2}\exp(u_k/2)\ind{i>k}.
\]
For every $1\le k_1<k_2<n$, direct calculation gives
\begin{equation}
\begin{aligned}
\|\Delta\mathbf a_0\|_2^2
&=\frac{k_2-k_1}{n}
\left(1-\frac{k_2-k_1}{n}\right)
\le\frac{k_2-k_1}{n},\\
\|\Delta\mathbf a_{1/2}\|_2^2
&=
2-2\langle\mathbf a_{1/2,k_1},\mathbf a_{1/2,k_2}\rangle=
2\left[
1-\exp\left\{-\frac{u_{k_2}-u_{k_1}}2\right\}
\right]
\le u_{k_2}-u_{k_1}.
\end{aligned}
\label{eq:quadratic-only-full-increment}
\end{equation}
By $\|\Delta\mathbf a_{\gamma}^I\|_2^2\le \|\Delta\mathbf a_{\gamma}\|_2^2$, we have
\begin{equation}
\|\Delta\mathbf a_{\gamma}^I\|_2^2
\lesssim
\begin{cases}
(k_2-k_1)/n, & \gamma=0,\\
|u_{k_2}-u_{k_1}|, & \gamma=1/2.
\end{cases}
\label{eq:quadratic-only-retained-increment}
\end{equation}
For $\gamma=0$, let
\(
d=d(k_1,k_2)=(k_2-k_1)/n,
m_{\mathcal J}=|\mathcal J|,
q_{\mathcal J}(k_1,k_2)=|\mathcal J\cap(k_1,k_2]|.
\)
Direct subtraction of the unstandardized CUSUM coefficients gives
\begin{align}
\|\Delta\mathbf a_0^{\mathrm{del}}\|_2^2
&\le \frac{q_{\mathcal J}(k_1,k_2)}n
+\frac{m_{\mathcal J}d^2}{n}.
\label{eq:quadratic-only-gamma0-deleted-increment}
\end{align}
For $\gamma=1/2$, let $\Delta u=u_{k_2}-u_{k_1}\ge0$.
For $i\le k_1$ and $i>k_2$, respectively,
\[
a_{i,(1/2,k_2)}=\exp(-\Delta u/2)a_{i,(1/2,k_1)},
\qquad
a_{i,(1/2,k_1)}=\exp(-\Delta u/2)a_{i,(1/2,k_2)}.
\]
Thus
\begin{align*}
&\sum_{i\in\mathcal J\cap([1,k_1]\cup(k_2,n])}
|a_{i,(1/2,k_2)}-a_{i,(1/2,k_1)}|^2\\
&\quad\le C\Delta u\left\{
\sum_{i\in\mathcal J\cap[1,k_1]}a_{i,(1/2,k_1)}^2
+\sum_{i\in\mathcal J\cap(k_2,n]}a_{i,(1/2,k_2)}^2
\right\}\\
&\quad\le C\{r_{n,1/2}^{\mathrm{del}}\}^2\Delta u.
\end{align*}
For $k_1<i\le k_2$,
\[
|a_{i,(1/2,k_2)}-a_{i,(1/2,k_1)}|^2
\le 2a_{i,(1/2,k_2)}^2+2a_{i,(1/2,k_1)}^2
\le \frac C{i\wedge(n-i)}.
\]
Consequently,
\begin{align}
\|\Delta\mathbf a_{1/2}^{\mathrm{del}}\|_2^2
&=\sum_{i\in\mathcal J\cap([1,k_1]\cup(k_2,n])}
|a_{i,(1/2,k_2)}-a_{i,(1/2,k_1)}|^2
+\sum_{i\in\mathcal J\cap(k_1,k_2]}
|a_{i,(1/2,k_2)}-a_{i,(1/2,k_1)}|^2\nonumber\\
&\le C\left\{
\{r_{n,1/2}^{\mathrm{del}}\}^2
|u_{k_2}-u_{k_1}|+\nu_{\mathcal J}(k_1,k_2)
\right\}.
\label{eq:quadratic-only-deleted-increment}
\end{align}
Here
\[
\nu_{\mathcal J}(k_1,k_2)
=\sum_{i\in\mathcal J\cap(k_1,k_2]}
\frac1{i\wedge(n-i)},
\qquad
\nu_{\mathcal J}(\lambda_n,n-\lambda_n)
\le C\ell_n \{r_{n,1/2}^{\mathrm{del}}\}^2.
\]
A dyadic decomposition and the deleted-index count give
\[
\sum_{i\in\mathcal J\cap[\lambda_n,n-\lambda_n]}
\frac1{i\wedge(n-i)}
\le C\left\{\frac {s_n}{b_n}\log\frac n{\lambda_n}
+\frac {b_n}{\lambda_n}\right\}
\le C\ell_n\{r_{n,1/2}^{\mathrm{del}}\}^2.
\]

For a deterministic coefficient vector $\mathbf x$, let
$\mathbf U(\mathbf x)=\sum_{i=1}^n x_i\bepsilon_i$, and define the
centered bilinear form
\begin{equation}
\mathcal B_n(\mathbf x,\mathbf y)
=\frac{
\mathbf U(\mathbf x)^\top\mathbf U(\mathbf y)
-\E\{\mathbf U(\mathbf x)^\top\mathbf U(\mathbf y)\}
}{\sqrt p}.
\label{eq:quadratic-only-bilinear-form}
\end{equation}
Then
\begin{align}
Q_{\gamma,n}(t_k;\mathbf h_k)
-Q_{\gamma,n}^{I}(t_k;\mathbf h_k)=
\mathfrak e_{\gamma}^{\mathrm{quad,del}}(k)
+\frac{2}{\sqrt p}\widetilde{\mathbf h}_{\gamma,k}^\top
\mathbf U_{\gamma,k}^{\mathrm{del}},
\label{eq:quadratic-only-deletion-shifted-decomposition}
\end{align}
where
\begin{align*}
\mathfrak e_{\gamma}^{\mathrm{quad,del}}(k)
=2\mathcal B_n(\mathbf a_{\gamma,k}^I,
                \mathbf a_{\gamma,k}^{\mathrm{del}})
+\mathcal B_n(\mathbf a_{\gamma,k}^{\mathrm{del}},
               \mathbf a_{\gamma,k}^{\mathrm{del}}).
\end{align*}

{The increment of $\mathfrak e_{\gamma}^{\mathrm{quad,del}}(k)$ in \eqref{eq:quadratic-only-deletion-shifted-decomposition} decomposes as}
\begin{align*}
\mathfrak e_\gamma^{\mathrm{quad,del}}(k_2)-
\mathfrak e_\gamma^{\mathrm{quad,del}}(k_1)=
\Delta\mathfrak e_{\gamma,1}+\Delta\mathfrak e_{\gamma,2}+\Delta\mathfrak e_{\gamma,3},
\end{align*}
with
\begin{align*}
\Delta\mathfrak e_{\gamma,1}
=2\mathcal B_n(
\Delta\mathbf a_{\gamma}^I,
\mathbf a_{\gamma,k_2}^{\mathrm{del}}
),
\quad
\Delta\mathfrak e_{\gamma,2}
=2\mathcal B_n(
\mathbf a_{\gamma,k_1}^I,
\Delta\mathbf a_{\gamma}^{\mathrm{del}}),
\quad
\Delta\mathfrak e_{\gamma,3}
=\mathcal B_n(
\Delta\mathbf a_{\gamma}^{\mathrm{del}},
\mathbf a_{\gamma,k_2}^{\mathrm{del}}
+\mathbf a_{\gamma,k_1}^{\mathrm{del}}).
\end{align*}
With
\eqref{eq:quadratic-only-coefficient-norms},
\eqref{eq:quadratic-only-retained-increment},
\eqref{eq:quadratic-only-gamma0-deleted-increment}, and
\eqref{eq:quadratic-only-deleted-increment}, Lemma~\ref{lem:bilinear-quadratic-bound} yields
\begin{align*}
\|\mathfrak e_\gamma^{\mathrm{quad,del}}(k_2)-
\mathfrak e_\gamma^{\mathrm{quad,del}}(k_1)\|_4
&\le C\|\Delta\mathbf a_{\gamma}^I\|_2\|\mathbf a_{\gamma,k_2}^{\mathrm{del}}\|_2
+C\|\mathbf a_{\gamma,k_1}^I\|_2
\|\Delta\mathbf a_{\gamma}^{\mathrm{del}}\|_2\\
&\quad+
C\|\Delta\mathbf a_{\gamma}^{\mathrm{del}}\|_2
\|\mathbf a_{\gamma,k_2}^{\mathrm{del}}
+\mathbf a_{\gamma,k_1}^{\mathrm{del}}\|_2\\
&\le C\mathfrak h_{\gamma}^{\mathrm{del}}(k_1,k_2)^{1/2},
\end{align*}
where
\begin{equation}
\mathfrak h_{\gamma}^{\mathrm{del}}(k_1,k_2)=
\begin{cases}
\{r_{n,0}^{\mathrm{del}}\}^2d
+\dfrac{q_{\mathcal J}(k_1,k_2)}n
+\dfrac{m_{\mathcal J}d^2}{n}, & \gamma=0,\\
\{r_{n,1/2}^{\mathrm{del}}\}^2
|u_{k_2}-u_{k_1}|+\nu_{\mathcal J}(k_1,k_2). & \gamma=1/2.
\end{cases}
\label{eq:quadratic-only-control-gamma}
\end{equation}
The control function \(\mathfrak h_0^{\mathrm{del}}(k_1,k_2)\) is superadditive on ordered index intervals, whereas \(\mathfrak h_{1/2}^{\mathrm{del}}(k_1,k_2)\) is additive on the corresponding logit intervals.
Hence
\[
\E\left|
\mathfrak e_{\gamma}^{\mathrm{quad,del}}(k_2)
-\mathfrak e_{\gamma}^{\mathrm{quad,del}}(k_1)
\right|^4
\le C\{\mathfrak h_{\gamma}^{\mathrm{del}}(k_1,k_2)\}^{2}.
\]
Moreover,
\begin{align}
\mathfrak h_0^{\mathrm{del}}(0,n)
&\le C\{r_{n,0}^{\mathrm{del}}\}^2
\le C\left(\frac{s_n}{b_n}+\frac{b_n}{n}\right),
\label{eq:quadratic-only-control-total-gamma0}\\
\mathfrak h_{1/2}^{\mathrm{del}}(\lambda_n,n-\lambda_n)
&\le C\ell_n\{r_{n,1/2}^{\mathrm{del}}\}^2
\le C\ell_n\left(
\frac{s_n}{b_n}+\frac{b_n}{n}+\frac{b_n}{\lambda_n}
\right).
\label{eq:quadratic-only-control-total-gammahalf}
\end{align}
The maximal-moment inequality of \citet[Theorem~1]{moricz1982}, with
moment order four and interval exponent two, therefore gives, for
$\gamma=0$,
\[
\E\max_{k\in\mathcal K_0}
|\mathfrak e_0^{\mathrm{quad,del}}(k)|^4
\le C\{\mathfrak h_0^{\mathrm{del}}(0,n)\}^{2}
\le C\{r_{n,0}^{\mathrm{del}}\}^{4},
\]
where $\mathfrak e_0^{\mathrm{quad,del}}(0)=0$.  For $\gamma=1/2$,
\[
\E\max_{k\in\mathcal K_{1/2}}
\left|
\mathfrak e_{1/2}^{\mathrm{quad,del}}(k)
-\mathfrak e_{1/2}^{\mathrm{quad,del}}(\lambda_n)
\right|^4
\le C\{\mathfrak h_{1/2}^{\mathrm{del}}
(\lambda_n,n-\lambda_n)\}^{2}
\le C\ell_n^2\{r_{n,1/2}^{\mathrm{del}}\}^{4}.
\]
Together with
$\|\mathfrak e_{1/2}^{\mathrm{quad,del}}(\lambda_n)\|_4
\le C r_{n,1/2}^{\mathrm{del}}$ and $\omega_p\asymp1$, this proves
\begin{equation}
\left\|
\max_{k\in\mathcal K_{\gamma}}
|\mathfrak e_{\gamma}^{\mathrm{quad,del}}(k)|
\right\|_4
\le
\begin{cases}
C r_{n,0}^{\mathrm{del}},&\gamma=0,\\
C\ell_n^{1/2}r_{n,1/2}^{\mathrm{del}},&\gamma=1/2,
\end{cases}
\label{eq:quadratic-only-deletion}
\end{equation}
\begingroup
because $c_{\mathrm q}=4+3c_{\mathrm{dep}}$ is larger than $1/2$.
\endgroup
There is no
maximum over the $p$ coordinates in
\eqref{eq:quadratic-only-deletion}.

For the second term in \eqref{eq:quadratic-only-deletion-shifted-decomposition},
let $q_n=\max\{4,\lceil\ell_n\rceil\}$. Lemma~\ref{lem:vector-basic}(ii),
\eqref{eq:quadratic-only-coefficient-norms},
$|\mathcal K_\gamma|\le n$, and $n^{1/q_n}\le e$ give
\begin{align*}
\left\|
\max_{k\in\mathcal K_\gamma}
\left|
\frac{2\widetilde{\mathbf h}_{\gamma,k}^\top\mathbf U_{\gamma,k}^{\mathrm{del}}}
{\sqrt p}
\right|
\right\|_4
&{}\le
\left\{
\sum_{k\in\mathcal K_\gamma}
\left\|
\frac{2\widetilde{\mathbf h}_{\gamma,k}^\top\mathbf U_{\gamma,k}^{\mathrm{del}}}
{\sqrt p}
\right\|_{q_n}^{q_n}
\right\}^{1/q_n}\\
&{}\le
Cn^{1/q_n}q_n^{c_{\mathrm{dep}}}
H_{\gamma,n}(\mathbf h)r_{n,\gamma}^{\mathrm{del}}
{
\lesssim \ell_n^{c_{\mathrm{dep}}}
H_{\gamma,n}(\mathbf h)r_{n,\gamma}^{\mathrm{del}}}.
\end{align*}
Combining this result with \eqref{eq:quadratic-only-deletion} and applying the triangle
inequality to the decomposition in \eqref{eq:quadratic-only-deletion-shifted-decomposition}, we have
\begin{equation}
\left\|
\max_{k\in\mathcal K_\gamma}
\left|
\frac{Q_{\gamma,n}(t_k;\mathbf h_k)
-Q_{\gamma,n}^{I}(t_k;\mathbf h_k)}{\omega_p}
\right|
\right\|_4
\lesssim {\ell_n^{c_{\mathrm{dep}}}}r_{n,\gamma}^{\mathrm{del}}
\{1+H_{\gamma,n}(\mathbf h)\}.
\label{eq:quadratic-only-deletion-shifted}
\end{equation}

Since $\operatorname{smax}_{\beta_n}$ is $1$-Lipschitz with respect to
the $\ell_\infty$ norm, the mean-value theorem and
\eqref{eq:quadratic-only-deletion-shifted} give
\begin{align*}
&\left|
\E f\!\left\{
\operatorname{smax}_{\beta_n}
\left(\omega_p^{-1}\mathcal Q_{\gamma,n}(\mathbf h)\right)
\right\}
-
\E f\!\left\{
\operatorname{smax}_{\beta_n}
\left(\omega_p^{-1}\mathcal Q_{\gamma,n}^{I}(\mathbf h)\right)
\right\}
\right|\\
&\quad\le
\|f'\|_\infty
\E\max_{k\in\mathcal K_\gamma}
\left|
\frac{Q_{\gamma,n}(t_k;\mathbf h_k)
-Q_{\gamma,n}^{I}(t_k;\mathbf h_k)}{\omega_p}
\right|\\
&%
\quad\lesssim {\|f\|_{C^1}
\ell_n^{c_{\mathrm{dep}}}}
r_{n,\gamma}^{\mathrm{del}}
\{1+H_{\gamma,n}(\mathbf h)\}.
\end{align*}

Step 2. We next reduce the retained blocks to mutually independent retained
blocks by the same three-branch construction as in Step~2 of
Lemma~\ref{lem:block-comparison-finite}.  Under (MD), the retained blocks are
already independent.  Under (PD), recall
\(\bepsilon_i^{(s_n)}\) and \(\mathbf R_i^{(s_n)}\) from
\eqref{eq:physical-dependent-truncation}.  With
\[
\mathbf U_{\gamma,k}^{I,(s_n)}
=\sum_{\ell=1}^L\sum_{i\in I_\ell}
a_{i,(\gamma,k)}\bepsilon_i^{(s_n)},
\qquad
\mathbf D_{\gamma,k}^{I,(s_n)}
=\sum_{\ell=1}^L\sum_{i\in I_\ell}
a_{i,(\gamma,k)}\mathbf R_i^{(s_n)},
\]
we have
\begin{align}
&Q_{\gamma,n}^{I}(t_k;\mathbf h_k)
-Q_{\gamma,n}^{I,(s_n)}(t_k;\mathbf h_k)\notag\\
&\quad=
\frac2{\sqrt p}\left[
(\mathbf U_{\gamma,k}^{I,(s_n)}
  +\widetilde{\mathbf h}_{\gamma,k})^\top
\mathbf D_{\gamma,k}^{I,(s_n)}
-\E\{(\mathbf U_{\gamma,k}^{I,(s_n)}
  +\widetilde{\mathbf h}_{\gamma,k})^\top
\mathbf D_{\gamma,k}^{I,(s_n)}\}
\right]\notag\\
&\qquad+
\frac1{\sqrt p}\left[
\|\mathbf D_{\gamma,k}^{I,(s_n)}\|_2^2
-\E\|\mathbf D_{\gamma,k}^{I,(s_n)}\|_2^2
\right].
\label{eq:quadratic-only-physical-decomposition}
\end{align}
For deterministic coefficient sequences \(x,y\), set
\[
\mathbf U^{(s_n)}(x)=\sum_i x_i\bepsilon_i^{(s_n)},
\qquad
\mathbf D^{(s_n)}(y)=\sum_i y_i\mathbf R_i^{(s_n)}.
\]
Then
\begingroup
\begin{align}
&\left\|
\frac{
\mathbf U^{(s_n)}(x)^\top\mathbf D^{(s_n)}(y)
-\E\{\mathbf U^{(s_n)}(x)^\top\mathbf D^{(s_n)}(y)\}}
{\sqrt p}
\right\|_4\notag\\
&\quad+
\left\|
\frac{
\mathbf D^{(s_n)}(x)^\top\mathbf D^{(s_n)}(y)
-\E\{\mathbf D^{(s_n)}(x)^\top\mathbf D^{(s_n)}(y)\}}
{\sqrt p}
\right\|_4\notag\\
&\qquad\le
C\sqrt p\,e^{-cs_n^{\varkappa}}\|x\|_2\|y\|_2 .
\label{eq:quadratic-only-physical-residual-moments}
\end{align}
Indeed, \eqref{eq:weighted-physical-residual} and
Lemma~\ref{lem:vector-basic}(ii) give
\[
\max_j\|U_j^{(s_n)}(x)\|_8\le C\|x\|_2,
\qquad \max_j\|D_j^{(s_n)}(y)\|_8
\le Ce^{-cs_n^{\varkappa}}\|y\|_2.
\]
Thus the first centered product satisfies
\begin{align*}
&\left\|p^{-1/2}\{\mathbf U^{(s_n)}(x)^\top\mathbf D^{(s_n)}(y)
-\E[\mathbf U^{(s_n)}(x)^\top\mathbf D^{(s_n)}(y)]\}\right\|_4\\
&\quad\le\frac2{\sqrt p}\sum_{j=1}^p
\|U_j^{(s_n)}(x)\|_8\|D_j^{(s_n)}(y)\|_8
\le C\sqrt p e^{-cs_n^{\varkappa}}\|x\|_2\|y\|_2.
\end{align*}
Replacing $\mathbf U^{(s_n)}(x)$ by $\mathbf D^{(s_n)}(x)$
gives the second bound, with the smaller factor $e^{-2cs_n^{\varkappa}}$.
\endgroup
Thus, by \eqref{eq:weighted-physical-residual},
\eqref{eq:quadratic-only-coefficient-norms},
\eqref{eq:quadratic-only-physical-decomposition}, and
\eqref{eq:quadratic-only-physical-residual-moments},
\begingroup
\begin{align}
\E\max_{k\in\mathcal K_\gamma}
\left|
\frac{
Q_{\gamma,n}^{I}(t_k;\mathbf h_k)
-Q_{\gamma,n}^{I,(s_n)}(t_k;\mathbf h_k)}
{\omega_p}
\right|
&\le
Cn^{1/4}\sqrt p\,e^{-cs_n^{\varkappa}}
\{1+H_{\gamma,n}(\mathbf h)\}\notag\\
&\le
Cn^{-A}\{1+H_{\gamma,n}(\mathbf h)\},
\label{eq:quadratic-only-sdependent-residual}
\end{align}
\endgroup
after increasing \(c_0=c_0(A)\), and we set
\begin{equation}
\mathcal Q_{\gamma,n}^{I,(s_n)}(\mathbf h)
=\{Q_{\gamma,n}^{I,(s_n)}(t_k;\mathbf h_k):k\in\mathcal K_\gamma\}.
\label{eq:quadratic-only-sdependent-array}
\end{equation}
In the (PD) branch set
\(\mathcal Q_{\gamma,n}^{I,\mathrm{ind}}(\mathbf h)
=\mathcal Q_{\gamma,n}^{I,(s_n)}(\mathbf h)\), and in the (MD) branch set
\(\mathcal Q_{\gamma,n}^{I,\mathrm{ind}}(\mathbf h)
=\mathcal Q_{\gamma,n}^{I}(\mathbf h)\).
\begingroup

Let $a_i(k)=n^{-1/2}v_k^{-\gamma}\{\mathbf1(i\leq k)-k/n\}$.
For a deterministic set $B\subseteq\{1,\ldots,n\}$, let
$\mathbf U_B(k)=\sum_{i\in B}a_i(k)\bepsilon_i$ and
$W_B=\max_{k\in\mathcal K_\gamma}(\sum_{i\in B}a_i(k)^2)^{1/2}$.
To bound the random Taylor derivatives, we first prove
\begin{equation}\label{eq:quadratic-block-maximal-moment}
\left\|\max_{k\in\mathcal K_\gamma}
\left|\frac{\mathbf U_B(k)^\top \mathbf U_R(k)-\mathrm E\{\mathbf U_B(k)^\top \mathbf U_R(k)\}}{\sqrt p}\right|\right\|_4
\leq C\{1+\log n\}^{1/4}W_B,
\end{equation}
for every deterministic $R\subseteq\{1,\ldots,n\}$. Taking $R=B$ improves
the right side to $C\{1+\log n\}^{1/4}W_B^2$.
For $\gamma=0$ the logarithmic factor can be omitted.

To prove \eqref{eq:quadratic-block-maximal-moment}, define, for
deterministic index sets $I,J$,
\[
T(I,J)=p^{-1/2}\sum_{i\in I,j\in J}
\{\bepsilon_i^\top\bepsilon_j-\E(\bepsilon_i^\top\bepsilon_j)\}.
\]
The claim is trivial if either set is empty. Otherwise let
$M=|I|$, $N=|J|$ and pad their ordered lists by zeros to lengths
$2^{a_0}\in[M,2M)$ and $2^{b_0}\in[N,2N)$.
Let $\mathcal D_a(I)$ and $\mathcal D_b(J)$ be their partitions into
$2^a$ and $2^b$ consecutive dyadic intervals. By
Lemma~\ref{lem:bilinear-quadratic-bound},
\begin{align*}
\|T(I',J')\|_4&\le C\sqrt{|I'||J'|},\\
\left\|\max_{\substack{I'\in\mathcal D_a(I)\\J'\in\mathcal D_b(J)}}
|T(I',J')|\right\|_4^4
&\le\sum_{I'\in\mathcal D_a(I)}\sum_{J'\in\mathcal D_b(J)}
\E|T(I',J')|^4
\le C M^2N^2 2^{-(a+b)}.
\end{align*}
Each prefix of either ordered list is the disjoint union of at most
one interval at each dyadic level. Hence
\begin{align}
\left\|\max_{s,t}|T(I\cap(-\infty,s],J\cap(-\infty,t])|\right\|_4
&\le\sum_{a=0}^{a_0}\sum_{b=0}^{b_0}
\left\|\max_{I'\in\mathcal D_a(I),J'\in\mathcal D_b(J)}
|T(I',J')|\right\|_4\notag\\
&\le C\sqrt{MN}\sum_{a,b\ge0}2^{-(a+b)/4}
\le C\sqrt{MN}.
\label{eq:quadratic-prefix-rectangle-maximal}
\end{align}
Reversing either ordered list proves the same bound for the other
three prefix/suffix combinations.

Let $B_-(k)=B\cap[1,k]$, $B_+(k)=B\cap(k,n]$, and define
$R_-(k),R_+(k)$ similarly. Let
\[
X_{BR}(k)=p^{-1/2}\{\mathbf U_B(k)^\top\mathbf U_R(k)
-\E[\mathbf U_B(k)^\top\mathbf U_R(k)]\}.
\]
If $a_i(k)=\alpha_k\ind{i\le k}-\beta_k\ind{i>k}$, then
\begin{align*}
X_{BR}(k)={}&\alpha_k^2 T(B_-(k),R_-(k))
-\alpha_k\beta_k\{T(B_-(k),R_+(k))+T(B_+(k),R_-(k))\}\\
&+\beta_k^2 T(B_+(k),R_+(k)).
\end{align*}
For $\gamma=0$, $|\alpha_k|,|\beta_k|\le n^{-1/2}$, so
\begin{align*}
\|\max_k|X_{BR}(k)|\|_4
&\le Cn^{-1}\sqrt{|B||R|}\le C\sqrt{|B|/n},\\
\|\max_k|X_{BB}(k)|\|_4&\le C|B|/n,\\
W_B^2&\ge\sum_{i\in B}a_i(\lfloor n/2\rfloor)^2
\ge c|B|/n.
\end{align*}
These are the required bounds without a logarithmic factor.

For $\gamma=1/2$, cover the trimmed grid by at most $C\ell_n$
nonempty sets $H$ contained in $[d,2d]\cap[1,n/2]$ and their
reflections. Let $k_1=\min H$, $k_2=\max H$ and
\[
M_-=|B\cap[1,k_2]|,\quad M_+=|B\cap(k_1,n]|,\quad
N_-=|R\cap[1,k_2]|,\quad N_+=|R\cap(k_1,n]|.
\]
On a left-side set,
\begin{align*}
\alpha_k&=\sqrt{(n-k)/(nk)}\le Cd^{-1/2},\qquad
\beta_k=\sqrt{k/\{n(n-k)\}}\le C\sqrt d/n,
\qquad \alpha_k\beta_k=1/n,\\
W_B^2&\ge\alpha_{k_2}^2M_-\ge cM_-/d,
\qquad W_B^2\ge\beta_{k_1}^2M_+\ge cdM_+/n^2,\\
N_-&\le2d,\qquad N_+\le n.
\end{align*}
Applying the four rectangle bounds to the expansion of $X_{BR}$ gives
\begin{align*}
\|\max_{k\in H}|X_{BR}(k)|\|_4
&\le C\left\{\frac{\sqrt{M_-d}}d
+\frac{\sqrt{M_-n}}n+\frac{\sqrt{M_+d}}n
+\frac{d\sqrt{M_+n}}{n^2}\right\}
\le CW_B,\\
\|\max_{k\in H}|X_{BB}(k)|\|_4
&\le C\left\{\frac{M_-}d+
\frac{\sqrt{M_-M_+}}n+\frac{dM_+}{n^2}\right\}
\le CW_B^2.
\end{align*}
Reflection gives the same inequalities on right-side sets. Therefore
\begin{align*}
\|\max_{k\in\mathcal K_{1/2}}|X_{BR}(k)|\|_4^4
&\le\sum_H\E\max_{k\in H}|X_{BR}(k)|^4
\le C\ell_n W_B^4,\\
\|\max_{k\in\mathcal K_{1/2}}|X_{BB}(k)|\|_4^4
&\le C\ell_n W_B^8.
\end{align*}
This proves \eqref{eq:quadratic-block-maximal-moment} for the original array.

For an array of independent original copies and original-covariance
Gaussian groups, let $g(i)$ denote the group containing time $i$.
Each time belongs to one group. Its primitive cumulants satisfy
\[
\Cum(\widetilde\epsilon_{i_1,j_1},\ldots,
      \widetilde\epsilon_{i_r,j_r})
=\begin{cases}
\Cum(\epsilon_{i_1,j_1},\ldots,\epsilon_{i_r,j_r}),
 &g(i_1)=\cdots=g(i_r)\text{ is original},\\
\Cov(\epsilon_{i_1,j_1},\epsilon_{i_2,j_2}),
 &r=2,\ g(i_1)=g(i_2)\text{ is Gaussian},\\
0,&\text{otherwise}.
\end{cases}
\]
Every absolute contraction is thus bounded by its original-process
counterpart. The rectangle proof applies unchanged to these arrays.

For projected (PD) groups, use the simultaneous copied-block and
Gaussian square-root couplings constructed in Step~3 of
Lemma~\ref{lem:block-comparison-finite}, specifically
\eqref{eq:fd-PD-gaussian-bilinear-transfer}. Let
$\mathbf Y_{\mathrm{ref}}(x)$ denote the reference hybrid and
$\mathbf Y_{\mathrm{PD}}(x)$ denote its projected hybrid. That construction gives
\begin{align*}
\max_j\|Y_{\mathrm{ref},j}(x)-Y_{\mathrm{PD},j}(x)\|_8
&\le C\sqrt n e^{-cs_n^{\varkappa}}\|x\|_2,\\
\|\mathcal B_{\mathbf Y_{\mathrm{ref}}}(x,y)
-\mathcal B_{\mathbf Y_{\mathrm{PD}}}(x,y)\|_4
&\le C\sqrt{np}e^{-cs_n^{\varkappa}}\|x\|_2\|y\|_2.
\end{align*}
Here $\mathcal B_{\mathbf Y}$ is the centered normalized bilinear form
defined in that step. Take $x_i=a_i(k)\ind{i\in B}$ and
$y_i=a_i(k)\ind{i\in R}$. Since $\|x\|_2\le W_B$ and $\|y\|_2\le1$,
\begin{align*}
\left\|\max_k|X_{BR}^{\mathrm{ref}}(k)-X_{BR}^{\mathrm{PD}}(k)|\right\|_4
&\le\left\{\sum_k
\|X_{BR}^{\mathrm{ref}}(k)-X_{BR}^{\mathrm{PD}}(k)\|_4^4\right\}^{1/4}\\
&\le Cn^{3/4}\sqrt p e^{-cs_n^{\varkappa}}W_B
\le CW_B.
\end{align*}
The last inequality follows from
$n^{3/4}\sqrt p e^{-cs_n^{\varkappa}}
\le(np)^{3/4-cc_0^{\varkappa}}\le1$.
For $R=B$, the same calculation has $W_B^2$ in place of $W_B$.
One coupling covers all coefficient vectors, so these bounds hold for
the grid maxima and establish the required hybrid versions.
\endgroup

Under (AM), use the copied retained blocks
\(\mathcal B_\ell^*\) from \eqref{eq:AM-copied-retained-blocks}.  For
\(\diamond\in\{0,*\}\), let
\begin{align}
\mathbf U_{\ell,\gamma,k}^{(0)}
&=\sum_{i\in I_\ell}a_{i,(\gamma,k)}\bepsilon_i,\qquad
\mathbf U_{\ell,\gamma,k}^{(*)}
=\sum_{i\in I_\ell}a_{i,(\gamma,k)}\bepsilon_i^{(\ell,*)},\notag\\
\mathbf S_{\ell,\gamma,k}
&=\sum_{j<\ell}\sum_{i\in I_j}a_{i,(\gamma,k)}
  \bepsilon_i^{(j,*)}
  +\sum_{j>\ell}\sum_{i\in I_j}a_{i,(\gamma,k)}\bepsilon_i,\notag\\
q_{\ell,\gamma,k}^{(\diamond)}(t)
&=\frac{
\|\mathbf S_{\ell,\gamma,k}
  +t\mathbf U_{\ell,\gamma,k}^{(\diamond)}
  +\widetilde{\mathbf h}_{\gamma,k}\|_2^2
-\E\|\mathbf S_{\ell,\gamma,k}
  +t\mathbf U_{\ell,\gamma,k}^{(\diamond)}\|_2^2}
{\sqrt p},\notag\\
\Psi_{\gamma,n}(\mathbf z)
&=f\!\left\{
\operatorname{smax}_{\beta_n}
\bigl((\omega_p^{-1}z_k):k\in\mathcal K_\gamma\bigr)
\right\}.
\label{eq:quadratic-only-AM-hybrid-coordinates}
\end{align}
The chain-rule inputs for this interpolation are
\begin{equation}
\|\dot q_{\ell,\gamma,k}^{(\diamond)}(t)\|_4
\le Cw_{\ell,(\gamma,k)}
\{1+H_{\gamma,n}(\mathbf h)\},
\qquad
\|\ddot q_{\ell,\gamma,k}^{(\diamond)}(t)\|_4
\le Cw_{\ell,(\gamma,k)}^2.
\label{eq:quadratic-only-AM-q-moments}
\end{equation}
\begingroup
{Let $W_\ell=\sup_k w_{\ell,(\gamma,k)}$.}
Equation~\eqref{eq:quadratic-block-maximal-moment}, the derivative
identities for the centered quadratic form, and the directional moment
bound with $q=\max\{4,\log n\}$ yield, uniformly in $t\in[0,1]$,
\begin{align}
\left\|\max_k|\dot q_{\ell,\gamma,k}^{(\diamond)}(t)|\right\|_4
&\le C\ell_n^{c_{\mathrm{dep}}}W_\ell\{1+H_{\gamma,n}(\mathbf h)\},\notag\\
\left\|\max_k|\ddot q_{\ell,\gamma,k}^{(\diamond)}(t)|\right\|_4
&\le C\ell_n^{1/4}W_\ell^2.
\label{eq:quadratic-block-maximal-derivatives}
\end{align}
To verify the deterministic-shift contribution, take
$q_n=\max\{4,\lceil\ell_n\rceil\}$ and
$\mathbf u=\widetilde{\mathbf h}_{\gamma,k}/
\|\widetilde{\mathbf h}_{\gamma,k}\|_2$ when the shift is nonzero.
For an original or copied current block,
Lemma~\ref{lem:vector-basic}(ii) gives
$\|\mathbf u^\top\mathbf U_{\ell,\gamma,k}\|_{q_n}
\le Cq_n^{c_{\mathrm{dep}}}w_{\ell,(\gamma,k)}$.
For its conditional projection and covariance-matched Gaussian block,
respectively,
\begin{align*}
\|\mathbf u^\top\mathbf U_{\ell,\gamma,k}^{(s_n)}\|_{q_n}
&\le C\{q_n^{c_{\mathrm{dep}}}
+q_n^{c_{\mathrm{dep}}+1}e^{-cs_n^{\varkappa}}\}
w_{\ell,(\gamma,k)}
\le Cq_n^{c_{\mathrm{dep}}}w_{\ell,(\gamma,k)},\\
\|\mathbf u^\top\mathbf U_{\ell,\gamma,k}^{G}\|_{q_n}
&\le C\sqrt{q_n}
\|\mathbf u^\top\mathbf U_{\ell,\gamma,k}^{G}\|_2
\le C\sqrt{q_n}\,w_{\ell,(\gamma,k)}.
\end{align*}
The first line uses \eqref{eq:weighted-physical-residual}; its exponential
factor absorbs $q_n$. The second uses covariance matching and the
order-two bound. Thus every current block in the interpolation satisfies
\begin{align*}
\left\|\max_k\frac{|
\widetilde{\mathbf h}_{\gamma,k}^{\top}\mathbf U_{\ell,\gamma,k}|}{\sqrt p}
\right\|_4
&\le n^{1/q_n}\max_k
\left\|\frac{\widetilde{\mathbf h}_{\gamma,k}^{\top}
\mathbf U_{\ell,\gamma,k}}{\sqrt p}\right\|_{q_n}\\
&\le C\ell_n^{c_{\mathrm{dep}}}
H_{\gamma,n}(\mathbf h)W_\ell.
\end{align*}
Together with \eqref{eq:quadratic-block-maximal-moment} and the exact
derivatives of the centered quadratic interpolation, this proves
\eqref{eq:quadratic-block-maximal-derivatives}.

For the Taylor remainder, let
$F_\ell^{(\diamond)}(t)=\Psi_{\gamma,n}
((q_{\ell,\gamma,k}^{(\diamond)}(t))_k)$ and define
\[
M_{1,\ell}(t)=\max_k|\dot q_{\ell,\gamma,k}^{(\diamond)}(t)|,
\qquad M_{2,\ell}(t)=\max_k|\ddot q_{\ell,\gamma,k}^{(\diamond)}(t)|.
\]
Since the interpolation coordinates are quadratic in $t$,
\begin{align*}
(F_\ell^{(\diamond)})^{(3)}(t)
&=\sum_{a,b,c}\partial_{abc}\Psi_{\gamma,n}(\mathbf q(t))
       \dot q_a(t)\dot q_b(t)\dot q_c(t)\\
&\quad+3\sum_{a,b}\partial_{ab}\Psi_{\gamma,n}(\mathbf q(t))
       \ddot q_a(t)\dot q_b(t),\\
\sup_{\mathbf z}\sum_{a,b}|\partial_{ab}\Psi_{\gamma,n}(\mathbf z)|
&\le C\|f\|_{C^2}(1+\beta_n),\\
\sup_{\mathbf z}\sum_{a,b,c}|\partial_{abc}\Psi_{\gamma,n}(\mathbf z)|
&\le C\|f\|_{C^3}(1+\beta_n+\beta_n^2).
\end{align*}
Consequently the pathwise derivative sums and H\"older's inequality give
\begin{align*}
\E|(F_\ell^{(\diamond)})^{(3)}(t)|
&\le C\|f\|_{C^3}
\{(1+\beta_n^2)\|M_{1,\ell}(t)\|_3^3
+(1+\beta_n)\|M_{2,\ell}(t)\|_2\|M_{1,\ell}(t)\|_2\}\\
&\le C\|f\|_{C^3}
\{\ell_n^{4+3c_{\mathrm{dep}}}
 +\ell_n^{2+c_{\mathrm{dep}}+1/4}\}
W_\ell^3\{1+H_{\gamma,n}(\mathbf h)\}^3,\\
R_{\ell,3}^{(\diamond)}
&=\frac12\int_0^1(1-t)^2(F_\ell^{(\diamond)})^{(3)}(t)\,dt,\\
\E|R_{\ell,3}^{(\diamond)}|
&\le C\|f\|_{C^3}\ell_n^{c_{\mathrm q}}W_\ell^3
\{1+H_{\gamma,n}(\mathbf h)\}^3.
\end{align*}
Here $c_{\mathrm q}=4+3c_{\mathrm{dep}}$ and
$c_{\mathrm{dep}}\ge1/2$.

The first two (AM) Taylor terms are the coordinatewise covariance
sums in \eqref{eq:fd-AM-first-Taylor-bound}--
\eqref{eq:fd-AM-second-Taylor-bound}. The number of time indices and
smooth-maximum derivatives is polynomial in $n$, and
$\max_{i,k}|a_i(k)|\le1$. Thus, for a fixed finite $C_*$,
\begin{align*}
&\sum_{\ell=1}^L
\left\{|\E[(F_\ell^{(0)})'(0)-(F_\ell^{(*)})'(0)]|
+|\E[(F_\ell^{(0)})''(0)-(F_\ell^{(*)})''(0)]|\right\}\\
&\quad\le C\|f\|_{C^3}\{1+H_{\gamma,n}(\mathbf h)\}^2
n(np)^{C_*}\ell_n^4e^{-cs_n^{\varkappa}}\\
&\quad\le C\|f\|_{C^3}\{1+H_{\gamma,n}(\mathbf h)\}^2n^{-A},
\end{align*}
after choosing $cc_0^{\varkappa}>A+C_*+6$.

\endgroup
Let
\(\mathcal Q_{\gamma,n}^{I,\mathrm{ind}}(\mathbf h)\) denote the retained
\begingroup
quadratic array produced by the corresponding branch. Summing the
Taylor identities just proved, and using
$\sum_\ell W_\ell^3\le C(b_n/\Lambda_{\gamma,n})^{1/2}$ from
\eqref{eq:quadratic-only-gamma0-block-weights} and
\eqref{eq:quadratic-only-block-weights}, gives
\begin{align}
&\left|
\E f\!\left\{\operatorname{smax}_{\beta_n}
(\omega_p^{-1}\mathcal Q_{\gamma,n}^{I}(\mathbf h))\right\}
-\E f\!\left\{\operatorname{smax}_{\beta_n}
(\omega_p^{-1}\mathcal Q_{\gamma,n}^{I,\mathrm{ind}}(\mathbf h))\right\}
\right|\notag\\
&\qquad\le
C\|f\|_{C^3}
\left\{
\ell_n^{c_{\mathrm q}}
\left(\frac{b_n}{\Lambda_{\gamma,n}}\right)^{1/2}
+n^{-A}
\right\}
\{1+H_{\gamma,n}(\mathbf h)\}^{3}.
\label{eq:quadratic-only-independent-blocks}
\end{align}
\endgroup

Step 3. We replace the mutually independent retained blocks by
covariance-matched Gaussian retained blocks.  Let
\(\bepsilon_i^{\mathrm{ind}}\) denote the retained observation produced in
Step~2.  For \(k\in\mathcal K_\gamma\),
define
\begin{equation}
\mathbf Z_{\ell,k}
=\sum_{i\in I_\ell}a_{i,(\gamma,k)}\bepsilon_i^{\mathrm{ind}},
\qquad
\mathbf Z_{\ell,k}^{G}
=\sum_{i\in I_\ell}a_{i,(\gamma,k)}\bepsilon_i^{G,\ell},
\label{eq:quadratic-only-step3-block-contributions}
\end{equation}
where
\begin{equation}
\E\bepsilon_i^{G,\ell}=\mathbf0,\qquad
\Cov(\bepsilon_i^{G,\ell},\bepsilon_j^{G,\ell})
=\Cov(\bepsilon_i^{\mathrm{ind}},\bepsilon_j^{\mathrm{ind}}),
\qquad i,j\in I_\ell .
\label{eq:quadratic-only-step3-covariance-matching}
\end{equation}
\begingroup
At replacement $\ell$, let
$\mathbf T_{\ell,k}=\sum_{j<\ell}\mathbf Z_{j,k}^G+
\sum_{j>\ell}\mathbf Z_{j,k}$ and
$\mathcal H_\ell=\sigma(\mathbf Z_{j,k},\mathbf Z_{j,k}^G:j\ne\ell,k)$.
For $\mathbf Z_{\ell,k}^{(0)}=\mathbf Z_{\ell,k}$ and
$\mathbf Z_{\ell,k}^{(G)}=\mathbf Z_{\ell,k}^G$,
\begin{align*}
\E(\mathbf Z_{\ell,k}^{(\diamond)}\mid\mathcal H_\ell)&=0,\\
\E(\mathbf Z_{\ell,k}^{(\diamond)}
\mathbf Z_{\ell,l}^{(\diamond)\top}\mid\mathcal H_\ell)
&=\Cov(\mathbf Z_{\ell,k},\mathbf Z_{\ell,l}),
\qquad \diamond\in\{0,G\}.
\end{align*}
Define $q_{\ell,k}^{(\diamond)}(t)$ and
$F_\ell^{(\diamond)}(t)$ by \eqref{eq:quadratic-only-AM-hybrid-coordinates},
with $\mathbf S_{\ell,\gamma,k}$ replaced by $\mathbf T_{\ell,k}$ and
$\mathbf U_{\ell,\gamma,k}^{(\diamond)}$ by $\mathbf Z_{\ell,k}^{(\diamond)}$.
The derivative identities in Step~3 of
Lemma~\ref{lem:block-comparison-finite} then give
\[
\E[(F_\ell^{(0)})'(0)-(F_\ell^{(G)})'(0)]=0,\qquad
\E[(F_\ell^{(0)})''(0)-(F_\ell^{(G)})''(0)]=0.
\]
The third-order remainders just proved are bounded by
\(
C\|f\|_{C^3}\ell_n^{c_{\mathrm q}}
\{1+H_{\gamma,n}(\mathbf h)\}^{3}
\sup_{k\in\mathcal K_\gamma}w_{\ell,(\gamma,k)}^3 .
\)
\endgroup
Therefore, using \eqref{eq:quadratic-only-gamma0-block-weights} and
\eqref{eq:quadratic-only-block-weights},
\begingroup
\begin{align}
&\left|
\E f\!\left\{\operatorname{smax}_{\beta_n}
(\omega_p^{-1}\mathcal Q_{\gamma,n}^{I,\mathrm{ind}}(\mathbf h))\right\}
-\E f\!\left\{\operatorname{smax}_{\beta_n}
(\omega_p^{-1}\mathcal Q_{\gamma,n}^{I,G,\mathrm{ind}}(\mathbf h))\right\}
\right|\notag\\
&\qquad\lesssim
\|f\|_{C^3}\ell_n^{c_{\mathrm q}}
\left(\frac{b_n}{\Lambda_{\gamma,n}}\right)^{1/2}
\{1+H_{\gamma,n}(\mathbf h)\}^{3}.
\label{eq:quadratic-only-gaussian-blocks}
\end{align}
\endgroup
Here \(\mathcal Q_{\gamma,n}^{I,G,\mathrm{ind}}(\mathbf h)\) is the array
formed from the Gaussian block sums
\((\mathbf Z_{\ell,k}^{G}:\ell\le L,k\in\mathcal K_\gamma)\).

Step 4. We restore the covariance of the retained Gaussian blocks and then
add back the deleted Gaussian indices.  Let
\begin{equation}
\widetilde{\mathbf U}_{\gamma,k}^G
=\sum_{\ell=1}^L\sum_{i\in I_\ell}
a_{i,(\gamma,k)}\bepsilon_i^{G,\ell},
\qquad
\mathbf U_{\gamma,k}^{I,G}
=\sum_{\ell=1}^L\sum_{i\in I_\ell}
a_{i,(\gamma,k)}\bepsilon_i^G,
\qquad
\mathbf U_{\gamma,k}^G
=\sum_{i=1}^n a_{i,(\gamma,k)}\bepsilon_i^G .
\label{eq:quadratic-only-Gaussian-retained-arrays}
\end{equation}
Define
\begin{equation}
\boldsymbol\Sigma_{0;k l}
=\Cov(\widetilde{\mathbf U}_{\gamma,k}^G,\widetilde{\mathbf U}_{\gamma,l}^G),
\quad
\boldsymbol\Sigma_{1;k l}
=\Cov(\mathbf U_{\gamma,k}^{I,G},\mathbf U_{\gamma,l}^{I,G}),
\quad
\boldsymbol\Delta_{\gamma;k l}^{\mathrm{rest}}
=\boldsymbol\Sigma_{1;k l}-\boldsymbol\Sigma_{0;k l}.
\label{eq:quadratic-only-GC-covariances}
\end{equation}
The covariance estimates in Step~4 of
Lemma~\ref{lem:block-comparison-finite}, with
\eqref{eq:quadratic-only-coefficient-norms}, give
\begin{align}
\max_{k,l\in\mathcal K_\gamma}
\|\boldsymbol\Delta_{\gamma;k l}^{\mathrm{rest}}\|_{\mathrm{op}}
&\le C\{\eta_{s_n}+e^{-cs_n^{\varkappa}}\},&
\max_{k,l\in\mathcal K_\gamma}
\|\boldsymbol\Delta_{\gamma;k l}^{\mathrm{rest}}\|_{\mathrm F}
&\le C\sqrt p\{\eta_{s_n}+e^{-cs_n^{\varkappa}}\},
\label{eq:quadratic-only-block-covariance}
\end{align}
and
\begin{equation}
\sup_{0\le\theta\le1}\max_{k,l\in\mathcal K_\gamma}
\|\boldsymbol\Sigma_{\theta;k l}\|_{\mathrm{op}}\le C,
\qquad
\sup_{0\le\theta\le1}\max_{k,l\in\mathcal K_\gamma}
\|\boldsymbol\Sigma_{\theta;k l}\|_{\mathrm F}\le C\sqrt p .
\label{eq:quadratic-only-interpolating-covariance}
\end{equation}
For
\(\boldsymbol\Sigma_\theta=(1-\theta)\boldsymbol\Sigma_0+\theta\boldsymbol\Sigma_1\),
let
\begin{equation}
\mathbf X_\theta
= (\mathbf X_{\theta,k}:k\in\mathcal K_\gamma)
\sim N(\mathbf0,\boldsymbol\Sigma_\theta),
\qquad
q_{\theta,k}
=\frac{\|\mathbf X_{\theta,k}
  +\widetilde{\mathbf h}_{\gamma,k}\|_2^2
-\Tr(\boldsymbol\Sigma_{\theta;k k})}{\sqrt p}.
\label{eq:quadratic-only-GC-interpolation-coordinates}
\end{equation}
Let
\[
\mathbf q_\theta=(q_{\theta,k}:k\in\mathcal K_\gamma),
\qquad
\Phi_\gamma(\theta)=\E\Psi_{\gamma,n}(\mathbf q_\theta).
\]
The Gaussian interpolation derivative is
\begin{equation}
\Phi_\gamma'(\theta)
=\frac2p\sum_{k,l\in\mathcal K_\gamma}
\E\left[
\partial_{kl}\Psi_{\gamma,n}(\mathbf q_\theta)
(\mathbf X_{\theta,k}+\widetilde{\mathbf h}_{\gamma,k})^\top
\boldsymbol\Delta_{\gamma;k l}^{\mathrm{rest}}
(\mathbf X_{\theta,l}+\widetilde{\mathbf h}_{\gamma,l})
\right],
\label{eq:quadratic-only-GC-derivative}
\end{equation}
\begingroup
To retain the random derivatives inside their expectation, bound the
maximum of the contractions. For every coordinate $j$, the Gaussian
variance bound gives
\begin{align*}
\Pr\{\max_k|X_{\theta,k,j}|>u\}
&\le\min\{1,2n\exp(-u^2/C)\},\\
\E\max_k|X_{\theta,k,j}|^2
&=\int_0^\infty\Pr\{\max_k|X_{\theta,k,j}|^2>v\}\,dv
\le C\log(en),\\
\frac1p\E\max_k\|\mathbf X_{\theta,k}
+\widetilde{\mathbf h}_{\gamma,k}\|_2^2
&\le\frac2p\sum_{j=1}^p\E\max_k|X_{\theta,k,j}|^2
+\frac2p\max_k\|\widetilde{\mathbf h}_{\gamma,k}\|_2^2
\le C\ell_n\{1+H_{\gamma,n}(\mathbf h)\}^2.
\end{align*}
The operator-norm discrepancy bound therefore yields
\begin{equation}
\frac1p\E\max_{k,l}
\left|(\mathbf X_{\theta,k}+\widetilde{\mathbf h}_{\gamma,k})^\top
\boldsymbol\Delta_{\gamma;k l}^{\mathrm{rest}}
(\mathbf X_{\theta,l}+\widetilde{\mathbf h}_{\gamma,l})\right|
\le C\ell_n\{\eta_{s_n}+e^{-cs_n^{\varkappa}}\}
\{1+H_{\gamma,n}(\mathbf h)\}^2.
\label{eq:quadratic-only-GC-contraction}
\end{equation}
\endgroup
Consequently, with \(\sum_{k,l}
\left|
\partial_{kl}\Psi_{\gamma,n}
\right|
\lesssim\|f\|_{C^3}\beta_n
\lesssim\|f\|_{C^3}\ell_n^2,\) and 
\(\eta_{s_n}\le Ce^{-cs_n^{\varkappa}},\)
\begingroup
\begin{align}
&\left|
\E f\!\left\{\operatorname{smax}_{\beta_n}
(\omega_p^{-1}\mathcal Q_{\gamma,n}^{I,G,\mathrm{ind}}(\mathbf h))\right\}
-\E f\!\left\{\operatorname{smax}_{\beta_n}
(\omega_p^{-1}\mathcal Q_{\gamma,n}^{G}(\mathbf h))\right\}
\right|\notag\\
&\qquad\le
C\|f\|_{C^3}\ell_n^3
\left\{r_{n,\gamma}^{\mathrm{del}}
+\eta_{s_n}+e^{-cs_n^{\varkappa}}\right\}
\{1+H_{\gamma,n}(\mathbf h)\}^{3}\notag\\
&\qquad\lesssim 
\|f\|_{C^3}
(\ell_n^3r_{n,\gamma}^{\mathrm{del}}
+n^{-A})
\{1+H_{\gamma,n}(\mathbf h)\}^{3},
\label{eq:quadratic-only-full-gaussian}
\end{align}
For (AM) and (PD), choose $c_0$ so that
\[
\ell_n^3\{\eta_{s_n}+e^{-cs_n^{\varkappa}}\}
\le C\ell_n^3(np)^{-cc_0^{\varkappa}}
=O(n^{-A}).
\]
For (MD), the restored covariance discrepancy is exactly zero because
$s_n>m_0$; only Gaussian deletion contributes, and the same final
bound follows without an exponential remainder.

\endgroup
Combining
\eqref{eq:quadratic-only-deletion-shifted},
\eqref{eq:quadratic-only-independent-blocks},
\eqref{eq:quadratic-only-gaussian-blocks}, and
\eqref{eq:quadratic-only-full-gaussian}, the triangle inequality gives
\begingroup
\begin{align*}
&\left|
\E f\!\left\{
\operatorname{smax}_{\beta_n}
\left(\omega_p^{-1}\mathcal Q_{\gamma,n}(\mathbf h)\right)
\right\}
-
\E f\!\left\{
\operatorname{smax}_{\beta_n}
\left(\omega_p^{-1}\mathcal Q_{\gamma,n}^{G}(\mathbf h)\right)
\right\}
\right|\\
&\quad\le
C\|f\|_{C^3}
\Bigg[
\ell_n^{c_{\mathrm{dep}}}r_{n,\gamma}^{\mathrm{del}}
+\ell_n^{c_{\mathrm q}}\left(\frac{b_n}{\Lambda_{\gamma,n}}\right)^{1/2}
+\ell_n^3r_{n,\gamma}^{\mathrm{del}}
+n^{-A}
\Bigg]
\{1+H_{\gamma,n}(\mathbf h)\}^{3}\\
&\quad\le
C\|f\|_{C^3}
\left[
\ell_n^{c_{\mathrm q}}
\left\{
\left(\frac{s_n}{b_n}\right)^{1/2}
+\left(\frac{b_n}{\Lambda_{\gamma,n}}\right)^{1/2}
\right\}
+n^{-A}
\right]
\{1+H_{\gamma,n}(\mathbf h)\}^{3}\\
&\quad=
C\|f\|_{C^3}r_{n,\gamma}^{\mathrm{quad}}
\{1+H_{\gamma,n}(\mathbf h)\}^{3},
\end{align*}
\endgroup
with
\begingroup
$c_{\mathrm q}=4+3c_{\mathrm{dep}}$.
Finally, $s_n/b_n\le C\ell_n^{-2c_{\mathrm q}-4}$ and
\[
\ell_n^{c_{\mathrm q}}(b_n/\Lambda_{\gamma,n})^{1/2}
\le C\ell_n^{2c_{\mathrm q}+2}(s_n/\Lambda_{\gamma,n})^{1/2}.
\]
Under (AM) or (PD),
$s_n=O(n^{\nu/\varkappa})$ and
$\epsilon_\gamma=(\lambda_\gamma-\nu/\varkappa)/2>0$; under (MD),
$s_n$ is fixed and set $\epsilon_\gamma=\lambda_\gamma/2>0$.
Therefore
\[
(\log\log n)^{3/2}r_{n,\gamma}^{\mathrm{quad}}
\le C(\log\log n)^{3/2}
\{\ell_n^{-2}+\ell_n^{2c_{\mathrm q}+2}n^{-\epsilon_\gamma}+n^{-A}\}
\longrightarrow0.
\]
\endgroup
This proves \eqref{eq:quadratic-only-block-comparison}.

\end{proof}

\begin{lemma}%
\label{lem:quadratic-only-bridge-comparison}
\begingroup
Suppose Assumptions~\ref{ass:C1}--\ref{ass:C2} hold.  When $\gamma=1/2$, assume
\endgroup
additionally that $\lambda_n\asymp n^\lambda$ for some fixed
$\lambda\in(0,1)$.  Let $\ell_n=\log(n)+1$ and $\beta_n=\ell_n^2$.
Then, for
$\gamma\in\{0,1/2\}$, a deterministic family
$\mathbf h=(\mathbf h_k:k\in\mathcal K_\gamma)$, 
and every $f\in C_b^3(\mathbb R)$,
\begin{align}
&\left|
\E f\!\left\{
\operatorname{smax}_{\beta_n}
\left(\omega_p^{-1}\mathcal Q_{\gamma,n}^{G}(\mathbf h)\right)
\right\}
-
\E f\!\left\{
\operatorname{smax}_{\beta_n}
\left(\mathcal Q_{\gamma,n}^{B}(\mathbf h)\right)
\right\}
\right|
\le
C\|f\|_{C^3}
\frac{\ell_n^{c_{\mathrm{br},\mathrm q}}}{\Lambda_{\gamma,n}}
\{1+H_{\gamma,n}(\mathbf h)\}^{3}.
\label{eq:quadratic-only-covariance-restoration}
\end{align}
where $\Lambda_{0,n}=n$, $\Lambda_{1/2,n}=\lambda_n$, 
\begingroup
\(
c_{\mathrm{br},\mathrm q}=3,
H_{\gamma,n}(\mathbf h)
=
\max_{k\in\mathcal K_\gamma}
v_k^{-\gamma}\|\mathbf h_k\|_2/\sqrt p.
\)
\endgroup
\end{lemma}

\begin{proof}
Let
\[
\mathbf U_{\gamma,k}^{G}=v_k^{-\gamma}\mathbf U_n^G(t_k),
\qquad
\mathbf U_{\gamma,k}^{B}=v_k^{-\gamma}\bOme^{1/2}\mathbf B(t_k),
\qquad
\widetilde{\mathbf h}_{\gamma,k}=v_k^{-\gamma}\mathbf h_k .
\]
Lemma~\ref{lem:cusum-cov-general} gives an operator-norm covariance
error $C/n$ before boundary standardization.  If $\gamma=0$, this error
is unchanged.  If $\gamma=1/2$, then
$v_k\ge c\lambda_n/n$ on $\mathcal K_{1/2}$.  Hence, for
$\gamma\in\{0,1/2\}$,
\begin{equation}
\begin{aligned}
&\max_{k,l\in\mathcal K_{\gamma}}
\left\|
\Cov\!\left(\mathbf U_{\gamma,k}^{G},\mathbf U_{\gamma,l}^{G}\right)
-\Cov\!\left(\mathbf U_{\gamma,k}^{B},\mathbf U_{\gamma,l}^{B}\right)
\right\|_{\mathrm{op}}
\le\frac C{\Lambda_{\gamma,n}},\\
&\max_{k,l\in\mathcal K_{\gamma}}
\left\|
\Cov\!\left(\mathbf U_{\gamma,k}^{G},\mathbf U_{\gamma,l}^{G}\right)
-\Cov\!\left(\mathbf U_{\gamma,k}^{B},\mathbf U_{\gamma,l}^{B}\right)
\right\|_{\mathrm F}
\le\frac{C\sqrt p}{\Lambda_{\gamma,n}}.
\end{aligned}
\label{eq:quadratic-only-bridge-covariance}
\end{equation}
\begingroup
Let $\boldsymbol\Sigma_0$ and $\boldsymbol\Sigma_1$ be the covariance
matrices of the stacked arrays $(\mathbf U_{\gamma,k}^{G})_k$ and
$(\mathbf U_{\gamma,k}^{B})_k$, and let
\[
\boldsymbol\Sigma_\theta=(1-\theta)\boldsymbol\Sigma_0+
\theta\boldsymbol\Sigma_1,\quad
\boldsymbol\Delta_{kl}=\boldsymbol\Sigma_{1;kl}-\boldsymbol\Sigma_{0;kl},
\quad \mathbf X_\theta\sim N(\mathbf0,\boldsymbol\Sigma_\theta).
\]
Define
\[
q_{\theta,k}=\frac{\|\mathbf X_{\theta,k}+
\widetilde{\mathbf h}_{\gamma,k}\|_2^2-
\Tr(\boldsymbol\Sigma_{\theta;kk})}{\sqrt p},\quad
\Psi(\mathbf z)=f\{\operatorname{smax}_{\beta_n}
((\omega_p^{-1}z_k)_k)\},\quad
\Phi(\theta)=\E\Psi((q_{\theta,k})_k).
\]
The trace cancellation in \eqref{eq:quadratic-only-GC-derivative} gives
\[
\Phi'(\theta)=\frac2p\sum_{k,l}\E\left[
\partial_{kl}\Psi(\mathbf q_\theta)
(\mathbf X_{\theta,k}+\widetilde{\mathbf h}_{\gamma,k})^\top
\boldsymbol\Delta_{kl}
(\mathbf X_{\theta,l}+\widetilde{\mathbf h}_{\gamma,l})\right].
\]
The covariance bound \eqref{eq:quadratic-only-bridge-covariance},
the Gaussian maximum calculation in
\eqref{eq:quadratic-only-GC-contraction}, and the pathwise derivative
sum imply
\begin{align*}
\frac1p\E\max_{k,l}
\left|(\mathbf X_{\theta,k}+\widetilde{\mathbf h}_{\gamma,k})^\top
\boldsymbol\Delta_{kl}
(\mathbf X_{\theta,l}+\widetilde{\mathbf h}_{\gamma,l})\right|
&\le\frac C{p\Lambda_{\gamma,n}}
\E\max_k\|\mathbf X_{\theta,k}+
\widetilde{\mathbf h}_{\gamma,k}\|_2^2\\
&\le\frac{C\ell_n}{\Lambda_{\gamma,n}}
\{1+H_{\gamma,n}(\mathbf h)\}^2,\\
\sup_{\mathbf z}\sum_{k,l}|\partial_{kl}\Psi(\mathbf z)|
&\le C\|f\|_{C^3}(1+\beta_n)
\le C\|f\|_{C^3}\ell_n^2,\\
|\Phi(1)-\Phi(0)|
&\le\int_0^1|\Phi'(\theta)|\,d\theta
\le C\|f\|_{C^3}\frac{\ell_n^3}{\Lambda_{\gamma,n}}
\{1+H_{\gamma,n}(\mathbf h)\}^2.
\end{align*}
Since $1+H_{\gamma,n}(\mathbf h)\ge1$, this proves
\endgroup
\eqref{eq:quadratic-only-covariance-restoration}.
\end{proof}

\begin{lemma}\label{lem:dense-boundary-extreme}
Suppose Assumption~\ref{ass:C2} holds, $\lambda_n\asymp n^\lambda$ for some $\lambda\in(0,1)$, and
\begingroup
$\log\log n=o(p^{1/7})$. Then
\[
A_{\mathrm{DE}}(\log h_n)
\max_{k\in\mathcal K_{1/2}}G_{p,1/2}(k/n)
-D_{\mathrm{DE}}(\log h_n)
\xrightarrow{d} Z_{\mathrm{Gu}}.
\]
\endgroup
\end{lemma}

\begin{proof}
\begingroup
We first compare the processes on a logit grid, then bound the
discretization error and apply the stationary Gaussian extreme-value theorem.
Let $\lambda_{j,p}$ be the eigenvalues of $\bOme$ and define
\[
c_{j,p}^{\Omega}=\frac{\lambda_{j,p}}{\{2\Tr(\bOme^2)\}^{1/2}},
\qquad t(u)=\frac{e^u}{1+e^u},\qquad
\widetilde G_{p,1/2}(u)=G_{p,1/2}\{t(u)\}.
\]
Diagonalizing the entire Gaussian bridge gives the process identity
\[
\{\widetilde G_{p,1/2}(u):u\in\mathbb R\}
\overset d=
\left\{\sum_{j=1}^p c_{j,p}^{\Omega}\{Y_j(u)^2-1\}:u\in\mathbb R\right\},
\]
\[
\Cov\{Y_j(u),Y_l(v)\}=\ind{j=l}e^{-|u-v|/2}.
\]
Here the $Y_j$ are independent stationary Gaussian processes. Assumption~\ref{ass:C2} gives
\[
\max_j\sqrt p\,|c_{j,p}^{\Omega}|\le C,\qquad
2\sum_j(c_{j,p}^{\Omega})^2=1,\qquad
\Cov\{\widetilde G_{p,1/2}(u),\widetilde G_{p,1/2}(v)\}=e^{-|u-v|}.
\]
For the process $V$ in Theorem~\ref{Th1}, let
\[
\widetilde V(u)=\frac{V\{t(u)\}}{t(u)\{1-t(u)\}}.
\]
Since $\Cov\{V(s),V(t)\}=K_B(s,t)^2$, for $u\le v$,
\[
\Cov\{\widetilde V(u),\widetilde V(v)\}
=\frac{t(u)\{1-t(v)\}}{t(v)\{1-t(u)\}}=e^{u-v}.
\]
Thus, for a stationary centered Gaussian process $Z_{\mathrm{OU}}$,
\begin{equation}
\{\widetilde V(u):u\in\mathbb R\}\overset d=
\{Z_{\mathrm{OU}}(u):u\in\mathbb R\},\qquad
\Cov\{Z_{\mathrm{OU}}(u),Z_{\mathrm{OU}}(v)\}=e^{-|u-v|}.
\label{eq:S14-OU-representation}
\end{equation}

Set
\[
L_{\mathrm{OU},n}=\log h_n\asymp\log n,\quad
R_{\mathrm{OU},n}=L_{\mathrm{OU},n}/2,\quad
\delta_{u,n}=(\log n)^{-8},\quad K_n=\log(2+\log n).
\]
Let $\mathcal U_n=\{u_r:1\le r\le N_{\mathrm{OU},n}\}$ be a regular grid
including both endpoints of $[-R_{\mathrm{OU},n},R_{\mathrm{OU},n}]$, of mesh
at most $\delta_{u,n}$. Then
\[
N_{\mathrm{OU},n}\le C(1+L_{\mathrm{OU},n}/\delta_{u,n}),\qquad
\log N_{\mathrm{OU},n}\le CK_n.
\]
Define independent vectors $\mathbf X_j=(X_{j,r})_{r\le N_{\mathrm{OU},n}}$ by
\[
X_{j,r}=\sqrt p\,c_{j,p}^{\Omega}\{Y_j(u_r)^2-1\}.
\]
For a constant $C_*$ independent of $n,p,j,r$,
\begin{align*}
\E X_{j,r}&=0, &
\frac1p\sum_j\E X_{j,r}^2&=1,\\
\max_{j,r}\E\exp(|X_{j,r}|/C_*)&\le2, &
\max_r\frac1p\sum_j\E|X_{j,r}|^{2+l}&\le C_*^l,
\quad l=1,2.
\end{align*}
Let $\mathbf X_j^G$ be independent centered Gaussian vectors with
$\Cov(\mathbf X_j^G)=\Cov(\mathbf X_j)$. Then
\begin{equation}
p^{-1/2}\sum_j\mathbf X_j^G
\overset d=(Z_{\mathrm{OU}}(u_r):u_r\in\mathcal U_n).
\label{eq:S14-finite-grid-identification}
\end{equation}
Proposition~2.1 of \citet{chernozhukov2017central}, under its conditions
(M.1), (M.2) and (E.1), now gives
\begin{align}
&\sup_x\left|
\Pr\{\max_{u_r\in\mathcal U_n}\widetilde G_{p,1/2}(u_r)\le x\}
-\Pr\{\max_{u_r\in\mathcal U_n}Z_{\mathrm{OU}}(u_r)\le x\}
\right|\nonumber\\
&\qquad\le C\left\{\frac{\log^7(pN_{\mathrm{OU},n})}{p}\right\}^{1/6}
\le C\left\{\frac{(K_n+\log p)^7}{p}\right\}^{1/6}=o(1).
\label{eq:S14-grid-gaussian-approximation}
\end{align}
The last equality follows from $K_n^7/p\to0$ and $(\log p)^7/p\to0$.
\endgroup

\begingroup
For either $X=\widetilde G_{p,1/2}$ or $X=Z_{\mathrm{OU}}$, Gaussian
hypercontractivity \citep[Theorem~5.10]{janson1997gaussian} yields
\begin{equation}
\|X(u)-X(v)\|_q\le Cq|u-v|^{1/2},\qquad q\ge2,\quad |u-v|\le1.
\label{eq:S14-hypercontractive-increment}
\end{equation}
To bound the modulus, bisect each grid interval. At level $m\ge0$ there
are at most $CN_{\mathrm{OU},n}2^m$ parent--child increments, each of
length at most $\delta_{u,n}2^{-m}$. Denote these increments by
$\Delta_{m,l}X$, and set $q_m=2\log(CN_{\mathrm{OU},n})+2m+2$. Markov's inequality gives
\begin{align*}
\Pr\{\exists m,l:\ |\Delta_{m,l}X|>Ce q_m\delta_{u,n}^{1/2}2^{-m/2}\}
&\le\sum_{m\ge0}CN_{\mathrm{OU},n}2^m e^{-q_m}
\le C/N_{\mathrm{OU},n}=o(1),\\
\sum_{m\ge0}q_m\delta_{u,n}^{1/2}2^{-m/2}
&\le C\delta_{u,n}^{1/2}\log(eN_{\mathrm{OU},n}).
\end{align*}
Path continuity and the two adjacent grid intervals therefore give
\begin{equation}
\omega_{X,n}(a):=
\sup_{\substack{|u|,|v|\le R_{\mathrm{OU},n}\\|u-v|\le a}}
|X(u)-X(v)|,\qquad
\omega_{X,n}(2\delta_{u,n})
=O_p\{\delta_{u,n}^{1/2}\log(eN_{\mathrm{OU},n})\}.
\label{eq:S14-Gtilde-grid-modulus-bound}
\end{equation}
The observed grid $\mathcal U_n^{\mathrm{obs}}=\{\log(k/(n-k)):k\in\mathcal K_{1/2}\}$
has the same endpoints and satisfies
\[
\max_k(u_{k+1}-u_k)\le C/\lambda_n=o(\delta_{u,n}),\qquad
A_{\mathrm{DE}}(L_{\mathrm{OU},n})\delta_{u,n}^{1/2}\log(eN_{\mathrm{OU},n})
\le C\frac{K_n^{3/2}}{(\log n)^4}\to0.
\]
Consequently,
\begin{equation}\label{eq:lem13-grid}
A_{\mathrm{DE}}(L_{\mathrm{OU},n})
\left|\max_{u\in\mathcal U_n^{\mathrm{obs}}}\widetilde G_{p,1/2}(u)
-\max_{u\in\mathcal U_n}\widetilde G_{p,1/2}(u)\right|=o_p(1),
\end{equation}
and
\begin{equation}
\begin{aligned}
A_{\mathrm{DE}}(L_{\mathrm{OU},n})
\left|\max_{k\in\mathcal K_{1/2}}G_{p,1/2}(k/n)
-\sup_{|u|\le R_{\mathrm{OU},n}}\widetilde G_{p,1/2}(u)\right|&=o_p(1),\\
A_{\mathrm{DE}}(L_{\mathrm{OU},n})
\left|\max_{u\in\mathcal U_n}Z_{\mathrm{OU}}(u)
-\sup_{|u|\le R_{\mathrm{OU},n}}Z_{\mathrm{OU}}(u)\right|&=o_p(1).
\end{aligned}
\label{eq:S14-grid-modulus}
\end{equation}
For $r(u)=e^{-|u|}$ and every $c>1$,
\[
1-r(u)=|u|+O(u^2),\qquad
L_{\mathrm{OU},n}\sup_{|u|\ge c\log L_{\mathrm{OU},n}}r(u)
\le L_{\mathrm{OU},n}^{1-c}\to0.
\]
The stationary Gaussian extreme-value theorem
\endgroup
\citep[Chapter~12]{leadbetter1983extremes} gives
\begingroup
\begin{equation}
\Pr\left\{A_{\mathrm{DE}}(L_{\mathrm{OU},n})
\sup_{|u|\le L_{\mathrm{OU},n}/2}Z_{\mathrm{OU}}(u)
-D_{\mathrm{DE}}(L_{\mathrm{OU},n})\le x\right\}
\to\exp\{-e^{-x}\}.
\label{eq:S14-OU-extreme-limit}
\end{equation}
Equations~\eqref{eq:S14-grid-gaussian-approximation},
\eqref{eq:lem13-grid} and \eqref{eq:S14-grid-modulus} transfer this limit
to the observed quadratic maximum.
\endgroup
\end{proof}

\medskip
\noindent\textit{Proof of the continuity of $F_V$.}
By \eqref{eq:S14-OU-representation}, if $B_{\mathrm{BM}}$ is a standard
Brownian motion, then
\[
\{Z_{\mathrm{OU}}(u):u\in\mathbb R\}
\overset d=
\{e^{-u}B_{\mathrm{BM}}(e^{2u}):u\in\mathbb R\}.
\]
With $x=e^u=t/(1-t)$ and $r=x^2$,
\begin{align*}
t(1-t)Z_{\mathrm{OU}}\!\left(\log\frac{t}{1-t}\right)
&\overset d=
\frac{x}{(1+x)^2}x^{-1}B_{\mathrm{BM}}(x^2)
=\frac{B_{\mathrm{BM}}(r)}{(1+\sqrt r)^2}.
\end{align*}
Consequently,
\[
\sup_{0\le t\le1}V(t)
\overset d=
\sup_{r\ge0}\frac{B_{\mathrm{BM}}(r)}{(1+\sqrt r)^2}.
\]
The sample paths of $V$ are continuous on $[0,1]$, so their supremum is
finite and nonnegative.  The Brownian representation and the reflection
principle give
\[
\Pr\left\{\sup_{0<r\le1}B_{\mathrm{BM}}(r)\le0\right\}=0,
\]
and hence $\Pr\{\sup_{0\le t\le1}V(t)=0\}=0$.  Fix $x>0$.  On the event
$\{\sup_{0\le t\le1}V(t)=x\}$, continuity and $V(0)=V(1)=0$ imply that the
maximum is attained in $[m^{-1},1-m^{-1}]$ for at least one integer
$m\ge3$.  Therefore
\[
\left\{\sup_{0\le t\le1}V(t)=x\right\}
\subseteq
\bigcup_{m=3}^{\infty}
\left\{\sup_{m^{-1}\le t\le1-m^{-1}}V(t)=x\right\}.
\]
For every fixed $m$, the restriction of $V$ to
$[m^{-1},1-m^{-1}]$ is a continuous Gaussian process with
\[
\inf_{m^{-1}\le t\le1-m^{-1}}\Var\{V(t)\}
=\inf_{m^{-1}\le t\le1-m^{-1}}t^2(1-t)^2>0,
\]
and it is nonconstant because its covariance is not rank one on any
nondegenerate interval.  The Gaussian-supremum density theorem of
\citet{tsirelson1976density} therefore gives an absolutely continuous
supremum distribution on this compact interval.  Each event on the
right-hand side has probability zero.  Thus the distribution function
$F_V$ of $\sup_{0\le t\le1}V(t)$ is continuous.

\begin{lemma}%
\label{lem:cf2-remainder}
Under model (1), suppose Assumptions~\ref{ass:C1}--\ref{ass:C3} and the conditions of
Theorem~\ref{Th2} hold.  Define
\begin{equation}
S_0^{(1)}=\max_{1\le k<n}\frac{W(k)-\widehat\mu_{k,M}}{\widehat\omega},
\qquad
\Delta_{S,n}=n\|\bdelta\|_2^2/\sqrt p.
\label{eq:def-first-order-unweighted-statistic}
\end{equation}
Then
\begin{equation}
|S_0-S_0^{(1)}|
=O_p\left(p^{-1/2}+\frac{\Delta_{S,n}^2}{\sqrt p}\right),
\label{eq:S12-first-order-local}
\end{equation}
In particular, under $H_0$,
\begin{equation}
|S_0-S_0^{(1)}|=O_p(p^{-1/2}).
\label{eq:S12-first-order-null}
\end{equation}
\end{lemma}

\begin{proof}
\begingroup
For $1\le k<n$, define $\Delta_k=W(k)-\widehat\mu_{k,M}$ and $\Delta_k^{\mathrm{prot}}=\min\{\Delta_k,3\widehat\mu_{k,M}/2\}$; then
\[
T_k=
\frac{\Delta_k}{\widehat\omega}
-\frac{(\Delta_k^{\mathrm{prot}})^2
-\widehat\omega^2v_k^2}
{3\widehat\mu_{k,M}^{+}\widehat\omega},
\qquad
S_0=\max_{1\le k<n}T_k.
\]
\endgroup
Since $|\max_kx_k-\max_ky_k|\le\max_k|x_k-y_k|$, we have
\(
|S_0-S_0^{(1)}|\le\max_{1\le k<n}|T_k-\Delta_k/\widehat\omega|.
\)

We first verify that the denominator in the protected correction is uniformly
positive.  
The inverse spectral representation gives
\begin{align*}
\mu_k
=\frac1{n\sqrt p}
\int_{-\pi}^{\pi}
\left|\sum_{i=1}^nc_{i,k}e^{\mathrm{i}i\lambda}\right|^2
\Tr\{\mathbf f_p(\lambda)\}\,d\lambda,
\end{align*}
\begingroup
Here $\mathbf f_p$ is the spectral density matrix of $\bepsilon_i$.
Parseval's identity gives
\endgroup
\begin{align*}
\int_{-\pi}^{\pi}
\left|\sum_{i=1}^nc_{i,k}e^{\mathrm{i}i\lambda}\right|^2d\lambda
=2\pi\sum_{i=1}^nc_{i,k}^2
=2\pi n v_k.
\end{align*}
Hence Assumption~\ref{ass:C2} and the covariance-tail estimate give,
uniformly in $1\le k<n$,
\begin{equation}\label{eq:lem14-mu-Mk}
|\mu_k-\mu_{k,M}|\le C\sqrt p\,v_k\eta_M,
\qquad
c\sqrt p\,v_k\le\mu_{k,M}\le C\sqrt p\,v_k.
\end{equation}
Moreover, 
\[
\frac{\widehat\mu_{k,M}}{\sqrt p\,v_k}
\ge
\frac{\mu_{k,M}}{\sqrt p\,v_k}
-
\frac{|\widehat\mu_{k,M}-\mu_{k,M}|}{\sqrt p\,v_k}
\ge c-o_p(1),
\]
where the last inequality follows from \eqref{eq:Th2-weighted-centering} and $r_{\mu,n}/\sqrt p\to0$. Furthermore, Theorem~\ref{Th2} gives
$\widehat\omega/\omega_p\to1$ in probability. Thus, there are fixed constants
$c_\mu,c_\omega,C_\omega>0$ such that
\begin{equation}
\Pr\left\{
\min_{1\le k<n}\frac{\widehat\mu_{k,M}}{\sqrt p\,v_k}
\ge c_\mu,\quad
c_\omega\le\widehat\omega\le C_\omega
\right\}\rightarrow1.
\label{eq:S12-positive-centering}
\end{equation}
\begingroup
On the event in \eqref{eq:S12-positive-centering}, for large $n$,
\[
\widehat\mu_{k,M}\ge c_\mu\sqrt p\,v_k>v_k/n,
\qquad \widehat\mu_{k,M}^{+}=\widehat\mu_{k,M},
\qquad |\Delta_k^{\mathrm{prot}}|\le|\Delta_k|.
\]
Hence
\endgroup
\begin{align*}
\left|T_k-\frac{\Delta_k}{\widehat\omega}\right|
\le \frac C{\sqrt p\,v_k}
\left\{\frac{\Delta_k^2}{\widehat\omega^2}+v_k^2\right\}
\le \frac C{\sqrt p}
\left\{\frac{\Delta_k^2}{\widehat\omega^2v_k}+1\right\}.
\end{align*}

It remains to control $\Delta_k^2/\widehat\omega^2v_k$. Let
\[
W^\epsilon(k)=\frac{\|\mathbf U_n(k/n)\|_2^2}{\sqrt p},
\qquad
\mathbf d_k=\E\left[
\frac1{\sqrt n}\left\{
\sum_{i=1}^k\bX_i-\frac{k}{n}\sum_{i=1}^n\bX_i
\right\}\right].
\]
Then
\[
W(k)=W^\epsilon(k)
+\frac{2\mathbf d_k^\top\mathbf U_n(k/n)}{\sqrt p}
+\frac{\|\mathbf d_k\|_2^2}{\sqrt p},
\qquad
\E W^\epsilon(k)=\mu_k.
\]
Accordingly, the identity
\[
\frac{\Delta_k}{\widehat\omega\sqrt{v_k}}
=\widetilde \Delta_{k,1}
+\widetilde \Delta_{k,2}
+\widetilde \Delta_{k,3}
+\widetilde \Delta_{k,4}
\]
holds with
\begin{align*}
\widetilde \Delta_{k,1}
&=\frac{W^\epsilon(k)-\mu_k}{\sqrt{v_k}\omega_p},
\qquad
\widetilde \Delta_{k,2}
=\frac{\|\mathbf d_k\|_2^2}{\sqrt{v_k}\sqrt p\,\widehat\omega},
\qquad
\widetilde \Delta_{k,3}
=\frac{2\mathbf d_k^\top\mathbf U_n(k/n)}
{\sqrt{v_k}\sqrt p\,\widehat\omega},
\\
\widetilde \Delta_{k,4}
&=\frac{W^\epsilon(k)-\mu_k}{\sqrt{v_k}}
\left(\frac1{\widehat\omega}-\frac1{\omega_p}\right)
+\frac{\mu_k-\mu_{k,M}}{\sqrt{v_k}\widehat\omega}
+\frac{\mu_{k,M}-\widehat\mu_{k,M}}{\sqrt{v_k}\widehat\omega}.
\end{align*}

We first consider term $\widetilde \Delta_{k,1}$. Recall
$t_k=k/n$, $u_k=\log\{t_k/(1-t_k)\}$, and construct
$\mathcal I_r^{\mathrm{logit}}=\{k:r\le u_k<r+1\}$.  If $k\in \mathcal I_r^{\mathrm{logit}}$,
$v_k\le Ce^{-|r|}$.  
By \eqref{eq:quadratic-only-bilinear-form},
\[
Z_k=\frac{\widetilde \Delta_{k,1}}{\sqrt{v_k}}=\frac{Q_{1/2,n,p}(t_k)}{\omega_p}=\omega_p^{-1}\mathcal B_n(\mathbf a_{1/2,k},\mathbf a_{1/2,k}).
\]
Lemma~\ref{lem:bilinear-quadratic-bound}, Assumption~\ref{ass:C2}, together with the coefficient
estimates \eqref{eq:quadratic-only-coefficient-norms} and
\eqref{eq:quadratic-only-full-increment} established in
Lemma~\ref{lem:quadratic-only-uniform-comparison}, give
\begin{align}
\sup_{1\le k<n}\left\|\frac{\widetilde \Delta_{k,1}}{\sqrt{v_k}}\right\|_4\le C,
\qquad
\E\left|\frac{\widetilde \Delta_{k_2,1}}{\sqrt{v_{k_2}}}-\frac{\widetilde \Delta_{k_1,1}}{\sqrt{v_{k_1}}}\right|^4
\le C|u_{k_2}-u_{k_1}|^2,
\label{eq:S12-noise-logit-increment}
\end{align}
with $k_1<k_2$.
For every nonempty set $\mathcal I_r^{\mathrm{logit}}$, let
$k_r^-=\min\mathcal I_r^{\mathrm{logit}}$.  Since its logit length is at
most one, the maximal-moment inequality of
\citet[Theorem~1]{moricz1982}, applied to the second inequality in 
\eqref{eq:S12-noise-logit-increment} with moment order four and interval
exponent two, gives
\[
\E\left[
\max_{k\in\mathcal I_r^{\mathrm{logit}}}
\left|\frac{\widetilde \Delta_{k,1}}{\sqrt{v_k}}-\frac{\widetilde \Delta_{k_r^-,1}}{\sqrt{v_{k_r^-}}}\right|^4
\right]\le C.
\]
Together with the first inequality in \eqref{eq:S12-noise-logit-increment}, this yields 
\begin{equation}
\sup_{r:\,\mathcal I_r^{\mathrm{logit}}\ne\varnothing}
\E\left[
\max_{k\in\mathcal I_r^{\mathrm{logit}}}
\left|\frac{\widetilde \Delta_{k,1}}{\sqrt{v_k}}\right|^4
\right]\le C.
\label{eq:S12-noise-shell-maximal}
\end{equation}
{Markov's inequality, \eqref{eq:S12-noise-shell-maximal}, and}
$v_k\le Ce^{-|r|}$ on $\mathcal I_r^{\mathrm{logit}}$ give, for $x>0$,
\begin{align*}
\Pr\left\{
\max_{1\le k<n}\widetilde \Delta_{k,1}^2>x
\right\}
\le
\sum_{r\in\mathbb Z}
\Pr\left\{
\max_{k\in\mathcal I_r^{\mathrm{logit}}}\left|\frac{\widetilde\Delta_{k,1}^2}{v_k}\right|>C^{-1}xe^{|r|}
\right\}
\le Cx^{-2}\sum_{r\in\mathbb Z}e^{-2|r|}\le Cx^{-2}.
\end{align*}
Consequently,
\(
\max_{1\le k<n}
\widetilde \Delta_{k,1}^2=O_p(1).
\)

For the second term $\widetilde\Delta_{k,2}$, let $\theta=\tau/n$ and $t_k=k/n$.  Direct summation under
model~1 gives
\[
\mathbf d_k=-\sqrt n\,\alpha_k\bdelta,
\qquad
\alpha_k=
\begin{cases}
t_k(1-\theta),&k\le\tau,\cr
\theta(1-t_k),&k>\tau.
\end{cases}
\]
Since $v_k=t_k(1-t_k)$,
\[
\frac{\alpha_k^2}{v_k}
\le
\begin{cases}
\theta(1-\theta),&k\le\tau,\cr
\theta(1-\theta),&k>\tau,
\end{cases}
\le\frac14.
\]
Consequently, on the event in
\eqref{eq:S12-positive-centering},
\begin{align*}
\max_{1\le k<n}\widetilde\Delta_{k,2}^2
=\max_{1\le k<n}
\frac{n^2\alpha_k^4\|\bdelta\|_2^4}
{p\widehat\omega^2v_k}
\le C\frac{n^2\|\bdelta\|_2^4}{p}
=C\Delta_{S,n}^2.
\end{align*}

For the third term $\widetilde\Delta_{k,3}$, we have
\[
\mathbf d_k^\top\mathbf U_n(t_k)
=-\alpha_k\left(
\sum_{i=1}^k\bdelta^\top\bepsilon_i
-\frac{k}{n}\sum_{i=1}^n\bdelta^\top\bepsilon_i
\right).
\]
For every interval $\{u,\ldots,v\}$, Lemma~\ref{lem:vector-basic}(ii)
with $q=4$ and Assumption~\ref{ass:C2} give
\[
\E\left|\sum_{i=u}^v\bdelta^\top\bepsilon_i\right|^4
\le C(v-u+1)^2(\bdelta^\top\bOme\bdelta)^2.
\]
The maximal-moment inequality of \citet[Theorem~1]{moricz1982} therefore
yields
\[
\E\max_{1\le j\le n}
\left|\sum_{i=1}^j\bdelta^\top\bepsilon_i\right|^4
\le Cn^2(\bdelta^\top\bOme\bdelta)^2.
\]
Since
\[
\max_{1\le k<n}
\left|
\sum_{i=1}^k\bdelta^\top\bepsilon_i
-\frac{k}{n}\sum_{i=1}^n\bdelta^\top\bepsilon_i
\right|
\le2\max_{1\le j\le n}
\left|\sum_{i=1}^j\bdelta^\top\bepsilon_i\right|,
\]
Markov's inequality gives
\[
\max_{1\le k<n}
\left|
\sum_{i=1}^k\bdelta^\top\bepsilon_i
-\frac{k}{n}\sum_{i=1}^n\bdelta^\top\bepsilon_i
\right|^2
=O_p\{n\bdelta^\top\bOme\bdelta\}.
\]
It follows that, on the same event,
\begin{align}
\max_{1\le k<n}\widetilde\Delta_{k,3}^2
=\max_{1\le k<n}
\frac{4\alpha_k^2}{p\widehat\omega^2v_k}
\left|
\sum_{i=1}^k\bdelta^\top\bepsilon_i
-\frac{k}{n}\sum_{i=1}^n\bdelta^\top\bepsilon_i
\right|^2
=O_p\left(\frac{n\bdelta^\top\bOme\bdelta}{p}\right).
\label{eq:S14-cross-term-bound}
\end{align}

For the last term, Equations~\eqref{eq:Th2-weighted-centering}, \eqref{eq:lem14-mu-Mk}
, and the scale bound in
Theorem~\ref{Th2} yield
\[
\max_{1\le k<n}
\widetilde\Delta_{k,4}^2
=O_p\left(r_{\mu,n}^2+r_{\omega,n}^2+p\eta_M^2\right).
\]
Since $n\boldsymbol\delta^\top\boldsymbol\Omega\boldsymbol\delta/p
\le Cn\|\boldsymbol\delta\|_2^2/p
=C\Delta_{S,n}/\sqrt p\le C(1+\Delta_{S,n}^2)$, 
\begin{align}
\max_{1\le k<n}
\frac{\Delta_k^2}{\widehat\omega^2v_k}
&{}\le 4\sum\limits_{j=1}^4\max_{1\le k<n}
\widetilde\Delta_{k,j}^2
=O_p\left(1+\Delta_{S,n}^2\right)+O_p\left(r_{\mu,n}^2+r_{\omega,n}^2+p\eta_M^2\right)=O_p(1+\Delta_{S,n}^2),
\label{eq:S12-score-bound}
\end{align}
under conditions in Theorem~\ref{Th2}.
It yields
\[
|S_0-S_0^{(1)}|\le\max_{1\le k<n}\left|T_k-\frac{\Delta_k}{\widehat\omega}\right|
=O_p\left(p^{-1/2}+\frac{\Delta_{S,n}^2}{\sqrt p}\right),
\]
which proves \eqref{eq:S12-first-order-local}.  Under $H_0$,
$\bdelta=\mathbf0$ and hence {the second and third terms vanish,} so the same bound reduces to
$|S_0-S_0^{(1)}|=O_p(p^{-1/2})$.
\end{proof}

\begin{proof}[Proof of Theorem~\ref{Tms-1}]
Throughout the proof, $H_0$ holds.  Assumption~\ref{ass:C2} gives
$0<c\le\omega_p\le C<\infty$, and Theorem~\ref{Th2} gives
\begin{equation}
\max_{\gamma\in\{0,1/2\}}
\max_{k\in\mathcal K_\gamma}
|\widehat\mu_{\gamma,k,M}-\mu_{\gamma,k,M}|
=O_p(r_{\mu,n}),
\qquad
\left|\frac{\widehat\omega}{\omega_p}-1\right|
=O_p(r_{\omega,n}).
\label{eq:Th3-feasible-rates}
\end{equation}

\textit{(i).}
Define
\[
S_{0}^{\dag}
=\max_{1\le k<n}
\frac{W(k)-\mu_{k,M}}{\omega_p},
\qquad
S_{0}^{\ddag}
=\max_{1\le k<n}
\frac{W(k)-\widehat\mu_{k,M}}{\omega_p},
\]
and $S_0^{(1)}=\omega_p/\widehat\omega S_{0}^{\ddag}$.
Theorem~\ref{Th1}, the continuity of the supremum map at every
continuous function, and the fact that $V$ has continuous sample paths give
\begin{equation}
S_{0}^{\dag}\ \xrightarrow{d}\
\sup_{0\le t\le1}V(t).
\label{eq:Th3-unweighted-population}
\end{equation}
In particular, $S_{0}^{\dag}=O_p(1)$.  By
\eqref{eq:Th3-feasible-rates},
\[
|S_{0}^{\ddag}-S_{0}^{\dag}|
\le\omega_p^{-1}
\max_{1\le k<n}
|\widehat\mu_{k,M}-\mu_{k,M}|
=O_p(r_{\mu,n})=o_p(1).
\]
Moreover, $S_0^{(1)}=\omega_p/\widehat\omega S_{0}^{\ddag}$,
so \eqref{eq:Th3-feasible-rates} and $S_{0}^{\ddag}=O_p(1)$ imply
$S_0^{(1)}-S_{0}^{\dag}=o_p(1)$.  Hence
\eqref{eq:Th3-unweighted-population} and Slutsky's theorem yield
\[
S_0^{(1)}\ \xrightarrow{d}\
\sup_{0\le t\le1}V(t).
\]
Finally, \eqref{eq:S12-first-order-null} gives
$|S_0-S_0^{(1)}|=O_p(p^{-1/2})=o_p(1)$, proving part~(i).

(ii).
Since $\lambda_n\asymp n^\lambda$, we have
$\log\log h_n\asymp\log\log n$,
$A_{\mathrm{DE}}(\log h_n)\asymp (\log\log h_n)^{1/2}$, and
$D_{\mathrm{DE}}(\log h_n)/A_{\mathrm{DE}}(\log h_n)\asymp (\log\log h_n)^{1/2}$.  Define
\begin{align*}
X_n^\epsilon
&=\max_{k\in\mathcal K_{1/2}}
\frac{Q_{1/2,n,p}(k/n)}{\omega_p},
&
X_n^G
&=\max_{k\in\mathcal K_{1/2}}
\frac{Q_{1/2,n,p}^G(k/n)}{\omega_p},
&
X_n^B
&=\max_{k\in\mathcal K_{1/2}}
G_{p,1/2}(k/n).
\end{align*}
{Recall $\ell_n=\log(n)+1$. Taking} $N=|\mathcal K_{1/2}|\le n$ and
$\beta=\ell_n^{2}$ in
\eqref{eq:smax-logsumexp-bound}, 
\begin{align}
0\le
\operatorname{smax}_{\ell_n^2}(\mathbf z)-\max \mathbf z
\le\frac{\log n}{\ell_n^2}\lesssim\frac{1}{\ell_n},
\qquad
\frac{A_{\mathrm{DE}}(\log h_n)}{\ell_n}\rightarrow0.
\label{eq:Th3-smooth-maximum-error}
\end{align}

For any $g\in C_c^3(\mathbb R)$, we take
\[
f(x)=g(A_{\mathrm{DE}}(\log h_n)x-D_{\mathrm{DE}}(\log h_n)),
\]
its first three derivatives are bounded by $C_g(1+\{A_{\mathrm{DE}}(\log h_n)\}^3)$. 
Applying
\eqref{eq:quadratic-only-block-comparison} and \eqref{eq:quadratic-only-covariance-restoration}, and \eqref{eq:Th3-smooth-maximum-error}, we have
\begingroup
\[
\begin{aligned}
&\left|
\E g\!\left\{A_{\mathrm{DE}}(\log h_n)X_n^\epsilon
-D_{\mathrm{DE}}(\log h_n)\right\}
-\E g\!\left\{A_{\mathrm{DE}}(\log h_n)X_n^B
-D_{\mathrm{DE}}(\log h_n)\right\}
\right|\\
&\quad\le C_g\left[
\frac{A_{\mathrm{DE}}(\log h_n)}
{\ell_n}
+(1+A_{\mathrm{DE}}(\log h_n)^3)
\left\{
r_{n,1/2}^{\mathrm{quad}}
+\frac{\ell_n^{c_{\mathrm{br},\mathrm q}}}{\lambda_n}
\right\}
\right]=o(1),
\end{aligned}
\]
\endgroup
{Lemma~\ref{lem:dense-boundary-extreme} and the convergence-determining property of $C_c^3(\mathbb R)$ imply}
\begin{equation}
A_{\mathrm{DE}}(\log h_n)X_n^\epsilon-D_{\mathrm{DE}}(\log h_n)\ \xrightarrow{d}\ Z_{\mathrm{Gu}}.
\label{eq:Th3-population-boundary-limit}
\end{equation}

Under $H_0$,
$Q_{1/2,n,p}(k/n)=W_{1/2}(k)-\mu_{1/2,k}$, and
\[
X_n^\epsilon=\max_{k\in\mathcal K_{1/2}}
\frac{W_{1/2}(k)-\mu_{1/2,k}}{\omega_p}.
\] 
Define
\[
S_{1/2}^{\dag}
=\max_{k\in\mathcal K_{1/2}}
\frac{W_{1/2}(k)-\mu_{1/2,k,M}}{\omega_p},
\qquad
S_{1/2}^{\ddag}
=\max_{k\in\mathcal K_{1/2}}
\frac{W_{1/2}(k)-\widehat\mu_{1/2,k,M}}{\omega_p},
\]
and $S_{1/2}=(\omega_p/\widehat\omega)S_{1/2}^{\ddag}$.
Since the maximum map is 1-Lipschitz with respect to the sup norm, with Lemma~\ref{lem:cf2-remainder},
\begin{equation}
A_{\mathrm{DE}}(\log h_n)|S_{1/2}^{\dag}-X_n^\epsilon|
\le
\frac{A_{\mathrm{DE}}(\log h_n)}{\omega_p}
\max_{k\in\mathcal K_{1/2}}
|\mu_{1/2,k}-\mu_{1/2,k,M}|
\le CA_{\mathrm{DE}}(\log h_n)\sqrt p\,\eta_M .
\label{eq:Th3-truncated-centering}
\end{equation}
Equations~\eqref{eq:Th3-population-boundary-limit} and
\eqref{eq:Th3-truncated-centering} give
\begin{equation}
A_{\mathrm{DE}}(\log h_n)S_{1/2}^{\dag}-D_{\mathrm{DE}}(\log h_n)
\ \xrightarrow{d}\ Z_{\mathrm{Gu}},
\qquad
S_{1/2}^{\dag}=O_p((\log\log h_n)^{1/2}).
\label{eq:Th3-infeasible-boundary-limit}
\end{equation}
The signal terms in Theorem~\ref{Th2} vanish under $H_0$, so
\[
r_{\mu,n}=\frac M{\sqrt n}+\sqrt p\,M\eta_M,
\qquad
r_{\omega,n}
=M^2\left(\frac1{np}+\frac1{n^2}\right)^{1/2}
+M^2\eta_M.
\]
\begingroup
By \eqref{eq:Th3-feasible-rates}, \eqref{eq:Th3-infeasible-boundary-limit}
and the rate calculation below,
\endgroup
\[
\left|S_{1/2}^{\ddag}-S_{1/2}^{\dag}\right|
=O_p(r_{\mu,n}),
\qquad
S_{1/2}^{\ddag}=O_p((\log\log h_n)^{1/2}).
\]
Thus,
\begin{align}
A_{\mathrm{DE}}(\log h_n)\left|S_{1/2}-S_{1/2}^{\dag}\right|
&\le
A_{\mathrm{DE}}(\log h_n)\left|S_{1/2}^{\ddag}-S_{1/2}^{\dag}\right|
+A_{\mathrm{DE}}(\log h_n)\left|\frac{\omega_p}{\widehat\omega}-1\right|
|S_{1/2}^{\ddag}|\nonumber\\
&=O_p\{(\log\log h_n)^{1/2}r_{\mu,n}+(\log\log h_n)r_{\omega,n}\}
\label{eq:Th3-feasible-boundary-error}
\end{align}
\begingroup
Let $K_n=\log(2+\log n)$. Since $M\le Cn^{1/8}$,
\begin{align*}
&K_n^{1/2}r_{\mu,n}+K_nr_{\omega,n}
+K_n^{1/2}\sqrt p\,\eta_M\\
&\quad\le C\{K_n^{1/2}n^{-3/8}
+K_nn^{-1/4}p^{-1/2}+K_nn^{-3/4}\}
+C(K_n^{1/2}\sqrt p M+K_nM^2)\eta_M.
\end{align*}
The first three terms vanish. Under (MD), $\eta_M=0$ eventually.
Under (AM) or (PD), $\eta_M\le C\exp(-cM^{\varkappa})$, and
\[
(K_n^{1/2}\sqrt p M+K_nM^2)\eta_M
\le
\begin{cases}
C(p^{1/14+5/8}+p^{1/7+1/4})e^{-cp^{\varkappa/8}},&p\le n,\\
\exp\{\tfrac12\log p+O(\log n)-cn^{\varkappa/8}\},&p>n,
\end{cases}
\longrightarrow0.
\]
Here the first case uses $K_n=o(p^{1/7})$, and the second uses
$\log p=o(n^{\varkappa/8})$. Since $\log\log h_n\asymp K_n$,
\endgroup
\begin{equation}
(\log\log h_n)^{1/2}r_{\mu,n}+(\log\log h_n)r_{\omega,n}\rightarrow0,
\qquad
A_{\mathrm{DE}}(\log h_n)\sqrt p\,\eta_M\rightarrow0.
\label{eq:Th3-strengthened-feasible-rate}
\end{equation}
Combining this with
\eqref{eq:Th3-infeasible-boundary-limit} and applying Slutsky's theorem
proves part~(ii).
\end{proof}

\begin{proof}[Proof of Proposition~\ref{prop:Tms-alter}]
\begingroup
For both parts, we bound the score at the true change point and compare it
with the corresponding critical value.
\endgroup
For $1\le k<n$, let
\begingroup
\[
\Delta_k=W(k)-\widehat\mu_{k,M},
\qquad
T_k=
\frac{\Delta_k}{\widehat\omega}
-\frac{(\Delta_k^{\mathrm{prot}})^2-\widehat\omega^2v_k^2}
{3\widehat\mu_{k,M}^{+}\widehat\omega},
\qquad
S_0=\max_{1\le k<n}T_k.
\]
\endgroup
Let $\theta=\tau/n$ and
\[
\mathbf d_k=\E\left[n^{-1/2}
\left\{\sum_{i=1}^k\bX_i-\frac{k}{n}\sum_{i=1}^n\bX_i\right\}\right].
\]
At $k=\tau$,
\[
\mathbf d_\tau=-\sqrt n\,\theta(1-\theta)\bdelta.
\]
For $\gamma\in\{0,1/2\}$, at $k=\tau$,
\begin{align*}
Q_{\gamma,\tau}
={}&v(\theta)^{-2\gamma}
\frac{\|\mathbf d_\tau\|_2^2}{\sqrt p}
+2v(\theta)^{-2\gamma}
\frac{\mathbf d_\tau^\top\mathbf U_n(\theta)}{\sqrt p}\\
&+Q_{\gamma,n,p}(\theta)
+\mu_{\gamma,\tau}-\mu_{\gamma,\tau,M}
+\mu_{\gamma,\tau,M}-\widehat\mu_{\gamma,\tau,M}.
\end{align*}
Because $\theta\in[\vartheta,1-\vartheta]$, the deterministic term is at
least $c_\vartheta\Delta_{S,n}$.  Lemma~\ref{lem:vector-basic}(ii) and the
bounded eigenvalues of $\bOme$ give
\[
\mathbf d_\tau^\top\mathbf U_n(\theta)
=O_p\{n^{1/2}(\bdelta^\top\bOme\bdelta)^{1/2}\}
=O_p\{p^{1/4}\Delta_{S,n}^{1/2}\},
\]
\begingroup
so the normalized cross term is
$O_p(p^{-1/4}\Delta_{S,n}^{1/2})$.
By Lemma~\ref{lem:quadratic-cumulants},
\eqref{eq:lem14-mu-Mk} multiplied by $v(\theta)^{-2\gamma}$,
and Theorem~\ref{Th2}, respectively,
\begin{align*}
Q_{\gamma,n,p}(\theta)&=O_p(1),\\
|\mu_{\gamma,\tau}-\mu_{\gamma,\tau,M}|
&\le C\sqrt p\,\eta_M\le Cr_{\mu,n}=o(1),\\
|\widehat\mu_{\gamma,\tau,M}-\mu_{\gamma,\tau,M}|
&=O_p(r_{\mu,n})=o_p(1),
\qquad \widehat\omega/\omega_p\xrightarrow{p}1.
\end{align*}
Here $0<c\le\omega_p\le C$ by Assumption~\ref{ass:C2}.
Hence
\endgroup
\[
\frac{\Delta_\tau}{\widehat\omega}
\ge c\Delta_{S,n}
-O_p\{\Delta_{S,n}^{1/2}+1\}.
\]
\begingroup
On the event in \eqref{eq:S12-positive-centering},
$a=\widehat\mu_{\tau,M}=\widehat\mu_{\tau,M}^{+}>0$.
For $x=\Delta_\tau\ge0$,
\endgroup
\[
x-\frac{\min\{x,3a/2\}^2-\widehat\omega^2v_\tau^2}{3a}
\ge x-\frac{\min\{x,3a/2\}^2}{3a}
\ge\frac{x}{2}.
\]
Consequently,
\[
S_0\ge T_\tau
\ge \frac{\Delta_\tau}{2\widehat\omega}\ge c\Delta_{S,n}
-O_p\{\Delta_{S,n}^{1/2}+1\}
\rightarrow+\infty,
\]
whenever $\Delta_{S,n}\to\infty$, which proves part (i).

\begingroup
For part (ii), $\lambda_n/n\to0$ and
$\theta\in[\vartheta,1-\vartheta]$ imply
$\tau\in\mathcal K_{1/2}$ eventually. The preceding decomposition
and bounds, with $\gamma=1/2$, give
\[
Q_{1/2,\tau}
=v(\theta)\Delta_{S,n}
+O_p\{1+p^{-1/4}\Delta_{S,n}^{1/2}\}+o_p(1).
\]
Since $\Delta_{S,n}\to\infty$ and $\widehat\omega/\omega_p\to1$
in probability, for some fixed $c>0$,
\[
\Pr\left\{S_{1/2}\ge
\frac{Q_{1/2,\tau}}{\widehat\omega}
\ge c\Delta_{S,n}\right\}\longrightarrow1.
\]
For fixed $\alpha\in(0,1)$, the test rejects when
$S_{1/2}>c_{n,\alpha}$, where
\[
c_{n,\alpha}=
\frac{D_{\mathrm{DE}}(\log h_n)+F_{\mathrm{Gu}}^{-1}(1-\alpha)}
{A_{\mathrm{DE}}(\log h_n)}
\sim\sqrt{2\log\log h_n}.
\]
Thus $\Delta_{S,n}/\sqrt{\log\log h_n}\to\infty$ implies
$\Pr(S_{1/2}>c_{n,\alpha})\to1$, proving part (ii).
\endgroup
\end{proof}

\section{Proofs of Theorems 4 and 5, and Propositions 2 and 3}
\label{sec:supp-testing-null}

{This section proves the max-$L_\infty$ limits, the joint null limits of the quadratic and coordinatewise maximum statistics, and the size result for the Cauchy combination tests.}  The componentwise long-run variance rate used for feasible normalization is proved in Section~\ref{subsec:supp-lrv-estimation}.

\subsection{Notation}
Let
\[
\sigma_j^2=(\bOme)_{jj},\qquad
\rho_{jj'}=\frac{(\bOme)_{jj'}}{\sigma_j\sigma_{j'}},
\qquad
t_k=\frac{k}{n},\qquad
v_k=t_k(1-t_k).
\]
Let $\mathbf e_j$ be the $j$th canonical basis vector of
$\mathbb R^p$, and define
\[
\mathbf Z(t)=\Diag(\bOme)^{-1/2}\bOme^{1/2}\mathbf B(t),
\qquad 0\le t\le1.
\]
Then every coordinate $Z_j$ is a standard Brownian bridge and
\(
\Cov\{Z_j(s),Z_{j'}(t)\}=K_B(s,t)\rho_{jj'}.
\)

For $\gamma\in\{0,1/2\}$, $j\le p$, and
$k\in\mathcal K_\gamma$, define the max-coordinate arrays
\begin{equation}
\begin{aligned}
&C_{\gamma,j}^{\epsilon,0}(k)=
\sum_{i=1}^n a_{i,(\gamma,k)}\frac{\epsilon_{ij}}{\sigma_j},
\qquad
C_{\gamma,j}^{G,0}(k)=
\sum_{i=1}^n a_{i,(\gamma,k)}\frac{\epsilon_{ij}^{G}}{\sigma_j},
\\
&C_{\gamma,j}^{B,0}(k)
=
v_k^{-\gamma}Z_j(t_k)
=
\left\{\frac{k}{n}\left(1-\frac{k}{n}\right)\right\}^{-\gamma}
\frac{\mathbf e_j^\top\bOme^{1/2}\mathbf B(k/n)}{\sigma_j}.
\end{aligned}
\label{eq:max-coordinate-arrays}
\end{equation}

For any deterministic array
$\mathbf d=(d_{\varsigma,j,k}:\varsigma\in\{-1,1\},j\le p,
k\in\mathcal K_\gamma)$, define the shifted max-coordinate arrays
\begin{equation}
\begin{aligned}
\mathcal C_{\gamma,n}^{\pm}(\mathbf d)
&=
\left\{
\varsigma C_{\gamma,j}^{\epsilon,0}(k)+d_{\varsigma,j,k}:
\varsigma\in\{-1,1\},\ j\le p,\ k\in\mathcal K_\gamma
\right\},\\
(\mathcal C_{\gamma,n}^{G})^{\pm}(\mathbf d)
&=
\left\{
\varsigma C_{\gamma,j}^{G,0}(k)+d_{\varsigma,j,k}:
\varsigma\in\{-1,1\},\ j\le p,\ k\in\mathcal K_\gamma
\right\},\\
(\mathcal C_{\gamma,n}^{B})^{\pm}(\mathbf d)
&=
\left\{
\varsigma C_{\gamma,j}^{B,0}(k)+d_{\varsigma,j,k}:
\varsigma\in\{-1,1\},\ j\le p,\ k\in\mathcal K_\gamma
\right\}.
\end{aligned}
\label{eq:def-shifted-max-coordinate-arrays}
\end{equation}
For zero shifts, with
$\mathbf0=(0:\varsigma\in\{-1,1\},j\le p,k\in\mathcal K_\gamma)$, let
\begin{align*}
\mathcal C_{\gamma,n}^{\pm}
&=\mathcal C_{\gamma,n}^{\pm}(\mathbf0),&
(\mathcal C_{\gamma,n}^{G})^{\pm}
&=(\mathcal C_{\gamma,n}^{G})^{\pm}(\mathbf0),&
(\mathcal C_{\gamma,n}^{B})^{\pm}
&=(\mathcal C_{\gamma,n}^{B})^{\pm}(\mathbf0).
\end{align*}
Furthermore, we define %
three corresponding max statistics, for $\gamma\in\{0,1/2\}$,
\begin{equation}
\begin{aligned}
M_{n,p}^{\epsilon,\gamma}=\max_{k\in\mathcal K_\gamma}\max_{j\le p}
|C_{\gamma,j}^{\epsilon,0}(k)|,
\quad
M_{n,p}^{G,\gamma}=\max_{k\in\mathcal K_\gamma}\max_{j\le p}
|C_{\gamma,j}^{G,0}(k)|,
\quad
M_{n,p}^{B,\gamma}=\max_{k\in\mathcal K_\gamma}\max_{j\le p}
|C_{\gamma,j}^{B,0}(k)|.
\end{aligned}
\label{eq:def-max-coordinate-statistics}
\end{equation}

\subsection{{Proof of Theorem~\ref{null:Max}--\ref{indnull}, Proposition~\ref{prop:consistency-Linf}--\ref{prop:cauchy-size}}%
}

{We first recall the correlation conditions in Assumption~2 of \citet{wang2023JRSSB}.}
\begin{assumption}[Componentwise correlations]\label{ass:C4}
Let $\sigma_j^2=(\bOme)_{jj}$ and
$\rho_{jj'}=(\bOme)_{jj'}/(\sigma_j\sigma_{j'})$. {(i) There exists a constant} $\varrho\in(0,1)$ such that $|\rho_{jj'}|\le\varrho$ for $j\ne j'$. (ii) $|\mathcal H_p|/p\to0$ for sequences $\delta_p=o\{(\log p)^{-1}\}$ and $\upsilon_p\to0$ as $p\rightarrow\infty$, where $\mathcal N_{p,j}=\{j'\ne j:|\rho_{jj'}|\ge\delta_p\}$ and $\mathcal H_p=\{j:|\mathcal N_{p,j}|\ge p^{\upsilon_p}\}$.
\end{assumption}

{The next lemma shows that Assumption~\ref{ass:C2} implies Assumption~\ref{ass:C4}.}
\begin{lemma}
\label{lem:bounded-spectrum-weak-corr}

Let $\{\mathbf\Sigma_d\}_{d\geq1}$ be a sequence of $d\times d$ positive-definite matrices such that, for some constants $0<C_0<C_1<\infty$ independent of $d$,
\(
C_0
\leq
\lambda_{\min}(\mathbf\Sigma_d)
\leq
\lambda_{\max}(\mathbf\Sigma_d)
\leq
C_1.
\)
Define
\(
\mathbf D_d
=
\Diag(\mathbf\Sigma_d),
\)
\(
\mathbf R_d
=
\mathbf D_d^{-1/2}
\mathbf\Sigma_d
\mathbf D_d^{-1/2}
=
(\rho_{jk}^{(d)})_{1\leq j,k\leq d}.
\)
Then the following properties hold.

\begin{enumerate}
\item[(i)]
There exists a constant $\varrho\in(0,1)$, independent of $d$,
such that $\max_{1\leq j\neq k\leq d}|\rho_{jk}^{(d)}|\leq\varrho$. %

\item[(ii)]

There exist sequences $\delta_d>0$ and $\nu_d>0$ such that, as $d\to\infty$, $\delta_d=o\{(\log d)^{-1}\}$, $\nu_d\rightarrow0$ with $\mathcal N_{d,j}=\left\{k\neq j:|\rho_{jk}^{(d)}|\geq\delta_d\right\}$ and $\mathcal H_d=\left\{j:|\mathcal N_{d,j}|\geq d^{\nu_d}\right\}$, we have $|\mathcal H_d|/d\rightarrow0$. %
\end{enumerate}

\end{lemma}

\begin{proof}
Since $C_0\leq\lambda_{\min}(\mathbf\Sigma_d)\leq(\mathbf\Sigma_d)_{jj}\leq\lambda_{\max}(\mathbf\Sigma_d)\leq C_1$, we have $C_0\mathbf I_d\preceq \mathbf D_d \preceq C_1\mathbf I_d$.
{For every $\mathbf x\in\mathbb R^d$, the identity $\mathbf x^\top\mathbf R_d\mathbf x=(\mathbf D_d^{-1/2}\mathbf x)^\top\mathbf\Sigma_d(\mathbf D_d^{-1/2}\mathbf x)$ gives}
\[
\begin{aligned}
\mathbf x^\top\mathbf R_d\mathbf x \geq C_0\|\mathbf D_d^{-1/2}\mathbf x\|_2^2\geq\frac{C_0}{C_1}\|\mathbf x\|_2^2,~~~~\mathbf x^\top\mathbf R_d\mathbf x\leq C_1 \|\mathbf D_d^{-1/2}\mathbf x\|_2^2\leq\frac{C_1}{C_0}\|\mathbf x\|_2^2.
\end{aligned}
\]
It follows that
\begin{equation}\label{eq:bounded-spectrum-corr}
\frac{C_0}{C_1}
\leq
\lambda_{\min}(\mathbf R_d)
\leq
\lambda_{\max}(\mathbf R_d)
\leq
\frac{C_1}{C_0}.
\end{equation}

For any $j\neq k$, the $2\times2$ principal submatrix of
$\mathbf R_d$ indexed by $\{j,k\}$ is
\[
\begin{pmatrix}
1 & \rho_{jk}^{(d)}\\
\rho_{jk}^{(d)} & 1
\end{pmatrix}.
\]
Its minimum eigenvalue is $1-|\rho_{jk}^{(d)}|$. By the Cauchy
interlacing theorem and Equation~\eqref{eq:bounded-spectrum-corr},
\[
1-|\rho_{jk}^{(d)}|
\geq
\lambda_{\min}(\mathbf R_d)
\geq
\frac{C_0}{C_1}.
\]
Consequently,
\(
|\rho_{jk}^{(d)}|
\leq
1-C_0/C_1,
\)
uniformly over $j\neq k$ and $d$, proving part~(i).

Next, for every $j$,
\begin{equation}\label{eq:correlation-row-square}
\sum_{k=1}^d
\{\rho_{jk}^{(d)}\}^2
=
(\mathbf R_d^2)_{jj}
=
\|\mathbf R_d\mathbf e_j\|_2^2
\leq
\|\mathbf R_d\|_{\mathrm{op}}^2
\leq
\left(\frac{C_1}{C_0}\right)^2,
\end{equation}
where $\mathbf e_j$ is the $j$th canonical basis vector. Hence, for
any $\delta_d>0$,
\[
|\mathcal N_{d,j}|\delta_d^2
\leq
\sum_{k\in\mathcal N_{d,j}}
\{\rho_{jk}^{(d)}\}^2
\leq
\left(\frac{C_1}{C_0}\right)^2,
\]
and therefore
\begin{equation}\label{eq:correlation-neighborhood-size}
\max_{1\leq j\leq d}
|\mathcal N_{d,j}|
\leq
\left(\frac{C_1}{C_0}\right)^2
\delta_d^{-2}.
\end{equation}

 Take, for all sufficiently large $d$, $\delta_d=(\log d)^{-2}$ and $\nu_d=5\log\log d/\log d$. Then $\delta_d=o\{(\log d)^{-1}\}$, $\nu_d\rightarrow0$, and $d^{\nu_d}=\exp\{5\log\log d\}=(\log d)^5$. On the other hand, Equation~\eqref{eq:correlation-neighborhood-size} gives
\[
\max_{1\leq j\leq d}
|\mathcal N_{d,j}|
\leq
\left(\frac{C_1}{C_0}\right)^2
(\log d)^4.
\]
For all sufficiently large $d$,
\[
\left(\frac{C_1}{C_0}\right)^2
(\log d)^4
<
(\log d)^5
=
d^{\nu_d}.
\]
Thus $\mathcal H_d=\varnothing$ eventually, and consequently $|\mathcal H_d|/d\rightarrow0$. This proves part~(ii).
\end{proof}

\begin{lemma}%
\label{lem:gaussian-extremes-general}
Suppose Assumption~\ref{ass:C2} holds. 
\begin{itemize}
\item [(i)] If $\log p=o\{n^{1/(2+\eta)}\}$,  {for some $\eta>0$, then, as $(n,p)\rightarrow\infty$,}
\[
2(M_{n,p}^{B,0})^2-\log(2p)
\rightarrow Z_{\mathrm{Gu}},
\qquad
\Pr(Z_{\mathrm{Gu}}\le x)=\exp\{-e^{-x}\}.
\]
\item [(ii)] If $\lambda_n\asymp n^\lambda$ for a fixed $\lambda\in(0,1)$, $\log p=o\{n^{\lambda/(2+\eta)}\}$ for some $\eta>0$, then, as $(n,p)\rightarrow\infty$,
\[
A_{\mathrm{DE}}(p\log h_n)M_{n,p}^{B,1/2}
-D_{\mathrm{DE}}(p\log h_n)
\rightarrow Z_{\mathrm{Gu}}.
\]
\end{itemize}
\end{lemma}

\begin{proof}
By Lemma~\ref{lem:bounded-spectrum-weak-corr}, Assumption~\ref{ass:C2} implies Assumption~\ref{ass:C4} for the correlation matrix $\mathbf R_{\Omega}=\Diag(\bOme)^{-1/2}\bOme\Diag(\bOme)^{-1/2}=(\rho_{jj'})_{j,j'\le p}$.

{Parts (i) and (ii) follow from Propositions~S.1 and~S.2, respectively, of}
\citet{wang2023JRSSB}.
\end{proof}

\begin{lemma}%
\label{lem:block-comparison-general}
Suppose Assumptions~\ref{ass:C1}--\ref{ass:C3} hold and
$\log p=O(n^\nu)$ with
\[
0<\nu<
\frac{\lambda_\gamma}
{\nu_{\mathrm{dep}}^{-1}+4c_{\mathrm{out}}-7},
\]
where $c_{\mathrm{out}}>0$ is fixed,
$\nu_{\mathrm{dep}}$ is defined before Theorem~\ref{null:Max}
in the main article, and
$\lambda_\gamma=\mathbb{I}(\gamma=0)+\lambda\mathbb{I}(\gamma=1/2)$.
When $\gamma=1/2$, assume additionally that $\lambda_n\asymp n^\lambda$
for a fixed $\lambda\in(0,1)$.
\begingroup
Let $L_n=\log(np)$, $\beta_n=L_n^{7/4}$, and
\endgroup
$b_n=\lceil L_n^{\kappa_b}\rceil$, with
\begingroup
\[
c_{\mathrm{cmp}}=\frac72+3c_{\mathrm{dep}},\qquad
d_{\mathrm{dep}}=\begin{cases}
1/\varkappa,&\text{under (AM) or (PD)},\\
0,&\text{under (MD)},
\end{cases}
\qquad
\kappa_b=2(c_{\mathrm{cmp}}+c_{\mathrm{out}})+d_{\mathrm{dep}}+1.
\]
\endgroup
Set
\[
s_n=
\begin{cases}
\lceil c_0L_n^{1/\varkappa}\rceil,&\text{under (AM) or (PD)},\\
m_0+1,&\text{under (MD)}.
\end{cases}
\]
For every fixed $A>0$, choose $c_0=c_0(A)$ sufficiently large
under (AM) or (PD).
Then, for every deterministic $\mathbf h=(\mathbf h_k:k\in\mathcal K_\gamma)$,
$\mathbf h_k\in\mathbb R^p$,
$\mathbf d=(d_{\varsigma,j,k}:\varsigma\in\{-1,1\},j\le p,
k\in\mathcal K_\gamma)$ and every $f\in C_b^3(\mathbb R^2)$,
\begingroup
\begin{align*}
&\Bigl|\E f\left\{
\operatorname{smax}_{\beta_n}\bigl(\omega_p^{-1}\mathcal Q_{\gamma,n}(\mathbf h)\bigr),
\operatorname{smax}_{\beta_n}\bigl(\mathcal C_{\gamma,n}^{\pm}(\mathbf d)\bigr)
\right\}\\
&\qquad-\E f\left\{
\operatorname{smax}_{\beta_n}\bigl(\omega_p^{-1}\mathcal Q_{\gamma,n}^{G}(\mathbf h)\bigr),
\operatorname{smax}_{\beta_n}\bigl((\mathcal C_{\gamma,n}^{G})^{\pm}(\mathbf d)\bigr)
\right\}\Bigr|\\
&\qquad\le
C\|f\|_{C^3}r_{n,\gamma}
\left\{
1+H_{\gamma,n}(\mathbf h)
\right\}^{3},
\end{align*}
\endgroup
where
\[
r_{n,\gamma}
\le
CL_n^{c_{\mathrm{cmp}}}\left\{
\left(\frac{s_n}{b_n}\right)^{1/2}
+\left(\frac{b_n}{\Lambda_{\gamma,n}}\right)^{1/2}
\right\}+(np)^{-A},
\quad
H_{\gamma,n}(\mathbf h)
=
\max_{k\in\mathcal K_\gamma}
v_k^{-\gamma}\frac{\|\mathbf h_k\|_2}{\sqrt p},
\]
with $\Lambda_{0,n}=n$, $\Lambda_{1/2,n}=\lambda_n$. Then, $L_n^{c_{\mathrm{out}}}r_{n,\gamma}\rightarrow0$.
\end{lemma}

\begin{proof}
Recall the quadratic arrays
\(\mathcal Q_{\gamma,n}(\mathbf h)\) and
\(\mathcal Q_{\gamma,n}^{G}(\mathbf h)\) from
\eqref{eq:def-boundary-quadratic-arrays-shifted}, and the shifted
max-coordinate arrays \(\mathcal C_{\gamma,n}^{\pm}(\mathbf d)\) and
\((\mathcal C_{\gamma,n}^{G})^{\pm}(\mathbf d)\) from
\eqref{eq:def-shifted-max-coordinate-arrays}.
Let
\[
\mathcal Z_{\gamma,n}(\mathbf h,\mathbf d)
=\bigl(
\mathcal Q_{\gamma,n}(\mathbf h),
\mathcal C_{\gamma,n}^{\pm}(\mathbf d)
\bigr),
\qquad
\mathcal Z_{\gamma,n}^{G}(\mathbf h,\mathbf d)
=\bigl(
\mathcal Q_{\gamma,n}^{G}(\mathbf h),
(\mathcal C_{\gamma,n}^{G})^{\pm}(\mathbf d)
\bigr).
\]
For finite vectors $\mathbf x$ and $\mathbf y$,
define
\[
\mathscr F_n(\mathbf x,\mathbf y)
=f\{\operatorname{smax}_{\beta_n}(\omega_p^{-1}\mathbf x),
     \operatorname{smax}_{\beta_n}(\mathbf y)\}.
\]
{Taking $N=|\mathcal K_\gamma|$ for the quadratic array,
$N=2p|\mathcal K_\gamma|$ for the signed coordinate array, and
$\beta=\beta_n$ in \eqref{eq:smax-derivative-bounds},
using \(\omega_p\asymp1\), and applying the chain rule give}
\begin{align*}
&\sum_a|\partial_a\mathscr F_n|\le C\|f\|_{C^1}, \qquad
\sum_{a,b}|\partial_{ab}\mathscr F_n|
\le C\|f\|_{C^2}(1+\beta_n),\\
&\sum_{a,b,c}|\partial_{abc}\mathscr F_n|
\le C\|f\|_{C^3}(1+\beta_n^2).
\end{align*}

We apply the four-step block construction directly to the shifted joint
arrays,
\[
\mathcal Z_{\gamma,n}(\mathbf h,\mathbf d)
\rightarrow
\mathcal Z_{\gamma,n}^{I}(\mathbf h,\mathbf d)
\rightarrow
\mathcal Z_{\gamma,n}^{I,\mathrm{ind}}(\mathbf h,\mathbf d)
\rightarrow
\mathcal Z_{\gamma,n}^{I,G,\mathrm{ind}}(\mathbf h,\mathbf d)
\rightarrow
\mathcal Z_{\gamma,n}^{G}(\mathbf h,\mathbf d).
\]

Step 1. Partition $\{1,\ldots,n\}$ into retained blocks $I_\ell$ of length $b_n$
separated by deleted blocks $J_\ell$ of length $s_n$, with the terminal
remainder included among the deleted indices.  Recall
\(a_{i,(\gamma,k)}\) and \(w_{\ell,(\gamma,k)}\) from
\eqref{eq:quadratic-only-weighted-coefficients}, with the present retained
blocks \(I_\ell\).  By
\eqref{eq:quadratic-only-gamma0-block-weights}--\eqref{eq:quadratic-only-deleted-weights},
\begin{align}
\sup_{k\in\mathcal K_\gamma}
\sum_{i\in\cup_\ell J_\ell}a_{i,(\gamma,k)}^2
\le C\{r_{n,\gamma}^{\mathrm{del}}\}^2,
\qquad
\sum_\ell\sup_{k\in\mathcal K_\gamma}
w_{\ell,(\gamma,k)}^3
\le C\left(\frac{b_n}{\Lambda_{\gamma,n}}\right)^{1/2}.
\label{eq:joint-block-weights}
\end{align}
For the deletion step, recall \(a_{i,(\gamma,k)}^I\) and
\(a_{i,(\gamma,k)}^{\mathrm{del}}\) from
\eqref{eq:quadratic-only-retained-deleted-coefficients}, now with
\(\mathcal J=\cup_\ell J_\ell\), and define
\begin{equation}
C_{\gamma,j}^{\epsilon,I}(k)
=\sum_{i=1}^n
a_{i,(\gamma,k)}^I\frac{\epsilon_{ij}}{\sigma_j},
\qquad
C_{\gamma,j}^{\epsilon,\mathrm{del}}(k)=C_{\gamma,j}^{\epsilon,0}(k)-C_{\gamma,j}^{\epsilon,I}(k)
=\sum_{i=1}^n
a_{i,(\gamma,k)}^{\mathrm{del}}\frac{\epsilon_{ij}}{\sigma_j}.
\label{eq:joint-deleted-coordinate-def}
\end{equation}
Let \(\mathcal C_{\gamma,n}^{I,\pm}(\mathbf d)=
\left\{
\varsigma C_{\gamma,j}^{\epsilon,I}(k)+d_{\varsigma,j,k}:
\varsigma\in\{-1,1\},\ j\le p,\ k\in\mathcal K_\gamma
\right\}\).

By the mean-value theorem and the first derivative bound for $\mathscr F_n$,
\begin{align}
&\left|
\E\mathscr F_n\bigl(
\mathcal Q_{\gamma,n}(\mathbf h),
\mathcal C_{\gamma,n}^{\pm}(\mathbf d)
\bigr)
-
\E\mathscr F_n\bigl(
\mathcal Q_{\gamma,n}^{I}(\mathbf h),
\mathcal C_{\gamma,n}^{I,\pm}(\mathbf d)
\bigr)
\right|\nonumber\\
&\quad\le
C\|f\|_{C^1}
\E\max_{k\in\mathcal K_\gamma}
\left|\frac{Q_{\gamma,n}(t_k;\mathbf h_k)
-Q_{\gamma,n}^{I}(t_k;\mathbf h_k)}{\omega_p}\right|
+C\|f\|_{C^1}
\E\max_{j\le p,\,k\in\mathcal K_\gamma}
|C_{\gamma,j}^{\epsilon,\mathrm{del}}(k)|.
\label{eq:joint-deletion-mvt-bound}
\end{align}
The first part in \eqref{eq:joint-deletion-mvt-bound} is supplied by
\eqref{eq:quadratic-only-deletion-shifted}:
\begingroup
\[
\E\max_{k\in\mathcal K_\gamma}
\left|\frac{Q_{\gamma,n}(t_k;\mathbf h_k)
-Q_{\gamma,n}^{I}(t_k;\mathbf h_k)}{\omega_p}\right|
\le
CL_n^{c_{\mathrm{dep}}}r_{n,\gamma}^{\mathrm{del}}
\{1+H_{\gamma,n}(\mathbf h)\}.
\]
\endgroup
It remains to bound the signed-coordinate deletion.  
{
For \(\gamma=0\),
combining \eqref{eq:quadratic-only-coefficient-norms} and
\eqref{eq:quadratic-only-gamma0-deleted-increment} with
Lemma~\ref{lem:vector-basic}(ii) and
\citet[Theorem~1]{moricz1982} gives the required bound directly.
For \(\gamma=1/2\), partition the trimmed logit range into
\[
\mathcal I_r^{\mathrm{logit}}
=\{k\in\mathcal K_{1/2}:r\le u_k<r+1\},
\qquad
\mathcal R_n
=\{r\in\mathbb Z:\mathcal I_r^{\mathrm{logit}}\ne\varnothing\}.
\]
Then
\(
|\mathcal R_n|
\le C\log(n/\lambda_n)
\le CL_n.
\)
If \(k_1<k_2\) belong to the same
\(\mathcal I_r^{\mathrm{logit}}\), then
\[
k_2-k_1
\le C\{k_1\wedge(n-k_2)\},
\qquad
\inf_{k_1<i\le k_2}\{i\wedge(n-i)\}
\ge c\{k_1\wedge(n-k_2)\}.
\]
Since the deleted blocks have length \(s_n\) and are separated by
retained blocks of length \(b_n\),
\[
|\mathcal J\cap(k_1,k_2]|
\le
C\left\{
\frac{s_n}{b_n}(k_2-k_1)+b_n
\right\}.
\]
Consequently,
\begin{align*}
\nu_{\mathcal J}(k_1,k_2)
&\le
\frac{C|\mathcal J\cap(k_1,k_2]|}
{k_1\wedge(n-k_2)}
\le C\left\{
\frac{s_n}{b_n}
+\frac{b_n}{\lambda_n}
\right\}
\le C\{r_{n,1/2}^{\mathrm{del}}\}^{2}.
\end{align*}
Lemma~\ref{lem:vector-basic}(ii) and
\eqref{eq:quadratic-only-deleted-increment} therefore give, for
\(k_1,k_2\in\mathcal I_r^{\mathrm{logit}}\),
\[
\left\|
C_{1/2,j}^{\epsilon,\mathrm{del}}(k_2)
-C_{1/2,j}^{\epsilon,\mathrm{del}}(k_1)
\right\|_q
\le
Cq^{c_{\mathrm{dep}}}
\left[
\{r_{n,1/2}^{\mathrm{del}}\}^{2}
|u_{k_2}-u_{k_1}|
+\nu_{\mathcal J}(k_1,k_2)
\right]^{1/2}.
\]
For each \(r\in\mathcal R_n\), let
\(k_r^-=\min\mathcal I_r^{\mathrm{logit}}\).  Applying
\citet[Theorem~1]{moricz1982} on this unit logit interval, with moment
order \(q\) and interval exponent \(q/2\), and using
\eqref{eq:quadratic-only-coefficient-norms} at \(k_r^-\), yields
\[
\left\|
\max_{k\in\mathcal I_r^{\mathrm{logit}}}
|C_{1/2,j}^{\epsilon,\mathrm{del}}(k)|
\right\|_q
\le
Cq^{c_{\mathrm{dep}}}r_{n,1/2}^{\mathrm{del}}.
\]
Thus, with \(q=2L_n\),
\begin{align*}
\left\|
\max_{k\in\mathcal K_{1/2}}
|C_{1/2,j}^{\epsilon,\mathrm{del}}(k)|
\right\|_q
&\le
\left\{
\sum_{r\in\mathcal R_n}
\left\|
\max_{k\in\mathcal I_r^{\mathrm{logit}}}
|C_{1/2,j}^{\epsilon,\mathrm{del}}(k)|
\right\|_q^q
\right\}^{1/q}\\
&\le
C|\mathcal R_n|^{1/q}
q^{c_{\mathrm{dep}}}r_{n,1/2}^{\mathrm{del}}
\lesssim
L_n^{c_{\mathrm{dep}}}r_{n,1/2}^{\mathrm{del}},
\end{align*}
where
\(
|\mathcal R_n|^{1/q}
\le
(CL_n)^{1/(2L_n)}
\le C.
\)
Combining the two cases gives
\[
\left\|
\max_{k\in\mathcal K_\gamma}
|C_{\gamma,j}^{\epsilon,\mathrm{del}}(k)|
\right\|_q
\le CL_n^{c_{\mathrm{dep}}}r_{n,\gamma}^{\mathrm{del}},
\qquad j\le p.
\]}
Therefore
\[
\E\max_{j\le p,\,k\in\mathcal K_\gamma}
|C_{\gamma,j}^{\epsilon,\mathrm{del}}(k)|
\le \left\{
\sum_{j=1}^p
\left\|
\max_{k\in\mathcal K_\gamma}
|C_{\gamma,j}^{\epsilon,\mathrm{del}}(k)|
\right\|_q^q
\right\}^{1/q}
{
\le Cp^{1/q}L_n^{c_{\mathrm{dep}}}r_{n,\gamma}^{\mathrm{del}}
\lesssim L_n^{c_{\mathrm{dep}}}r_{n,\gamma}^{\mathrm{del}}.}
\]
The last inequality uses $p^{1/q}\le e^{1/2}$ because $q=2\log(np)$.  Thus 
\begin{equation}
\eqref{eq:joint-deletion-mvt-bound}\lesssim {\|f\|_{C^3}L_n^{c_{\mathrm{dep}}}}r_{n,\gamma}^{\mathrm{del}}
\{1+H_{\gamma,n}(\mathbf h)\}^{3}.
\label{eq:joint-deletion-bound}
\end{equation}

Step 2. We only need to consider cases (PD) and (AM).
Define
\[
\Delta_{2,n}
=\left|
\E\mathscr F_n\bigl(
\mathcal Q_{\gamma,n}^{I}(\mathbf h),
\mathcal C_{\gamma,n}^{I,\pm}(\mathbf d)
\bigr)
-
\E\mathscr F_n\bigl(
\mathcal Q_{\gamma,n}^{I,\mathrm{ind}}(\mathbf h),
\mathcal C_{\gamma,n}^{I,\mathrm{ind},\pm}(\mathbf d)
\bigr)
\right|.
\]
We show that
\begingroup
\begin{equation}
\Delta_{2,n}
\le
C\|f\|_{C^3}\left\{
L_n^{c_{\mathrm{cmp}}}
\left(\frac{b_n}{\Lambda_{\gamma,n}}\right)^{1/2}
+(np)^{-A}
\right\}
\{1+H_{\gamma,n}(\mathbf h)\}^{3}.
\label{eq:joint-dependence-reduction}
\end{equation}
\endgroup

Under (PD), recall \(\bepsilon_i^{(s_n)}\) and
\(\mathbf R_i^{(s_n)}\) from
\eqref{eq:physical-dependent-truncation}, and \(\mathcal Q_{\gamma,n}^{I,(s_n)}(\mathbf h)\) from
\eqref{eq:quadratic-only-sdependent-array}.  Define $C_{\gamma,j}^{\epsilon,I,(s_n)}(k)$ and 
\(\mathcal C_{\gamma,n}^{I,(s_n),\pm}(\mathbf d)\) by replacing
\(\bepsilon_i\) with \(\bepsilon_i^{(s_n)}\) in $C_{\gamma,j}^{\epsilon,I}(k)$ and 
\(\mathcal C_{\gamma,n}^{I,\pm}(\mathbf d)\).
The mean-value theorem gives
\begin{align*}
\Delta_{2,n}
&\le
C\|f\|_{C^1}\E\max_{k\in\mathcal K_\gamma}
\left|
\frac{Q_{\gamma,n}^{I}(t_k;\mathbf h_k)
-Q_{\gamma,n}^{I,(s_n)}(t_k;\mathbf h_k)}{\omega_p}
\right|\\
&\quad+
C\|f\|_{C^1}\E\max_{j\le p,\,k\in\mathcal K_\gamma}
\left|
C_{\gamma,j}^{\epsilon,I}(k)
-C_{\gamma,j}^{\epsilon,I,(s_n)}(k)
\right|.
\end{align*}
The same calculation as \eqref{eq:quadratic-only-sdependent-residual}, with
the present \(s_n=\lceil c_0L_n^{1/\varkappa}\rceil\), gives
\[
\E\max_{k\in\mathcal K_\gamma}
\left|
\frac{Q_{\gamma,n}^{I}(t_k;\mathbf h_k)
-Q_{\gamma,n}^{I,(s_n)}(t_k;\mathbf h_k)}{\omega_p}
\right|
\le
C(np)^{-A}\{1+H_{\gamma,n}(\mathbf h)\}.
\]
The signed-coordinate residual follows from Lemma~\ref{lem:vector-basic}(ii), the retained
coefficient norm and increment bounds in
\eqref{eq:quadratic-only-coefficient-norms} and
\eqref{eq:quadratic-only-retained-increment}, and M\'oricz's maximal
inequality.  With $q=2L_n$,
{
\[
\E\max_{j\le p,\,k\in\mathcal K_\gamma}
\left|
C_{\gamma,j}^{\epsilon,I}(k)
-C_{\gamma,j}^{\epsilon,I,(s_n)}(k)
\right|
\le
CL_n^{c_{\mathrm{dep}}}e^{-cs_n^{\varkappa}}.
\]
Together with the quadratic residual bound, \eqref{eq:joint-dependence-reduction} holds.}

Under (AM), recall
\(\mathcal B_\ell=(\bepsilon_i:i\in I_\ell)\), \(\mathcal B_\ell^*=(\bepsilon_i^{(\ell,*)}:i\in I_\ell)\) in
\eqref{eq:AM-copied-retained-blocks}, and \(\mathcal B_\ell^*\overset d=\mathcal B_\ell\).
For the \(\ell\)th replacement and \(\diamond\in\{0,*\}\), recall the quadratic block-replacement coordinates
\(\mathbf q_{\ell,\gamma}^{(\diamond)}(t)\) from
\eqref{eq:quadratic-only-AM-hybrid-coordinates}.  Define the signed-coordinate block-replacement coordinates by
\[
Z_{\ell,\gamma,j,k}^{(\diamond)}
=\sum_{i\in I_\ell}a_{i,(\gamma,k)}
\frac{\epsilon_{ij}^{(\diamond)}}{\sigma_j},
\qquad
z_{\ell,\varsigma,j,k}^{(\diamond)}(t)
=
\varsigma\left\{
S_{\ell,\gamma,j,k}^{C}
+tZ_{\ell,\gamma,j,k}^{(\diamond)}
\right\}
+d_{\varsigma,j,k},
\]
where 
\begingroup
\[S_{\ell,\gamma,j,k}^{C}
=
\sum_{r<\ell}\sum_{i\in I_r}
a_{i,(\gamma,k)}
\frac{\epsilon_{ij}^{(r,*)}}{\sigma_j}
+
\sum_{r>\ell}\sum_{i\in I_r}
a_{i,(\gamma,k)}
\frac{\epsilon_{ij}}{\sigma_j},
\]
\[
\epsilon_{ij}^{(0)}=\epsilon_{ij},
\qquad
\epsilon_{ij}^{(*)}=\epsilon_{ij}^{(\ell,*)},
\quad i\in I_\ell .
\]
\endgroup
Define
\[
\Theta_\ell^{(\diamond)}(t)
=
\left(
\mathbf q_{\ell,\gamma}^{(\diamond)}(t),
\{z_{\ell,\varsigma,j,k}^{(\diamond)}(t):
  \varsigma\in\{-1,1\},\,j\le p,\,k\in\mathcal K_\gamma\}
\right),
\quad
G_\ell^{(\diamond)}(t)=\mathscr F_n(\Theta_\ell^{(\diamond)}(t)).
\]
Replacing the retained blocks one at a time gives
\[
\Delta_{2,n}
\le
\sum_{\ell=1}^L
\left|
\E G_\ell^{(0)}(1)-\E G_\ell^{(*)}(1)
\right|.
\]
The deterministic array \(\mathbf d\) only
translates \(z_{\ell,\varsigma,j,k}^{(\diamond)}(t)\), and hence is absent
from all \(t\)-derivatives.  By
\eqref{eq:quadratic-only-AM-q-moments}, the definition above, and
Lemma~\ref{lem:vector-basic}(ii), for every \(q\ge2\),
\begin{equation}
\begin{aligned}
\|\dot q_{\ell,\gamma,k}^{(\diamond)}(t)\|_4
&\le Cw_{\ell,(\gamma,k)}
\{1+H_{\gamma,n}(\mathbf h)\},
\qquad
\|\ddot q_{\ell,\gamma,k}^{(\diamond)}(t)\|_4
\le Cw_{\ell,(\gamma,k)}^2,\\
\dot z_{\ell,\varsigma,j,k}^{(\diamond)}(t)
&=\varsigma Z_{\ell,\gamma,j,k}^{(\diamond)},
\qquad
\ddot z_{\ell,\varsigma,j,k}^{(\diamond)}(t)=0,
\qquad
\|\dot z_{\ell,\varsigma,j,k}^{(\diamond)}(t)\|_q
\le Cq^{c_{\mathrm{dep}}}w_{\ell,(\gamma,k)} .
\end{aligned}
\label{eq:joint-AM-hybrid-derivative-bounds}
\end{equation}
Thus the chain rule gives
\begin{align*}
\dot G_\ell^{(\diamond)}(t)
&=
\sum_{k\in\mathcal K_\gamma}
\partial_{q_k}\mathscr F_n\dot q_{\ell,\gamma,k}^{(\diamond)}(t)+
\sum_{\varsigma,j,k}
\partial_{z_{\varsigma,j,k}}\mathscr F_n
\dot z_{\ell,\varsigma,j,k}^{(\diamond)}(t),
\\
\ddot G_\ell^{(\diamond)}(t)
&=
\sum_{k_1,k_2}
\partial_{q_{k_1}q_{k_2}}\mathscr F_n
\dot q_{\ell,\gamma,k_1}^{(\diamond)}(t)
\dot q_{\ell,\gamma,k_2}^{(\diamond)}(t)+
\sum_{k}
\partial_{q_k}\mathscr F_n
\ddot q_{\ell,\gamma,k}^{(\diamond)}(t)\\
&\quad+
2\sum_{k,\varsigma,j,k'}
\partial_{q_k z_{\varsigma,j,k'}}\mathscr F_n
\dot q_{\ell,\gamma,k}^{(\diamond)}(t)
\dot z_{\ell,\varsigma,j,k'}^{(\diamond)}(t)\\
&\quad+
\sum_{\varsigma_1,j_1,k_1}
\sum_{\varsigma_2,j_2,k_2}
\partial_{z_{\varsigma_1,j_1,k_1}z_{\varsigma_2,j_2,k_2}}\mathscr F_n
\dot z_{\ell,\varsigma_1,j_1,k_1}^{(\diamond)}(t)
\dot z_{\ell,\varsigma_2,j_2,k_2}^{(\diamond)}(t),
\end{align*}
where all derivatives of \(\mathscr F_n\) above are evaluated at \(\Theta_\ell^{(\diamond)}(t)\).  Removing all
\(z\)-coordinates from the two displays leaves precisely the
quadratic-only Taylor derivatives controlled by
\eqref{eq:quadratic-only-AM-q-moments}; the additional terms are linear or
mixed signed-coordinate terms.  Since
\(G_\ell^{(0)}(0)=G_\ell^{(*)}(0)\),
\begin{align*}
&\left|
\E G_\ell^{(0)}(1)-\E G_\ell^{(*)}(1)
\right|\\
&\quad\le
\left|\E\dot G_\ell^{(0)}(0)-\E\dot G_\ell^{(*)}(0)\right|
+\frac12
\left|\E\ddot G_\ell^{(0)}(0)-\E\ddot G_\ell^{(*)}(0)\right|
+\E|R_\ell^{(0)}|+\E|R_\ell^{(*)}|,
\end{align*}
with
\[
R_\ell^{(\diamond)}
=\frac12\int_0^1(1-t)^2
\frac{d^3}{dt^3}G_\ell^{(\diamond)}(t)\,dt.
\]
\begingroup
Let $W_\ell=\sup_{k\in\mathcal K_\gamma}w_{\ell,(\gamma,k)}$.
\endgroup
The covariance representations used for
\eqref{eq:fd-AM-first-Taylor-bound} and
\eqref{eq:fd-AM-second-Taylor-bound}, with the derivative sums controlled
by \eqref{eq:smax-derivative-bounds}, give for a fixed finite \(C_*\),
\begin{align*}
&\sum_{\ell=1}^L
\left\{
\left|\E\dot G_\ell^{(0)}(0)-\E\dot G_\ell^{(*)}(0)\right|
+
\left|\E\ddot G_\ell^{(0)}(0)-\E\ddot G_\ell^{(*)}(0)\right|
\right\}\\
&\qquad\le
C\|f\|_{C^3}
L_n^2(np)^{C_*}
\alpha_{\epsilon,p}(s_n)^{1/2}
\{1+H_{\gamma,n}(\mathbf h)\}^{3}\\
&\qquad\le
C\|f\|_{C^3}(np)^{-A}
\{1+H_{\gamma,n}(\mathbf h)\}^{3},
\end{align*}
after increasing \(c_0=c_0(A)\).  For the remainder, the chain rule gives
\[
\dddot G_\ell^{(\diamond)}(t)
=
\sum_{a,b,c}\partial_{abc}\mathscr F_n\,
\dot\theta_{\ell,a}^{(\diamond)}
\dot\theta_{\ell,b}^{(\diamond)}
\dot\theta_{\ell,c}^{(\diamond)}
+3\sum_{a,b}\partial_{ab}\mathscr F_n\,
\ddot\theta_{\ell,a}^{(\diamond)}
\dot\theta_{\ell,b}^{(\diamond)},
\qquad \dddot\theta_{\ell,a}^{(\diamond)}=0,
\]
where \(\theta_{\ell,a}^{(\diamond)}\) ranges over the coordinates of
\(\Theta_\ell^{(\diamond)}\).  {
\begingroup
For $q=\max\{4,\log(2np)\}$,
\[
\left\|\max_{\varsigma,j,k}|\dot z_{\ell,\varsigma,j,k}^{(\diamond)}|\right\|_4
\le (2np)^{1/q}\max_{\varsigma,j,k}
\|\dot z_{\ell,\varsigma,j,k}^{(\diamond)}\|_q
\le CL_n^{c_{\mathrm{dep}}}W_\ell.
\]
Together with \eqref{eq:quadratic-block-maximal-derivatives},
\begin{align*}
\|\max_a|\dot\theta_{\ell,a}^{(\diamond)}|\|_4
&\le CL_n^{c_{\mathrm{dep}}}W_\ell\{1+H_{\gamma,n}(\mathbf h)\},\\
\|\max_a|\ddot\theta_{\ell,a}^{(\diamond)}|\|_4
&\le CL_n^{1/4}W_\ell^2.
\end{align*}
Writing $D_1=\max_a|\dot\theta_{\ell,a}^{(\diamond)}|$ and
$D_2=\max_a|\ddot\theta_{\ell,a}^{(\diamond)}|$, the chain rule and H\"older's inequality give
\begin{align*}
|\dddot G_\ell^{(\diamond)}|
&\le C\|f\|_{C^3}\{(1+\beta_n^2)D_1^3+(1+\beta_n)D_1D_2\},\\
\E D_1^3&\le\|D_1\|_4^3,\qquad
\E(D_1D_2)\le\|D_1\|_4\|D_2\|_4,\\
\frac72+3c_{\mathrm{dep}}&=c_{\mathrm{cmp}},\qquad
\frac74+c_{\mathrm{dep}}+\frac14\le c_{\mathrm{cmp}}.
\end{align*}
Since $\beta_n^2=L_n^{7/2}$,
\[
\E|R_\ell^{(0)}|+\E|R_\ell^{(*)}|
\lesssim
\|f\|_{C^3}L_n^{c_{\mathrm{cmp}}}
\{1+H_{\gamma,n}(\mathbf h)\}^{3}
\sup_{k\in\mathcal K_\gamma}w_{\ell,(\gamma,k)}^3.
\]
\endgroup
}
Summing the last display over
\(\ell\) and using \eqref{eq:joint-block-weights} gives
\begingroup
\[
\sum_{\ell=1}^L
\{\E|R_\ell^{(0)}|+\E|R_\ell^{(*)}|\}
\le
C\|f\|_{C^3}L_n^{c_{\mathrm{cmp}}}
\left(\frac{b_n}{\Lambda_{\gamma,n}}\right)^{1/2}
\{1+H_{\gamma,n}(\mathbf h)\}^{3}.
\]
\endgroup
Combining the preceding displays proves the (AM) part of
\eqref{eq:joint-dependence-reduction}.

Step 3. We next replace the mutually independent retained blocks by
covariance-matched Gaussian blocks.  Define
\[
\Delta_{3,n}
=\left|
\E\mathscr F_n\bigl(
\mathcal Q_{\gamma,n}^{I,\mathrm{ind}}(\mathbf h),
\mathcal C_{\gamma,n}^{I,\mathrm{ind},\pm}(\mathbf d)
\bigr)
-
\E\mathscr F_n\bigl(
\mathcal Q_{\gamma,n}^{I,G,\mathrm{ind}}(\mathbf h),
\mathcal C_{\gamma,n}^{I,G,\mathrm{ind},\pm}(\mathbf d)
\bigr)
\right|.
\]
Recall \(\bepsilon_i^{\mathrm{ind}}\), \(\mathbf Z_{\ell,k}\), and
\(\mathbf Z_{\ell,k}^{G}\) from
\eqref{eq:quadratic-only-step3-block-contributions}, and recall the
covariance matching in
\eqref{eq:quadratic-only-step3-covariance-matching}.  Let
\[
\mathbf Z_{\ell,k}^{(\mathrm{ind})}=\mathbf Z_{\ell,k},
\qquad
\mathbf Z_{\ell,k}^{(G)}=\mathbf Z_{\ell,k}^{G}.
\]
For this step, instead of the functions \(G_\ell^{(\diamond)}\) used in
Step~2(AM), take
\[
G_{\ell,3}^{(\diamond)}(t)
=\mathscr F_n\{\Theta_{\ell,3}^{(\diamond)}(t)\},
\qquad
\diamond\in\{\mathrm{ind},G\},
\]
where \(\Theta_{\ell,3}^{(\diamond)}(t)\) has coordinates
\[
\Theta_{\ell,3}^{(\diamond)}(t)
=\left(
\{q_{\ell,3,\gamma,k}^{(\diamond)}(t):k\in\mathcal K_\gamma\},
\{z_{\ell,3,\varsigma,j,k}^{(\diamond)}(t):
  \varsigma\in\{-1,1\},\,j\le p,\,k\in\mathcal K_\gamma\}
\right),
\]
\begingroup
\[
\begin{aligned}
q_{\ell,3,\gamma,k}^{(\diamond)}(t)
&=
\frac{
\|\mathbf S_{\ell,\gamma,k}^{\mathrm{mix}}
  +t\mathbf Z_{\ell,k}^{(\diamond)}
  +\widetilde{\mathbf h}_{\gamma,k}\|_2^2
-m_{\ell,\gamma,k}(t)}
{\sqrt p},
\\
z_{\ell,3,\varsigma,j,k}^{(\diamond)}(t)
&=
\varsigma\{S_{\ell,\gamma,j,k}^{\mathrm{mix}}
  +tZ_{\ell,\gamma,j,k}^{(\diamond)}\}
+d_{\varsigma,j,k},
\end{aligned}
\]
\endgroup
with
\[
\mathbf S_{\ell,\gamma,k}^{\mathrm{mix}}
=
\sum_{r<\ell}\mathbf Z_{r,k}^{G}
+\sum_{r>\ell}\mathbf Z_{r,k}.
\]
Define the signed-coordinate counterparts by
\[
\begin{aligned}
Z_{r,\gamma,j,k}^{\mathrm{ind}}
&=\sum_{i\in I_r}a_{i,(\gamma,k)}
  \frac{\epsilon_{ij}^{\mathrm{ind}}}{\sigma_j},
&
Z_{r,\gamma,j,k}^{G}
&=\sum_{i\in I_r}a_{i,(\gamma,k)}
  \frac{\epsilon_{ij}^{G,r}}{\sigma_j},\\
Z_{r,\gamma,j,k}^{(\mathrm{ind})}
&=Z_{r,\gamma,j,k}^{\mathrm{ind}},
&
Z_{r,\gamma,j,k}^{(G)}
&=Z_{r,\gamma,j,k}^{G},
\\
S_{\ell,\gamma,j,k}^{\mathrm{mix}}
&=\sum_{r<\ell}Z_{r,\gamma,j,k}^{G}
  +\sum_{r>\ell}Z_{r,\gamma,j,k}^{\mathrm{ind}}.
\end{aligned}
\]
Here
\[
m_{\ell,\gamma,k}(t)
=\E\|\mathbf S_{\ell,\gamma,k}^{\mathrm{mix}}
  +t\mathbf Z_{\ell,k}\|_2^2
=\E\|\mathbf S_{\ell,\gamma,k}^{\mathrm{mix}}
  +t\mathbf Z_{\ell,k}^{G}\|_2^2.
\]
\begingroup
Conditional on the blocks other than $\ell$, the two replacement
vectors are centered and have equal covariance. Thus
\[
\E\{G_{\ell,3}^{\mathrm{ind}}(0)-G_{\ell,3}^{G}(0)\}=0,\qquad
\E\{(G_{\ell,3}^{\mathrm{ind}})^{(r)}(0)
-(G_{\ell,3}^{G})^{(r)}(0)\}=0,\quad r=1,2.
\]
Taylor's formula leaves the remainders
\endgroup
\[
R_{\ell,3}^{(\diamond)}
=\frac12\int_0^1(1-t)^2
\dddot G_{\ell,3}^{(\diamond)}(t)\,dt,
\qquad
\diamond\in\{\mathrm{ind},G\}.
\]
The same derivative calculation as in
\eqref{eq:joint-AM-hybrid-derivative-bounds}, together with the
Gaussian moment bound for \(\mathbf Z_{\ell,k}^{G}\), gives
\[
\|\dot\theta_{\ell,3,a}^{(\diamond)}(t)\|_4
\le
CW_\ell\{1+H_{\gamma,n}(\mathbf h)\},
\qquad
\|\ddot\theta_{\ell,3,a}^{(\diamond)}(t)\|_4
\le
CW_\ell^2,
\qquad
W_\ell=\sup_{k\in\mathcal K_\gamma}w_{\ell,(\gamma,k)}.
\]
\begingroup
Equations~\eqref{eq:quadratic-block-maximal-derivatives} and
\eqref{eq:joint-AM-hybrid-derivative-bounds} give
\[
\left\|\max_a|\dot\theta_{\ell,3,a}^{(\diamond)}|\right\|_4
\le CL_n^{c_{\mathrm{dep}}}W_\ell\{1+H_{\gamma,n}(\mathbf h)\},\qquad
\left\|\max_a|\ddot\theta_{\ell,3,a}^{(\diamond)}|\right\|_4
\le CL_n^{1/4}W_\ell^2.
\]
The preceding chain-rule bound therefore yields
\[
\E|R_{\ell,3}^{\mathrm{ind}}|+\E|R_{\ell,3}^{G}|
\le
C\|f\|_{C^3}L_n^{c_{\mathrm{cmp}}}
\{1+H_{\gamma,n}(\mathbf h)\}^{3}
\sup_{k\in\mathcal K_\gamma}w_{\ell,(\gamma,k)}^3.
\]
\endgroup
Summing over \(\ell\) and using \eqref{eq:joint-block-weights} gives
\begingroup
\begin{equation}
\Delta_{3,n}
\le \sum_{\ell=1}^L
\left|
\E G_{\ell,3}^{\mathrm{ind}}(1)-\E G_{\ell,3}^{G}(1)
\right| \le
C\|f\|_{C^3}L_n^{c_{\mathrm{cmp}}}
\left(\frac{b_n}{\Lambda_{\gamma,n}}\right)^{1/2}
\{1+H_{\gamma,n}(\mathbf h)\}^{3}.
\label{eq:joint-gaussian-block-replacement-bound}
\end{equation}
\endgroup

Step 4. We restore the covariance between distinct retained Gaussian blocks
and then add back the deleted Gaussian coordinates. 
Recall
\(\widetilde{\mathbf U}_{\gamma,k}^G\),
\(\mathbf U_{\gamma,k}^{I,G}\), and
\(\mathbf U_{\gamma,k}^G\) from
\eqref{eq:quadratic-only-Gaussian-retained-arrays}; the associated
quadratic arrays are formed through
\eqref{eq:def-boundary-quadratic-shifted}.  Define
\[
C_{\gamma,j}^{I,G}(k)
=\sum_{i=1}^n a_{i,(\gamma,k)}^I
\frac{\epsilon_{ij}^G}{\sigma_j},
\quad
\mathcal C_{\gamma,n}^{I,G,\pm}(\mathbf d)
=\{\varsigma C_{\gamma,j}^{I,G}(k)+d_{\varsigma,j,k}:
\varsigma\in\{-1,1\},\,j\le p,\,k\in\mathcal K_\gamma\},
\]
and let
\(
\mathcal Z_{\gamma,n}^{I,G,\mathrm{ind}}(\mathbf h,\mathbf d)
=\bigl(\mathcal Q_{\gamma,n}^{I,G,\mathrm{ind}}(\mathbf h),
\mathcal C_{\gamma,n}^{I,G,\mathrm{ind},\pm}(\mathbf d)\bigr),
\mathcal Z_{\gamma,n}^{I,G}(\mathbf h,\mathbf d)
=\bigl(\mathcal Q_{\gamma,n}^{I,G}(\mathbf h),
\mathcal C_{\gamma,n}^{I,G,\pm}(\mathbf d)\bigr).
\)
Then
\[
\left|
\E\mathscr F_n\bigl(
\mathcal Z_{\gamma,n}^{I,G,\mathrm{ind}}(\mathbf h,\mathbf d)
\bigr)
-
\E\mathscr F_n\bigl(
\mathcal Z_{\gamma,n}^{G}(\mathbf h,\mathbf d)
\bigr)
\right|\le
\Delta_{4,n}^{\mathrm{cov}}
+
\Delta_{4,n}^{\mathrm{del},G},
\]
where
\begingroup
\[
\begin{aligned}
\Delta_{4,n}^{\mathrm{cov}}
&=
\left|
\E\mathscr F_n\bigl(
\mathcal Z_{\gamma,n}^{I,G,\mathrm{ind}}(\mathbf h,\mathbf d)
\bigr)
-
\E\mathscr F_n\bigl(
\mathcal Z_{\gamma,n}^{I,G}(\mathbf h,\mathbf d)
\bigr)
\right|,\\
\Delta_{4,n}^{\mathrm{del},G}
&=
\left|
\E\mathscr F_n\bigl(
\mathcal Z_{\gamma,n}^{I,G}(\mathbf h,\mathbf d)
\bigr)
-
\E\mathscr F_n\bigl(
\mathcal Z_{\gamma,n}^{G}(\mathbf h,\mathbf d)
\bigr)
\right|.
\end{aligned}
\]
Let $\delta_n=\eta_{s_n}+e^{-cs_n^{\varkappa}}$ under (AM) or (PD),
and $\delta_n=0$ under (MD). For \(\boldsymbol\Delta_{\gamma;k,l}^{\mathrm{rest}}\) in
\endgroup
\eqref{eq:quadratic-only-GC-covariances},
\eqref{eq:quadratic-only-block-covariance} gives, uniformly over
\(k,l\in\mathcal K_\gamma\),
\(
\|\boldsymbol\Delta_{\gamma;k,l}^{\mathrm{rest}}\|_{\mathrm{op}}
\le C\delta_n,
\|\boldsymbol\Delta_{\gamma;k,l}^{\mathrm{rest}}\|_{\mathrm F}
\le C\sqrt p\,\delta_n .
\)
Under (MD), \(\boldsymbol\Delta_{\gamma;k,l}^{\mathrm{rest}}=\mathbf0\).

For the Gaussian interpolation, use \(\mathbf X_\theta\) and
\(q_{\theta,k}\) from
\eqref{eq:quadratic-only-GC-interpolation-coordinates}, with
\(\boldsymbol\Sigma_0,\boldsymbol\Sigma_1\) given by
\eqref{eq:quadratic-only-GC-covariances}.  Define
\[
z_{\theta,\varsigma,j,k}
=\varsigma\,\frac{\mathbf e_j^\top\mathbf X_{\theta,k}}{\sigma_j}
+d_{\varsigma,j,k},
\]
and
\[
\Phi_\gamma(\theta)
=
\E\mathscr F_n\bigl(
\{q_{\theta,k}:k\in\mathcal K_\gamma\},
\{z_{\theta,\varsigma,j,k}:
\varsigma\in\{-1,1\},j\le p,k\in\mathcal K_\gamma\}
\bigr).
\]
Then
\(
\Delta_{4,n}^{\mathrm{cov}}
\le \int_0^1|\Phi_\gamma'(\theta)|\,d\theta .
\)
The trace-centering cancellation in \eqref{eq:fd-GC-derivative} gives
\[
\Phi_\gamma'(\theta)
=T_{QQ}(\theta)+T_{QC}(\theta)+T_{CC}(\theta),
\]
where, with all derivatives of \(\mathscr F_n\) evaluated at the
interpolated coordinates,
\[
\begin{aligned}
T_{QQ}(\theta)
&=
\frac2p\sum_{k,l}
\E\!\left[
\partial_{q_kq_l}\mathscr F_n\,
(\mathbf X_{\theta,k}+\widetilde{\mathbf h}_{\gamma,k})^\top
\boldsymbol\Delta_{\gamma;k,l}^{\mathrm{rest}}
(\mathbf X_{\theta,l}+\widetilde{\mathbf h}_{\gamma,l})
\right],\\
T_{QC}(\theta)
&=
\frac2{\sqrt p}\sum_{k,l,\varsigma,j}
\E\!\left[
\partial_{q_kz_{\varsigma,j,l}}\mathscr F_n\,
\varsigma
(\mathbf X_{\theta,k}+\widetilde{\mathbf h}_{\gamma,k})^\top
\boldsymbol\Delta_{\gamma;k,l}^{\mathrm{rest}}
\frac{\mathbf e_j}{\sigma_j}
\right],\\
T_{CC}(\theta)
&=
\frac12\sum_{k,l}\sum_{\varsigma_1,j_1,\varsigma_2,j_2}
\E\!\left[
\partial_{z_{\varsigma_1,j_1,k}z_{\varsigma_2,j_2,l}}\mathscr F_n\,
\varsigma_1\varsigma_2
\frac{\mathbf e_{j_1}^\top
\boldsymbol\Delta_{\gamma;k,l}^{\mathrm{rest}}
\mathbf e_{j_2}}{\sigma_{j_1}\sigma_{j_2}}
\right].
\end{aligned}
\]
\begingroup
Let $H_{\gamma,n}=H_{\gamma,n}(\mathbf h)$. The maximum-contraction bound
\eqref{eq:quadratic-only-GC-contraction}, with $\ell_n\le L_n$,
and the operator-norm discrepancy bound give
\begin{align*}
\frac1p\E\max_{k,l}|(\mathbf X_{\theta,k}
+\widetilde{\mathbf h}_{\gamma,k})^\top
\boldsymbol\Delta_{\gamma;k,l}^{\mathrm{rest}}
(\mathbf X_{\theta,l}+\widetilde{\mathbf h}_{\gamma,l})|
&\le C\delta_n L_n(1+H_{\gamma,n})^2,\\
\frac1{\sqrt p}\E\max_{k,l,j}|(\mathbf X_{\theta,k}
+\widetilde{\mathbf h}_{\gamma,k})^\top
\boldsymbol\Delta_{\gamma;k,l}^{\mathrm{rest}}\mathbf e_j/\sigma_j|
&\le C\delta_n L_n^{1/2}(1+H_{\gamma,n}),\\
\max_{k,l,j_1,j_2}|\mathbf e_{j_1}^\top
\boldsymbol\Delta_{\gamma;k,l}^{\mathrm{rest}}\mathbf e_{j_2}/
(\sigma_{j_1}\sigma_{j_2})|&\le C\delta_n.
\end{align*}
For the second bound, Cauchy--Schwarz gives
\[
\frac1{\sqrt p}\E\max_{k,l,j}
|(\mathbf X_{\theta,k}+\widetilde{\mathbf h}_{\gamma,k})^\top
\boldsymbol\Delta_{\gamma;k,l}^{\mathrm{rest}}\mathbf e_j/\sigma_j|
\le C\delta_n\left\{\frac1p\E\max_k
\|\mathbf X_{\theta,k}+\widetilde{\mathbf h}_{\gamma,k}\|_2^2\right\}^{1/2}.
\]
The derivative sums in \eqref{eq:smax-derivative-bounds} imply
\[
\begin{aligned}
|T_{QQ}(\theta)|+|T_{QC}(\theta)|+|T_{CC}(\theta)|
&\le C\|f\|_{C^3}(1+\beta_n)L_n\delta_n(1+H_{\gamma,n})^3\\
&\le C\|f\|_{C^3}L_n^3\delta_n(1+H_{\gamma,n})^3.
\end{aligned}
\]
Integrating gives
\begin{equation}
\Delta_{4,n}^{\mathrm{cov}}
\le C\|f\|_{C^3}L_n^3\delta_n
\{1+H_{\gamma,n}(\mathbf h)\}^3.
\label{eq:joint-covariance-restoration-bound}
\end{equation}
\endgroup

For the deleted Gaussian coordinates, the argument leading to
\eqref{eq:joint-deletion-bound}, with \(\bepsilon_i\) replaced by
\(\bepsilon_i^G\), gives
\begin{equation}
\Delta_{4,n}^{\mathrm{del},G}
\le
C\|f\|_{C^3}{L_n^{1/2}}
r_{n,\gamma}^{\mathrm{del}}
\{1+H_{\gamma,n}(\mathbf h)\}^{3}.
\label{eq:joint-gaussian-deletion-bound}
\end{equation}
Under (AM) or (PD), since \(\eta_{s_n}\le Ce^{-cs_n^{\varkappa}}\), increasing
\(c_0=c_0(A)\) makes
\begingroup
\(
L_n^3(\eta_{s_n}+e^{-cs_n^{\varkappa}})
\lesssim (np)^{-A}.
\)
\endgroup
Under (MD), the retained blocks are independent because $s_n>m_0$,
and the dependence-reduction and covariance-restoration errors
are zero; in particular, no exponential-tail term is needed.
\begingroup
{
Combining \eqref{eq:joint-deletion-bound},
\eqref{eq:joint-dependence-reduction},
\eqref{eq:joint-gaussian-block-replacement-bound},
\eqref{eq:joint-covariance-restoration-bound}, and
\eqref{eq:joint-gaussian-deletion-bound} gives
\begin{align*}
&\left|
\E\mathscr F_n\bigl(
\mathcal Q_{\gamma,n}(\mathbf h),
\mathcal C_{\gamma,n}^{\pm}(\mathbf d)
\bigr)
-
\E\mathscr F_n\bigl(
\mathcal Q_{\gamma,n}^{G}(\mathbf h),
(\mathcal C_{\gamma,n}^{G})^{\pm}(\mathbf d)
\bigr)
\right|\\
&\quad\le
C\|f\|_{C^3}
\Bigg[
L_n^{c_{\mathrm{dep}}}r_{n,\gamma}^{\mathrm{del}}
+L_n^{c_{\mathrm{cmp}}}
\left(\frac{b_n}{\Lambda_{\gamma,n}}\right)^{1/2}
+L_n^{c_{\mathrm{cmp}}}
\left(\frac{b_n}{\Lambda_{\gamma,n}}\right)^{1/2}\\
&\hspace{42mm}
+L_n^{1/2}r_{n,\gamma}^{\mathrm{del}}
+(np)^{-A}
\Bigg]
\{1+H_{\gamma,n}(\mathbf h)\}^{3}\\
&\quad\le
C\|f\|_{C^3}
\left[
L_n^{c_{\mathrm{cmp}}}
\left\{
\left(\frac{s_n}{b_n}\right)^{1/2}
+\left(\frac{b_n}{\Lambda_{\gamma,n}}\right)^{1/2}
\right\}
+(np)^{-A}
\right]
\{1+H_{\gamma,n}(\mathbf h)\}^{3}.
\end{align*}
All logarithmic factors are bounded by
$L_n^{c_{\mathrm{cmp}}}$ because
$c_{\mathrm{cmp}}=7/2+3c_{\mathrm{dep}}$ and
$c_{\mathrm{dep}}\ge1/2$.}
\endgroup

It remains to prove $L_n^{c_{\mathrm{out}}}r_{n,\gamma}\to0$.
The branchwise choices above give
\[
L_n^{c_{\mathrm{out}}+c_{\mathrm{cmp}}}
(s_n/b_n)^{1/2}\le CL_n^{-1/2},
\qquad
2(c_{\mathrm{out}}+c_{\mathrm{cmp}})+\kappa_b
=\nu_{\mathrm{dep}}^{-1}+4c_{\mathrm{out}}-7.
\]
Since $L_n=O(n^\nu)$ and
$\Lambda_{\gamma,n}\asymp n^{\lambda_\gamma}$, it follows that
\begin{align}
L_n^{c_{\mathrm{out}}}r_{n,\gamma}
&\le
CL_n^{-1/2}
+CL_n^{(\nu_{\mathrm{dep}}^{-1}+4c_{\mathrm{out}}-7)/2}
\Lambda_{\gamma,n}^{-1/2}
+L_n^{c_{\mathrm{out}}}(np)^{-A} \notag\\
&=
o(1)
+O\!\left(
n^{\{\nu(\nu_{\mathrm{dep}}^{-1}+4c_{\mathrm{out}}-7)
-\lambda_\gamma\}/2}
\right)+o(1)=o(1).
\label{eq:joint-rate-vanishing}
\end{align}

\end{proof}

{
\begin{lemma}%
\label{lem:joint-gaussian-cov-restoration}
Suppose Assumptions~\ref{ass:C1}--\ref{ass:C2} hold. For
\begingroup
$\gamma\in\{0,1/2\}$, let $L_n=\log(np)$, $
\beta_n=L_n^{7/4}$. For every deterministic $\mathbf h=(\mathbf h_k:k\in\mathcal K_\gamma)$,
\endgroup
$\mathbf h_k\in\mathbb R^p$, and  $\mathbf d
=
(d_{\varsigma,j,k}:
\varsigma\in\{-1,1\},\,
j\leq p,\,
k\in\mathcal K_\gamma)$ and every $f\in C_b^3(\mathbb R^2)$,
\begingroup
\begin{align*}
&\Bigg|
\E f\left[
\operatorname{smax}_{\beta_n}
\bigl(\omega_p^{-1}\mathcal Q_{\gamma,n}^{G}(\mathbf h)\bigr),
\operatorname{smax}_{\beta_n}
\bigl((\mathcal C_{\gamma,n}^{G})^{\pm}(\mathbf d)\bigr)
\right]\\
&\qquad-
\E f\left[
\operatorname{smax}_{\beta_n}
\bigl(\mathcal Q_{\gamma,n}^{B}(\mathbf h)\bigr),
\operatorname{smax}_{\beta_n}
\bigl((\mathcal C_{\gamma,n}^{B})^{\pm}(\mathbf d)\bigr)
\right]
\Bigg|\\
&\qquad\leq
C\|f\|_{C^3}
L_n^{3}
\Lambda_{\gamma,n}^{-1}
\left\{
1+H_{\gamma,n}(\mathbf h)
\right\}^{3},
\end{align*}
\endgroup
where $\Lambda_{0,n}=n$ and $\Lambda_{1/2,n}=\lambda_n$, $H_{\gamma,n}(\mathbf h)
=
\max_{k\in\mathcal K_\gamma}v_k^{-\gamma}\|\mathbf h_k\|_2/\sqrt p$. Moreover, for every $g\in C_b^2(\mathbb R)$,
\begin{align*}
&\Bigg|
\E g\left[
\operatorname{smax}_{\beta_n}
\bigl((\mathcal C_{\gamma,n}^{G})^\pm(\mathbf d)\bigr)
\right]
-
\E g\left[
\operatorname{smax}_{\beta_n}
\bigl((\mathcal C_{\gamma,n}^{B})^\pm(\mathbf d)\bigr)
\right]
\Bigg|\leq
C\|g\|_{C^2}
L_n^2\Lambda_{\gamma,n}^{-1}.
\end{align*}
\end{lemma}

\begin{proof}
For $k\in\mathcal K_\gamma$, let $t_k=k/n$, $v_k=t_k(1-t_k)$, $\widetilde{\mathbf h}_{\gamma,k}=v_k^{-\gamma}\mathbf h_k$ and $\tau_p=\sqrt p\,\omega_p=\{2\operatorname{tr}(\boldsymbol{\Omega}^2)\}^{1/2}$. Also let $\sigma_j^2=(\boldsymbol{\Omega})_{jj}$ and let $\mathbf e_j$ be the $j$th canonical basis vector of
$\mathbb R^p$. By Assumption~\ref{ass:C2}, $\tau_p^2\asymp p$, $0<c\leq \sigma_j\leq C$, $ 1\leq j\leq p$. Define $\mathbf X_k^G=v_k^{-\gamma}\mathbf U_n^G(t_k)$, $\mathbf X_k^B=v_k^{-\gamma}\boldsymbol{\Omega}^{1/2}\mathbf B(t_k)$, and stack these vectors as $\mathbf X^G=(\mathbf X_k^G:k\in\mathcal K_\gamma)$, $\mathbf X^B=(\mathbf X_k^B:k\in\mathcal K_\gamma)$. Let $\boldsymbol{\Sigma}^G$ and $\boldsymbol{\Sigma}^B$ be their covariance matrices, and let $\boldsymbol{\Sigma}_{k\ell}^G$ and $\boldsymbol{\Sigma}_{k\ell}^B$ denote the corresponding $p\times p$ blocks. Let $\boldsymbol{\Delta}_{k\ell}=\boldsymbol{\Sigma}_{k\ell}^B-\boldsymbol{\Sigma}_{k\ell}^G$.

Lemma~\ref{lem:cusum-cov-general} gives $\max_{k,\ell\in\mathcal K_\gamma}
\|\boldsymbol{\Delta}_{k\ell}\|_{\mathrm{op}}
\leq
C\Lambda_{\gamma,n}^{-1}$. Indeed, the right-hand side is $Cn^{-1}$ for $\gamma=0$. For
$\gamma=1/2$, since $v_k\geq c\lambda_n/n$, $k\in\mathcal K_{1/2}$, we have $C/\{n(v_kv_\ell)^{1/2}\}\leq C\lambda_n^{-1}$. Moreover,
\[
\max_{\nu\in\{G,B\}}
\max_{k,\ell\in\mathcal K_\gamma}
\|\boldsymbol{\Sigma}_{k\ell}^{\nu}\|_{\mathrm{op}}
\leq C.
\]
For the Brownian-bridge array, this follows from
\[
|K_B(s,t)|
\leq
\{s(1-s)t(1-t)\}^{1/2},
\]
whereas the finite-sample bound follows from the preceding covariance
approximation.

For $0\leq\theta\leq1$, define
$\boldsymbol{\Sigma}_\theta
=
(1-\theta)\boldsymbol{\Sigma}^G
+
\theta\boldsymbol{\Sigma}^B$, and let $\mathbf X_\theta
=
(\mathbf X_{\theta,k}:k\in\mathcal K_\gamma)
\sim
N(0,\boldsymbol{\Sigma}_\theta)$. Let $\boldsymbol{\Sigma}_{\theta,k\ell}$ denote the $(k,\ell)$ block
of $\boldsymbol{\Sigma}_\theta$. By convexity, $\sup_{0\leq\theta\leq1}
\max_{k,\ell\in\mathcal K_\gamma}
\|\boldsymbol{\Sigma}_{\theta,k\ell}\|_{\mathrm{op}}
\leq C$. If an interpolating covariance matrix is singular, the calculation
below is first applied after adding $\varepsilon I$ and then letting
$\varepsilon\downarrow0$.

For a stacked vector $\mathbf x=(\mathbf x_k:k\in\mathcal K_\gamma)$, define
\[
q_{\theta,k}(\mathbf x)
=
\frac{
\|\mathbf x_k+\widetilde{\mathbf h}_{\gamma,k}\|_2^2
-
\operatorname{tr}(\boldsymbol{\Sigma}_{\theta,kk})
}{
\tau_p
},
\qquad
k\in\mathcal K_\gamma.
\]
Let $\mathcal I_\gamma
=
\{-1,1\}
\times
\{1,\ldots,p\}
\times
\mathcal K_\gamma$. For $a=(\varsigma,j,k)\in\mathcal I_\gamma$, define $k(a)=k$, $
\mathbf b_a
=
\varsigma\frac{\mathbf e_j}{\sigma_j}$, $
d_a=d_{\varsigma,j,k}$, and set $z_a(\mathbf x)
=
\mathbf b_a^\top\mathbf x_{k(a)}+d_a$. For finite vectors $\mathbf q=(q_k)_k$ and $
\mathbf z=(z_a)_a$, define $\mathcal F_n(\mathbf q,\mathbf z)
=
f\left\{
\operatorname{smax}_{\beta_n}(\mathbf q),
\operatorname{smax}_{\beta_n}(\mathbf z)
\right\}$. Set $
\Psi(\theta)
=
\E\,\mathcal F_n
\left(
\bigl(q_{\theta,k}(\mathbf X_\theta):
k\in\mathcal K_\gamma\bigr),
\bigl(z_a(\mathbf X_\theta):
a\in\mathcal I_\gamma\bigr)
\right)$. By the definitions of the arrays in the lemma, $\Psi(0)$ and
$\Psi(1)$ are the two expectations in the first assertion.

Let $\pi_k$ and $\pi_{k\ell}$ denote the first and second derivatives
of the smooth maximum with respect to the quadratic coordinates, and
let $\rho_a$ and $\rho_{ab}$ denote the corresponding derivatives
with respect to the linear coordinates. The derivative bounds for the
log-sum-exp map give
\[
\sum_k|\pi_k|=1,
\qquad
\sum_{k,\ell}|\pi_{k\ell}|
\leq
2\beta_n,
\sum_a|\rho_a|=1,
\qquad
\sum_{a,b}|\rho_{ab}|
\leq
2\beta_n.
\]
The chain rule yields
\[
\partial_{q_kq_\ell}\mathcal F_n
=
f_{11}\pi_k\pi_\ell+f_1\pi_{k\ell},
\qquad
\partial_{q_kz_a}\mathcal F_n
=
f_{12}\pi_k\rho_a,\quad\partial_{z_az_b}\mathcal F_n
=
f_{22}\rho_a\rho_b+f_2\rho_{ab},
\]
where the derivatives of $f$ are evaluated at the two smooth maxima.
Consequently, pointwise,
\[
\sum_{k,\ell}
|\partial_{q_kq_\ell}\mathcal F_n|
\leq
C\|f\|_{C^2}(1+\beta_n),
\sum_{k,a}
|\partial_{q_kz_a}\mathcal F_n|
\leq
C\|f\|_{C^2},
\sum_{a,b}
|\partial_{z_az_b}\mathcal F_n|
\leq
C\|f\|_{C^2}(1+\beta_n).
\]

We next differentiate $\Psi$. The Gaussian covariance-interpolation
identity, together with
\[
\partial_\theta q_{\theta,k}(\mathbf x)
=
-
\frac{
\operatorname{tr}(\boldsymbol{\Delta}_{kk})
}{
\tau_p
},
\qquad
\nabla_{\mathbf x_k}q_{\theta,k}(\mathbf x)
=
\frac{
2(\mathbf x_k+\widetilde{\mathbf h}_{\gamma,k})
}{
\tau_p
},
\]
and
\[
\nabla_{\mathbf x_k}^2q_{\theta,k}(\mathbf x)
=
\frac{2I_p}{\tau_p},
\qquad
\nabla_{\mathbf x_{k(a)}}z_a(\mathbf x)
=
\mathbf b_a,
\qquad
\nabla^2z_a(\mathbf x)=0,
\]
\begingroup
gives the exact cancellation
\[
\partial_\theta q_{\theta,k}
+\frac12\Tr\{\boldsymbol\Delta_{kk}\nabla_{\mathbf x_k}^2q_{\theta,k}\}
=-\frac{\Tr(\boldsymbol\Delta_{kk})}{\tau_p}
+\frac{\Tr(\boldsymbol\Delta_{kk})}{\tau_p}=0.
\]
\endgroup

Define $\mathbf Y_{\theta,k}
=
\mathbf X_{\theta,k}
+
\widetilde{\mathbf h}_{\gamma,k}$.
After the preceding cancellation,
\[
\begin{aligned}
|\Psi'(\theta)|
\leq{}&
\frac{C}{\tau_p^2}
\E\sum_{k,\ell}
|\partial_{q_kq_\ell}\mathcal F_n|
\left|
\mathbf Y_{\theta,k}^{\top}
\boldsymbol{\Delta}_{k\ell}
\mathbf Y_{\theta,\ell}
\right|+
\frac{C}{\tau_p}
\E\sum_{k,a}
|\partial_{q_kz_a}\mathcal F_n|
\left|
\mathbf Y_{\theta,k}^{\top}
\boldsymbol{\Delta}_{k,k(a)}
\mathbf b_a
\right|
\\
&+
C\E\sum_{a,b}
|\partial_{z_az_b}\mathcal F_n|
\left|
\mathbf b_a^\top
\boldsymbol{\Delta}_{k(a),k(b)}
\mathbf b_b
\right|.
\end{aligned}
\]

\begingroup
Let $M_\theta=\max_k\|\mathbf Y_{\theta,k}\|_2$.
The covariance bound and Gaussian concentration yield, uniformly in $\theta$,
\begin{align*}
\Pr\{\max_k\|\mathbf X_{\theta,k}\|_2>C\sqrt p+u\}
&\le \min(1,ne^{-cu^2}),\\
\E\max_k\|\mathbf X_{\theta,k}\|_2^2
&\le Cp+C\int_0^\infty u\min(1,ne^{-cu^2})\,du
\le C(p+\log n),\\
\max_k\|\widetilde{\mathbf h}_{\gamma,k}\|_2
&\le\sqrt p\,H_{\gamma,n}(\mathbf h).
\end{align*}
Hence
\endgroup
\[
\frac{\E M_\theta^2}{p}
\leq
C\left\{
1+H_{\gamma,n}(\mathbf h)^2
+\frac{\log n}{p}
\right\}
\leq
CL_n
\{1+H_{\gamma,n}(\mathbf h)\}^2.
\]
Moreover, $\sup_{a\in\mathcal I_\gamma}
\|\mathbf b_a\|_2
\leq C$.

Using $\tau_p^2\asymp p$, the covariance-discrepancy bound, and the
preceding derivative-sum bounds, we obtain
\[
\begin{aligned}
|\Psi'(\theta)|
\leq{}&
C\|f\|_{C^2}
\Lambda_{\gamma,n}^{-1}
\Big[
(1+\beta_n)L_n
\{1+H_{\gamma,n}(\mathbf h)\}^2+
L_n^{1/2}
\{1+H_{\gamma,n}(\mathbf h)\}
+
(1+\beta_n)
\Big].
\end{aligned}
\]
\begingroup
Since $\beta_n=L_n^{7/4}$, for all sufficiently large $n,p$,
\endgroup
\[
\sup_{0\leq\theta\leq1}
|\Psi'(\theta)|
\leq
C\|f\|_{C^3}
L_n^3
\Lambda_{\gamma,n}^{-1}
\{1+H_{\gamma,n}(\mathbf h)\}^3.
\]
Integrating over $\theta\in[0,1]$ proves the first assertion.

For the max-only assertion, define $\mathcal G_n(\mathbf z)
=
g\{\operatorname{smax}_{\beta_n}(\mathbf z)\}$ and $\Psi_M(\theta)
=
\E\,\mathcal G_n
\left(
z_a(\mathbf X_\theta):
a\in\mathcal I_\gamma
\right)$. The chain rule gives $\partial_{z_az_b}\mathcal G_n
=
g''\rho_a\rho_b+g'\rho_{ab}$, and hence $\sum_{a,b}
|\partial_{z_az_b}\mathcal G_n|
\leq
C\|g\|_{C^2}(1+\beta_n)$. Since $\max_{a,b}
\left|
\mathbf b_a^\top
\boldsymbol{\Delta}_{k(a),k(b)}
\mathbf b_b
\right|
\leq
C\Lambda_{\gamma,n}^{-1}$, the Gaussian covariance-interpolation identity gives
\[
|\Psi_M'(\theta)|
\leq
C\|g\|_{C^2}
(1+\beta_n)
\Lambda_{\gamma,n}^{-1}
\leq
C\|g\|_{C^2}
L_n^2
\Lambda_{\gamma,n}^{-1}.
\]
Integrating over $\theta\in[0,1]$ proves the max-only assertion.
\end{proof}}

\begingroup

\endgroup

\begin{lemma}%
\label{lem:max-transfer-general}
Suppose Assumptions~\ref{ass:C1}--\ref{ass:C3} hold and
$\log p=O(n^{\nu})$ with %
$0<\nu<\lambda_\gamma\nu_{\mathrm{dep}}$. Let
\(
M_{n,p}^{\epsilon,\gamma}=\max\mathcal C_{\gamma,n}^{\pm}.
\)
Then
\begin{itemize}
\item [(i)] As $n,p\rightarrow \infty$,
\[
2(M_{n,p}^{\epsilon,0})^2-\log(2p)
\xrightarrow{d} Z_{\mathrm{Gu}}.
\]
\item [(ii)]
If $\lambda_n\asymp n^\lambda$, as $n,p\rightarrow \infty$,
\[
A_{\mathrm{DE}}(p\log h_n)M_{n,p}^{\epsilon,1/2}
-D_{\mathrm{DE}}(p\log h_n)
\xrightarrow{d} Z_{\mathrm{Gu}}.
\]
\end{itemize}
\end{lemma}

\begin{proof}
\begingroup
Recall $L_n=\log(np)$, $\beta_n=L_n^{7/4}$, and define
\endgroup
\[
M_{n,p}^{G,\gamma}=\max(\mathcal C_{\gamma,n}^{G})^\pm,
\qquad
M_{n,p}^{B,\gamma}=\max(\mathcal C_{\gamma,n}^{B})^\pm.
\]
For $\gamma=0$ and $\gamma=1/2$, respectively, define
\begin{align}
\mathsf N_{n,0}^{\mathrm M}(x)=2x^2-\log(2p),
\qquad
\mathsf N_{n,1/2}^{\mathrm M}(x)
=A_{\mathrm{DE}}(p\log h_n)x-D_{\mathrm{DE}}(p\log h_n).
\label{eq:max-normalizations} 
\end{align}
Fix $z\in\mathbb R$.  For \(a\in\mathbb R\), let
\(\psi_{a,\varepsilon}\in C^3\) satisfying
\begingroup
\[
\ind{x\le a}\le \psi_{a,\varepsilon}(x)
\le\ind{x\le a+\varepsilon},
\qquad
\|\psi_{a,\varepsilon}^{(r)}\|_\infty
\le C\varepsilon^{-r},\qquad r=1,2,3.
\]
\endgroup
Take $\varepsilon_n=(\log L_n)^{-1}$. Let
\(\mathfrak t_{n,C_{\mathrm{clip}}}\in C^3(\mathbb R)\) satisfy
\begin{equation}
\mathfrak t_{n,C_{\mathrm{clip}}}(y)=y
\quad\text{for } |y|\le C_{\mathrm{clip}}\sqrt{L_n},
\qquad
\|\mathfrak t_{n,C_{\mathrm{clip}}}\|_\infty
\le 2C_{\mathrm{clip}}\sqrt{L_n},
\qquad
\max_{1\le r\le3}
\|\mathfrak t_{n,C_{\mathrm{clip}}}^{(r)}\|_\infty
\le C.
\label{eq:max-transfer-clipping-map}
\end{equation}

For \(a\in\mathbb R\), define
\[
g_{n,\gamma}(y)
=\mathsf N_{n,\gamma}^{\mathrm M}
\{\mathfrak t_{n,C_{\mathrm{clip}}}(y)\},
\qquad
f_{n,a}(x,y)
=
\psi_{a,\varepsilon_n}\!
\{g_{n,\gamma}(y)\}.
\]
By the definitions of the normalizing maps,
\[
\begin{aligned}
&(\mathsf N_{n,0}^{\mathrm M})'(u)=4u,
\qquad
(\mathsf N_{n,0}^{\mathrm M})''(u)=4,
\qquad
(\mathsf N_{n,0}^{\mathrm M})'''(u)=0,\\
&(\mathsf N_{n,1/2}^{\mathrm M})'(u)
=A_{\mathrm{DE}}(p\log h_n),
\qquad
(\mathsf N_{n,1/2}^{\mathrm M})''(u)
=(\mathsf N_{n,1/2}^{\mathrm M})'''(u)=0.
\end{aligned}
\]
\begingroup
Moreover, \(A_{\mathrm{DE}}(p\log h_n)
\le C\sqrt{L_n}\).
\endgroup
Together with the bounds on
\(\mathfrak t_{n,C_{\mathrm{clip}}}\), this gives, uniformly over
\(\gamma\in\{0,1/2\}\),
\(
\max_{1\le r\le3}\|g_{n,\gamma}^{(r)}\|_\infty
\le CL_n^{1/2}.
\)
\begingroup

\endgroup
By the chain rule, 
\[
\sup_{a\in\mathbb R}\|f_{n,a}\|_{C^3}
\le C\left\{
1+\varepsilon_n^{-1}L_n^{1/2}
+\varepsilon_n^{-2}L_n
+\varepsilon_n^{-3}L_n^{3/2}
\right\}.
\]
Substituting $\varepsilon_n=(\log L_n)^{-1}$ gives
\[
\sup_{a\in\mathbb R}\|f_{n,a}\|_{C^3}
\le CL_n^{3/2}(\log L_n)^3
=o\!\left(L_n^{3/2+\delta}\right).
\]

Fix $0<\delta<1/4$ and take
{$c_{\mathrm{out}}=3/2+\delta$}.
Since
\[
\nu_{\mathrm{dep}}^{-1}+4c_{\mathrm{out}}-7
=\nu_{\mathrm{dep}}^{-1}-1+4\delta
<\nu_{\mathrm{dep}}^{-1},
\]
the condition $0<\nu<\lambda_\gamma\nu_{\mathrm{dep}}$ permits
Lemma~\ref{lem:block-comparison-general} to be applied, giving
$L_n^{c_{\mathrm{out}}}r_{n,\gamma}\to0$.
Furthermore, since $c_{\mathrm{br}}=2$,
\(\Lambda_{\gamma,n}\asymp n^{\lambda_\gamma}\) and
\(L_n=O(n^{\nu})\), 
{
\(
L_n^{c_{\mathrm{out}}+c_{\mathrm{br}}}
\Lambda_{\gamma,n}^{-1}=O\left\{n^{\nu(7/2+\delta)-\lambda_\gamma}\right\}=o(1).
\)
}

Define
\[
\begin{aligned}
\overline T_{n,\gamma}^{\epsilon}
&=\mathsf N_{n,\gamma}^{\mathrm M}
\left\{\mathfrak t_{n,C_{\mathrm{clip}}}
\bigl(\operatorname{smax}_{\beta_n}(\mathcal C_{\gamma,n}^{\pm})\bigr)
\right\},\\
\overline T_{n,\gamma}^{G}
&=\mathsf N_{n,\gamma}^{\mathrm M}
\left\{\mathfrak t_{n,C_{\mathrm{clip}}}
\bigl(\operatorname{smax}_{\beta_n}((\mathcal C_{\gamma,n}^{G})^\pm)\bigr)
\right\},\\
\overline T_{n,\gamma}^{B}
&=\mathsf N_{n,\gamma}^{\mathrm M}
\left\{\mathfrak t_{n,C_{\mathrm{clip}}}
\bigl(\operatorname{smax}_{\beta_n}((\mathcal C_{\gamma,n}^{B})^\pm)\bigr)
\right\},
\end{aligned}
\qquad
\begin{aligned}
T_{n,\gamma}^{\epsilon}
&=\mathsf N_{n,\gamma}^{\mathrm M}(M_{n,p}^{\epsilon,\gamma}),\\
T_{n,\gamma}^{B}
&=\mathsf N_{n,\gamma}^{\mathrm M}(M_{n,p}^{B,\gamma}).
\end{aligned}
\]
For every real sequence \(a_n\), Lemmas~\ref{lem:block-comparison-general}--\ref{lem:joint-gaussian-cov-restoration}
{applied to \(f_{n,a_n}\), with \(\mathbf h=\mathbf0\) and
\(\mathbf d=\mathbf0\), give}
{
\begin{align}
\left|
\E\psi_{a_n,\varepsilon_n}(\overline T_{n,\gamma}^{\epsilon})
-
\E\psi_{a_n,\varepsilon_n}(\overline T_{n,\gamma}^{B})
\right|&{}\le
C L_n^{c_{\mathrm{out}}}r_{n,\gamma}
+C L_n^{c_{\mathrm{out}}+c_{\mathrm{br}}}
\Lambda_{\gamma,n}^{-1}\rightarrow 0.
\label{eq:max-transfer-smooth-expectation}
\end{align}}

For the finite Gaussian array, by triangle inequality, Lemma~\ref{lem:cusum-cov-general} gives
\[
\sup_{j,k}\Var\{C_{0,j}^{G,0}(k)\}\le C,
\qquad
\sup_{j,k\in\mathcal K_{1/2}}
\Var\{C_{1/2,j}^{G,0}(k)\}\le C.
\]
Let \(\eta_{a,\varepsilon}(x)=1-\psi_{a,\varepsilon}(x)\).  Since
\(
\ind{x>a+\varepsilon}\le \eta_{a,\varepsilon}(x)\le \ind{x>a},
\)
taking \(a=C_{\mathrm{clip}}\sqrt{L_n}-1\) and \(\varepsilon=1\), by Lemma~\ref{lem:block-comparison-general}, with \(f(x,y)=
\eta_{C_{\mathrm{clip}}\sqrt{L_n}-1,1}(y)\), \(\mathbf h=\mathbf0\), and
\(\mathbf d=\mathbf0\), gives
\begin{align*}
\Pr\{M_{n,p}^{\epsilon,\gamma}>C_{\mathrm{clip}}\sqrt{L_n}\}
&\le\E\eta_{C_{\mathrm{clip}}\sqrt{L_n}-1,1}
\{\operatorname{smax}_{\beta_n}(\mathcal C_{\gamma,n}^{\pm})\}\\
&\le\E\eta_{C_{\mathrm{clip}}\sqrt{L_n}-1,1}
\{\operatorname{smax}_{\beta_n}((\mathcal C_{\gamma,n}^{G})^\pm)\}
+Cr_{n,\gamma}\\
&\le\Pr\left\{
\operatorname{smax}_{\beta_n}((\mathcal C_{\gamma,n}^{G})^\pm)
>C_{\mathrm{clip}}\sqrt{L_n}-1
\right\}+o(1)\\
&\le\Pr\{M_{n,p}^{G,\gamma}>C_{\mathrm{clip}}\sqrt{L_n}-2\}+o(1)\\
&\le 2pn\exp\{-c(C_{\mathrm{clip}}\sqrt{L_n}-2)^2\}+o(1),
\end{align*}
where the last two inequalities use \eqref{eq:smax-logsumexp-bound}, with
\(N\le 2pn\) and \(\beta=\beta_n\), and the Gaussian tail inequality.
Here \(r_{n,\gamma}=o(1)\) follows from
\(L_n^{3/2+\delta}r_{n,\gamma}\to0\).
Thus, after choosing \(C_{\mathrm{clip}}\) sufficiently large,
\(
\Pr\{M_{n,p}^{\epsilon,\gamma}>C_{\mathrm{clip}}\sqrt{L_n}\}\rightarrow0.
\)
Moreover, Lemma~\ref{lem:gaussian-extremes-general} gives, after enlarging
\(C_{\mathrm{clip}}\) if necessary,
\(
\Pr\{M_{n,p}^{B,\gamma}>C_{\mathrm{clip}}\sqrt{L_n}\}\rightarrow0.
\)

\begingroup
Let \(\Delta_n=CL_n^{-1/4}\).
\endgroup
By \eqref{eq:smax-logsumexp-bound}, taking \(N\le2pn\) and
\begingroup
\(\beta=\beta_n=L_n^{7/4}\),
\(
0\le
\operatorname{smax}_{\beta_n}(\mathbf z)-\max\mathbf z
\le CL_n^{-3/4}.
\)
\endgroup
On the event $A_{n,\gamma}^{\epsilon}
=\{M_{n,p}^{\epsilon,\gamma}\le C_{\mathrm{clip}}\sqrt{L_n}\}$, since \(\|\mathfrak t^{(1)}_{n,C_{\mathrm{clip}}}\|_\infty\le C\),
\begingroup
\begin{align*}
\left|
\mathfrak t_{n,C_{\mathrm{clip}}}
\{\operatorname{smax}_{\beta_n}(\mathcal C_{\gamma,n}^{\pm})\}
-
M_{n,p}^{\epsilon,\gamma}
\right|&{}=
\left|
\mathfrak t_{n,C_{\mathrm{clip}}}
\{\operatorname{smax}_{\beta_n}(\mathcal C_{\gamma,n}^{\pm})\}
-
\mathfrak t_{n,C_{\mathrm{clip}}}(M_{n,p}^{\epsilon,\gamma})
\right|
\\
&{}\le
C\left|
\operatorname{smax}_{\beta_n}(\mathcal C_{\gamma,n}^{\pm})
-
M_{n,p}^{\epsilon,\gamma}
\right|
\le
CL_n^{-3/4}.
\end{align*}
\endgroup
Furthermore, since
\(
\sup_{|u|\le2C_{\mathrm{clip}}\sqrt{L_n}}
\left|(\mathsf N_{n,\gamma}^{\mathrm M})'(u)\right|
\le C\sqrt{L_n},
\)
then
\begin{equation}
\begin{aligned}
\Pr\left(
\left|
\overline T_{n,\gamma}^{\epsilon}
-
T_{n,\gamma}^{\epsilon}
\right|
>
\Delta_n
\right)
\le
\Pr\{(A_{n,\gamma}^{\epsilon})^c\}
=
\Pr\{M_{n,p}^{\epsilon,\gamma}>C_{\mathrm{clip}}\sqrt{L_n}\}
\to0.
\end{aligned}
\label{eq:max-transfer-remove-smooth-epsilon}
\end{equation}
Similarly, we define event $A_{n,\gamma}^{B}
=\{M_{n,p}^{B,\gamma}\le C_{\mathrm{clip}}\sqrt{L_n}\}$, and have
\begin{equation}
\begin{aligned}
\Pr\left(
\left|\overline T_{n,\gamma}^{B}
-T_{n,\gamma}^{B}\right|
>
\Delta_n
\right)
\to0.
\end{aligned}
\label{eq:max-transfer-remove-smooth-B}
\end{equation}
Consequently, by \eqref{eq:max-transfer-remove-smooth-epsilon} and
\eqref{eq:max-transfer-remove-smooth-B}, for any \(b\in\mathbb{R}\),
\begin{equation}
\begin{aligned}
\Pr(\overline T_{n,\gamma}^{\epsilon}\le 
b-\Delta_n)-o(1)\le\Pr(T_{n,\gamma}^{\epsilon}\le b)
&\le
\Pr(\overline T_{n,\gamma}^{\epsilon}\le b+\Delta_n)+o(1),
\\
\Pr(T_{n,\gamma}^{B}\le b-\Delta_n)-o(1)
\le
\Pr(\overline T_{n,\gamma}^{B}\le b)
&\le
\Pr(T_{n,\gamma}^{B}\le b+\Delta_n)+o(1).
\end{aligned}
\label{eq:max-transfer-cdf-clipping}
\end{equation}
Using
\(
\ind{x\le a}\le\psi_{a,\varepsilon_n}(x)
\le\ind{x\le a+\varepsilon_n},
\)
\eqref{eq:max-transfer-smooth-expectation} with
\(a_n=z+\Delta_n\) gives
\begin{align}
\Pr(T_{n,\gamma}^{\epsilon}\le z)
&\le
\Pr(\overline T_{n,\gamma}^{\epsilon}\le z+\Delta_n)+o(1)\le
\E\psi_{z+\Delta_n,\varepsilon_n}
(\overline T_{n,\gamma}^{\epsilon})+o(1) \notag\\
&\le
\E\psi_{z+\Delta_n,\varepsilon_n}
(\overline T_{n,\gamma}^{B})+o(1)\le
\Pr(\overline T_{n,\gamma}^{B}\le z+\Delta_n+\varepsilon_n)+o(1) \notag\\
&\le
\Pr(T_{n,\gamma}^{B}\le z+2\Delta_n+\varepsilon_n)+o(1),
\label{eq:max-transfer-cdf-upper}
\end{align}
where the first and last inequalities use
\eqref{eq:max-transfer-cdf-clipping}.  Similarly,
\eqref{eq:max-transfer-smooth-expectation} with
\(a_n=z-\Delta_n-\varepsilon_n\) gives
\begin{align}
\Pr(T_{n,\gamma}^{\epsilon}\le z)
&\ge
\Pr(T_{n,\gamma}^{B}\le z-2\Delta_n-\varepsilon_n)-o(1).
\label{eq:max-transfer-cdf-lower}
\end{align}

Since $\nu_{\mathrm{dep}}\le1/24$, the assumed upper bound on
$\nu$ implies $\nu<\lambda_\gamma/2$.
Hence one can choose a fixed \(\eta>0\) such that
\[
\nu
<
\frac{\lambda_\gamma}{2+\eta}.
\]
It follows from
\(\log p=O(n^{\nu})\) that $\log p=o\{
n^{\lambda_\gamma/(2+\eta)}\}$. Therefore Lemma~\ref{lem:gaussian-extremes-general} applies and gives $T_{n,\gamma}^{B}\Rightarrow Z_{\mathrm{Gu}}$.

Since \(Z_{\mathrm{Gu}}\) has a continuous distribution function and
\(\Delta_n+\varepsilon_n\rightarrow0\),
\eqref{eq:max-transfer-cdf-upper}--\eqref{eq:max-transfer-cdf-lower} imply $\Pr(T_{n,\gamma}^{\epsilon}\le z)\rightarrow
\Pr(Z_{\mathrm{Gu}}\le z)$ which proves the lemma.
\end{proof}

\begin{proof}[Proof of Theorem~\ref{null:Max}]
For every $j,k$, $C_{\gamma,j}(k)=C_{\gamma,j}^{\epsilon,0}(k)\sigma_j/\widehat\sigma_j$.
On the event $\{r_{\sigma,n}\le1/2\}$ whose probability tends to one by Lemma~\ref{lem:lrv-general} and \ref{lemma:LRV-rate-log} with $A=1$, {we have} $\max_j\left|\sigma_j/\widehat\sigma_j-1\right|\le2r_{\sigma,n}$.

(i). Lemma~\ref{lem:max-transfer-general} implies $M_{n,p}^{\epsilon,0}=O_p\{(\log p)^{1/2}\}$.  Hence
\begin{align*}
|M_{0}-M_{n,p}^{\epsilon,0}|
&\le2r_{\sigma,n}M_{n,p}^{\epsilon,0}=O_p\{a_{\sigma,n}(\log p)^{1/2}\},
\end{align*}
and,
\begin{equation}
\begin{aligned}
\left|
2M_0^2-2(M_{n,p}^{\epsilon,0})^2
\right|
&=
2\left|
M_0-M_{n,p}^{\epsilon,0}
\right|
\left(
M_0+M_{n,p}^{\epsilon,0}
\right)=
O_p\{a_{\sigma,n}\log p\}
=o_p(1),
\end{aligned}
\label{eq:max-feasible-unweighted-error}
\end{equation}
where the last equality follows from Lemma~\ref{lemma:LRV-rate-log} with $A=1$.

(ii). Lemma~\ref{lem:max-transfer-general} implies $M_{n,p}^{\epsilon,1/2}=O_p[\{\log(p\log h_n)\}^{1/2}]$. On the event \(\{r_{\sigma,n}\le1/2\}\),
\[
\begin{aligned}
\left|
M_{1/2}-M_{n,p}^{\epsilon,1/2}
\right|
&\le
2r_{\sigma,n}M_{n,p}^{\epsilon,1/2}=
O_p\left[
a_{\sigma,n}
\{\log(p\log h_n)\}^{1/2}
\right].
\end{aligned}
\]
Since $A_{\mathrm{DE}}(p\log h_n)\asymp\{\log(p\log h_n)\}^{1/2}$, we obtain
\begin{equation}
\begin{aligned}
&A_{\mathrm{DE}}(p\log h_n)
\left|
M_{1/2}-M_{n,p}^{\epsilon,1/2}
\right|=
O_p\{a_{\sigma,n}\log(p\log h_n)\}
=o_p(1),
\end{aligned}
\label{eq:max-feasible-boundary-error}
\end{equation}
where the last equality follows from Lemma~\ref{lemma:LRV-rate-log} with $A=1$ and $\log(p\log h_n)\lesssim\log(np)$.
{Slutsky's theorem and Lemma~\ref{lem:max-transfer-general} prove parts (i) and (ii).}
\end{proof}

\begin{proof}[Proof of Proposition~\ref{prop:consistency-Linf}]

By Lemma~\ref{lem:lrv-general}, {we have}
\[
r_{\sigma,n}
=
\max_{1\le j\le p}
\left|
\frac{\widehat\sigma_j}{\sigma_j}-1
\right|
=
O_p\{a_{\sigma,n}^{(\delta)}\},
\]
where
\[
a_{\sigma,n}^{(\delta)}
=
a_{\sigma,n}
+
\frac{
\ell_{\max}^{3/2}\|\bdelta\|_\infty
\{\log(np)\}^{c_{\mathrm{LRV}}}
}{n}
+
\frac{
\ell_{\max}^{2}\|\bdelta\|_\infty^2
}{n}.
\]

Under the assumptions in part~(i), {we have} $\ell_{\max}^{2}\|\bdelta\|_\infty^2\log(2p)/n\rightarrow0$. Moreover,
\[
\begin{aligned}
\frac{
\ell_{\max}^{3/2}\|\bdelta\|_\infty
\{\log(np)\}^{c_{\mathrm{LRV}}}\log(2p)
}{n}&=
o\left\{
\frac{
\ell_{\max}^{1/2}\{\log(np)\}^{c_{\mathrm{LRV}}}
\sqrt{\log(2p)}
}{\sqrt n}
\right\}
\\
&\le
o\left\{
\frac{
\ell_{\max}\{\log(np)\}^{c_{\mathrm{LRV}}}
\log(2p)
}{\sqrt n}
\right\}
=o(1),
\end{aligned}
\]
where the last equality follows from
{Lemma~\ref{lemma:LRV-rate-log} with $A=1$}. Hence $a_{\sigma,n}^{(\delta)}\log(2p)\rightarrow0$. Similarly, {we have} $a_{\sigma,n}^{(\delta)}\log(p\log h_n)\rightarrow0$.

For \(k\in\mathcal K_\gamma\), define
\[
C_{\gamma,j}^{X,0}(k)=v_k^{-\gamma}
\frac{\sum_{i=1}^n c_{i,k} X_{ij}}{\sqrt n\,\sigma_j},
\qquad
d_{\gamma,j}(k)=v_k^{-\gamma}\frac{d_{k,j}}{\sigma_j}.
\]
Here $d_{k,j}$ is the $j$th component of $\mathbf d_k$.
Then
\[
C_{\gamma,j}^{X,0}(k)=C_{\gamma,j}^{\epsilon,0}(k)+d_{\gamma,j}(k),
\qquad
C_{\gamma,j}(k)=\frac{\sigma_j}{\widehat\sigma_j}C_{\gamma,j}^{X,0}(k).
\]
Choose \(j_0\) such that $|\delta_{j_0}|=\|\bdelta\|_\infty$.  Consequently,
$d_{\gamma,j_0}(\tau)=-\sqrt n\,v_\tau^{1-\gamma}\delta_{j_0}/\sigma_{j_0}$. Since $\tau/n\in[\vartheta,1-\vartheta]$ and Assumption~\ref{ass:C2}, there exists a constant \(c_\vartheta>0\) such that $|d_{\gamma,j_0}(\tau)|\ge c_\vartheta\sqrt n\,\|\bdelta\|_\infty$, $\gamma\in\{0,1/2\}$.

Recall $M_{n,p}^{\epsilon,\gamma}=\max_{k\in\mathcal K_\gamma}\max_{1\le j\le p}|C_{\gamma,j}^{\epsilon,0}(k)|$. By Theorem~\ref{null:Max}, {we have}
\[
M_{n,p}^{\epsilon,0}=O_p\{\sqrt{\log(2p)}\},
\qquad
M_{n,p}^{\epsilon,1/2}=O_p\{\sqrt{\log(p\log h_n)}\},
\]
with the fact that $D_{\mathrm{DE}}(p\log h_n)/A_{\mathrm{DE}}(p\log h_n)=\sqrt{2\log (p\log h_n)}\{1+o(1)\}$.
Since $\widehat\sigma_j/\sigma_j\le 1+r_{\sigma,n}$, we have $\frac{\sigma_j}{\widehat\sigma_j}\ge(1+r_{\sigma,n})^{-1}$. Therefore,
\[
\begin{aligned}
M_\gamma
&\ge|C_{\gamma,j_0}(\tau)|=\frac{\sigma_{j_0}}{\widehat\sigma_{j_0}}\left|C_{\gamma,j_0}^{\epsilon,0}(\tau)+d_{\gamma,j_0}(\tau)\right|\\
&\ge\frac{|d_{\gamma,j_0}(\tau)|-|C_{\gamma,j_0}^{\epsilon,0}(\tau)|}{1+r_{\sigma,n}}\ge\frac{c_\vartheta\sqrt n\,\|\bdelta\|_\infty-M_{n,p}^{\epsilon,\gamma}}{1+r_{\sigma,n}}.
\end{aligned}
\]

For part~(i), it follows that
\[
M_0
\ge
\frac{
c_\vartheta\sqrt n\,\|\bdelta\|_\infty
-
O_p\{\sqrt{\log(2p)}\}
}{
1+O_p\{a_{\sigma,n}^{(\delta)}\}
}.
\]
The assumptions $
\sqrt n\,\|\bdelta\|_\infty/\sqrt{\log(2p)}\rightarrow\infty$ and $a_{\sigma,n}^{(\delta)}\log(2p)\rightarrow0$ give $2M_0^2-\log(2p)
\xrightarrow{p}\infty$. Hence under $H_1$,
\[
\Pr
\left\{
2M_0^2-\log(2p)
>
F_{\mathrm{Gu}}^{-1}(1-\alpha)
\right\}
\rightarrow1.
\]

For part~(ii),
\[
M_{1/2}
\ge
\frac{
c_\vartheta\sqrt n\,\|\bdelta\|_\infty
-
O_p\{\sqrt{\log(p\log h_n)}\}
}{
1+O_p\{a_{\sigma,n}^{(\delta)}\}
}.
\]
The assumptions $\sqrt n\|\bdelta\|_\infty/\sqrt{\log(p\log h_n)}\rightarrow\infty$
and $a_{\sigma,n}^{(\delta)}
\log(p\log h_n)\rightarrow0$ in part (ii) give $A_{\mathrm{DE}}(p\log h_n)M_{1/2}-D_{\mathrm{DE}}(p\log h_n)\xrightarrow{p}\infty$. Hence, under $H_1$,
\[
\Pr
\left[
A_{\mathrm{DE}}(p\log h_n)M_{1/2}
-
D_{\mathrm{DE}}(p\log h_n)
>
F_{\mathrm{Gu}}^{-1}(1-\alpha)
\right]
\rightarrow1.
\]
The proof is completed.
\end{proof}

Let
\[
\mathsf N_{n,0}^{\mathrm S}(x)=x,
\qquad
\mathsf N_{n,1/2}^{\mathrm S}(x)
=A_{\mathrm{DE}}(\log h_n)x-D_{\mathrm{DE}}(\log h_n),
\]
and let $\mathsf N_{n,0}^{\mathrm M}$ and
$\mathsf N_{n,1/2}^{\mathrm M}$ be as in \eqref{eq:max-normalizations}.
Recall $t_k=k/n$, $v_k=t_k(1-t_k)$, and
$\mathbf B_\Omega(t)=\bOme^{1/2}\mathbf B(t)$, with
$\sigma_j^2=(\bOme)_{jj}$.

\begingroup
For $\gamma\in\{0,1/2\}$ and a
deterministic array $\mathbf h=(\mathbf h_k:k\in\mathcal K_\gamma)$
with $\mathbf h_k\in\mathbb R^p$, define
\begin{equation}
\begin{aligned}
&S_{n,\gamma}^{B}(\mathbf h)
=\max_{k\in\mathcal K_\gamma}
\frac{v_k^{-2\gamma}}{\sqrt p\,\omega_p}
\{\|\mathbf B_\Omega(t_k)+\mathbf h_k\|_2^2-v_k\Tr(\bOme)\},\\
&M_{n,p}^{B,\gamma}(\mathbf h)
=\max_{k\in\mathcal K_\gamma,\,j\le p}
v_k^{-\gamma}\frac{|B_{\Omega,j}(t_k)+h_{k,j}|}{\sigma_j},\\
&\mathcal X_{n,\gamma}^{B}
=\mathsf N_{n,\gamma}^{\mathrm S}\{S_{n,\gamma}^{B}(\mathbf h)\},
\qquad
\mathcal Y_{n,\gamma}^{B}
=\mathsf N_{n,\gamma}^{\mathrm M}\{M_{n,p}^{B,\gamma}(\mathbf h)\}.
\end{aligned}
\label{eq:gaussian-decoupling-drifted-statistics}
\end{equation}
Let
\[
\begin{aligned}
\mathcal A&=\bigcup_{k\in\mathcal K_\gamma}\{j:h_{k,j}\ne0\},
\qquad \mathcal A^c=\{1,\ldots,p\}\setminus\mathcal A.
\end{aligned}
\]
For $0<|\mathcal A|<p$, define
\[
\chi_{\mathcal A,p}
=\max_{j\in\mathcal A^c}
\frac{\bOme_{j,\mathcal A}\bOme_{\mathcal A,\mathcal A}^{-1}
\bOme_{\mathcal A,j}}{(\bOme)_{jj}},
\]
and set $\chi_{\mathcal A,p}=0$ when $\mathcal A=\varnothing$.
The support condition in the following lemma ensures $|\mathcal A|<p$
for all sufficiently large $n,p$. 
\begin{lemma}
\label{lem:gaussian-decoupling-general}
Suppose Assumption~\ref{ass:C2} holds,
$\log\log n=o(p^{1/7})$, and $\log p=O(n^\nu)$ for a fixed $\nu>0$.
For the deterministic array $\mathbf h$ defined above, assume
\[
\chi_{\mathcal A,p}\{\log(np)\}^2\to0,
\qquad
\sqrt{\frac{\log\log n}{p}}\,
\E\max_{k\in\mathcal K_\gamma}
\frac{|\mathbf h_k^\top\mathbf B_\Omega(t_k)|}{v_k}\to0.
\]
Consider either of the following cases:
\begin{enumerate}
\item[(i)] When $\gamma=0$, $0<\nu<1/2$, $|\mathcal A|/p\to0$, and
$\mathbf h_k=-\sqrt nK_B(t_k,\tau/n)\bdelta$ for some
$\bdelta\in\mathbb R^p$ and $\tau\in\{1,\ldots,n-1\}$;
\item[(ii)] When $\gamma=1/2$, $|\mathcal A|(\log\log n)^2/p\to0$,
$\lambda_n\asymp n^\lambda$ for a fixed $\lambda\in(0,1)$,
and $0<\nu<\lambda/2$.
\end{enumerate}
Then, for every fixed $x,y\in\mathbb R$,
\[
\Pr(\mathcal X_{n,\gamma}^{B}\le x,\mathcal Y_{n,\gamma}^{B}\le y)
-\Pr(\mathcal X_{n,\gamma}^{B}\le x)
 \Pr(\mathcal Y_{n,\gamma}^{B}\le y)\to0.
\]
\end{lemma}
\endgroup

\begin{proof}
\begingroup
Let $L_n=\log(np)$, $K_n=\log(2+\log n)\sim\log\log n$,
and $J_n=K_n+\log p$.
\endgroup
Recall
$v(t)=t(1-t)$ and $\varphi_\gamma(t)=v(t)^{-2\gamma}$.
\begingroup
Recall $K_B(s,t)=s\wedge t-st$.
\endgroup
Since $\E\|\mathbf B_\Omega(t_k)\|_2^2=v_k\Tr(\bOme)$, expanding \eqref{eq:def-boundary-quadratic-shifted} directly gives
\begin{equation}
\begin{aligned}
\mathcal X_{n,\gamma}^{B}
&=\mathsf N_{n,\gamma}^{\mathrm S}\left[
\max_{k\in\mathcal K_\gamma}
\left\{G_{p,\gamma}(t_k)+\mathfrak d_{\gamma,n}^{\mathrm{det}}(k)
+R_{h,\gamma,n}^{G}(k)\right\}\right],
\end{aligned}
\label{eq:matched-decoupling-shift-expansion-1}
\end{equation}
where 
\begin{equation}
G_{p,\gamma}(t_k)
=\frac{v_k^{-2\gamma}}{\sqrt p\,\omega_p}
\left\{\|\mathbf B_\Omega(t_k)\|_2^2-v_k\Tr(\bOme)\right\},
\label{eq:matched-decoupling-unshifted-quadratic-process}
\end{equation}
and
\begin{equation}
\begin{aligned}
\mathfrak d_{\gamma,n}^{\mathrm{det}}(k)
&=\frac{v_k^{-2\gamma}\|\mathbf h_k\|_2^2}{\sqrt p\,\omega_p},
\qquad
R_{h,\gamma,n}^{G}(k)
=\frac{2v_k^{-2\gamma}\mathbf h_k^\top\mathbf B_\Omega(t_k)}
{\sqrt p\,\omega_p}.
\end{aligned}
\label{eq:matched-decoupling-shift-expansion-2}
\end{equation}

For $\mathcal A\ne\varnothing$, define
\begin{equation}
\begin{aligned}
\mathbf R_{\mathcal A}(t)
&:=\mathbf B_{\Omega,\mathcal A^c}(t)
-\bOme_{\mathcal A^c,\mathcal A}\bOme_{\mathcal A,\mathcal A}^{-1}
\mathbf B_{\Omega,\mathcal A}(t),\\
\bOme_R
&:=\bOme_{\mathcal A^c,\mathcal A^c}
-\bOme_{\mathcal A^c,\mathcal A}\bOme_{\mathcal A,\mathcal A}^{-1}
\bOme_{\mathcal A,\mathcal A^c}.
\end{aligned}
\label{eq:matched-decoupling-regression}
\end{equation}
Then $\mathbf B_{\Omega,\mathcal A^c}(t)$ can be decomposed as
\(
\mathbf B_{\Omega,\mathcal A^c}(t)
=\mathbf R_{\mathcal A}(t)
+\bOme_{\mathcal A^c,\mathcal A}\bOme_{\mathcal A,\mathcal A}^{-1}
\mathbf B_{\Omega,\mathcal A}(t),
\)
and
\begin{equation*}
\begin{aligned}
\Cov\{\mathbf R_{\mathcal A}(s),\mathbf R_{\mathcal A}(t)\}
=K_B(s,t)\bOme_R,\qquad
\{\mathbf R_{\mathcal A}(t):0\le t\le1\} 
\perp
\{\mathbf B_{\Omega,\mathcal A}(t):0\le t\le1\}.
\end{aligned}
\end{equation*}
When $\mathcal A=\varnothing$, take
$\mathbf R_{\mathcal A}(t)=\mathbf B_\Omega(t)$, $\bOme_R=\bOme$, and
$\chi_{\mathcal A,p}=0$.
Replacing $(\mathbf B_\Omega,\bOme)$ by
$(\mathbf R_{\mathcal A},\bOme_R)$ in
\eqref{eq:matched-decoupling-unshifted-quadratic-process}, define
\begin{equation}
\begin{aligned}
G_{R,p,\gamma}(t_k)
:=\frac{\varphi_\gamma(t_k)}{\sqrt p\,\omega_p}
\left\{
\|\mathbf R_{\mathcal A}(t_k)\|_2^2
-v_k\Tr(\bOme_R)
\right\}.
\end{aligned}
\label{eq:matched-decoupling-residual-quadratic-process}
\end{equation}
Taking $\gamma=1/2$ in
\eqref{eq:matched-decoupling-residual-quadratic-process} gives
\[
G_{R,p,1/2}(t_k)
=\frac{\|\mathbf R_{\mathcal A}(t_k)\|_2^2/v_k
-\Tr(\bOme_R)}{\sqrt p\,\omega_p}.
\]
Since $\varphi_\gamma(t_k)=v_k^{1-2\gamma}\varphi_{1/2}(t_k)$,
for either $\gamma\in\{0,1/2\}$ we have
\[
\begin{aligned}
G_{R,p,\gamma}(t_k)=v_k^{1-2\gamma}G_{R,p,1/2}(t_k),
\qquad
G_{p,\gamma}(t_k)=v_k^{1-2\gamma}G_{p,1/2}(t_k).
\end{aligned}
\]
Let $q=p-|\mathcal A|$. Since $|\mathcal A|/p\to0$, we have
$q\sim p$ and $q>0$ for all sufficiently large $n,p$.
Coordinates of the residual bridge below retain
their original indices $j\in\mathcal A^c$.

By Assumption~\ref{ass:C2},
\(
C_0\le\lambda_{\min}(\bOme_R)
\le\lambda_{\max}(\bOme_R)\le C_1.
\)
After ordering the coordinates as $(\mathcal A,\mathcal A^c)$,
$\bOme-\operatorname{diag}(\mathbf0_{|\mathcal A|\times|\mathcal A|},\bOme_R)$
is positive semidefinite, has rank at most $|\mathcal A|$, and has bounded
operator norm. Consequently,
\[
|\Tr(\bOme^2)-\Tr(\bOme_R^2)|\le C|\mathcal A|,\qquad
\frac{\Tr(\bOme_R^2)}q
=\frac{\Tr(\bOme^2)}p+O(|\mathcal A|/p)\to\vartheta_\Omega.
\]

Define
\begingroup
\begin{align*}
X_{R,n,\gamma}^{G}
&=\mathsf N_{n,\gamma}^{\mathrm S}\left(
\max_{k\in\mathcal K_\gamma}
\{G_{R,p,\gamma}(t_k)+\mathfrak d_{\gamma,n}^{\mathrm{det}}(k)\}
\right),\\
Y_{R,n,\gamma}^{G}
&=\mathsf N_{n,\gamma}^{\mathrm M}\left(
\max_{k\in\mathcal K_\gamma,\,j\in\mathcal A^c}
v_k^{-\gamma}\frac{|R_{\mathcal A,j}(t_k)|}
{\sqrt{(\bOme_R)_{jj}}}\right),\\
Y_{\mathcal A,n,\gamma}^{G}
&=\mathsf N_{n,\gamma}^{\mathrm M}\left(
\max_{k\in\mathcal K_\gamma,\,j\in\mathcal A}
v_k^{-\gamma}\frac{|B_{\Omega,j}(t_k)+h_{k,j}|}{\sigma_j}
\right),\\
Y_{\mathcal A^c,n,\gamma}^{G}
&=\mathsf N_{n,\gamma}^{\mathrm M}\left(
\max_{k\in\mathcal K_\gamma,\,j\in\mathcal A^c}
v_k^{-\gamma}\frac{|B_{\Omega,j}(t_k)|}{\sigma_j}
\right).
\end{align*}
\endgroup
Set $Y_{\mathcal A,n,\gamma}^{G}=-\infty$ when $\mathcal A=\varnothing$.
Since $h_{k,j}=0$ for every $k$ and $j\in\mathcal A^c$,
\[
\mathcal Y_{n,\gamma}^{B}
=\max(Y_{\mathcal A^c,n,\gamma}^{G},Y_{\mathcal A,n,\gamma}^{G}).
\]
For $x,y\in\mathbb R$, define $A_{\gamma}:=A_{\gamma,n}(x)=\{X_{R,n,\gamma}^{G}\le x\}$,
$D_{\gamma}:=D_{\gamma,n}(y)=\{Y_{\mathcal A,n,\gamma}^{G}\le y\}$,
\[
\begin{aligned}
B_{\gamma,j}&:=B_{\gamma,j,n}(y)=\left\{
\mathsf N_{n,\gamma}^{\mathrm M}\left(
\max_{k\in\mathcal K_\gamma}
v_k^{-\gamma}\frac{|R_{\mathcal A,j}(t_k)|}
{\sqrt{(\bOme_R)_{jj}}}\right)>y
\right\},\qquad j\in\mathcal A^c,\\
E_{\gamma}&:=E_{\gamma,n}(y)=D_{\gamma}\cap
\bigcap_{j\in\mathcal A^c}B_{\gamma,j}^c.
\end{aligned}
\]
The independence of $\mathbf R_{\mathcal A}(t)$ and $\mathbf B_{\Omega,\mathcal A}(t)$ implies that $A_{\gamma}$ is
independent of $D_{\gamma}$. 
The triangle inequality yields
\begin{equation}
\begin{aligned}
&|\Pr(\mathcal X_{n,\gamma}^{B}\le x,\mathcal Y_{n,\gamma}^{B}\le y)
-\Pr(\mathcal X_{n,\gamma}^{B}\le x)
 \Pr(\mathcal Y_{n,\gamma}^{B}\le y)|\\
&\quad\le
|\Pr(A_{\gamma}\cap E_{\gamma})
-\Pr(A_{\gamma})\Pr(E_{\gamma})|\\
&\qquad+
\left|\Pr(\mathcal X_{n,\gamma}^{B}\le x,\mathcal Y_{n,\gamma}^{B}\le y)
-\Pr(A_{\gamma}\cap E_{\gamma})\right|\\
&\qquad+
\left|\Pr(\mathcal X_{n,\gamma}^{B}\le x)
-\Pr(A_{\gamma})\right|
+\left|\Pr(\mathcal Y_{n,\gamma}^{B}\le y)
-\Pr(E_{\gamma})\right|.
\end{aligned}
\label{eq:matched-decoupling-factorization-reduction}
\end{equation}

\begingroup
For the first term in \eqref{eq:matched-decoupling-factorization-reduction},
we use Bonferroni's inequalities after removing a fixed number of residual coordinates.
\endgroup
For $J\subset\mathcal A^c$, let
$B_{\gamma,J}=\bigcap_{j\in J}B_{\gamma,j}$, and for $r\ge1$, define
\begin{align*}
\zeta_{\gamma,n}(r)
&=
\sum_{J\subset\mathcal A^c,|J|=r}
\Bigl|
\Pr(A_{\gamma}\cap D_{\gamma}\cap B_{\gamma,J})
-\Pr(A_{\gamma})
 \Pr(D_{\gamma}\cap B_{\gamma,J})
\Bigr|,\\
H_{\gamma,n}(r)
&=
\sum_{J\subset\mathcal A^c,|J|=r}
\Pr(B_{\gamma,J}).
\end{align*}
The dependence on the thresholds $x,y$ is suppressed in these two notations.
Since
\(
\bigcap_{j\in\mathcal A^c}B_{\gamma,j}^c
=\{Y_{R,n,\gamma}^{G}\le y\},
\)
Bonferroni's inequalities and the independence of $A_{\gamma}$
and $D_{\gamma}$ give, for every integer $m\ge1$,
\begin{align}
\Bigl|
\Pr(A_{\gamma}\cap E_{\gamma})
-\Pr(A_{\gamma})\Pr(E_{\gamma})
\Bigr|
\le
H_{\gamma,n}(2m+1)
+\sum_{r=1}^{2m}\zeta_{\gamma,n}(r).
\label{eq:matched-decoupling-bonferroni}
\end{align}

Let
\(
\widetilde{\mathsf N}_{n,\gamma}^{\mathrm M}
:=\left.\mathsf N_{n,\gamma}^{\mathrm M}\right|_{p\mapsto q}.
\)
Define $\widetilde H_{\gamma,n}(r)$ and
$\widetilde Y_{R,n,\gamma}^{G}$ by replacing
$\mathsf N_{n,\gamma}^{\mathrm M}$ with
$\widetilde{\mathsf N}_{n,\gamma}^{\mathrm M}$
in the definitions of $H_{\gamma,n}(r)$ and
$Y_{R,n,\gamma}^{G}$, respectively.
For $j\in\mathcal A^c$ and $y\in\mathbb R$, define
\[
\widetilde B_{\gamma,j}(y)
:=\left\{
\widetilde{\mathsf N}_{n,\gamma}^{\mathrm M}\left(
\max_{k\in\mathcal K_\gamma}
v_k^{-\gamma}
\frac{|R_{\mathcal A,j}(t_k)|}{\sqrt{(\bOme_R)_{jj}}}
\right)>y
\right\},
\qquad
\widetilde H_{\gamma,n}(r)
:=\sum_{J\subset\mathcal A^c,\,|J|=r}
\Pr\left(\bigcap_{j\in J}\widetilde B_{\gamma,j}(y)\right).
\]
Recall the definition of $\mathsf N_{n,\gamma}^{\mathrm M}$ in \eqref{eq:max-normalizations}, for every $j\in\mathcal A^c$,
\[
B_{0,j}(y)
=\widetilde B_{0,j}\bigl(y+\log(p/q)\bigr),
\qquad
B_{1/2,j}(y)
=\widetilde B_{1/2,j}\left(\frac{y-b_n}{a_n}\right),
\]
where
\[
a_n:=\frac{A_{\mathrm{DE}}(p\log h_n)}
{A_{\mathrm{DE}}(q\log h_n)}>0,
\qquad
b_n:=a_nD_{\mathrm{DE}}(q\log h_n)
-D_{\mathrm{DE}}(p\log h_n).
\]
Since $q/p\to1$, we have
$\log(p/q)=o(1)$, $a_n=1+o(1)$, and $b_n=o(1)$.
Thus, for each fixed $y$, $B_{\gamma,j}(y)
=\widetilde B_{\gamma,j}\bigl(y+o(1)\bigr)$.
Then, for every fixed
$\epsilon>0$ and all sufficiently large $n,p$,
\[
\widetilde B_{\gamma,j}(y+\epsilon)
\subseteq B_{\gamma,j}(y)
\subseteq\widetilde B_{\gamma,j}(y-\epsilon),
\qquad j\in\mathcal A^c.
\]
{Hence, for every fixed $\epsilon>0$ and all sufficiently
large $n,p$,}
\[
\widetilde H_{\gamma,n}(r;y+\epsilon)
\le H_{\gamma,n}(r;y)
\le\widetilde H_{\gamma,n}(r;y-\epsilon).
\]
The condition
\(\log p=O(n^\nu)\) {with \(\nu<1/2\) or \(\nu<\lambda/2\) allows a fixed} $\eta_0>0$ with
$\nu<1/(2+\eta_0)$ when $\gamma=0$, or
$\nu<\lambda/(2+\eta_0)$ when $\gamma=1/2$.
Thus $\log p=O(n^\nu)$ implies the corresponding growth condition
in Lemma~\ref{lem:gaussian-extremes-general}.
Since $q\sim p$, these growth conditions also hold with $q$ in
place of $p$.
The spectral bounds above and Lemma~\ref{lem:bounded-spectrum-weak-corr}
verify the correlation conditions for the residual bridge.
More precisely, writing
$\rho^R_{jj'}=(\bOme_R)_{jj'}/\{(\bOme_R)_{jj}(\bOme_R)_{j'j'}\}^{1/2}$,
we have
\[
\max_{j\ne j'}|\rho^R_{jj'}|\le1-C_0/C_1<1,
\qquad
\max_j\sum_{j'\in\mathcal A^c}(\rho^R_{jj'})^2\le C.
\]
Taking $\delta_q=(\log q)^{-2}$ and
$\upsilon_q=5\log\log q/\log q$, we obtain
\[
\max_j\bigl|\{j'\in\mathcal A^c\setminus\{j\}:|\rho^R_{jj'}|\ge\delta_q\}\bigr|
\le C(\log q)^4<q^{\upsilon_q}.
\]
Then, by the arguments in the proofs of Propositions~S.1--S.2 of
\citet{wang2023JRSSB}, applied to the residual
bridge \(\mathbf R_{\mathcal A}=(R_{\mathcal A,j})_{j\in\mathcal A^c}\), for every fixed integer $r\ge1$ and $y\in\mathbb R$,
we have
\[
\widetilde H_{\gamma,n}(r;y)
\rightarrow\frac{\exp(-ry)}{r!}.
\]
Consequently,
\[
\begin{aligned}
\frac{\exp\{-r(y+\epsilon)\}}{r!}
\le\liminf_{n,p\to\infty}H_{\gamma,n}(r;y)
\le\limsup_{n,p\to\infty}H_{\gamma,n}(r;y)
\le\frac{\exp\{-r(y-\epsilon)\}}{r!}.
\end{aligned}
\]
Letting $\epsilon\downarrow0$ yields
\begin{equation}
H_{\gamma,n}(r;y)\rightarrow\frac{\exp(-ry)}{r!}.
\label{eq:matched-decoupling-factorial-limit}
\end{equation}
Letting $n,p\to\infty$ at each fixed truncation order and then
letting that order tend to infinity, \eqref{eq:matched-decoupling-factorial-limit}
gives, for every fixed $y\in\mathbb R$,
\begin{align}
\Pr(Y_{R,n,\gamma}^{G}\le y)&=\Pr\left(\bigcap_{j\in\mathcal A^c}B_{\gamma,j}^c\right)
=\sum_{r=0}^{q}(-1)^rH_{\gamma,n}(r)\notag\\
&\rightarrow
\sum_{r=0}^{\infty}\frac{(-1)^r\exp(-ry)}{r!}
=\exp(-\exp(-y))
=F_{\mathrm{Gu}}(y),
\label{eq:matched-decoupling-residual-gumbel}
\end{align}
with $H_{\gamma,n}(0):=1.$
In particular, $H_{\gamma,n}(r)=O(1)$ for every fixed $r$, and
\[
\begin{aligned}
\lim_{m\to\infty}\lim_{n,p\to\infty}
H_{\gamma,n}(2m+1)=
\lim_{m\to\infty}
\frac{e^{-(2m+1)y}}{(2m+1)!}=0.
\end{aligned}
\]

{It remains to prove that}, for every fixed integer $r\ge1$,
\begin{equation}
\sup_x\zeta_{\gamma,n}(r)\rightarrow0.
\label{eq:matched-decoupling-zeta-vanishing}
\end{equation}
For $x\in\mathbb R$, define
\begin{equation}
\beta_k(x)
:=\frac{(\mathsf N_{n,\gamma}^{\mathrm S})^{-1}(x)
-\mathfrak d_{\gamma,n}^{\mathrm{det}}(k)}
{v_k^{1-2\gamma}},\qquad k\in\mathcal K_\gamma.
\label{eq:matched-decoupling-dense-thresholds}
\end{equation}
Then the event $A_{\gamma}$ can be written as
\(
A_{\gamma}
=\left\{G_{R,p,1/2}(t_k)
\le \beta_k(x),
\ \forall k\in\mathcal K_\gamma\right\}.
\)
For each $J\subset\mathcal A^c$ with $|J|<q$, let
$J^c=\mathcal A^c\setminus J$ within this step, and define
\[
\begin{aligned}
\mathbf R_J(t)
&:=\mathbf R_{\mathcal A,J^c}(t)
-(\bOme_R)_{J^c,J}(\bOme_R)_{J,J}^{-1}
\mathbf R_{\mathcal A,J}(t),\\
\bOme_{R,J^c\mid J}
&:=(\bOme_R)_{J^c,J^c}
-(\bOme_R)_{J^c,J}(\bOme_R)_{J,J}^{-1}
(\bOme_R)_{J,J^c}.
\end{aligned}
\]
Define $G_{J,p,\gamma}(t_k)$ by
\eqref{eq:matched-decoupling-residual-quadratic-process}, replacing
$(\mathbf R_{\mathcal A},\bOme_R)$ by
$(\mathbf R_J,\bOme_{R,J^c\mid J})$ and retaining the denominator
$\sqrt p\,\omega_p$.
For $J=\varnothing$, take $\mathbf R_J=\mathbf R_{\mathcal A}$ and
$\bOme_{R,J^c\mid J}=\bOme_R$, so that
$G_{\varnothing,p,\gamma}=G_{R,p,\gamma}$.
Fix an integer $r\ge1$; then $r<q$ for all sufficiently
large $n,p$. For every $J$ with $|J|=r$, define
\begingroup
\[
A_J^\pm
:=\left\{G_{J,p,1/2}(t_k)
\le \beta_k(x)
\pm\frac{C_rJ_n}{\sqrt p},
\ \forall k\in\mathcal K_\gamma\right\}.
\]
\endgroup
The constant $C_r$ is chosen to satisfy \eqref{eq:matched-decoupling-tail-threshold}.
These events satisfy $A_J^-\subseteq A_J^+$ and are
independent of $D_{\gamma}\cap B_{\gamma,J}$, because $\mathbf R_J$
is jointly independent of $\mathbf R_{\mathcal A,J}$ and
$\mathbf B_{\Omega,\mathcal A}$.
Let
\[
E_{\gamma,J,n}
:=(A_J^-\setminus A_{\gamma})
\cup(A_{\gamma}\setminus A_J^+).
\]
The dependence on $x$ is suppressed in $A_J^\pm$ and $E_{\gamma,J,n}$.
Bounding $A_{\gamma}$ between $A_J^-$ and
$A_J^+$ outside event $E_{\gamma,J,n}$ gives
\[
\begin{aligned}
&\left|\Pr(A_{\gamma}\cap D_{\gamma}\cap B_{\gamma,J})
-\Pr(A_{\gamma})\Pr(D_{\gamma}\cap B_{\gamma,J})\right|\\
&\quad\le
\{\Pr(A_J^+)-\Pr(A_J^-)\}
\Pr(D_{\gamma}\cap B_{\gamma,J})+2\Pr(E_{\gamma,J,n}).
\end{aligned}
\]
Summing over $J\subset\mathcal A^c$ with $|J|=r$ and using
$\Pr(D_{\gamma}\cap B_{\gamma,J})\le\Pr(B_{\gamma,J})$
therefore yields
\[
\begin{aligned}
\zeta_{\gamma,n}(r)
&\le H_{\gamma,n}(r)
\max_{|J|=r}\{\Pr(A_J^+)-\Pr(A_J^-)\}+2\sum_{|J|=r}\Pr(E_{\gamma,J,n}).
\end{aligned}
\]
Since $H_{\gamma,n}(r)=O(1)$ by
\eqref{eq:matched-decoupling-factorial-limit}, it suffices to verify

\begin{align}
\sup_x\sum_{|J|=r}\Pr(E_{\gamma,J,n})
\rightarrow0,
\qquad
\sup_x\max_{|J|=r}
\{\Pr(A_J^+)-\Pr(A_J^-)\}\rightarrow0.
\label{eq:matched-decoupling-error-targets}
\end{align}

\begingroup
To prove the first limit in \eqref{eq:matched-decoupling-error-targets},
it suffices to bound the entire quadratic remainder path, since
\[
E_{\gamma,J,n}\subset \left\{
\max_{k\in\mathcal K_\gamma}
|G_{R,p,1/2}(t_k)-G_{J,p,1/2}(t_k)|>\frac{C_rJ_n}{\sqrt p}
\right\},
\]
We first establish the required path bound. Let $\mathbf Z(u)$ be a
stationary Gaussian process in $\mathbb R^d$ with independent coordinates
and covariance $e^{-|u-v|/2}\mathbf I_d$, and let $\mathbf A$ be
deterministic and symmetric. For
$Q_{\mathbf A}(u)=\mathbf Z(u)^\top\mathbf A\mathbf Z(u)-\Tr(\mathbf A)$,
Lemma~\ref{lem:gaussian-tail-general} gives, for every $q\ge2$,
\[
\|Q_{\mathbf A}(u)\|_q
\le C\{\sqrt q\|\mathbf A\|_{\mathrm F}
+q\|\mathbf A\|_{\mathrm{op}}\}.
\]
For $|u-v|\le1$, let $\rho=e^{-|u-v|/2}$. The normalized sum and
difference of $\mathbf Z(u)$ and $\mathbf Z(v)$ are independent
$N_d(\mathbf0,\mathbf I_d)$ vectors $\mathbf X,\mathbf Y$, so
\[
\begin{aligned}
Q_{\mathbf A}(u)-Q_{\mathbf A}(v)
&=2\sqrt{1-\rho^2}\,\mathbf X^\top\mathbf A\mathbf Y,\\
\|Q_{\mathbf A}(u)-Q_{\mathbf A}(v)\|_q
&\le C\sqrt{|u-v|}
\{\sqrt q\|\mathbf A\|_{\mathrm F}+q\|\mathbf A\|_{\mathrm{op}}\}.
\end{aligned}
\]
The second inequality follows by applying the same Gaussian quadratic-form
bound to the symmetric block matrix with off-diagonal blocks
$\mathbf A/2$ and zero diagonal blocks.
Partition an interval of length $T\ge1$ into at most $\lceil T\rceil$
unit intervals, shortening the last one, and use dyadic grids in each.
At level $m\ge0$ there are at most $C(2+T)2^m$ increments of length
at most $2^{-m}$. Markov's inequality at $e$ times their moment bounds, with
$q_m=z+\log\{C(2+T)\}+2m+2$, $z\ge1$, gives a total exceptional
probability at most
\[
C\sum_{m\ge0}(2+T)2^m e^{-q_m}\le Ce^{-z}.
\]
Outside this event, summing the increment bounds and bounding the interval
starting points by the fixed-time inequality yields
\[
\begin{aligned}
\sum_{m\ge0}2^{-m/2}\sqrt{q_m}
&\le C\sqrt{z+\log(2+T)},\\
\sum_{m\ge0}2^{-m/2}q_m
&\le C\{z+\log(2+T)\}.
\end{aligned}
\]
Continuity extends the bounds from the dyadic grids to the full interval.
Thus, uniformly in its location and in $d,\mathbf A,T,z$,
\begin{equation}
\begin{aligned}
\Pr\bigg[\sup_{0\le u\le T}|Q_{\mathbf A}(u)|
&>C\big\{\|\mathbf A\|_{\mathrm F}\sqrt{z+\log(2+T)}\\
&\hspace{18mm}+\|\mathbf A\|_{\mathrm{op}}\{z+\log(2+T)\}\big\}\bigg]
\le Ce^{-z}.
\end{aligned}
\label{eq:ou-quadratic-maximal-bound}
\end{equation}

\endgroup
For every \(\mathbf x\in\mathbb R^q\), let \(\mathbf E_J\) and
\(\mathbf E_{J^c}\) denote the coordinate-selection matrices satisfying
\(\mathbf E_J\mathbf x=\mathbf x_J\) and
\(\mathbf E_{J^c}\mathbf x=\mathbf x_{J^c}\), and define
\(
\mathbf L_J
=\mathbf E_{J^c}
-(\bOme_R)_{J^c,J}(\bOme_R)_{J,J}^{-1}\mathbf E_J.
\)
\begingroup
For $t(u)=e^u/(1+e^u)$, the whitened process
$\mathbf Z(u)=\bOme_R^{-1/2}\mathbf R_{\mathcal A}(t(u))/\sqrt{v(t(u))}$
has exactly the OU covariance above. Jointly over all $1\le k<n$,
\begin{equation}
\begin{aligned}
&G_{R,p,1/2}(t_k)-G_{J,p,1/2}(t_k)
=Q_{\mathbf A_J}\!\left(\log\frac{k}{n-k}\right),\\
&\mathbf A_J
:=\frac{1}{\sqrt p\,\omega_p}
\bOme_R^{1/2}
\bigl(\mathbf I_q-\mathbf L_J^\top\mathbf L_J\bigr)
\bOme_R^{1/2}.
\end{aligned}
\label{eq:matched-decoupling-remainder}
\end{equation}
After ordering coordinates as $(J,J^c)$, the $J^c\times J^c$ block of
$\mathbf I_q-\mathbf L_J^\top\mathbf L_J$ is zero. Thus, uniformly
over $J$, $\operatorname{rank}(\mathbf A_J)\le2|J|$, and
\endgroup
\begin{equation}
\begin{aligned}
\|\mathbf A_J\|_{\mathrm{op}}
\lesssim\frac 1{\sqrt p},
\qquad
\|\mathbf A_J\|_{\mathrm F}^2
\le\operatorname{rank}(\mathbf A_J)\|\mathbf A_J\|_{\mathrm{op}}^2
\lesssim\frac{|J|}{p}.
\end{aligned}
\label{eq:matched-decoupling-matrix-bounds}
\end{equation}
\begingroup
The logit observation times lie in an interval of length
$2\log(n-1)$, whose logarithm is $O(K_n)$. Hence
\eqref{eq:ou-quadratic-maximal-bound} gives, for $z\ge1$,
\begin{equation}
\begin{aligned}
\Pr\bigg\{\max_{1\le k<n}|G_{R,p,1/2}(t_k)-G_{J,p,1/2}(t_k)|
&>C\bigg(\sqrt{\frac{|J|(z+K_n)}p}\\
&\hspace{8mm}+\frac{z+K_n}{\sqrt p}\bigg)\bigg\}\le Ce^{-z}.
\end{aligned}
\label{eq:matched-decoupling-pointwise-tail}
\end{equation}
For fixed $r$, take $z=(r+3)\log p$ and choose $C_r$ so that
\begin{equation}
C\left(\sqrt{\frac{r\{(r+3)\log p+K_n\}}p}
+\frac{(r+3)\log p+K_n}{\sqrt p}\right)
\le\frac{C_rJ_n}{\sqrt p}.
\label{eq:matched-decoupling-tail-threshold}
\end{equation}
The constants are uniform over subsets of size $r$. Taking the union only
over those subsets gives
\begin{equation}
\begin{aligned}
\sup_x\sum_{|J|=r}\Pr(E_{\gamma,J,n})
&\le C\binom qr p^{-(r+3)}\le Cp^{-3}\to0.
\end{aligned}
\label{eq:matched-decoupling-removal}
\end{equation}
\endgroup

We now verify the second result in \eqref{eq:matched-decoupling-error-targets}. 
By the definitions of $A_J^\pm$, we have
\begingroup
\begin{equation}
\begin{aligned}
\Pr(A_J^+)-\Pr(A_J^-)
&\le\Pr\left\{
\left|\max_{k\in\mathcal K_\gamma}
\{G_{J,p,1/2}(t_k)-\beta_k(x)\}\right|
\le\frac{C_rJ_n}{\sqrt p}
\right\}.
\end{aligned}
\label{eq:matched-decoupling-band-inclusion}
\end{equation}
\endgroup
To control the right-hand side of \eqref{eq:matched-decoupling-band-inclusion},
we apply Proposition~\ref{prop:quadratic-bridge-anticoncentration}
to the residual bridge $\mathbf R_J$.
Let
\begingroup
\[
\begin{aligned}
a_J&:=\frac{\sqrt{2\Tr(\bOme_{R,J^c\mid J}^2)}}{\sqrt p\,\omega_p}
=\left\{\frac{\Tr(\bOme_{R,J^c\mid J}^2)}{\Tr(\bOme^2)}\right\}^{1/2}\asymp1,\\
G_{J,p,1/2}(t_k)&=a_J Q_{\mathbf R_J}(t_k).
\end{aligned}
\]
\endgroup
and hence
\[
\max_{k\in\mathcal K_\gamma}
\{G_{J,p,1/2}(t_k)-\beta_k(x)\}
=a_J\max_{k\in\mathcal K_\gamma}
\left\{Q_{\mathbf R_J}(t_k)-\frac{\beta_k(x)}{a_J}\right\}.
\]
Taking 
\begingroup
\(
\mathcal I_n=\mathcal K_\gamma,
c_k=\frac{\beta_k(x)}{a_J},
\epsilon=\frac{C_rJ_n}{a_J\sqrt p}
\)
\endgroup
in Proposition~\ref{prop:quadratic-bridge-anticoncentration}, 
with \eqref{eq:matched-decoupling-band-inclusion}, we have
\begingroup
\[
\sup_x\max_{|J|=r}
\{\Pr(A_J^+)-\Pr(A_J^-)\}
\lesssim\left\{\frac{C_rJ_n\sqrt{K_n}}{\sqrt p}
+\frac{K_n}{\log n}
+\left(\frac{J_n^7}p\right)^{1/6}\right\}\to0,
\]
because $q-r\sim p$ and
\[
\frac{J_n^7}p\le C\left\{\frac{K_n^7}p+\frac{(\log p)^7}p\right\}\to0,
\qquad
\frac{J_n\sqrt{K_n}}{\sqrt p}
\le\frac{J_n^{3/2}}{\sqrt p}\to0.
\]
\endgroup
This proves the second result in \eqref{eq:matched-decoupling-error-targets}.
Thus, \eqref{eq:matched-decoupling-zeta-vanishing} is proved.
\begingroup
Together with \eqref{eq:matched-decoupling-bonferroni} and
\eqref{eq:matched-decoupling-factorial-limit}, this gives, for fixed $x,y$,
\begin{equation}
\begin{aligned}
\limsup_{n,p\to\infty}
\big|\Pr(A_{\gamma}\cap E_{\gamma})
-\Pr(A_{\gamma})\Pr(E_{\gamma})\big|
\le\frac{e^{-(2m+1)y}}{(2m+1)!}\xrightarrow[m\to\infty]{}0.
\end{aligned}
\label{eq:matched-decoupling-residual-factorization}
\end{equation}
\endgroup

\begingroup
For the remaining terms in \eqref{eq:matched-decoupling-factorization-reduction},
we compare the dense and maximum events separately, retaining arbitrary auxiliary events.
\endgroup
Under the conditions in Lemma~\ref{lem:gaussian-decoupling-general},
\[
\begin{aligned}
\E\left[\sqrt{K_n}\max_{k\in\mathcal K_\gamma}
\frac{|R_{h,\gamma,n}^{G}(k)|}{v_k^{1-2\gamma}}\right]
&=\frac{2}{\omega_p}\sqrt{\frac{K_n}{p}}\,
\E\max_{k\in\mathcal K_\gamma}
\frac{|\mathbf h_k^\top\mathbf B_\Omega(t_k)|}{v_k}
\rightarrow0,
\end{aligned}
\]
where $\omega_p\ge\sqrt2 C_0$.
Markov's inequality therefore yields
\begin{equation}
\max_{k\in\mathcal K_\gamma}
\frac{|R_{h,\gamma,n}^{G}(k)|}{v_k^{1-2\gamma}}
=o_p(K_n^{-1/2}).
\label{eq:matched-decoupling-linear-remainder}
\end{equation}
\begingroup
For $\mathcal A\ne\varnothing$, order the coordinates as
$(\mathcal A,\mathcal A^c)$ and let
\[
\begin{aligned}
\mathbf L=\bOme_{\mathcal A^c,\mathcal A}
\bOme_{\mathcal A,\mathcal A}^{-1},\qquad
\mathbf H_{\mathcal A}
=\begin{pmatrix}
\mathbf I-\mathbf L^\top\mathbf L\mathbf L^\top\\
\mathbf L\mathbf0
\end{pmatrix},\qquad
\mathbf A_{\mathcal A}
=\frac{\bOme^{1/2}\mathbf H_{\mathcal A}\bOme^{1/2}}
{\sqrt p\,\omega_p}.
\end{aligned}
\]
For $\mathcal A=\varnothing$, let
$\mathbf L=\mathbf0$ and
$\mathbf H_{\mathcal A}=\mathbf A_{\mathcal A}=\mathbf0$.
Since
\[
\|\mathbf L\|_{\mathrm{op}}\le C_1/C_0,\qquad
\operatorname{rank}(\mathbf H_{\mathcal A})\le2|\mathcal A|,
\qquad \|\mathbf H_{\mathcal A}\|_{\mathrm{op}}\le C,
\]
\eqref{eq:matched-decoupling-regression} gives
\[
\begin{aligned}
G_{p,\gamma}(t_k)-G_{R,p,\gamma}(t_k)
&=v_k^{-2\gamma}
\{\mathbf B(t_k)^\top\mathbf A_{\mathcal A}\mathbf B(t_k)
-v_k\Tr(\mathbf A_{\mathcal A})\},\\
\|\mathbf A_{\mathcal A}\|_{\mathrm F}
&\le C\sqrt{|\mathcal A|/p},\qquad
\|\mathbf A_{\mathcal A}\|_{\mathrm{op}}\le C/\sqrt p.
\end{aligned}
\]

\noindent\textit{(i)} Let $\gamma=0$, $\theta=\tau/n$, and
\[
\mu_n=\frac{n\|\bdelta\|_2^2}{\sqrt p\,\omega_p},\qquad
d_n(t)=\mu_nK_B(t,\theta)^2,
\qquad
\mathfrak d_{0,n}^{\mathrm{det}}(k)=d_n(t_k).
\]
For $0\le t\le1$, let
\[
G_{p,0}(t)
=\frac{\|\mathbf B_\Omega(t)\|_2^2-v(t)\Tr(\bOme)}
{\sqrt p\,\omega_p},\qquad
G_{R,p,0}(t)
=\frac{\|\mathbf R_{\mathcal A}(t)\|_2^2-v(t)\Tr(\bOme_R)}
{\sqrt p\,\omega_p}.
\]
Then
\[
G_{p,0}(t)-G_{R,p,0}(t)
=\mathbf B(t)^\top\mathbf A_{\mathcal A}\mathbf B(t)
-v(t)\Tr(\mathbf A_{\mathcal A}).
\]

For a symmetric matrix $\mathbf A$, the joint logit representation, with
$t(u)=e^u/(1+e^u)$, gives
\[
\begin{aligned}
\mathbf B(t(u))^\top\mathbf A\mathbf B(t(u))
-v(t(u))\Tr(\mathbf A)
&=v(t(u))Q_{\mathbf A}(u),\\
v(t(u))&=\frac{e^u}{(1+e^u)^2}\le e^{-|u|}.
\end{aligned}
\]
Integrating \eqref{eq:ou-quadratic-maximal-bound} with $T=1$ yields
\[
\left\|\sup_{j\le u\le j+1}|Q_{\mathbf A}(u)|\right\|_2
\le C(\|\mathbf A\|_{\mathrm F}+\|\mathbf A\|_{\mathrm{op}}),
\qquad j\in\mathbb Z.
\]
Minkowski's inequality gives
\begin{equation}
\begin{aligned}
&\left\|\sup_{0\le t\le1}
|\mathbf B(t)^\top\mathbf A\mathbf B(t)-v(t)\Tr(\mathbf A)|\right\|_2\\
&\quad\le
\sum_{j\in\mathbb Z}e^{1-|j|}
\left\|\sup_{j\le u\le j+1}|Q_{\mathbf A}(u)|\right\|_2
\le C(\|\mathbf A\|_{\mathrm F}+\|\mathbf A\|_{\mathrm{op}}).
\end{aligned}
\label{eq:compact-bridge-quadratic-maximal-bound}
\end{equation}
Taking $\mathbf A=\mathbf A_{\mathcal A}$ gives
\[
\left\|\sup_{0\le t\le1}|G_{p,0}(t)-G_{R,p,0}(t)|\right\|_2
\le C\sqrt{|\mathcal A|/p}=o(1).
\]
Together with \eqref{eq:matched-decoupling-linear-remainder}, this implies
\begin{equation}
\begin{aligned}
|\mathcal X_{n,0}^{B}-X_{R,n,0}^{G}|
&\le\max_{k\in\mathcal K_0}|G_{p,0}(t_k)-G_{R,p,0}(t_k)|\\
&\quad+\frac{2}{\sqrt p\,\omega_p}
\max_{k\in\mathcal K_0}|\mathbf h_k^\top\mathbf B_\Omega(t_k)|
=o_p(1).
\end{aligned}
\label{eq:joint-alter-unweighted-quadratic-approximation}
\end{equation}
The spectral bounds on $\bOme_R$,
$\Tr(\bOme_R^2)/\Tr(\bOme^2)\to1$, and
Lemma~\ref{lem:ideal-gaussian-dense} imply
$G_{R,p,0}\Rightarrow V$ in $D[0,1]$.
Since $V$ is continuous, the Skorokhod representation gives a coupling
such that $\|G_{R,p,0}-V\|_\infty=o_p(1)$.
For $t\in[0,1]$, define
\[
k(t)=
\begin{cases}
\max\{1,\lceil nt\rceil\},&t\le\theta,\\
\min\{n-1,\lfloor nt\rfloor\},&t>\theta.
\end{cases}
\]
Since $\theta=\tau/n$ with $1\le\tau<n$,
\[
\begin{gathered}
k(t)\in\mathcal K_0,\qquad
|t-t_{k(t)}|\le n^{-1},\qquad
d_n(t)\le d_n(t_{k(t)}),\\
0\le\sup_t\{V(t)+d_n(t)\}
-\max_{k\in\mathcal K_0}\{V(t_k)+d_n(t_k)\}
\le\sup_{|s-t|\le n^{-1}}|V(s)-V(t)|=o_p(1).
\end{gathered}
\]
Consequently, on this coupling,
\[
X_{R,n,0}^{G}=\sup_{0\le t\le1}\{V(t)+d_n(t)\}+o_p(1).
\]
As identities in distribution of processes,
\[
\begin{aligned}
V(t)&\overset d=(1-t)^2W\!\left(\frac{t^2}{(1-t)^2}\right),
\qquad 0\le t<1,\\
V\!\left(\frac{\sqrt r}{1+\sqrt r}\right)
+d_n\!\left(\frac{\sqrt r}{1+\sqrt r}\right)
&\overset d=\frac{W(r)+H_n(r)}{g(r)},
\end{aligned}
\]
where $W$ is standard Brownian motion,
$g(r)=(1+\sqrt r)^2$, and
$H_n(r)=\mu_n\min\{(1-\theta)^2r,\theta^2\}$ is nondecreasing.
Let $Z_n(t)=V(t)+d_n(t)$. For $x>0$ and
$0<\eta<\min(1,x/2)$, define
\[
\begin{aligned}
E_+&=\{|\sup_tZ_n(t)-x|\le\eta,\quad
\sup_{0\le t\le1/2}Z_n(t)\ge x-\eta\},\\
R&=\inf\{r\in[0,1]:(W(r)+H_n(r))/g(r)\ge x-\eta\}\wedge1,
\end{aligned}
\]
with $\inf\varnothing=\infty$.
On $E_+$, $W(R)+H_n(R)=(x-\eta)g(R)$ and, for $0<u\le1$,
\[
\begin{aligned}
W(R+u)-W(R)
&\le2\eta g(R)+(x+\eta)\{g(R+u)-g(R)\}\\
&\le8\eta+3(x+1)\sqrt u.
\end{aligned}
\]
The strong Markov property, together with
$\{V(1-t)\}\overset d=\{V(t)\}$ and
$K_B(1-t,\theta)=K_B(t,1-\theta)$, implies
\[
\begin{aligned}
\Pr(|\sup_tZ_n(t)-x|\le\eta)
&\le2\Pr\{W(u)\le8\eta+3(x+1)\sqrt u
\text{ for all }0<u\le1\}\\
&\xrightarrow[\eta\downarrow0]{}0.
\end{aligned}
\]
The right-hand side is independent of $\mu_n\ge0$ and
$\theta\in(0,1)$; its limit is zero because the Brownian law of the
iterated logarithm gives
\[
\limsup_{u\downarrow0}\frac{W(u)}{\sqrt u}=\infty
\quad\text{almost surely}.
\]
For $x\le0$, $d_n(t)\ge0$ gives
\[
\Pr(|\sup_tZ_n(t)-x|\le\eta)
\le\Pr(\sup_tV(t)\le\eta)
\xrightarrow[\eta\downarrow0]{}\Pr(\sup_tV(t)\le0)=0.
\]
Thus, for every fixed $x\in\mathbb R$,
\begin{equation}
\lim_{\eta\downarrow0}\limsup_{n,p\to\infty}
\Pr(|X_{R,n,0}^{G}-x|\le\eta)=0.
\label{eq:matched-decoupling-unweighted-structured-continuity}
\end{equation}
For every $\eta>0$,
\[
\begin{aligned}
&\sup_E|\Pr(\mathcal X_{n,0}^{B}\le x,E)
-\Pr(X_{R,n,0}^{G}\le x,E)|\\
&\quad\le
\Pr(|\mathcal X_{n,0}^{B}-X_{R,n,0}^{G}|>\eta)
+\Pr(|X_{R,n,0}^{G}-x|\le\eta).
\end{aligned}
\]
Taking first $n,p\to\infty$ and then $\eta\downarrow0$ gives
\begin{equation}
\sup_E|\Pr(\mathcal X_{n,0}^{B}\le x,E)
-\Pr(X_{R,n,0}^{G}\le x,E)|\to0.
\label{eq:matched-decoupling-structured-event-replacement}
\end{equation}

\noindent\textit{(ii)} Let $\gamma=1/2$.
Equation~\eqref{eq:ou-quadratic-maximal-bound} gives
\[
\begin{aligned}
\max_{k\in\mathcal K_\gamma}
|G_{p,\gamma}(t_k)-G_{R,p,\gamma}(t_k)|
&=O_p\left(\sqrt{\frac{|\mathcal A|K_n}{p}}
+\frac{K_n}{\sqrt p}\right),\\
\sqrt{K_n}\max_{k\in\mathcal K_\gamma}
|G_{p,\gamma}(t_k)-G_{R,p,\gamma}(t_k)|
&=O_p\left(\sqrt{\frac{|\mathcal A|K_n^2}{p}}
+\frac{K_n^{3/2}}{\sqrt p}\right)=o_p(1).
\end{aligned}
\]
By \eqref{eq:matched-decoupling-linear-remainder}, for every fixed $\eta>0$,
\[
\rho_{X,n,\gamma}(\eta):=
\Pr\left(\sqrt{K_n}\max_{k\in\mathcal K_\gamma}
\{|G_{p,\gamma}(t_k)-G_{R,p,\gamma}(t_k)|
+|R_{h,\gamma,n}^{G}(k)|\}>\eta\right)\to0.
\]
With $\beta_k(x)$ from \eqref{eq:matched-decoupling-dense-thresholds}, let
\[
A_{\gamma,n}^{\pm}(x;\eta)
=\left\{G_{R,p,1/2}(t_k)\le\beta_k(x)\pm\frac{\eta}{\sqrt{K_n}},
\ \forall k\in\mathcal K_\gamma\right\}.
\]
Then $A_{\gamma,n}^{-}(x;\eta)\subseteq A_\gamma
\subseteq A_{\gamma,n}^{+}(x;\eta)$ and, for every event $E$,
\begin{equation}
\begin{aligned}
&\Pr(A_{\gamma,n}^{-}(x;\eta)\cap E)-\rho_{X,n,\gamma}(\eta)\\
&\quad\le\Pr(\mathcal X_{n,\gamma}^{B}\le x,E)
\le\Pr(A_{\gamma,n}^{+}(x;\eta)\cap E)+\rho_{X,n,\gamma}(\eta).
\end{aligned}
\label{eq:matched-decoupling-dense-threshold-sandwich}
\end{equation}
Let
\[
a_R=\left\{\frac{\Tr(\bOme_R^2)}{\Tr(\bOme^2)}\right\}^{1/2}\to1,
\qquad G_{R,p,1/2}(t_k)=a_RQ_{\mathbf R_{\mathcal A}}(t_k).
\]
Taking $\mathbf B_\Sigma=\mathbf R_{\mathcal A}$,
$\mathcal I_n=\mathcal K_\gamma$, $c_k=\beta_k(x)/a_R$, and
$\epsilon=\eta/(a_R\sqrt{K_n})$ in
Proposition~\ref{prop:quadratic-bridge-anticoncentration} gives
\begin{equation}
\begin{aligned}
\sup_x\{\Pr(A_{\gamma,n}^{+}(x;\eta))
-\Pr(A_{\gamma,n}^{-}(x;\eta))\}
&\le C\eta+C\left\{\frac{K_n}{\log n}
+\left(\frac{J_n^7}{p}\right)^{1/6}\right\}.
\end{aligned}
\label{eq:matched-decoupling-dense-event-comparison}
\end{equation}
Consequently,
\[
\begin{aligned}
\sup_x\sup_E
|\Pr(\mathcal X_{n,\gamma}^{B}\le x,E)-\Pr(A_\gamma\cap E)|
&\le\rho_{X,n,\gamma}(\eta)
+\sup_x\{\Pr(A_{\gamma,n}^{+}(x;\eta))
-\Pr(A_{\gamma,n}^{-}(x;\eta))\}\\
&\le C\eta+o(1).
\end{aligned}
\]
Taking $n,p\to\infty$ and then $\eta\downarrow0$ yields
\begin{equation}
\sup_x\sup_E
|\Pr(\mathcal X_{n,1/2}^{B}\le x,E)-\Pr(A_{1/2}\cap E)|\to0.
\label{eq:matched-decoupling-dense-event-replacement}
\end{equation}
\endgroup

\begingroup

For both $\gamma\in\{0,1/2\}$, the maximum-event comparison is as follows:
\endgroup
\[
\Var\left\{v_k^{-\gamma}\frac{B_{\Omega,j}(t_k)}{\sigma_j}\right\}
=\Var\left\{v_k^{-\gamma}
\frac{R_{\mathcal A,j}(t_k)}{\sqrt{(\bOme_R)_{jj}}}\right\}
=v_k^{1-2\gamma}\le1.
\]
By the Gaussian tail inequality and the union bound, 
for $u>0$,
\begingroup
\[
\begin{aligned}
\Pr\left[\max_{k\in\mathcal K_\gamma,\,j\in\mathcal A^c}
v_k^{-\gamma}\left|
\frac{B_{\Omega,j}(t_k)}{\sigma_j}\right|>u\right]
&\le\sum_{k\in\mathcal K_\gamma}\sum_{j\in\mathcal A^c}
\Pr\left\{v_k^{-\gamma}\left|\frac{B_{\Omega,j}(t_k)}{\sigma_j}\right|>u\right\}\\
&\le2np\exp(-u^2/2).
\end{aligned}
\]
\endgroup
Taking $u=C\sqrt{L_n}$ with $C^2/2>1$ gives
$2np\exp(-u^2/2)=2(np)^{1-C^2/2}\to0$. {Applying the same Gaussian union bound to $\max_{k\in\mathcal K_\gamma,\,j\in\mathcal A^c}
v_k^{-\gamma}\frac{|R_{\mathcal A,j}(t_k)|}{\sqrt{(\bOme_R)_{jj}}}$ gives}
\begin{equation}
\max_{k\in\mathcal K_\gamma,\,j\in\mathcal A^c}
v_k^{-\gamma}\frac{|B_{\Omega,j}(t_k)|}{\sigma_j}
+\max_{k\in\mathcal K_\gamma,\,j\in\mathcal A^c}
v_k^{-\gamma}\frac{|R_{\mathcal A,j}(t_k)|}{\sqrt{(\bOme_R)_{jj}}}
=O_p(\sqrt{L_n}).
\label{eq:matched-decoupling-central-maxima-order}
\end{equation}
Furthermore, for any $j\in \mathcal A^c$, let %
\(
V_j(t)=B_{\Omega,j}(t)-R_{\mathcal A,j}(t)
=\bOme_{j,\mathcal A}\bOme_{\mathcal A,\mathcal A}^{-1}
\mathbf B_{\Omega,\mathcal A}(t).
\)
Then, we have
\[
\Var\left\{v_k^{-\gamma}\frac{V_j(t_k)}{\sigma_j}\right\}=v_k^{1-2\gamma}\frac{\bOme_{j,\mathcal A}\bOme_{\mathcal A,\mathcal A}^{-1}
\bOme_{\mathcal A,j}}{(\bOme)_{jj}}
\le\chi_{\mathcal A,p},\qquad
0\le1-\frac{(\bOme_R)_{jj}}{(\bOme)_{jj}}\le\chi_{\mathcal A,p}.
\]
A Gaussian union bound gives
\begingroup
\begin{align}
&\max_{k\in\mathcal K_\gamma,\,j\in\mathcal A^c}
v_k^{-\gamma}\left|
\frac{B_{\Omega,j}(t_k)}{\sigma_j}
-\frac{R_{\mathcal A,j}(t_k)}{\sqrt{(\bOme_R)_{jj}}}
\right|\nonumber\\
&\quad\le \max_{k\in\mathcal K_\gamma,\,j\in\mathcal A^c}v_k^{-\gamma}\frac{|V_j(t_k)|}{\sigma_j}
+\chi_{\mathcal A,p}\max_{k\in\mathcal K_\gamma,\,j\in\mathcal A^c}v_k^{-\gamma}\frac{|R_{\mathcal A,j}(t_k)|}{\sqrt{(\bOme_R)_{jj}}}\nonumber\\
&\quad=O_p\{(\sqrt{\chi_{\mathcal A,p}}+\chi_{\mathcal A,p})\sqrt{L_n}\}.
\label{eq:matched-decoupling-standardized-process-error}
\end{align}
\endgroup
For $\gamma=0$, use $|u^2-v^2|\le(u+v)|u-v|$ for nonnegative $u,v$;
for $\gamma=1/2$, use $A_{\mathrm{DE}}(p\log h_n)=O(\sqrt{L_n})$.
Combining \eqref{eq:matched-decoupling-central-maxima-order}--\eqref{eq:matched-decoupling-standardized-process-error}, with
$\chi_{\mathcal A,p}L_n^2\to0$, we obtain
\[
|Y_{\mathcal A^c,n,\gamma}^{G}-Y_{R,n,\gamma}^{G}|
=O_p\{(\sqrt{\chi_{\mathcal A,p}}+\chi_{\mathcal A,p})L_n\}
=o_p(1).
\]
Thus, for every fixed $\eta>0$, 
\[
\rho_{Y,n,\gamma}(\eta)
:=\Pr\left(\left|Y_{\mathcal A^c,n,\gamma}^{G}-Y_{R,n,\gamma}^{G}\right|>\eta\right)
\to0.
\]
For any event $E$,
\begin{equation}
\begin{aligned}
&\Pr\left(E\cap D_{\gamma}\cap\{Y_{R,n,\gamma}^{G}\le y-\eta\}\right)
-\rho_{Y,n,\gamma}(\eta)\\
&\le\Pr(\mathcal Y_{n,\gamma}^{B}\le y,\ E)
\le\Pr\left(E\cap D_{\gamma}\cap\{Y_{R,n,\gamma}^{G}\le y+\eta\}\right)
+\rho_{Y,n,\gamma}(\eta).
\end{aligned}
\label{eq:matched-decoupling-max-event-comparison}
\end{equation}
Since $E_{\gamma}=D_{\gamma}\cap\{Y_{R,n,\gamma}^{G}\le y\}$,
this and \eqref{eq:matched-decoupling-residual-gumbel} imply, for each fixed $\eta>0$,
\[
\begin{aligned}
\sup_E\left|\Pr(\mathcal Y_{n,\gamma}^{B}\le y,\ E)
-\Pr(E_{\gamma}\cap E)\right|
&\le\rho_{Y,n,\gamma}(\eta)
+\Pr\{y-\eta<Y_{R,n,\gamma}^{G}\le y+\eta\}\\
&=\rho_{Y,n,\gamma}(\eta)+F_{\mathrm{Gu}}(y+\eta)-F_{\mathrm{Gu}}(y-\eta)+o(1).
\end{aligned}
\]
Taking first $n,p\to\infty$ and then $\eta\downarrow0$, we have
\begin{equation}
\sup_E\left|\Pr\{\mathcal Y_{n,\gamma}^{B}\le y,\ E\}
-\Pr\{E_{\gamma}\cap E\}\right|\to0.
\label{eq:matched-decoupling-max-event-replacement}
\end{equation}

\begingroup
Equations~\eqref{eq:matched-decoupling-structured-event-replacement},
\eqref{eq:matched-decoupling-dense-event-replacement}, and
\eqref{eq:matched-decoupling-max-event-replacement} give, for each fixed
$x,y\in\mathbb R$ and $\gamma\in\{0,1/2\}$,
\[
\begin{aligned}
&\sup_E|\Pr(\mathcal X_{n,\gamma}^{B}\le x,E)-\Pr(A_\gamma\cap E)|\\
&\quad+\sup_E|\Pr(\mathcal Y_{n,\gamma}^{B}\le y,E)
-\Pr(E_\gamma\cap E)|\to0.
\end{aligned}
\]
Together with \eqref{eq:matched-decoupling-residual-factorization},
\[
\begin{aligned}
&|\Pr(\mathcal X_{n,\gamma}^{B}\le x,\mathcal Y_{n,\gamma}^{B}\le y)
-\Pr(\mathcal X_{n,\gamma}^{B}\le x)\Pr(\mathcal Y_{n,\gamma}^{B}\le y)|\\
&\quad\le|\Pr(A_\gamma\cap E_\gamma)-\Pr(A_\gamma)\Pr(E_\gamma)|\\
&\qquad+2\sup_E|\Pr(\mathcal X_{n,\gamma}^{B}\le x,E)-\Pr(A_\gamma\cap E)|\\
&\qquad+2\sup_E|\Pr(\mathcal Y_{n,\gamma}^{B}\le y,E)
-\Pr(E_\gamma\cap E)|\to0.
\end{aligned}
\]
\endgroup
\end{proof}

\begin{proof}[Proof of Theorem~\ref{indnull}]
\begingroup
We transfer the population statistics to the Gaussian bridge, factorize
the Gaussian probabilities, and then replace the population quantities
by their feasible counterparts.
Let $L_n=\log(np)$, $\beta_n=L_n^{7/4}$, and fix
\endgroup
$\gamma\in\{0,1/2\}$.  Recall
$S_{n,\gamma}^{\epsilon}$, $S_{n,\gamma}^{G}$, and
$S_{n,\gamma}^{B}$ from \eqref{eq:def-dense-max-statistics}, and
$M_{n,p}^{\epsilon,\gamma}$, $M_{n,p}^{G,\gamma}$, and
$M_{n,p}^{B,\gamma}$ from
\eqref{eq:def-max-coordinate-statistics}.
\begingroup
Let $\mathsf N_{n,\gamma}^{\mathrm S}$ and
$\mathsf N_{n,\gamma}^{\mathrm M}$ be the normalizations defined
before Lemma~\ref{lem:gaussian-decoupling-general} and in
\eqref{eq:max-normalizations}, respectively,
\endgroup
and set
\[
\mathbf T_{n,\gamma}^{\diamond}
=\left\{
\mathsf N_{n,\gamma}^{\mathrm S}(S_{n,\gamma}^{\diamond}),
\mathsf N_{n,\gamma}^{\mathrm M}(M_{n,p}^{\diamond,\gamma})
\right\}^{\top}
=:
(
T_{n,\gamma,\mathrm S}^{\diamond},
T_{n,\gamma,\mathrm M}^{\diamond}
)^{\top},
\qquad \diamond\in\{\epsilon,G,B\}.
\]

Step 1.
\begingroup
We compare the population pair with its Gaussian bridge counterpart:
\endgroup
\begin{equation}
\sup_{\substack{\|\psi\|_\infty\le1,
\operatorname{Lip}(\psi)\le1}}
\left|
\E\psi(\mathbf T_{n,\gamma}^{\epsilon})
-\E\psi(\mathbf T_{n,\gamma}^{B})
\right|\rightarrow0.
\label{eq:joint-null-population-transfer}
\end{equation}
Fix such a $\psi$ and fixed $0<\eta\le 1,x_0>0$.  Let
$\psi_\eta\in C_b^3(\mathbb R^2)$ be a mollification satisfying
\[
\|\psi_\eta-\psi\|_\infty\le C\eta,
\qquad
\|\psi_\eta\|_{C^3}\le C\eta^{-3}.
\]
Let $\mathfrak t_{n,2x_0}\in C^3(\mathbb R)$ be the clipping map from
\eqref{eq:max-transfer-clipping-map} in Lemma~\ref{lem:max-transfer-general} with $C_{\mathrm{clip}}=2x_0$; that is,
\[
\begin{aligned}
\mathfrak t_{n,2x_0}(y)=y,
\qquad |y|\le2x_0\sqrt{L_n},
\qquad
\|\mathfrak t_{n,2x_0}\|_\infty\le4x_0\sqrt{L_n},
\qquad
\max_{1\le r\le3}\|\mathfrak t_{n,2x_0}^{(r)}\|_\infty\le C.
\end{aligned}
\]
Define
\[
\begin{aligned}
\overline{\mathbf T}_{n,\gamma}^{\epsilon}
&=\left[
\mathsf N_{n,\gamma}^{\mathrm S}
 \{\operatorname{smax}_{\beta_n}(\omega_p^{-1}\mathcal Q_{\gamma,n})\},
\mathsf N_{n,\gamma}^{\mathrm M}
 \{\mathfrak t_{n,2x_0}(\operatorname{smax}_{\beta_n}(\mathcal C_{\gamma,n}^{\pm}))\}
\right]^{\top},\\
\overline{\mathbf T}_{n,\gamma}^{G}
&=\left[
\mathsf N_{n,\gamma}^{\mathrm S}
 \{\operatorname{smax}_{\beta_n}(\omega_p^{-1}\mathcal Q_{\gamma,n}^{G})\},
\mathsf N_{n,\gamma}^{\mathrm M}
 \{\mathfrak t_{n,2x_0}(\operatorname{smax}_{\beta_n}((\mathcal C_{\gamma,n}^{G})^{\pm}))\}
\right]^{\top},\\
\overline{\mathbf T}_{n,\gamma}^{B}
&=\left[
\mathsf N_{n,\gamma}^{\mathrm S}
 \{\operatorname{smax}_{\beta_n}(\mathcal Q_{\gamma,n}^{B})\},
\mathsf N_{n,\gamma}^{\mathrm M}
 \{\mathfrak t_{n,2x_0}(\operatorname{smax}_{\beta_n}((\mathcal C_{\gamma,n}^{B})^{\pm}))\}
\right]^{\top}.
\end{aligned}
\]
Define
\(
f_{\eta,n,\gamma}(u,v)
=\psi_\eta\!\left[
\mathsf N_{n,\gamma}^{\mathrm S}(u),
\mathsf N_{n,\gamma}^{\mathrm M}\{\mathfrak t_{n,2x_0}(v)\}
\right].
\)
Moreover,
\[
\left|(\mathsf N_{n,\gamma}^{\mathrm S})'\right|
\le C(\log L_n)^{1/2},
\qquad
\max_{1\le r\le3}
\left\|\left(\mathsf N_{n,\gamma}^{\mathrm M}
\circ\mathfrak t_{n,2x_0}\right)^{(r)}\right\|_\infty
\le C_{x_0}\sqrt{L_n}.
\]
Hence, by the chain rule,
\(
\|f_{\eta,n,\gamma}\|_{C^3}
\le C\eta^{-3}L_n^{3/2}.
\)
Take $c_{\mathrm{out}}=3/2$ in
Lemma~\ref{lem:block-comparison-general}, it follows that
$L_n^{3/2}r_{n,\gamma}\to0$.
Moreover, since $L_n=O(n^\nu)$ and
$\Lambda_{\gamma,n}\asymp n^{\lambda_\gamma}$,
\(
L_n^{9/2}\Lambda_{\gamma,n}^{-1}
=O\!\left(n^{9\nu/2-\lambda_\gamma}\right)=o(1),
\)
\begingroup
because $\nu<\lambda_\gamma\nu_{\mathrm{dep}}
\le\lambda_\gamma/28<2\lambda_\gamma/9$.
\endgroup
Applying Lemmas~\ref{lem:block-comparison-general}--
\ref{lem:joint-gaussian-cov-restoration} with
$\mathbf h=\mathbf0$ and $\mathbf d=\mathbf0$ therefore gives,
for fixed $\eta$ and $x_0$,
\begin{align}
\left|
\E\psi_\eta(\overline{\mathbf T}_{n,\gamma}^{\epsilon})
-\E\psi_\eta(\overline{\mathbf T}_{n,\gamma}^{B})
\right|
\le C\eta^{-3}L_n^{3/2}
\left\{r_{n,\gamma}
+L_n^{3}\Lambda_{\gamma,n}^{-1}\right\}=o(1),
\label{eq:joint-null-data-to-gaussian}
\end{align}

For $\diamond\in\{\epsilon,G,B\}$, set
\(
\mathcal A_{n,x_0}^{\diamond}
=\{M_{n,p}^{\diamond,\gamma}\le x_0\sqrt{L_n}\}.
\)
Lemma~\ref{lem:max-transfer-general},
Lemma~\ref{lem:gaussian-extremes-general}, and the Gaussian union bound give
\begin{equation}
\lim_{x_0\to\infty}\limsup_{n,p\to\infty}
\max_{\diamond\in\{\epsilon,G,B\}}
\Pr\{(\mathcal A_{n,x_0}^{\diamond})^c\}=0.
\label{eq:joint-null-clipping-probability}
\end{equation}
On $\cap_{\diamond}\mathcal A_{n,x_0}^{\diamond}$,
\eqref{eq:smax-logsumexp-bound} and the mean-value theorem yield
\begingroup
\begin{align}
&\max\Bigl\{
\|\overline{\mathbf T}_{n,\gamma}^{\epsilon}
-\mathbf T_{n,\gamma}^{\epsilon}\|_2,
\|\overline{\mathbf T}_{n,\gamma}^{G}
-\mathbf T_{n,\gamma}^{G}\|_2,
\|\overline{\mathbf T}_{n,\gamma}^{B}
-\mathbf T_{n,\gamma}^{B}\|_2
\Bigr\}\notag\\
&\quad\quad\quad\le C\left\{
\frac{\sqrt{\log L_n}}{L_n^{3/4}}
+\frac1{L_n^{1/4}}
\right\}=o(1).
\label{eq:joint-null-remove-smoothing}
\end{align}
\endgroup
Combining
\eqref{eq:joint-null-data-to-gaussian}--\eqref{eq:joint-null-remove-smoothing},
\[
\limsup_{n,p\to\infty}
\left|
\E\psi(\mathbf T_{n,\gamma}^{\epsilon})
-\E\psi(\mathbf T_{n,\gamma}^{B})
\right|
\le C\eta
+C\limsup_{n,p\to\infty}
\sum_{\diamond\in\{\epsilon,G,B\}}
\Pr\{(\mathcal A_{n,x_0}^{\diamond})^c\}.
\]
Letting $x_0\to\infty$ and then $\eta\downarrow0$ proves
\eqref{eq:joint-null-population-transfer}.

Step 2.
By Lemma~\ref{lem:gaussian-decoupling-general} with
$\mathbf h_k=\mathbf0$ for every $k$ {(take
$\bdelta=\mathbf0$ in part~(i))}, we obtain, for every
$x,y\in\mathbb R$,
\[
\Pr\{T_{n,\gamma,\mathrm S}^{B}\le x,
T_{n,\gamma,\mathrm M}^{B}\le y\}
-\Pr\{T_{n,\gamma,\mathrm S}^{B}\le x\}
 \Pr\{T_{n,\gamma,\mathrm M}^{B}\le y\}
\rightarrow0.
\]
By Lemmas~\ref{lem:ideal-gaussian-dense},
\ref{lem:dense-boundary-extreme}, 
\ref{lem:gaussian-extremes-general}, \eqref{eq:joint-null-population-transfer} and the bounded-Lipschitz
characterization of weak convergence,
\begin{equation}
\begin{gathered}
\Pr\{
T_{n,0,\mathrm S}^{\epsilon}\le x,\,
T_{n,0,\mathrm M}^{\epsilon}\le y\}
\rightarrow F_V(x)F_{\mathrm{Gu}}(y),
\\
\Pr\{
T_{n,1/2,\mathrm S}^{\epsilon}\le x,\,
T_{n,1/2,\mathrm M}^{\epsilon}\le y\}
\rightarrow F_{\mathrm{Gu}}(x)F_{\mathrm{Gu}}(y).
\end{gathered}
\label{eq:joint-null-population-limits}
\end{equation}

Step 3.
Recall $S_0^{(1)}$ from
\eqref{eq:def-first-order-unweighted-statistic}.
For $\gamma=0$, since $\mathsf N_{n,0}^{\mathrm S}(x)=x$,
\begin{align*}
|S_0-S_{n,0}^{\epsilon}|
&\le
|S_0-S_0^{(1)}|
+
|S_0^{(1)}-S_0^\dagger|
+
|S_0^\dagger-S_{n,0}^{\epsilon}|.
\end{align*}
By \eqref{eq:S12-first-order-null}, \eqref{eq:Th3-feasible-rates}, 
\eqref{eq:Th3-unweighted-population} and \eqref{eq:lem14-mu-Mk},
\[
|S_0-S_{n,0}^{\epsilon}|
=O_p\bigl(p^{-1/2}+r_{\mu,n}+r_{\omega,n}+\sqrt p\,\eta_M\bigr)
=o_p(1).
\]

For $\gamma=1/2$, since $\mathsf N_{n,1/2}^{\mathrm S}(x)=A_{\mathrm{DE}}(\log h_n)x-D_{\mathrm{DE}}(\log h_n)$,
with
\eqref{eq:Th3-feasible-boundary-error} 
\eqref{eq:Th3-truncated-centering}, and 
\eqref{eq:Th3-strengthened-feasible-rate}, we have
\begin{align*}
A_{\mathrm{DE}}(\log h_n)
|S_{1/2}-S_{n,1/2}^{\epsilon}|
&{}\le
A_{\mathrm{DE}}(\log h_n)
|S_{1/2}-S_{1/2}^\dagger|+A_{\mathrm{DE}}(\log h_n)
|S_{1/2}^\dagger-S_{n,1/2}^{\epsilon}|\\
&{}=O_p\!\left\{
(\log\log h_n)^{1/2}r_{\mu,n}
+(\log\log h_n)r_{\omega,n}
+A_{\mathrm{DE}}(\log h_n)\sqrt p\,\eta_M
\right\}\\
&{}=o_p(1),
\end{align*}
Hence
\begin{equation}
\left|
\mathsf N_{n,\gamma}^{\mathrm S}(S_\gamma)
-\mathsf N_{n,\gamma}^{\mathrm S}(S_{n,\gamma}^{\epsilon})
\right|=o_p(1),
\qquad \gamma\in\{0,1/2\}.
\label{eq:joint-null-feasible-dense}
\end{equation}
Moreover, 
\eqref{eq:max-feasible-unweighted-error} and
\eqref{eq:max-feasible-boundary-error} give
\begin{equation}
\left|
\mathsf N_{n,\gamma}^{\mathrm M}(M_\gamma)
-\mathsf N_{n,\gamma}^{\mathrm M}(M_{n,p}^{\epsilon,\gamma})
\right|=o_p(1),
\qquad \gamma\in\{0,1/2\}.
\label{eq:joint-null-feasible-max}
\end{equation}
Combining
\eqref{eq:joint-null-feasible-dense}--\eqref{eq:joint-null-feasible-max} with
\eqref{eq:joint-null-population-limits} and
Slutsky's theorem proves Theorem~\ref{indnull}(i) and (ii).
\end{proof}

\begin{proof}[Proof of Proposition~\ref{prop:cauchy-size}]
\begingroup
For either branch, Theorem~\ref{indnull} and continuity of the marginal
distribution functions give
\[
(p_{S_\gamma},p_{M_\gamma})\Rightarrow(U_1,U_2),
\qquad U_1,U_2\ \text{independent }U(0,1).
\]
Let $T_j=\tan\{\pi(1/2-U_j)\}$, $j=1,2$. These variables are
independent standard Cauchy, and
\[
\E e^{it(T_1+T_2)/2}=e^{-|t|/2}e^{-|t|/2}=e^{-|t|}.
\]
The transform is continuous on $(0,1)^2$, which has limiting probability
one. The continuous mapping theorem therefore yields
\[
T_{CC,\gamma}\Rightarrow (T_1+T_2)/2\sim C(0,1),
\qquad \Pr(p_{CC,\gamma}\le\alpha)\to\alpha,
\quad 0<\alpha<1.
\]
\endgroup
\end{proof}

\begingroup
\begin{proposition}%
\label{prop:quadratic-bridge-anticoncentration}
Let $\mathbf B_\Sigma$ be a centered $d$-dimensional Gaussian bridge with
\[
\Cov\{\mathbf B_\Sigma(s),\mathbf B_\Sigma(t)\}
=K_B(s,t)\boldsymbol\Sigma,\qquad
K_B(s,t)=s\wedge t-st,
\]
where
\(
C_0\le\lambda_{\min}(\boldsymbol\Sigma)
\le\lambda_{\max}(\boldsymbol\Sigma)\le C_1
\)
\begingroup
for fixed $0<C_0\le C_1<\infty$. Let $t_k=k/n$,
$K_n=\log(2+\log n)$, and
\[
Q_{\mathbf B_\Sigma}(t)
=\frac{\|\mathbf B_\Sigma(t)\|_2^2-\E\|\mathbf B_\Sigma(t)\|_2^2}
{\sqrt{\Var(\|\mathbf B_\Sigma(t)\|_2^2)}},\qquad 0<t<1.
\]
\endgroup
There exists a constant $C$, depending only on $C_0,C_1$, such that,
for every nonempty set
\begingroup
$\mathcal I_n\subseteq\{1,\ldots,n-1\}$, every $\epsilon>0$, and all $n\ge3$, $d\ge1$,
\[
\begin{aligned}
&\sup_{\mathbf c\in\mathbb R^{|\mathcal I_n|}}
\Pr\left\{\left|\max_{k\in\mathcal I_n}
\{Q_{\mathbf B_\Sigma}(t_k)-c_k\}\right|\le\epsilon\right\}\\
&\quad\le C\epsilon\sqrt{K_n}
+C\left\{\frac{K_n}{\log n}
+\left(\frac{(K_n+\log d)^7}d\right)^{1/6}\right\},
\end{aligned}
\]
\endgroup
\end{proposition}

\begin{proof}[Proof of Proposition~\ref{prop:quadratic-bridge-anticoncentration}]
\begingroup
We discretize logit time using the smallest threshold in each bin, apply
the rectangle central limit theorem to the grid, and bound the within-bin
oscillation. With $v(t)=t(1-t)$,
\[
\E\|\mathbf B_\Sigma(t)\|_2^2=v(t)\Tr(\boldsymbol\Sigma),\qquad
\Var(\|\mathbf B_\Sigma(t)\|_2^2)=2v(t)^2\Tr(\boldsymbol\Sigma^2),
\]
and hence
\[
Q_{\mathbf B_\Sigma}(t)
=\frac{\|\mathbf B_\Sigma(t)\|_2^2-v(t)\Tr(\boldsymbol\Sigma)}
{v(t)\sqrt{2\Tr(\boldsymbol\Sigma^2)}},\qquad 0<t<1.
\]
\endgroup
Let $u=\log\{t/(1-t)\}$, equivalently $t=e^u/(1+e^u)$.
The eigenvalue decomposition of $\boldsymbol\Sigma$ gives independent centered stationary
Gaussian processes $\{\xi_\ell(u):u\in\mathbb R\}$ with covariance
$\exp(-|u-v|/2)$ such that
\[
Q_{\mathbf B_\Sigma}\!\left(\frac{e^u}{1+e^u}\right)
=\sum_{\ell=1}^{d}w_\ell\{\xi_\ell(u)^2-1\},\qquad
w_\ell=\frac{\lambda_\ell(\boldsymbol\Sigma)}
{\sqrt{2\Tr(\boldsymbol\Sigma^2)}},
\quad u\in\mathbb R.
\]
The condition
\(
C_0\le\lambda_{\min}(\boldsymbol\Sigma)
\le\lambda_{\max}(\boldsymbol\Sigma)\le C_1
\)
implies
\begingroup
\(
\max_{\ell\le d}|w_\ell|\le Cd^{-1/2},\qquad
2\sum_{\ell=1}^{d}w_\ell^2=1.
\)
\endgroup
The logit transformation $u=\log\{t/(1-t)\}$ maps the observation
times $t_k=k/n$, $1\le k<n$, to
$u_k=\log\{k/(n-k)\}\in[-\log(n-1),\log(n-1)]$.
{Partition $[-\log(n-1),\log(n-1)]$ into} $\left\lceil2\log(n-1)(\log n)^6\right\rceil$ intervals of length
$h=(\log n)^{-6}$, with the last interval possibly shorter.
Let $I_1,\ldots,I_N$ be those containing observation times $u_k=\log\{t_k/(1-t_k)\}$ for $k\in\mathcal I_n$, and choose
\begingroup
\[
\begin{aligned}
k_r&=\min\argmin_{\{k\in\mathcal I_n:\,u_k\in I_r\}}c_k,
\qquad r=1,\ldots,N,
\end{aligned}
\]
\endgroup
where 
$N\le\left\lceil2\log(n-1)(\log n)^6\right\rceil\lesssim(\log n)^7$, and $\log(2N)\lesssim K_n$.
For any $u_k\in I_r$, we have $Q_{\mathbf B_\Sigma}(t_k)-c_k
\le Q_{\mathbf B_\Sigma}(t_{k_r})-c_{k_r}
+|Q_{\mathbf B_\Sigma}(t_k)-Q_{\mathbf B_\Sigma}(t_{k_r})|$. Consequently,
\begingroup
\begin{align*}
0&\le\max_{k\in\mathcal I_n}\{Q_{\mathbf B_\Sigma}(t_k)-c_k\}
-\max_{r\le N}\{Q_{\mathbf B_\Sigma}(t_{k_r})-c_{k_r}\}\\
&\le\max_{r\le N}\sup_{u,v\in I_r}
\left|Q_{\mathbf B_\Sigma}\!\left(\frac{e^u}{1+e^u}\right)
-Q_{\mathbf B_\Sigma}\!\left(\frac{e^v}{1+e^v}\right)\right|.
\end{align*}
\endgroup
For arbitrary $\epsilon>0$, we have
\begingroup
\begin{equation}
\begin{aligned}
&\Pr\left\{\left|\max_{k\in\mathcal I_n}
\{Q_{\mathbf B_\Sigma}(t_k)-c_k\}\right|\le\epsilon\right\}\\
&\quad\le
\Pr\left\{\left|\max_{r\le N}
\{Q_{\mathbf B_\Sigma}(t_{k_r})-c_{k_r}\}\right|
\le\epsilon+\frac1{(\log n)^2}\right\}\\
&\qquad+
\Pr\left\{\max_{r\le N}\sup_{u,v\in I_r}
\left|Q_{\mathbf B_\Sigma}\!\left(\frac{e^u}{1+e^u}\right)
-Q_{\mathbf B_\Sigma}\!\left(\frac{e^v}{1+e^v}\right)\right|
>\frac1{(\log n)^2}\right\}.
\end{aligned}
\label{eq:matched-decoupling-band-discretization}
\end{equation}
\endgroup

\begingroup
To bound the grid term in \eqref{eq:matched-decoupling-band-discretization},
define the independent, not necessarily identically distributed vectors
\[
\mathbf X_\ell
=\sqrt d\,w_\ell\bigl(\xi_\ell(u_{k_r})^2-1\bigr)_{r=1}^N,
\qquad \ell=1,\ldots,d.
\]
If $N<3$, repeat coordinates and their thresholds until the dimension is
$N'=\max(N,3)$; the maximum is unchanged. Since
$\sqrt d\max_\ell|w_\ell|\le C$, the summands satisfy
\[
\begin{aligned}
\E X_{\ell r}&=0,
&d^{-1}\sum_{\ell=1}^d\E X_{\ell r}^2&=1,\\
\max_{r\le N'}d^{-1}\sum_{\ell=1}^d\E|X_{\ell r}|^{2+j}
&\le C\quad(j=1,2),
&\max_{\ell,r}\E\exp(|X_{\ell r}|/C)&\le2.
\end{aligned}
\]
Let $\mathbf Z$ be a centered Gaussian vector with covariance
$\Cov\{d^{-1/2}\sum_\ell\mathbf X_\ell\}$. Its coordinate variances
are one and its covariance entries on the original grid are
$e^{-|u_{k_r}-u_{k_s}|}$. For $d\ge4$, Proposition~2.1 of
\citet{chernozhukov2017central}, under its conditions (M.1), (M.2)
and (E.1), gives
\[
\sup_{\mathbf b\in\mathbb R^{N'}}
\left|\Pr\left\{d^{-1/2}\sum_{\ell=1}^d\mathbf X_\ell\le\mathbf b\right\}
-\Pr\{\mathbf Z\le\mathbf b\}\right|
\le C\left\{\frac{\log^7(dN')}d\right\}^{1/6}.
\]
For $d\le3$ the final bound follows by increasing $C$, since its
dimension remainder is bounded below by a positive constant.
The rectangle bound is uniform in the thresholds, including the selected
values $c_{k_r}$. Nazarov's inequality
\citep[Lemma~A.1]{chernozhukov2017central} consequently gives, for $s>0$,
\[
\begin{aligned}
&\sup_{\mathbf c}\Pr\left\{\left|\max_{r\le N}
\{Q_{\mathbf B_\Sigma}(t_{k_r})-c_{k_r}\}\right|\le s\right\}\\
&\quad\le Cs\sqrt{\log(2N)}
+C\left\{\frac{\log^7(dN')}d\right\}^{1/6}.
\end{aligned}
\]
Since $\log(2N)\le CK_n$ and $\log(dN')\le C(K_n+\log d)$,
taking $s=\epsilon+(\log n)^{-2}$ yields
\begin{equation}
\begin{aligned}
&\Pr\left\{\left|\max_{r\le N}
\{Q_{\mathbf B_\Sigma}(t_{k_r})-c_{k_r}\}\right|
\le\epsilon+(\log n)^{-2}\right\}\\
&\quad\le C\epsilon\sqrt{K_n}
+C\left\{\frac{K_n}{\log n}
+\left(\frac{(K_n+\log d)^7}d\right)^{1/6}\right\}.
\end{aligned}
\label{eq:matched-decoupling-grid-anticoncentration}
\end{equation}
\endgroup
We next control the second term of
\eqref{eq:matched-decoupling-band-discretization}.
Let $t(u)=e^u/(1+e^u)$.
For each $r\le N$, divide the closure $\overline I_r$ of $I_r$ into $2^m$ equal pieces,
with endpoints
\[
u_{r,m,j}=\inf I_r+\frac{j}{2^m}|I_r|,
\qquad 0\le j\le2^m,\quad m=0,1,\ldots.
\]
Denote the increments between adjacent endpoints by
\[
\Delta_{r,m,j}
:=Q_{\mathbf B_\Sigma}(t(u_{r,m,j}))
-Q_{\mathbf B_\Sigma}(t(u_{r,m,j-1})),
\qquad 1\le j\le2^m.
\]
For fixed $r$ and $u\in\overline I_r$, let $u^{(m)}$ be the largest
level-$m$ grid point not exceeding $u$.
Then $u^{(m)}\to u$, and $u^{(m-1)}$ and $u^{(m)}$ either coincide or
are adjacent on the level-$m$ grid.
By continuity and telescoping,
\[
\begin{aligned}
&|Q_{\mathbf B_\Sigma}(t(u))-Q_{\mathbf B_\Sigma}(t(\inf I_r))|\\
&\quad\le
|Q_{\mathbf B_\Sigma}(t(u^{(0)}))-Q_{\mathbf B_\Sigma}(t(\inf I_r))|
+\sum_{m\ge1}
|Q_{\mathbf B_\Sigma}(t(u^{(m)}))-Q_{\mathbf B_\Sigma}(t(u^{(m-1)}))|\\
&\quad\le\sum_{m\ge0}\max_{1\le j\le2^m}|\Delta_{r,m,j}|.
\end{aligned}
\]
Applying this bound to both $u$ and $v$ gives
\[
\max_{r\le N}\sup_{u,v\in I_r}
|Q_{\mathbf B_\Sigma}(t(u))-Q_{\mathbf B_\Sigma}(t(v))|
\le2\sum_{m\ge0}
\max_{\substack{r\le N\\1\le j\le2^m}}|\Delta_{r,m,j}|.
\]
We next bound the moments of each grid increment.
Each $\Delta_{r,m,j}$ is a centered quadratic form of a Gaussian vector.
By Gaussian hypercontractivity
\citep[Theorem~5.10 and Remark~5.11]{janson1997gaussian}, for every $q\ge2$,
\[
\begin{aligned}
\|\Delta_{r,m,j}\|_q
\le(q-1)\|\Delta_{r,m,j}\|_2
&=2(q-1)\left(\{1-e^{-|I_r|2^{-m}}\}
\sum_{\ell=1}^d w_\ell^2\right)^{1/2}\\
&\le\sqrt2\,(q-1)\sqrt h\,2^{-m/2},
\end{aligned}
\]
where we used $\sum_\ell w_\ell^2=1/2$, $|I_r|\le h$, and
$1-e^{-x}\le x$ for $x\ge0$.
Finally, there are $N2^m$ increments at level $m$.
Set $q_m=\max\{2,\log(2N2^m)\}$, so that
$(N2^m)^{1/q_m}\le e$ and $q_m\le C\{\log(2N)+m\}$.
The moment bound above implies
\[
\begin{aligned}
\E\max_{\substack{r\le N\\1\le j\le2^m}}|\Delta_{r,m,j}|
\le\left(\sum_{r=1}^N\sum_{j=1}^{2^m}
\E|\Delta_{r,m,j}|^{q_m}\right)^{1/q_m}
&\le(N2^m)^{1/q_m}
\max_{\substack{r\le N\\1\le j\le2^m}}\|\Delta_{r,m,j}\|_{q_m}\\
&\le C\sqrt h\,2^{-m/2}\{\log(2N)+m\}.
\end{aligned}
\]
{Markov's inequality and Tonelli's theorem give}
\begingroup
\begin{equation}
\begin{aligned}
&\Pr\left\{\max_{r\le N}\sup_{u,v\in I_r}
|Q_{\mathbf B_\Sigma}(t(u))-Q_{\mathbf B_\Sigma}(t(v))|
>(\log n)^{-2}\right\}\\
&\quad\le(\log n)^2\E\max_{r\le N}\sup_{u,v\in I_r}
|Q_{\mathbf B_\Sigma}(t(u))-Q_{\mathbf B_\Sigma}(t(v))|\\
&\quad\le2(\log n)^2\sum_{m\ge0}
\E\max_{\substack{r\le N\\1\le j\le2^m}}|\Delta_{r,m,j}|\\
&\quad\le C(\log n)^2\sqrt h
\sum_{m\ge0}2^{-m/2}\{\log(2N)+m\}\\
&\quad\le C(\log n)^2\sqrt h\log(2N)
\le\frac{CK_n}{\log n}.
\end{aligned}
\label{eq:matched-decoupling-oscillation-tail}
\end{equation}
\endgroup

Combining
\eqref{eq:matched-decoupling-band-discretization},
\eqref{eq:matched-decoupling-grid-anticoncentration}, and
\eqref{eq:matched-decoupling-oscillation-tail}
gives, for every
$\epsilon>0$,
\begingroup
\begin{equation}
\begin{aligned}
&\sup_{\mathbf c\in\mathbb R^{|\mathcal I_n|}}
\Pr\left\{\left|\max_{k\in\mathcal I_n}
\{Q_{\mathbf B_\Sigma}(t_k)-c_k\}\right|\le\epsilon\right\}
\le C\epsilon\sqrt{K_n}
+C\left\{\frac{K_n}{\log n}
+\left(\frac{(K_n+\log d)^7}d\right)^{1/6}\right\},
\end{aligned}
\label{eq:matched-decoupling-translated-anticoncentration}
\end{equation}
\endgroup
where the constants are independent of $\epsilon$ and $\mathbf c$.
\end{proof}
\endgroup

\subsection{{Proof of Theorem~\ref{indalter} and Proposition~\ref{prop:cauchy-consistency}}%
}

\begin{lemma}%
\label{lem:projected-cusum-general}
Under Assumptions~\ref{ass:C1}--\ref{ass:C2}, let $\mathbf a\in\mathbb R^p$ be a deterministic vector and define
\[
Z_k(\mathbf a)=\frac{1}{\sqrt n}\left\{
\sum_{i=1}^k\mathbf a^\top\bepsilon_i
-\frac{k}{n}\sum_{i=1}^n\mathbf a^\top\bepsilon_i\right\}.
\]
\begin{enumerate}
\item[(i)]
\(
\max_{k\in\mathcal K_0}
|Z_k(\mathbf a)|
=O_p\left[\{\mathbf a^\top\bOme\mathbf a\}^{1/2}\right].
\)
\item[(ii)] If $\lambda_n\asymp n^\lambda$ for some $\lambda\in(0,1)$, then
\[
\max_{k\in\mathcal K_{1/2}}
\frac{|Z_k(\mathbf a)|}{\sqrt{v_k}}
=O_p\left[\{\mathbf a^\top\bOme\mathbf a\}^{1/2}
\sqrt{\log\log h_n}\right].
\]
\end{enumerate}
\end{lemma}

\begin{proof}
\begingroup
The result is immediate when $\mathbf a=\mathbf0$; assume henceforth
that $\mathbf a\ne\mathbf0$.

\endgroup

(i) Let $Y_i(\mathbf a)=\mathbf a^\top\bepsilon_i$ and $S_k(\mathbf a)=\sum_{i=1}^kY_i(\mathbf a)$.
For $1\le u\le v\le n$, Lemma~\ref{lem:vector-basic}(ii), with
$q=4$, $w_i=\ind{u\le i\le v}$, 
$\mathbf u=\mathbf a/\|\mathbf a\|_2$, and Assumption~\ref{ass:C2} gives
\[
\E\left|\sum_{i=u}^vY_i(\mathbf a)\right|^4
\le C(v-u+1)^2(\mathbf a^\top\bOme\mathbf a)^2.
\]
The maximal-moment inequality of \citet[Theorem~1]{moricz1982} therefore
yields
\[
\E\max_{1\le k\le n}|S_k(\mathbf a)|^4
\le Cn^2(\mathbf a^\top\bOme\mathbf a)^2.
\]
Since
\(
\max_{1\le k<n}|Z_k(\mathbf a)|
\le \frac{1}{\sqrt n}\left\{
\max_{1\le k\le n}|S_k(\mathbf a)|+|S_n(\mathbf a)|\right\}
\le\frac{2}{\sqrt n}\max_{1\le k\le n}|S_k(\mathbf a)|.
\)
Markov's inequality gives, for every $x>0$,
\begin{equation}
\Pr\left\{
\max_{k\in\mathcal K_0}|Z_k(\mathbf a)|
>x\{\mathbf a^\top\bOme\mathbf a\}^{1/2}
\right\}
\le Cx^{-4}.
\label{eq:projected-cusum-unweighted-tail}
\end{equation}
This proves part~(i).

(ii). Let $\sigma_{\mathbf a}^2=\mathbf a^\top\bOme\mathbf a$ and
\[
\mathcal X_{n,\mathbf a}^{\pm}
=\left\{
\varsigma\frac{Z_k(\mathbf a)}{\sigma_{\mathbf a}\sqrt{v_k}}:
\varsigma\in\{-1,1\},\ k\in\mathcal K_{1/2}
\right\},
\qquad
T_{n,\mathbf a}:=\max\mathcal X_{n,\mathbf a}^{\pm}.
\]
The process $X_i(\mathbf a)=\mathbf a^\top\bepsilon_i/\sigma_{\mathbf a}$ has
long-run variance one.  By Assumption~\ref{ass:C2}, \(\|\mathbf a\|_2/\sigma_{\mathbf a}\le C\). Hence \(\{X_i(\mathbf a)\}\) satisfies Assumption~\ref{ass:C1} with constants uniform in \(p\) and \(\mathbf a\).

Let $\mathcal G_{n,\mathbf a}^{\pm}$ be the centered Gaussian array with
the same covariance matrix as $\mathcal X_{n,\mathbf a}^{\pm}$, and set
$T_{n,\mathbf a}^{G}=\max\mathcal G_{n,\mathbf a}^{\pm}$.  The
signed-coordinate arguments in the proof of
Lemma~\ref{lem:block-comparison-general}, leading respectively to
\eqref{eq:joint-deletion-bound},
\eqref{eq:joint-dependence-reduction},
\eqref{eq:joint-gaussian-block-replacement-bound},
\eqref{eq:joint-covariance-restoration-bound}, and
\eqref{eq:joint-gaussian-deletion-bound}, apply to $\{X_i(\mathbf a)\}$.
Indeed, these arguments use only Assumption~\ref{ass:C1} and the unit
long-run variance of $X_i(\mathbf a)$.  Let
\begingroup
\[
\ell_n=\log(en),\qquad
\beta_n=\ell_n^2,\qquad
s_n=\begin{cases}\lceil c_0\ell_n^{1/\varkappa}\rceil,&\text{under (AM) or (PD)},\\m_0+1,&\text{under (MD)},\end{cases}\qquad
b_n=\lceil\ell_n^{\kappa_b}\rceil,
\]
\endgroup
and, for every fixed $A>0$,
\begingroup
\[
\rho_n
=\ell_n^{4+3c_{\mathrm{dep}}}
\left\{
\left(\frac{s_n}{b_n}\right)^{1/2}
+\left(\frac{b_n}{\lambda_n}\right)^{1/2}
\right\}+n^{-A}=o(1),
\]
\endgroup
where $c_0$ and $\kappa_b$ are chosen sufficiently large.  The cited
bounds give, for every $g\in C_b^3(\mathbb R)$,
\begin{equation}
\left|
\E g\{\operatorname{smax}_{\beta_n}
(\mathcal X_{n,\mathbf a}^{\pm})\}
-\E g\{\operatorname{smax}_{\beta_n}
(\mathcal G_{n,\mathbf a}^{\pm})\}
\right|\le C\|g\|_{C^3}\rho_n.
\label{eq:projected-cusum-scalar-comparison}
\end{equation}

It remains to bound $T_{n,\mathbf a}^{G}$.  Let
$G_{n,\mathbf a}(k)$ be the positive-sign coordinate of
$\mathcal G_{n,\mathbf a}^{\pm}$ and let $u_k=\log\{k/(n-k)\}$.  
The CUSUM covariance identity and Lemma~\ref{lem:vector-basic}(i) give,
uniformly in $k,l\in\mathcal K_{1/2}$,
\begin{align}
\left|\Cov\{G_{n,\mathbf a}(k),G_{n,\mathbf a}(l)\}
-e^{-|u_k-u_l|/2}\right|\le C\lambda_n^{-1},
\label{eq:projected-cusum-gaussian-covariance}
\end{align}
and hence, whenever $u_k,u_l$ lie in the same unit logit interval \(\mathcal I_r=\{k\in\mathcal K_{1/2}:r\le u_k<r+1\}\),
\[
\E\{G_{n,\mathbf a}(k)-G_{n,\mathbf a}(l)\}^2
\le C\{|u_k-u_l|+\lambda_n^{-1}\}.
\]
There are at most $C\log h_n$ such intervals.  {On each nonempty interval $\mathcal I_r$, choose $k_r\in\mathcal I_r$.
The Dudley entropy integral in \citet[Theorem~2.1.1]{adler2007random} is uniformly bounded:}
indeed, for the metric
$d_n(k,l):=\left[\E\{G_{n,\mathbf a}(k)-G_{n,\mathbf a}(l)\}^2\right]^{1/2}$,
let $N_r(\varepsilon)$ be the $d_n$-covering number of $\mathcal I_r$.
Taking $\varepsilon_0=C\lambda_n^{-1/2}$, when $\varepsilon\ge\varepsilon_0$,
a grid in the $u$-coordinate with mesh
of order $\varepsilon^2$ gives $N_r(\varepsilon)\le C\varepsilon^{-2}$;
when $0<\varepsilon<\varepsilon_0$, the trivial bound $N_r(\varepsilon)\le |\mathcal I_r|\le n$. Thus,
\[
\int_0^{\operatorname{diam}(\mathcal I_r,d_n)}
 \sqrt{\log N_r(\varepsilon)}\,d\varepsilon
\le C\lambda_n^{-1/2}\sqrt{\log n}
 +C\int_{C\lambda_n^{-1/2}}^1
       \sqrt{\log(C\varepsilon^{-2})}\,d\varepsilon
\le C,
\]
where the first term is $o(1)$ because $\lambda_n\asymp n^\lambda$. {Dudley's entropy inequality gives} 
\[
\mathbb E\max_{k\in\mathcal I_r}
|G_{n,\mathbf a}(k)-G_{n,\mathbf a}(k_r)|
\le \int_0^{\operatorname{diam}(\mathcal I_r,d_n)}
 \sqrt{\log N_r(\varepsilon)}\,d\varepsilon \le C.
\]
Furthermore, by taking $k=l$ in \eqref{eq:projected-cusum-gaussian-covariance}, we have $\sup_k\mathbb E|G_{n,\mathbf a}(k)|
\le\sup_k\{\operatorname{Var}(G_{n,\mathbf a}(k))\}^{1/2}\le C$.
Thus
\[
\sup_r\E\max_{k\in\mathcal I_r}|G_{n,\mathbf a}(k)|\le C.
\]
For any $x>0$, the Borell--TIS inequality yields that
\begin{equation}
\Pr\{T_{n,\mathbf a}^{G}>C+x\}
\le\sum_{r}
\Pr\{\max_{k\in\mathcal I_r}|G_{n,\mathbf a}(k)|>C_0+x\}
\le C\log h_ne^{-cx^2}.
\label{eq:projected-cusum-gaussian-logit-tail}
\end{equation}
For every sufficiently large fixed $M$, by taking $x=M\sqrt{\log\log h_n}$, $\Pr\{T_{n,\mathbf a}^{G}>C+M\sqrt{\log\log h_n}\}\le C(\log h_n)^{1-cM^2}\rightarrow 0$.
Then $T_{n,\mathbf a}^{G}=O_p\{\sqrt{\log\log h_n}\}$.

Denote $S_n=\operatorname{smax}_{\beta_n}(\mathcal X_{n,\mathbf a}^{\pm})$,
$S_n^G=\operatorname{smax}_{\beta_n}(\mathcal G_{n,\mathbf a}^{\pm})$, 
and choose a nondecreasing function
$f\in C_b^3(\mathbb R)$ with $f(x)=0$ on $(-\infty,1]$ and $f(x)=1$
on $[2,\infty)$.  By \eqref{eq:smax-logsumexp-bound},
$\delta_n:=\log\{2|\mathcal K_{1/2}|\}/\beta_n=o(1)$ bounds the
difference between each smooth maximum in
\eqref{eq:projected-cusum-scalar-comparison} and its ordinary maximum.
Applying \eqref{eq:projected-cusum-scalar-comparison} with $g(x)=f(x/(M\sqrt{\log\log h_n}))$ and then using
\eqref{eq:projected-cusum-gaussian-logit-tail} gives, for every sufficiently
large fixed $M$,
\begin{align*}
\Pr(T_{n,\mathbf a}>2M\sqrt{\log\log h_n})
&\le\E f\left(\frac{S_n}{M\sqrt{\log\log h_n}}\right)\le
\E f\left(\frac{S_n^G}{M\sqrt{\log\log h_n}}\right)+o(1)\\
&\le
\Pr(S_n^G>M\sqrt{\log\log h_n})+o(1)\\
&\le
\Pr(T_{n,\mathbf a}^G>M\sqrt{\log\log h_n}-\delta_n)+o(1)=o(1).
\end{align*}
Consequently $T_{n,\mathbf a}=O_p(\sqrt{\log\log h_n})$, which proves part~(ii).
\end{proof}

\begin{proof}[Proof of Theorem~\ref{indalter}]
Let $L_n=\log(np)$, $K_n=\log(2+\log n)$, and
\(
b_{\delta,n}=\left(n\bdelta^\top\bOme\bdelta/p\right)^{1/2}.
\)
In part~(i), Assumption~\ref{ass:C2} and
$\Delta_{S,n}p^{-1/4}\to0$ give
$b_{\delta,n}^2\le C\Delta_{S,n}/\sqrt p=o(1)$.
In part~(ii), the energy bound in $H_{1;n,p}$ gives
$b_{\delta,n}\sqrt{\log\log n}\to0$, and hence $b_{\delta,n}\to0$.
Recall $t_k=k/n$, $\theta=\tau/n$, $v_k=t_k(1-t_k)$ and
$\varphi_\gamma(t_k)=v_k^{-2\gamma}$.
Let
\[
\mathbf d_k
=-\sqrt n(t_k\wedge\theta-t_k\theta)\bdelta
=\frac{c_k}{\sqrt n}\bdelta,
\qquad c_k=-n(t_k\wedge\theta-t_k\theta).
\]
Since $0\le(t\wedge\theta-t\theta)/\{t(1-t)\}\le1$ for
$0<t<1$, we have $|c_k|\le nv_k$. 
Since $\inf_j\sigma_j\ge C_0^{1/2}$, for $\gamma\in\{0,1/2\}$,
\begin{align*}
\max_{j,k}|d_{\gamma,j}(k)|
\le C\sqrt n\|\bdelta\|_\infty\max_{k\in\mathcal K_\gamma}v_k^{1-\gamma}
\le C\sqrt n\|\bdelta\|_\infty
\le C\{\log(np)\}^{b_\delta}.
\end{align*}

For $\gamma\in\{0,1/2\}$, $k\in\mathcal K_\gamma$ and $j\le p$, define
\[
C_{\gamma,j}^{X,0}(k)=
\left\{\frac{k}{n}\left(1-\frac{k}{n}\right)\right\}^{-\gamma}
\frac{\sum_{i=1}^nc_{i,k}X_{ij}}
{\sqrt n\,\sigma_j},~~~~d_{\gamma,j}(k)=
\left\{\frac{k}{n}\left(1-\frac{k}{n}\right)\right\}^{-\gamma}
\frac{d_{k,j}}{\sigma_j}.
\]
Then
\[
C_{\gamma,j}^{X,0}(k)=C_{\gamma,j}^{\epsilon,0}(k)+d_{\gamma,j}(k),
\qquad
M_{\gamma,n,p}^{X,0}=\max_{k\in\mathcal K_\gamma}\max_{j\le p}|C_{\gamma,j}^{X,0}(k)|.
\]
Let $\mathcal C_{\gamma,n}^{X,\pm}$ contain both signs of
$C_{\gamma,j}^{X,0}(k)$, and define
$(\mathcal C_{\gamma,n}^{X,G})^\pm$ and
$(\mathcal C_{\gamma,n}^{X,B})^\pm$ by replacing the noise coordinate by its
Gaussian and ideal-bridge counterpart.
Recall $\mathsf N_{n,0}^{\mathrm M}(x)$ and $\mathsf N_{n,1/2}^{\mathrm M}(x)$ in \eqref{eq:max-normalizations},
and the shifted quadratic arrays defined in
\eqref{eq:def-boundary-quadratic-arrays-shifted}.
With $\mathbf d=(\mathbf d_k:k\in\mathcal K_\gamma)$, define
\begin{equation}
\begin{alignedat}{2}
\mathcal X_{n,\gamma}
&=\mathsf N_{n,\gamma}^{\mathrm S}\!\left[
\max\{\omega_p^{-1}\mathcal Q_{\gamma,n}(\mathbf d)\}
\right],
\qquad
&\mathcal Y_{n,\gamma}
&=\mathsf N_{n,\gamma}^{\mathrm M}
\{\max\mathcal C_{\gamma,n}^{X,\pm}\},\\
\mathcal X_{n,\gamma}^{G}
&=\mathsf N_{n,\gamma}^{\mathrm S}\!\left[
\max\{\omega_p^{-1}\mathcal Q_{\gamma,n}^{G}(\mathbf d)\}
\right],
\qquad
&\mathcal Y_{n,\gamma}^{G}
&=\mathsf N_{n,\gamma}^{\mathrm M}
\{\max(\mathcal C_{\gamma,n}^{X,G})^\pm\},\\
\mathcal X_{n,\gamma}^{B}
&=\mathsf N_{n,\gamma}^{\mathrm S}\!\left[
\max\{\mathcal Q_{\gamma,n}^{B}(\mathbf d)\}
\right],
\qquad
&\mathcal Y_{n,\gamma}^{B}
&=\mathsf N_{n,\gamma}^{\mathrm M}
\{\max(\mathcal C_{\gamma,n}^{X,B})^\pm\}.
\end{alignedat}
\label{eq:joint-alter-normalized-statistics}
\end{equation}

\begingroup
Step 1.
Let $\mathbf h_k=\mathbf d_k=-\sqrt nK_B(t_k,\theta)\bdelta$.
Since $0\le K_B(t_k,\theta)\le v_k$, $H_{1;n,p}$ gives
\begin{equation}
\begin{aligned}
\sqrt{\frac{K_n}{p}}\,
\E\max_{k\in\mathcal K_\gamma}
\frac{|\mathbf d_k^\top\mathbf B_\Omega(t_k)|}{v_k}
&\le\sqrt{\frac{nK_n}{p}}\,
\E\sup_{0\le t\le1}|\bdelta^\top\mathbf B_\Omega(t)|\\
&\le C\sqrt{\frac{nK_n\bdelta^\top\bOme\bdelta}{p}}=o(1).
\end{aligned}
\label{eq:joint-alter-gaussian-drift-linear}
\end{equation}
The remaining conditions of Lemma~\ref{lem:gaussian-decoupling-general}
follow from Assumption~\ref{ass:C2} and
\[
\chi_{\mathcal A,p}L_n^2\to0,\qquad
\log\log n=o(p^{1/7}),\qquad \log p=O(n^\nu),
\]
\[
\begin{array}{ll}
\gamma=0:& |\mathcal A|/p\to0,\quad
0<\nu<\nu_{\mathrm{dep}}<1/2;\\[2mm]
\gamma=1/2:& |\mathcal A|K_n^2/p\to0,\quad
\lambda_n\asymp n^\lambda,\quad
0<\nu<\lambda\nu_{\mathrm{dep}}<\lambda/2.
\end{array}
\]
Lemma~\ref{lem:gaussian-decoupling-general}(i)--(ii) therefore gives,
for every fixed $x,y\in\mathbb R$,
\begin{equation}
\Pr(\mathcal X_{n,\gamma}^{B}\le x,\mathcal Y_{n,\gamma}^{B}\le y)
-\Pr(\mathcal X_{n,\gamma}^{B}\le x)
 \Pr(\mathcal Y_{n,\gamma}^{B}\le y)\to0.
\label{eq:joint-alter-gaussian-factorization}
\end{equation}
\endgroup

Step 2.
{We prove} 
\begin{equation}
\left|\Pr(\mathcal X_{n,\gamma}\le x,\mathcal Y_{n,\gamma}\le y)
-\Pr(\mathcal X_{n,\gamma}\le x)
 \Pr(\mathcal Y_{n,\gamma}\le y)
\right|\rightarrow0.
\label{eq:joint-alter-original-factorization}
\end{equation}
Fix $x,y\in\mathbb R$ and $0<\eta\le1$. Choose a nonincreasing
$\psi\in C^\infty(\mathbb R)$ with $0\le\psi\le1$, $\psi(z)=1$ for
$z\le0$, and $\psi(z)=0$ for $z\ge1$, and define
\[
\psi_{\eta,x}(z)=\psi\{(z-x)/\eta\},
\qquad
f_{\eta,n,\gamma}(u,v)
=\psi_{\eta,x}\{\mathsf N_{n,\gamma}^{\mathrm S}(u)\}
 \psi_{\eta,y}\{\mathsf N_{n,\gamma}^{\mathrm M}(v)\}.
\]
Thus $\ind{z\le x}\le\psi_{\eta,x}(z)\le\ind{z\le x+\eta}$ and
$\|\psi_{\eta,x}^{(r)}\|_\infty\le C_r\eta^{-r}$ for $r=1,2,3$.
Since 
\(\|\frac{d^r}{du^r}
\psi_{\eta,x}\!(\mathsf N_{n,\gamma}^{\mathrm S}(u))\|_\infty
\lesssim \eta^{-r}L_n^{r/2}\)
and 
\(\|\frac{d^r}{dv^r}
\psi_{\eta,y}\!(\mathsf N_{n,\gamma}^{\mathrm M}(v))\|_\infty
\lesssim \eta^{-r}L_n^{r/2}\)
for $r=1,2,3$.
The chain rule therefore gives the global bound
\begin{equation}
\|f_{\eta,n,\gamma}\|_{C^3}
\lesssim \eta^{-3}L_n^{3/2}.
\label{eq:joint-alter-threshold-derivatives}
\end{equation}
\begingroup
Let $\beta_n=L_n^{7/4}$. Define
\endgroup
$\widetilde{\mathcal X}_{n,\gamma}$,
$\widetilde{\mathcal Y}_{n,\gamma}$,
$\widetilde{\mathcal X}_{n,\gamma}^{G}$,
$\widetilde{\mathcal Y}_{n,\gamma}^{G}$,
$\widetilde{\mathcal X}_{n,\gamma}^{B}$ and
$\widetilde{\mathcal Y}_{n,\gamma}^{B}$ by replacing $\max$ in
the corresponding definitions in
\eqref{eq:joint-alter-normalized-statistics} by
$\operatorname{smax}_{\beta_n}$, with all normalizations unchanged.
For every fixed $x,y\in\mathbb R$ and $0<\eta\le1$,
Lemma~\ref{lem:block-comparison-general}, with $c_{\mathrm{out}}=3/2$,
and \eqref{eq:joint-alter-threshold-derivatives} give
\begingroup
\[
\begin{aligned}
&\left|
\E\left[
\psi_{\eta,x}(\widetilde{\mathcal X}_{n,\gamma})
\psi_{\eta,y}(\widetilde{\mathcal Y}_{n,\gamma})\right]
-\E\left[
\psi_{\eta,x}(\widetilde{\mathcal X}_{n,\gamma}^{G})
\psi_{\eta,y}(\widetilde{\mathcal Y}_{n,\gamma}^{G})\right]
\right|\\
&\qquad\lesssim\eta^{-3}L_n^{3/2}r_{n,\gamma}
\{1+H_{\gamma,n}(\mathbf d)\}^{3}=o(1).
\end{aligned}
\]
\endgroup
Here, since $|c_k|\le nv_k$, Assumption~\ref{ass:C2} gives
\[
H_{\gamma,n}(\mathbf d)
=\max_{k\in\mathcal K_\gamma}
v_k^{-\gamma}\frac{\|\mathbf d_k\|_2}{\sqrt p}
\le \frac{\sqrt n\|\bdelta\|_2}{\sqrt p}
\max_{k\in\mathcal K_\gamma}v_k^{1-\gamma}
\lesssim \left(\frac{n\bdelta^\top\bOme\bdelta}{p}\right)^{1/2}=o(1).
\]
Similarly, Lemma~\ref{lem:joint-gaussian-cov-restoration} gives
\[
\begin{aligned}
&\left|
\E\left[
\psi_{\eta,x}(\widetilde{\mathcal X}_{n,\gamma}^{G})
\psi_{\eta,y}(\widetilde{\mathcal Y}_{n,\gamma}^{G})\right]
-\E\left[
\psi_{\eta,x}(\widetilde{\mathcal X}_{n,\gamma}^{B})
\psi_{\eta,y}(\widetilde{\mathcal Y}_{n,\gamma}^{B})\right]
\right|\\
&\qquad\lesssim\eta^{-3}L_n^{9/2}\Lambda_{\gamma,n}^{-1}
\{1+H_{\gamma,n}(\mathbf d)\}^{3}=O(n^{9\nu/2-\lambda_\gamma})=o(1).
\end{aligned}
\]
Since each array has at most $2np$ entries,
\eqref{eq:smax-logsumexp-bound} gives a smooth-maximum error of at most
\begingroup
$\log(2np)/\beta_n\lesssim L_n^{-3/4}$. The first-derivative bound
\endgroup
\(
\|\frac{d}{du}
\psi_{\eta,x}\!(\mathsf N_{n,\gamma}^{\mathrm S}(u))\|_\infty
+
\|\frac{d}{dv}
\psi_{\eta,y}\!(\mathsf N_{n,\gamma}^{\mathrm M}(v))\|_\infty\\
\lesssim \eta^{-1}L_n^{1/2}.
\)
{shows that replacing
the smooth maxima by the maxima changes the}
\begingroup
expectation by at most $\eta^{-1}L_n^{-1/4}=o(1)$.
The block, covariance and smoothing errors therefore give, for fixed $\eta>0$,
\begin{equation}
\begin{aligned}
&\left|\E\{\psi_{\eta,x}(\mathcal X_{n,\gamma})
                  \psi_{\eta,y}(\mathcal Y_{n,\gamma})\}
-\E\{\psi_{\eta,x}(\mathcal X_{n,\gamma}^{B})
       \psi_{\eta,y}(\mathcal Y_{n,\gamma}^{B})\}\right|\\
&\quad\le C\eta^{-3}
\{L_n^{3/2}r_{n,\gamma}+L_n^{9/2}\Lambda_{\gamma,n}^{-1}\}
\{1+H_{\gamma,n}(\mathbf d)\}^{3}
+C\eta^{-1}L_n^{-1/4}=o(1).
\end{aligned}
\label{eq:joint-alter-smooth-threshold-transfer}
\end{equation}
\endgroup
Similarly, we have 
\begin{equation}
\left|
\E\left[
\psi_{\eta,x}(\mathcal X_{n,\gamma})\right]
-\E\left[
\psi_{\eta,x}(\mathcal X_{n,\gamma}^{B})\right]
\right|\rightarrow 0,
\qquad
\left|
\E\left[
\psi_{\eta,y}(\mathcal Y_{n,\gamma})\right]
-\E\left[
\psi_{\eta,y}(\mathcal Y_{n,\gamma}^{B})\right]
\right|\rightarrow 0.
\label{eq:joint-alter-marginal-smooth-comparisons}
\end{equation}
Let
\[
F_{n,\gamma}(x,y)
=\Pr(\mathcal X_{n,\gamma}\le x,\,
\mathcal Y_{n,\gamma}\le y),
\qquad
F_{n,\gamma}^{B}(x,y)
=\Pr(\mathcal X_{n,\gamma}^{B}\le x,\,
\mathcal Y_{n,\gamma}^{B}\le y).
\]
The inequalities
$\ind{z\le x}\le\psi_{\eta,x}(z)\le\ind{z\le x+\eta}$,
together with \eqref{eq:joint-alter-smooth-threshold-transfer} at
$(x,y)$ and $(x-\eta,y-\eta)$ and the two marginal comparisons in
\eqref{eq:joint-alter-marginal-smooth-comparisons},
give, for each fixed $0<\eta\le1$,
\begin{equation}
\begin{aligned}
F_{n,\gamma}^{B}(x-\eta,y-\eta)-o(1)
&\le F_{n,\gamma}(x,y)
\le F_{n,\gamma}^{B}(x+\eta,y+\eta)+o(1),\\
\Pr(\mathcal X_{n,\gamma}^{B}\le x-\eta)-o(1)
&\le \Pr(\mathcal X_{n,\gamma}\le x)
\le \Pr(\mathcal X_{n,\gamma}^{B}\le x+\eta)+o(1),\\
\Pr(\mathcal Y_{n,\gamma}^{B}\le y-\eta)-o(1)
&\le \Pr(\mathcal Y_{n,\gamma}\le y)
\le \Pr(\mathcal Y_{n,\gamma}^{B}\le y+\eta)+o(1).
\end{aligned}
\label{eq:joint-alter-cdf-transfer-sandwich}
\end{equation}
Define
\(
\varpi_{n,\gamma}(\eta;x,y)
=\Pr(|\mathcal X_{n,\gamma}^{B}-x|\le\eta)
+\Pr(|\mathcal Y_{n,\gamma}^{B}-y|\le\eta).
\)
We claim that
\begin{equation}
\lim_{\eta\downarrow0}\limsup_{n,p\to\infty}
\varpi_{n,\gamma}(\eta;x,y)=0,
\qquad \gamma\in\{0,1/2\}.
\label{eq:joint-alter-threshold-continuity}
\end{equation}
\begingroup
Once \eqref{eq:joint-alter-threshold-continuity} is proved,
\endgroup
taking $n,p\to\infty$ first and then $\eta\downarrow0$ gives
\[
\limsup_{n,p\to\infty}
|F_{n,\gamma}(x,y)-F_{n,\gamma}^{B}(x,y)|
\le
\lim_{\eta\downarrow0}\limsup_{n,p\to\infty}\varpi_{n,\gamma}(\eta;x,y)=0.
\]
The marginal sandwich also yields, for every fixed $x,y\in\mathbb R$,
\begin{equation}
\begin{aligned}
\left|\Pr(\mathcal X_{n,\gamma}\le x)
-\Pr(\mathcal X_{n,\gamma}^{B}\le x)\right|&\to0,\\
\left|\Pr(\mathcal Y_{n,\gamma}\le y)
-\Pr(\mathcal Y_{n,\gamma}^{B}\le y)\right|&\to0.
\end{aligned}
\label{eq:joint-alter-marginal-cdf-comparisons}
\end{equation}
Combining these comparisons with
\eqref{eq:joint-alter-gaussian-factorization}, we have
\begin{align*}
&\left|
\Pr(\mathcal X_{n,\gamma}\le x,\mathcal Y_{n,\gamma}\le y)
-\Pr(\mathcal X_{n,\gamma}\le x)
 \Pr(\mathcal Y_{n,\gamma}\le y)
\right|\\
&\quad\le |F_{n,\gamma}(x,y)-F_{n,\gamma}^{B}(x,y)|+\left|
F_{n,\gamma}^{B}(x,y)
-\Pr(\mathcal X_{n,\gamma}^{B}\le x)
 \Pr(\mathcal Y_{n,\gamma}^{B}\le y)
\right|\\
&\qquad+\left|
\Pr(\mathcal X_{n,\gamma}\le x)
-\Pr(\mathcal X_{n,\gamma}^{B}\le x)
\right|\Pr(\mathcal Y_{n,\gamma}^{B}\le y)\\
&\qquad+\left|
\Pr(\mathcal Y_{n,\gamma}\le y)
-\Pr(\mathcal Y_{n,\gamma}^{B}\le y)
\right|\Pr(\mathcal X_{n,\gamma}\le x)
\rightarrow0.
\end{align*}

We now verify \eqref{eq:joint-alter-threshold-continuity}.
\begingroup
For $\gamma=0$, let $X_{R,n,0}^{G}$ be as in the proof of
Lemma~\ref{lem:gaussian-decoupling-general}(i), with
$\mathbf h_k=\mathbf d_k$. For every $\eta>0$,
\[
\begin{aligned}
\Pr(|\mathcal X_{n,0}^{B}-x|\le\eta)
&\le\Pr(|X_{R,n,0}^{G}-x|\le2\eta)
+\Pr(|\mathcal X_{n,0}^{B}-X_{R,n,0}^{G}|>\eta).
\end{aligned}
\]
Equations~\eqref{eq:joint-alter-unweighted-quadratic-approximation}
and~\eqref{eq:matched-decoupling-unweighted-structured-continuity}
therefore imply
\begin{equation}
\lim_{\eta\downarrow0}\limsup_{n,p\to\infty}
\Pr(|\mathcal X_{n,0}^{B}-x|\le\eta)=0,
\qquad x\in\mathbb R.
\label{eq:joint-alter-unweighted-dense-band}
\end{equation}
\endgroup
We next consider $\mathcal Y_{n,0}^{B}$. Let 
\[
a_j=-\sqrt n\,\delta_j/\sigma_j,\qquad
U_j(t)=B_{\Omega,j}(t)/\sigma_j+a_jK_B(t,\theta).
\]
Then, we rewrite
\[
\mathcal Y_{n,0}^{B}
=2\left(\max_{k\in\mathcal K_0,j\le p}|U_j(t_k)|\right)^2-\log(2p).
\]
For fixed $y\in\mathbb R$ and $\eta>0$, set $\ell_n=\sqrt{\{\log(2p)+y-\eta\}/2}$,
$u_n=\sqrt{\{\log(2p)+y+\eta\}/2}$.
Both are positive eventually, $u_n\asymp\ell_n\asymp\sqrt{\log p}$,
and $u_n(u_n-\ell_n)\le C\eta$. For sufficiently large $p$,
\[
\begin{aligned}
\{|\mathcal Y_{n,0}^{B}-y|\le\eta\}
&=\left\{\ell_n\le \max_{k\in\mathcal K_0,j\le p}|U_j(t_k)|\le u_n\right\}.
\end{aligned}
\]
Define the continuous maximum and its grid approximation error by
\[
\mathcal M_n=\max_{j\le p}\sup_{0\le t\le1}|U_j(t)|,
\qquad
\varepsilon_n=\mathcal M_n-
\max_{k\in\mathcal K_0,j\le p}|U_j(t_k)|\ge0.
\]
On $\{|\mathcal Y_{n,0}^{B}-y|\le\eta\}$ the grid maximum is at most
$u_n$, and hence
\[
0\le 2\mathcal M_n^2-\log(2p)-\mathcal Y_{n,0}^{B}
\le4u_n\varepsilon_n+2\varepsilon_n^2\le\eta
\quad\text{if }\varepsilon_n\le\frac{\eta}{8u_n},
\]
where the last inequality holds for sufficiently large $p$. Thus
\begin{equation}
\begin{aligned}
\Pr(|\mathcal Y_{n,0}^{B}-y|\le\eta)
&\le\Pr(|2\mathcal M_n^2-\log(2p)-y|\le2\eta)+\Pr\!\left(\varepsilon_n>\frac{\eta}{8u_n}\right).
\end{aligned}
\label{eq:joint-alter-unweighted-max-band-transfer}
\end{equation}

We consider the first term in \eqref{eq:joint-alter-unweighted-max-band-transfer}.
With $\ell_n$ and $u_n$, $\Pr(|2\mathcal M_n^2-\log(2p)-y|\le\eta)=
\Pr(\ell_n\le\mathcal M_n\le u_n)$.
If $\max_{j\le p}|a_j|\theta(1-\theta)>2u_n$, choose a coordinate $j$
such that $|a_j|\theta(1-\theta)>2u_n$, and let
$s=\Sgn(a_j)$. Evaluating $sU_j(t)$ at $t=\theta$ gives
\begin{equation}
\Pr(\ell_n\le\mathcal M_n\le u_n)
\le\Pr\left(s\frac{B_{\Omega,j}(\theta)}{\sigma_{j}}\le-u_n\right)
\le e^{-2u_n^2}=o(1).
\label{eq:joint-alter-unweighted-max-large-signal-band}
\end{equation}
It remains to consider $\max_j|a_j|\theta(1-\theta)\le2u_n$.
Define the two half-interval events
\[
\begin{aligned}
E_{n,+}=\left\{\mathcal M_n\le u_n,\,
\max_{j\le p}\sup_{0\le t\le1/2}|U_j(t)|\ge\ell_n\right\},~~
E_{n,-}=\left\{\mathcal M_n\le u_n,\,
\max_{j\le p}\sup_{1/2\le t\le1}|U_j(t)|\ge\ell_n\right\}.
\end{aligned}
\]
Since
\(
\{\ell_n\le\mathcal M_n\le u_n\}=E_{n,+}\cup E_{n,-},
\)
it suffices to bound their probabilities. We first treat $E_{n,+}$,
and then use time reversal to obtain the same bound for $E_{n,-}$.

Use the Brownian bridge representation
\[
\left(\frac{B_{\Omega,j}(t)}{\sigma_j}\right)_{j\le p,\,0\le t<1}
\overset d=
\left((1-t)B_{\mathrm{BM},j}\!\left(\frac{t}{1-t}\right)
\right)_{j\le p,\,0\le t<1},
\]
where $(B_{\mathrm{BM},j})_{j\le p}$ is a Brownian vector with
\(
\Cov\{B_{\mathrm{BM},i}(r),B_{\mathrm{BM},j}(q)\}
=\frac{\bOme_{ij}}{\sigma_i\sigma_j}(r\wedge q).
\)
Working with this representation and setting $r=t/(1-t)$ gives
\[
U_j\!\left(\frac r{1+r}\right)
=\frac{B_{\mathrm{BM},j}(r)+a_jh_\theta(r)}{1+r},
\qquad h_\theta(r)=\min\{(1-\theta)r,\theta\}.
\]
Let
\[
R=\inf\left\{r\in[0,1]:
\max_{j\le p}\left|U_j\!\left(\frac r{1+r}\right)\right|
\ge\ell_n\right\}\wedge1,
\qquad \inf\varnothing=\infty.
\]
If $\max_{i\le p}|U_i(t)|$ reaches $\ell_n$ for some $t\in[0,1/2]$,
continuity allows us to choose a smallest index $j$ and sign
$s\in\{-1,1\}$ such that $sU_j(R/(1+R))=\ell_n$.
Choose a fixed pair otherwise; then $(j,s)$ is measurable at $R$.
On event $E_{n,+}$,
\(
sB_{\mathrm{BM},j}(R)+sa_jh_\theta(R)
=\ell_n(1+R).
\)
Moreover, $E_{n,+}$ implies $\mathcal M_n\le u_n$, so for any $v\ge0$,
\(
sU_j\!\left(\frac{R+v}{1+R+v}\right)\le u_n,
\)
and hence
\(
sB_{\mathrm{BM},j}(R+v)+sa_jh_\theta(R+v)
\le u_n(1+R+v).
\)
Subtracting the equality at $R$ gives
\begin{equation}
\begin{aligned}
&s\{B_{\mathrm{BM},j}(R+v)-B_{\mathrm{BM},j}(R)\}\\
&\quad\le
u_n(1+R+v)-\ell_n(1+R)
-sa_j\{h_\theta(R+v)-h_\theta(R)\}\\
&\quad=
(u_n-\ell_n)(1+R)+u_nv
-sa_j\{h_\theta(R+v)-h_\theta(R)\}.
\end{aligned}
\label{eq:joint-alter-unweighted-max-hit-increment}
\end{equation}
We bound the two terms on the right-hand side separately. First,
since $R\le1$ and $u_n\ge\ell_n$,
\[
(u_n-\ell_n)(1+R)+u_nv\le2(u_n-\ell_n)+u_nv.
\]
Second, consider the term $-sa_j\{h_\theta(R+v)-h_\theta(R)\}$.
The function $h_\theta$ is nondecreasing, satisfies
\[
0\le h_\theta(R+v)-h_\theta(R)\le(1-\theta)v,
\qquad v\ge0,
\]
and is constant on $[\theta/(1-\theta),\infty)$.
Let $c_*=1/16$. The following four cases cover all possibilities
for this signal term on $E_{n,+}$.

(i). $sa_j\ge0$. In this case, 
\(
-sa_j\{h_\theta(R+v)-h_\theta(R)\}\le0
\)
by the monotonicity of $h_\theta$.

(ii). $sa_j<0$ and $R\ge\theta/(1-\theta)$.
In this case, 
\(
-sa_j\{h_\theta(R+v)-h_\theta(R)\}=0.
\)

(iii). $sa_j<0$, $R<\theta/(1-\theta)$, and $\theta\ge c_*$.
The signal bound
$|a_j|\theta(1-\theta)\le2u_n$ implies
\(
|a_j|(1-\theta)\le2u_n/\theta\le2u_n/c_*.
\)
Consequently,
\[
\begin{aligned}
-sa_j\{h_\theta(R+v)-h_\theta(R)\}
&=|a_j|\{h_\theta(R+v)-h_\theta(R)\}\le |a_j|(1-\theta)v
\le\frac{2u_n}{c_*}v.
\end{aligned}
\]

(iv). $sa_j<0$, $R<\theta/(1-\theta)$, and $\theta<c_*$.
Since $R/(1+R)<\theta<c_*$ and 
$sa_jK_B(\frac{R}{1+R},\theta)\le 0$,
the signed bridge must reach at least $\ell_n$ before $c_*$. This case
is therefore contained in the exceptional event
\[
D_{n,+}:=\left\{\max_{j\le p}\sup_{0\le t\le c_*}
\frac{|B_{\Omega,j}(t)|}{\sigma_j}\ge\ell_n\right\},
\]
whose probability satisfies
\[
\Pr(D_{n,+})
\le4p\exp(-15\ell_n^2/2)=o(1).
\]
The probability bound follows from the bridge representation,
$c_*/(1-c_*)=1/15$, the reflection principle, and a union bound over $j$.

Combining (i)--(iv), together with the bound for the first term in \eqref{eq:joint-alter-unweighted-max-hit-increment},
this gives, on $E_{n,+}\setminus D_{n,+}$,
\[
\begin{aligned}
s\{B_{\mathrm{BM},j}(R+v)-B_{\mathrm{BM},j}(R)\}
&\le2(u_n-\ell_n)+u_nv+\frac{2u_n}{c_*}v\\
&=2(u_n-\ell_n)+C_0u_nv,\qquad v\ge0,
\end{aligned}
\]
where $C_0=1+2/c_*$.
Since the selected pair $(j,s)$ is measurable at $R$, the strong Markov
property of the Brownian vector and the scaling property of Brownian
motion imply that the process
\[
\widehat B_n(q):=u_ns\left\{
B_{\mathrm{BM},j}\!\left(R+\frac q{u_n^2}\right)
-B_{\mathrm{BM},j}(R)\right\},\qquad q\ge0,
\]
is standard Brownian motion. %
In particular,
$(\widehat B_n(q))_{q\ge0}\overset d=(B_{\mathrm{BM}}(q))_{q\ge0}$.
On $E_{n,+}\setminus D_{n,+}$, substituting
$v=q/u_n^2$ in the increment bound and multiplying by $u_n$ gives
\[
\widehat B_n(q)-C_0q\le2u_n(u_n-\ell_n),
\qquad 0\le q\le1.
\]
Consequently,
\(
E_{n,+}\setminus D_{n,+}
\subseteq\left\{
\sup_{0\le q\le1}\{\widehat B_n(q)-C_0q\}
\le2u_n(u_n-\ell_n)\right\}.
\)
Since $u_n^2-\ell_n^2=\eta$, set
\[
z:=2u_n(u_n-\ell_n)
=\frac{2u_n}{u_n+\ell_n}\eta\le2\eta.
\]
Equality in law of the Brownian processes and the reflection formula yield
\[
\begin{aligned}
\Pr(E_{n,+})
&\le\Pr\left(\sup_{0\le q\le1}\{\widehat B_n(q)-C_0q\}
\le2u_n(u_n-\ell_n)\right)+\Pr(D_{n,+})\\
&=\Pr\left(\sup_{0\le q\le1}\{B_{\mathrm{BM}}(q)-C_0q\}
\le2u_n(u_n-\ell_n)\right)+o(1)\\
&=\Pr(B_{\mathrm{BM}}(1)\le C_0+z)-\Pr\!\left(\sup_{0\le q\le1}\{B_{\mathrm{BM}}(q)-C_0q\}>z,\ B_{\mathrm{BM}}(1)\le C_0+z\right)+o(1)\\
&=\{\Phi(C_0+z)-\Phi(C_0-z)\}
+(1-e^{-2C_0z})\Phi(C_0-z)+o(1)\\
&\le\bigl(\sqrt{2/\pi}+2C_0\bigr)z+o(1)
\le C\eta+o(1).
\end{aligned}
\]
Here $\Phi$ is the standard normal distribution function; the first
inequality after the reflection formula uses its density bound
$1/\sqrt{2\pi}$ and $1-e^{-2C_0z}\le2C_0z$, and the last inequality
uses $z\le2\eta$.
The same argument applied to $(U_j(1-t))_{j\le p}$, with $\theta$ replaced by
$1-\theta$, gives the identical bound for $\Pr(E_{n,-})$.
Since $\{\ell_n\le\mathcal M_n\le u_n\}=E_{n,+}\cup E_{n,-}$,
we have
\begin{align}
\Pr(\ell_n\le\mathcal M_n\le u_n)
&\le \Pr(E_{n,+})+\Pr(E_{n,-})+o(1)\le C\eta+o(1).
\label{eq:joint-alter-unweighted-continuous-max-band}
\end{align}
Together with \eqref{eq:joint-alter-unweighted-max-large-signal-band},  
we have
\begin{equation}
\Pr(|2\mathcal M_n^2-\log(2p)-y|\le2\eta) \le C\eta+o(1).
\label{eq:joint-alter-unweighted-continuous-max-normalized-band}
\end{equation}

{We now show that the second term in \eqref{eq:joint-alter-unweighted-max-band-transfer} tends to zero.}
Let
\[
\tilde B_{\Omega,j}(t)=\frac{B_{\Omega,j}(t)}{\sigma_j},
\qquad I_i=[i/n,(i+1)/n],\quad i=0,\ldots,n-1.
\]
Each $\tilde B_{\Omega,j}$ is a standard Brownian bridge. Then $U_j(t)=\tilde B_{\Omega,j}(t)+a_jK_B(t,\theta)$.
Since
$\theta=\tau/n$, the signal $a_jK_B(t,\theta)$ is linear on every
$I_i$. For each $j\le p$ and $t\in I_i$, choose $s\in\{-1,1\}$
such that $sU_j(t)=|U_j(t)|$, taking $s=1$ if $U_j(t)=0$.
Then choose an endpoint $v$ of $I_i$ such that
$sa_jK_B(t,\theta)\le sa_jK_B(v,\theta)$. Hence
\[
|U_j(t)|\le|U_j(v)|+|\tilde B_{\Omega,j}(t)-\tilde B_{\Omega,j}(v)|.
\]
The endpoints $0$ and $1$ cause no difficulty since $U_j(0)=U_j(1)=0$.
Taking suprema over $t\in I_i$ and maxima over the coordinates and
intervals yields
\begin{equation}
0\le\varepsilon_n
\le\max_{j\le p}\max_{0\le i<n}
\sup_{s,t\in I_i}|\tilde B_{\Omega,j}(s)-\tilde B_{\Omega,j}(t)|.
\label{eq:joint-alter-unweighted-max-grid-pathwise}
\end{equation}
For each fixed $j$, the standard Brownian-bridge representation gives
\(
\bigl(\tilde B_{\Omega,j}(t)\bigr)_{0\le t\le1}
\overset d=
\bigl(B_{\mathrm{BM}}(t)-tB_{\mathrm{BM}}(1)\bigr)_{0\le t\le1}.
\)
By triangle inequality, we have
\begin{equation}
\begin{aligned}
\sup_{s,t\in I_i}|\tilde B_{\Omega,j}(s)-\tilde B_{\Omega,j}(t)|
&\le2\sup_{0\le r\le1/n}
|B_{\mathrm{BM}}(i/n+r)-B_{\mathrm{BM}}(i/n)|+\frac1n|B_{\mathrm{BM}}(1)|.
\end{aligned}
\label{eq:joint-alter-unweighted-bridge-oscillation-bound}
\end{equation}
Applying a union bound to the two terms on the right-hand side of
\eqref{eq:joint-alter-unweighted-bridge-oscillation-bound}, each at
threshold $z/2$, %
for every $z>0$,
\[
\begin{aligned}
\Pr\!\left(
\sup_{s,t\in I_i}|\tilde B_{\Omega,j}(s)-\tilde B_{\Omega,j}(t)|>z
\right)
&\le
\Pr\!\left(
\sup_{0\le r\le1/n}
|B_{\mathrm{BM}}(i/n+r)-B_{\mathrm{BM}}(i/n)|>\frac z4
\right)\\
&\quad
+\Pr\!\left(|B_{\mathrm{BM}}(1)|>\frac{nz}{2}\right)\\
&\le4e^{-nz^2/32}+2e^{-n^2z^2/8}\le6e^{-nz^2/32}.
\end{aligned}
\]
{The bound for the first term follows from} 
\(
\bigl(B_{\mathrm{BM}}(i/n+r)-B_{\mathrm{BM}}(i/n)\bigr)_{r\ge0}
\overset d=
\bigl(B_{\mathrm{BM}}(r)\bigr)_{r\ge0}
\) and
\[
\Pr\!\left(\sup_{0\le r\le h}|B_{\mathrm{BM}}(r)|>a\right)
\le4\Pr\!\left(B_{\mathrm{BM}}(h)>a\right)
\le4\exp\!\left(-\frac{a^2}{2h}\right),\qquad h,a>0.
\]
A union bound over the $p$ coordinates and $n$ intervals therefore gives
\[
\Pr\!\left(
\max_{j\le p}\max_{0\le i<n}
\sup_{s,t\in I_i}|\tilde B_{\Omega,j}(s)-\tilde B_{\Omega,j}(t)|>z
\right)
\le6pn\,e^{-nz^2/32}\to0.
\]
by taking $z=A\sqrt{L_n/n}$ with any fixed $A>\sqrt{32}$ and 
$L_n=\log(np)$. %
By \eqref{eq:joint-alter-unweighted-max-grid-pathwise},
then
\begin{equation}
0\le\varepsilon_n
=\mathcal M_n-\max_{k\in\mathcal K_0,j\le p}|U_j(t_k)|
=O_p\!\left(\sqrt{\frac{L_n}{n}}\right).
\label{eq:joint-alter-unweighted-max-grid-error}
\end{equation}
Since $\log p=O(n^\nu)$ with
{$0<\nu<\min\{\nu_{\mathrm{dep}},\nu_c\}$, we have}
$\nu<1/2$ and $L_n=\log n+\log p=O(n^\nu)$. Hence
\(
u_n\varepsilon_n
=O_p\!\left(\sqrt{\log p}\sqrt{L_n/n}\right)
=O_p\!\left(n^{\nu-1/2}\right)=o_p(1).
\)
Consequently, for every fixed $\eta>0$,
\begin{equation}
\Pr\!\left(\varepsilon_n>\frac{\eta}{8u_n}\right)
\rightarrow0.
\label{eq:joint-alter-unweighted-max-grid-error-prob}
\end{equation}
Combining \eqref{eq:joint-alter-unweighted-continuous-max-normalized-band} and \eqref{eq:joint-alter-unweighted-max-grid-error-prob}, we have
\begin{equation}
\lim_{\eta\downarrow0}\limsup_{n,p\to\infty}
\Pr(|\mathcal Y_{n,0}^{B}-y|\le\eta)=0.
\label{eq:joint-alter-unweighted-max-band}
\end{equation}
Equations~\eqref{eq:joint-alter-unweighted-dense-band} and
\eqref{eq:joint-alter-unweighted-max-band} prove
\eqref{eq:joint-alter-threshold-continuity} for $\gamma=0$.

Next, we consider $\gamma=1/2$ and begin with $\mathcal X_{n,1/2}^{B}$.
Let $R_{\delta,\gamma,n}^{G}(k)$ be
$R_{h,\gamma,n}^{G}(k)$ in
\eqref{eq:matched-decoupling-shift-expansion-2} with
$\mathbf h_k=\mathbf d_k$. Since
$A_{\mathrm{DE}}(\log h_n)=O(\sqrt{K_n})$,
\eqref{eq:matched-decoupling-linear-remainder} gives
\[
A_{\mathrm{DE}}(\log h_n)
\max_{k\in\mathcal K_{1/2}}|R_{\delta,1/2,n}^{G}(k)|
=o_p(1).
\]
Recall also
$\mathfrak d_{\gamma,n}^{\mathrm{det}}(k)
=v_k^{-2\gamma}\|\mathbf d_k\|_2^2/(\sqrt p\,\omega_p)$
from \eqref{eq:matched-decoupling-shift-expansion-2} with
$\mathbf h_k=\mathbf d_k$.
For the dense component, we apply
Proposition~\ref{prop:quadratic-bridge-anticoncentration}
to the Gaussian bridge $\mathbf B_\Omega$.
Since $\sqrt p\,\omega_p=\sqrt{2\Tr(\bOme^2)}$, we have
\[
G_{p,1/2}(t_k)=Q_{\mathbf B_\Omega}(t_k).
\]
Taking
\(
\mathcal I_n=\mathcal K_{1/2},
c_k=(\mathsf N_{n,1/2}^{\mathrm S})^{-1}(x)
-\mathfrak d_{1/2,n}^{\mathrm{det}}(k),
\epsilon=\frac{2\eta}{A_{\mathrm{DE}}(\log h_n)}
\)
in Proposition~\ref{prop:quadratic-bridge-anticoncentration},
whose covariance conditions hold by Assumption~\ref{ass:C2},
we obtain, for every fixed $\eta>0$,
\begingroup
\begin{equation}
\begin{aligned}
\sup_{x\in\mathbb R}
\Pr(|\mathcal X_{n,1/2}^{B}-x|\le\eta)
&\le
\frac{C\eta\sqrt{K_n}}{A_{\mathrm{DE}}(\log h_n)}
+C\left\{\frac{K_n}{\log n}
+\left\{\frac{(K_n+\log p)^7}p\right\}^{1/6}\right\}+o(1)\\
&\le C\eta+o(1),
\end{aligned}
\label{eq:joint-alter-weighted-dense-band}
\end{equation}
\endgroup
since $\sqrt{K_n}/A_{\mathrm{DE}}(\log h_n)=O(1)$.
Next consider $\mathcal Y_{n,1/2}^{B}$. The transformed time points $u_k=\log\{t_k/(1-t_k)\}$,
$k\in\mathcal K_{1/2}$, lie in the logit interval
\(
[
-\log\frac{n-\lambda_n}{\lambda_n},
\ \log\frac{n-\lambda_n}{\lambda_n}
],
\)
whose length is $\log h_n$, with \(h_n=\left\{(\lambda_n/n)^{-1}-1\right\}^{2}.\)
Partition the logit interval into at most
$N_n\le C[1+(\log h_n)\{\log(p\log h_n)\}^4]$ intervals of length
at most $\{\log(p\log h_n)\}^{-4}$. 
Let $t=e^u/(1+e^u)$, and define
\[
Z_j(u)=\frac{B_{\Omega,j}(t)}{\sigma_j\sqrt{t(1-t)}}.
\]
$Z_j$ is a centered stationary Gaussian process with
$\Cov\{Z_j(u),Z_j(v)\}=\exp(-|u-v|/2)$.
For each interval $I$ containing an observation time, each $j\le p$
and each $s\in\{-1,1\}$, choose deterministically
\[
k_{I,j,s}\in{\arg\max\nolimits_{k\in\mathcal K_{1/2},\,u_k\in I}}
\frac{s\,d_{k,j}}{\sigma_j\sqrt{v_k}},
\]
and let $\mathcal I_n$ be the set of selected triples $(k,j,s)$; thus
$|\mathcal I_n|\le2pN_n$. Define
\[
\begin{aligned}
\mathcal Y_{n,1/2}^{B,\mathrm{sel}}
&=A_{\mathrm{DE}}(p\log h_n)
\max_{(k,j,s)\in\mathcal I_n}
s\left\{Z_j(u_k)+\frac{d_{k,j}}{\sigma_j\sqrt{v_k}}\right\}
-D_{\mathrm{DE}}(p\log h_n).
\end{aligned}
\]
It follows that
\[
0\le\mathcal Y_{n,1/2}^{B}
-\mathcal Y_{n,1/2}^{B,\mathrm{sel}}
\le A_{\mathrm{DE}}(p\log h_n)\Delta_n^{\mathrm{osc}},
\]
where 
\[
\Delta_n^{\mathrm{osc}}=\max_{j\le p}\max_{I}\sup_{u,v\in I}|Z_j(u)-Z_j(v)|.
\]
Then
\begin{align}
\sup_{y\in\mathbb R}
\Pr(|\mathcal Y_{n,1/2}^{B}-y|\le\eta)
\le\sup_{y\in\mathbb R}\Pr(|\mathcal Y_{n,1/2}^{B,\mathrm{sel}}-y|\le2\eta)
+\Pr(A_{\mathrm{DE}}(p\log h_n)\Delta_n^{\mathrm{osc}}>\eta).
\label{eq:joint-alter-weighted-max-band}
\end{align}
Applying the dyadic decomposition used in the proof of
Proposition~\ref{prop:quadratic-bridge-anticoncentration}
to the present partition, we divide each interval into $2^m$
equal subintervals at level $m$. Since
$\Var\{Z_j(u)-Z_j(v)\}=2\{1-e^{-|u-v|/2}\}\le|u-v|$,
the Gaussian maximal inequality applied to the resulting increments
of $Z_j$, together with Markov's inequality, gives, for every fixed
$\eta>0$,
\begin{align*}
\Pr(A_{\mathrm{DE}}(p\log h_n)\Delta_n^{\mathrm{osc}}>\eta)
&\le\frac{A_{\mathrm{DE}}(p\log h_n)}{\eta}
\E\Delta_n^{\mathrm{osc}}\\
&\le\frac{C}{\eta\{\log(p\log h_n)\}^{3/2}}
\sum_{m\ge0}2^{-m/2}\sqrt{\log(2pN_n)+m}\\
&\le\frac{C}{\eta\log(p\log h_n)},
\end{align*}
where $A_{\mathrm{DE}}(p\log h_n)=\sqrt{2\log(p\log h_n)}$
and $\log(2pN_n)=O\{\log(p\log h_n)\}$.
Similarly, Nazarov's inequality
\citep[Lemma~A.1]{chernozhukov2017central} gives
\[
\begin{aligned}
\sup_{y\in\mathbb R}
\Pr(|\mathcal Y_{n,1/2}^{B,\mathrm{sel}}-y|\le2\eta)
&\le\frac{4\eta}{A_{\mathrm{DE}}(p\log h_n)}
\left\{\sqrt{2\log|\mathcal I_n|}+2\right\}
\le C\eta,
\end{aligned}
\]
with $|\mathcal I_n|\le2pN_n$.
Thus, by \eqref{eq:joint-alter-weighted-max-band},
\begin{align*}
\sup_{y\in\mathbb R}
\Pr(|\mathcal Y_{n,1/2}^{B}-y|\le\eta)
\le C\eta+\frac{C}{\eta\log(p\log h_n)}
=C\eta+o(1)
\end{align*}
for every fixed $\eta>0$.
Thus, \eqref{eq:joint-alter-weighted-dense-band} and \eqref{eq:joint-alter-weighted-max-band} prove
\eqref{eq:joint-alter-threshold-continuity} when $\gamma=1/2$.
Now, \eqref{eq:joint-alter-original-factorization} has been proved.

Step 3.
{We prove}
\[
\Pr(Z_{S,\gamma,n}\le x,Z_{M,\gamma,n}\le y)
-\Pr(Z_{S,\gamma,n}\le x)\Pr(Z_{M,\gamma,n}\le y)\rightarrow0,
\]
where
$Z_{S,\gamma,n}=\mathsf N_{n,\gamma}^{\mathrm S}(S_\gamma)$, $Z_{M,\gamma,n}=\mathsf N_{n,\gamma}^{\mathrm M}(M_\gamma)$.
With equations~\eqref{eq:joint-alter-gaussian-factorization} and \eqref{eq:joint-alter-original-factorization}, it remains to prove
\begin{equation}
\begin{aligned}
\left|\Pr(Z_{S,\gamma,n}\le x,Z_{M,\gamma,n}\le y)
-\Pr(\mathcal X_{n,\gamma}\le x,\mathcal Y_{n,\gamma}\le y)\right|
\rightarrow0,
\end{aligned}
\label{eq:joint-alter-feasible-transfer-target-1}
\end{equation}
and
\begin{equation}
\begin{aligned}
\left|\Pr(Z_{S,\gamma,n}\le x)-\Pr(\mathcal X_{n,\gamma}\le x)\right|+\left|\Pr(Z_{M,\gamma,n}\le y)-\Pr(\mathcal Y_{n,\gamma}\le y)\right|
\rightarrow0.
\end{aligned}
\label{eq:joint-alter-feasible-transfer-target-2}
\end{equation}

We first control the centering, scale and second-order correction for
the dense statistic $Z_{S,\gamma,n}$, and then consider the componentwise variance
estimators for the max statistic $Z_{M,\gamma,n}$.
{Recall $S_0^{(1)}$ from
\eqref{eq:def-first-order-unweighted-statistic}; no quadratic correction is applied to $S_{1/2}$.}
By the mean-change model and \eqref{eq:def-noise-cusum-processes},
\[
\frac1{\sqrt n}\sum_{i=1}^n c_{i,k}\mathbf X_i
=\mathbf U_n(t_k)+\mathbf d_k.
\]
Recall the shifted quadratic process in
\eqref{eq:def-boundary-quadratic-shifted} and let
\[
\mu_{\gamma,k}
=v_k^{-2\gamma}\frac{\E\|\mathbf U_n(t_k)\|_2^2}{\sqrt p}.
\]
Then, for $\gamma\in\{0,1/2\}$ and $k\in\mathcal K_\gamma$,
\begin{equation}
\frac{W_\gamma(k)-\mu_{\gamma,k,M}}{\omega_p}
=\frac{Q_{\gamma,n,p}(t_k;\mathbf d_k)}{\omega_p}
+\frac{\mu_{\gamma,k}-\mu_{\gamma,k,M}}{\omega_p}.
\label{eq:joint-alter-shifted-statistic}
\end{equation}
By \eqref{eq:lem14-mu-Mk} and $\omega_p\asymp1$,
\[
\max_{k\in\mathcal K_\gamma}
\left|\frac{\mu_{\gamma,k}-\mu_{\gamma,k,M}}{\omega_p}\right|
\le C\sqrt p\,\eta_M
\max_{k\in\mathcal K_\gamma}v_k^{1-2\gamma}
\le C\sqrt p\,\eta_M.
\]
For the magnitude bounds needed here, $|c_k|\le nv_k$ and
Lemma~\ref{lem:projected-cusum-general}(i) give, for
$\gamma\in\{0,1/2\}$,
\[
\max_{k\in\mathcal K_\gamma}
\left|\frac{2v_k^{-2\gamma}\mathbf d_k^\top\mathbf U_n(t_k)}
{\sqrt p\,\omega_p}\right|
\le\frac{2\sqrt n}{\sqrt p\,\omega_p}
\max_{1\le k<n}|\bdelta^\top\mathbf U_n(t_k)|
=O_p(b_{\delta,n}).
\]
Also,
$\max_{k\in\mathcal K_\gamma}
v_k^{-2\gamma}\|\mathbf d_k\|_2^2/(\sqrt p\,\omega_p)
\le C\Delta_{S,n}$.
Together with the centered-noise maximum bounds in
\eqref{eq:Th3-unweighted-population} and
\eqref{eq:Th3-infeasible-boundary-limit}, applied to the noise sequence,
these yield
\[
\begin{aligned}
\max_{k\in\mathcal K_0}
\frac{W_0(k)-\mu_{0,k,M}}{\omega_p}
&=O_p(1+\Delta_{S,n}),\\
\max_{k\in\mathcal K_{1/2}}
\frac{W_{1/2}(k)-\mu_{1/2,k,M}}{\omega_p}
&=O_p\{\sqrt{\log\log h_n}+\Delta_{S,n}\}.
\end{aligned}
\]
Theorem~\ref{Th2} and these magnitude bounds give
\[
\left|
S_{0}^{(1)}
-\max_{k\in\mathcal K_0}
\frac{W_0(k)-\mu_{0,k,M}}{\omega_p}
\right|
=O_p\{r_{\mu,n}+r_{\omega,n}(1+\Delta_{S,n})\},
\]
and
\begin{align*}
&A_{\mathrm{DE}}(\log h_n)
\left|
S_{1/2}
-\max_{k\in\mathcal K_{1/2}}
 \frac{W_{1/2}(k)-\mu_{1/2,k,M}}{\omega_p}
\right|\\
&\quad=O_p\left[
\sqrt{\log\log h_n}\,r_{\mu,n}
+\{\log\log h_n+\sqrt{\log\log h_n}\,\Delta_{S,n}\}r_{\omega,n}
\right].
\end{align*}
Equation~\eqref{eq:S12-first-order-local} gives
\[
|S_{0}-S_{0}^{(1)}|
=O_p\left(p^{-1/2}+\frac{\Delta_{S,n}^2}{\sqrt p}\right).
\]
{The condition $\Delta_{S,n}p^{-1/4}\to0$ makes the preceding bound $o_p(1)$, so $S_0-S_0^{(1)}=o_p(1)$.}
No such correction is applied for $\gamma=1/2$.

The triangle inequality and
$|\max_k a_k-\max_k b_k|\le\max_k|a_k-b_k|$, together with
\eqref{eq:joint-alter-shifted-statistic}, give
\begin{equation}
\begin{aligned}
|S_0-\mathcal X_{n,0}|
&\le |S_0-S_0^{(1)}|
+\left|S_0^{(1)}-\max_{k\in\mathcal K_0}
\frac{W_0(k)-\mu_{0,k,M}}{\omega_p}\right|\\
&\quad+\left|\max_{k\in\mathcal K_0}
\frac{W_0(k)-\mu_{0,k,M}}{\omega_p}
-\max_{k\in\mathcal K_0}
\frac{Q_{0,n,p}(t_k;\mathbf d_k)}{\omega_p}\right|\\
&\le |S_0-S_0^{(1)}|
+\left|S_0^{(1)}-\max_{k\in\mathcal K_0}
\frac{W_0(k)-\mu_{0,k,M}}{\omega_p}\right|+\max_{k\in\mathcal K_0}
\left|\frac{\mu_{0,k}-\mu_{0,k,M}}{\omega_p}\right|\\
&=O_p\left\{p^{-1/2}+\frac{\Delta_{S,n}^2}{\sqrt p}
+r_{\mu,n}+r_{\omega,n}(1+\Delta_{S,n})+\sqrt p\,\eta_M\right\}
=o_p(1),
\end{aligned}
\label{eq:joint-alter-dense-feasible-bound}
\end{equation}
Here $\Delta_{S,n}p^{-1/4}\to0$ implies
$\Delta_{S,n}^2/\sqrt p\to0$.
Recall $r_{\omega,n}$ in Theorem~\ref{Th2}.
Under the conditions of Theorem~\ref{indalter}(i),
\(
r_{\mu,n}+r_{\omega,n}(1+\Delta_{S,n})+\sqrt p\,\eta_M\to0,
\)
which gives the $o_p(1)$ conclusion in
\eqref{eq:joint-alter-dense-feasible-bound}.
Furthermore, we have
\begin{equation}
\begin{aligned}
\left|\mathsf N_{n,1/2}^{\mathrm S}(S_{1/2})
-\mathcal X_{n,1/2}\right|
&=A_{\mathrm{DE}}(\log h_n)
\left|S_{1/2}-\max_{k\in\mathcal K_{1/2}}
\frac{Q_{1/2,n,p}(t_k;\mathbf d_k)}{\omega_p}\right|\\
&\le A_{\mathrm{DE}}(\log h_n)
\left|S_{1/2}-\max_{k\in\mathcal K_{1/2}}
\frac{W_{1/2}(k)-\mu_{1/2,k,M}}{\omega_p}\right|\\
&\quad+A_{\mathrm{DE}}(\log h_n)
\left|\max_{k\in\mathcal K_{1/2}}
\frac{W_{1/2}(k)-\mu_{1/2,k,M}}{\omega_p}
-\max_{k\in\mathcal K_{1/2}}
\frac{Q_{1/2,n,p}(t_k;\mathbf d_k)}{\omega_p}\right|\\
&\le A_{\mathrm{DE}}(\log h_n)
\left|S_{1/2}-\max_{k\in\mathcal K_{1/2}}
\frac{W_{1/2}(k)-\mu_{1/2,k,M}}{\omega_p}\right|\\
&\quad+A_{\mathrm{DE}}(\log h_n)
\max_{k\in\mathcal K_{1/2}}
\left|\frac{\mu_{1/2,k}-\mu_{1/2,k,M}}{\omega_p}\right|\\
&=O_p\left[\sqrt{\log\log h_n}\,r_{\mu,n}
+\{\log\log h_n+\sqrt{\log\log h_n}\,\Delta_{S,n}\}r_{\omega,n}\right.\\
&\hspace{30mm}\left.+A_{\mathrm{DE}}(\log h_n)\sqrt p\,\eta_M\right]
=o_p(1),
\end{aligned}
\label{eq:joint-alter-weighted-dense-feasible-bound}
\end{equation}
under the conditions of Theorem~\ref{indalter}(ii).

Let $E$ be any event and fix $0<\eta\le1$.
Combining \eqref{eq:joint-alter-dense-feasible-bound} and
\eqref{eq:joint-alter-weighted-dense-feasible-bound} with the marginal
comparison in \eqref{eq:joint-alter-marginal-cdf-comparisons} at
$x+\eta$ and $x-2\eta$, and \eqref{eq:joint-alter-threshold-continuity}, we obtain, uniformly over $E$,
\begin{equation}
\begin{aligned}
&\lim_{\eta\downarrow0}\limsup_{n,p\to\infty}
\left|\Pr(E,Z_{S,\gamma,n}\le x)
-\Pr(E,\mathcal X_{n,\gamma}\le x)\right|\\
&\le
\lim_{\eta\downarrow0}\limsup_{n,p\to\infty}\Pr(|Z_{S,\gamma,n}-\mathcal X_{n,\gamma}|>\eta)
+\lim_{\eta\downarrow0}\limsup_{n,p\to\infty}\Pr(|\mathcal X_{n,\gamma}-x|\le\eta)=0.
\end{aligned}
\label{eq:joint-alter-dense-feasible-threshold-bound}
\end{equation}
Taking $E$ to be the whole probability space gives 
\begin{equation}
\begin{aligned}
\lim_{\eta\downarrow0}\limsup_{n,p\to\infty}
\left|\Pr(Z_{S,\gamma,n}\le x)
-\Pr(\mathcal X_{n,\gamma}\le x)\right|=0;
\end{aligned}
\label{eq:joint-alter-dense-marginal-transfer}
\end{equation}
taking
$E=\{Z_{M,\gamma,n}\le y\}$ gives%
\begin{equation}
\begin{aligned}
&\lim_{\eta\downarrow0}\limsup_{n,p\to\infty}
\left|\Pr(Z_{S,\gamma,n}\le x,Z_{M,\gamma,n}\le y)
-\Pr(\mathcal X_{n,\gamma}\le x,Z_{M,\gamma,n}\le y)\right|=0.
\end{aligned}
\label{eq:joint-alter-dense-joint-transfer}
\end{equation}

For the max components, we compare the threshold events for
$Z_{M,\gamma,n}$ and $\mathcal Y_{n,\gamma}$ uniformly over any additional
event $E$. We first verify that $L_na_{\sigma,n}^{(\delta)}\to0$.
Lemma~\ref{lem:lrv-general} gives, under the alternative,
\begin{equation}
\begin{aligned}
r_{\sigma,n}
&=\max_{j\le p}|\widehat\sigma_j/\sigma_j-1|
=O_p\{a_{\sigma,n}^{(\delta)}\},\\
a_{\sigma,n}^{(\delta)}
&=a_{\sigma,n}
+\frac{\ell_{\max}^{3/2}\|\bdelta\|_\infty L_n^{c_{\mathrm{LRV}}}}{n}
+\frac{\ell_{\max}^2\|\bdelta\|_\infty^2}{n}.
\end{aligned}
\label{eq:joint-alter-lrv-rate}
\end{equation}
The condition $\nu<\nu_c$ in Theorem~\ref{indalter} and
Lemma~\ref{lemma:LRV-rate-log} with $A=1$ give $L_na_{\sigma,n}\to0$.
Together with $\|\bdelta\|_\infty=o(\sqrt{n}/\ell_{\max})$ in Assumption~\ref{ass:lrv-M}(iii), this yields
\[
\begin{aligned}
L_n\frac{\ell_{\max}^{3/2}\|\bdelta\|_\infty
L_n^{c_{\mathrm{LRV}}}}{n}
&=\left\{\frac{\ell_{\max}
L_n^{1+c_{\mathrm{LRV}}}}{\sqrt n}\right\}
\frac{\sqrt{\ell_{\max}}\|\bdelta\|_\infty}{\sqrt n}
=o(1),
\end{aligned}
\]
since the first factor is bounded by $L_na_{\sigma,n}$ and the second
is $o(\ell_{\max}^{-1/2})$.
The signal bound in $H_{1;n,p}$, Assumption~\ref{ass:lrv-M}(ii), and
$\nu<(1-c_{\ell^+})/(b_\delta+1/2)$ also give
\begin{equation}
\begin{aligned}
\frac{L_n\ell_{\max}^2\|\bdelta\|_\infty^2}{n}
&=O\left(\frac{\ell_{\max}^2L_n^{2b_\delta+1}}{n^2}\right)
=O\left(n^{2c_{\ell^+}+\nu(2b_\delta+1)-2}\right)=o(1).
\end{aligned}
\label{eq:joint-alter-lrv-signal-requirement}
\end{equation}
\begingroup
Thus $L_na_{\sigma,n}^{(\delta)}\to 0$.
\endgroup
Let
$r_n=\{a_{\sigma,n}^{(\delta)}/L_n\}^{1/2}$. Then
\[
\frac{r_n}{a_{\sigma,n}^{(\delta)}}
=\{L_na_{\sigma,n}^{(\delta)}\}^{-1/2}\to\infty,
\qquad
r_nL_n=\{L_na_{\sigma,n}^{(\delta)}\}^{1/2}\to0.
\]
In particular, $r_n\to0$, and \eqref{eq:joint-alter-lrv-rate} implies
$\Pr(r_{\sigma,n}>r_n)\to0$. We now use the event
\begingroup
$G_{r_n}=\{r_{\sigma,n}\le r_n,\ r_n<1/2\}$ to compare the threshold events
\endgroup
$\{Z_{M,\gamma,n}\le y\}$ and $\{\mathcal Y_{n,\gamma}\le y\}$.
For fixed $y\in\mathbb R$ and all sufficiently large $n,p$, let
\[
\begin{aligned}
u_{n,0}(y)&=\left\{\frac{y+\log(2p)}2\right\}^{1/2},
\qquad
u_{n,1/2}(y)&=
\frac{y+D_{\mathrm{DE}}(p\log h_n)}{A_{\mathrm{DE}}(p\log h_n)}.
\end{aligned}
\]
Then 
\(
\{Z_{M,\gamma,n}\le y\}
=\{M_\gamma\le u_{n,\gamma}(y)\}.
\)
\begingroup
Both thresholds are positive. On $G_{r_n}$ we have $0<r_n<1/2$ and
\endgroup
\[
\frac{M_{\gamma,n,p}^{X,0}}{1+r_n}
\le M_\gamma\le
\frac{M_{\gamma,n,p}^{X,0}}{1-r_n}.
\]
Consequently, on the same event,
\begin{equation}
\begin{aligned}
\{M_{\gamma,n,p}^{X,0}\le(1-r_n)u_{n,\gamma}(y)\}
&\subseteq\{Z_{M,\gamma,n}\le y\}\subseteq
\{M_{\gamma,n,p}^{X,0}\le(1+r_n)u_{n,\gamma}(y)\}.
\end{aligned}
\label{eq:joint-alter-max-threshold-sandwich}
\end{equation}
With the definitions in \eqref{eq:max-normalizations},
\[
\begin{aligned}
\mathsf N_{n,0}^{\mathrm M}\{(1\pm r_n)u_{n,0}(y)\}-y
&=(\pm2r_n+r_n^2)\{y+\log(2p)\},\\
\mathsf N_{n,1/2}^{\mathrm M}\{(1\pm r_n)u_{n,1/2}(y)\}-y
&=\pm r_n\{y+D_{\mathrm{DE}}(p\log h_n)\}.
\end{aligned}
\]
Hence, for a constant $C_y$ depending only on the fixed threshold $y$,
\(
\left|
\mathsf N_{n,\gamma}^{\mathrm M}
\{(1\pm r_n)u_{n,\gamma}(y)\}-y
\right|\le C_y r_nL_n.
\)
Recalling that $\mathcal Y_{n,\gamma}
=\mathsf N_{n,\gamma}^{\mathrm M}(M_{\gamma,n,p}^{X,0})$,
the threshold sandwich therefore gives
\begingroup
\begin{equation}
\begin{aligned}
&\sup_E\left|\Pr(E,Z_{M,\gamma,n}\le y)
-\Pr(E,\mathcal Y_{n,\gamma}\le y)\right|\\
&\qquad\le \Pr(r_{\sigma,n}>r_n)+\Pr(r_n\ge1/2)
+\Pr(|\mathcal Y_{n,\gamma}-y|\le C_y r_nL_n).
\end{aligned}
\label{eq:joint-alter-max-feasible-bound}
\end{equation}
\endgroup

The marginal comparison in \eqref{eq:joint-alter-marginal-cdf-comparisons}
and \eqref{eq:joint-alter-threshold-continuity} imply
\[
\lim_{\eta\downarrow0}\limsup_{n,p\to\infty}
\Pr(|\mathcal Y_{n,\gamma}-y|\le\eta)=0.
\]
\begingroup
For every fixed $\eta>0$,
\[
\Pr(|\mathcal Y_{n,\gamma}-y|\le C_y r_nL_n)
\le\Pr(|\mathcal Y_{n,\gamma}-y|\le\eta)
+\Pr(C_y r_nL_n>\eta).
\]
The second term tends to zero because $r_nL_n\to0$ in probability. Thus, taking
\endgroup
$n,p\to\infty$ and then $\eta\downarrow0$ shows that the band probability
in \eqref{eq:joint-alter-max-feasible-bound} tends to zero. Together with
\begingroup
$\Pr(r_{\sigma,n}>r_n)+\Pr(r_n\ge1/2)\to0$, this gives
\endgroup
\[
\sup_E\left|\Pr(E,Z_{M,\gamma,n}\le y)
-\Pr(E,\mathcal Y_{n,\gamma}\le y)\right|\to0.
\]
Taking $E$ to be the whole probability space and
$E=\{\mathcal X_{n,\gamma}\le x\}$, and combining with
\eqref{eq:joint-alter-dense-marginal-transfer} and
\eqref{eq:joint-alter-dense-joint-transfer}, proves
\eqref{eq:joint-alter-feasible-transfer-target-1} and
\eqref{eq:joint-alter-feasible-transfer-target-2}.
\end{proof}

\begin{proof}[Proof of Proposition~\ref{prop:cauchy-consistency}]
Let $L_n=\log(np)$, $H_n=\log\log h_n$ and $J_n=\log\{p\log h_n\}$. For $\mathfrak c\in\{S,M\}$ and $\gamma\in\{0,1/2\}$, let 
\[
p_{\mathfrak c,\gamma,n,p}=1-F_{\mathfrak c,\gamma}(Z_{\mathfrak c,\gamma,n}),\qquad
F_{\mathfrak c,\gamma}=
\begin{cases}
F_V, & \text{if }(\mathfrak c,\gamma)=(S,0),\\
F_{\mathrm{Gu}}, & \text{otherwise}.
\end{cases}
\]
We first prove that, for every fixed $u\in(0,1)$, 
\begin{equation}\label{eq:cauchy-no-upper-accumulation}
\limsup_{n,p\to\infty}\Pr\left(p_{\mathfrak c,\gamma,n,p}>1-u\right)\le u
\end{equation} holds for any $\mathfrak c\in\{S,M\}$ and $\gamma\in\{0,1/2\}$.

Recall $Z_{S,0,n}=S_0$, $Z_{S,1/2,n}=A_{\mathrm{DE}}(\log h_n)S_{1/2}-D_{\mathrm{DE}}(\log h_n)$, $Z_{M,0,n}=2M_0^2-\log(2p)$ and $Z_{M,1/2,n}=A_{\mathrm{DE}}(p\log h_n)M_{1/2}-D_{\mathrm{DE}}(p\log h_n)$.
Let $q_{\mathfrak c,\gamma}(u)
=F_{\mathfrak c,\gamma}^{-1}(u)$. Since the limiting distribution functions $F_V$ and $F_{\mathrm{Gu}}$ are continuous, $F_{\mathfrak c,\gamma}
\{q_{\mathfrak c,\gamma}(u)\}=u$. Moreover, by the monotonicity of $F_{\mathfrak c,\gamma}$,
$\left\{
 p_{\mathfrak c,\gamma,n,p}>1-u
\right\}
\subseteq
\left\{
 Z_{\mathfrak c,\gamma,n}
 \le q_{\mathfrak c,\gamma}(u)
\right\}$.

Let $Z_{\mathfrak c,\gamma,n}^{G}(\bdelta)$ denote the corresponding
covariance-matched Gaussian statistic with the deterministic CUSUM drift
induced by $\bdelta$. For $(\mathfrak c,\gamma)=(S,0)$, explicitly,
\[
Z_{S,0,n}^{G}(\bdelta)
:=\max_{k\in\mathcal K_0}
\frac{\|\mathbf U_n^G(t_k)+\mathbf d_k\|_2^2
-\E\|\mathbf U_n^G(t_k)\|_2^2}{\sqrt p\,\omega_p}
=\mathcal X_{n,0}^{G},
\]
where $\mathbf d_k=-\sqrt nK_B(t_k,\theta)\bdelta$, and
$\mathbf U_n^G$ and $\mathcal X_{n,0}^{G}$ are defined in
\eqref{eq:def-noise-cusum-processes} and
\eqref{eq:joint-alter-normalized-statistics}, respectively.
This is the Gaussian counterpart of the first-order statistic
$S_0^{(1)}$ in \eqref{eq:def-first-order-unweighted-statistic}.
The error in replacing $S_0$ by $S_0^{(1)}$, bounded in
\eqref{eq:S12-first-order-local}, is included in $\rho_{S,0,n}(\eta)$ below. Let
$Z_{\mathfrak c,\gamma}^{\infty}$ denote a random
variable with distribution function $F_{\mathfrak c,\gamma}$. The Gaussian
null-limit results established for the component statistics give
$Z_{\mathfrak c,\gamma,n}^{G}(\mathbf0)
\ \Rightarrow\
Z_{\mathfrak c,\gamma}^{\infty}$. Consequently, since $q_{\mathfrak c,\gamma}(u)$ is a continuity point of $F_{\mathfrak c,\gamma}$,
\begin{equation}
\Pr\left\{
Z_{\mathfrak c,\gamma,n}^{G}(\mathbf0)
\le q_{\mathfrak c,\gamma}(u)
\right\}
=
u+o(1).
\label{eq:cauchy-gaussian-null-quantile}
\end{equation}
\begingroup
For fixed $x$, let
\[
u_{S,\gamma,n}(x)=(\mathsf N_{n,\gamma}^{\mathrm S})^{-1}(x),\qquad
u_{M,\gamma,n}(x)=(\mathsf N_{n,\gamma}^{\mathrm M})^{-1}(x),
\]
on the range of the corresponding unnormalized statistic. In the stacked
raw Gaussian CUSUM coordinates $\mathbf z=(\mathbf z_k)_k$, the
zero-drift acceptance sets are
\[
\begin{aligned}
\mathcal D_{S,\gamma,n}(x)
&=\bigcap_{k\in\mathcal K_\gamma}
\left\{\mathbf z:\|\mathbf z_k\|_2^2
\le\E\|\mathbf U_n^G(t_k)\|_2^2
+\sqrt p\,\omega_p v_k^{2\gamma}u_{S,\gamma,n}(x)\right\},\\
\mathcal D_{M,\gamma,n}(x)
&=\bigcap_{k\in\mathcal K_\gamma,\,j\le p}
\left\{\mathbf z:|z_{k,j}|
\le\sigma_jv_k^\gamma u_{M,\gamma,n}(x)\right\}.
\end{aligned}
\]
A negative radius makes the corresponding set empty; otherwise these
sets are convex and symmetric. Thus Anderson's inequality
\citep{anderson1955integral}, applied to the deterministic translation
$\mathbf d=(\mathbf d_k)_k$, gives
\begin{equation}
\Pr\{Z_{\mathfrak c,\gamma,n}^{G}(\bdelta)\le x\}
=\Pr\{\mathbf U_n^G+\mathbf d\in\mathcal D_{\mathfrak c,\gamma,n}(x)\}
\le\Pr\{Z_{\mathfrak c,\gamma,n}^{G}(\mathbf0)\le x\}.
\label{eq:cauchy-anderson}
\end{equation}
{For $S_0$, this comparison applies to $S_0^{(1)}$; the error $S_0-S_0^{(1)}$ is included in $\rho_{S,0,n}(\eta)$ below.}
\endgroup

For each $\gamma\in\{0,1/2\}$, apply
Lemma~\ref{lem:block-comparison-general} with $c_{\mathrm{out}}=3/2$
under the corresponding conditions of Proposition~\ref{prop:cauchy-consistency}.
Writing $\lambda_0=1$ and $\lambda_{1/2}=\lambda$, its growth condition holds because
\[
0<\nu<\lambda_\gamma\nu_{\mathrm{dep}}
<\frac{\lambda_\gamma}{\nu_{\mathrm{dep}}^{-1}-1}.
\]
For this choice of block parameters, \eqref{eq:joint-rate-vanishing} gives
\(
L_n^{3/2}r_{n,\gamma}\to0.
\)
Fix $0<\eta\le1$. The marginal cutoffs used in
\eqref{eq:joint-alter-threshold-derivatives}, centered at
$q_{\mathfrak c,\gamma}(u)+\eta$, have global $C^3$ norms bounded by
$C_u\eta^{-3}L_n^{3/2}$.
\begingroup
With $\beta_n=L_n^{7/4}$, replacing a smooth maximum by the ordinary maximum
changes each cutoff expectation by at most $C_u\eta^{-1}L_n^{-1/4}$,
\endgroup
by \eqref{eq:smax-logsumexp-bound}.

Let $Z_{S,\gamma,n}^{\circ}=\mathcal X_{n,\gamma}$ and
$Z_{M,\gamma,n}^{\circ}=\mathcal Y_{n,\gamma}$ denote the population-standardized
statistics in \eqref{eq:joint-alter-normalized-statistics}.
In the marginal comparison, take $\mathbf h_S=\mathbf d$ for the dense
component and $\mathbf h_M=\mathbf0$ for the maximum component, whose
deterministic coordinate shifts are retained in the linear array.
The comparison lemma and the cutoff bounds therefore give
\begin{equation}
\Pr\left\{
 Z_{\mathfrak c,\gamma,n}
 \le q_{\mathfrak c,\gamma}(u)
\right\}
\le
\Pr\left\{
 Z_{\mathfrak c,\gamma,n}^{G}(\bdelta)
 \le q_{\mathfrak c,\gamma}(u)+2\eta
\right\}
+\rho_{\mathfrak c,\gamma,n}(\eta),
\label{eq:cauchy-feasible-comparison}
\end{equation}
where, with $H_{\gamma,n}$ as in Lemma~\ref{lem:block-comparison-general},
we may take
\begingroup
\[
\begin{aligned}
\rho_{\mathfrak c,\gamma,n}(\eta)
&:= C_u\left[
\eta^{-3}L_n^{3/2}r_{n,\gamma}
\{1+H_{\gamma,n}(\mathbf h_{\mathfrak c})\}^{3}
+\eta^{-1}L_n^{-1/4}\right]\\
&\quad+\Pr\left(
|Z_{\mathfrak c,\gamma,n}-Z_{\mathfrak c,\gamma,n}^{\circ}|>\eta
\right).
\end{aligned}
\]
\endgroup
We next verify that $\rho_{\mathfrak c,\gamma,n}(\eta)\to0$
for every fixed $0<\eta\le1$ on the corresponding bounded-signal subsequences.

For the dense components, the centering and scaling bounds from
Theorem~\ref{Th2}, together with \eqref{eq:S12-first-order-local}, give,
on subsequences with $\Delta_{S,n}=O(1)$ and
$\Delta_{S,n}/H_n^{1/2}=O(1)$, respectively,
\[
\begin{aligned}
|Z_{S,0,n}-Z_{S,0,n}^{\circ}|
&=O_p\left\{
r_{\mu,n}+(1+\Delta_{S,n})r_{\omega,n}
+p^{-1/2}
+\frac{\Delta_{S,n}^2}{\sqrt p}\right\},\\
|Z_{S,1/2,n}-Z_{S,1/2,n}^{\circ}|
&=O_p\left[H_n^{1/2}r_{\mu,n}
+\{H_n+H_n^{1/2}\Delta_{S,n}\}r_{\omega,n}\right].
\end{aligned}
\]
Assumption~\ref{ass:C2} yields
\[
\frac{n\bdelta^\top\mathbf{\Omega}\bdelta}{p}
\le C\frac{n\|\bdelta\|_2^2}{p}
=C\frac{\Delta_{S,n}}{\sqrt p}.
\]
Under $\Delta_{S,n}=O(1)$,
\begin{equation*}
\frac{n\bdelta^\top\mathbf{\Omega}\bdelta}{p}
=O(p^{-1/2}),
\qquad
p^{-1/2}+\frac{\Delta_{S,n}^2}{\sqrt p}
=O(p^{-1/2}).
\end{equation*}
Theorem~\ref{Th2}, together with $M=\left\lceil(n\wedge p)^{1/8}\right\rceil,
\|\bdelta\|_2^2
=o(n\sqrt p/M^2)$ gives $r_{\mu,n}+r_{\omega,n}\rightarrow0$. Consequently,
$|Z_{S,0,n}-Z_{S,0,n}^{\circ}|=o_p(1)$.
Under $\Delta_{S,n}/H_n^{1/2}=O(1)$,
\begin{equation*}
H_n^{1/2}\Delta_{S,n}=O(H_n),
\qquad
\frac{n\bdelta^\top\mathbf{\Omega}\bdelta}{p}
=O\left(\frac{H_n^{1/2}}{\sqrt p}\right),
\end{equation*}
and hence
$H_n^{1/2}r_{\mu,n}+H_nr_{\omega,n}\rightarrow0$,
$|Z_{S,1/2,n}-Z_{S,1/2,n}^{\circ}|=o_p(1)$.

For the maximum components, recall $r_{\sigma,n}$ and
$a_{\sigma,n}^{(\delta)}$ from \eqref{eq:joint-alter-lrv-rate};
Lemma~\ref{lem:lrv-general} gives
$r_{\sigma,n}=O_p\{a_{\sigma,n}^{(\delta)}\}$.
Let $K_{0,n}=\log(2p)$ and $K_{1/2,n}=J_n$. Let $M_\gamma^{\circ}$ denote the corresponding statistic standardized by
the true long-run standard deviations. On a subsequence satisfying $\sqrt n\|\bdelta\|_\infty/\sqrt{\log p}=O(1)$ for $\gamma=0$, or $\sqrt n\|\bdelta\|_\infty/\sqrt{J_n}=O(1)$ for $\gamma=1/2$, the residual maximum and deterministic drift bounds in the proof of Proposition~\ref{prop:consistency-Linf} give $M_0^{\circ}=O_p\{K_{0,n}^{1/2}\}$ and $M_{1/2}^{\circ}=O_p\{K_{1/2,n}^{1/2}\}$. On $\{r_{\sigma,n}\le1/2\}$,
$|M_\gamma-M_\gamma^{\circ}|
\le2r_{\sigma,n}M_\gamma^{\circ}$. Therefore,
\begin{equation*}
\left|
Z_{M,0,n}-Z_{M,0,n}^{\circ}
\right|
=O_p\{a_{\sigma,n}^{(\delta)}K_{0,n}\},
~~~~\left|
Z_{M,1/2,n}-Z_{M,1/2,n}^{\circ}
\right|
=O_p\{a_{\sigma,n}^{(\delta)}K_{1/2,n}\}.
\end{equation*}
The conditions $\nu<\nu_c$ and
$\|\bdelta\|_\infty=o\{\sqrt n/(\ell_{\max}\sqrt{K_{\gamma,n}})\}$
give $a_{\sigma,n}^{(\delta)}K_{\gamma,n}\to0$.
Hence $|Z_{M,\gamma,n}-Z_{M,\gamma,n}^{\circ}|=o_p(1)$.

For the dense components, the bounds above and $|c_k|\le nv_k$ imply
\[
H_{\gamma,n}(\mathbf h_S)
\le\frac{\sqrt n\|\bdelta\|_2}{\sqrt p}
\max_{k\in\mathcal K_\gamma}v_k^{1-\gamma}
\le C\left(\frac{n\bdelta^\top\bOme\bdelta}{p}\right)^{1/2}=o(1),
\]
whereas $H_{\gamma,n}(\mathbf h_M)=0$.
Thus the preceding bounds verify
$\rho_{\mathfrak c,\gamma,n}(\eta)\to0$ for every fixed $\eta>0$
on each corresponding bounded-signal subsequence.

Define the normalized signal strength by
\[
\mathfrak s_{\mathfrak c,\gamma,n,p}=
\begin{cases}
\Delta_{S,n}, & (\mathfrak c,\gamma)=(S,0),\\
\Delta_{S,n}/\sqrt{H_n}, & (\mathfrak c,\gamma)=(S,1/2),\\
\sqrt n\|\bdelta\|_\infty/\sqrt{\log p},
& (\mathfrak c,\gamma)=(M,0),\\
\sqrt n\|\bdelta\|_\infty/\sqrt{J_n},
& (\mathfrak c,\gamma)=(M,1/2).
\end{cases}
\]
Combining \eqref{eq:cauchy-anderson} and
\eqref{eq:cauchy-feasible-comparison} with the Gaussian null limit,
on every subsequence with $\mathfrak s_{\mathfrak c,\gamma,n,p}=O(1)$
and for every fixed $0<\eta\le1$,
\begin{equation*}
\begin{aligned}
\Pr\left\{
 p_{\mathfrak c,\gamma,n,p}>1-u
\right\}
&\le
\Pr\left\{
 Z_{\mathfrak c,\gamma,n}
 \le q_{\mathfrak c,\gamma}(u)
\right\}\\
&\le
\Pr\left\{
 Z_{\mathfrak c,\gamma,n}^{G}(\bdelta)
 \le q_{\mathfrak c,\gamma}(u)+2\eta
\right\}
+\rho_{\mathfrak c,\gamma,n}(\eta)\\
&\le
\Pr\left\{
 Z_{\mathfrak c,\gamma,n}^{G}(\mathbf0)
 \le q_{\mathfrak c,\gamma}(u)+2\eta
\right\}
+\rho_{\mathfrak c,\gamma,n}(\eta)\\
&=F_{\mathfrak c,\gamma}\{q_{\mathfrak c,\gamma}(u)+2\eta\}
+o(1)+\rho_{\mathfrak c,\gamma,n}(\eta)\\
&=F_{\mathfrak c,\gamma}\{q_{\mathfrak c,\gamma}(u)+2\eta\}+o(1).
\end{aligned}
\end{equation*}
Taking $n,p\to\infty$ along this subsequence and then $\eta\downarrow0$,
continuity of $F_{\mathfrak c,\gamma}$ gives
\[
\limsup_{n,p\to\infty}
\Pr\left(p_{\mathfrak c,\gamma,n,p}>1-u\right)
\le\lim_{\eta\downarrow0}
F_{\mathfrak c,\gamma}\{q_{\mathfrak c,\gamma}(u)+2\eta\}=u.
\]
On any subsequence with $\mathfrak s_{\mathfrak c,\gamma,n,p}\to\infty$,
Proposition~\ref{prop:Tms-alter} or Proposition~\ref{prop:consistency-Linf}
gives
\[
p_{\mathfrak c,\gamma,n,p}\xrightarrow{p}0,
\qquad
\Pr\left(p_{\mathfrak c,\gamma,n,p}>1-u\right)\to0.
\]
For an arbitrary subsequence $(n_m,p_m)\to(\infty,\infty)$, choose a
further subsequence such that
\[
\mathfrak s_{\mathfrak c,\gamma,n_{m_\ell},p_{m_\ell}}
\longrightarrow s_\star\in[0,\infty].
\]
The two cases above imply
\[
\limsup_{\ell\to\infty}
\Pr\left(p_{\mathfrak c,\gamma,n_{m_\ell},p_{m_\ell}}>1-u\right)
\le
\begin{cases}
u, & s_\star<\infty,\\
0, & s_\star=\infty.
\end{cases}
\]
Since the initial subsequence was arbitrary,
\[
\limsup_{n,p\to\infty}
\Pr\left(p_{\mathfrak c,\gamma,n,p}>1-u\right)\le u,
\]
which proves \eqref{eq:cauchy-no-upper-accumulation}.
Consequently, for every fixed $K>1$,
\begin{equation}
\limsup_{n,p\to\infty}
\Pr\left[
\{1-p_{\mathfrak c,\gamma,n,p}\}^{-1}>K
\right]
\le K^{-1},
\label{eq:cauchy-upper-tightness}
\end{equation}
and hence $\{1-p_{\mathfrak c,\gamma,n,p}\}^{-1}=O_p(1)$.

\begingroup
Define $G(u)=\cot(\pi u)$ for $0<u<1$, with $G(1)=-\infty$ as in the endpoint convention of the main article. The Cauchy-combination statistic of \citet{LiuXie2020Cauchy} is
\endgroup
\begin{equation}
T_{CC,\gamma}
=\frac12G(p_{S_\gamma,n,p})
 +\frac12G(p_{M_\gamma,n,p}).
\label{eq:cauchy-decomposition}
\end{equation}
As $u\downarrow0$,
\begin{equation}
G(u)=\frac{1}{\pi u}+O(u)
\rightarrow+\infty.
\label{eq:cauchy-positive-tail}
\end{equation}
For every $u\in(0,1)$,
\begin{equation}
G(u)\ge-\frac{1}{\pi(1-u)}.
\label{eq:cauchy-negative-tail}
\end{equation}
Indeed, \eqref{eq:cauchy-negative-tail} is immediate for $u\le1/2$.  If
$u>1/2$, then $x=\pi(1-u)\in(0,\pi/2)$ and $\tan x\ge x$, so
\begin{equation*}
-G(u)=\cot x\le x^{-1}=\{\pi(1-u)\}^{-1}.
\end{equation*}

If $\Delta_{S,n}\to\infty$, Proposition~\ref{prop:Tms-alter}(i) gives $p_{S_0,n,p}\xrightarrow{p}0$. Equations~\eqref{eq:cauchy-positive-tail},
\eqref{eq:cauchy-upper-tightness}, and
\eqref{eq:cauchy-negative-tail} yield
$G(p_{S_0,n,p})\xrightarrow{p}+\infty$ and $G(p_{M_0,n,p})\ge-O_p(1)$. Hence $T_{CC,0}\xrightarrow{p}+\infty$. If
$\sqrt n\|\bdelta\|_\infty/\sqrt{\log p}\rightarrow\infty$, Proposition~\ref{prop:consistency-Linf}(i) gives
$p_{M_0,n,p}\xrightarrow{p}0$, and the same calculation with $S_0$ and $M_0$ interchanged again gives
$T_{CC,0}\xrightarrow{p}+\infty$. Since, by definition,
\begin{equation*}
p_{CC,0}
=1-F_C(T_{CC,0})
=\frac12-\frac1\pi\arctan(T_{CC,0}),
\end{equation*}
we obtain $p_{CC,0}\xrightarrow{p}0$.

If $\Delta_{S,n}/H_n^{1/2}\to\infty$,
Proposition~\ref{prop:Tms-alter}(ii) gives $p_{S_{1/2},n,p}\xrightarrow{p}0$. If
$\sqrt n\|\bdelta\|_\infty/\sqrt{J_n}\rightarrow\infty$, Proposition~\ref{prop:consistency-Linf}(ii) gives$p_{M_{1/2},n,p}\xrightarrow{p}0$. In either case,
\eqref{eq:cauchy-positive-tail} makes the Cauchy transform of the consistent
component diverge to $+\infty$, while
\eqref{eq:cauchy-upper-tightness} and
\eqref{eq:cauchy-negative-tail} bound the negative contribution of the other
component by $O_p(1)$.  Thus
$T_{CC,1/2}\xrightarrow{p}+\infty$ and $p_{CC,1/2}\xrightarrow{p}0$. For every fixed $\alpha\in(0,1)$ and $\gamma\in\{0,1/2\}$,
$\Pr\{p_{CC,\gamma}\le\alpha\}\rightarrow1$. This proves the result.
\end{proof}

\section{Proof of Theorems and Corollary in Section~\ref{sec:location}%
}
Recall $C_{\gamma,j}^{\epsilon,0}(k)$ and $M_{n,p}^{\epsilon,\gamma}$ from
\eqref{eq:max-coordinate-arrays} and \eqref{eq:def-max-coordinate-statistics}.

\begin{proof}[Proof of Theorem~\ref{consistency-S}]
(i). Let $t_k=k/n$ and $\theta=\tau/n$. By the decomposition $Q_{0,k}=\mathfrak g_{S,0}(k)+\mathfrak r_{S,0}(k)$, where
\begin{align*}
&\mathfrak g_{S,0}(k)
=\frac{\|\mathbf d_k\|_2^2}{\sqrt p},
\qquad
\mathbf d_k
=
-\frac{k(n-\tau)}{n\sqrt n}\bdelta\ind{k\le\tau}
-\frac{\tau(n-k)}{n\sqrt n}\bdelta\ind{k>\tau},\\
&\mathfrak r_{S,0}(k)
=
\frac{\|\mathbf U_n(t_k)\|_2^2
-\E\|\mathbf U_n(t_k)\|_2^2}{\sqrt p}
+\frac{2\mathbf d_k^\top\mathbf U_n(t_k)}{\sqrt p}
+\mu_k-\widehat\mu_{k,M}.
\end{align*}
Then, 
\[
\mathfrak g_{S,0}(\tau)-\mathfrak g_{S,0}(k)
=
\begin{cases}
\dfrac{(n-\tau)^2(\tau-k)(\tau+k)}
      {n^3\sqrt p}\|\bdelta\|_2^2,
& k\le\tau,\\[3mm]
\dfrac{\tau^2(k-\tau)(2n-k-\tau)}
      {n^3\sqrt p}\|\bdelta\|_2^2,
& k>\tau.
\end{cases}
\]
Since $\tau/n\in[\vartheta,1-\vartheta]$, {we have}
\begin{equation}\label{eq:thm7-proof-1}
    \mathfrak g_{S,0}(\tau)-\mathfrak g_{S,0}(k)
\ge \vartheta^3\frac{|k-\tau|\|\bdelta\|_2^2}{\sqrt p}.
\end{equation}
By Theorem~\ref{Th1}, Lemma~\ref{lem:projected-cusum-general} %
and Lemma~\ref{lem:centering-general}, 
\begin{equation}\label{eq:thm7-proof-2}
    \max_{1\leq k<n}|\mathfrak r_{S,0}(k)|
=O_p\left[1
+\left\{\frac{n\|\bdelta\|_2^2}{p}\right\}^{1/2}+\sqrt p\eta_M+r_{\mu,n}\right].
\end{equation}
Since $\widehat\tau_{S,0}$ maximizes $Q_{\gamma,k}=\mathfrak g_{S,0}(k)+\mathfrak r_{S,0}(k)$,
\begin{equation}\label{eq:thm7-proof-3}
    \mathfrak g_{S,0}(\tau)-\mathfrak g_{S,0}(\widehat\tau_{S,0})\le
    \mathfrak r_{S,0}(\widehat\tau_{S,0})-\mathfrak r_{S,0}(\tau)\le
    2\max_k|\mathfrak r_{S,0}(k)|.
\end{equation}
Applying \eqref{eq:thm7-proof-1} with $k=\widehat\tau_{S,0}$, with \eqref{eq:thm7-proof-2}--\eqref{eq:thm7-proof-3}, {we have}
\[
\frac{|\widehat\tau_{S,0}-\tau|}{n}\le
\frac{2\sqrt p}{\vartheta^3 n\|\bdelta\|_2^2}
\max_k|\mathfrak r_{S,0}(k)|
=O_p\left\{\frac{\sqrt p}{n\|\bdelta\|_2^2}
+\frac1{\sqrt{n}\|\bdelta\|_2}\right\}.
\]

(ii) {As in part (i), decompose} 
\[
W_{1/2}(k)-\widehat{\mu}_{1/2,k,M}
=\mathfrak g_{S,1/2}(k)+\mathfrak r_{S,1/2}(k),
\]
where $\mathfrak g_{S,1/2}(k)=v_k^{-1}\mathfrak g_{S,0}(k)$, and $\mathfrak r_{S,1/2}(k)=v_k^{-1}\mathfrak r_{S,0}(k)$ with
$v_k=(k/n)(1-k/n)$.
Then
\[
\mathfrak g_{S,1/2}(\tau)-\mathfrak g_{S,1/2}(k)
=\frac{\|\bdelta\|_2^2}{\sqrt p}
\begin{cases}
\dfrac{(1-\tau/n)(\tau-k)}{1-k/n},&k\le\tau,\\[2mm]
\dfrac{(\tau/n)(k-\tau)}{k/n},&k>\tau,
\end{cases}
\]
and therefore
\[
\mathfrak g_{S,1/2}(\tau)-\mathfrak g_{S,1/2}(k)
\ge \vartheta^3\frac{|k-\tau|\|\bdelta\|_2^2}{\sqrt p}.
\]
\begingroup
Apply Lemmas~\ref{lem:quadratic-only-uniform-comparison} and
\ref{lem:quadratic-only-bridge-comparison} to the quadratic array and
its negative. Their proofs are unchanged by multiplication of any
coordinate by a deterministic sign; all bounds use absolute derivative
sums. The Gaussian grid argument in Lemma~\ref{lem:dense-boundary-extreme}
also applies to the signed grid, whose size is at most twice the original
size. The corresponding Gaussian Ornstein--Uhlenbeck maximum in absolute
value is $O_p(\sqrt{\log\log h_n})$. The additional dimension condition
in part (ii) makes both quadratic comparison errors negligible at this
scale. Hence
\endgroup
\[
\max_{k\in\mathcal K_{1/2}}
\left|v_k^{-1}
\frac{\|\mathbf U_n(t_k)\|_2^2-\E\|\mathbf U_n(t_k)\|_2^2}{\sqrt p}
\right|
=O_p\{\sqrt{\log\log h_n}\}.
\]
The weighted result in Lemma~\ref{lem:projected-cusum-general} gives
\[
\max_{k\in\mathcal K_{1/2}}
\left|v_k^{-1}\frac{2\mathbf d_k^\top\mathbf U_n(t_k)}{\sqrt p}\right|
=O_p\left[
\sqrt{\log\log h_n}
\left\{\frac{n\|\bdelta\|_2^2}{p}\right\}^{1/2}
\right].
\]
and by Lemma~\ref{lem:centering-general}, 
\[
\max_{k\in\mathcal K_{1/2}}|\mathfrak r_{S,1/2}(k)|
=O_p\left[\sqrt{\log\log h_n}
\left\{1+\left(\frac{n\|\bdelta\|_2^2}{p}\right)^{1/2}\right\}
+r_{\mu,n}+\sqrt p\eta_M\right].
\]
Finally, {we have}
\[
\frac{|\widehat\tau_{S,1/2}-\tau|}{n}
=O_p\left[\sqrt{\log\log h_n}\left\{
\frac{\sqrt p}{n\|\bdelta\|_2^2}
+\frac1{\sqrt{n}\|\bdelta\|_2}\right\}\right].
\]
\end{proof}

\begin{proof}[Proof of Theorem~\ref{consistency-M}]
\begingroup
{We first bound the deterministic CUSUM gap by the uniform noise error, then show that the selected coordinate has signal magnitude comparable to $\|\bdelta\|_\infty$.} Choose a maximizing pair
\[
(\widehat k,\widehat j)\in
\operatorname*{arg\,max}_{k\in\mathcal K_\gamma,\,j\le p}
|C_{\gamma,j}(k)|^2,
\qquad \widehat k=\widehat\tau_{M,\gamma}.
\]
\endgroup
Recall $t_k=k/n$, $\theta=\tau/n$, $v(t)=t(1-t)$, $v_k=v(t_k)$,
and $K_B(s,t)=s\wedge t-st$.
Since $\mathbf d_k=-\sqrt n\,K_B(t_k,\theta)\bdelta$, for $\gamma\in\{0,1/2\}$,
\[
d_{\gamma,j}(k)=v_k^{-\gamma}\frac{d_{k,j}}{\sigma_j}
=-\sqrt n\,\frac{\delta_j}{\sigma_j}v_k^{-\gamma}K_B(t_k,\theta),
\]
and the population-standardized decomposition is
\[
C_{\gamma,j}^{X,0}(k)=C_{\gamma,j}^{\epsilon,0}(k)+d_{\gamma,j}(k).
\]
For $0<t<1$, direct calculation gives
\[
K_B(\theta,\theta)^2-K_B(t,\theta)^2
=\begin{cases}
(1-\theta)^2(\theta-t)(\theta+t),&t\le\theta,\\
\theta^2(t-\theta)(2-\theta-t),&t>\theta,
\end{cases}
\]
and
\[
\frac{K_B(\theta,\theta)^2}{v(\theta)}
-\frac{K_B(t,\theta)^2}{v(t)}
=\begin{cases}
(1-\theta)(\theta-t)/(1-t),&t\le\theta,\\
\theta(t-\theta)/t,&t>\theta.
\end{cases}
\]
Thus, {we have} for $k\in\mathcal K_\gamma$,
\[
d_{\gamma,j}^2(\tau)-d_{\gamma,j}^2(k)
\ge c_\vartheta|k-\tau|\frac{\delta_j^2}{\sigma_j^2}.
\]
Lemma~\ref{lem:max-transfer-general} implies
\[
M_{n,p}^{\epsilon,\gamma}
=O_p(\sqrt{L_n}),
\qquad L_n=\log(np).
\]
Let $\mathfrak a_j=|\delta_j|/\sigma_j$, $\mathfrak a_{\max}=\max_{1\leq j\le p}\mathfrak a_j$ and $j_\star\in\argmax_{1\leq j\le p}\mathfrak a_j$.
Since $\theta\in[\vartheta,1-\vartheta]$,
$v(\theta)^{-\gamma}K_B(\theta,\theta)=v(\theta)^{1-\gamma}\ge c_\vartheta$
for $\gamma\in\{0,1/2\}$. We consider the event
\[
\mathcal E_{n,\gamma}^{\mathrm{loc,M}}=\left\{
M_{n,p}^{\epsilon,\gamma}\le C_0\sqrt{\log(np)},
\quad r_{\sigma,n}\le 1/8,
\quad \sqrt n\,\mathfrak a_{\max}v(\theta)^{1-\gamma}
\ge8M_{n,p}^{\epsilon,\gamma}
\right\}.
\]
For sufficiently large fixed $C_0$, the tail bound established in the proof of
Lemma~\ref{lem:max-transfer-general} gives
$\Pr\{M_{n,p}^{\epsilon,\gamma}\le C_0\sqrt{L_n}\}\to1$.
Together with Lemma~\ref{lem:lrv-general} and
$\sqrt n\,\mathfrak a_{\max}/\sqrt{L_n}\to\infty$, this yields
$\Pr(\mathcal E_{n,\gamma}^{\mathrm{loc,M}})\to1$.
\begingroup
For the selected coordinate, its positive denominator is constant in $k$;
hence the maximizing property gives
\[
|C_{\gamma,\widehat j}^{\epsilon,0}(\widehat k)
  +d_{\gamma,\widehat j}(\widehat k)|^2
\ge
|C_{\gamma,\widehat j}^{\epsilon,0}(\tau)
  +d_{\gamma,\widehat j}(\tau)|^2.
\]
Expanding the two squares on $\mathcal E_{n,\gamma}^{\mathrm{loc,M}}$ gives
\endgroup
\begin{align*}
d_{\gamma,\widehat j}^2(\tau)
-d_{\gamma,\widehat j}^2(\widehat k)
&\le
2\{|d_{\gamma,\widehat j}(\tau)|
  +|d_{\gamma,\widehat j}(\widehat k)|\}M_{n,p}^{\epsilon,\gamma}
+2(M_{n,p}^{\epsilon,\gamma})^2\\
&\le C\sqrt n\,\mathfrak a_{\widehat j}M_{n,p}^{\epsilon,\gamma}
+2(M_{n,p}^{\epsilon,\gamma})^2.
\end{align*}
Together with the deterministic gap bound, this yields
\[
c|\widehat k-\tau|\mathfrak a_{\widehat j}^2
\le C\sqrt n\,\mathfrak a_{\widehat j}M_{n,p}^{\epsilon,\gamma}
+2(M_{n,p}^{\epsilon,\gamma})^2.
\]

 Besides, on $\mathcal E_{n,\gamma}^{\mathrm{loc,M}}$, the definition of $r_{\sigma,n}$ gives $8/9\le\sigma_j/\widehat\sigma_j\le 8/7$, for $1\leq j\le p$ and $|C_{\gamma,j_{\star}}(\tau)|^2$ is at least $(8/9)^2\{\sqrt n\,\mathfrak a_{\max}v(\theta)^{1-\gamma}-M_{n,p}^{\epsilon,\gamma}\}^2$.
If a coordinate $j$ satisfies $\mathfrak a_j<\mathfrak a_{\max}/3$, then
$v(t)^{-\gamma}K_B(t,\theta)\le v(\theta)^{1-\gamma}$ for every $t\in(0,1)$,
and its corresponding $|C_{\gamma,j}(k)|^2$ is at most
$(8/7)^2\{(1/2)\sqrt n\,\mathfrak a_{\max}v(\theta)^{1-\gamma}+M_{n,p}^{\epsilon,\gamma}\}^2$
for every $k\in\mathcal K_\gamma$, {which is smaller than the preceding lower bound for $|C_{\gamma,j_\star}(\tau)|^2$.}
Hence the coordinate $\widehat j$ selected satisfies $\mathfrak a_{\widehat j}\ge \mathfrak a_{\max}/3$ on $\mathcal E_{n,\gamma}^{\mathrm{loc,M}}$.
From Assumption~\ref{ass:C2} we know that $\mathfrak a_{\max}\asymp\|\bdelta\|_\infty$. Finally, {we have}
\[
\frac{|\widehat k-\tau|}{n}
=O_p\left\{
\frac{\sqrt{\log(np)}}{\sqrt n\|\bdelta\|_\infty}
+\frac{\log(np)}{n\|\bdelta\|_\infty^2}\right\}.
\]
\end{proof}

\begingroup
\begin{proof}[Proof of Corollary~\ref{cor:adaptive-location}]
We compare the two negative log p-values when only one localizer has the required rate. When both localizers have that rate, the data-dependent selection needs no further probability bound.%
Let $L_n=\log(np)$, $H_n=\log\log h_n$ and $J_n=\log\{p\log h_n\}$. {We have} $H_n=o(L_n)$, $H_n\le J_n\le C L_n$ and $H_n^3=o(L_nJ_n)$. Define the two selection events $\mathcal A_{S,\gamma,n}=\{p_{S_\gamma}<p_{M_\gamma}\}$ and $\mathcal A_{M,\gamma,n}=\mathcal A_{S,\gamma,n}^{c}$. Then 
\begin{equation}
\widehat\tau_\gamma
=
\widehat\tau_{S,\gamma}\mathbb{I}(\mathcal A_{S,\gamma,n})
+
\widehat\tau_{M,\gamma}\mathbb{I}(\mathcal A_{M,\gamma,n}).
\label{eq:adaptive-S2}
\end{equation}
If $F_V$ is the distribution
function of $\sup_{0\le t\le1}V(t)$, then the Borell--TIS inequality\citep{adler2007random} and a
fixed-point Gaussian Mills bound\citep{duembgen2010bounding} give
\begin{equation}
c(x_+)^2-C
\le -\log\{1-F_V(x)\}
\le C\{1+(x_+)^2\},
\qquad x\in\mathbb R.
\label{eq:adaptive-S3}
\end{equation}
For the standard Gumbel distribution
$F_{\rm Gu}(x)=\exp(-e^{-x})$,
\begin{equation}
c x_+-C
\le -\log\{1-F_{\rm Gu}(x)\}
\le C(1+x_+),
\qquad x\in\mathbb R.
\label{eq:adaptive-S4}
\end{equation}
where $x_+=x\vee0$.
{We consider four cases:} (i) $\Delta_{S,n}^2/L_n\to\infty$ and $\Delta_{M,n}=O(1)$ for $\widehat{\tau}_0$; (ii)  $\Delta_{M,n}\to\infty$ for $\widehat{\tau}_0$; (iii) $H_n^{1/2}\Delta_{S,n}/L_n\rightarrow\infty$ and $\Delta_{M,n}=O(1)$ for $\widehat{\tau}_{1/2}$; and (iv) $\Delta_{M,n}\rightarrow\infty$ for $\widehat{\tau}_{1/2}$.

\noindent\textit{\textbf{Part (i)}}: $\gamma=0$ under the dense-selector condition, i.e.,
\begin{equation}
\frac{\Delta_{S,n}^2}{L_n}\to\infty,
\qquad
\Delta_{M,n}=O(1).
\label{eq:adaptive-S5}
\end{equation}

From the proof of Proposition~\ref{prop:Tms-alter}, we have
\begin{equation}
S_0
\ge c\Delta_{S,n}-O_p(1+\Delta_{S,n}^{1/2}),
\qquad
S_0=O_p(1+\Delta_{S,n}).
\label{eq:adaptive-S6}
\end{equation}

Since $\Delta_{S,n}\to\infty$, the first relation implies
$S_0\ge c\Delta_{S,n}$ with probability tending to one.  As
$p_{S_0}=1-F_V(S_0)$, Equation~(\ref{eq:adaptive-S3}) yields
\begin{equation}
\Pr\{-\log p_{S_0}\ge c\Delta_{S,n}^2\}\to1.
\label{eq:adaptive-S7}
\end{equation}

The maximal-CUSUM bound used in the proof of Proposition~\ref{prop:consistency-Linf} gives
\begin{equation}
M_0=O_p\{L_n^{1/2}(1+\Delta_{M,n})\}.
\label{eq:adaptive-S8}
\end{equation}

Recall that
\begin{equation*}
Z_{M,0,n}=2M_0^2-\log(2p),
\qquad
p_{M_0}=1-F_{\rm Gu}(Z_{M,0,n}).
\end{equation*}

By Equations (\ref{eq:adaptive-S4}), (\ref{eq:adaptive-S8}), and $\Delta_{M,n}=O(1)$, {we have}
\begin{equation}
-\log p_{M_0}=O_p(L_n).
\label{eq:adaptive-S9}
\end{equation}

Combining Equations~(\ref{eq:adaptive-S5}), (\ref{eq:adaptive-S7}), and (\ref{eq:adaptive-S9}), $-\log p_{S_0}/L_n\xrightarrow{p}\infty$ and $-\log p_{M_0}/L_n=O_p(1)$. Hence
\begin{equation}
\Pr(\mathcal A_{S,0,n})=1-\Pr(\mathcal A_{S,0,n}^{c}) \to1,
\label{eq:adaptive-S10}
\end{equation}
since
\begin{equation*}
\begin{aligned}
\Pr(\mathcal A_{S,0,n}^{c})
&\le
\Pr\{-\log p_{S_0}\le KL_n\}
+
\Pr\{-\log p_{M_0}>KL_n\},
\end{aligned}
\end{equation*}
where the second is uniformly small for large enough $K$ and the first term tends to 0.
Combining Theorem~\ref{consistency-S} and Equation~(\ref{eq:adaptive-S2}) and (\ref{eq:adaptive-S10}), we have $|\widehat\tau_0-\tau|/n=O_p(R_{S,0,n})$.

\noindent\textit{\textbf{Part (ii)}}: $\gamma=0$ under the max-signal condition, i.e.,
\begin{equation}
\Delta_{M,n}\to\infty.
\label{eq:adaptive-S12}
\end{equation}

We divide the deterministic sequence into two index sets according to whether
$\Delta_{S,n}\le\Delta_{M,n}$ or
$\Delta_{S,n}>\Delta_{M,n}$.

We first consider the case:
\begin{equation}
\Delta_{S,n}\le\Delta_{M,n}.
\label{eq:adaptive-S13}
\end{equation}
By Equations~(\ref{eq:adaptive-S6}) and (\ref{eq:adaptive-S3}), $-\log p_{S_0}=O_p(1+\Delta_{S,n}^2)$. For $M_0$, {we have}%
\begin{equation}
M_0
\ge c\sqrt n\|\bdelta\|_\infty-O_p(1)
=c\Delta_{M,n}L_n^{1/2}-O_p(1).
\label{eq:adaptive-S15}
\end{equation}

Since $\Delta_{M,n}\to\infty$,
\begin{equation*}
Z_{M,0,n}
=2M_0^2-\log(2p)
\ge cL_n\Delta_{M,n}^2
\end{equation*}
with probability tending to one.  By Equation~(\ref{eq:adaptive-S4}),
\begin{equation}
\Pr\{-\log p_{M_0}
      \ge cL_n\Delta_{M,n}^2\}\to1.
\label{eq:adaptive-S16}
\end{equation}

By Equations~(\ref{eq:adaptive-S13})--(\ref{eq:adaptive-S16}),
\begin{equation*}
\frac{-\log p_{S_0}}
     {L_n\Delta_{M,n}^2}
=O_p\left(
\frac{1}{L_n\Delta_{M,n}^2}+
\frac{1}{L_n}
\right)
=o_p(1),
\end{equation*}
therefore $\Pr(\mathcal A_{M,0,n})\to1$. By Theorem~\ref{consistency-M}(i), the adaptive estimator then satisfies
\begin{equation}
\frac{|\widehat\tau_0-\tau|}{n}=O_p(R_{M,n}).
\label{eq:adaptive-S18}
\end{equation}

We next consider the case:
\begin{equation}
\Delta_{S,n}>\Delta_{M,n}.
\label{eq:adaptive-S19}
\end{equation}

 So, {we have}
\begin{equation*}
\begin{aligned}
R_{S,0,n}
&=\Delta_{S,n}^{-1}
  +\{n\|\bdelta\|_2^2\}^{-1/2}\le \Delta_{M,n}^{-1}
 +\frac1{\sqrt n\|\bdelta\|_\infty}\\
&=\Delta_{M,n}^{-1}
 +\frac1{\Delta_{M,n}L_n^{1/2}}
=O(R_{M,n}).
\end{aligned}
\end{equation*}

Since the adaptive estimator equals one of them,
\begin{equation}
\frac{|\widehat\tau_0-\tau|}{n}=O_p(R_{M,n}).
\label{eq:adaptive-S21}
\end{equation}

Combining (\ref{eq:adaptive-S18}) and (\ref{eq:adaptive-S21}), {we have} $|\widehat\tau_0-\tau|/n=O_p(R_{M,n})$.

\noindent\textit{\textbf{Part (iii)}}: $\gamma=1/2$ under the dense-selector condition, i.e.,
\begin{equation}
\frac{H_n^{1/2}\Delta_{S,n}}{L_n}\to\infty,
\qquad
\Delta_{M,n}=O(1).
\label{eq:adaptive-S23}
\end{equation}

Similarly, from the proof of Proposition~\ref{prop:Tms-alter}, {we have}
\begin{equation}
S_{1/2}
\ge c\Delta_{S,n}
-O_p\{H_n^{1/2}(1+\Delta_{S,n}^{1/2})\}.
\label{eq:adaptive-S24}
\end{equation}

Condition (\ref{eq:adaptive-S23}) implies
$\Delta_{S,n}/H_n\to\infty$.  Hence the remainder in Equation~(\ref{eq:adaptive-S24}) is
$o_p(\Delta_{S,n})$, and therefore
\begin{equation}
S_{1/2}\ge c\Delta_{S,n},
\label{eq:adaptive-S25}
\end{equation}
{with probability tending to one.} Recall that $Z_{S,1/2,n}=A_{\rm DE}(\log h_n)S_{1/2}-D_{\rm DE}(\log h_n)$, where $A_{\rm DE}(\log h_n)\asymp H_n^{1/2}$ and $D_{\rm DE}(\log h_n)=2H_n+O(\log H_n)$. By Equations~(\ref{eq:adaptive-S23}) and (\ref{eq:adaptive-S25}), $Z_{S,1/2,n}\ge cH_n^{1/2}\Delta_{S,n}$ with probability tending to one.  Since
$p_{S_{1/2}}=1-F_{\rm Gu}(Z_{S,1/2,n})$, Equation~(\ref{eq:adaptive-S4}) gives
\begin{equation}
\Pr\{-\log p_{S_{1/2}}
\ge cH_n^{1/2}\Delta_{S,n}\}\to1.
\label{eq:adaptive-S26}
\end{equation}

From the proof of Proposition~\ref{prop:consistency-Linf},
\begin{equation}
M_{1/2}=O_p\{J_n^{1/2}+L_n^{1/2}\Delta_{M,n}\}.
\label{eq:adaptive-S27}
\end{equation}

Recall that $Z_{M,1/2,n}=A_{\rm DE}(p\log h_n)M_{1/2}-D_{\rm DE}(p\log h_n)$,
where $A_{\rm DE}(p\log h_n)\asymp J_n^{1/2}$ and $D_{\rm DE}(p\log h_n)=2J_n+O(\log J_n)$. By (\ref{eq:adaptive-S4}), (\ref{eq:adaptive-S27}), and $\Delta_{M,n}=O(1)$,
\begin{equation}
-\log p_{M_{1/2}}
=O_p\{J_n+(L_nJ_n)^{1/2}\}
=O_p(L_n).
\label{eq:adaptive-S28}
\end{equation}
Similarly, Equations (\ref{eq:adaptive-S23}), (\ref{eq:adaptive-S26}), and (\ref{eq:adaptive-S28}) imply $\Pr(\mathcal A_{S,1/2,n})\to1$. Combining Theorem~\ref{consistency-S}, {we have} $|\widehat\tau_{1/2}-\tau|/n=O_p(R_{S,1/2,n})$.

\noindent\textit{\textbf{Part (iv)}}: $\gamma=1/2$ under the max-signal condition, i.e.,
\begin{equation}
\Delta_{M,n}\to\infty.
\label{eq:adaptive-S31}
\end{equation}

First consider the case $\Delta_{S,n}\le \Delta_{M,n}H_n^{1/2}$. Similarly, for $p_{S_{1/2}}$, {we have}
\begin{equation*}
-\log p_{S_{1/2}}=O_p\{H_n+H_n^{1/2}\Delta_{S,n}+H_n\Delta_{S,n}^{1/2}\}=O_p\{H_n+\Delta_{M,n}H_n+\Delta_{M,n}^{1/2}H_n^{5/4}\}.
\end{equation*}

As $\Delta_{M,n}\to\infty$, {we have}
\begin{equation*}
\frac{H_n}{\Delta_{M,n}(L_nJ_n)^{1/2}}\to0,\frac{H_n}{(L_nJ_n)^{1/2}}\to0,
\end{equation*}
and
\begin{equation*}
\frac{\Delta_{M,n}^{1/2}H_n^{5/4}}
     {\Delta_{M,n}(L_nJ_n)^{1/2}}
=
\frac{H_n^{5/4}}
     {\Delta_{M,n}^{1/2}(L_nJ_n)^{1/2}}
\to0.
\end{equation*}

Consequently,
\begin{equation}
\frac{-\log p_{S_{1/2}}}
{\Delta_{M,n}(L_nJ_n)^{1/2}}
\xrightarrow{p}0.
\label{eq:adaptive-S35}
\end{equation}

The same argument as in (\ref{eq:adaptive-S15}) gives $M_{1/2}\ge c\Delta_{M,n}L_n^{1/2}-O_p(1)$. Therefore $Z_{M,1/2,n}\ge c\Delta_{M,n}(L_nJ_n)^{1/2}-CJ_n$. Since $J_n\le CL_n$ and $\Delta_{M,n}\to\infty$, the negative $J_n$ term
is negligible, and (\ref{eq:adaptive-S4}) yields
\begin{equation}
\Pr\{-\log p_{M_{1/2}}
\ge c\Delta_{M,n}(L_nJ_n)^{1/2}\}\to1.
\label{eq:adaptive-S36}
\end{equation}

So, Equations (\ref{eq:adaptive-S35})--(\ref{eq:adaptive-S36}) imply $\Pr(\mathcal A_{M,1/2,n})\to1$. Hence, as $\Delta_{S,n}\le \Delta_{M,n}H_n^{1/2}$,
\begin{equation}
\frac{|\widehat\tau_{1/2}-\tau|}{n}=O_p(R_{M,n}).\label{eq:adaptive-S38}
\end{equation}

Now consider the case $\Delta_{S,n}>\Delta_{M,n}H_n^{1/2}$. {We have} $\Delta_{S,n}/H_n^{1/2}>\Delta_{M,n}\to\infty$. Furthermore,
\begin{equation*}
\begin{aligned}
R_{S,1/2,n}
&=H_n^{1/2}
\left[
\Delta_{S,n}^{-1}
+\{n\|\bdelta\|_2^2\}^{-1/2}
\right]\le
\Delta_{M,n}^{-1}
+
\frac{H_n^{1/2}}{\sqrt n\|\bdelta\|_\infty}\\
&=
\Delta_{M,n}^{-1}
+
\frac{H_n^{1/2}}{\Delta_{M,n}L_n^{1/2}}
=O(R_{M,n}),
\end{aligned}
\end{equation*}

Thus, as $\Delta_{S,n}>\Delta_{M,n}H_n^{1/2}$,
\begin{equation}
\frac{|\widehat\tau_{1/2}-\tau|}{n}=O_p(R_{M,n}).
\label{eq:adaptive-S41}
\end{equation}
Combining (\ref{eq:adaptive-S38}) and (\ref{eq:adaptive-S41}), we have $|\widehat\tau_{1/2}-\tau|/n=O_p(R_{M,n})$.

Finally, if the dense-selector condition holds but $\Delta_{M,n}$
neither stays bounded nor tends to infinity, every subsequence has a
further subsequence on which $\Delta_{M,n}$ converges in $[0,\infty]$.
Parts (i)/(iii) give consistency for a finite limit and parts (ii)/(iv)
for an infinite limit. Thus the stated consistency also holds without
an extra dichotomy assumption on the signal sequence.
\end{proof}
\endgroup

\section{Proof of Theorems and corollary in Section~\ref{sec:wbs}%
}

This section proves the multiple-change results in Section~\ref{sec:wbs}. %
Throughout this section, $\lambda_m^{\mathrm W}=\lceil m^{\lambda_{\mathrm W}}\rceil$ is the WBS-specific trimming sequence, $\mathfrak c\in\{\mathrm S,\mathrm M\}$ denotes a component label, and $q\ge2$ is reserved for a moment order.

The component thresholds are parameterized as
\[
\Lambda_{\mathfrak c,\gamma,n}^{\mathrm W}
=c_{\mathfrak c,\gamma}^{\mathrm W}
 \bar\Lambda_{\mathfrak c,\gamma,n}^{\mathrm W},
\qquad
0<c_-\le c_{\mathfrak c,\gamma}^{\mathrm W}\le c_+<\infty,
\qquad
\mathfrak c\in\{\mathrm S,\mathrm M\},\quad \gamma\in\{0,1/2\},
\]
where the calibration multipliers are fixed with respect to $n$ and $p$.

\begin{lemma}%
\label{lem:wbs-threshold-calibration}
For any positive sequences $x_n$ and $y_n$ and each
$\mathfrak c\in\{\mathrm S,\mathrm M\}$,
\[
x_n=o(\Lambda_{\mathfrak c,\gamma,n}^{\mathrm W})
\quad\Longleftrightarrow\quad
x_n=o(\bar\Lambda_{\mathfrak c,\gamma,n}^{\mathrm W}),
\qquad
\Lambda_{\mathfrak c,\gamma,n}^{\mathrm W}=o(y_n)
\quad\Longleftrightarrow\quad
\bar\Lambda_{\mathfrak c,\gamma,n}^{\mathrm W}=o(y_n).
\]
Moreover, the preliminary localization rates in the main article can be written exactly as
\begingroup
\begin{align*}
r_{S,\gamma,n}^{\mathrm{pre,W}}
={}&\ind{\gamma=1/2}\Delta_n^{\lambda_{\mathrm W}}
+\Delta_n\left\{
(c_{S,\gamma}^{\mathrm W})^{-1}
\frac{B_{S,\gamma,n}^{\mathrm W}}
 {\bar\Lambda_{S,\gamma,n}^{\mathrm W}}
+(c_{S,\gamma}^{\mathrm W})^{-1/2}
\frac{\{\log(np)\}^{c_{\mathrm{dep}}}}
 {p^{1/4}\{\bar\Lambda_{S,\gamma,n}^{\mathrm W}\}^{1/2}}
\right\},\\
r_{M,\gamma,n}^{\mathrm{pre,W}}
={}&\ind{\gamma=1/2}\Delta_n^{\lambda_{\mathrm W}}
+\Delta_n(c_{M,\gamma}^{\mathrm W})^{-1}
\frac{\{\log(np)\}^{c_{\mathrm{dep}}}}
 {\bar\Lambda_{M,\gamma,n}^{\mathrm W}},\\
r_{A,\gamma,n}^{\mathrm{pre,W}}
={}&\ind{\gamma=1/2}\Delta_n^{\lambda_{\mathrm W}}
+\Delta_n\Biggl\{
(c_{S,\gamma}^{\mathrm W})^{-1}
\frac{B_{S,\gamma,n}^{\mathrm W}}
 {\bar\Lambda_{S,\gamma,n}^{\mathrm W}}
+(c_{S,\gamma}^{\mathrm W})^{-1/2}
\frac{\{\log(np)\}^{c_{\mathrm{dep}}}}
 {p^{1/4}\{\bar\Lambda_{S,\gamma,n}^{\mathrm W}\}^{1/2}}
\notag\\
&\hspace{7.5em}
+(c_{M,\gamma}^{\mathrm W})^{-1}
\frac{\{\log(np)\}^{c_{\mathrm{dep}}}}
 {\bar\Lambda_{M,\gamma,n}^{\mathrm W}}
\Biggr\}.
\end{align*}
\endgroup
Consequently, replacing a baseline threshold by a fixed positive multiple changes only the finite multiplicative constants in the WBS probability and localization bounds; it does not change any exponent condition or asymptotic order in Theorems~\ref{thm:wbs-dense}--\ref{thm:wbs-max} or Corollary~\ref{cor:wbs-adaptive}.
\end{lemma}

\begin{proof}
The bounds on the calibration multipliers give
\[
c_-\bar\Lambda_{\mathfrak c,\gamma,n}^{\mathrm W}
\le \Lambda_{\mathfrak c,\gamma,n}^{\mathrm W}
\le c_+\bar\Lambda_{\mathfrak c,\gamma,n}^{\mathrm W}.
\]
Consequently,
\[
\frac1{c_+}\frac{x_n}{\bar\Lambda_{\mathfrak c,\gamma,n}^{\mathrm W}}
\le
\frac{x_n}{\Lambda_{\mathfrak c,\gamma,n}^{\mathrm W}}
\le
\frac1{c_-}\frac{x_n}{\bar\Lambda_{\mathfrak c,\gamma,n}^{\mathrm W}},
\qquad
c_-\frac{\bar\Lambda_{\mathfrak c,\gamma,n}^{\mathrm W}}{y_n}
\le
\frac{\Lambda_{\mathfrak c,\gamma,n}^{\mathrm W}}{y_n}
\le
c_+\frac{\bar\Lambda_{\mathfrak c,\gamma,n}^{\mathrm W}}{y_n}.
\]
These inequalities prove the two equivalences.  Substitution of
$\Lambda_{S,\gamma,n}^{\mathrm W}=c_{S,\gamma}^{\mathrm W}
\bar\Lambda_{S,\gamma,n}^{\mathrm W}$ and
$\Lambda_{M,\gamma,n}^{\mathrm W}=c_{M,\gamma}^{\mathrm W}
\bar\Lambda_{M,\gamma,n}^{\mathrm W}$ into the three definitions of the
preliminary rates gives the displayed identities.  Finally,
\[
(c_+)^{-1}\le(c_{\mathfrak c,\gamma}^{\mathrm W})^{-1}\le(c_-)^{-1},
\qquad
(c_+)^{-1/2}\le(c_{\mathfrak c,\gamma}^{\mathrm W})^{-1/2}\le(c_-)^{-1/2},
\]
so these factors are absorbed into the universal constants in all subsequent bounds.
\end{proof}

For an interval $I=(s,e]$, let $m_I=e-s$ and use the convention that a change belongs to $I$ only when $s<\tau_j<e$.  Then $\mathbf U_I(b)=\mathbf m_I(b)+\mathbf Z_I(b)$ where $\mathbf Z_I(b)
=\sum_{i=s+1}^{b}\bepsilon_i
-(b-s)/(e-s)\sum_{i=s+1}^{e}\bepsilon_i$ and $\mathbf m_I(b)=\sum_{i=s+1}^{b}\E(\bX_i)-(b-s)/(e-s)\sum_{i=s+1}^{e}\E(\bX_i)$. The deterministic dense and coordinatewise parts are
\begin{align*}
G_{S,\gamma,I}(b)
&=v_{I,b}^{-2\gamma}
  \frac{\|\mathbf m_I(b)\|_2^2}{m_I\sqrt p},\\
G_{M,\gamma,I}(b)
&=\max_{1\le j\le p}
  v_{I,b}^{-\gamma}
  \frac{|m_{I,j}(b)|}{\sqrt{m_I}\,\sigma_j},\\
\widehat G_{M,\gamma,I}(b)
&=\max_{1\le j\le p}
  v_{I,b}^{-\gamma}
  \frac{|m_{I,j}(b)|}{\sqrt{m_I}\,\widehat\sigma_j}.
\end{align*}

Let
\begin{equation*}
    r_{\mu,n}^{\mathrm{mc}}={}\frac{M}{\sqrt n}+\sqrt p\,M\eta_M+\frac{M^2\displaystyle\max_{1\le r\le K}\|\bdelta_r\|_2^2}{n\sqrt p}+\frac{M^2\displaystyle\max_{1\le r\le K}\|\bdelta_r\|_2}{n\sqrt p},
\end{equation*}
\[
\begin{aligned}
r_{\omega,n}^{\mathrm{mc}}
={}&
M^2\left(\frac{1}{np}+\frac {1}{n^2}\right)^{1/2}+M^2\eta_M+M^2\varpi_{\mathrm{dep}}(n/3)+\left\{\frac{M^2\displaystyle\max_{1\le r\le K}\|\bdelta_r\|_2^2}{n\sqrt p}\right\}^{2}\\
&+\frac{M^2\displaystyle\max_{1\le r\le K}\|\bdelta_r\|_2^2}{n\sqrt p}+\frac{M^3\displaystyle\max_{1\le r\le K}\|\bdelta_r\|_2}{n\sqrt p},
\end{aligned}
\]
and
\[
a_{\sigma,n}^{\mathrm{mc}}
=
a_{\sigma,n}
+
\frac{\ell_{\max}^2\max_{1\le r\le K}\|\bdelta_r\|_\infty ^2}{n}.
\]
\begin{lemma}%
\label{lem:wbs-multiple-nuisance}
(i) Suppose Assumptions~\ref{ass:C1}--\ref{ass:C3}, \ref{ass:WBS-geometry} %
and $M=\left\lceil (n\wedge p)^{1/8}\right\rceil$ hold. Under Assumption~\ref{ass:C1}(AM) or (PD), assume additionally that $\log p=o(n^{\varkappa/8})$. {We have}
\[
\max_{\gamma\in\{0,1/2\}}\max_{\substack{I=(s,e]:\,m_I\geq\ell_{\mathrm W,n}\\
b\in\mathcal K_\gamma(I)}}\left|\widehat\mu_{\gamma,M,I}(b)-\mu_{\gamma,M,I}(b)\right|=O_p\left(r_{\mu,n}^{\mathrm{mc}}\right),
\]
and $\left|\widehat\omega^2-\omega_p^2\right|=O_p\left(r_{\omega,n}^{\mathrm{mc}}\right)$. Moreover, $\vert\widehat\omega/\omega_p-1\vert=O_p(r_{\omega,n}^{\mathrm{mc}})$ with $r_{\mu,n}^{\mathrm{mc}}+r_{\omega,n}^{\mathrm{mc}}\rightarrow0$.

\begingroup
(ii) Suppose Assumptions~\ref{ass:C1}--\ref{ass:C3}, \ref{ass:lrv-M}(i)--(ii), \ref{ass:WBS-geometry}(i),~\ref{ass:WBS-tuning}(iv) and  {$\log p=O(n^{\nu})$} with $0<\nu<\nu_c$ hold, where $\nu_c$ is defined before Theorem~\ref{null:Max}. {We have} $\max_{1\le j\le p}\left|\widehat\sigma_j/\sigma_j-1\right|=O_p\left(a_{\sigma,n}^{\mathrm{mc}}\right)$ with $a_{\sigma,n}^{\mathrm{mc}}\rightarrow0$.
\endgroup
\end{lemma}

\begin{proof}
\begingroup
For part (i), we separate the noise-only estimators from their deterministic
mean contamination and control each term of the centering and quartic scale
expansions. Part (ii) uses the uniform bandwidth bounds in
Lemma~\ref{lem:lrv-general} and the corresponding two-term contamination
expansion.
\endgroup
(i) For $0\leq h\leq M$, let $\mathbf E_{i,h}^{(M)}=\mathbf d_{i,h}^{\mu}+\mathbf E_{i,h}^{(M),\epsilon}$ and $\mathbf d_{i,h}^{\mu}=\E(\bX_i)-\E(\bX_{i+M+h+1})$. {The contribution of the $r$th change to} $\mathbf d_{i,h}^{\mu}$ is nonzero only if $i\leq\tau_r<i+M+h+1$. Since $\Delta_n\geq c_\Delta n$ and $M\leq(n\wedge p)^{1/8}+1$,
\[
\frac{2M+1}{\Delta_n}
\leq
\frac{2(n\wedge p)^{1/8}+3}{c_\Delta n}
\leq
Cn^{-7/8}
\rightarrow0.
\]
Consequently, for all sufficiently large $n$, an interval $[i,i+M+h+1]$ contains at most one change point. Hence $\mathbf d_{i,h}^{\mu}\in\{\mathbf0,\pm\bdelta_1,\ldots,\pm\bdelta_{K}\}$. {For every $r\leq K$, the number of indices $i$ for which $\mathbf d_{i,h}^{\mu}=\pm\bdelta_r$ is at most $CM$. Therefore, since $K$ is fixed,
\[
\sum_i
\|\mathbf d_{i,h}^{\mu}\|_2^2
\leq
CM
\sum_{r=1}^{K}\|\bdelta_r\|_2^2
\lesssim CM
\max_{1\le r\le K}\|\bdelta_r\|_2^2.
\]
}

\begingroup
Let $N_h=n-M-2h-1$ and let
$\mathbf e_{i,h}=\mathbf E_{i,h}^{(M),\epsilon}$.
For the noise-only estimator $\widehat g_h^\epsilon$, the exact expansion is
\[
\begin{aligned}
\widehat g_h-\widehat g_h^\epsilon
&=A_h^{\mu\mu}+A_h^{\mu\epsilon}+A_h^{\epsilon\mu},\\
A_h^{\mu\mu}
&=\frac1{2N_h}\sum_{i=1}^{N_h}(\mathbf d_{i,h}^{\mu})^\top\mathbf d_{i+h,h}^{\mu},\\
A_h^{\mu\epsilon}
&=\frac1{2N_h}\sum_{i=1}^{N_h}(\mathbf d_{i,h}^{\mu})^\top\mathbf e_{i+h,h},\qquad
A_h^{\epsilon\mu}
=\frac1{2N_h}\sum_{i=1}^{N_h}\mathbf e_{i,h}^\top\mathbf d_{i+h,h}^{\mu}.
\end{aligned}
\]
The deterministic term satisfies
\endgroup
\[
|A_h^{\mu\mu}|
\leq
\frac{C}{N_h}
\sum_i
\|\mathbf d_{i,h}^{\mu}\|_2^2
\lesssim
C\frac{M\displaystyle\max_{1\le r\le K}\|\bdelta_r\|_2^2}{n}.
\]

{
And for the term $A_h^{\mu\epsilon}$, there are deterministic coefficients $w_{r,i,h}^{(a)}$, $a\in\{0,1\}$, satisfying $|w_{r,i,h}^{(a)}|\leq C$ and $\sum_i\{w_{r,i,h}^{(a)}\}^2\lesssim CM$ such that $A_h^{\mu\epsilon}=(2N_h)^{-1}\sum_{r=1}^{K}\sum_{a=0}^{1}\sum_i w_{r,i,h}^{(a)}\bdelta_r^\top\bepsilon_{i+q_{a,h}}$, where $q_{a,h}$ is a deterministic temporal shift. By Lemma~\ref{lem:vector-basic}(ii) with
$\boldsymbol u=\bdelta_r/\|\bdelta_r\|_2$, $\|\sum_i w_{r,i,h}^{(a)}\bdelta_r^\top\bepsilon_{i+q_{a,h}}\|_2\leq C\sqrt M\,\|\bdelta_r\|_2$. It follows that
\[
\begin{aligned}
\|A_h^{\mu\epsilon}\|_2
&\lesssim
\frac{C\sqrt M}{n}
\sum_{r=1}^{K}\|\bdelta_r\|_2\leq
C\frac{\sqrt M}{n}
\max_{1\le r\le K}\|\bdelta_r\|_2.
\end{aligned}
\]
The same calculation gives $\|A_h^{\epsilon\mu}\|_2\lesssim C\sqrt M/n \max_{1\le r\le K}\|\bdelta_r\|_2$.
}

Consequently,
\[
\begin{aligned}
&
\E\left[\max_{0\leq h\leq M}\bigl\{|A_h^{\mu\epsilon}|+|A_h^{\epsilon\mu}|\bigr\}^2\right]\leq 2\sum_{h=0}^{M}\E\bigl\{|A_h^{\mu\epsilon}|+|A_h^{\epsilon\mu}|\bigr\}^2\lesssim C\frac{M^2\displaystyle\max_{1\le r\le K}\|\bdelta_r\|_2^2}{n^2}.
\end{aligned}
\]
Markov's inequality therefore gives
\[
\max_{0\leq h\leq M}
|\widehat g_h-\widehat g_h^\epsilon|
=
O_p\left(
\frac{M\displaystyle\max_{1\le r\le K}\|\bdelta_r\|_2^2}{n}
+
\frac{M\displaystyle\max_{1\le r\le K}\|\bdelta_r\|_2}{n}
\right).
\]

The coefficient of $\widehat g_h$ in the local centering
term is
\[
c_{\gamma,I,b,h}
=
v_{I,b}^{-2\gamma}
\frac{
m_{I,-}(b)^2m_{I,+}(b)^2
}{
m_I^3\sqrt p
}
\sum_{i=s+1}^{e-h}
\{2-\ind{h=0}\}
a_{i,b}^{I}a_{i+h,b}^{I}.
\]
Then, we have\[
\left|
\sum_{i=s+1}^{e-h}
a_{i,b}^{I}a_{i+h,b}^{I}
\right|
\leq
C\left\{
\frac{1}{m_{I,-}(b)}
+
\frac{1}{m_{I,+}(b)}
+
\frac{h}{
m_{I,-}(b)m_{I,+}(b)
}
\right\}.
\]
If $h\geq m_I$, the sum defining $c_{\gamma,I,b,h}$
is empty. If $h<m_I$, then $h/m_I\leq1$, we have
\[
|c_{0,I,b,h}|
\leq
\frac{C}{\sqrt p}
\left(
v_{I,b}
+
\frac{h}{m_I}v_{I,b}
\right)
\leq
\frac{C}{\sqrt p},~~~~|c_{1/2,I,b,h}|
\leq
\frac{C}{\sqrt p}
\left(1+\frac{h}{m_I}\right)
\leq
\frac{C}{\sqrt p},
\]
and
\[
\max_{\gamma,I,b,h}
|c_{\gamma,I,b,h}|
\leq
\frac{C}{\sqrt p},
\qquad
\max_{\gamma,I,b}
\sum_{h=0}^{M}
|c_{\gamma,I,b,h}|
\leq
\frac{CM}{\sqrt p},~~~~\max_{\gamma,I,b}
\sum_{h=0}^{M}
c_{\gamma,I,b,h}^2
\leq
\frac{CM}{p}.
\]

Thus, 
\begingroup
\[
\begin{aligned}
&
\max_{\gamma,I,b}
\left|
\sum_{h=0}^{M}
c_{\gamma,I,b,h}
\{
\widehat g_h-\widehat g_h^\epsilon
\}
\right|\leq
\frac{CM}{\sqrt p}
\max_{0\leq h\leq M}
|\widehat g_h-\widehat g_h^\epsilon|\\
&\qquad=
O_p\left(
\frac{M^2\displaystyle\max_{1\le r\le K}\|\bdelta_r\|_2^2}{n\sqrt p}
+
\frac{M^2\displaystyle\max_{1\le r\le K}\|\bdelta_r\|_2}{n\sqrt p}
\right).
\end{aligned}
\]
\endgroup

Let $\mathbf X_g^\epsilon
=
\left(
\widehat g_0^\epsilon-\E\widehat g_0^\epsilon,
\ldots,
\widehat g_M^\epsilon-\E\widehat g_M^\epsilon
\right)^\top$.
By Lemma~\ref{lem:centering-general}, $\E\|\mathbf X_g^\epsilon\|_2^2=\sum_{h=0}^{M}\Var(\widehat g_h^\epsilon)\leq C Mp/n$. Therefore
\[
\begin{aligned}
&
\E\left[
\max_{\gamma,I,b}
\left|
\sum_{h=0}^{M}
c_{\gamma,I,b,h}
\{
\widehat g_h^\epsilon
-\E\widehat g_h^\epsilon
\}
\right|^2
\right]\leq
\left(
\max_{\gamma,I,b}
\sum_{h=0}^{M}
c_{\gamma,I,b,h}^2
\right)
\E\|\mathbf X_g^\epsilon\|_2^2
\leq
C\frac{M^2}{n}.
\end{aligned}
\]
Hence $\max_{\gamma,I,b}|\sum_{h=0}^{M}c_{\gamma,I,b,h}\{\widehat g_h^\epsilon-\E\widehat g_h^\epsilon\}|=O_p(M/\sqrt n)$. By Lemma~\ref{lem:vector-basic} and \ref{lem:centering-general}, we have $\max_{0\leq h\leq M}|\E\widehat g_h^\epsilon-g_h|\leq Cp\eta_M$, then
\[
\max_{\gamma,I,b}
\left|
\sum_{h=0}^{M}
c_{\gamma,I,b,h}
(
\E\widehat g_h^\epsilon-g_h
)
\right|
\leq
C\sqrt p\,M\eta_M.
\]

Combining the above equations yields $\max_{\gamma,I,b}
|
\widehat{\mu}_{\gamma,M,I}(b)
-\mu_{\gamma,M,I}(b)
|
=
O_p(r_{\mu,n}^{\mathrm{mc}})$.

{%
For $\widehat{\omega}$, let
\[
\widetilde\omega^2
=
\frac{2}{p}
\sum_{h,\ell=0}^{M}
\sum_{\nu\in\{\mathrm F,\mathrm C\}}
w_{h\ell}^{\nu}\widehat{\mathcal T}_{h\ell}^{\nu},
\qquad
\widetilde\omega_{\epsilon}^2
=
\frac{2}{p}
\sum_{h,\ell=0}^{M}
\sum_{\nu\in\{\mathrm F,\mathrm C\}}
w_{h\ell}^{\nu}\widehat{\mathcal T}_{h\ell}^{\nu,\epsilon},
\]
\begingroup
where the nonzero deterministic coefficients are $w_{00}^{\mathrm F}=1$, $w_{h0}^{\mathrm F}=w_{0h}^{\mathrm F}=2$ for $1\le h\le M$ and $w_{h\ell}^{\mathrm F}=w_{h\ell}^{\mathrm C}=2$ for $1\le h,\ell\le M$. Set all unlisted weights $w_{h\ell}^{\nu}$ equal to zero. Thus, $\max_{h,\ell,\nu}|w_{h\ell}^{\nu}|\le2$ and $\widehat\omega^2=\max\{\widetilde\omega^2,n^{-1}\}$. For each $(i_1,i_2)\in\mathcal I_{h\ell}$, abbreviate the four noise and
mean-difference vectors by
\[
\begin{aligned}
(\mathbf e_1,\mathbf e_2,\mathbf e_3,\mathbf e_4)
&=(\mathbf E_{i_1,h}^{(M),\epsilon},\mathbf E_{i_1+h,h}^{(M),\epsilon},
   \mathbf E_{i_2,\ell}^{(M),\epsilon},\mathbf E_{i_2+\ell,\ell}^{(M),\epsilon}),\\
(\mathbf m_1,\mathbf m_2,\mathbf m_3,\mathbf m_4)
&=(\mathbf d_{i_1,h}^{\mu},\mathbf d_{i_1+h,h}^{\mu},
   \mathbf d_{i_2,\ell}^{\mu},\mathbf d_{i_2+\ell,\ell}^{\mu}).
\end{aligned}
\]
Define the two multilinear products
\[
\Phi_{\mathrm F}(\mathbf z_1,\mathbf z_2,\mathbf z_3,\mathbf z_4)
=(\mathbf z_1^\top\mathbf z_3)(\mathbf z_2^\top\mathbf z_4),\qquad
\Phi_{\mathrm C}(\mathbf z_1,\mathbf z_2,\mathbf z_3,\mathbf z_4)
=(\mathbf z_1^\top\mathbf z_4)(\mathbf z_2^\top\mathbf z_3).
\]
For $U\subseteq\{1,2,3,4\}$, set
$\mathbf z_a^U=\mathbf m_a$ if $a\in U$ and
$\mathbf z_a^U=\mathbf e_a$ otherwise. Multilinearity gives the exact
expansion
\[
\widehat{\mathcal T}_{h\ell}^{\nu}-\widehat{\mathcal T}_{h\ell}^{\nu,\epsilon}
=\sum_{a=1}^4\Delta_{h\ell}^{\nu,(a)},\qquad
\Delta_{h\ell}^{\nu,(a)}
=\frac1{4|\mathcal I_{h\ell}|}
\sum_{(i_1,i_2)\in\mathcal I_{h\ell}}\sum_{|U|=a}
\Phi_\nu(\mathbf z_1^U,\mathbf z_2^U,\mathbf z_3^U,\mathbf z_4^U).
\]
Let $D_n=\max_{r\le K}\|\bdelta_r\|_2$. Each nonzero $\mathbf m_a$ has
norm at most $D_n$. The primitive times on the two halves are separated by
\[
\min_{u\in\{i_1,i_1+h,i_1+M+h+1,i_1+M+2h+1\}}
\min_{v\in\{i_2,i_2+\ell,i_2+M+\ell+1,i_2+M+2\ell+1\}}|v-u|
\ge\lfloor n/2\rfloor-3M-1\ge n/3.
\]
The deterministic coefficient vectors defining each $\mathbf e_a$ have
bounded Euclidean norm. Lemmas~\ref{lem:bilinear-quadratic-bound} and
\ref{lem:vector-basic} therefore give, for indices on opposite halves,
\[
\begin{aligned}
\|\mathbf m_a^\top\mathbf e_b\|_4&\le CD_n,\\
\|\mathbf e_a^\top\mathbf e_b\|_4
&\le\|\mathbf e_a^\top\mathbf e_b-\E(\mathbf e_a^\top\mathbf e_b)\|_4
   +|\E(\mathbf e_a^\top\mathbf e_b)|\\
&\le C\sqrt p+Cp\varpi_{\mathrm{dep}}(n/3)\le C\sqrt p.
\end{aligned}
\]
The last inequality follows from $\log p=o(n^{\varkappa/8})$ under
(AM)/(PD), and from exact independence of the halves under (MD).
For $a=2$, let $\Delta_{h\ell,\mathrm{same}}^{\nu,(2)}$ retain
$U=\{1,2\},\{3,4\}$ in the expansion, and let
$\Delta_{h\ell,\mathrm{cross}}^{\nu,(2)}$ retain the other four subsets.
A contaminated half has at most $CM$ possible starting times; hence a
monomial is nonzero on at most $CMn$ pairs if only one half is
contaminated, and on at most $CM^2$ pairs if both are contaminated.
Applying $\|AB\|_2\le\|A\|_4\|B\|_4$ to each product gives
\[
\begin{array}{c|c|c}
\text{term}&\text{number of nonzero pairs}&L_2\text{ bound}\\ \hline
\Delta_{h\ell}^{\nu,(1)}&CMn&CM\sqrt p\,D_n/n\\
\Delta_{h\ell,\mathrm{same}}^{\nu,(2)}&CMn&CM D_n^2/n\\
\Delta_{h\ell,\mathrm{cross}}^{\nu,(2)}&CM^2&CM^2\sqrt p\,D_n^2/n^2\\
\Delta_{h\ell}^{\nu,(3)}&CM^2&CM^2D_n^3/n^2\\
\Delta_{h\ell}^{\nu,(4)}&CM^2&CM^2D_n^4/n^2
\end{array}
\]
(the last term is deterministic). There are $O(M^2)$ weighted lag pairs.
Writing $B_n=M^2D_n^2/(n\sqrt p)$ and using
$M\le C\sqrt p$, $M^2\le Cn$ and $M/(\sqrt n\,p^{1/4})\le C$, we obtain
\[
\begin{aligned}
\left\|\frac2p\sum_{h,\ell,\nu}w_{h\ell}^\nu\Delta_{h\ell}^{\nu,(1)}\right\|_2
&\le \frac{CM^3D_n}{n\sqrt p},\\
\left\|\frac2p\sum_{h,\ell,\nu}w_{h\ell}^\nu\Delta_{h\ell}^{\nu,(2)}\right\|_2
&\le C\left\{\frac{M^3D_n^2}{np}+\frac{M^4D_n^2}{n^2\sqrt p}\right\}
\le CB_n,\\
\left\|\frac2p\sum_{h,\ell,\nu}w_{h\ell}^\nu\Delta_{h\ell}^{\nu,(3)}\right\|_2
&\le CB_n^{3/2}\frac{M}{\sqrt n\,p^{1/4}}\le C(B_n+B_n^2),\\
\left|\frac2p\sum_{h,\ell,\nu}w_{h\ell}^\nu\Delta_{h\ell}^{\nu,(4)}\right|
&\le CB_n^2.
\end{aligned}
\]
\endgroup

For $\widetilde{\omega}_{\epsilon}$, let
\[
\bOme_M
=
\bGam(0)+\sum_{h=1}^{M}\{\bGam(h)+\bGam(h)^\top\},
\qquad
\omega_{p,M}^2
=
\frac{2}{p}\Tr (\bOme_M^2).
\]

By Lemma~\ref{lem:difference-product-general}, the expectation-bias and lag-truncation calculations in the proof of Lemma~\ref{lem:scale-general} and Lemma~\ref{lem:vector-basic}(i), we have\[
\left\|
\widetilde\omega_{\epsilon}^2
-
\E(\widetilde\omega_{\epsilon}^2)
\right\|_2
\le
CM^2\left(\frac{1}{np}+\frac {1}{n^2}\right)^{1/2},
\]
\[
\left|
\E(\widetilde\omega_{\epsilon}^2)
-
\omega_{p,M}^2
\right|
\le
CM^2\{\eta_M+\varpi_{\mathrm{dep}}(n/3)+n^{-8}\},
\qquad
\left|\omega_{p,M}^2-\omega_p^2\right|
\le C\eta_M.
\]
Thus,
\[
\left|
\widetilde\omega_{\epsilon}^2-
\omega_p^2
\right|
=
O_p\left\{
M^2\left(\frac{1}{np}+\frac {1}{n^2}\right)^{1/2}
+M^2\eta_M
+M^2\varpi_{\mathrm{dep}}(n/3)
\right\}.
\]
Together with
\[
\widetilde\omega^2-
\widetilde\omega_{\epsilon}^2
=
\frac{2}{p}
\sum_{h,\ell=0}^{M}\sum_{\nu}
w_{h\ell}^{\nu}
\sum_{a=1}^{4}\Delta_{h\ell}^{\nu,(a)},
\]
the preceding bounds imply
\[
\left|\widetilde\omega^2-
\omega_p^2\right|
=
O_p\left(r_{\omega,n}^{\mathrm{mc}}\right).
\]
\begingroup
Finally, $x\mapsto\max(x,n^{-1})$ is one-Lipschitz and $\omega_p^2\ge2C_0^2>n^{-1}$ for all sufficiently large $n$, so
\[
\begin{aligned}
\left|\widehat\omega^2-
\omega_p^2\right|=
\left|\max(\widetilde\omega^2,n^{-1})-\max(\omega_p^2,n^{-1})\right|\le
\left|\widetilde\omega^2-
\omega_p^2\right|=
O_p\left(r_{\omega,n}^{\mathrm{mc}}\right).
\end{aligned}
\]
\endgroup

Assumption~\ref{ass:C2} gives $\omega_p^2=(2/p)\Tr (\bOme^2)\ge 2C_0^2$. Since $r_{\omega,n}^{\mathrm{mc}}\to0$, with probability tending to one, $\widehat\omega^2\ge\omega_p^2-|\widehat\omega^2-\omega_p^2|\ge C_0^2$. On this event,
\[
\begin{aligned}
\left|\frac{\widehat\omega}{\omega_p}-1\right|
&=
\frac{|\widehat\omega^2-\omega_p^2|}
{\omega_p(\widehat\omega+\omega_p)}\le
C|\widehat\omega^2-\omega_p^2|,
\end{aligned}
\]
which proves $\left|\widehat\omega/\omega_p-1\right|=O_p\left(r_{\omega,n}^{\mathrm{mc}}\right)$.
}

\begingroup
(ii) Work with fixed bandwidths $\ell_j$ in the set
$\mathcal L_n$ defined in the proof of Lemma~\ref{lem:lrv-general}.
All bounds below hold uniformly over at most $n$ candidates. In each
maximum-moment or tail bound, taking their union introduces at most an
additional factor $n^{1/q}$, bounded uniformly when $q\ge\log(np)$.
For coordinate $j$, let $D_{i,j}^{(3)}=D_{i,j}^{(3),\epsilon}+d_{i,j}^{(3),\mu}$, and let $\widetilde\sigma_{j,\epsilon}^2$ denote the same lag-window estimator as $\widetilde\sigma_j^2$, with every $D_{i,j}^{(3)}$ replaced by $D_{i,j}^{(3),\epsilon}$.
\endgroup

{%
Since $\ell_{\max}/\Delta_n\rightarrow0$ and $K$ is fixed, we have $\#\{i:d_{i,j}^{(3),\mu}\neq0\}\leq C\ell_j$ and $\sum_i\{d_{i,j}^{(3),\mu}\}^2\leq C\ell_j\max_{1\le r\le K}\|\bdelta_r\|_\infty ^2$.
}

For a fixed $|v|<\ell_j$, the deterministic
signal--signal term satisfies
\[
\begin{aligned}
&
\frac1n
\sum_i
\left|
d_{i,j}^{(3),\mu}
d_{i-|v|,j}^{(3),\mu}
\right|\leq
\frac{1}{2n}
\sum_i
\left[
\{d_{i,j}^{(3),\mu}\}^2
+
\{d_{i-|v|,j}^{(3),\mu}\}^2
\right]\leq
C\frac{\ell_j\max_{1\le r\le K}\|\bdelta_r\|_\infty ^2}{n}.
\end{aligned}
\]

{%

Lemma~\ref{lem:vector-basic}(ii) gives, for every $q\geq2$,
\[
\begin{aligned}
&
\left\|
\frac1n
\sum_i
d_{i,j}^{(3),\mu}
D_{i-|v|,j}^{(3),\epsilon}
\right\|_q\leq
\frac{Cq^{c_{\mathrm{dep}}}}{n}
\left[
\sum_i
\{d_{i,j}^{(3),\mu}\}^2
\right]^{1/2}\leq
Cq^{c_{\mathrm{dep}}}
\frac{\sqrt{\ell_j}\,\max_{1\le r\le K}\|\bdelta_r\|_\infty }{n},
\end{aligned}
\]
where $D_{i-|v|,j}^{(3),\epsilon}=\sum_{a=0}^{3}d_a^{(3)}\epsilon_{i-|v|-2a\ell_j,j}$.

The same calculation yields
\[
\left\|
\frac1n
\sum_i
D_{i,j}^{(3),\epsilon}
d_{i-|v|,j}^{(3),\mu}
\right\|_q
\leq
Cq^{c_{\mathrm{dep}}}
\frac{\sqrt{\ell_j}\,\max_{1\le r\le K}\|\bdelta_r\|_\infty }{n}.
\]
}

\begingroup
Let $\widetilde\sigma_j^2-\widetilde\sigma_{j,\epsilon}^2=R_{j}^{\mathrm{ss}}+R_{j}^{\mathrm{sn}}$, where
\[
R_j^{\mathrm{ss}}
=
\sum_{|v|<\ell_j}K_{\mathrm{LRV}}(v/\ell_j)
\frac1n\sum_i
d_{i,j}^{(3),\mu}d_{i-|v|,j}^{(3),\mu},
\]
\endgroup
and
\begingroup
\[
R_j^{\mathrm{sn}}
=
\sum_{|v|<\ell_j}K_{\mathrm{LRV}}(v/\ell_j)
\frac1n\sum_i
\left\{
d_{i,j}^{(3),\mu}D_{i-|v|,j}^{(3),\epsilon}
+
D_{i,j}^{(3),\epsilon}d_{i-|v|,j}^{(3),\mu}
\right\}.
\]
\endgroup
Since $K_{\mathrm{LRV}}$ is bounded and supported on
$[-1,1]$, the preceding fixed-lag bounds and Minkowski's
inequality yield
\[
|R_j^{\mathrm{ss}}|
\le
C\frac{\ell_j^2\max_{1\le r\le K}\|\bdelta_r\|_\infty ^2}{n},
\qquad
\|R_j^{\mathrm{sn}}\|_q
\le
Cq^{c_{\mathrm{dep}}}
\frac{\ell_j^{3/2}\max_{1\le r\le K}\|\bdelta_r\|_\infty }{n}.
\]

{%
Taking $q=C_A\log(np)$ and using $\| \max_{1\leq j\leq p}|Y_j| \|_q \leq p^{1/q}\max_{1\leq j\leq p}\|Y_j\|_q\lesssim \max_{1\leq j\leq p}\|Y_j\|_q$, gives
\begingroup
\[
\begin{aligned}
\max_{1\leq j\leq p}
\left|
\widetilde\sigma_j^2
-\widetilde\sigma_{j,\epsilon}^2
\right|
=
O_p\Bigg(
&
\frac{
\ell_{\max}^2\max_{1\le r\le K}\|\bdelta_r\|_\infty ^2
}{n}+
\frac{
\ell_{\max}^{3/2}\max_{1\le r\le K}\|\bdelta_r\|_\infty 
\{\log(np)\}^{c_{\mathrm{LRV}}}
}{n}
\Bigg).
\end{aligned}
\]
\endgroup
}
\begingroup

\endgroup
For all sufficiently large $n$, Young's inequality gives
\[
\begin{aligned}
\frac{\ell_{\max}^{3/2}\max_{1\le r\le K}\|\bdelta_r\|_\infty 
\{\log(np)\}^{c_{\mathrm{LRV}}}}{n}
\le&
\frac12
\left\{
\frac{\ell_{\max}\{\log(np)\}^{c_{\mathrm{LRV}}}}{\sqrt n}
\right\}^{2}
+
\frac12\frac{\ell_{\max}\max_{1\le r\le K}\|\bdelta_r\|_\infty ^2}{n}
\\
\le&
C\left(
a_{\sigma,n}
+
\frac{\ell_{\max}^2\max_{1\le r\le K}\|\bdelta_r\|_\infty ^2}{n}
\right).
\end{aligned}
\]
Consequently,
\begingroup
\[
\max_{1\leq j\leq p}
\left|
\widetilde\sigma_j^2-
\widetilde\sigma_{j,\epsilon}^2
\right|
=
O_p\left(
a_{\sigma,n}
+
\frac{\ell_{\max}^2\max_{1\le r\le K}\|\bdelta_r\|_\infty ^2}{n}
\right).
\]
Lemma~\ref{lem:lrv-general} applied to the residual sequence gives $\max_{1\leq j\leq p}|\widetilde\sigma_{j,\epsilon}^2-\sigma_j^2 |=O_p(a_{\sigma,n})$, and hence $\max_{1\leq j\leq p}|\widetilde\sigma_j^2-\sigma_j^2|=O_p(a_{\sigma,n}^{\mathrm{mc}})$. Assumption~\ref{ass:C2} gives $\min_{1\leq j\leq p}\sigma_j^2\geq C_0$. Since $a_{\sigma,n}^{\mathrm{mc}}\to0$, with probability tending to one, $\min_{1\leq j\leq p}\widetilde\sigma_j^2\geq C_0/2$. On this event,
\[
\begin{aligned}
\max_{1\leq j\leq p}
\left|
\frac{\widehat\sigma_j}{\sigma_j}-1
\right|
&=
\max_{1\leq j\leq p}
\frac{
|\widetilde\sigma_j^2-\sigma_j^2|
}{
\sigma_j(\widehat\sigma_j+\sigma_j)
}\leq
C
\max_{1\leq j\leq p}
|\widetilde\sigma_j^2-\sigma_j^2|=
O_p\left(a_{\sigma,n}^{\mathrm{mc}}\right).
\end{aligned}
\]
\endgroup

It remains to verify the convergence of the two dense
nuisance rates. By Assumption~\ref{ass:WBS-geometry}(ii), {we have}
\[
\begin{aligned}
\frac{
M^2\displaystyle\max_{1\le r\le K}\|\bdelta_r\|_2^2
}{n\sqrt p}
&\rightarrow 0,~~~~
\frac{
M^3\displaystyle\max_{1\le r\le K}\|\bdelta_r\|_2
}{n\sqrt p}
=
\left\{
\frac{
M^2\displaystyle\max_{1\le r\le K}\|\bdelta_r\|_2^2
}{n\sqrt p}
\right\}^{1/2}
\frac{M^2}{\sqrt n\,p^{1/4}}\rightarrow 0.
\end{aligned}
\]

Also, $M/\sqrt n\leq Cn^{-3/8} \rightarrow0$. Under Assumption~\ref{ass:C1}(MD), since $M\rightarrow\infty$, we have $M>m_0$ and $n/3>m_0$ for all sufficiently large $n$. Hence, $\eta_M=0$ and $\varpi_{\mathrm{dep}}(n/3)=0$ for all sufficiently large $n$.

Under Assumption~\ref{ass:C1}(AM) or (PD), the dependence-tail bounds give
\[
\eta_M\leq C\exp(-cM^{\varkappa}),\qquad\varpi_{\mathrm{dep}}(n/3)\leq C\exp(-cn^{\varkappa}),
\]
where the constants may change between displays. If $p\leq n$, then $M\asymp p^{1/8}$, and therefore $\sqrt p\,M\eta_M\leq Cp^{5/8}\exp(-cp^{\varkappa/8})=o(1)$ and $M^2\eta_M\leq Cp^{1/4}\exp(-cp^{\varkappa/8})=o(1)$. If $p>n$, then $M\asymp n^{1/8}$. Since $\log p=o(n^{\varkappa/8})$ and $\log n=o(n^{\varkappa/8})$, we have
\[
\begin{aligned}
\sqrt p\,M\eta_M&\leq C\exp\left\{ \frac12\log p+\frac18\log n-cn^{\varkappa/8}\right\}=o(1),
\\
M^2\eta_M&\leq Cn^{1/4}\exp(-cn^{\varkappa/8})=o(1).
\end{aligned}
\]
Moreover, under Assumption~\ref{ass:C1}(AM) or (PD), $M^2\varpi_{\mathrm{dep}}(n/3)\leq Cn^{1/4}\exp(-cn^{\varkappa}) =o(1)$.

Finally, since $M^2\leq C(n\wedge p)^{1/4}$, $\sqrt{np}\geq n\wedge p$ and $n\geq n\wedge p$, we obtain $M^2\{1/(np)+1/n^2\}^{1/2}\leq C(n\wedge p)^{-3/4}=o(1)$. 
\begingroup
Combining these bounds proves $r_{\mu,n}^{\mathrm{mc}}+r_{\omega,n}^{\mathrm{mc}}\rightarrow0$. Finally, Lemma~\ref{lemma:LRV-rate-log} with $A=1$ gives $a_{\sigma,n}\to0$, and Assumption~\ref{ass:WBS-tuning}(iv) gives $\ell_{\max}^2\max_r\|\bdelta_r\|_\infty^2/n\to0$. Hence $a_{\sigma,n}^{\mathrm{mc}}\to0$, as required.%
\endgroup
\end{proof}

For $1\le j\le  K $, define
\[
\mathcal L_{j,n}^{\mathrm W}
=\left\{\tau_j-\left\lfloor\frac{\Delta_n}{4}\right\rfloor,
\ldots,
\tau_j-\left\lceil\frac{\Delta_n}{8}\right\rceil\right\},~~~~\mathcal R_{j,n}^{\mathrm W}
=\left\{\tau_j+\left\lceil\frac{\Delta_n}{8}\right\rceil,
\ldots,
\tau_j+\left\lfloor\frac{\Delta_n}{4}\right\rfloor\right\}.
\]
Let $\mathcal C_n^{\mathrm W}$ be the event that every change has at least one sampled interval with its endpoints in $\mathcal L_{j,n}^{\mathrm W}$ and $\mathcal R_{j,n}^{\mathrm W}$, i.e.,
$$
\mathcal C_n^{\mathrm W}
=
\bigcap_{j=1}^{K}
\bigcup_{r=1}^{N_{\mathrm W,n}}
\left\{
S_r^{\mathrm W}\in\mathcal L_{j,n}^{\mathrm W},
\quad
E_r^{\mathrm W}\in\mathcal R_{j,n}^{\mathrm W}
\right\},
$$
where $I_r^{\mathrm W}=(S_r^{\mathrm W},E_r^{\mathrm W}]$ is the $r$th sampled interval, $r=1,\ldots,N_{W,n}$.
\begin{lemma}%
\label{lem:wbs-coverage-main}
Under Assumptions~\ref{ass:WBS-geometry}--\ref{ass:WBS-tuning}, with $c_{\mathrm W}=1/100$ and all sufficiently large $n$,
\[
\Pr\{(\mathcal C_n^{\mathrm W})^c\}
\le  K \{1-c_{\mathrm W}(\Delta_n/n)^2\}^{N_{\mathrm W,n}}.
\]
On $\mathcal C_n^{\mathrm W}$, each $\tau_j$ has an interval
$I_{j,n}^{\mathrm{iso}}=(s_{j,n}^{\mathrm{iso}},e_{j,n}^{\mathrm{iso}}]$
that contains no other change and satisfies
\[
\frac{\Delta_n}{8}-1
\le \tau_j-s_{j,n}^{\mathrm{iso}}
\le\frac{\Delta_n}{4}+1,~~~~\frac{\Delta_n}{8}-1
\le e_{j,n}^{\mathrm{iso}}-\tau_j
\le\frac{\Delta_n}{4}+1.
\]
In particular,
\[
\frac{\Delta_n}{4}-2
\le |I_{j,n}^{\mathrm{iso}}|
\le\frac{\Delta_n}{2}+2.
\]
\end{lemma}

\begin{proof}
The number of admissible endpoint pairs is no larger than
\[
\sum_{s=0}^{n-2}(n-s-1)=\frac{n(n-1)}2.
\]
For all sufficiently large $n$,
\[
|\mathcal L_{j,n}^{\mathrm W}|
\ge\frac{\Delta_n}{9},
\qquad
|\mathcal R_{j,n}^{\mathrm W}|
\ge\frac{\Delta_n}{9}.
\]
Because $\ell_{\mathrm W,n}=o(\Delta_n)$, every pair formed from these endpoint sets is admissible eventually.  Hence one draw isolates the $j$th change with probability at least
\[
\frac{(\Delta_n/9)^2}{n(n-1)/2}
\ge \frac1{100}(\Delta_n/n)^2
=c_{\mathrm W}(\Delta_n/n)^2.
\]
Independence of the draws gives the failure probability for one change.  A union bound over $j=1,\ldots, K $ proves the displayed inequality.  The endpoint construction gives the stated margins.  Since both endpoints lie strictly inside the two segments adjacent to $\tau_j$, no other change lies in the interval.
\end{proof}

\begingroup

\endgroup
Define $a_{\mathrm{quad},0,n}^{\mathrm W}=\{N_{\mathrm W,n}\log(2n)\}^{1/4}$, $a_{\mathrm{quad},1/2,n}^{\mathrm W}=N_{\mathrm W,n}^{1/4}\{\log(2n)\}^{1/2}$, and 
\[
Q_{\gamma,I}(b)
=v_{I,b}^{-2\gamma}
\frac{\|\mathbf Z_I(b)\|_2^2-
      \E\|\mathbf Z_I(b)\|_2^2}{m_I\sqrt p},
\qquad
\mathcal V_n=\operatorname{span}(\bdelta_1,\ldots,\bdelta_{ K }).
\]

\begingroup
{Under the conditions of Lemma~\ref{lem:wbs-multiple-nuisance}(i), choose a deterministic sequence $h_{S,n}\to\infty$ sufficiently slowly that $h_{S,n}r_{\mu,n}^{\mathrm{mc}}\to0$ and $h_{S,n}r_{\omega,n}^{\mathrm{mc}}\to0$. Under the conditions of Lemma~\ref{lem:wbs-multiple-nuisance}(ii), choose a deterministic sequence $h_{M,n}\to\infty$ sufficiently slowly that $h_{M,n}a_{\sigma,n}^{\mathrm{mc}}\xrightarrow{p}0$. Define the dense feasible event
\[
\begin{aligned}
\mathcal E_{\mathrm{feas},S,n}
:=\Bigg\{&\max_{\gamma\in\{0,1/2\}}\max_{\substack{I=(s,e]\\ m_I\geq \ell_{\mathrm W,n}}}\max_{b\in\mathcal K_\gamma(I)}
\left|\widehat\mu_{\gamma,M,I}(b)-\mu_{\gamma,M,I}(b)\right|
\leq h_{S,n}r_{\mu,n}^{\mathrm{mc}},\\
&\left|\frac{\widehat\omega}{\omega_p}-1\right|
\leq h_{S,n}r_{\omega,n}^{\mathrm{mc}}\Bigg\},
\end{aligned}
\]
and the coordinatewise feasible event
\[
\mathcal E_{\mathrm{feas},M,n}
:=\left\{
\max_{1\leq j\leq p}
\left|\frac{\widehat\sigma_j}{\sigma_j}-1\right|
\leq h_{M,n}a_{\sigma,n}^{\mathrm{mc}}\le\frac14
\right\}.
\]
Their failure probabilities are denoted by
\[
\varepsilon_{\mathrm{feas},S,n}
=\Pr\left(\mathcal E_{\mathrm{feas},S,n}^{c}\right),
\qquad
\varepsilon_{\mathrm{feas},M,n}
=\Pr\left(\mathcal E_{\mathrm{feas},M,n}^{c}\right).
\]}
\endgroup

\begin{lemma}%
\label{lem:wbs-uniform-noise}
\begingroup
Suppose Assumptions~\ref{ass:C1}--\ref{ass:C3} and \ref{ass:WBS-tuning}(i)--(ii) hold.
\endgroup
Let $r=1,\ldots,N_{\mathrm W,n}$,
$b\in\mathcal K_\gamma(I_r^{\mathrm W})$, and $j=1,\ldots,p$.

(i) Under the conditions of Lemma~\ref{lem:wbs-multiple-nuisance}(i),
there is an event
$\mathcal E_{n,\mathrm{noise},S}^{\mathrm W}
\subseteq\mathcal E_{\mathrm{feas},S,n}$ on which
\begin{align}
\max_{r,b}|Q_{\gamma,I_r^{\mathrm W}}(b)|
&\le Ca_{\mathrm{quad},\gamma,n}^{\mathrm W},
\qquad \gamma\in\{0,1/2\},
\label{eq:wbs-dense-noise0-final}\\
\max_{\gamma\in\{0,1/2\}}\max_{r,b}
\sup_{\substack{\mathbf u\in\mathcal V_n\\\|\mathbf u\|_2\le1}}
\frac{v_{I_r^{\mathrm W},b}^{-\gamma}
|\mathbf u^\top\mathbf Z_{I_r^{\mathrm W}}(b)|}
{\sqrt{m_{I_r^{\mathrm W}}}}
&\le C\{\log(np)\}^{c_{\mathrm{dep}}},
\label{eq:wbs-projection-noise-final}
\end{align}
and
\[
\Pr\{(\mathcal E_{n,\mathrm{noise},S}^{\mathrm W})^c\}
\le\frac C{\log(2n)}+\frac C{nN_{\mathrm W,n}}
+\varepsilon_{\mathrm{feas},S,n}\to0.
\]

(ii) Under the conditions of Lemma~\ref{lem:wbs-multiple-nuisance}(ii),
there is an event
$\mathcal E_{n,\mathrm{noise},M}^{\mathrm W}
\subseteq\mathcal E_{\mathrm{feas},M,n}$ on which
\begin{equation}
\max_{\gamma\in\{0,1/2\}}\max_{r,b,j}
\frac{v_{I_r^{\mathrm W},b}^{-\gamma}
|Z_{I_r^{\mathrm W},j}(b)|}
{\sqrt{m_{I_r^{\mathrm W}}}\sigma_j}
\le C\{\log(np)\}^{c_{\mathrm{dep}}},
\label{eq:wbs-coordinate-noise-final}
\end{equation}
and
\[
\Pr\{(\mathcal E_{n,\mathrm{noise},M}^{\mathrm W})^c\}
\le\frac C{npN_{\mathrm W,n}}
+\varepsilon_{\mathrm{feas},M,n}\to0.
\]
\end{lemma}

\begin{proof}
For a deterministic interval $I=(s,e]$ of length $m=e-s$ and
$b\in\mathcal K_\gamma(I)$, define the coefficient vector
$\boldsymbol w_{I,b}^{(\gamma)}=(w_{I,b,i}^{(\gamma)})_{i=1}^n$ by
\[
w_{I,b,i}^{(\gamma)}
=v_{I,b}^{-\gamma}m^{-1/2}
\left\{\mathbb{I}(s<i\le b)
-\frac{b-s}{m}\mathbb{I}(s<i\le e)\right\}.
\]
Then
\[
\sum_iw_{I,b,i}^{(\gamma)}\bepsilon_i
=\frac{v_{I,b}^{-\gamma}\mathbf Z_I(b)}{\sqrt m},
\qquad
\|\boldsymbol w_{I,b}^{(\gamma)}\|_2^2
=v_{I,b}^{1-2\gamma}\le1.
\]

\noindent\textit{(i)} For a deterministic coefficient vector $\boldsymbol w$, let
\[
Q(\boldsymbol w)
=\frac{\|\sum_iw_i\bepsilon_i\|_2^2
-\E\|\sum_iw_i\bepsilon_i\|_2^2}{\sqrt p},
\qquad
Q_{\gamma,I}(b)=Q(\boldsymbol w_{I,b}^{(\gamma)}).
\]
Lemma~\ref{lem:bilinear-quadratic-bound}, applied to the weight pairs
$(\boldsymbol w,\boldsymbol w)$ and
$(\boldsymbol w_2-\boldsymbol w_1,\boldsymbol w_2+\boldsymbol w_1)$,
gives
\[
\|Q(\boldsymbol w)\|_4\le C\|\boldsymbol w\|_2^2,
\qquad
\|Q(\boldsymbol w_2)-Q(\boldsymbol w_1)\|_4
\le C\|\boldsymbol w_2-\boldsymbol w_1\|_2
       \|\boldsymbol w_2+\boldsymbol w_1\|_2.
\]

For $s<b_1<b_2<e$, direct calculation yields
\[
\|\boldsymbol w_{I,b_2}^{(0)}
-\boldsymbol w_{I,b_1}^{(0)}\|_2^2
=\frac{b_2-b_1}{m}
 \left(1-\frac{b_2-b_1}{m}\right)
\le\frac{b_2-b_1}{m}.
\]
Consequently,
\[
\E|Q_{0,I}(b_2)-Q_{0,I}(b_1)|^4
\le C\left(\frac{b_2-b_1}{m}\right)^2.
\]
Extend $Q_{0,I}$ to the endpoints by $Q_{0,I}(s)=Q_{0,I}(e)=0$.
The same increment bound holds there. Applying
\citet[Corollary~1]{moricz1982} to the successive increments,
with moment order $4$ and additive weights $1/m$, gives
\[
\E\left[\max_{b\in\mathcal K_0(I)}|Q_{0,I}(b)|^4\right]\le C.
\]

For $\gamma=1/2$, let
$u_{I,b}=\log\{t_{I,b}/(1-t_{I,b})\}$.
For $b_1<b_2$, the coefficient inner product is
\[
\left\langle\boldsymbol w_{I,b_1}^{(1/2)},
\boldsymbol w_{I,b_2}^{(1/2)}\right\rangle
=\left\{\frac{t_{I,b_1}(1-t_{I,b_2})}
{t_{I,b_2}(1-t_{I,b_1})}\right\}^{1/2}
=\exp\{-(u_{I,b_2}-u_{I,b_1})/2\}.
\]
Since both coefficient vectors have norm one,
\[
\|\boldsymbol w_{I,b_2}^{(1/2)}
-\boldsymbol w_{I,b_1}^{(1/2)}\|_2^2
=2\left[1-\exp\{-(u_{I,b_2}-u_{I,b_1})/2\}\right]
\le u_{I,b_2}-u_{I,b_1},
\]
and hence
\[
\E|Q_{1/2,I}(b_2)-Q_{1/2,I}(b_1)|^4
\le C|u_{I,b_2}-u_{I,b_1}|^2.
\]
Partition the trimmed logit range into at most $C\log(2m)$
intervals of length at most one. Within each nonempty block,
order the split points as $b_1<\cdots<b_L$ and apply
\citet[Corollary~1]{moricz1982} to their successive process
increments, with additive weights
$u_{I,b_\ell}-u_{I,b_{\ell-1}}$, $\ell=2,\ldots,L$.
These weights sum to at most one, so
\[
\E\max_{\ell\le L}
|Q_{1/2,I}(b_\ell)-Q_{1/2,I}(b_1)|^4\le C.
\]
The single-point bound $\E|Q_{1/2,I}(b_1)|^4\le C$
therefore bounds the fourth moment of the maximum on each block.
Summing over the blocks gives
\[
\E\left[\max_{b\in\mathcal K_\gamma(I)}
|Q_{\gamma,I}(b)|^4\right]
\le C\{\log(2m)\}^{2\gamma},
\qquad \gamma\in\{0,1/2\}.
\]
All these bounds are uniform in deterministic $I$. Since the sampled
intervals are independent of the observations, they also hold
conditionally on $(I_r^{\mathrm W})_{r\le N_{\mathrm W,n}}$.
Conditional Markov and union bounds, followed by expectation over
the sampled intervals, yield
\[
\Pr\left\{
\max_{r,b}|Q_{\gamma,I_r^{\mathrm W}}(b)|
>Ca_{\mathrm{quad},\gamma,n}^{\mathrm W}
\right\}
\le
\frac{CN_{\mathrm W,n}\{\log(2n)\}^{2\gamma}}
{(a_{\mathrm{quad},\gamma,n}^{\mathrm W})^4}
\le\frac C{\log(2n)}.
\]

For the projection bound, $K$ is fixed under the conditions of
Lemma~\ref{lem:wbs-multiple-nuisance}(i), and
$\dim(\mathcal V_n)\le K$.
If $\mathcal V_n=\{\mathbf0\}$, the bound is immediate.
Otherwise, let $\mathcal N_n$ be a $1/2$-net of its unit sphere.
Then $|\mathcal N_n|\le5^K$ and
\[
\sup_{\substack{\mathbf u\in\mathcal V_n\\\|\mathbf u\|_2\le1}}
|\mathbf u^\top\mathbf z|
\le2\max_{\mathbf u\in\mathcal N_n}|\mathbf u^\top\mathbf z|.
\]
Lemma~\ref{lem:vector-basic}(ii) gives, for every $q\ge2$,
\[
\left\|
\frac{v_{I,b}^{-\gamma}\mathbf u^\top\mathbf Z_I(b)}{\sqrt m}
\right\|_q\le Cq^{c_{\mathrm{dep}}},
\qquad \|\mathbf u\|_2=1.
\]
Take $q=2\log(2\cdot5^K nN_{\mathrm W,n})$.
Assumption~\ref{ass:WBS-tuning}(ii) gives $q\le C\log(np)$.
Using $\Pr(|Y|>e\|Y\|_q)\le e^{-q}$ and a union bound over
the net, the two values of $\gamma$, and at most $nN_{\mathrm W,n}$
interval--split pairs, conditionally on the sampled intervals, gives
\[
\Pr\left\{
\max_{\gamma\in\{0,1/2\}}\max_{r,b}
\sup_{\substack{\mathbf u\in\mathcal V_n\\\|\mathbf u\|_2\le1}}
\frac{v_{I_r^{\mathrm W},b}^{-\gamma}
|\mathbf u^\top\mathbf Z_{I_r^{\mathrm W}}(b)|}
{\sqrt{m_{I_r^{\mathrm W}}}}
>C\{\log(np)\}^{c_{\mathrm{dep}}}
\right\}
\le 2\cdot5^K nN_{\mathrm W,n}e^{-q}
\le\frac C{nN_{\mathrm W,n}}.
\]
Here the displayed probability is unconditional, obtained by averaging
the uniform conditional bound.

Intersect the two quadratic-form events and the projection event with
$\mathcal E_{\mathrm{feas},S,n}$ to define
$\mathcal E_{n,\mathrm{noise},S}^{\mathrm W}$.
The union bound proves the failure-probability bound in part~(i).
Lemma~\ref{lem:wbs-multiple-nuisance}(i) and the choice of $h_{S,n}$
give $\varepsilon_{\mathrm{feas},S,n}\to0$.

\noindent\textit{(ii)} Since $\sigma_j\ge\sqrt{C_0}$ by Assumption~\ref{ass:C2},
Lemma~\ref{lem:vector-basic}(ii) implies
\[
\left\|
\frac{v_{I,b}^{-\gamma}Z_{I,j}(b)}{\sqrt m\,\sigma_j}
\right\|_q\le Cq^{c_{\mathrm{dep}}}.
\]
Taking $q=2\log(2npN_{\mathrm W,n})\le C\log(np)$ and
applying the same conditional union bound gives
\[
\Pr\left\{
\max_{\gamma\in\{0,1/2\}}\max_{r,b,j}
\frac{v_{I_r^{\mathrm W},b}^{-\gamma}|Z_{I_r^{\mathrm W},j}(b)|}
{\sqrt{m_{I_r^{\mathrm W}}}\sigma_j}
>C\{\log(np)\}^{c_{\mathrm{dep}}}
\right\}
\le 2npN_{\mathrm W,n}e^{-q}
\le\frac C{npN_{\mathrm W,n}}.
\]

Intersect the coordinate event with $\mathcal E_{\mathrm{feas},M,n}$ to define
$\mathcal E_{n,\mathrm{noise},M}^{\mathrm W}$.
The union bound proves the failure-probability bound in part~(ii).
Lemma~\ref{lem:wbs-multiple-nuisance}(ii) and the choice of $h_{M,n}$
give $\varepsilon_{\mathrm{feas},M,n}\to0$.
For later use, on $\mathcal E_{\mathrm{feas},M,n}$ and for all
sufficiently large $n$,
\[
\max_{j\le p}\left|\frac{\sigma_j}{\widehat\sigma_j}-1\right|
\le\frac{h_{M,n}a_{\sigma,n}^{\mathrm{mc}}}
{1-h_{M,n}a_{\sigma,n}^{\mathrm{mc}}}
\le2h_{M,n}a_{\sigma,n}^{\mathrm{mc}}.
\]
\end{proof}

Under the conditions of both parts of Lemma~\ref{lem:wbs-uniform-noise},
define the adaptive event
\[
\mathcal E_{n,\mathrm{noise},A}^{\mathrm W}
:=\mathcal E_{n,\mathrm{noise},S}^{\mathrm W}
\cap\mathcal E_{n,\mathrm{noise},M}^{\mathrm W}.
\]
By the union bound,
\[
\Pr\{(\mathcal E_{n,\mathrm{noise},A}^{\mathrm W})^c\}
\le
\Pr\{(\mathcal E_{n,\mathrm{noise},S}^{\mathrm W})^c\}
+\Pr\{(\mathcal E_{n,\mathrm{noise},M}^{\mathrm W})^c\}
\to0.
\]

Let $I=(s,e]$ contain exactly one change $\tau_j$, and $a=\tau_j-s$, $c=e-\tau_j$, $m=a+c$. {Let $b_{\gamma,I}^0$ be the closest point to $\tau_j$ in $\mathcal K_\gamma(I)$. Then} $|b_{0,I}^0-\tau_j|=0$ and $|b_{1/2,I}^0-\tau_j|\le\lambda_m^{\mathrm W}$.

\begin{lemma}%
\begingroup
\label{lem:wbs-one-change-interval} For parts (ii) and (iii), assume that $I$ belongs to the sampled interval pool.
\endgroup
{Suppose Assumptions~}\ref{ass:C1}--\ref{ass:C3}, \ref{ass:WBS-geometry}(i) %
hold. 
 (i) The point $b_{\gamma,I}^0$ maximizes $G_{S,\gamma,I}$ and $G_{M,\gamma,I}$ over $\mathcal K_\gamma(I)$. For a universal $c_0>0$, and for every $b\in\mathcal{K}_{\gamma}(I)$,
\begin{equation}\label{eq:wbs-relative-gap-final}
\begin{aligned}
    &G_{S,\gamma,I}(b_{\gamma,I}^0)-G_{S,\gamma,I}(b)\ge c_0\frac{|b-b_{\gamma,I}^0|}{m}G_{S,\gamma,I}(b_{\gamma,I}^0),\\
    &G_{M,\gamma,I}(b_{\gamma,I}^0)-G_{M,\gamma,I}(b)\ge c_0\frac{|b-b_{\gamma,I}^0|}{m}G_{M,\gamma,I}(b_{\gamma,I}^0).
\end{aligned}
\end{equation}
Suppose further that Assumptions~\ref{ass:lrv-M}(i)--(ii) and \ref{ass:WBS-tuning}(iv) hold. Then, on ${\mathcal E_{\mathrm{feas},M,n}}$, the point $b_{\gamma,I}^0$ also maximizes $\widehat G_{M,\gamma,I}$ and the same inequality holds with $G_{S,\gamma,I}$ replaced by either $G_{M,\gamma,I}$ or $\widehat G_{M,\gamma,I}$.

(ii) Suppose further that Assumption~\ref{ass:WBS-geometry}(ii), \ref{ass:WBS-tuning}(i)--(ii), $M=\lceil(n\wedge p)^{1/8}\rceil$ and $B_{S,\gamma,n}^{\mathrm W}=o(\Lambda_{S,\gamma,n}^{\mathrm W})$ hold. Then, on ${\mathcal E_{n,\mathrm{noise},S}^{\mathrm W}}$, if
$\mathfrak R_{S,\gamma}^{\mathrm W}(I)>1$, 
\[
|\widehat{\tau}_{S,\gamma}(I)-\tau_j|
\le\ind{\gamma=1/2}\lambda_m^{\mathrm W}
+C m\left\{
\frac{B_{S,\gamma,n}^{\mathrm W}}{\Lambda_{S,\gamma,n}^{\mathrm W}}
+\frac{\{\log(np)\}^{c_{\mathrm{dep}}}}
{p^{1/4}\{\Lambda_{S,\gamma,n}^{\mathrm W}\}^{1/2}}
\right\}.
\]

(iii) Suppose further that Assumption~\ref{ass:lrv-M}(i)--(ii), \ref{ass:WBS-tuning}(i)--(ii), (iv) and $\{\log(np)\}^{c_{\mathrm{dep}}}=o(\Lambda_{M,\gamma,n}^{\mathrm W})$ hold. Then, on ${\mathcal E_{n,\mathrm{noise},M}^{\mathrm W}}$, if $\mathfrak R_{M,\gamma}^{\mathrm W}(I)>1$,
\[
|\widehat{\tau}_{M,\gamma}(I)-\tau_j|\le\ind{\gamma=1/2}\lambda_m^{\mathrm W}+C m\frac{\{\log(np)\}^{c_{\mathrm{dep}}}}{\Lambda_{M,\gamma,n}^{\mathrm W}}.
\]

(iv) If $I=I_{j,n}^{\mathrm{iso}}$ is a coverage interval, then
\[
G_{S,\gamma,I}(\tau_j)
\ge c\frac{\Delta_n\|\bdelta_j\|_2^2}{\sqrt p},
\qquad
G_{M,\gamma,I}(\tau_j)
\ge c\sqrt{\Delta_n}\|\bdelta_j\|_\infty.
\]
Suppose further that Assumption~\ref{ass:lrv-M}(i)--(ii) and \ref{ass:WBS-tuning}(iv) hold. Then on ${\mathcal E_{\mathrm{feas},M,n}}$, the second lower bound also holds with $G_{M,\gamma,I}$ replaced by $\widehat G_{M,\gamma,I}$ after changing $c$.
\end{lemma}

\begin{proof}
(i) The deterministic parts in CUSUM statistics are
\[
\mathbf m_I(b)=-\bdelta_j
\begin{cases}
(b-s)(e-\tau_j)/m,&b\le\tau_j,\\
(\tau_j-s)(e-b)/m,&b>\tau_j.
\end{cases}
\]

For max-$L_2$ statistics,
\begin{align*}
G_{S,0,I}(b)&=\frac{(b-s)^2(e-\tau_j)^2}{m^3\sqrt p}\|\bdelta_j\|_2^2,~~~~G_{S,1/2,I}(b)=\frac{(b-s)(e-\tau_j)^2}{m(m-b+s)\sqrt p}\|\bdelta_j\|_2^2,
\end{align*}
for $b\le\tau_j$; and for $b>\tau_j$,%
\begin{align*}
G_{S,0,I}(b)
&=\frac{(\tau_j-s)^2(e-b)^2}{m^3\sqrt p}\|\bdelta_j\|_2^2,~~~~G_{S,1/2,I}(b)=\frac{(\tau_j-s)^2(e-b)}{m(m-e+b)\sqrt p}\|\bdelta_j\|_2^2.
\end{align*}
\begingroup
If $\bdelta_j=0$, every deterministic profile and gap is zero; assume otherwise. Let $a_\star=\max_{1\le l\le p}|\delta_{j,l}|/\sigma_l$ and $\widehat a_\star=\max_{1\le l\le p}|\delta_{j,l}|/\widehat\sigma_l$. On ${\mathcal E_{\mathrm{feas},M,n}}$, $(1-{C h_{M,n} a_{\sigma,n}^{\mathrm{mc}}})a_\star\le \widehat a_\star\le (1+{C h_{M,n} a_{\sigma,n}^{\mathrm{mc}}})a_\star$ for all sufficiently large $n$, i.e., $ \widehat a_\star\asymp a_{\star}$. For max-$L_\infty$ statistics, 
\endgroup
\begin{align*}
G_{M,0,I}(b)
&=\frac{(b-s)(e-\tau_j)}{m^{3/2}}a_\star,~~~~G_{M,1/2,I}(b)=\frac{(e-\tau_j)\sqrt{b-s}}{\sqrt m\sqrt{m-b+s}}a_\star,
\end{align*}
for $b\leq\tau_j$; and for $b>\tau_j$ 
\begin{align*}
G_{M,0,I}(b)
&=\frac{(\tau_j-s)(e-b)}{m^{3/2}}a_\star,~~~~G_{M,1/2,I}(b)=\frac{(\tau_j-s)\sqrt{e-b}}{\sqrt m\sqrt{m-e+b}}a_\star.
\end{align*}
The formulas for $\widehat G_{M,\gamma,I}$ are identical with $a_\star$ replaced by $\widehat a_\star$. %
Therefore $G_{S,\gamma,I}$, $G_{M,\gamma,I}$, and, on ${\mathcal E_{\mathrm{feas},M,n}}$, $\widehat G_{M,\gamma,I}$ are maximized at $b_{\gamma,I}^0$ in $\mathcal{K}_{\gamma}(I)$.

 For max-$L_2$ statistics, {first consider the case} $b_{\gamma,I}^0=\tau_j$. If $b\le\tau_j$, %
\begin{align*}
G_{S,0,I}(\tau_j)-G_{S,0,I}(b)
&=\frac{(e-\tau_j)^2(\tau_j-b)(2(\tau_j-s)-\tau_j+b)}{m^3\sqrt p}\|\bdelta_j\|_2^2\ge\frac{\vert\tau_j-b\vert}{m}G_{S,0,I}(\tau_j),\\
\frac{G_{S,1/2,I}(\tau_j)-G_{S,1/2,I}(b)}
{G_{S,1/2,I}(\tau_j)}
&=\frac{(\tau_j-b)m}{(\tau_j-s)(e-\tau_j+\tau_j-b)}
\ge\frac{\vert \tau_j-b\vert}{m};
\end{align*}
{the same calculation applies for $b\ge\tau_j$.}%
If $b_{1/2,I}^0\ne\tau_j$, {without loss of generality,} $b_{1/2,I}^0<\tau_j$, %
\begin{align*}
&\frac{G_{S,1/2,I}(b_{1/2,I}^0)-G_{S,1/2,I}(b)}{G_{S,1/2,I}(b_{1/2,I}^0)}=\frac{(b_{1/2,I}^0-b)m}{(b_{1/2,I}^0-s)(m-b+s)}\ge\frac{\vert b_{1/2,I}^0-b\vert}{m}.
\end{align*}
Thus {the first inequality in} \eqref{eq:wbs-relative-gap-final} holds with $c_0=1$.

\begingroup
For max-$L_\infty$ statistics,\[
G_{M,\gamma,I}(b)^2
=\frac{a_\star^2\sqrt p}{\|\bdelta_j\|_2^2}
 G_{S,\gamma,I}(b),
\qquad
\widehat G_{M,\gamma,I}(b)^2
=\frac{\widehat a_\star^2\sqrt p}{\|\bdelta_j\|_2^2}
 G_{S,\gamma,I}(b),
\]
\endgroup
for every $b\in\mathcal K_\gamma(I)$. Consequently, on ${\mathcal E_{\mathrm{feas},M,n}}$ for
$H\in\{G_{M,\gamma,I},\widehat G_{M,\gamma,I}\}$, 
\begin{align*}
H(b_{\gamma,I}^0)-H(b)
&=\frac{H(b_{\gamma,I}^0)^2-H(b)^2}
        {H(b_{\gamma,I}^0)+H(b)}\ge \frac12\frac{|b-b_{\gamma,I}^0|}{m} H(b_{\gamma,I}^0).
\end{align*}
Taking $c_0=1/2$ {proves the second inequality in} \eqref{eq:wbs-relative-gap-final}.

(ii) For the max-$L_2$ statistic,
\begin{align*}
&W_{\gamma,I}(b)-\widehat{\mu}_{\gamma,M,I}(b)
-G_{S,\gamma,I}(b)\\
&\quad=Q_{\gamma,I}(b)
+\frac{2v_{I,b}^{-2\gamma}
\mathbf m_I(b)^\top\mathbf Z_I(b)}{m\sqrt p}
+\mu_{\gamma,M,I}(b)-\widehat{\mu}_{\gamma,M,I}(b)
+R_{\gamma,M,I}^{\mathrm{tail}}(b),
\end{align*}
where \[
R_{\gamma,M,I}^{\mathrm{tail}}(b)
=\frac{v_{I,b}^{-2\gamma}\E\|\mathbf Z_I(b)\|_2^2}
       {m\sqrt p}
 -\mu_{\gamma,M,I}(b).
\]
 From the proof in Lemma~\ref{lem:wbs-multiple-nuisance} and Lemma~\ref{lem:vector-basic}(i), \begin{equation}\label{eq:wbs-dense-centering-tail}
     |R_{\gamma,M,I}^{\mathrm{tail}}(b)|\le C\sqrt p\,\eta_M\le C r_{\mu,n}^{\mathrm{mc}}.
 \end{equation}
Since $\mathbf m_I(b)$ is a scalar multiple of $\bdelta_j$, {Equation~}\eqref{eq:wbs-projection-noise-final} gives, on ${\mathcal E_{n,\mathrm{noise},S}^{\mathrm W}}$,
\[
\left|
\frac{2v_{I,b}^{-2\gamma}
\mathbf m_I(b)^\top\mathbf Z_I(b)}{m\sqrt p}
\right|
\le C\frac{\{\log(np)\}^{c_{\mathrm{dep}}}}{p^{1/4}}
\{G_{S,\gamma,I}(b)\}^{1/2}.
\]
Hence, on ${\mathcal E_{n,\mathrm{noise},S}^{\mathrm W}}$, uniformly in $b$,
\begin{equation}\label{eq:wbs-one-change-dense-error}
\left|W_{\gamma,I}(b)-\widehat{\mu}_{\gamma,M,I}(b)
-G_{S,\gamma,I}(b)\right|
\le a_{\gamma,n}^{\mathrm{err,W}}+b_n^{\mathrm{cross,W}}\{G_{S,\gamma,I}(b)\}^{1/2},
\end{equation}
where $a_{\gamma,n}^{\mathrm{err,W}}=C\{a_{\mathrm{quad},\gamma,n}^{\mathrm W}+{h_{S,n}} r_{\mu,n}^{\mathrm{mc}}\}$ and $
b_n^{\mathrm{cross,W}}=C\{\log(np)\}^{c_{\mathrm{dep}}}p^{-1/4}$. Young's inequality gives $b_n^{\mathrm{cross,W}}\{G_{S,\gamma,I}(b)\}^{1/2}\le\frac14 G_{S,\gamma,I}(b)+C\{b_n^{\mathrm{cross,W}}\}^2$. By Assumption~\ref{ass:WBS-tuning}(ii) and ${h_{S,n}} r_{\mu,n}^{\mathrm{mc}}\to0$, we have $$a_{\gamma,n}^{\mathrm{err,W}}+C\{b_n^{\mathrm{cross,W}}\}^2\le C\{ a_{\mathrm{quad},\gamma,n}^{\mathrm W}+p^{-1/2}\{\log(np)\}^{2c_{\mathrm{dep}}}\}=C B_{S,\gamma,n}^{\mathrm W}.$$

Since $\mathfrak R_{S,\gamma}^{\mathrm W}(I)>1$, $B_{S,\gamma,n}^{\mathrm W}=o(\Lambda_{S,\gamma,n}^{\mathrm W})$, and $\widehat\omega_{\mathrm W}$ is bounded away from zero and infinity, Equation~\eqref{eq:wbs-one-change-dense-error} implies $G_{S,\gamma,I}(b_{\gamma,I}^0)\ge c\Lambda_{S,\gamma,n}^{\mathrm W}$. 

  The maximizing property and Equation~\eqref{eq:wbs-one-change-dense-error} give
\begin{align*}
G_{S,\gamma,I}(b_{\gamma,I}^0)-G_{S,\gamma,I}(\widehat{\tau}_{S,\gamma}(I))
&\le 2a_{\gamma,n}^{\mathrm{err,W}}+b_n^{\mathrm{cross,W}}\left[G_{S,\gamma,I}(b_{\gamma,I}^0)^{1/2}+\{G_{S,\gamma,I}(\widehat{\tau}_{S,\gamma}(I))\}^{1/2}\right]\\
&\le C\left[B_{S,\gamma,n}^{\mathrm W}+\{\log(np)\}^{c_{\mathrm{dep}}}p^{-1/4}\{G_{S,\gamma,I}(b_{\gamma,I}^0)\}^{1/2}\right].
\end{align*}
Combining Equation~\eqref{eq:wbs-relative-gap-final} yields
\[
|\widehat{\tau}_{S,\gamma}(I)-b_{\gamma,I}^0|
\le Cm\left\{
\frac{B_{S,\gamma,n}^{\mathrm W}}{\Lambda_{S,\gamma,n}^{\mathrm W}}
+\frac{\{\log(np)\}^{c_{\mathrm{dep}}}}
{p^{1/4}\{\Lambda_{S,\gamma,n}^{\mathrm W}\}^{1/2}}
\right\}.
\]
Combining the gap $|b_{\gamma,I}^0-\tau_j|$ proves (ii).

(iii) For max-$L_{\infty}$ statistics, on
${\mathcal E_{n,\mathrm{noise},M}^{\mathrm W}}$,
Equation~\eqref{eq:wbs-coordinate-noise-final} implies that, uniformly over
$b\in\mathcal K_\gamma(I)$,
\[
\left|
\max_{1\le l\le p}
\left|
v_{I,b}^{-\gamma}
\frac{m_{I,l}(b)+Z_{I,l}(b)}
{\sqrt{m_I}\widehat\sigma_l}
\right|
-
\widehat G_{M,\gamma,I}(b)
\right|
\le
C\{\log(np)\}^{c_{\mathrm{dep}}}.
\]
Since $\widehat\tau_{M,\gamma}(I)$ maximizes the max-$L_\infty$ statistic and, by part (i), $b_{\gamma,I}^{0}$ maximizes $\widehat G_{M,\gamma,I}(b)$, we have
\[
\left|
\mathcal M_\gamma(I)
-
\widehat G_{M,\gamma,I}(b_{\gamma,I}^{0})
\right|
\le
C\{\log(np)\}^{c_{\mathrm{dep}}}.
\]

If $\mathfrak R_{M,\gamma}^{\mathrm W}(I)>1$, by the definition of
$\mathfrak R_{M,\gamma}^{\mathrm W}(I)$, $\mathcal M_\gamma(I)>\Lambda_{M,\gamma,n}^{\mathrm W}$. Hence, $\widehat G_{M,\gamma,I}(b_{\gamma,I}^{0})\ge\Lambda_{M,\gamma,n}^{\mathrm W}-C\{\log(np)\}^{c_{\mathrm{dep}}}$. Under $\{\log(np)\}^{c_{\mathrm{dep}}}
=o(\Lambda_{M,\gamma,n}^{\mathrm W})$, $\widehat G_{M,\gamma,I}(b_{\gamma,I}^{0})\ge
c\Lambda_{M,\gamma,n}^{\mathrm W}$.

Moreover, applying the uniform approximation in Lemma~\ref{lem:wbs-uniform-noise} at
$b=b_{\gamma,I}^{0}$ and
$b=\widehat\tau_{M,\gamma}(I)$, and using the definition of
$\widehat\tau_{M,\gamma}(I)$ as the maximizer of the max-$L_\infty$ statistics, gives
\[
\widehat G_{M,\gamma,I}(b_{\gamma,I}^{0})
-
\widehat G_{M,\gamma,I}
(\widehat\tau_{M,\gamma}(I))
\lesssim
C\{\log(np)\}^{c_{\mathrm{dep}}}.
\]

By Equation~\eqref{eq:wbs-relative-gap-final},
\[
\widehat G_{M,\gamma,I}(b_{\gamma,I}^{0})
-
\widehat G_{M,\gamma,I}
(\widehat\tau_{M,\gamma}(I))
\ge
c
\frac{
|\widehat\tau_{M,\gamma}(I)-b_{\gamma,I}^{0}|
}{m}
\widehat G_{M,\gamma,I}(b_{\gamma,I}^{0}).
\]
Combining Equation~\eqref{eq:wbs-relative-gap-final} yields \[ |\widehat\tau_{M,\gamma}(I)-b_{\gamma,I}^0|\le Cm\frac{\{\log(np)\}^{c_{\mathrm{dep}}}}{\Lambda_{M,\gamma,n}^{\mathrm W}}.%
\]
Combining the gap $|b_{\gamma,I}^0-\tau_j|$ proves (iii).

\begingroup
(iv) For a coverage interval, $a=\tau_j-s$, $c=e-\tau_j$ and
$m=a+c$ are all comparable to $\Delta_n$. Evaluation at the change gives
\[
\begin{array}{c|cc}
&\gamma=0&\gamma=1/2\\ \hline
G_{S,\gamma,I}(\tau_j)
&\dfrac{a^2c^2}{m^3\sqrt p}\|\bdelta_j\|_2^2
&\dfrac{ac}{m\sqrt p}\|\bdelta_j\|_2^2\\[2mm]
G_{M,\gamma,I}(\tau_j)
&\dfrac{ac}{m^{3/2}}a_\star
&\sqrt{\dfrac{ac}{m}}a_\star
\end{array}
\]
Since $a_\star\asymp\|\bdelta_j\|_\infty$, these are the asserted lower
bounds. On $\mathcal E_{\mathrm{feas},M,n}$,
$\widehat a_\star\asymp a_\star$ gives the feasible bound as well.
\endgroup
\end{proof}

\begin{lemma}%
\label{lem:wbs-shortest-one-change}

Suppose Assumptions \ref{ass:C1}--\ref{ass:C3}, \ref{ass:WBS-geometry}(i), and \ref{ass:WBS-tuning}(i)--(ii) hold%
. Fix
$\gamma\in\{0,1/2\}$. Let $(s,e]$ be a recursive segment such that every undetected change point $\tau_j\in(s,e]$ has a coverage interval $I_{j,n}^{\mathrm{iso}}$ fully contained in $(s,e]$.

(i) Suppose further that Assumption~\ref{ass:WBS-geometry}(ii), $M=\left\lceil(n\wedge p)^{1/8}\right\rceil$ and $
B_{S,\gamma,n}^{\mathrm W}=o\!\left(\Lambda_{S,\gamma,n}^{\mathrm W}\right)$ hold. {On $\mathcal C_n^{\mathrm W}\cap\mathcal E_{n,\mathrm{noise},S}^{\mathrm W}$,} for every sampled interval $I\subseteq(s,e]$ containing no true change point and satisfying $\mathcal K_\gamma(I)\ne\varnothing$, $\mathfrak R_{S,\gamma}^{\mathrm W}(I)=o(1)$. If $(s,e]$ contains an undetected change point $\tau_j$ satisfying
\[
\frac{\Delta_n\|\bdelta_j\|_2^2}
{\sqrt p\,\Lambda_{S,\gamma,n}^{\mathrm W}}
\rightarrow\infty,
\]
then, for all sufficiently large $n$, $\mathfrak R_{S,\gamma}^{\mathrm W}\bigl(I_{j,n}^{\mathrm{iso}}\bigr)>1$. Consequently, every shortest sampled interval satisfying $\mathfrak R_{S,\gamma}^{\mathrm W}(I)>1$ has length at most $\Delta_n/2+2$ and contains exactly one true change point. If $(s,e]$ contains no true change point, no sampled interval in $(s,e]$ satisfies $\mathfrak R_{S,\gamma}^{\mathrm W}(I)>1$.

\begingroup
(ii) Suppose further that \ref{ass:lrv-M}(i)--(ii), \ref{ass:WBS-tuning}(iv), $\{\log(np)\}^{c_{\mathrm{dep}}}=o(\Lambda_{M,\gamma,n}^{\mathrm W})$ hold. {On $\mathcal C_n^{\mathrm W}\cap\mathcal E_{n,\mathrm{noise},M}^{\mathrm W}$,} for every sampled interval $I\subseteq(s,e]$ containing no true change point and satisfying $\mathcal K_\gamma(I)\ne\varnothing$, $\mathfrak R_{M,\gamma}^{\mathrm W}(I)=o(1)$. If $(s,e]$ contains an undetected change point $\tau_j$ satisfying
\endgroup
\[
\frac{\sqrt{\Delta_n}\|\bdelta_j\|_\infty}{\Lambda_{M,\gamma,n}^{\mathrm W}}\rightarrow\infty,
\]
then, for all sufficiently large $n$, $\mathfrak R_{M,\gamma}^{\mathrm W}\bigl(I_{j,n}^{\mathrm{iso}}\bigr)>1$. Consequently, every shortest sampled interval satisfying $\mathfrak R_{M,\gamma}^{\mathrm W}(I)>1$ has length at most $\Delta_n/2+2$ and contains exactly one true change point. If $(s,e]$ contains no true change point, no sampled interval in $(s,e]$ satisfies $\mathfrak R_{M,\gamma}^{\mathrm W}(I)>1$.

(iii) Suppose the conditions in parts (i) and (ii) hold. {On $\mathcal C_n^{\mathrm W}\cap\mathcal E_{n,\mathrm{noise},A}^{\mathrm W}$,} if $(s,e]$ contains an undetected change point $\tau_j$ satisfying at least one of
\[
\frac{\Delta_n\|\bdelta_j\|_2^2}
{\sqrt p\,\Lambda_{S,\gamma,n}^{\mathrm W}}
\rightarrow\infty,
\qquad
\frac{\sqrt{\Delta_n}\|\bdelta_j\|_\infty}
{\Lambda_{M,\gamma,n}^{\mathrm W}}
\rightarrow\infty,
\]
then $\mathcal E_\gamma^{\mathrm W}(s,e)\ne\varnothing$. Every element of $\mathcal E_\gamma^{\mathrm W}(s,e)$ attaining the minimum interval length indexes an interval of length at most $\Delta_n/2+2$ containing exactly one true change point. If $(s,e]$ contains no true change point, then $\mathcal E_\gamma^{\mathrm W}(s,e)=\varnothing$.

\end{lemma}

\begin{proof}
Let $I_{j,n}^{\mathrm{iso}}=(s_j,e_j]$ be the coverage interval of an undetected change point $\tau_j$, and let $m_j=e_j-s_j$. By definition, $\tau_j-s_j,\ e_j-\tau_j\ge\Delta_n/8-1$ and $m_j\le \frac{\Delta_n}{2}+2$. For $\gamma=1/2$, Assumption~\ref{ass:WBS-geometry}(i), $\lambda_{\mathrm W}\in(0,1/4)$, and $m_j\asymp\Delta_n$ imply $\lambda_{m_j}^{\mathrm W}=O(\Delta_n^{\lambda_{\mathrm W}})=o(\Delta_n)$. Hence, for all sufficiently large $n$, $\lambda_{m_j}^{\mathrm W}\le\min\{\tau_j-s_j,e_j-\tau_j\}$, so $\tau_j\in\mathcal K_\gamma(I_{j,n}^{\mathrm{iso}})$ for both $\gamma=0$ and $\gamma=1/2$.

Lemma~\ref{lem:wbs-one-change-interval}(iv) therefore gives
\begin{equation}
\label{eq:S26-coverage-lower}
G_{S,\gamma,I_{j,n}^{\mathrm{iso}}}(\tau_j)
\ge
c\frac{\Delta_n\|\bdelta_j\|_2^2}{\sqrt p},
\qquad
\widehat G_{M,\gamma,I_{j,n}^{\mathrm{iso}}}(\tau_j)
\ge
c\sqrt{\Delta_n}\|\bdelta_j\|_\infty,
\end{equation}
where the second inequality holds on
${\mathcal E_{\mathrm{feas},M,n}}$.

 (i) By Assumption~\ref{ass:C2} and Lemma~\ref{lem:wbs-multiple-nuisance}, on ${\mathcal E_{\mathrm{feas},S,n}}$, {there exist constants $c_\omega<C_\omega$ such that} $0<c_\omega\le\widehat\omega_{\mathrm W}\le C_\omega<\infty$ for all sufficiently large $n$. Since
$\tau_j\in\mathcal K_\gamma(I_{j,n}^{\mathrm{iso}})$, Lemma \ref{lem:wbs-one-change-interval}(ii) yields
\begingroup
\[
\mathcal S_\gamma(I_{j,n}^{\mathrm{iso}})
\ge
cG_{S,\gamma,I_{j,n}^{\mathrm{iso}}}(\tau_j)
-CB_{S,\gamma,n}^{\mathrm W}.
\]
\endgroup
Combining Equation~\eqref{eq:S26-coverage-lower}, {we have}
\[
\mathfrak R_{S,\gamma}^{\mathrm W}
(I_{j,n}^{\mathrm{iso}})
\ge
c\frac{\Delta_n\|\bdelta_j\|_2^2}
{\sqrt p\,\Lambda_{S,\gamma,n}^{\mathrm W}}
-C\frac{B_{S,\gamma,n}^{\mathrm W}}
{\Lambda_{S,\gamma,n}^{\mathrm W}}
\rightarrow\infty,
\]
i.e., $I_{j,n}^{\mathrm{iso}}$ is significant.

If a sampled interval $I$ contains no true change point, then
$\mathbf m_I(b)=\mathbf0$ for every
$b\in\mathcal K_\gamma(I)$. By Lemma~\ref{lem:wbs-multiple-nuisance}(i), Lemma~\ref{lem:wbs-uniform-noise} and Equation~\eqref{eq:wbs-dense-centering-tail}, {we have} on ${\mathcal E_{n,\mathrm{noise},S}^{\mathrm W}}$, $\mathcal S_\gamma(I)\le CB_{S,\gamma,n}^{\mathrm W}$ and hence $\mathfrak R_{S,\gamma}^{\mathrm W}(I)\le CB_{S,\gamma,n}^{\mathrm W}/\Lambda_{S,\gamma,n}^{\mathrm W}=o(1)$. Therefore no interval containing no true change point is significant.
 
\begingroup
 Since the coverage interval is significant, every shortest significant interval has length at most $\Delta_n/2+2$. By Assumption~\ref{ass:WBS-geometry}(i), $\Delta_n\to\infty$, so $\Delta_n/2+2<\Delta_n$ for all sufficiently large $n$. Such an interval contains at most one true change point, whereas the preceding no-change bound shows that it contains at least one. This proves part (i).
\endgroup

(ii) By Lemma~\ref{lem:wbs-multiple-nuisance} and Lemma~\ref{lem:wbs-uniform-noise}, on ${\mathcal E_{n,\mathrm{noise},M}^{\mathrm W}}$,
\[
\mathcal M_\gamma(I_{j,n}^{\mathrm{iso}})
\ge
\widehat G_{M,\gamma,I_{j,n}^{\mathrm{iso}}}(\tau_j)
-C\{\log(np)\}^{c_{\mathrm{dep}}}.
\]
It follows from Equation~\eqref{eq:S26-coverage-lower} that
\[
\mathfrak R_{M,\gamma}^{\mathrm W}
(I_{j,n}^{\mathrm{iso}})
\ge
c\frac{\sqrt{\Delta_n}\|\bdelta_j\|_\infty}
{\Lambda_{M,\gamma,n}^{\mathrm W}}
-C\frac{\{\log(np)\}^{c_{\mathrm{dep}}}}
{\Lambda_{M,\gamma,n}^{\mathrm W}}
\rightarrow\infty.
\]
i.e., $I_{j,n}^{\mathrm{iso}}$ is significant.

If $I$ contains no true change point, then $m_{I,l}(b)=0$ for every
$b\in I$ and $l=1,\ldots,p$. By %
 Lemma~\ref{lem:wbs-multiple-nuisance}(ii) and  Lemma~\ref{lem:wbs-uniform-noise}, on ${\mathcal E_{n,\mathrm{noise},M}^{\mathrm W}}$, $\mathcal M_\gamma(I)\le C\{\log(np)\}^{c_{\mathrm{dep}}}$, and therefore $\mathfrak R_{M,\gamma}^{\mathrm W}(I)\le C\{\log(np)\}^{c_{\mathrm{dep}}}/\Lambda_{M,\gamma,n}^{\mathrm W}=o(1)$. The same shortest-interval argument as in part (i) proves part (ii).

(iii) If
\[
\frac{\Delta_n\|\boldsymbol\delta_j\|_2^2}{\sqrt p\,\Lambda_{S,\gamma,n}^{\mathrm W}}\rightarrow\infty,
\]
then part~(i) gives $\mathfrak R_{S,\gamma}^{\mathrm W}(I_{j,n}^{\mathrm{iso}})>1$. If instead
\[
\frac{
\sqrt{\Delta_n}\|\boldsymbol\delta_j\|_\infty
}{
\Lambda_{M,\gamma,n}^{\mathrm W}
}
\rightarrow\infty,
\]
then part~(ii) gives $\mathfrak R_{M,\gamma}^{\mathrm W}(I_{j,n}^{\mathrm{iso}})>1$. Hence, in either case, $\mathcal E_\gamma^{\mathrm W}(s,e)\ne\varnothing$.

Since $|I_{j,n}^{\mathrm{iso}}|\le\Delta_n/2+2$, every pair in $\mathcal E_\gamma^{\mathrm W}(s,e)$ attaining the minimum interval length indexes an interval $I^\star$ satisfying $|I^\star|\le\Delta_n/2+2<\Delta_n$ for all sufficiently large $n$. Therefore, by the definition of $\Delta_n$, the interval $I^\star$ contains at most one true change point. Parts~(i) and~(ii) also show that, for every sampled interval $I$ containing no true change point, $\mathfrak R_{S,\gamma}^{\mathrm W}(I)=o(1)$ and $\mathfrak R_{M,\gamma}^{\mathrm W}(I)=o(1)$. Thus an interval containing no change point cannot be indexed by $\mathcal E_\gamma^{\mathrm W}(s,e)$ for all sufficiently large $n$, and consequently $I^\star$ contains exactly one true change point.

Finally, if $(s,e]$ contains no true change point, then every
sampled interval $I\subseteq(s,e]$ contains no true change point. Hence, by parts~(i) and~(ii), $\mathfrak R_{S,\gamma}^{\mathrm W}(I)<1$ and $\mathfrak R_{M,\gamma}^{\mathrm W}(I)<1$ for all sufficiently large $n$. Therefore, $\mathcal E_\gamma^{\mathrm W}(s,e)=\varnothing$.
\end{proof}

\begin{proof}[Proof of Theorem~\ref{thm:wbs-dense}]
Recall that 
\[
\begin{aligned}
r_{S,\gamma,n}^{\mathrm{pre,W}}
:={}&
\ind{\gamma=1/2}\Delta_n^{\lambda_{\mathrm W}}+
\Delta_n\left\{
\frac{B_{S,\gamma,n}^{\mathrm W}}
     {\Lambda_{S,\gamma,n}^{\mathrm W}}
+
\frac{\{\log(np)\}^{c_{\mathrm{dep}}}}
     {p^{1/4}\{\Lambda_{S,\gamma,n}^{\mathrm W}\}^{1/2}}
\right\}.
\end{aligned}
\]

{
Let $\mathcal E_{n,\mathrm{all},S}^{\mathrm W}
=\mathcal C_n^{\mathrm W}\cap\mathcal E_{n,\mathrm{noise},S}^{\mathrm W}$. By Lemmas~\ref{lem:wbs-coverage-main} and~\ref{lem:wbs-uniform-noise},
\[
\begin{aligned}
\Pr\{(\mathcal E_{n,\mathrm{all},S}^{\mathrm W})^c\}
\le{}&
K\{1-c_{\mathrm W}(\Delta_n/n)^2\}^{N_{\mathrm W,n}}
+\frac{C}{\log(2n)}
+\frac{C}{nN_{\mathrm W,n}}
+\varepsilon_{\mathrm{feas},S,n}
\rightarrow0.
\end{aligned}
\]
We work on $\mathcal E_{n,\mathrm{all},S}^{\mathrm W}$ throughout the proof.
}

For every $j=1,\ldots,K$, the signal condition in the theorem gives
\[
\frac{\Delta_n\|\bdelta_j\|_2^2}
     {\sqrt p\,\Lambda_{S,\gamma,n}^{\mathrm W}}\rightarrow\infty.
\]

Consider a segment $(s,e]$ arising in the recursion such that the
coverage interval of every undetected change point in $(s,e]$ is
fully contained in $(s,e]$. This property holds for the initial
segment $(0,n]$ on $\mathcal C_n^{\mathrm W}$.

Suppose that $(s,e]$ contains an undetected change point. By
Lemma~\ref{lem:wbs-shortest-one-change}(i), there exists
$r\in\mathcal I_{\mathrm W}(s,e)$ such that $\mathcal K_\gamma(I_r^{\mathrm W})\ne\varnothing$ and $\mathfrak R_{S,\gamma}^{\mathrm W}(I_r^{\mathrm W})>1$. Let $I^\star=(u,v]$ be a shortest sampled interval among those satisfying the above conditions. By Lemma~\ref{lem:wbs-shortest-one-change}(i), $m^\star:=v-u\le\Delta_n/2+2$ and $I^\star$ contains exactly one true change point, denoted by $\tau_j$. Let $\widehat\tau=\widehat\tau_{S,\gamma}(I^\star)$ be the split selected on $I^\star$. Since $\mathfrak R_{S,\gamma}^{\mathrm W}(I^\star)>1$,
Lemma~\ref{lem:wbs-one-change-interval}(ii) gives
\[
\begin{aligned}
|\widehat\tau-\tau_j|
\le C\Biggl[
&\ind{\gamma=1/2}\lambda_{m^\star}^{\mathrm W}+
m^\star\left\{
\frac{B_{S,\gamma,n}^{\mathrm W}}
     {\Lambda_{S,\gamma,n}^{\mathrm W}}
+
\frac{\{\log(np)\}^{c_{\mathrm{dep}}}}
     {p^{1/4}\{\Lambda_{S,\gamma,n}^{\mathrm W}\}^{1/2}}
\right\}
\Biggr].
\end{aligned}
\]
Since $m^\star\le\frac{\Delta_n}{2}+2\le C\Delta_n$ and $\lambda_{m^\star}^{\mathrm W}=\left\lceil(m^\star)^{\lambda_{\mathrm W}}\right\rceil\le C\Delta_n^{\lambda_{\mathrm W}}$, we obtain $|\widehat\tau-\tau_j|
\le
Cr_{S,\gamma,n}^{\mathrm{pre,W}}$, which implies for all
sufficiently large $n$, $\tau_j\in(\widehat\tau-d_{\mathrm W,n},\widehat\tau+d_{\mathrm W,n}]$, since
$r_{S,\gamma,n}^{\mathrm{pre,W}}=o(d_{\mathrm W,n})$. Thus the guard removes $\tau_j$. Moreover, for every $k\ne j$,
\[
\begin{aligned}
|\tau_k-\widehat\tau|\ge|\tau_k-\tau_j|-|\widehat\tau-\tau_j|\ge\Delta_n-Cr_{S,\gamma,n}^{\mathrm{pre,W}}>d_{\mathrm W,n}
\end{aligned}
\]
for all sufficiently large $n$, since
$r_{S,\gamma,n}^{\mathrm{pre,W}}=o(d_{\mathrm W,n})$ and $d_{\mathrm W,n}=o(\Delta_n)$. Hence no other true change point is removed.

We next verify that the coverage intervals of the remaining changes are retained in the appropriate child segments. Let $\tau_k>\tau_j$
be an undetected change point, and let $(s_{k,n}^{\mathrm{iso}},e_{k,n}^{\mathrm{iso}}]$ be its coverage interval. By Lemma~\ref{lem:wbs-coverage-main}, $s_{k,n}^{\mathrm{iso}}\ge\tau_k-\Delta_n/4-1$.
Therefore,
\[
\begin{aligned}
s_{k,n}^{\mathrm{iso}}
-(\widehat\tau+d_{\mathrm W,n})
&\ge
(\tau_k-\tau_j)
-\frac{\Delta_n}{4}-1
-|\widehat\tau-\tau_j|
-d_{\mathrm W,n}\\
&\ge
\frac{3\Delta_n}{4}-1
-Cr_{S,\gamma,n}^{\mathrm{pre,W}}
-d_{\mathrm W,n}>0
\end{aligned}
\]
for all sufficiently large $n$. Thus the entire coverage interval of
$\tau_k$ is contained in the right child segment. If $\tau_k<\tau_j$, Lemma~\ref{lem:wbs-coverage-main} gives $e_{k,n}^{\mathrm{iso}}\le\tau_k+\frac{\Delta_n}{4}+1$. Similarly,
\[
\begin{aligned}
(\widehat\tau-d_{\mathrm W,n})
-e_{k,n}^{\mathrm{iso}}
&\ge
(\tau_j-\tau_k)
-\frac{\Delta_n}{4}-1
-|\widehat\tau-\tau_j|
-d_{\mathrm W,n}\\
&\ge
\frac{3\Delta_n}{4}-1
-Cr_{S,\gamma,n}^{\mathrm{pre,W}}
-d_{\mathrm W,n}>0.
\end{aligned}
\]
Hence the entire coverage interval of $\tau_k$ is contained in the left child segment.

Starting from $(0,n]$, the preceding argument applies whenever a current segment contains an undetected change point: a shortest sampled interval satisfying $\mathfrak R_{S,\gamma}^{\mathrm W}(I)>1$ contains exactly one such change, the resulting split is within $Cr_{S,\gamma,n}^{\mathrm{pre,W}}$ of that change, and the coverage intervals of all remaining changes are retained in the appropriate
child segments.

If a current segment contains no true change point, Lemma~\ref{lem:wbs-shortest-one-change}(i) shows that no sampled interval in that segment satisfies $\mathfrak R_{S,\gamma}^{\mathrm W}(I)>1$, so the recursion stops on that segment. Consequently, exactly one true change point is removed at each split, the recursion cannot stop while an undetected change remains, and no split is made on a segment containing no true change point. Hence the recursion returns exactly $K$ estimates.

For $j=1,\ldots,K$, denote by $\widehat\tau^{[j]}$ the estimate associated with the selected interval containing $\tau_j$. On ${\mathcal E_{n,\mathrm{all},S}^{\mathrm W}}$,
\[
\widehat K_{S,\gamma}=K,
\qquad
\max_{1\le j\le K}
|\widehat\tau^{[j]}-\tau_j|
\le
Cr_{S,\gamma,n}^{\mathrm{pre,W}}.
\]

Finally,
$r_{S,\gamma,n}^{\mathrm{pre,W}}=o(d_{\mathrm W,n})$ and
$d_{\mathrm W,n}=o(\Delta_n)$ imply
$r_{S,\gamma,n}^{\mathrm{pre,W}}=o(\Delta_n)$. Thus, for
$1\le j<K$,
\[
\begin{aligned}
\widehat\tau^{[j+1]}-\widehat\tau^{[j]}
&\ge(\tau_{j+1}-\tau_j)-2Cr_{S,\gamma,n}^{\mathrm{pre,W}}\ge\Delta_n-2Cr_{S,\gamma,n}^{\mathrm{pre,W}}>0,
\end{aligned}
\]
for all sufficiently large $n$. Hence the estimates associated with
$\tau_1,\ldots,\tau_K$ are already in increasing order and therefore
coincide with the ordered output $\widehat\tau_{S,\gamma,1},\ldots,\widehat\tau_{S,\gamma,K}$. Since
${\Pr\{(\mathcal E_{n,\mathrm{all},S}^{\mathrm W})^c\}\to0}$, we conclude
that
\[
\Pr(\widehat K_{S,\gamma}=K)\rightarrow1,\qquad \max_{1\le j\le K}
|\widehat\tau_{S,\gamma,j}-\tau_j|
=
O_p\!\left(r_{S,\gamma,n}^{\mathrm{pre,W}}\right).
\]
 
 $r_{S,\gamma,n}^{\mathrm{pre,W}}=o(\Delta_n)$ also yields
\[
\frac{1}{\Delta_n}
\max_{1\le j\le K}
|\widehat\tau_{S,\gamma,j}-\tau_j|
\xrightarrow{p}0.
\]
\end{proof}

\begin{proof}[Proof of Theorem~\ref{thm:wbs-max}]
\begingroup
It suffices to verify the significance and localization inputs to the
recursion already proved for Theorem~\ref{thm:wbs-dense}. Define
\[
\mathcal E_{n,\mathrm{all},M}^{\mathrm W}
=\mathcal C_n^{\mathrm W}\cap\mathcal E_{n,\mathrm{noise},M}^{\mathrm W}.
\]
Lemmas~\ref{lem:wbs-coverage-main} and~\ref{lem:wbs-uniform-noise} give
\[
\Pr\{(\mathcal E_{n,\mathrm{all},M}^{\mathrm W})^c\}
\le K\{1-c_{\mathrm W}(\Delta_n/n)^2\}^{N_{\mathrm W,n}}
+\frac C{npN_{\mathrm W,n}}+\varepsilon_{\mathrm{feas},M,n}\to0.
\]
On this event, Lemma~\ref{lem:wbs-shortest-one-change}(ii) makes every
coverage interval significant and every no-change interval insignificant.
The shortest significant interval $I^\star$ therefore contains exactly
one change $\tau_j$ and has $m^\star\le\Delta_n/2+2$.
Lemma~\ref{lem:wbs-one-change-interval}(iii) yields
\[
\begin{aligned}
|\widehat\tau-\tau_j|
&\le C\left\{\ind{\gamma=1/2}\lambda_{m^\star}^{\mathrm W}
+m^\star\frac{\{\log(np)\}^{c_{\mathrm{dep}}}}
{\Lambda_{M,\gamma,n}^{\mathrm W}}\right\}\\
&\le Cr_{M,\gamma,n}^{\mathrm{pre,W}}=o(d_{\mathrm W,n}).
\end{aligned}
\]
To check that the previous recursion applies, let
$r_n=r_{M,\gamma,n}^{\mathrm{pre,W}}$. Its guard and coverage inequalities
become, for every other undetected change $\tau_k$,
\[
\begin{aligned}
|\widehat\tau-\tau_j|&\le Cr_n<d_{\mathrm W,n},&
|\tau_k-\widehat\tau|&\ge\Delta_n-Cr_n>d_{\mathrm W,n},\\
s_{k,n}^{\mathrm{iso}}-(\widehat\tau+d_{\mathrm W,n})
&\ge3\Delta_n/4-1-Cr_n-d_{\mathrm W,n}>0
&& (\tau_k>\tau_j),\\
(\widehat\tau-d_{\mathrm W,n})-e_{k,n}^{\mathrm{iso}}
&\ge3\Delta_n/4-1-Cr_n-d_{\mathrm W,n}>0
&& (\tau_k<\tau_j).
\end{aligned}
\]
Thus each split removes exactly one change and preserves every remaining
coverage interval. The recursion stops precisely on no-change segments.
It follows that $\widehat K_{M,\gamma}=K$ and
$\max_j|\widehat\tau^{[j]}-\tau_j|\le Cr_n$ on this event, where
$\widehat\tau^{[j]}$ is the estimate associated with $\tau_j$.
Since $r_n=o(\Delta_n)$,
\[
\widehat\tau^{[j+1]}-\widehat\tau^{[j]}
\ge\Delta_n-2Cr_n>0,
\]
so the association agrees with the ordered output. Consequently
\[
\Pr(\widehat K_{M,\gamma}=K)\to1,\qquad
\max_{j\le K}|\widehat\tau_{M,\gamma,j}-\tau_j|=O_p(r_n),\qquad
\Delta_n^{-1}\max_{j\le K}|\widehat\tau_{M,\gamma,j}-\tau_j|
\xrightarrow{p}0.
\]
\endgroup
\end{proof}

\begingroup
\begin{proof}[Proof of Corollary~\ref{cor:wbs-adaptive}]
We verify the two inputs to the recursion in Theorem~\ref{thm:wbs-dense}:
a significant one-change coverage interval exists, and either selected
component localizes within the guard. Work on
\[
\mathcal E_{n,\mathrm{all},A}^{\mathrm W}
=\mathcal C_n^{\mathrm W}\cap\mathcal E_{n,\mathrm{noise},A}^{\mathrm W},
\]
whose failure probability is at most
\[
K\{1-c_{\mathrm W}(\Delta_n/n)^2\}^{N_{\mathrm W,n}}
+\frac C{\log(2n)}+\frac C{nN_{\mathrm W,n}}+\frac C{npN_{\mathrm W,n}}
+\varepsilon_{\mathrm{feas},S,n}+\varepsilon_{\mathrm{feas},M,n}=o(1).
\]
For each change, the assumed dense or coordinatewise signal condition
and Lemma~\ref{lem:wbs-shortest-one-change}(iii) make its coverage
interval significant. In a recursive segment containing an undetected
change, let $(\mathfrak c^\star,r^\star)$ be the selected shortest
significant pair, let $I^\star=I_{r^\star}^{\mathrm W}$ and
$m^\star=|I^\star|$. The same lemma gives
\[
m^\star\le\Delta_n/2+2,\qquad
I^\star\text{ contains exactly one change }\tau_j,\qquad
\mathfrak R_{\mathfrak c^\star,\gamma}^{\mathrm W}(I^\star)>1.
\]
Let $L_n=\log(np)$ and
\[
\begin{aligned}
b_{S,n}&=\frac{B_{S,\gamma,n}^{\mathrm W}}
{\Lambda_{S,\gamma,n}^{\mathrm W}}
+\frac{L_n^{c_{\mathrm{dep}}}}
{p^{1/4}(\Lambda_{S,\gamma,n}^{\mathrm W})^{1/2}},&
b_{M,n}&=\frac{L_n^{c_{\mathrm{dep}}}}{\Lambda_{M,\gamma,n}^{\mathrm W}}.
\end{aligned}
\]
Parts (ii) and (iii) of Lemma~\ref{lem:wbs-one-change-interval}, applied
to the actual selected component, imply
\[
\begin{aligned}
|\widehat\tau-\tau_j|
&\le C\{\ind{\gamma=1/2}\lambda_{m^\star}^{\mathrm W}
             +m^\star b_{\mathfrak c^\star,n}\}\\
&\le C\{\ind{\gamma=1/2}\Delta_n^{\lambda_{\mathrm W}}
             +\Delta_n(b_{S,n}+b_{M,n})\}
=Cr_{A,\gamma,n}^{\mathrm{pre,W}}.
\end{aligned}
\]
This bound holds for both choices on the same event, so no independence
between the selected component and the observations is required.
Set $r_n=r_{A,\gamma,n}^{\mathrm{pre,W}}$. The guard and coverage
inequalities established in the dense proof apply with this $r_n$,
because
\[
Cr_n<d_{\mathrm W,n},\qquad
\Delta_n-Cr_n>d_{\mathrm W,n},\qquad
3\Delta_n/4-1-Cr_n-d_{\mathrm W,n}>0.
\]
They remove only $\tau_j$ and retain all other coverage intervals in the
appropriate children. Lemma~\ref{lem:wbs-shortest-one-change}(iii) also
makes the adaptive significant set empty on every no-change segment.
Induction therefore gives exactly $K$ estimates and
$\max_j|\widehat\tau^{[j]}-\tau_j|\le Cr_n$. Finally,
\[
\widehat\tau^{[j+1]}-\widehat\tau^{[j]}
\ge\Delta_n-2Cr_n>0,
\]
so association with the true changes agrees with the ordered output.
Since the failure probability tends to zero and $r_n=o(\Delta_n)$,
\[
\Pr(\widehat K_\gamma=K)\to1,\qquad
\max_{j\le K}|\widehat\tau_{\gamma,j}-\tau_j|=O_p(r_n),\qquad
\Delta_n^{-1}\max_{j\le K}|\widehat\tau_{\gamma,j}-\tau_j|
\xrightarrow{p}0.
\]
\end{proof}
\endgroup

\section{Uniform componentwise long-run variance estimation}\label{subsec:supp-lrv-estimation}

\begingroup
For each coordinate $j\in\{1,\dots,p\}$, let $c_{\epsilon,j}(r)=\Cov(\epsilon_{0j},\epsilon_{rj})=\mathbf\Gamma(r)_{jj}$, $\sigma_j^2=\sum_{r\in\mathbb Z}c_{\epsilon,j}(r)=\mathbf\Omega_{jj}$. Recall that $\widetilde\sigma_j^2=\sum_{|r|<\ell_j}K_{\mathrm{LRV}}(r/\ell_j)\widehat c_{r,j}^{D}.
$ where $\widehat c_{r,j}^{D}
=n^{-1}\sum_{i=6\ell_j+|r|+1}^n
D_{i,j}^{(3)}D_{i-|r|,j}^{(3)}$, $D_{i,j}^{(3)}=\sum_{a=0}^3d_a^{(3)}X_{i-2a\ell_j,j}$ for $ i=6\ell_j+1,\ldots,n,$ and $\mathbf d^{(3)}=(d_0^{(3)},d_1^{(3)},d_2^{(3)},d_3^{(3)})$ with $\sum_{a=0}^3d_a^{(3)}=0$ and $\sum_{a=0}^3\{d_a^{(3)}\}^2=1$.
\endgroup

\begin{lemma}%
\label{lem:lrv-general}
Suppose Assumptions~\ref{ass:C1}--\ref{ass:lrv-M} hold. Under model (1), %
\begingroup
take $c_{\mathrm{LRV}}=3/2+1/\varkappa$ under (AM) and
$c_{\mathrm{LRV}}=3/2$ under (MD) or (PD). These choices satisfy
$c_{\mathrm{LRV}}\ge c_{\mathrm{dep}}$, where $c_{\mathrm{dep}}$ is
defined in Equation~\eqref{eq:branchwise-cdep-choice}. Let
\endgroup
\[
a_{\sigma,n}=
\frac{\ell_{\max}\{\log(np)\}^{c_{\mathrm{LRV}}}}{\sqrt n}
+\ell_{\min}^{-\widetilde q}
+\eta_{\ell_{\min}}
+\frac{\ell_{\max}^2}{n},
\]
and
\[
a_{\sigma,n}^{(\delta)}
=
a_{\sigma,n}
+
\frac{\ell_{\max}^{3/2}\|\bdelta\|_\infty
\{\log(np)\}^{c_{\mathrm{LRV}}}}{n}
+
\frac{\ell_{\max}^2\|\bdelta\|_\infty^2}{n}.
\]
\begingroup
Then $\max_{j\le p}|\widetilde\sigma_j^2-\sigma_j^2|
=O_p\{a_{\sigma,n}^{(\delta)}\}$ and \(
r_{\sigma,n}
=
\max_{j\le p}
\left|\widehat\sigma_j/\sigma_j-1
\right|
=
O_p\{a_{\sigma,n}^{(\delta)}\}.
\)
\endgroup
\end{lemma}

\begin{proof}
\begingroup
We decompose the estimation error into centered noise, deterministic bias,
and signal contamination. The noise term follows from a lag-uniform
product moment bound; the signal terms use deterministic coefficient norms.
Let
$\mathcal L_n=\{\ell\in\mathbb N:
c_-n^{c_{\ell^-}}\le\ell\le C_+n^{c_{\ell^+}}\}$,
where $c_-,C_+>0$ are fixed envelope constants in
Assumption~\ref{ass:lrv-M}(ii), and take deterministic bandwidths
$\ell_j\in\mathcal L_n$, $j\le p$.
All probability bounds below hold uniformly over $j\le p$ and
$\ell\in\mathcal L_n$, with $|\mathcal L_n|\le n$ for all sufficiently large $n$.

\endgroup
Under model (1),
\(
D_{i,j}^{(3)}
=
D_{i,j}^{(3),\epsilon}
+m_{i,j}^{(3)},
\)
where
\[
D_{i,j}^{(3),\epsilon}
=
\sum_{a=0}^3d_a^{(3)}\epsilon_{i-2a\ell_j,j},
\qquad
m_{i,j}^{(3)}
=
\E D_{i,j}^{(3)}.
\]
\begingroup
Let $\widetilde\sigma_{j,\epsilon}^2$ and $\widehat c_{r,j}^{D,\epsilon}$
be defined in the same way as $\widetilde\sigma_j^2$ and
\endgroup
$\widehat c_{r,j}^{D}$, respectively, with every $D_{i,j}^{(3)}$
replaced by $D_{i,j}^{(3),\epsilon}$.
Then
\begingroup
\begin{equation}
\label{eq:lrv-error-decomposition}
\max_{j\le p}|\widetilde\sigma_j^2-\sigma_j^2|
\le
\max_{j\le p}|\widetilde\sigma_{j,\epsilon}^2-\sigma_j^2|
+
\max_{j\le p}|\widetilde\sigma_j^2-\widetilde\sigma_{j,\epsilon}^2|.
\end{equation}
\endgroup

Step 1. We first consider the first term in \eqref{eq:lrv-error-decomposition}. Define
\[
\mathfrak z_{i,r,j}^{(D)}
=D_{i,j}^{(3),\epsilon}D_{i-|r|,j}^{(3),\epsilon}
-\E\{D_{i,j}^{(3),\epsilon}D_{i-|r|,j}^{(3),\epsilon}\},
\qquad
\mathfrak S_{r,j}^{(D)}
=\sum_{i=6\ell_j+|r|+1}^n\mathfrak z_{i,r,j}^{(D)}.
\]
Then
\begin{equation}
\label{eq:lrv-covariance-sum-reduction}
\widehat c_{r,j}^{D,\epsilon}
-
\E\widehat c_{r,j}^{D,\epsilon}
=
n^{-1}\mathfrak S_{r,j}^{(D)}.
\end{equation}
We first prove the auxiliary moment bound
\begin{equation}
\label{eq:lrv-S-moment-target}
\|\mathfrak S_{r,j}^{(D)}\|_q
\le Cq^{c_{\mathrm{LRV}}}\sqrt n,
\qquad q\ge2.
\end{equation}
\begingroup
To control the stochastic part of the long-run variance estimator, it is
enough to bound a centered sum of lagged products uniformly in its lag.
Fix a coordinate and let $X_i=\epsilon_{ij}$. For an integer $h\geq0$, let
\[
Y_i^{(h)}=X_iX_{i-h}-\E(X_iX_{i-h}).
\]
For every finitely supported deterministic real sequence $(w_i)$ and every
$q\geq2$, we will prove
\begin{equation}\label{eq:lrv-lag-product-moment}
\left\|\sum_iw_iY_i^{(h)}\right\|_q
\leq Cq^{c_{\mathrm{LRV}}}\left(\sum_iw_i^2\right)^{1/2},
\qquad
c_{\mathrm{LRV}}=
\begin{cases}
3/2+1/\varkappa,&\text{under (AM)},\\
3/2,&\text{under (MD) or (PD)}.
\end{cases}
\end{equation}
All constants below are uniform in the coordinate, the lag, and the sample
size. The sub-Gaussian marginal assumption gives
$\|X_i\|_u\leq C\sqrt{u}$ for $u\geq2$, and consequently
$\|Y_i^{(h)}\|_u\leq Cu$.

Under (AM), let $\mathcal F_k=\sigma(\bepsilon_t:t\leq k)$ and
$P_kZ=\E(Z\mid\mathcal F_k)-\E(Z\mid\mathcal F_{k-1})$.
The conditional version of Davydov's inequality, obtained by duality
against an $\mathcal F_k$-measurable variable of conjugate norm one,
gives, for a centered future variable $Z$ separated from $\mathcal F_k$
by $d$ time units,
\[
\|\E(Z\mid\mathcal F_k)\|_u
\leq C\|Z\|_{2u}\alpha_{\epsilon,p}(d)^{1/(2u)}.
\]
For $d=0$ the same expression is replaced by the contraction bound
$\|\E(Z\mid\mathcal F_k)\|_u\leq\|Z\|_u$.
If $0\leq r\leq h$ and $k=i-r$, then $X_{i-h}$ is
$\mathcal F_k$-measurable. The tower property, followed by conditional
contraction and H\"older's inequality, gives
\begin{align*}
P_k(X_iX_{i-h})
&=P_k\{X_{i-h}\E(X_i\mid\mathcal F_k)\},\\
\|P_{i-r}Y_i^{(h)}\|_q
&\leq2\|X_{i-h}\|_{2q}
       \|\E(X_i\mid\mathcal F_{i-r})\|_{2q}
\leq Cq\exp(-cr^\varkappa/q).
\end{align*}
This reasoning also covers $r=h$: measurability with respect to
$\mathcal F_{k-1}$ is not required in the first identity.
For $r>h$, apply the conditional bound directly to the centered variable
$Y_i^{(h)}$, whose earliest observation occurs at $i-h$. Thus
\[
\|P_{i-r}Y_i^{(h)}\|_q
\leq Cq\exp\{-c(r-h)^\varkappa/q\}.
\]
Summing these two bounds proves
\[
\sum_{r=0}^{\infty}\|P_{i-r}Y_i^{(h)}\|_q
\leq Cq\sum_{r=0}^{\infty}\exp(-cr^\varkappa/q)
\leq Cq^{1+1/\varkappa}.
\]
Moreover, the same conditional bound implies
$\E(Y_i^{(h)}\mid\mathcal F_{i-R})\to0$ in $L_q$ as $R\to\infty$.
Consequently $Y_i^{(h)}=\sum_{r\geq0}P_{i-r}Y_i^{(h)}$ in $L_q$.
For each fixed $r$, the summands indexed by $i$ form a martingale
difference sequence. For martingale differences $Z_i$ and $q\ge2$, the martingale moment inequality is
$\|\sum_i Z_i\|_q\le\sqrt{q-1}(\sum_i\|Z_i\|_q^2)^{1/2}$.
Applying this inequality to $Z_i=w_iP_{i-r}Y_i^{(h)}$ and then
Minkowski's inequality gives
\begin{align*}
\left\|\sum_iw_iY_i^{(h)}\right\|_q
&\leq\sum_{r\geq0}
\left\|\sum_iw_iP_{i-r}Y_i^{(h)}\right\|_q\\
&\leq\sqrt{q-1}\sum_{r\geq0}
\left\{\sum_iw_i^2\|P_{i-r}Y_i^{(h)}\|_q^2\right\}^{1/2}
\leq Cq^{3/2+1/\varkappa}\left(\sum_iw_i^2\right)^{1/2}.
\end{align*}
Under (MD), the same proof applies with the projection bound zero for
$r>m_0$ in the range $r\leq h$, and zero for $r-h>m_0$ in the range
$r>h$. At most $2m_0+2$ terms remain. {This proves \eqref{eq:lrv-lag-product-moment} with exponent $3/2$.}
\endgroup

\begingroup
Under (PD), use the innovation filtration and let $X_i^{\{r\}}$ denote
the variable obtained by replacing the innovation at time $i-r$ by an
independent copy. Apply the same replacement to both factors of the
product defining $Y_i^{(h)}$; denote the result by
$(Y_i^{(h)})^{\{r\}}$. The two products have the same expectation.
H\"older's inequality and the physical dependence bound imply
\begin{align*}
\|Y_i^{(h)}-(Y_i^{(h)})^{\{r\}}\|_q
&\leq Cq\left\{\exp(-cr^\varkappa)
+\mathbf1\{r\geq h\}\exp[-c(r-h)^\varkappa]\right\}.
\end{align*}
The projection--coupling inequality bounds
$\|P_{i-r}Y_i^{(h)}\|_q$ by the left-hand side. Its sum over
$r\geq0$ is at most $Cq$, independently of $h$.
The innovation reverse tail is trivial, so the same projection
decomposition and martingale moment inequality prove
\eqref{eq:lrv-lag-product-moment} with exponent $3/2$.

To apply \eqref{eq:lrv-lag-product-moment}, let
$u_a=2a\ell_j$, $v_b=|r|+2b\ell_j$,
$h_{ab}=|v_b-u_a|$ and $c_{ab}=\min(u_a,v_b)$. Then
\[
\begin{aligned}
\mathfrak S_{r,j}^{(D)}
&=\sum_{a,b=0}^3d_a^{(3)}d_b^{(3)}
\sum_{i=6\ell_j+|r|+1}^{n}Y_{i-c_{ab}}^{(h_{ab})},
\\
\|\mathfrak S_{r,j}^{(D)}\|_q
&\le C\sqrt n\,q^{c_{\mathrm{LRV}}}
\sum_{a,b=0}^3|d_a^{(3)}d_b^{(3)}|
\le C\sqrt n\,q^{c_{\mathrm{LRV}}}.
\end{aligned}
\]
This proves \eqref{eq:lrv-S-moment-target}, uniformly in the coordinate,
lag and candidate bandwidth.

Finally,
\endgroup
\eqref{eq:lrv-covariance-sum-reduction} gives
\[
\left\|
\widehat c_{r,j}^{D,\epsilon}-\E\widehat c_{r,j}^{D,\epsilon}
\right\|_q
\le\frac{Cq^{c_{\mathrm{LRV}}}}{\sqrt n}.
\]

\begingroup
Fix $A>0$, let $L_n=\log(np)$ and
$q_A=2\lceil(A+3)L_n\rceil$. There are at most $2pn^2$ triples
$(j,\ell,r)$ with $\ell\in\mathcal L_n$ and $|r|<\ell$. Thus
\begin{align*}
&\Pr\left\{\max_{j\le p,\,\ell\in\mathcal L_n,\,|r|<\ell}
|\widehat c_{r,j,\ell}^{D,\epsilon}
 -\E\widehat c_{r,j,\ell}^{D,\epsilon}|
>C_A L_n^{c_{\mathrm{LRV}}}/\sqrt n\right\}\\
&\hspace{25mm}\le2pn^2e^{-q_A}\le2(np)^{-A}.
\end{align*}
Since $\sum_{|r|<\ell}|K_{\mathrm{LRV}}(r/\ell)|\le C\ell$, it follows that
\begin{equation}
\Pr\left\{\max_{j\le p,\,\ell\in\mathcal L_n}
\frac{|\widetilde\sigma_{j,\epsilon,\ell}^2
        -\E\widetilde\sigma_{j,\epsilon,\ell}^2|}{\ell}
>C_A L_n^{c_{\mathrm{LRV}}}/\sqrt n\right\}
\le2(np)^{-A}.
\label{eq:lrv-noise-stochastic-bound}
\end{equation}
In particular, taking $\ell=\ell_j$ for each $j$ gives the claimed
stochastic term with $\ell_{\max}=\max_j\ell_j$.
\endgroup
 
Furthermore, we need to consider 
\begingroup
\[
\E\widetilde\sigma_{j,\epsilon}^2-\sigma_j^2
=
\sum_{|r|<\ell_j}K_{\rm LRV}(r/\ell_j)
\{\E\widehat c_{r,j}^{D,\epsilon}-c_{\epsilon,j}(r)\}
+
\sum_{r\in\mathbb Z}
\{K_{\rm LRV}(r/\ell_j)-1\}c_{\epsilon,j}(r).
\]
\endgroup
\begin{align*}
\E\{D_{i,j}^{(3),\epsilon}D_{i-|r|,j}^{(3),\epsilon}\}
&=\sum_{a,b=0}^3d_a^{(3)}d_b^{(3)}
c_{\epsilon,j}\{-|r|+2(a-b)\ell_j\}\\
&=\sum_{a,b=0}^3d_a^{(3)}d_b^{(3)}
c_{\epsilon,j}\{r+2(b-a)\ell_j\}.
\end{align*}
The second equality uses that \(c_{\epsilon,j}(-h)=c_{\epsilon,j}(h)\) and, when
\(r<0\), relabels \(a\) and \(b\).  Since \(\E\widehat c_{r,j}^{D,\epsilon}
=(1-\frac{6\ell_j+|r|}{n})\E\{D_{i,j}^{(3),\epsilon}D_{i-|r|,j}^{(3),\epsilon}\}\). Therefore, uniformly for $|r|<\ell_j$,
\begin{align*}
\E\widehat c_{r,j}^{D,\epsilon}
={}&c_{\epsilon,j}(r)
+\sum_{a\ne b}d_a^{(3)}d_b^{(3)}
c_{\epsilon,j}\{r+2(b-a)\ell_j\}
+O(\ell_j/n),
\end{align*}
with the fact that $\sum_a\{d_a^{(3)}\}^2=1$.  For \(a\ne b\) and
\(|r|<\ell_j\),
\(
|r+2(b-a)\ell_j|\ge 2\ell_j-|r|>\ell_j.
\)
Thus
\begingroup
\begin{align*}
\sum_{|r|<\ell_j}
\left|
\sum_{a\ne b}d_a^{(3)}d_b^{(3)}
c_{\epsilon,j}\{r+2(b-a)\ell_j\}
\right|
&\le
C\sum_{|h|>\ell_j}|c_{\epsilon,j}(h)|
\le C\eta_{\ell_j},\\
\sum_{|r|<\ell_j}O(\ell_j/n)
&=O(\ell_j^2/n).
\end{align*}
\endgroup
Since \(K_{\mathrm{LRV}}\) is bounded,
\[
\left|
\sum_{|r|<\ell_j}K_{\mathrm{LRV}}(r/\ell_j)
\{\E\widehat c_{r,j}^{D,\epsilon}-c_{\epsilon,j}(r)\}
\right|
\le C\eta_{\ell_j}+C\frac{\ell_j^2}{n}.
\]

Extend $K_{\mathrm{LRV}}$ by zero outside $[-1,1]$.  Its expansion at the origin gives, for a fixed $c_0\in(0,1)$, for $|x|\le c_0$,
\(
|K_{\mathrm{LRV}}(x)-1|
\le C|x|^{\widetilde q}.
\)
Lemma~\ref{lem:vector-basic}(i) yields
\begingroup
\begin{align*}
\left|
\sum_{r\in\mathbb Z}
\{K_{\mathrm{LRV}}(r/\ell_j)-1\}c_{\epsilon,j}(r)
\right|
&\le
C\ell_j^{-\widetilde q}
\sum_{|r|\le c_0\ell_j}|r|^{\widetilde q}|c_{\epsilon,j}(r)|
+C\sum_{|r|>c_0\ell_j}|c_{\epsilon,j}(r)|\\
&\le C\ell_j^{-\widetilde q}.
\end{align*}
{The stretched-exponential tail beginning at $c_0\ell_j$ is $O(\ell_j^{-\widetilde q})$.}
\endgroup
Combining the preceding displays gives
\begingroup
\begin{equation}
\label{eq:lrv-noise-bias-bound}
\max_{j\le p}|\E\widetilde\sigma_{j,\epsilon}^2-\sigma_j^2|
\le C\left\{
\ell_{\min}^{-\widetilde q}
+\eta_{\ell_{\min}}
+\frac{\ell_{\max}^2}{n}
\right\}.
\end{equation}
\endgroup

\begingroup

\endgroup
Step 2. We next consider the second term in
\eqref{eq:lrv-error-decomposition}. For every $j\le p$,
\begingroup
\begin{align}
\widetilde\sigma_j^2-\widetilde\sigma_{j,\epsilon}^2
={}&
\sum_{|r|<\ell_j}
K_{\mathrm{LRV}}(r/\ell_j)
\frac1n
\sum_{i=6\ell_j+|r|+1}^n
\Bigl\{
D_{i,j}^{(3),\epsilon}m_{i-|r|,j}^{(3)}
+m_{i,j}^{(3)}D_{i-|r|,j}^{(3),\epsilon}
+m_{i,j}^{(3)}m_{i-|r|,j}^{(3)}
\Bigr\}.
\label{eq:lrv-contamination-decomposition}
\end{align}
\endgroup
Since \(\sum_a d_a^{(3)}=0\), under model (1),
\(
m_{i,j}^{(3)}
=\delta_j\sum_{a=0}^3d_a^{(3)}
\ind{i-2a\ell_j>\tau}.
\)
Since \(\sum_{a=0}^3d_a^{(3)}=0\),
$m_{i,j}^{(3)}=0$, if $i\le\tau$ or $i-6\ell_j>\tau$.
Hence
\(
|\{i:m_{i,j}^{(3)}\ne0\}|\le C\ell_j,
|m_{i,j}^{(3)}|\le C|\delta_j|,
\)
and
\(
\sum_i\{m_{i,j}^{(3)}\}^2\le C\ell_j\delta_j^2.
\)
Since $K_{\mathrm{LRV}}$ is bounded,
\[
\left|
\sum_{|r|<\ell_j}
K_{\mathrm{LRV}}(r/\ell_j)
\frac1n\sum_i
m_{i,j}^{(3)}m_{i-|r|,j}^{(3)}
\right|
\le
C\frac{\ell_j^2\delta_j^2}{n}.
\]

For term $n^{-1}\sum_i m_{i,j}^{(3)}D_{i-|r|,j}^{(3),\epsilon}$, we have
\[
\frac1n\sum_i
m_{i,j}^{(3)}D_{i-|r|,j}^{(3),\epsilon}
=
\frac1n\sum_{a=0}^3d_a^{(3)}
\sum_i
m_{i,j}^{(3)}
\epsilon_{i-|r|-2a\ell_j,j}.
\]
By Lemma~\ref{lem:vector-basic}(ii), uniformly over $a=0,\ldots,3$,
$|r|<\ell_j$, and $q\ge2$,
\[
\left\|
\sum_i m_{i,j}^{(3)}
\epsilon_{i-|r|-2a\ell_j,j}
\right\|_q
\le
Cq^{c_{\mathrm{dep}}}
\left\{\sum_i\{m_{i,j}^{(3)}\}^2\right\}^{1/2}
\le
Cq^{c_{\mathrm{dep}}}\sqrt{\ell_j}|\delta_j|.
\]
Minkowski's inequality then gives
\[
\left\|
\frac1n\sum_i
m_{i,j}^{(3)}D_{i-|r|,j}^{(3),\epsilon}
\right\|_q
\le
\frac1n
\sum_{a=0}^3 |d_a^{(3)}|
\left\|
\sum_i m_{i,j}^{(3)}
\epsilon_{i-|r|-2a\ell_j,j}
\right\|_q 
\le 
Cq^{c_{\mathrm{dep}}}
\frac{\sqrt{\ell_j}|\delta_j|}{n}.
\]
Thus
\[
\left\|
\sum_{|r|<\ell_j}
K_{\mathrm{LRV}}(r/\ell_j)
\frac1n\sum_i
m_{i,j}^{(3)}D_{i-|r|,j}^{(3),\epsilon}
\right\|_q
\le
Cq^{c_{\mathrm{dep}}}
\frac{\ell_j^{3/2}|\delta_j|}{n}.
\]
\begingroup
Taking \(q=q_A=2\lceil (A+3)\log(np)\rceil,\) and taking the union over $j$ and all candidate bandwidths, then \(q_A^{c_{\mathrm{dep}}}
\le
C_A\{\log(np)\}^{c_{\mathrm{dep}}}.\)
\endgroup
Applying Markov's inequality and using a union bound over $j\le p$ give
\begingroup
\begin{align*}
&\Pr\left(
\max_{j\le p}\left|\sum_{|r|<\ell_j}
K_{\mathrm{LRV}}(r/\ell_j)
\frac1n\sum_i
m_{i,j}^{(3)}D_{i-|r|,j}^{(3),\epsilon}\right|>
C_A
\frac{\ell_{\max}^{3/2}\|\bdelta\|_\infty
\{\log(np)\}^{c_{\mathrm{dep}}}}{n}
\right)\\
&\qquad\le
pn e^{-q_A}
\le
C_A(np)^{-A}.
\end{align*}
\endgroup
The same bound holds for $n^{-1}\sum_i
D_{i,j}^{(3),\epsilon}m_{i-|r|,j}^{(3)}$.
\begingroup
The stated branchwise choices satisfy $c_{\mathrm{LRV}}\ge c_{\mathrm{dep}}$.
\endgroup
Together with \eqref{eq:lrv-contamination-decomposition}, this gives
\begingroup
\begin{equation}
\label{eq:lrv-contamination-bound}
\Pr\left\{
\max_{j\le p}
\left|
\widetilde\sigma_j^2-\widetilde\sigma_{j,\epsilon}^2
\right|
>
C_A\left(
\frac{\ell_{\max}^{3/2}\|\bdelta\|_\infty
\{\log(np)\}^{c_{\mathrm{LRV}}}}{n}
+
\frac{\ell_{\max}^2\|\bdelta\|_\infty^2}{n}
\right)
\right\}
\le C_A(np)^{-A}.
\end{equation}
\endgroup

Combining \eqref{eq:lrv-error-decomposition}, \eqref{eq:lrv-noise-stochastic-bound}, \eqref{eq:lrv-noise-bias-bound}, and \eqref{eq:lrv-contamination-bound} gives, for every fixed $A>0$ and all sufficiently large $n$,
\begingroup
\begin{equation}
\label{eq:lrv-square-error-bound}
\Pr\left\{
\max_{j\le p}|\widetilde\sigma_j^2-\sigma_j^2|
>C_Aa_{\sigma,n}^{(\delta)}
\right\}
\le C_A(np)^{-A}.
\end{equation}
\endgroup

Assumption~\ref{ass:C2} gives
\(
\inf_{j\le p}\sigma_j^2\ge C_0.
\)
\begingroup
For all sufficiently large $n$, $n^{-1}\le C_0$ and the map
$x\mapsto\max(x,n^{-1})$ cannot increase its distance from $\sigma_j^2$.
Thus, without requiring the displayed rate to tend to zero,
\begin{align*}
\left|\frac{\widehat\sigma_j}{\sigma_j}-1\right|
&=\frac{|\max(\widetilde\sigma_j^2,n^{-1})-\sigma_j^2|}
{\sigma_j(\widehat\sigma_j+\sigma_j)}
\le C_0^{-1}|\widetilde\sigma_j^2-\sigma_j^2|.
\end{align*}
\endgroup
This proves $r_{\sigma,n}=O_p\{a_{\sigma,n}^{(\delta)}\}$. 
\end{proof}

\begin{lemma}\label{lemma:LRV-rate-log}
{Suppose Assumptions~}\ref{ass:C1}, \ref{ass:lrv-M}(i)--(ii) hold. For each fixed $A>0$, 
if $\log p=O(n^{\nu})$ with 
\begin{equation}
\label{eq:lrv-log-rate-nu}
0<\nu<\min\left\{
\frac{1/2-c_{\ell^+}}{A+c_{\mathrm{LRV}}},
\frac{c_{\ell^-}\widetilde q}{A}%
\right\},
\end{equation}
$\{\log(np)\}^Aa_{\sigma,n}\to0$.
\end{lemma}

\begin{proof}
By Assumptions~\ref{ass:C1} and~\ref{ass:lrv-M}(i)--(ii),
\(
1/2-c_{\ell^+}>0,
c_{\ell^-}\widetilde q>0,
1-2c_{\ell^+}>0,
\varkappa c_{\ell^-}>0.
\)
Since $\log p=O(n^{\nu})$, {we have} $\log(np)=\log n+\log p=O(n^{\nu})$ and $\{\log(np)\}^B=O(n^{B\nu})$ for every fixed $B>0$. By the definition of $a_{\sigma,n}$ in
Lemma~\ref{lem:lrv-general},
\begingroup
\begin{align*}
 \{\log (np)\}^A a_{\sigma,n}
={}&
\frac{\ell_{\max}\{\log (np)\}^{A+c_{\mathrm{LRV}}}}{\sqrt n}
+
\ell_{\min}^{-\widetilde q}\{\log (np)\}^A\\
&+
\eta_{\ell_{\min}}\{\log (np)\}^A
+
\frac{\ell_{\max}^2\{\log (np)\}^A}{n}.
\end{align*}
The polynomial terms are, respectively,
\[
\begin{aligned}
\frac{\ell_{\max}L_n^{A+c_{\mathrm{LRV}}}}{\sqrt n}
&=O\{n^{c_{\ell^+}-1/2+\nu(A+c_{\mathrm{LRV}})}\}=o(1),\\
\ell_{\min}^{-\widetilde q}L_n^A
&=O\{n^{-c_{\ell^-}\widetilde q+A\nu}\}=o(1),\\
\frac{\ell_{\max}^2L_n^A}{n}
&=O\{n^{2c_{\ell^+}-1+A\nu}\}=o(1),\qquad L_n=\log(np).
\end{aligned}
\]
Indeed, the first restriction in \eqref{eq:lrv-log-rate-nu} gives
$A\nu<1/2-c_{\ell^+}<1-2c_{\ell^+}$, so no additional restriction is
needed for the fourth term.
\endgroup

Under Assumption~\ref{ass:C1}(AM) or Assumption~\ref{ass:C1}(PD),
\[
\eta_{\ell_{\min}}
\le
C\exp(-c\ell_{\min}^{\varkappa})
\le
C\exp\{-c_1n^{\varkappa c_{\ell^-}}\}
\]
for some constant $c_1>0$. Under
Assumption~\ref{ass:C1}(MD),
$\eta_{\ell_{\min}}=0$ for all sufficiently large $n$, so the same
upper bound is valid in all three cases. Since
$A\log\log(np)=O(\log n)=o(n^{\varkappa c_{\ell^-}})$,
\[
\log\left[\eta_{\ell_{\min}}\{\log (np)\}^A\right]\le
\log\left[
\exp\{-c_1n^{\varkappa c_{\ell^-}}\}
\{\log(np)\}^A
\right]
=
-c_1n^{\varkappa c_{\ell^-}}
+
A\log\log(np)
\rightarrow-\infty,
\]
and hence $\eta_{\ell_{\min}}\{\log (np)\}^A\rightarrow0$. Combining the
four bounds proves the claim.
\end{proof}

\begin{proof}[Proof of Remark~\ref{rem:lrv-polynomial-rate}]
Let
\[
c_\sigma^\star
=
\frac{2\widetilde q-1}{2(1+2\widetilde q)}.
\]
If $\widetilde q>1/2$ and $\ell_j\asymp n^{1/(1+2\widetilde q)}$ uniformly in $j$, then
\[
\frac{\ell_{\max}}{\sqrt n}=O(n^{-c_\sigma^\star}),
\qquad
\ell_{\min}^{-\widetilde q}
=O\{n^{-\widetilde q/(1+2\widetilde q)}\},
\qquad
\frac{\ell_{\max}^2}{n}=O(n^{-2c_\sigma^\star}).
\]
The term $\eta_{\ell_{\min}}$ decreases faster than every negative power of $n$.  Moreover,
\[
\log(np)=O(n^{\nu}),
\qquad
\frac{\ell_{\max}\{\log(np)\}^{c_{\mathrm{LRV}}}}{\sqrt n}
=O\left(n^{-c_\sigma^\star+\nu c_{\mathrm{LRV}}}\right),
\]
Together with
\(
\frac{\widetilde q}{1+2\widetilde q}>c_\sigma^\star,
\)
every fixed \(c<c_\sigma^\star\), after taking
\(\nu<(c_\sigma^\star-c)/c_{\mathrm{LRV}}\), satisfies
$a_{\sigma,n}=O(n^{-c})$ under $H_0$.
\end{proof}

\begingroup

\endgroup
\subsection{Implementation of finite-sample long-run variance and max-tail calibration}

{The max-$L_\infty$ theory uses the third-difference estimator analyzed above,} the explicit rate $r_{\sigma,n}=O_p(a_{\sigma,n})$, and the
Gumbel $p$-values in Theorem~\ref{null:Max}.  The numerical implementation
additionally applies the RCP preprocessing of \citet{chan2022AoS} inside each
componentwise long-run variance calculation and uses the analytic finite-LRV
tails displayed in the main article.  These operations do not alter either
CUSUM numerator.  {These finite-sample adjustments are not covered by Theorem~\ref{null:Max} or the asymptotic independence theorems.}
The code forms autocovariances from the available third-difference series of
length $n-6\ell_j$.  Relative to the original-sample denominator in the formal
estimator, this changes a coordinatewise term by the multiplicative factor
\[
\frac{n}{n-6\ell_j}=1+O\left(\frac{\ell_j}{n}\right).
\]
Consequently, uniformly in $j$,
\[
O\left(\frac{\ell_j}{n}\right)
=O\left(\frac{\ell_{\max}^2}{n}\right),
\]
where the last bound uses $\ell_{\min}\to\infty$.  This implementation
choice is therefore already covered by the deterministic term in
$a_{\sigma,n}$.

The exact coefficient vector in the formal estimator satisfies
$\sum_{a=0}^3d_a^{(3)}=0$ and
$\sum_{a=0}^3\{d_a^{(3)}\}^2=1$.  {The displayed vector $(0.1942,0.2809,0.3832,-0.8582)$ is a four-decimal approximation. The proof uses the exact coefficients, not these rounded values.}

{The Gamma calculation underlying the finite-LRV tail formulas is as follows.}  If
$R\sim\operatorname{Gamma}(\nu/2,\nu/2)$ in shape--rate notation, then
$\E R=1$, $\Var(R)=2/\nu$, and, for every fixed $x\ge0$,
\[
-\frac\nu2\log\left(1+\frac{4x^2}{\nu}\right)
=-2x^2+O(\nu^{-1}).
\]
Also $R\to1$ in probability.  Since $0\le e^{-b\sqrt R}\le1$, bounded
convergence gives, for every fixed $b\ge0$,
\[
\E\{e^{-b\sqrt R}\}\rightarrow e^{-b}.
\]
Consequently the analytic formulas recover the first-order coordinate tails
at fixed arguments as $\nu\to\infty$.  No uniform-in-threshold expansion is
used in the formal theory: Theorem~\ref{null:Max} instead relies on the
{ratio-consistency rate proved above}.  The Gamma formulas are retained
only to document the finite-sample implementation.

\section{Auxiliary tools}
\label{sec:supp-auxiliary-tools}

\begin{lemma}%
\label{lem:vector-basic}
Under Assumption~\ref{ass:C1}, the following bounds hold uniformly in $p$.
\begin{enumerate}
\begingroup
\item[(i)] For every $h\in\mathbb Z$,
$\|\bGam(h)\|_{\mathrm{op}}\leq\varpi_{\mathrm{dep}}(|h|)$, where
\[
\varpi_{\mathrm{dep}}(r)=
\begin{cases}
C\exp(-cr^{\varkappa}),&\text{under (AM) or (PD)},\\
C\mathbb I(r\le m_0),&\text{under (MD)},
\end{cases}
\]
for constants $C,c>0$ independent of $p$.
Consequently, $\sum_{h\in\mathbb Z}(1+|h|)^8\|\bGam(h)\|_{\mathrm{op}}<\infty$.
\item[(ii)] For every $q\ge2$, every deterministic unit vector $\mathbf u\in\mathbb R^p$, and every finitely supported deterministic sequence $\boldsymbol w=(w_i)_{i\in \mathbb{Z}}$, 
\endgroup
\begin{equation}\label{eq:c_dep}
\left\|\sum_{i\in \mathbb{Z}}w_i\mathbf u^\top\bepsilon_i\right\|_q
\le Cq^{c_{\mathrm{dep}}}\|\boldsymbol w\|_2,
\end{equation}
\begingroup
with the branchwise choices
\begin{equation}
\label{eq:branchwise-cdep-choice}
c_{\mathrm{dep}}
=
\begin{cases}
1+1/\varkappa, & \text{under (AM)},\\
1/2,           & \text{under (MD)},\\
1,             & \text{under (PD)}.
\end{cases}
\end{equation}
\endgroup
\end{enumerate}
\end{lemma}

\begin{proof}
(i).
\begingroup
By the variational definition of the operator norm,
\endgroup
\[
\|\bGam(h)\|_{\mathrm{op}} = \sup\limits_{\|\mathbf u\|=1,\|\mathbf v\|=1}|\mathbf u^\top\bGam(h)\mathbf v|.
\]
It suffices to consider \(|\mathbf u^\top\bGam(h)\mathbf v|\).
For unit vectors $\mathbf u,\mathbf v\in\mathbb R^p$, let
$Y_0=\mathbf u^\top\bepsilon_0$ and
$Y_h=\mathbf v^\top\bepsilon_h$.  

Under Assumption~\ref{ass:C1}(AM), Davydov's covariance inequality \citep{davydov1968} gives
\begin{align*}
|\mathbf u^\top\bGam(h)\mathbf v|&=|\Cov(Y_0,Y_h)|\lesssim\alpha_{\epsilon,p}(|h|)^{1/2}\|Y_0\|_4\|Y_h\|_4\le C\exp(-c|h|^{\varkappa}).
\end{align*}

Under Assumption~\ref{ass:C1}(MD), for $|h|>m_0$, $\mathbf u^\top\bGam(h)\mathbf v=\mathbf 0$; for $|h|\leq m_0$, 
\begin{align*}
|\mathbf u^\top\bGam(h)\mathbf v|
&\le \|Y_0\|_2\|Y_h\|_2 \le C.
\end{align*}

Under Assumption~\ref{ass:C1}(PD), it is enough to consider $h\ge0$ since $\bGam(-h)=\bGam(h)^\top$. Let \(\mathcal F_j=\sigma(\boldsymbol\xi_j,\boldsymbol\xi_{j-1},\ldots),\) {and define the martingale projection by} $P_jX=\E(X\mid\mathcal F_j)-\E(X\mid\mathcal F_{j-1})$. Then $P_jX$ is $\mathcal F_{j}$-measurable, and $\E(P_j X\mid\mathcal F_{j-1})=0$.
 {For $j<k$ and $X,Z\in L_2$,} $\E(P_jXP_kZ)=\E[\E(P_jXP_kZ\mid\mathcal F_{k-1})]=0$. By Assumption~\ref{ass:C1}, {we have} $Y_0=\sum_{r=0}^{\infty}P_{-r}Y_0$ in $L_2$, $Y_h=\sum_{s=0}^{\infty}P_{h-s}Y_h$ in $L_2$. Hence $\Cov(Y_0,Y_h)$ {can be expressed as} 
 $$
 \Cov(Y_0,Y_h)=\sum_{r=0}^{\infty}\sum_{s=0}^{\infty}\E(P_{-r}Y_0\,P_{h-s}Y_h)=\sum_{r=0}^{\infty}\E(P_{-r}Y_0\,P_{-r}Y_h).$$ Assumption~\ref{ass:C1}(PD) and Jensen's inequality yield
\[
\|P_{-r}Y_0\|_2 = \|\E(Y_0-Y_0^{\{r\}}\mid\mathcal F_{-r})\|_2\le \|Y_0-Y_0^{\{r\}}\|_2\lesssim \exp(-cr^{\varkappa}),\text{ and }
\]
\[
\|P_{-r}Y_h\|_2\le \|Y_h-Y_h^{\{h+r\}}\|_2\lesssim\exp\{-c(h+r)^{\varkappa}\}.
\]
Therefore, by Cauchy--Schwarz inequality,
\begin{align*}
|\mathbf u^\top\bGam(h)\mathbf v|
&\lesssim\sum_{r=0}^{\infty}
\exp(-cr^{\varkappa})
\exp\{-c(h+r)^{\varkappa}\}\lesssim \exp(-ch^{\varkappa}).
\end{align*}
Taking the supremum over $\mathbf u$ and $\mathbf v$ proves part (i).

\begingroup
(ii). Let $X_i=\mathbf u^\top\bepsilon_i$. By Assumption~\ref{ass:C1}, $\E X_i=0$ and $\sup_{i\in\mathbb Z}\|X_i\|_q\leq K_{\mathrm{sg}}q^{1/2}$ for every $q\ge2$, with $K_{\mathrm{sg}}<\infty$.
\endgroup

Under Assumption~\ref{ass:C1}(AM), define $\mathcal F_j^\epsilon=\sigma(\bepsilon_s:s\leq j)$ and $P_j^\epsilon Z=\E(Z\mid\mathcal F_j^\epsilon)-\E(Z\mid\mathcal F_{j-1}^\epsilon)$. Since the process is strongly mixing, by the reverse martingale convergence theorem, $X_i=\sum_{r=0}^{\infty}P_{i-r}^\epsilon X_i$ in $L_q$.
Let $q'=q/(q-1)$. For $r\geq1$, by the
$L_q$--$L_{q'}$ duality and Davydov's covariance inequality,
\[
\begin{aligned}
\left\|\E(X_i\mid\mathcal F_{i-r}^\epsilon)\right\|_q&=\sup_{\substack{H\in L_{q'}(\mathcal F_{i-r}^\epsilon)\\\|H\|_{q'}\leq1}}\left|\E\left[H\E(X_i\mid\mathcal F_{i-r}^\epsilon)\right]\right|=\sup_{\substack{H\in L_{q'}(\mathcal F_{i-r}^\epsilon)\\\|H\|_{q'}\leq1}}|\E(HX_i)|\\
&=\sup_{\substack{H\in L_{q'}(\mathcal F_{i-r}^\epsilon)\\\|H\|_{q'}\leq1}}|\Cov(H,X_i)|\\
&\leq 8\alpha_{\epsilon,p}(r)^{1-1/q'-1/(2q)}\sup_{\substack{H\in L_{q'}(\mathcal F_{i-r}^\epsilon)\\\|H\|_{q'}\leq1}}\|H\|_{q'}\|X_i\|_{2q}\\
&=8\alpha_{\epsilon,p}(r)^{1/(2q)}\|X_i\|_{2q}
\lesssim q^{1/2}\exp\left(-\frac{cr^\varkappa}{q}\right).
\end{aligned}
\]
Consequently, for every $r\geq0$,
\[
\|P_{i-r}^\epsilon X_i\|_q\leq \left\|\E(X_i\mid\mathcal F_{i-r}^\epsilon)\right\|_q+\left\|\E(X_i\mid\mathcal F_{i-r-1}^\epsilon)\right\|_q\leq Cq^{1/2}\exp\left(-\frac{cr^\varkappa}{q}\right),
\]
where the case $r=0$ follows directly from $\|P_i^\epsilon X_i\|_q\leq2\|X_i\|_q$.

For each fixed $r$, the sequence $\{P_{i-r}^\epsilon X_i\}_{i\in\mathbb Z}$ is a martingale-difference sequence with respect to $\{\mathcal F_{i-r}^\epsilon\}_{i\in\mathbb Z}$. Hence Burkholder's inequality gives
\[
\begin{aligned}
\left\|\sum_{i\in\mathbb{Z}} w_iP_{i-r}^\epsilon X_i\right\|_q\lesssim q^{1/2}\left\{\sum_{i\in\mathbb{Z}}w_i^2\|P_{i-r}^\epsilon X_i\|_q^2\right\}^{1/2}\lesssim q\exp\left(-\frac{cr^\varkappa}{q}\right)\|\boldsymbol w\|_2.
\end{aligned}
\]
Therefore,
\[
\begin{aligned}
\left\|\sum_{i\in\mathbb{Z}}w_iX_i\right\|_q\lesssim q\|\boldsymbol w\|_2\sum_{r=0}^{\infty}\exp\left(-\frac{cr^\varkappa}{q}\right)\lesssim q^{1+1/\varkappa}\|\boldsymbol w\|_2.
\end{aligned}
\]
Thus, \eqref{eq:c_dep} holds by taking $c_{\mathrm{dep}}=1+1/\varkappa$ under (AM).

Under Assumption~\ref{ass:C1}(MD), for $r=0,\ldots,m_0$, let $\mathcal I_r=\{i\in\mathbb Z:i=r+k(m_0+1)\text{ for some }k\in\mathbb Z\}$. For each $r$, the collection $\{X_i:i\in\mathcal I_r\}$ is jointly independent. The marginal sub-Gaussian moment condition therefore gives
\[
\left\|\sum_{i\in\mathcal I_r}w_iX_i\right\|_q\lesssim q^{1/2}\left(\sum_{i\in\mathcal I_r}w_i^2\right)^{1/2}.
\]
By Minkowski's inequality and the Cauchy--Schwarz inequality,
\[
\begin{aligned}
\left\|\sum_{i\in\mathbb{Z}}w_iX_i\right\|_q&\lesssim q^{1/2}\sum_{r=0}^{m_0}\left(\sum_{i\in\mathcal I_r}w_i^2\right)^{1/2}\lesssim q^{1/2}\sqrt{m_0+1}\,\left(\sum_{r=0}^{m_0}\sum_{i\in\mathcal I_r}w_i^2\right)^{1/2}\leq Cq^{1/2}\|\boldsymbol w\|_2.
\end{aligned}
\]
Thus, \eqref{eq:c_dep} holds by taking $c_{\mathrm{dep}}=1/2$ under (MD).

Under Assumption~\ref{ass:C1}(PD), let $\mathcal G_j=\sigma(\boldsymbol\xi_j,\boldsymbol\xi_{j-1},\ldots)$ and $P_j^\xi Z=\E(Z\mid\mathcal G_j)-\E(Z\mid\mathcal G_{j-1})$. The martingale projection decomposition gives $X_i=\sum_{r=0}^{\infty}P_{i-r}^\xi X_i$ in $L_q$. Similarly, {we have} %
\[
\begin{aligned}
\|P_{i-r}^\xi X_i\|_q\leq\|X_i-X_i^{\{r\}}\|_q\lesssim q^{1/2}\exp(-cr^\varkappa).
\end{aligned}
\]
For each fixed $r$, $\{P_{i-r}^\xi X_i\}_{i\in\mathbb Z}$ is a martingale-difference sequence. Therefore, by Burkholder's inequality,
\[
\begin{aligned}
\left\|\sum_{i\in\mathbb{Z}}w_iP_{i-r}^\xi X_i\right\|_q \lesssim q^{1/2}\left\{\sum_{i\in\mathbb{Z}}w_i^2\|P_{i-r}^\xi X_i\|_q^2\right\}^{1/2}\lesssim q\exp(-cr^\varkappa)\|\boldsymbol w\|_2.
\end{aligned}
\]
Summing over $r\geq0$ yields $\|\sum_{i\in\mathbb{Z}}w_iX_i\|_q\leq Cq\|\boldsymbol w\|_2$. Thus, \eqref{eq:c_dep} holds by taking $c_{\mathrm{dep}}=1$ under (PD).
This proves part~(ii).

\end{proof}

\begin{lemma}\label{lem:cumulant-coefficient-contraction}
Suppose Assumption~\ref{ass:C1} holds. Let $q\in\{2,3,4\}$, let $\pi$ be a partition of $\{1,\ldots,2q\}$, and for each $a=1,\ldots,2q$, let $x_a=(x_{a,1},\ldots,x_{a,n})^\top\in\mathbb R^n$ be a deterministic coefficient vector.
For notational convenience, set $x_{a,i}=0$, $i\notin\{1,\ldots,n\}$. (i) With $j(a)=j_{\lceil a/2\rceil}$ and $\mathfrak C_{\pi,p}$ defined before Assumption~\ref{ass:C3},
\begin{align*}
&\sum_{j_1,\ldots,j_q=1}^p
\sum_{i_1,\ldots,i_{2q}\in\mathbb Z}
\prod_{B\in\pi}
\left|\Cum\{\epsilon_{i_a,j(a)}:a\in B\}\right|
\prod_{a=1}^{2q}|x_{a,i_a}|\le
\mathfrak C_{\pi,p}\prod_{a=1}^{2q}\|x_a\|_2.
\end{align*}
(ii) If $\pi$ has a single block $B=\{1,\ldots,2q\}$, then, for every $a_0\in B$,
\begin{align*}
&\sum_{j_1,\ldots,j_q=1}^p
\sum_{i_1,\ldots,i_{2q}\in\mathbb Z}
\left|\Cum\{\epsilon_{i_a,j(a)}:a\in B\}\right|
\prod_{a=1}^{2q}|x_{a,i_a}|\le
\mathfrak C_{\pi,p}\,
\|x_{a_0}\|_1\prod_{a\ne a_0}\|x_a\|_\infty.
\end{align*}
\end{lemma}

\begin{proof}
Fix $j_1,\ldots,j_q$ and a block $B\in\pi$. Let $a_B=\min B$ and $r_{a_B}=0$, and define $r_a=i_a-i_{a_B}$ for $a\in B\setminus\{a_B\}$. By strict stationarity,
\[
\begin{aligned}
S_B
&:=
\sum_{(i_a:a\in B)\in\mathbb Z^{|B|}}
\left|
\operatorname{cum}
\{\epsilon_{i_a,j(a)}:a\in B\}
\right|
\prod_{a\in B}|x_{a,i_a}|\\
&=
\sum_{(r_a:a\in B\setminus\{a_B\})}
\left|
\operatorname{cum}
\left\{
\epsilon_{0,j(a_B)},
\epsilon_{r_a,j(a)}:
a\in B\setminus\{a_B\}
\right\}
\right|
\sum_{t\in\mathbb Z}
\prod_{a\in B}|x_{a,t+r_a}|.
\end{aligned}
\]

If $|B|\geq2$, Hölder's inequality gives
\[
\begin{aligned}
\sum_{t\in\mathbb Z}\prod_{a\in B}|x_{a,t+r_a}|\leq \prod_{a\in B}\left(\sum_{t\in\mathbb Z}|x_{a,t+r_a}|^{|B|}\right)^{1/|B|}=\prod_{a\in B}\|x_a\|_{|B|}\leq \prod_{a\in B}\|x_a\|_2.
\end{aligned}
\]
Since the weight in $\mathfrak K_B^{\mathrm{cum}}(j_1,\ldots,j_q)$ is at least one,
\[
S_B\leq\mathfrak K_B^{\mathrm{cum}}(j_1,\ldots,j_q)\prod_{a\in B}\|x_a\|_2.
\]
If $B$ is a singleton, then $S_B=0$ because the process is
centered, so the same bound holds trivially.

Because the blocks of $\pi$ are disjoint, Tonelli's theorem yields
\[
\begin{aligned}
&\sum_{i_1,\ldots,i_{2q}\in\mathbb Z}\prod_{B\in\pi}\left|\operatorname{cum}\{\epsilon_{i_a,j(a)}:a\in B\}\right|\prod_{a=1}^{2q}|x_{a,i_a}|=\prod_{B\in\pi}S_B\\
\leq&\prod_{B\in\pi}\mathfrak K_B^{\mathrm{cum}}(j_1,\ldots,j_q)\prod_{a=1}^{2q}\|x_a\|_2.
\end{aligned}
\]
Summing over $j_1,\ldots,j_q$ {proves part (i).}

If $\pi$ consists of the single block $B=\{1,\ldots,2q\}$, then, for any $a_0\in B$,
\[
\begin{aligned}
\sum_{t\in\mathbb Z}\prod_{a\in B}|x_{a,t+r_a}|
&\leq\sum_{t\in\mathbb Z}|x_{a_0,t+r_{a_0}}|\prod_{a\neq a_0}\|x_a\|_\infty=\|x_{a_0}\|_1\prod_{a\neq a_0}\|x_a\|_\infty.
\end{aligned}
\]
Substituting this bound into the same relative-lag expansion and summing over $j_1,\ldots,j_q$ {proves part (ii).}
\end{proof}

\begin{lemma}
\label{lem:bilinear-quadratic-bound}
Suppose Assumptions~\ref{ass:C1}--\ref{ass:C3} hold. Let $\boldsymbol a=(a_i)_{i=1}^n,\boldsymbol b=(b_i)_{i=1}^n$ be deterministic weight vectors. 
 Define $\mathbf Y_a=\sum_{i=1}^n a_i\bepsilon_i$ and $\mathbf Y_b=\sum_{i=1}^n b_i\bepsilon_i$. Then
\[
\left\|p^{-1/2}\left[\mathbf Y_a^\top\mathbf Y_b-\E(\mathbf Y_a^\top\mathbf Y_b)\right]\right\|_4\le C\|\boldsymbol a\|_2\|\boldsymbol b\|_2.
\]
\end{lemma}

\begin{proof}
For notational convenience, we extend $a$ and $b$ to sequences indexed by $\mathbb Z$ by setting $a_i=b_i=0$, $i\notin\{1,\ldots,n\}$. For $1\le j\le p$, let $Y_{a,j}=\sum_{i\in\mathbb Z} a_i\epsilon_{i,j}$, $Y_{b,j}=\sum_{i\in\mathbb Z} b_i\epsilon_{i,j}$ and $ Z=p^{-1/2}\{\mathbf Y_a^\top\mathbf Y_b-\E(\mathbf Y_a^\top\mathbf Y_b)\}$. Since $Z$ is centered, $\E (Z^4)=\Cum(Z,Z,Z,Z)+3\{\Var(Z)\}^2$.

We first bound $\Var(Z)$.  For $\boldsymbol\alpha,\boldsymbol\beta\in\{\boldsymbol a,\boldsymbol b\}$, let
\[
\boldsymbol\Sigma_{\boldsymbol\alpha\boldsymbol\beta}=\Cov(\mathbf Y_{\boldsymbol\alpha},\mathbf Y_{\boldsymbol\beta})=\sum_{h\in\mathbb Z}\left(\sum_{i\in\mathbb Z}\alpha_i\beta_{i+h}\right)\bGam(h).
\]
Lemma~\ref{lem:vector-basic}(i) and Cauchy--Schwarz give
\begin{align*}
\|\boldsymbol\Sigma_{\boldsymbol\alpha\boldsymbol\beta}\|_{\mathrm{op}}\le\sum_{h\in\mathbb{Z}}\|\bGam(h)\|_{\mathrm{op}}\sum_{i\in\mathbb Z}|\alpha_i\beta_{i+h}|\le C\|\boldsymbol\alpha\|_2\|\boldsymbol\beta\|_2,\text{ and }\|\boldsymbol\Sigma_{\boldsymbol\alpha\boldsymbol\beta}\|_{\mathrm F}\le C\sqrt p\,\|\boldsymbol\alpha\|_2\|\boldsymbol\beta\|_2.
\end{align*}
Since the linear forms are centered, the fourth-order moment--cumulant identity yields
\begin{align*}
\Cov\{Y_{a,j}Y_{b,j},Y_{a,k}Y_{b,k}\}
&=\Cov(Y_{a,j},Y_{a,k})\Cov(Y_{b,j},Y_{b,k})+\Cov(Y_{a,j},Y_{b,k})\Cov(Y_{b,j},Y_{a,k})\\
&\qquad\quad+\Cum(Y_{a,j},Y_{b,j},Y_{a,k},Y_{b,k}).
\end{align*}
Therefore,
\begin{align*}
\Var(Z)\le p^{-1}|\Tr(\boldsymbol\Sigma_{aa}\boldsymbol\Sigma_{bb})|+p^{-1}|\Tr(\boldsymbol\Sigma_{ab}\boldsymbol\Sigma_{ba}^{\top})|+p^{-1}\left|\sum_{j,k=1}^p\Cum(Y_{a,j},Y_{b,j},Y_{a,k},Y_{b,k})\right|.
\end{align*}
Lemma~\ref{lem:cumulant-coefficient-contraction} with $q=2$ gives
\[
\left|\sum_{j,k=1}^p\Cum(Y_{a,j},Y_{b,j},Y_{a,k},Y_{b,k})\right|\le Cp\|\boldsymbol a\|_2^2\|\boldsymbol b\|_2^2.
\]
Combining this estimate with the preceding Frobenius-norm bounds yields \(\Var(Z)\le C\|\boldsymbol a\|_2^2\|\boldsymbol b\|_2^2.\) It remains to control the fourth cumulant of $Z$.  By the Leonov--Shiryaev product--cumulant formula,
\[
\Cum(Z,Z,Z,Z)=p^{-2}\sum_{\pi:\,\pi\ \mathrm{connected\ relative\ to}\ \mathfrak P_4}T_\pi,
\]
\begingroup
where, writing $x_{2r-1}=\boldsymbol a$, $x_{2r}=\boldsymbol b$
and $j(a)=j_{\lceil a/2\rceil}$,
\[
T_\pi=\sum_{j_1,\ldots,j_4=1}^p\sum_{i_1,\ldots,i_8\in\mathbb Z}
\left(\prod_{a=1}^8x_{a,i_a}\right)
\prod_{B\in\pi}\Cum\{\epsilon_{i_a,j(a)}:a\in B\}.
\]
\endgroup
For each connected partition $\pi$, multilinearity, Lemma~\ref{lem:cumulant-coefficient-contraction} with $q=4$, and Assumption~\ref{ass:C3} give $|T_\pi|\le Cp\|\boldsymbol a\|_2^4\|\boldsymbol b\|_2^4$. There are only finitely many such partitions; hence $|\Cum(Z,Z,Z,Z)|\le Cp^{-1}\|\boldsymbol a\|_2^4\|\boldsymbol b\|_2^4\le C\|\boldsymbol a\|_2^4\|\boldsymbol b\|_2^4$. Consequently, $\E (Z^4)\le C\|\boldsymbol a\|_2^4\|\boldsymbol b\|_2^4$. Taking fourth roots proves the result.
\end{proof}

\begin{lemma}%
\label{lem:gaussian-tail-general}
Let $\mathbf g\sim N_d(\mathbf0,\mathbf I_d)$ and let $\mathbf A$ {be a symmetric matrix.} Then, for every $x>0$,
\[
\Pr\{|\mathbf g^\top\mathbf A\mathbf g-\Tr(\mathbf A)|>x\}
\le2\exp\left[-c\min\left\{
\frac{x^2}{\|\mathbf A\|_{\mathrm F}^2},
\frac{x}{\|\mathbf A\|_{\mathrm{op}}}
\right\}\right].
\]
If $\mathbf g$ and $\mathbf h$ are independent standard Gaussian vectors, the same bound, after changing $c$, holds for $|\mathbf g^\top\mathbf A\mathbf h|$.
\end{lemma}

\begin{proof}
\begingroup
The assertion is immediate for $\mathbf A=0$. For $\mathbf A\ne0$, we
bound the centered quadratic moment generating function and optimize its
Chernoff bound; the bilinear assertion follows by a block-matrix embedding.
\endgroup
Let $\lambda_1,\ldots,\lambda_d$ be the eigenvalues of $\mathbf A$.  For
$|t|\le(4\|\mathbf A\|_{\mathrm{op}})^{-1}$,
\begin{align*}
\log\E\exp\{t(\mathbf g^\top\mathbf A\mathbf g-\Tr\mathbf A)\}
&=\sum_{r=1}^d\left[-\frac12\log(1-2t\lambda_r)-t\lambda_r\right]\\
&\le Ct^2\sum_{r=1}^d\lambda_r^2
=Ct^2\|\mathbf A\|_{\mathrm F}^2.
\end{align*}
For every $0<t\le(4\|\mathbf A\|_{\mathrm{op}})^{-1}$, with Chernoff's
inequality, it follows that
\begin{align*}
\Pr(\mathbf g^\top\mathbf A\mathbf g-\Tr(\mathbf A)\ge x)
&{}\le\exp\left[-tx+\log\E\exp\{t(\mathbf g^\top\mathbf A\mathbf g-\Tr\mathbf A)\}\right],\\
&{}\le\exp(-tx+Ct^2\|\mathbf A\|_{\mathrm F}^2).
\end{align*}
To minimize the Chernoff exponent over its admissible range, take
\(
t_x
=
\min\left\{
\frac{x}{2C\|\mathbf A\|_{\mathrm F}^2},
\frac{1}{4\|\mathbf A\|_{\mathrm{op}}}
\right\}.
\)
Consequently, 
\[
\Pr(\mathbf g^\top\mathbf A\mathbf g-\Tr(\mathbf A)\ge x)
\le
\exp\left[
-c\min\left\{
\frac{x^2}{\|\mathbf A\|_{\mathrm F}^2},
\frac{x}{\|\mathbf A\|_{\mathrm{op}}}
\right\}
\right].
\]
Since $\|-\mathbf A\|_{\mathrm F}=\|\mathbf A\|_{\mathrm F}$,
$\|-\mathbf A\|_{\mathrm{op}}=\|\mathbf A\|_{\mathrm{op}}$, replacing $\mathbf A$ by $-\mathbf A$ in the upper-tail bound gives the same upper bound for $\Pr(\mathbf g^\top\mathbf A\mathbf g-\Tr(\mathbf A)\le -x)$.
Therefore,
\[
\Pr(|\mathbf g^\top\mathbf A\mathbf g-\Tr(\mathbf A)|>x)
\le
2\exp\left[
-c\min\left\{
\frac{x^2}{\|\mathbf A\|_{\mathrm F}^2},
\frac{x}{\|\mathbf A\|_{\mathrm{op}}}
\right\}
\right].
\]

If $\mathbf g$ and $\mathbf h$ are independent standard Gaussian vectors, let
\[
\mathbf z=(\mathbf g^\top,\mathbf h^\top)^\top,
\qquad
\widetilde{\mathbf A}
=\frac12\begin{pmatrix}\mathbf0&\mathbf A\\
\mathbf A^\top&\mathbf0\end{pmatrix}.
\]
Then $\mathbf g^\top\mathbf A\mathbf h=\mathbf z^\top\widetilde{\mathbf A}\mathbf z$,
$\Tr(\widetilde{\mathbf A})=0$,
$\|\widetilde{\mathbf A}\|_{\mathrm F}^2=\|\mathbf A\|_{\mathrm F}^2/2$, and
$\|\widetilde{\mathbf A}\|_{\mathrm{op}}=\|\mathbf A\|_{\mathrm{op}}/2$.  The first assertion completes the proof.
\end{proof}

\end{document}